\documentclass[10pt,twoside]{article}
\usepackage{amsmath,amssymb,latexsym,graphicx,hyperref,epstopdf,color,subfigure}
\usepackage[noend]{algpseudocode}
\usepackage{algorithm,epstopdf}
\usepackage[capitalise]{cleveref}
\normalsize
\newtheorem{theorem}{Theorem}
\newtheorem{lemma}{Lemma}
\newtheorem{demma}{Discharging-Lemma}
\newtheorem{corollary}{Corollary}

\newtheorem{proposition}{Proposition}[subsubsection]

\newtheorem{conjecture}{Conjecture}
\newenvironment{proof}{{\sc Proof. }}{\hfill$\Box$\vspace{0.2in}}

\title{Planar graphs are acyclically edge $(\Delta + 5)$-colorable}
\author{
	Qiaojun~Shu\thanks{Department of Mathematics, Hangzhou Dianzi University.  Hangzhou 310018, China.
	\texttt{qjshu@hdu.edu.cn}}
\and
	Guohui~Lin\thanks{Correspondence author.
	Department of Computing Science, University of Alberta.  Edmonton, Alberta T6G 2E8, Canada.
	\texttt{guohui@ualberta.ca}}
}
\date{\today}
\begin{document}
\maketitle

\begin{abstract}
An edge coloring of a graph $G$ is to color all the edges in the graph such that adjacent edges receive different colors.
It is acyclic if each cycle in the graph receives at least three colors.
Fiam{\v{c}}ik (1978) and Alon, Sudakov and Zaks (2001) conjectured that
every simple graph with maximum degree $\Delta$ is acyclically edge $(\Delta + 2)$-colorable --- 
the well-known {\em acyclic edge coloring conjecture} (AECC).
Despite many major breakthroughs and minor improvements, the conjecture remains open even for planar graphs.
In this paper, we prove that planar graphs are acyclically edge $(\Delta + 5)$-colorable.
Our proof has two main steps:
Using discharging methods, we first show that every non-trivial planar graph must have one of the eight groups of well characterized local structures;
and then acyclically edge color the graph using no more than $\Delta + 5$ colors by an induction on the number of edges.

\end{abstract}

\subsubsection*{Acknowledgments.}
This research is supported by the NSERC Canada and the China Scholarship Council.

\section{Introduction}
We consider simple and $2$-connected graphs in this paper,
i.e., graphs without self-loops, without parallel edges but with at least two vertex-disjoint paths between any pair of vertices,
as explained below.

We use $V(G)$ and $E(G)$ to denote the vertex set and the edge set of a graph $G$, respectively, 
and simplify them as $V$ and $E$ respectively when the graph is clear from the context.
This way, $G$ is also denoted as $G = (V, E)$.

For an integer $k \ge 2$, we use the integers $1, 2, \ldots, k$ to represent the $k$ distinct colors.
An {\em edge $k$-coloring} of the graph $G$ is a mapping $c: E(G) \to \{1, 2, \ldots, k\}$ such that any two adjacent edges receive different colors.
If such a mapping exists, then $G$ is said {\em edge $k$-colorable}.
The {\em chromatic index} of $G$, denoted as $\chi'(G)$, is the smallest integer $k$ such that $G$ is edge $k$-colorable.
An edge $k$-coloring of $G$ is called {\em acyclic} if each cycle receives at least three different colors,
or equivalently there is no {\em bichromatic} cycle in $G$, or equivalently the subgraph of $G$ induced by any two colors is a forest.
Correspondingly, $a'(G)$ denotes the smallest integer $k$ such that $G$ is acyclically edge $k$-colorable, called the {\em acyclic chromatic index} of $G$.

One sees from the definition of acyclic edge coloring the reason why we consider simple and connected graphs only.
Furthermore, a connected graph can be represented as a block-cut tree,
in which each block is a $2$-connected component of the graph and two blocks are adjacent if and only if they overlap at a cut-vertex.
Given an acyclic edge $k$-coloring of the graph $G$, after swapping any two colors inside a block it remains as an acyclic edge $k$-coloring.
Therefore, we may constrain ourselves on $2$-connected graphs.
In the sequel, the input graph is simple and $2$-connected.

We use $\Delta(G)$ to denote the maximum degree of the graph $G$, and simply as $\Delta$ when the graph is clear from the context.
It holds trivially that $\Delta \le \chi'(G) \le a'(G)$.
By Vizing's theorem \cite{Vizing64}, $\chi'(G) \le \Delta + 1$;
it follows from $\chi'(K_3) = 3 = \Delta + 1$ that $\Delta + 1$ is tight for the chromatic index $\chi'(G)$.
For the acyclic chromatic index, one sees that $a'(K_4) = 5 = \Delta + 2$.
Fiam{\v{c}}ik \cite{F78} and Alon, Sudakov and Zaks \cite{ABZ01} independently conjectured that it is upper bounded by $\Delta + 2$,
called the {\em acyclic edge coloring conjecture} (AECC):

\begin{conjecture}\label{conj01} {\rm (AECC)}
For any graph $G$, $a'(G) \le \Delta + 2$.
\end{conjecture}

The AECC has been confirmed true for graphs with $\Delta \in \{3, 4\}$~\cite{Sku04,AMM12,BC09,SWMW19,WMSW19}.
However, for an arbitrary graph $G$, the conjecture is far away from reached.
Yet, several major milestones bounding $a'(G)$ have been achieved:
Alon, McDiarmid and Reed \cite{AMR91} proved that $a'(G) \le 64 \Delta$ by a probabilistic argument;
and then the upper bound was improved to $16 \Delta$ \cite{MR98}, to $\lceil 9.62 (\Delta - 1) \rceil$ \cite{NPS12}, to $4 \Delta - 4$ \cite{EP12},
and most recently in 2017 to $\lceil 3.74 (\Delta - 1) \rceil + 1$ by Giotis et al.~\cite{GKP17} using Lov\'{a}sz local lemma.

A {\em planar} graph is one that can be drawn in the two-dimensional plane so that its edges intersect only at their ending vertices.
The AECC on planar graphs also attracts many studies with much tighter upper bounds.
For an arbitrary planar graph $G$, Basavaraju et al.~\cite{BLSNHT11} first showed that $a'(G) \le \Delta + 12$;
and then the upper bound was improved to $\Delta + 7$ by Wang, Shu and Wang~\cite{WSW127} and the latest to $\Delta + 6$ by Wang and Zhang~\cite{WZ16}.
For some sub-classes of planar graphs, the AECC has been confirmed true,
for example, planar graphs without $i$-cycles for each $i \in \{3, 4, 5, 6\}$ \cite{SW11,WSW4,SWW12,WSWZ14},
and planar graphs without intersecting triangles (i.e., $3$-cycles) \cite{SCH20,SCH21}.

In this paper, we focus on arbitrary planar graphs, aiming to improve the upper bound from $\Delta + 6$ to $\Delta + 5$, stated in the following theorem.
\begin{theorem}
\label{thm01}
If $G$ is a planar graph, then $a'(G)\leq \Delta + 5$.
\end{theorem}

Recall that when $\Delta \le 4$, the AECC has been confirmed true.
We thus prove Theorem~\ref{thm01} assuming $\Delta \ge 5$.
The rest of the paper is organized as follows.
In Section 2, we characterize eight groups of local structures (also called {\em configurations}),
and by discharging methods we prove that $G$ must contain at least one of these local structures.
We prove the theorem in Section 3 by an induction on the number of edges in $G$,
where we employ an important known property of an acyclic edge coloring (Lemma~\ref{lemma06}) from \cite{SWMW19}.
We conclude the paper in Section 4.

\section{Notations and the eight groups of local structures}
We fix a simple and $2$-connected planar graph $G$ with $\Delta\ge 5$ for discussion.

Let $d(v)$ denote the degree of the vertex $v$ in $G$.
A vertex of degree $k$ (at least $k$, at most $k$, respectively) is called a $k$-{\em vertex} ({$k^+$-{\em vertex}, $k^-$-{\em vertex}, respectively).
For convenience, let $n_k(v)$ ($n_{k^{+}}(v)$, $n_{k^{-}}(v)$, respectively) denote the number of $k$-vertices ($k^{+}$-vertices, $k^{-}$-vertices, respectively) adjacent to $v$ in $G$.

We use $[a, b]$ to denote the set of integers in between $a$ and $b$, inclusive.
Let $S_i$ be a set of integers, $i = 1, 2$;
then an {\em $(S_1, S_2)$ vertex} is a $k$-vertex $x$ of $G$ such that $k \in S_1$ and $n_{6^{+}}(x) \in S_2$ (and in this case, we say that $x$ is $(S_1, S_2)$).
When $S_1 = \{k\}$ ($S_2 = \{j\}$, $S_1 = \{k_1, k_2\}$, $S_2 = \{j_1, j_2\}$, respectively), the above notation is simplified as $(k, S_2)$ ($(S_1, j)$, $(k_1/k_2, S_2)$, $(S_1, j_1/j_2)$, respectively).
Similarly, when $S_1 = [k, \infty]$ ($S_1 = [1, k]$, $S_2 = [j, \infty]$, $S_2 = [1, j]$, respectively),
we can define a $(k^+, S_2)$ ($(k^-, S_2)$, $(S_1, j^+)$, $(S_1, j^-)$, respectively) vertex.
For instance, $(9, 5)$ is a degree-$9$ vertex $x$ having exactly five $6^+$-neighbors,
and $(7^+, 4^-)$ is a $7^+$-vertex $x$ with no more than four $6^+$-neighbors.

\begin{figure}[]
\begin{center}
\includegraphics[height=0.9\textheight]{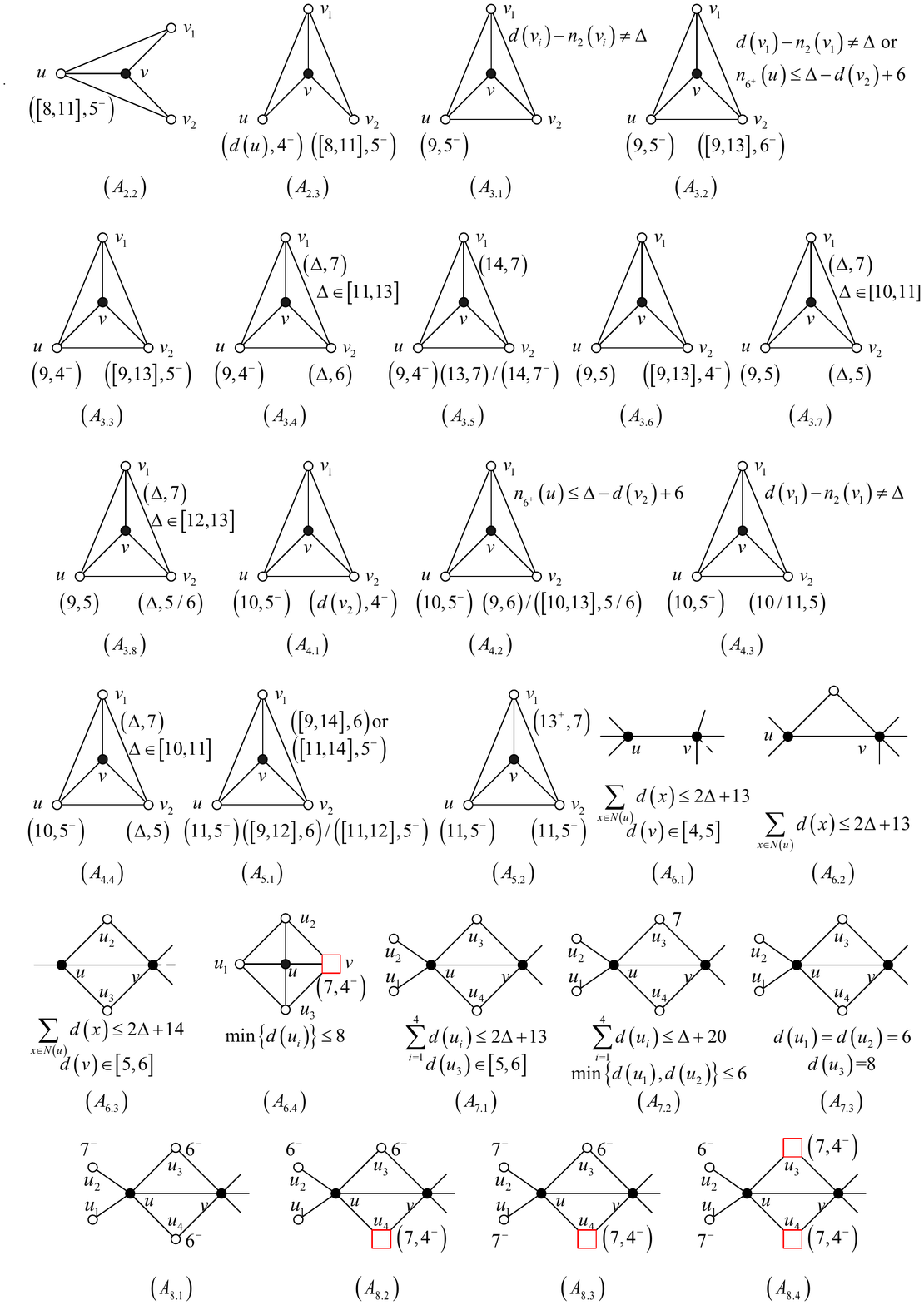}
\caption{Some local structures of the groups ($A_2$)--($A_8$) as defined in Theorem~\ref{thm02}.
	In each local structure, filled vertices are distinct from each other and each has its degree determined,
	while a non-filled vertex has an undetermined degree (i.e., could be incident with more edges not shown).
	In ($A_{6.1}$) and ($A_{6.3}$), a dashed edge indicates the existence of a possible edge incident at the vertex.\label{fig01}}
\end{center}
\end{figure}

\begin{theorem}
\label{thm02}
The graph $G$ contains at least one of the following local structures {\rm ($A_1$)--($A_8$)}, most of which are depicted in Figure~\ref{fig01}.
\begin{description}
\parskip=0pt
\item[{\rm ($A_1$)}]
	A path $u v w$ with $d(v)= 2$ such that at least one of the following holds:
	\begin{description}
	\parskip=0pt
	\item[{\rm ($A_{1.1}$)}]
		$u w\not\in E(G)$;
	\item[{\rm ($A_{1.2}$)}]
		$u w \in E(G)$ and $d(u)\le 7$;
	\item[{\rm ($A_{1.3}$)}]
		$u w \in E(G)$ and $n_{5^-}(u)\ge d(u) - 6$.
	\end{description}
\item[{\rm ($A_2$)}]
	An edge $u v$ with $N(v) = \{u, v_1, v_2\}$ such that at least one of the following holds:
	\begin{description}
	\parskip=0pt
	\item[{\rm ($A_{2.1}$)}]
		$d(u)\le 7$;
	\item[{\rm ($A_{2.2}$)}]
		$u$ is $([8, 11], 5^-)$, and $u v_1, u v_2\in E(G)$, $v_1 v_2\not\in E(G)$;
	\item[{\rm ($A_{2.3}$)}]
		$u$ is $(8, 4^-)$, $v_2$ is $([8, 11], 5^-)$, and $u v_1, v_1 v_2\in E(G)$, $u v_2\not\in E(G)$.
	\item[{\rm ($A_{2.4}$)}]
		$d(u) = 8$, and $u v_1, u v_2, v_1 v_2\in E(G)$.
	\end{description}
\item[{\rm ($A_3$)}]
	A $(9, 5^-)$ vertex $u$ adjacent to a $3$-vertex $v$ with $N(v) = \{u, v_1, v_2\}$ and $u v_1, u v_2, v_1 v_2\in E(G)$ such that at least one of the following holds:
	\begin{description}
	\parskip=0pt
	\item[{\rm ($A_{3.1}$)}]
		$d(v_i)- n_2(v_i)\ne \Delta$ for each $i\in \{1, 2\}$;
	\item[{\rm ($A_{3.2}$)}]
		$v_2$ is $([9, 13], 6^-)$, and $d(v_1)- n_2(v_1)\ne \Delta$ or $n_{6^+}(v_1)\le \Delta - d(v_2) + 6$;
	\item[{\rm ($A_{3.3}$)}]
		$u$ is $(9, 4^-)$ and $v_2$ is $([9, 13], 5^-)$;
	\item[{\rm ($A_{3.4}$)}]
		$u$ is $(9, 4^-)$, $v_1$ is $(\Delta, 7)$, and $v_2$ is $(\Delta, 6)$, where $\Delta \in [11, 13]$;
	\item[{\rm ($A_{3.5}$)}]
		$u$ is $(9, 4^-)$, $v_1$ is $(14, 7)$, and $v_2$ is $(13, 7)$ or $(14, 7^-)$;
    \item[{\rm ($A_{3.6}$)}]
		$u$ is $(9, 5)$ and $v_2$ is $([9, 13], 4^-)$;
	\item[{\rm ($A_{3.7}$)}]
		$u$ is $(9, 5)$, $v_1$ is $(\Delta, 7)$, and $v_2$ is $(\Delta, 5)$, where $\Delta \in [10, 11]$;
	\item[{\rm ($A_{3.8}$)}]
		$u$ is $(9, 5)$, $v_1$ is $(\Delta, 7)$, and $v_2$ is $(\Delta, 5/6)$, where $\Delta \in [12, 13]$.
     \end{description}
\item[{\rm ($A_4$)}]
	A $(10, 5^-)$ vertex $u$ adjacent to a $3$-vertex $v$ with $N(v)= \{u, v_1, v_2\}$ and $uv_1, uv_2, v_1v_2\in E(G)$ such that at least one of the following holds:
    \begin{description}
	\parskip=0pt
	\item[{\rm ($A_{4.1}$)}]
		$v_2$ is $(d(v_2), 4^-)$;
	\item[{\rm ($A_{4.2}$)}]
		$v_2$ is $(9, 6)$ or $([10, 13], 5/6)$, and $n_{6^+}(v_1)\le \Delta - d(v_2) + 6$;
	\item[{\rm ($A_{4.3}$)}]
		$v_2$ is $(10/11, 5)$ and $d(v_1)- n_2(v_1)\ne \Delta$;
	\item[{\rm ($A_{4.4}$)}]
		$v_1$ is $(\Delta, 7)$, and $v_2$ is $(\Delta, 5)$, where $\Delta \in [10, 11]$.
     \end{description}
\item[{\rm ($A_5$)}]
	A $(11, 5^-)$ vertex $u$ adjacent to a $3$-vertex $v$ with $N(v) = \{u, v_1, v_2\}$ and $u v_1, u v_2, v_1 v_2\in E(G)$ such that at least one of the following holds:
    \begin{description}
	\parskip=0pt
	\item[{\rm ($A_{5.1}$)}]
		$v_1$ is $([9, 14], 6)$ or $([11, 14], 5^-)$, and $v_2$ is $([9, 12], 6)$ or $([11, 12], 5^-)$;
	\item[{\rm ($A_{5.2}$)}]
		$v_1$ is $(13^+, 7)$ and $v_2$ is $(11, 5^-)$.
     \end{description}
\item[{\rm ($A_6$)}]
	A $4$-vertex $u$ with $N(u) = \{v, u_1, u_2, u_3\}$ such that at least one of the following cases holds:
    \begin{description}
	\parskip=0pt
	\item[{\rm ($A_{6.1}$)}]
		$d(v)\in \{4, 5\}$ and $\sum_{x\in N(u)}d(x)\le 2\Delta + 13$;
	\item[{\rm ($A_{6.2}$)}]
		$d(v) = 6$, $uv$ is contained in a $3$-cycle, and $\sum_{x\in N(u)}d(x)\le 2\Delta + 13$;
	\item[{\rm ($A_{6.3}$)}]
		$d(v)\in \{5, 6\}$, $uv$ is contained in two $3$-cycles, and $\sum_{x\in N(u)}d(x)\le 2\Delta + 14$;
	\item[{\rm ($A_{6.4}$)}]
		$v$ is $(7, 4^-)$, $u_1 u_2$, $u_1 u_3$, $v u_2$, $v u_3\in E(G)$, and $\min \{d(u_1), d(u_2), d(u_3)\}\le 8$.
    \end{description}
\item[{\rm ($A_7$)}]
	A $5$-vertex $u$ adjacent to a $5$-vertex $v$ with $N(u) = \{v, u_1, u_2, u_3, u_4\}$ and $uv$ contained in two $3$-cycles $u v u_3$,
    $uvu_4$ such that at least one of the following holds:
   \begin{description}
	\parskip=0pt
	\item[{\rm ($A_{7.1}$)}]
		$d(u_3)\in \{5, 6\}$ and $\sum_{x\in N(u)\setminus \{v\}}d(x)\le 2\Delta + 13$;
	\item[{\rm ($A_{7.2}$)}]
		$d(u_3) = 7$, $\min \{d(u_1), d(u_2)\}\le 6$, and $\sum_{x\in N(u)\setminus \{v\}}d(x)\le \Delta + 20$;
	\item[{\rm ($A_{7.3}$)}]
		$d(u_3) = 8$ and $d(u_1) = d(u_2) = 6$.
     \end{description}
\item[{\rm ($A_8$)}]
	A $5$-vertex $u$ adjacent to a $6$-vertex $v$ with $N(u) = \{v, u_1, u_2, u_3, u_4\}$ and $u v$ contained in two $3$-cycles $u v u_3$,
    $u v u_4$ such that at least one of the following holds:
   \begin{description}
	\parskip=0pt
	\item[{\rm ($A_{8.1}$)}]
		$d(u_3)\le 6$, $d(u_4)\le 6$, and $\min\{d(u_1), d(u_2)\}\le 7$;
	\item[{\rm ($A_{8.2}$)}]
		$d(u_3)\le 6$, $u_4$ is $(7, 4^-)$, and $\min\{d(u_1), d(u_2)\}\le 6$;
	\item[{\rm ($A_{8.3}$)}]
		$d(u_3)\le 6$, $u_4$ is $(7, 4^-)$, and $d(u_i)\le 7$ for each $i\in \{1, 2\}$;
	\item[{\rm ($A_{8.4}$)}]
		$u_i$ is $(7, 4^-)$ for each $i\in \{3, 4\}$ and $d(u_1)\le 6$, $d(u_2)\le 7$.
     \end{description}
\end{description}
\end{theorem}
The rest of the section is devoted to the proof of Theorem~\ref{thm02}, by contradiction built on discharging methods.
To this purpose, we assume to the contrary that $G$ contains none of the local structures ($A_1$)--($A_8$).

\subsection{Definitions and notations}
Since $G$ is $2$-connected, $d(v)\ge 2$ for every vertex $v \in V(G)$.
Let $G'$ be the graph obtained by deleting all the $2$-vertices of $G$, and let $H$ be a connected component of $G'$.
One sees that every vertex $v \in V(H)$ has its (called {\em original}) degree $d(v) \ge 3$ in $G$, and $H$ is clearly also planar.
In what follows, we assume that $H$ is embedded in the two dimensional plane such that its edges intersect only at their ending vertices,
and we refer $H$ to as a {\em plane} graph.

For objects in the plane graph $H$, we use the following notations:
\begin{itemize}
\parskip=0pt
\item
	$N_H(v) = \{u \mid u v\in E(H)\}$ and $d_H(v) = |N_H(v)|$ --- the degree of the vertex $v\in V(H)$;
\item
	similarly define what a $k$-{\em vertex}, a $k^+$-{\em vertex}, and a $k^-$-{\em vertex} are;
\item
	$n'_k(v)$ ($n'_{k^+}(v)$, $n'_{k^-}(v)$, respectively) --- the number of $k$- ($k^+$-, $k^{-}$-, respectively) neighbors of $v$;
\item
	an $(S_1, S_2)_H$ vertex --- a vertex $x$ with $d_H(x) \in S_1$ and $n'_{6^{+}}(x)\in S_2$;
\item
	an $(S_1, S_2)_{HG}$ vertex --- a vertex $x$ with $d_H(x) = d(x) \in S_1$ and $n'_{6^{+}}(x) = n_{6^{+}}(x) \in S_2$.
\end{itemize}

Note that we do not discuss the embedding of the graph $G$.
Every face in the proof of Theorem~\ref{thm02} is in the embedding of the graph $H$, or is simplified as in the embedding of the graph $H$.
We reserve the character $f$ to denote a face in the embedding.

\begin{itemize}
\parskip=0pt
\item
	$F(H)$ --- the face set of $H$;
\item
	$V_k(H)$ --- the set of $k$-vertices in $V(H)$;
\item
	$V(f)$ --- the set of vertices on (the boundary of) the face $f$;
\item
	a vertex $v$ is {\em incident with} a face $f$ if $v\in V(f)$;
\item
	an edge $u v$ is {\em incident with} a face $f$ if $u, v\in V(f)$;
\item
	$n_k(f)$ ($n_{k^+}(f)$, $n_{k^-}(f)$, respectively) --- the number of $k$- ($k^+$-, $k^{-}$-, respectively) vertices in $V(f)$;
\item
	$\delta(f)$ --- the minimum degree of the vertices in $V(f)$;
\item
	$d(f)$ --- the degree of the face $f$, which is the number of edges on the face $f$ with cut-edges counted twice,
	and similarly define what a $k$-{\em face}, a $k^+$-{\em face} and a $k^-$-{\em face} are;
\item
	$F(v) = \{f \in F(H) \mid v\in V(f)\}$ --- the set of faces incident with $v$;
\item
	$F_3(v) = \{f \in F(H) \mid v\in V(f), d(f)= 3\}$ --- the set of $3$-faces incident with $v$;
\item
	$F(u v) = \{f \in F(H) \mid uv\in E(G), u, v\in V(f)\}$ --- the set of faces incident with the edge $uv$;
	note that $|F(u v)| = 2$ for any edge $uv$;
\item
	$f_{u v} = F(u v)\setminus \{f\}$ --- the face incident with the edge $uv$ and adjacent to the face $f\in F(u v)$;
\item
	$m_k(v)$ ($m_{k^+}(v)$, $m_{k^-}(v)$, respectively) --- the number of $k$- ($k^+$-, $k^-$-, respectively) faces in $F(v)$;
\item
	$m_3(u v)$--- the number of $3$-faces in $F(u v)$.
\end{itemize}

\subsection{Structural properties}
Note from the construction of $H$ that $d_H(u) = d(u)- n_2(u)$ and thus $d_H(u) = d(u)$ if $n_2(u) = 0$.
In this subsection, we characterize structural properties for $5^-$-vertices in $H$, which are used in the discharging.

\begin{lemma}
\label{lemma01}
For any $u \in V(H)$,
\begin{description}
\parskip=0pt
\item[{\rm (1)}]
	if $3 \le d(u) \le 7$, then $d_H(u) = d(u)$;
\item[{\rm (2)}]
	if $d(u) \ge 8$, then $d_H(u) \ge 7$;
\item[{\rm (3)}]
	$\delta(H)\ge 3$, $n'_2(u) = 0$, $n'_{k}(u) = n_{k}(u)$, $k\in [3, 6]$, and $n'_{5^-}(u) = n_{5^-}(u) - n_2(u)$, $n'_{6^+}(u) = n_{6^+}(u)$;
\item[{\rm (4)}]
	if $d(u) \ge 8$, then $n'_{5^-}(u)\le d_H(u) - 7$ and $n'_{6^+}(u)\ge 7$.
\end{description}
\end{lemma}
\begin{proof}
The non-existence of ($A_{1.1}$) in $G$ states that $u w\in E(G)$ for each $2$-vertex $v$ with $N(v) = \{u, w\}$.

The non-existence of ($A_{1.2}$) in $G$ states that no $7^-$-vertex is adjacent to any $2$-vertex in $G$;
thus $3 \le d(u) \le 7$ implies $d_H(u) = d(u)$.

The non-existences of $(A_{1.3})$ in $G$ implies that every $8^+$-vertex $u$ is adjacent to at most ($d(u) - 7$) $2$-vertices in $G$;
thus if $d(u) \ge 8$, then $d_H(u) = d(u)- n_2(u) \ge d(u) - (d(u) - 7) = 7$.

Note that all $2$-vertices are deleted from $G$, a $k$-neighbor of $u$ in $G$ remains as a $k$-neighbor of $u$ in $H$, $k\in [3, 6]$,
and a $7^+$-neighbor of $u$ in $G$ remains as a $7^+$-neighbor of $u$ in $H$,
Therefore, $\delta(H)\ge 3$, $n'_2(u) = 0$, $n'_{k}(u) = n_{k}(u)$, for any $k\in [3, 6]$, and $n'_{5^-}(u) = n_{5^-}(u) - n_2(u)$, $n'_{6^+}(u) = n_{6^+}(u)$.

The non-existences of $(A_{1.3})$ in $G$ implies that every $8^+$-vertex $u$ is adjacent to at most ($d(u) - 6$) $5^-$-vertices in $G$;
thus if $d(u) \ge 8$, then $n'_{5^-}(u) = n_{5^-}(u) - n_2(u)\le d(u) - 7 - (d(u) - d_H(u)) = d_H(u) - 7$, and $n'_{6^+}(u) \ge 7$.
\end{proof}

By Lemma~\ref{lemma01}, for any $k$-vertex, where $k\in [3, 6]$, we do not distinguish which graph, $G$ or $H$, it is in;
furthermore, we use $d(u)$ instead of $d_H(u)$, use $n_k(u)$ instead of $n'_k(u)$, and use $n_{6^+}(u)$ instead of $n'_{6^+}(u)$.

\begin{corollary}
\label{coro01}
For any vertex $u \in V(H)$, if $d(u)> d_H(u)\ge 7$, then $n'_{5^-}(u)\le d_H(u) - 7$, $n_{6^+}(u) = n'_{6^+}(u) \ge 7$;
and if $n'_{5^-}(u) = n_{5^-}(u) > 0$, then $d_H(u)\ge 7$.
\end{corollary}

\begin{lemma}
\label{lemma02}
For any $uv \in E(H)$,
\begin{description}
\parskip=0pt
\item[{\rm (1)}]
	if $d(v) = 3$ with $N(v) = \{u, v_1, v_2\}$, then $d(u)\ge 8$ and $d_H(u)\ge 8$;
\item[{\rm (2)}]
	if $d(u) = 4$ with $N(u) = \{v, u_1, u_2, u_3\}$, then for each $i = 1, 2, 3$, we have $d_H(u_i)\ge 8$ and
	\begin{description}
	\parskip=0pt
	\item[{\rm (2.1)}]
		$d(u_i)\ge 10$ if $d(v) = 4$;
	\item[{\rm (2.2)}]
		$d(u_i)\ge 9$ if $d(v) = 5$;
	\item[{\rm (2.3)}]
		$d(u_i)\ge 8$ if $d(v) = 6$ and $u v$ is contained in a $3$-cycle;
	\item[{\rm (2.4)}]
		$d(u_i)\ge 10$ if $d(v) = 5$ and $u v$ is contained in two $3$-cycles;
	\item[{\rm (2.5)}]
		$d(u_i)\ge 9$ if $d(v) = 6$ and $u v$ is contained in two $3$-cycles;
	\item[{\rm (2.6)}]
		$d(u_i)\ge 9$ if $v$ is $(7, 4^-)$ and $u v$ is contained in two $3$-cycles $u u_2 v$, $u u_3 v$ with $u_1 u_2, u_1 u_3\in E(G)$.
	\end{description}
\end{description}
\end{lemma}
\begin{proof}
The non-existence of ($A_{2.1}$) in $G$ states that no $7^-$-vertex is adjacent to a $3$-vertex in $G$;
thus $d(u)\ge 8$ and by Corollary~\ref{coro01}, $d_H(u)\ge 8$.
The non-existence of ($A_{1.1}$), ($A_{1.2}$), ($A_{2.1}$) and ($A_{6}$) in $G$ and Corollary~\ref{coro01} imply that (2.1)-(2.6) hold.
\end{proof}

\begin{lemma}
\label{lemma03}
Let $u$, $v$, $u_1$, $v_1$ be four $4^-$-vertices in $H$.
\begin{description}
\parskip=0pt
\item[{\rm (1)}]
  If $d(u) = 3$ and $d(v)\le 4$, or $d(u) = d(v) = 4$ with $u v\not\in E(H)$,
  then $F_3(u)\cap F_3(v) = \emptyset$.
\item[{\rm (2)}]
  If $d(u) = d(v) = 4$ with $u v\in E(H)$ and $d(u_1)\le 4$,
  then $(F_3(u)\cup F_3(v))\cap F_3(u_1) = \emptyset$.
\item[{\rm (3)}]
  If $d(u) = d(v) = 4$ with $u v\in E(H)$ and $d(u_1) = d(v_1) = 4$ with $u_1 v_1\in E(H)$,
  then $(F_3(u)\cup F_3(v))\cap (F_3(u_1)\cup F_3(v_1)) = \emptyset$. \hfill $\Box$
\end{description}
\end{lemma}

\subsection{Discharging to show contradictions}
To derive a contradiction, we make use of the discharging methods, which are very similar to an amortized analysis through token re-distribution.
First, by Euler's formula on the embedding $|V(H)| - |E(H)| + |F(H)|= 2$ and the relation $\sum_{v\in V(H)}d_H(v) = \sum_{f\in F(H)}d(f) = 2|E(H)|$,
we have the following equality:
\begin{equation}\label{eq01}
\sum_{u\in V(H)}(4 - d_H(u)) + \sum_{f\in F(H)}(4 - d(f)) = 8.
\end{equation}
Next, we define an initial {\em weight} (i.e., {\em token}, or {\em charge}) function $\omega(\cdot)$ on $V(H)\cup F(H)$ by
setting $\omega(u) = 4 - d_H(u)$ for each vertex $u\in V(H)$ and
setting $\omega(f) = 4 - d(f)$ for each face $f\in F(H)$.
It follows from Eq.~(\ref{eq01}) that the total weight is equal to $8$.
We define next a set of discharging rules (R1)--(R5) to move portions of weights from $x$ to $z$ (possibly through $y$), that is,
the function $\tau(x \rightarrow z)$ (or $\tau(x \rightarrow z)_y$), where $x, y, z\in V(H)\cup F(H)$.
The effect of this function decreases the charge at $x$ by $\tau(x \rightarrow z)_y$, leaves the charge of $y$ unchanged, and increases the charge of $z$ by $\tau(x \rightarrow z)_y$.
For counting purposes, we say that $x$ sends the charge to $y$, and $z$ receives the charge from $y$.
At the end of the discharging, a new weight function $\omega'(\cdot)$ is achieved and we will show that $\omega'(x) \le 0$ for every element $x \in V(H)\cup F(H)$.
This contradicts the positive total weight.

Before presenting the discharging rules, we introduce an important parameter $\alpha_u$ for a vertex $u\in V(H)$ satisfying $d_H(u) = k\ge 6$ and $n'_{5^-}(u)> 0$:
\begin{equation}
\label{eq02}
\alpha_u = \left\{
\begin{array}{ll}
\frac{k - 4 - \frac 13(k - 2n'_{5^-}(u))}{2n'_{5^-}(u)} 	= \frac{k - 6}{3n'_{5^-}(u)} + \frac 13,	&\mbox{ if }	n'_{5^-}(u)\le \lfloor\frac k2\rfloor;\\
\frac{k - 4}k = 1 - \frac 4k, 						&\mbox{ otherwise.}
\end{array}\right.
\end{equation}
The following properties hold for $\alpha_u$, of which some values are given in Figure~\ref{fig02}.
\begin{lemma}
\label{lemma04properties}
Let $d_H(u)= k\ge 6$ with $n'_{5^-}(u)> 0$.
\begin{description}
\parskip=0pt
\item[{\bf\rm (1)}]
	$\alpha_u \ge 1 - \frac 4k$.
\item[{\bf\rm (2)}]
	If $d_H(u) = 6$, then $\alpha_u = \frac 13$;
	if $d_H(u)\ge 12$, then $\alpha_u \ge \frac 23$.
\item[{\bf\rm (3)}]
	If $u$ is $(k, 7^+)_H$, then {\rm (i)} if $n'_{5^-}(u)\le \lfloor\frac k2\rfloor$, then $\alpha_u \ge \frac 23 + \frac 1{3n'_{5^-}(u)}> \frac 23$;
	{\rm (ii)} $\alpha_u \ge \frac 57> \frac 23$.
\item[{\bf\rm (4)}]
	If $d(u) > d_H(u)\ge 8$, then $\alpha_u \ge \frac 57> \frac 23$.
\item[{\bf\rm (5)}]
	If $u$ is $(k, 6^+)_H$, then $\alpha_u \ge \frac 23$.
\item[{\bf\rm (6)}]
	If $\alpha_u < \frac 23$, then $d(u) = d_H(u)= k\in [7, 11]$, $u$ is $(k, 5^-)$, and if $n'_{3}(u)> 0$, then $k\in [8, 11]$.
\end{description}
\end{lemma}
\begin{proof}
By the definition of $\alpha_u$ in Eq.~(\ref{eq02}), to prove (1) it suffices to assume that $n'_{5^-}(u)\le \lfloor\frac k2\rfloor$.
In this case, we have $\alpha_u - (1 - \frac 4k) = \frac{k - 6}{3n'_{5^-}(u)} - \frac 23 + \frac 4k \ge \frac{k - 6}{3\times \frac{k}{2}} - \frac 23 + \frac 4k = 0$,
i.e., $\alpha_u \ge 1 - \frac 4k$.

For (2), if $d_H(u) = 6$, then by the definition $\alpha_u = \frac 13$;
if $d_H(u) \ge 12 $, then by (1) $\alpha_u \ge \frac 23$.

To prove (3), if $d_H(u)\ge 14$, then by (1), $\alpha_u \ge 1 - \frac 4k \ge \frac 57$.
If $d_H(u)\le 14$, then $n'_{5^-}(u) \le k - 7 \le \lfloor\frac k2\rfloor\le 7$;
thus, $\alpha_u = \frac{k - 6}{3n'_{5^-}(u)} + \frac 13 \ge \frac 1{3n'_{5^-}(u)} + \frac 23 \ge \frac 57$.

By Corollary~\ref{coro01}, if $d(u) > d_H(u)$, then we have $n'_{5^-}(u)\le d_H(u) - 7$ and thus $u$ is $(k, 7^+)_H$.
It follows from (3) that $\alpha_u \ge \frac 57> \frac 23$.
This proves (4).

To prove (5), if $d_H(u)\ge 12$, then by (2) $\alpha_u \ge \frac 23$.
If $d_H(u)\le 11$, then $n'_{5^-}(u)\le k - 6 \le \lfloor\frac k2\rfloor$;
thus, $\alpha_u  =\frac{k - 6}{3n'_{5^-}(u)} + \frac 13
\ge \frac{k - 6}{3(k - 6)} + \frac 13
= \frac{2}{3}. $

Lastly, if $\alpha_u < \frac 23$, then by (2) and (5) $d(u) = d_H(u) = k\in [7, 11]$ and $u$ is $(k, 5^-)$.
If $n'_{3}(u)> 0$, then by the non-existence of ($A_{2.1}$) in $G$, we have $k\in [8, 11]$.
This proves (6).
\end{proof}

\begin{figure}[]
\begin{center}
\includegraphics[width=0.8\textwidth]{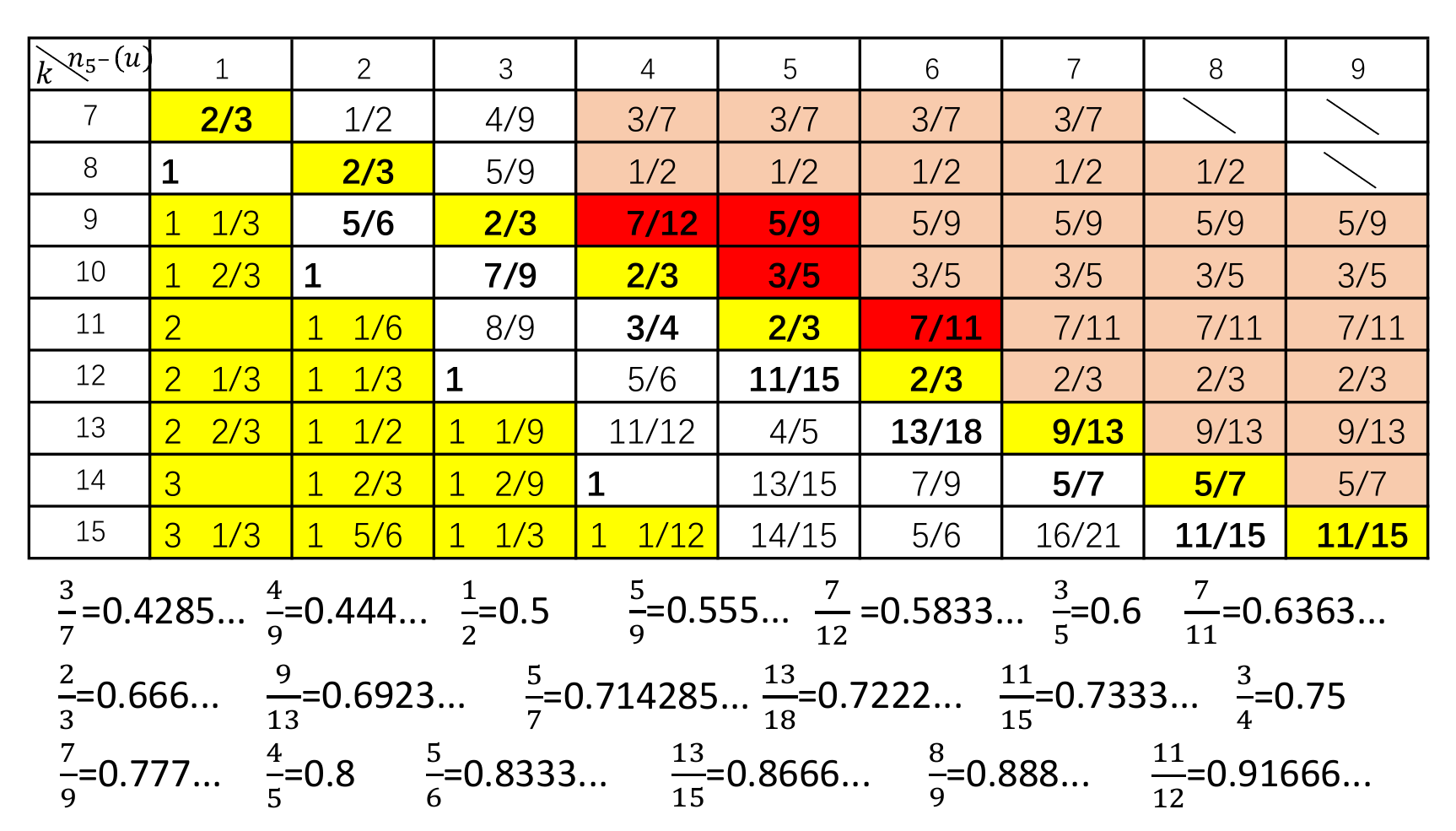}
\caption{The $\alpha_u$ value for a vertex $u$ with $d_H(u) = k\in [7, 15]$ and $n'_{5^-}(u)\le 9$.\label{fig02}}
\end{center}
\end{figure}

We now present the discharging rules (R1)--(R5).
A graphical view of how weights are moved from an element to another is depicted in Figure~\ref{fig03}.

\begin{description}
\parskip=0pt
\item[{\rm (R1)}]
	For each $4^-$-vertex $u$ adjacent to a $6^+$-vertex $v$ with $F(u v) = \{f_1, f_2\}$,
	if $d(f_1)\ge 4$ and $d(f_2)\ge 4$, then $\tau(u \rightarrow v)_{f_1} = \tau(u \rightarrow v)_{f_2} = \frac{\alpha_v}2$.
\end{description}
In the other four rules (R2)--(R5), $f= [uvw]$ is a $3$-face.

\begin{description}
\parskip=0pt
\item[{\rm (R2)}]
	If $\delta(f)\ge 6$, then for any $x\in V(f)$, $\tau(f \rightarrow x) = \frac 13$.
\end{description}
In the rest three rules (R3)--(R5), $\delta(f) \le 5$ and $d(u)\le 5$.

\begin{description}
\parskip=0pt
\item[{\rm (R3)}]
	For any $x\in \{v, w\}$ with $d_H(x)\ge 6$,
	\begin{description}
	\parskip=0pt
	\item[{\rm (R3.1)}]
		$\tau(f \rightarrow x)_f = \alpha_x$;
	\item[{\rm (R3.2)}]
		$\tau(f \rightarrow x)_{f_{ux}} = \begin{cases}
			\frac {\alpha_x}2,	&\mbox{ if } d(f_{ux})\ge 4;\\
			0,							&\mbox{ if } d(f_{ux}) = 3.
			\end{cases}$
	\end{description}
\item[{\rm (R4)}]
	If $d(u) = 4$, $d(v) = 5$ and $m_3(u) = 4$ with $f_{u v} = [u v w^*]$, then
	\begin{description}
	\parskip=0pt
	\item[{\rm (R4.1)}]
		$\tau(F(u v) \rightarrow v) = \frac 1{10}$ if $\min\{d(f_{v w}), d((f_{u v})_{v w^*})\} = 3$ and $\max\{d(f_{v w}), d((f_{u v})_{v w^*})\}\ge 4$;
	\item[{\rm (R4.2)}]
		$\tau(F(u v) \rightarrow v) = \frac 25$ if $d(f_{v w}) = d((f_{u v})_{v w^*}) = 3$.
	\end{description}
\item[{\rm (R5)}]
	If $\delta(f) = d(u) = 5$ with $n_5(f) = k$, then
	\begin{description}
	\parskip=0pt
	\item[{\rm (R5.1)}]
		if after applying (R3), the face $f$ still has charge $\beta \ge 0$, then $\tau(f \rightarrow u) = \frac{\beta}{k}$;
	\item[{\rm (R5.2)}]
		if $d_H(v)\ge 6$, $d_H(w)\ge 6$, and $\alpha_v + \alpha_w\ge 1$, then $\tau(u \rightarrow f) = \alpha_v + \alpha_w - 1$.
	\end{description}
\end{description}

\begin{figure}[h]
\begin{center}
\includegraphics[width=0.7\textwidth]{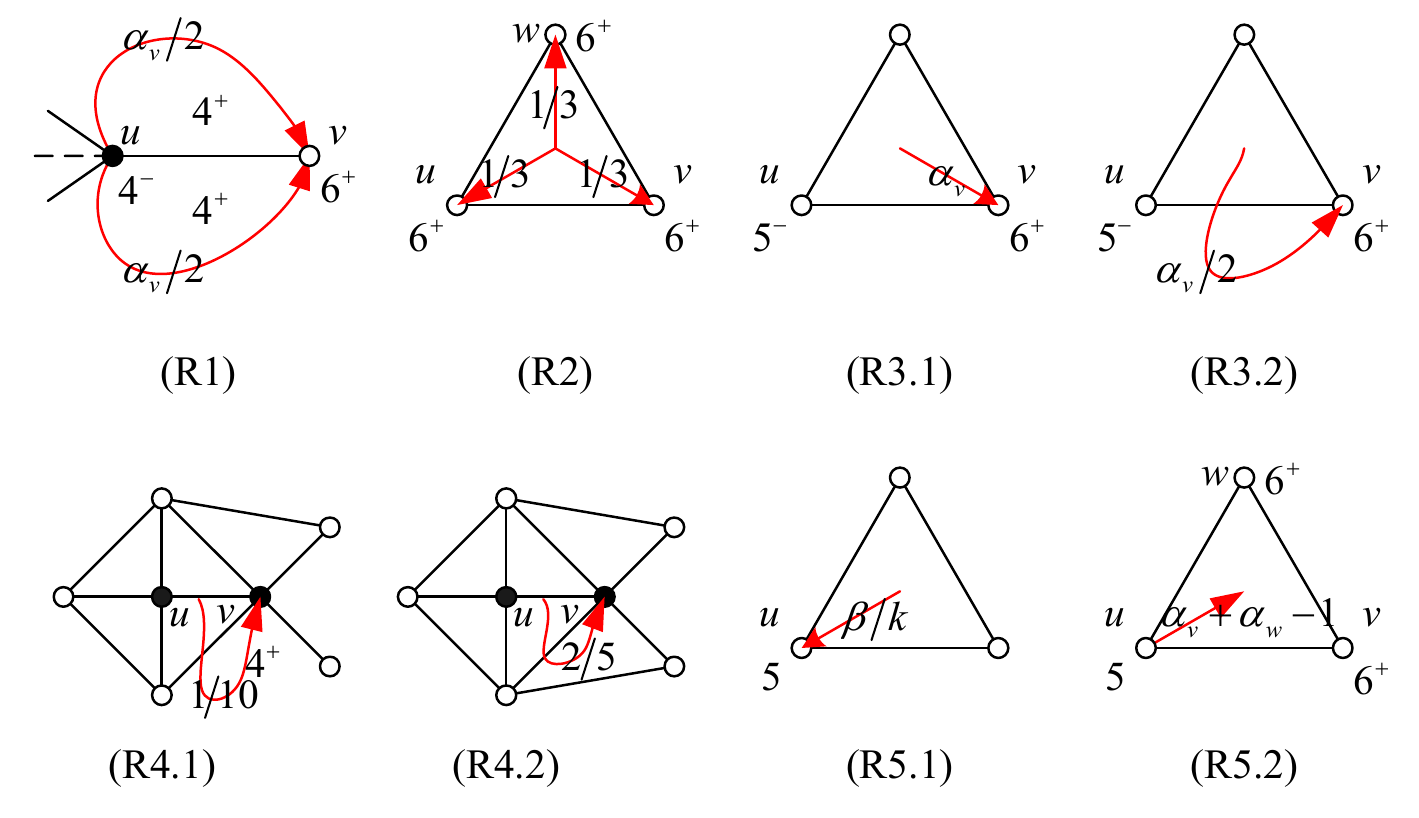}
\caption{The discharging rules (R1)--(R5): A graphical view of how weights are moved from an element to another by red arrows,
	each is labeled with the weight.\label{fig03}}
\end{center}
\end{figure}

In the rest of this section we prove that at the end of discharging the new weight function $\omega'(\cdot) \le 0$, through ten Discharging-Lemmas.
The first of which is obvious.

\begin{demma}
\label{demma01}
For every face $f$, if $d(f)\ge 4$, then $f$ does not receive or give any weight and thus $\omega'(f) = \omega(f) = 4 - d(f)\le 0$.
\end{demma}
\begin{proof}
This holds trivially since the discharging rules (R1)--(R5) do not apply to $f$ with $d(f)\ge 4$.
\end{proof}

\begin{demma}
\label{demma02}
For every vertex $v\in V(H)$ with $d_H(v) = k \ge 6$, $\omega'(v)\le 0$.
\end{demma}
\begin{proof}
Let $f = [\ldots u v w\ldots]$ be a face incident to $v$.
If $d_H(u)\ge 6$ and $d_H(w)\ge 6$, then by (R2), $\tau(f \rightarrow v)\le \frac 13$.
If $d_H(u)\le 5$ and $d(f) = 3$, then by (R3.1), $\tau(f \rightarrow v)_f = \alpha_v$.
If $d_H(u)\le 5$ and $d(f)\geq 4$, then by (R1), (R3.1) and (R3.2), $\tau(u \rightarrow v)_f + \tau(f_{v u} \rightarrow v)_f = \frac {\alpha_v}2$,
and $\tau(w \rightarrow v)_f + \tau(f_{vw} \rightarrow v)_f \le \frac {\alpha_v}2$.

Let $k_1$ be the number of faces $f = [\ldots u v w\ldots]$ incident to $v$ such that $d_H(u)\ge 6$ and $d_H(w)\ge 6$.
The total weight transferred to $v$ is at most $k_1 \times \frac{1}{3} + (d_H(v) - k_1)\times \alpha_v$.
It follows that
$\omega'(v) \le \omega(v) + k_1 \times \frac 13 + (d_H(v) - k_1)\times \alpha_v
= 4 - d_H(v) + k_1 \times \frac 13 + (d_H(v) - k_1)\times \alpha_v
= 4 - k + k\times \alpha_v - (\alpha_v - \frac 13)\times k_1$.

If $n'_{5^-}(u)\ge \lceil\frac k2\rceil$,
then $\omega'(v)\le 4 - k + k\times \alpha_v 
= 4 - k + k\times \frac {k - 4}k = 0$.
If $n'_{5^-}(u)\le \lfloor\frac k2\rfloor$, then it follows from $k_1\ge k - 2n'_{5^-}(u)$ that
$\omega'(v)\le 4 - k + k\times \alpha_v - (\alpha_v - \frac 13)\times (k - 2n'_{5^-}(u))
= 4 - \frac 23\times k + \frac 23\times (k - 6)
= 0$.
\end{proof}

\begin{demma}
\label{demma03}
For every vertex $v\in V(H)$ with $d(v) = 3$, $\sum_{x\in \{v\}\cup F_3(v)}\omega'(x)\le 0$.
\end{demma}
\begin{proof}
Let $u_i v\in E(H)$, and $f_i = [\ldots u_i v u_{(i\bmod3) + 1} \ldots]$, $i\in \{1, 2, 3\}$, be the three faces incident to $v$.
One sees from Lemma~\ref{lemma02}(1) that $d(u_i)\ge 8$,
$d_H(u_i)\ge 8$ and $\alpha_{u_i}\ge \frac 12$ for each $i = 1, 2, 3$.
We distinguish the following four cases for $m_3(v)$.

{\bf Case 1.} $m_3(v) = 0$.
In this case, by (R1), $\tau(v \rightarrow u_1)_{f_3} + \tau(v \rightarrow u_1)_{f_1} = \alpha_{u_1}$,
$\tau(v \rightarrow u_2)_{f_1} + \tau(v \rightarrow u_2)_{f_2} = \alpha_{u_2}$, and
$\tau(v \rightarrow u_3)_{f_2} + \tau(v \rightarrow u_3)_{f_3} = \alpha_{u_3}$.
Thus,
\[
\sum_{x\in \{v\}\cup F_3(v)}\omega'(x)
= \omega'(v) = \omega(v) - \sum_{i\in [1, 3]}\tau(v \rightarrow u_i)
= 1- \sum_{i\in [1, 3]}\alpha_{u_i} \le 1- 3\times \frac 12 < 0.
\]

{\bf Case 2.} $m_3(v) = 1$ with $d(f_2) = 3$.
In this case, by (R1), (R3.1) and (R3.2), $\tau(v \rightarrow u_1)_{f_3} + \tau(v \rightarrow u_1)_{f_1} = \alpha_{u_1}$,
$\tau(f_2 \rightarrow u_2)_{f_1} = \frac {\alpha_{u_2}}2$,
$\tau(f_2 \rightarrow u_3)_{f_3} = \frac {\alpha_{u_3}}2$, and 
for each $i\in [2, 3]$, $\tau(f_2 \rightarrow u_i)_{f_2} = \alpha_{u_i}$.
Thus,
\begin{equation*}
\begin{aligned}
\sum_{x\in \{v\}\cup F_3(v)}\omega'(x)
& = \sum_{x\in \{v, f_2\}}\omega(x)- \left(\tau(v \rightarrow u_1)
	+ \sum_{i\in [2, 3]}\tau(f_2 \rightarrow u_i) + \tau(f_2 \rightarrow u_2)_{f_1} +  \tau(f_2 \rightarrow u_3)_{f_3}\right)\\
& = \sum_{x\in \{v, f_2\}}(4 - 3) - (\alpha_{u_1} +  \alpha_{u_2} + \alpha_{u_3} + \frac 12 \alpha_{u_2} + \frac 12 \alpha_{u_3})
	\le 2\times 1 - 4\times \frac 12 = 0.
\end{aligned}
\end{equation*}

{\bf Case 3.} $m_3(v) = 2$ with $d(f_2) = d(f_3) = 3$.
In this case, by (R3.1) and (R3.2), $\tau(f_3 \rightarrow u_1)_{f_3} = \alpha_{u_1}$,
$\tau(f_3 \rightarrow u_1)_{f_1} = \frac {\alpha_{u_1}}2$,
$\tau(f_2 \rightarrow u_2)_{f_2} = \alpha_{u_2}$,
$\tau(f_2 \rightarrow u_2)_{f_1} = \frac {\alpha_{u_2}}2$,
and $\tau(f_2 \rightarrow u_3)_{f_2} = \tau(f_3 \rightarrow u_3)_{f_3} = \alpha_{u_3}$.
Thus,
\begin{equation}
\label{eq03}
\begin{aligned}
\sum_{x\in \{v\}\cup F_3(v)}\omega'(x)
& = \sum_{x\in \{v, f_2, f_3\}}\omega(x) - \left(\sum_{i\in [2, 3]}\tau(f_2 \rightarrow u_i)
	+ \sum_{i\in \{1, 3\}}\tau(f_3 \rightarrow u_i) + \tau(f_3 \rightarrow u_1)_{f_1} + \tau(f_2 \rightarrow u_2)_{f_1}\right)\\
& = \sum_{x\in \{v, f_2, f_3\}}(4 - 3)- (\alpha_{u_2} + \alpha_{u_3} + \alpha_{u_1} + \alpha_{u_3} + \frac 12 \alpha_{u_1} + \frac 12 \alpha_{u_2})\\
& = 3\times 1 - (2\alpha_{u_3} + \frac 32 \alpha_{u_1} + \frac 32 \alpha_{u_2}).
\end{aligned}
\end{equation}

Below we distinguish two cases where $u_1 u_2\in E(G)$ and $u_1 u_2\not\in E(G)$, respectively.

{Case 3.1.} $u_1u_2\in E(G)$.
In this case, for each $i\in [1, 3]$, the non-existence of ($A_{2.1}$) and ($A_{2.4}$) in $G$ states that $d(u_i)\ge 9$;
by Corollary~\ref{coro01}, no $7$-vertex is adjacent to any $5^-$-vertex in $H$ and $d_H(u_i)\ge 8$.
Thus, we have $\alpha_{u_i}\ge \frac 59$ for each $i\in [1, 3]$.
Note that if $d_H(u_i)= 8$ for some $i\in [1, 3]$, then $d(u_i)\ge 9 > d_H(u_i)$, and it follows from Lemma~\ref{lemma04properties} that $\alpha_{u_i}\ge \frac 57 > \frac23$.

When $\alpha_{u_3}\ge \frac 23$, Eq.~(\ref{eq03}) becomes
$\sum_{x\in \{v\}\cup F_3(v)}\omega'(x) 
 \le 3 - (2\times \frac 23 + \frac 32 \times \frac 59 + \frac 32\times \frac 59) = 0$.

When $\alpha_{u_1}\ge \frac{7}{11}$ and $\alpha_{u_2}\ge \frac 7{11}$, Eq.~(\ref{eq03}) becomes
$\sum_{x\in \{v\}\cup F_3(v)}\omega'(x)
 \le 3 - (2\times \frac 59 + \frac 32 \times \frac 7{11} + \frac 32 \times \frac 7{11}) 
= - \frac 19 + \frac 1{11} < 0$.

When $\alpha_{u_3}< \frac 23$ and $\alpha_{u_i}< \frac 7{11}$ for some $i\in [1, 2]$,
it follows from Lemma~\ref{lemma04properties} that $d_H(u_3) = d(u_3)\in [9, 11]$, $u_3$ is $(d(u_3), 5^-)$,
and there exists an $i\in [1, 2]$ such that $d_H(u_i) = d(u_i)\in [9, 10]$, $u_i$ is $(d(u_i), 5^-)$.
It follows that there exists a $w\in N(v)$ such that $w$ is $(10, 5^-)$.

Assume that $w_1$ is $(9, 5^-)$.
The non-existence of ($A_{3.3}$) in $G$ implies that $w_1$ is $(9, 5)$;
the non-existence of ($A_{3.6}$) in $G$ implies that there exists a $w_2\in N(v)\setminus \{w_1\}$ such that $w_2$ is $([9, 11], 5)$.
Let $w_3\in N(v)\setminus \{w_1, w_2\}$.
Then by the non-existence of ($A_{3.2}$) in $G$, we have $d_H(w_3)= d(w_3) = \Delta$ and $w_3$ is $(\Delta, 7^+)_{HG}$.
It is clear that $w_3\in \{u_1, u_2\}$, i.e., $w_3 \ne u_3$.
Therefore, $\alpha_{w_1}\ge \frac 7{12}$, $\alpha_{w_2}\ge \frac 7{12}$, and by Lemma~\ref{lemma04properties} $\alpha_{w_3}\ge \frac 57$.
Hence, Eq.~(\ref{eq03}) becomes
$\sum_{x\in \{v\}\cup F_3(v)}\omega'(x)
 \le 3 - (2\times \frac 7{12} + \frac 32\times \frac 7{12} + \frac 32\times \frac 57)
= - \frac 1{24} - \frac 1{14} < 0$.

Assume that there is no $(9, 5^-)$ in $\{u_1, u_2, u_3\}$ and $w_1$ is $(10, 5^-)$.
The non-existence of ($A_{4.1}$) in $G$ implies that there exists a $w_2\in N(v)\setminus \{w_1\}$ such that $w_2$ is $(10/11, 5)$.
Let $w_3\in N(v)\setminus \{w_1, w_2\}$.
Thus, by the non-existence of ($A_{4.3}$) in $G$, we have $d_H(w_3) = d(w_3) = \Delta$ and by the non-existence of ($A_{4.2}$) in $G$, $w_3$ is $(\Delta, 7^+)_{H G}$.
It is clear that $w_3\in \{u_1, u_2\}$, i.e., $w_3 \ne u_3$.
Therefore, $\alpha_{w_1}\ge \frac 35$, $\alpha_{w_2}\ge \frac 35$, and by Lemma~\ref{lemma04properties} $\alpha_{w_3}\ge \frac 57$.
Hence, Eq.~(\ref{eq03}) becomes
$\sum_{x\in \{v\}\cup F_3(v)}\omega'(x) 
 \le 3 - (2\times \frac 35 + \frac 32\times \frac 35 + \frac 32\times \frac 57)
= - \frac 1{10} - \frac 1{14} < 0$.

{Case 3.2.} $u_1 u_2\not\in E(G)$.
In this case, $d_H(u_3)\ge 12$, or $d(u_3) > d_H(u_3)$,
or $d_H(u_3) = d(u_3)\in [8, 11]$ and by the non-existence of ($A_{2.2}$) in $G$ $u_3$ is $(d(u_3), 6^+)_{HG}$.
Then by Lemma~\ref{lemma04properties}(2), (5) and (6), $\alpha_{u_3}\ge \frac 23$.

When $\alpha_{u_1}\ge \frac 59$ and $\alpha_{u_2}\ge \frac 59$, Eq.~(\ref{eq03}) becomes
$\sum_{x\in \{v\}\cup F_3(v)}\omega'(x) 
 \le 3 - (2\times \frac 23 + \frac 32 \times \frac 59 + \frac 32\times \frac 59) = 0$.

When $\alpha_{u_i}< \frac 59$ for some $i\in [1, 2]$,
it follows from Lemma~\ref{lemma04properties}(2) and (4) that there exists a $u_i\in \{u_1, u_2\}$ such that $d(u_i) = d_H(u_i) = 8$ and $u_i$ is $(8, 4^-)_{H G}$.
Assuming w.l.o.g. $u_1$ is such a vertex.
From the non-existence of ($A_{2.2}$) in $G$, $d_H(u_2)\ge 12$, or $d(u_2) > d_H(u_2)$, or $d_H(u_2) = d(u_2)\in [8, 11]$ and $u_2$ is $(d(u_2), 6^+)_{HG}$.
We have $\alpha_{u_2}\ge \frac 23$ by Lemma~\ref{lemma04properties}(2), (5) and (6).
Thus, Eq.~(\ref{eq03}) becomes
$\sum_{x\in \{v\}\cup F_3(v)}\omega'(x) 
 \le 3 - (2\times \frac 23 + \frac 32\times \frac 12 + \frac 32 \times \frac 23) = -\frac 1{12} < 0$.

{\bf Case 4.} $m_3(v) = 3$.
In this case, by a similar discussion as in Case 3.1 where $u_1 u_2\in E(G)$,
it suffices to assume that for each $i\in [1, 3]$, $\alpha_{u_i}\ge \frac 59$, and either $d(u_i)\ge 9$ or $u_i$ is $(8, 7^+)$.
It follows from (R3.1) that $\tau(f_3 \rightarrow u_1)_{f_3} = \tau(f_1 \rightarrow u_1)_{f_1} = \alpha_{u_1}$,
$\tau(f_1 \rightarrow u_2)_{f_1} = \tau(f_2 \rightarrow u_2)_{f_2} = \alpha_{u_2}$,
$\tau(f_2 \rightarrow u_3)_{f_2} = \tau(f_3 \rightarrow u_3)_{f_3} = \alpha_{u_3}$,
and thus,
\begin{equation}
\label{eq04}
\begin{aligned}
\sum_{x\in \{v\}\cup F_3(v)}\omega'(x)
& = \sum_{x\in \{v, f_1, f_2, f_3\}}\omega(x)- \left(\sum_{i\in \{1, 3\}}\tau(f_i \rightarrow u_1) + \sum_{i\in [1, 2]}\tau(f_i \rightarrow u_2)
	+ \sum_{i\in [2, 3]}\tau(f_i \rightarrow u_3)\right)\\
& = 4 - 2 (\alpha_{u_1} + \alpha_{u_2} + \alpha_{u_3}).
\end{aligned}
\end{equation}
If $\alpha_{u_i}\ge \frac 23$ for each $i\in [1, 3]$, then Eq.~(\ref{eq04}) clearly is $\le 0$;
otherwise, there exists an $i_0\in [1, 3]$ such that $\alpha_{u_{i_0}} < \frac 23$, and then by Lemma~\ref{lemma04properties}(6), $u_{i_0}$ is $([9, 11], 5^-)_{HG}$.
We discuss the following three subcases on the existence of a $(9/10/11, 5^-)$ vertex, which w.l.o.g. is assumed $u_3$.

{Case 4.1.} $u_3$ is $(9, 5^-)$.
In this case, the non-existence of ($A_{3.1}$) in $G$ states that there exits an $i^*\in [1, 2]$ such that $d(u_{i^*}) = d_H(u_{i^*}) = \Delta$.
If $\alpha_{u_i}\ge \frac {13}{18}$ for each $i\in [1, 2]$, then Eq.~(\ref{eq04}) becomes
$\sum_{x\in \{v\}\cup F_3(v)}\omega'(x) 
 \le 4 - 2 (\frac {13}{18} + \frac {13}{18} + \frac 59) = 0$.
Otherwise, there exits an $i_1\in [1, 2]$ such that $\alpha_{u_{i_1}} < \frac {13}{18} < \frac {11}{15}$.
One sees that if $d_H(u_{i_1})\ge 15$, or $u_{i_1}$ is $(14, 8^+)_H$, or $u_{i_1}$ is $([8, 13], 7^+)_H$,
then $\alpha_{u_{i_1}}\ge \frac {13}{18}$.
Therefore, $u_{i_1}$ is $([9, 13], 6^-)_{H G}$, or $u_{i_1}$ is $(14, 7^-)_H$.
We consider them separately below, assuming $i_1 = 2$:

(Case 4.1.1.) $u_2$ is $([9, 13], 6^-)_{H G}$.
In this case, the non-existence of ($A_{3.2}$) in $G$ states that $u_1$ is $(\Delta, (\Delta - d(u_2) + 7)^+)_{H G}$, i.e., $n'_{5^-}(u)\le d(u_2)- 7\le 6$.
Let $\{a, b\} = \{\alpha_{u_1}, \alpha_{u_2}\}$.
When $a\ge \frac 23$ and $b\ge \frac 79$, using $\alpha_{u_3} \ge \frac 59$ Eq.~(\ref{eq04}) becomes
$\sum_{x\in \{v\}\cup F_3(v)}\omega'(x) 
 \le 4 - 2 (\frac 79 + \frac 23 + \frac 59) = 0$.
We thus assume below that $a< \frac 23$ or $b < \frac 79$.
Similarly, if $\alpha_{u_3}\ge \frac 7{12}$, then we assume that $a< \frac 23$ or $b< \frac 34$.
\begin{itemize}
\parskip=0pt
\item[{\rm (1)}]
	$u_3$ is $(9, 4^-)$.
	The non-existence of ($A_{3.3}$) in $G$ states that $u_2$ is $([9, 13], 6)_{HG}$.
	It follows from $\alpha_{u_2}\ge \frac 23$ that $\alpha_{u_1}< \frac 79$.
	If $d(u_1) = \Delta\ge 14$, then $n'_{5^-}(u_1)\le d(u_2)- 7\le 6 \le \lfloor\frac k2\rfloor$,
	and thus $\alpha_{u_1} = \frac{k - 6}{3n'_{5^-}(u)} + \frac 13 \ge
	\frac 79$, a contradiction.
	If $u_1$ is $(9/10, 7^+)_{H G}$, or $u_1$ is $([11, 13], 8^+)_{H G}$, then $\alpha_{u_1}\ge \frac 79$, also a contradiction.
	Therefore, $d_H(u_1) = d(u_1)= \Delta\in [11, 13]$, and furthermore $u_1$ is $([11, 13], 7)_{H G}$.
	Since $d(u_2) < d(u_1)$ implies $n_{6^+}(u_1)\ge 8$, we have $u_2$ is $(\Delta, 6)_{H G}$,
	by which the local structure ($A_{3.4}$) is obtained, a contradiction again.
\item[{\rm (2)}]
	$u_3$ is $(9, 5)$ and $\alpha_{u_3}\ge \frac 7{12}$.
	The non-existence of ($A_{3.6}$) in $G$ states that $u_2$ is $([9, 13], 5/6)_{H G}$ and thus $\alpha_{u_2}\ge \frac 7{12}$.
	When $\alpha_{u_1}\ge \frac{5}{6}$, Eq.~(\ref{eq04}) becomes
	$\sum_{x\in \{v\}\cup F_3(v)}\omega'(x) 
	 \le 4 - 2 (\frac 56 + \frac 7{12} + \frac 7{12}) = 0$.
	When $\alpha_{u_1} < \frac 56$, one sees that $d_H(u_1) = d(u_1) = k$,
	and $u_1$ cannot be $(9^+, (k - 2)^+)_{H G}$, or $(11^+, (k - 3)^+)_{H G}$, or $(12^+, (k - 4)^+)_{H G}$, or $(14^+, (k - 5)^+)_{H G}$.
	Since $u_1$ is $(\Delta, (\Delta - d(u_2) + 7)^+)_{H G}$, we have $n'_{5^-}(u)\le d(u_2)- 7\le 6$, and $k\ge 12$ or $u_1$ is $(10/11, 7)_{H G}$.
	We discuss the following three subcases for three possible values of $d(u_2)$.
	\begin{itemize}
	\parskip=0pt
	\item
		$u_2$ is $(9, 5/6)_{H G}$.
		It follows that $d(u_1) = \Delta\ge 10$, and then $n'_{5^-}(u_1)\le d(u_2)- 7\le 2$.
		Thus, $u_1$ is $(10^+, (k-2)^+)_{H G}$.
	\item
		$u_2$ is $(10, 5/6)_{H G}$.
		If $d(u_1) = \Delta\ge 11$, then $n'_{5^-}(u_1)\le d(u_2)- 7\le 3$, and thus $u_1$ is $(11^+, (k - 3)^+)$.
		Otherwise, $u_1$ is $(10, 7)_{H G}$.
		It follows from $\alpha_{u_1} = \frac 79$ that $\alpha_{u_2}< \frac 23$.
		Then $u_2$ is $(10, 5)_{H G}$, and thus ($A_{3.7}$) is obtained, a contradiction.
	\item
		$u_2$ is $(11, 5/6)_{H G}$.
		If $d(u_1) = \Delta\ge 12$, then $n'_{5^-}(u_1)\le d(u_2)- 7\le 4$, and $u_1$ is $(12^+, (k - 4)^+)$.
		Otherwise, $u_1$ is $(11, 7)_{H G}$.
		It follows from $\alpha_{u_1} = \frac 34$ that $\alpha_{u_2}< \frac 23$.
		Then $u_2$ is $(11, 5)_{H G}$, and thus ($A_{3.7}$) is obtained, a contradiction.
	\item
		$u_2$ is $(12/13, 5/6)_{H G}$.
		It follows from $\alpha_{u_2}\ge \frac 23$ that $\alpha_{u_1}< \frac 34$.
		One sees that $d_H(u_1) = d(u_1) = k\ge 12$,
		and if $u_1$ is $(12^+, (k - 4)^+)_{H G}$, or $(13^+, (k - 5)^+)_{H G}$, or $(14^+, (k - 6)^+)_{H G}$, then $\alpha_{u_1}\ge \frac 34$.
		Hence, $k\ge 14$, or $u_1$ is $(12/13, 7)_{H G}$.
		\begin{itemize}
		\parskip=0pt
		\item
			$u_2$ is $(12, 5/6)_{H G}$.
			If $d(u_1) = \Delta\ge 13$, then $n'_{5^-}(u_1)\le d(u_2)- 7\le 5$, and thus $u_1$ is $(13^+, (k - 5)^+)$.
			Otherwise, $u_1$ is $(12, 7)_{H G}$, and thus ($A_{3.8}$) is obtained, a contradiction.
		\item
			$u_2$ is $(13, 5/6)_{H G}$.
			If $d(u_1) = \Delta\ge 14$, then $n'_{5^-}(u_1)\le d(u_2)- 7\le 6$, and thus $u_1$ is $(14^+, (k - 6)^+)$.
			Otherwise, $u_1$ is $(13, 7)_{H G}$, and thus ($A_{3.8}$) is obtained, a contradiction.
		\end{itemize}
	\end{itemize}
\end{itemize}

(Case 4.1.2.) None of $\{u_1, u_2\}$ is $([9, 13], 6^-)_{H G}$ and $u_2$ is $(14, 7^-)_H$.
In this case, $d_H(u_1)\ge 14$ or $u_1$ is $([9, 13], 7^+)_H$.
It follows from Lemma~\ref{lemma04properties}(1) and (3) that $\alpha_{u_1}\ge \frac 57$ and $\alpha_{u_2}\ge \frac 57$.
If $u_3$ is $(9, 5)$, then Eq.~(\ref{eq04}) becomes
$\sum_{x\in \{v\}\cup F_3(v)}\omega'(x) 
 \le 4 - 2 (\frac 57 + \frac 57 + \frac 7{12})= - \frac 1{42} < 0$.
Otherwise, $u_3$ is $(9, 4^-)$.
Note that $2 - (\alpha_{u_3} + \alpha_{u_2}) = 2 - (\frac 59 +  \frac 57) < \frac {11}{15}$.
One sees that if $d_H(u_1)\ge 15$, or $u_1$ is $(13/14, 8^+)_H$, or $u_1$ is $([9, 12], 7^+)_H$,
then $\alpha_{u_1}\ge \frac {11}{15}$ and thus Eq.~(\ref{eq04}) becomes $\sum_{x\in \{v\}\cup F_3(v)}\omega'(x) < 0$.
In the other case, $u_1$ is $(13, 7)_H$ or $(14, 7^-)_H$, either of which leads to a structure ($A_{3.5}$), a contradiction.

{Case 4.2.} None of $\{u_1, u_2, u_3\}$ is $(9, 5^-)$ and $u_3$ is $(10, 5^-)$.
In this case, for each $i\in [1, 3]$, $\alpha_{u_i}\ge \frac 35$.
When $\alpha_{u_i}\ge \frac 7{10}$ for each $i\in [1, 2]$, Eq.~(\ref{eq04}) becomes
$\sum_{x\in \{v\}\cup F_3(v)}\omega'(x) \le 4 - 2 (\frac 7{10} + \frac 7{10} + \frac 35) = 0$.
When there exits an $i_1\in [1, 2]$ such that $\alpha_{u_{i_1}}< \frac 7{10}$,
by Lemma~\ref{lemma04properties}(1) and (3), $u_{i_1}$ is $(9, 6)_{HG}$ or $([10, 13], 6^-)_{H G}$.
Assuming w.l.o.g. $i_1 = 2$.
The non-existence of ($A_{4.1}$) in $G$ implies that $u_2$ is $(9, 6)_{H G}$ or $([10, 13], 5/6)_{H G}$;
the non-existence of ($A_{4.2}$) in $G$ implies that $n'_{6^+}(u_1) = n_{6^+}(u_1)\ge \Delta - d(u_2) + 7\ge 7$ and $n'_{5^-}(u_1)\le d(u_2) - 7$.
It follows from Lemma~\ref{lemma04properties} that $\alpha_{u_1}\ge \frac 57$.
Recall that $\alpha_{u_3}\ge \frac 35$, and $2 - (\alpha_{u_3} + \alpha_{u_1}) = 2 - (\frac 35 +  \frac 57) < \frac 9{13}$.
If $d_H(u_1)\ge 13$, then $\alpha_{u_1}\ge \frac 9{13}$ and thus Eq.~(\ref{eq04}) becomes $\sum_{x\in \{v\}\cup F_3(v)}\omega'(x) < 0$.
In the other case, $u_1$ is $(9, 6)_{H G}$ or $([10, 12], 5/6)_{H G}$, and we discuss them separately below:

({Case 4.2.1}) $u_2$ is $([9, 11], 6)$ or $(12, 5/6)$.
In this case, $\alpha_{u_2}= \frac 23$.
When $d_H(u_1)\ge 15$, or $u_1$ is $(12^-, 7^+)_H$ or $(13/14, 8^+)_H$, $\alpha_{u_1}\ge \frac {11}{15}$, and thus Eq.~(\ref{eq04}) becomes
$\sum_{x\in \{v\}\cup F_3(v)}\omega'(x) \le 4 - 2(\frac {11}{15} + \frac 23 + \frac 35) = 0$.
In the other case, $u_1$ is $(13/14, 7)_H$ and thus $\Delta\ge d_H(u_1)\ge 13$.
It follows that $7 \ge n'_{6^+}(u_1)\ge \Delta - d(u_2) + 7\ge 13 - 12 + 7 = 8$, a contradiction.

({Case 4.2.2}) $u_2$ is $(10/11, 5)$.
The non-existence of ($A_{4.3}$) in $G$ implies that $d_H(u_1) = d(u_1) = \Delta$.
Hence, $u_1$ is $(\Delta, (\Delta - d(u_2) + 7)^+){H G}$ where $\Delta\ge d(u_1)\ge 10$.
When $u_1$ is $(10, 8^+)_{H G}$, $(11^+, (k - 3)^+)_{H G}$, or $(12^+, (k - 4)^+)_{H G}$, we have $\alpha_{u_1}\ge \frac{4}{5}$,
and thus Eq.~(\ref{eq04}) becomes
$\sum_{x\in \{v\}\cup F_3(v)}\omega'(x) \le 4 - 2(\frac 45 + \frac 35 + \frac 35) = 0$.
In the other case, $d_H(u_1) = d(u_1) = k \ge 12$ or $u_1$ is $(10/11, 7)_{H G}$.
We discuss them separately below to derive contradictions:

\begin{itemize}
\parskip=0pt
\item[{\rm (1)}]
	$u_2$ is $(10, 5)$.
	If $d(u_1) = \Delta\ge 11$, then $n'_{5^-}(u_1)\le d(u_2)- 7\le 3$;
	thus $u_1$ is $(11^+, (k - 3)^+)_{H G}$, a contradiction.
	Otherwise, $u_1$ is $(10, 7)_{H G}$, and thus ($A_{4.4}$) is obtained, a contradiction too.
\item[{\rm (2)}]
	$u_2$ is $(11, 5)$.
	If $d(u_1) = \Delta\ge 12$, then $n'_{5^-}(u_1)\le d(u_2)- 7\le 4$;
	thus $u_1$ is $(12^+, (k - 4)^+)_{H G}$, a contradiction.
	Otherwise, $u_1$ is $(11, 7)_{H G}$, and thus ($A_{4.4}$) is obtained, a contradiction too.
\end{itemize}

{Case 4.3.} None of $\{u_1, u_2, u_3\}$ is $(9/10, 5^-)$ and $u_3$ is $(11, 5^-)$.
In this case, for each $i\in [1, 3]$, $\alpha_{u_i}\ge \frac 7{11}$.
When $\alpha_{u_i}\ge \frac {15}{22}$ for each $i\in [1, 2]$, Eq.~(\ref{eq04}) becomes
$\sum_{x\in \{v\}\cup F_3(v)}\omega'(x) \le 4 - 2(\frac {15}{22} + \frac {15}{22} + \frac 7{11}) = 0$.
In the other case, there exits an $i_1\in [1, 2]$ such that $\alpha_{u_{i_1}}< \frac{15}{22} < \frac{9}{13}$.
By Lemma~\ref{lemma04properties}(1) and (3), $u_{i_2}$ is $([9, 12], 6)_{H G}$, or $(12, 5^-)_{H G}$, or $(11, 5^-)_{H G}$.
Assume w.l.o.g. $i_2 = 2$.
We have the following two scenarios:

({Case 4.3.1}) $u_2$ is $([9, 12], 6)_{H G}$, or $(12, 5^-)_{H G}$.
In this case, $\alpha_{u_2} = \frac 23$.
When $\alpha_{u_1}\ge \frac {23}{33}$, Eq.~(\ref{eq04}) becomes
$\sum_{x\in \{v\}\cup F_3(v)}\omega'(x) \le 4 - 2(\frac {23}{33} + \frac 23 + \frac 7{11}) = 0$.
In the other case, $\alpha_{u_1} < \frac{23}{33} < \frac 57$.
One sees that if $d_H(u_1)\ge 14$ or $n'_{6^+}(u_1)\ge 7$ by Lemma~\ref{lemma04properties}(1) and (3),
then $\alpha_{u_1}\ge \frac 57$ and thus Eq.~(\ref{eq04}) becomes $\sum_{x\in \{v\}\cup F_3(v)}\omega'(x) < 0$.
It follows that $n_{6^+}(u_1) = n'_{6^+}(u_1)\le 6$, and $u_1$ is $([9, 13], 6)_{H G}$, or $([11, 14], 5^-)$;
but then ($A_{5.1}$) is obtained, a contradiction.

({Case 4.4.2}) None of $\{u_1, u_2\}$ is $([9, 12], 6)_{H G}$ or $(12, 5^-)_{H G}$, and $u_2$ is $(11, 5^-)$.
In this case, $\alpha_{u_2} = \frac 23$.
When $d_H(u_1)\ge 15$, or $u_1$ is $(12^-, 7^+)_H$, or $(13/14, 8^+)_H$, we have $\alpha_{u_1}\ge \frac 8{11}$,
and thus Eq.~(\ref{eq04}) becomes
$\sum_{x\in \{v\}\cup F_3(v)}\omega'(x) \le 4 - 2(\frac 8{11} + \frac 7{11} + \frac 7{11}) = 0$.
In the other case, $d_H(u_1)\le 14$, and $n_{6^+}(u_1) = n'_{6^+}(u_1)\le 6$ or $u_1$ is $(13/14, 7)_H$.
Recall that $d(u_1) = d_H(u_1)$ if $n'_{6^+}(u_1)\le 6$.
The non-existence of ($A_{5.1}$) in $G$ implies that $u_1$ is $(13/14, 7)_H$, that is, $u_1$ is $(13^+, 7)$ in $G$,
which leads to a structure ($A_{5.2}$), a contradiction.
This finishes all cases, and thus proves the lemma.
\end{proof}

\begin{demma}
\label{demma04}
For every vertex $u\in V(H)$ with $d_H(u) = d(u) = 4$, we have
\begin{itemize}
\parskip=0pt
\item[{\rm (1)}]
	for each $6^+$-vertex $x$ in $N_H(u)$, the total weight transferred from $u$ and the $3$-faces in $F(u x)$ to $x$ is $(1 + \frac 12 m_3(u x))\alpha_{x}$;
\item[{\rm (2)}]
	if $u$ is adjacent to some $4$-vertex $v$ in $H$, then $\sum_{x\in \{u, v\}\cup F_3(u)\cup F_3(v)}\omega'(x)\le 0$;
\item[{\rm (3)}]
	if $d_H(v)\ge 5$ for each $v\in N_H(u)$, then $\sum_{x\in \{u\}\cup F_3(u)}\omega'(x)\le 0$.
\end{itemize}
\end{demma}
\begin{proof}
For each $x\in N_H(u)$ with $d_H(x)\ge 6$, let $\Gamma(u \rightarrow x)$ be the total weight transferred from $u$ and the $3$-faces in $F(u x)$ to $x$.
To prove (1), we discuss the following three subcases for the three possible values of $m_3(u x)$.
\begin{itemize}
\parskip=0pt
\item
	$m_3(u x) = 0$. 
	By (R1), $\tau(u \rightarrow x)_f = \tau(u \rightarrow x)_{f_{u x}} = \frac 12 \times \alpha_{x}$;
	and thus $\Gamma(u \rightarrow x) = \tau(u \rightarrow x)_f + \tau(u \rightarrow x)_{f_{u x}} = \alpha_{x}$.
\item
	$m_3(u x) = 1$ with $d(f_3) = 3$. 
	By (R3.1) and (R3.2), $\tau(f \rightarrow x)_f = \alpha_{x}$, $\tau(f \rightarrow x)_{f_{u x}} = \frac 12 \times \alpha_{x}$;
	and thus $\Gamma(u \rightarrow x) = \tau(f \rightarrow x)_f + \tau(f \rightarrow x)_{f_{u x}} = \alpha_{x} + \frac 12 \times \alpha_{x} = \frac 32\alpha_{x}$.
\item
	$m_3(u x) = 2$. 
	By (R3.1), $\tau(f \rightarrow x)_f = \tau(f_{ux} \rightarrow x)_{f_{u x}} = \alpha_{x}$;
	and thus $\Gamma(u \rightarrow x) = \tau(f \rightarrow x)_f + \tau(f_{ux} \rightarrow x)_{f_{u x}} = \alpha_{x} + \alpha_{x} = 2\alpha_{x}$.
\end{itemize}
Hence, the total weight transferred from $u$ and the $3$-faces in $F(u x)$ to $x$ is $(1 + \frac 12 m_3(u x))\alpha_{x}$.

Now let $u_1, u_2, v, u_3 $ be all the neighbors of $u$ in clockwise order,
and denote the face containing $u_1 u u_2$ ($u_2 u v$, $v u u_3$, $u_3 u u_1$, respectively) as $f_1$ ($f_2, f_3, f_4$, respectively).
Note that
\begin{equation}
\label{eq05}
\sum_{x\in \{u, v\}\cup F_3(u)\cup F_3(v)}\omega(x)
= \sum_{x\in F_3(u)\cup F_3(v)}\omega(x) = m_3(u) + m_3(v) - m_3(uv).
\end{equation}

To prove (2), we have $d(v) = d_H(v) = 4$ and let $v_1, v_2, v_3 $ be the other neighbors of $v$ in clockwise order,
and denote the face containing $v_1 v v_2$ ($v_2 v v_3$, respectively) as $h_1$ ($h_2$, respectively).

For each $x\in N_H(u)\cup N_H(v)$, it follows from Lemma~\ref{lemma02} that $d(x)\ge 10$ and $d_H(x) \ge 8$;
and then by Lemma~\ref{lemma04properties}(1) and (4), we have $\alpha_x\ge \frac 35$.
Using (R1), (R3.1) and (R3.1), we discuss the following three subcases for the three possible values of $m_3(uv)$.

{Case 1.} $m_3(uv) = 0$.
In this case, $m_3(u)\le 2$ and $m_3(v)\le 2$.
It follows from Eq.~(\ref{eq05}) that
\begin{equation}
\label{eq06}
\begin{aligned}
\sum_{x\in \{u, v\}\cup (F_3(u)\cup F_3(v))}\omega'(x)
& = \sum_{x\in \{u, v\}\cup (F_3(u)\cup F_3(v))}\omega(x) - \sum_{x\in \{u_1, u_2, u_3, v_1, v_2, v_3\}}\Gamma(u \rightarrow x)\\
& = m_3(u)+ m_3(v) - \sum_{x\in \{u_1, u_2, u_3, v_1, v_2, v_3\}}\Gamma(u \rightarrow x).
\end{aligned}
\end{equation}
By (1), $\sum_{x\in N_H(u)\setminus \{v\}}\Gamma(u \rightarrow x)\ge \alpha_{u_1} + \alpha_{u_2} + \alpha_{u_3}$,
$\sum_{x\in N_H(v)\setminus \{u\}}\Gamma(u \rightarrow x)\ge \alpha_{v_1} + \alpha_{v_2} + \alpha_{v_3}$.

When $m_3(u)+ m_3(v)\le 3$, Eq.~(\ref{eq06}) becomes
$\sum_{x\in \{u, v\}\cup (F_3(u)\cup F_3(v))}\omega'(x)
 \le 3 - (\alpha_{u_1} + \alpha_{u_2} + \alpha_{u_3} + \alpha_{v_1} + \alpha_{v_2} + \alpha_{v_3}) \le 3 - 6\times \frac 35 = - \frac 35 < 0$.

In the other case, $m_3(u)+ m_3(v) = 4$, i.e., $m_3(u) = m_3(v)= 2$.
Then by (1), $\sum_{x\in N_H(u)\setminus \{v\}}\Gamma(u \rightarrow x)\ge 2\alpha_{u_1} + \alpha_{u_2} + \alpha_{u_3}$,
$\sum_{x\in N_H(v)\setminus \{u\}}\Gamma(u \rightarrow x)\ge \alpha_{v_1} + 2\alpha_{v_2} + \alpha_{v_3}$, and Eq.~(\ref{eq06}) becomes
$\sum_{x\in \{u, v\}\cup (F_3(u)\cup F_3(v))}\omega'(x)
 \le 4 - (2\alpha_{u_1} + \alpha_{u_2} + \alpha_{u_3} + \alpha_{v_1} + 2\alpha_{v_2} + \alpha_{v_3}) \le 4 - 8\times \frac 35 = - \frac 45< 0$.

{Case 2.} $m_3(u v) = 1$ with $d(f_2) = 3$.
In this case, $u_2 = v_1$.
Let $\Gamma(u/v \rightarrow u_2)$ be the total weight transferred from $u$, $v$ and  the $3$-faces in $F(u u_2)\cup F(v u_2)$ to $u_2$.
If $d(f_1) = d(h_1) = 3$,
then $\Gamma(u/v \rightarrow u_2) = \tau(f_1 \rightarrow u_2)_{f_1} + \tau(f_2 \rightarrow u_2)_{f_2} + \tau(h_1 \rightarrow u_2)_{h_1} = 3\alpha_{u_2}$;
if $d(f_1), d(h_1)\ge 4$,
then $\Gamma(u/v \rightarrow u_2) = \tau(f_2 \rightarrow u_2)_{f_1} + \tau(f_2 \rightarrow u_2)_{f_2} + \tau(f_2 \rightarrow u_2)_{h_1} = 2\alpha_{u_2}$;
otherwise, only one of $\{f_1, h_1\}$ is a $3$-face, assuming w.l.o.g. $d(f_1) = 3$,
thus $\Gamma(u/v \rightarrow u_2) = \tau(f_1 \rightarrow u_2)_{f_1} + \tau(f_2 \rightarrow u_2)_{f_2} + \tau(f_2 \rightarrow u_2)_{h_1}= \frac 52\alpha_{u_2}$.
That is,
\begin{equation}
\label{eq07}
\begin{aligned}
\Gamma(u/v \rightarrow u_2) = \begin{cases}
	3\alpha_{u_2},				&\mbox{ if } d(f_1) = d(h_1)= 3;\\
	2\alpha_{u_2},				&\mbox{ if } d(f_1), d(h_1)\ge 4;\\
	\frac 52\alpha_{u_2},	&\mbox{ otherwise}.\end{cases}
\end{aligned}
\end{equation}
By (1) and Eq.~(\ref{eq05}), $\sum_{x\in \{u_1, u_3\}}\Gamma(u \rightarrow x)\ge \alpha_{u_1} + \alpha_{u_3}$,
$\sum_{x\in \{v_2, v_3\}}\Gamma(u \rightarrow x)\ge \alpha_{v_2} + \alpha_{v_3}$, and
\begin{equation}
\label{eq08}
\sum_{x\in \{u, v\}\cup F_3(u)\cup F_3(v)}\omega'(x) = m_3(u)+ m_3(v)- 1 - \sum_{x\in \{u_1, u_3, v_2, v_3\}}\Gamma(u \rightarrow x) - \Gamma(u/v \rightarrow u_2).
\end{equation}

When $m_3(u)+ m_3(v)\le 4$, Eq.~(\ref{eq08}) becomes
$\sum_{x\in \{u, v\}\cup F_3(u)\cup F_3(v)}\omega'(x)
 \le 3 - (\alpha_{u_1} + \alpha_{u_3} + \alpha_{v_2} + \alpha_{v_3} + 2\alpha_{u_2}) \le 3 - 6\times \frac 35 = - \frac 35< 0$.

In the other case, $5\le m_3(u)+ m_3(v)\le 6$.
It follows that $\max\{m_3(u), m_3(v)\}\ge 3$.
Assume w.l.o.g. $m_3(u) = 3$, i.e., $d(f_1) = d(f_4) = 3$, and $m_3(v) \ge 2$.
Then by (1) and Eq.~(\ref{eq07}), $\sum_{x\in \{u_1, u_3\}}\Gamma(u \rightarrow x)\ge 2\alpha_{u_1} + \frac 32 \alpha_{u_3}$,
$\Gamma(u/v \rightarrow u_2) \ge 52 \alpha_{u_2}$, $\sum_{x\in \{v_2, v_3\}}\Gamma(u \rightarrow x)\ge \frac 32 \alpha_{v_2} + \alpha_{v_3}$,
and hence Eq.~(\ref{eq08}) becomes
$\sum_{x\in \{u, v\}\cup F_3(u)\cup F_3(v)}\omega'(x)
 \le 6 - 1 - (2\alpha_{u_1} + \frac 32 \alpha_{u_3} + \frac 32 \alpha_{v_2} + \alpha_{v_3} + \frac 52 \alpha_{u_2}) \le 5 - \frac {17}2\times \frac 35 = - \frac 1{10}< 0$.

{Case 3.} $m_3(u v) = 2$.
In this case, $u_2 = v_1$ and $u_3 = v_3$.
Let $\Gamma(u/v \rightarrow u_3)$ be the total weight transferred from $u$, $v$ and the $3$-faces in $F(u u_2)\cup F(v u_3)$ to $u_3$.
We have the following similar to Eq.~(\ref{eq07}):
\begin{equation}
\label{eq09}
\begin{aligned}
\Gamma(u/v \rightarrow u_3) = \begin{cases}
	3\alpha_{u_3},				&\mbox{ if } d(f_4) = d(h_2) = 3;\\
	2\alpha_{u_3},				&\mbox{ if } d(f_4), d(h_2)\ge 4;\\
	\frac 52 \alpha_{u_3},	&\mbox{ otherwise}.\end{cases}
\end{aligned}
\end{equation}
By (1), $\Gamma(u \rightarrow u_1)= (1 + \frac{1}{2}m_3(uu_1))\alpha_{u_1}$, $\Gamma(u \rightarrow v_2) = (1 + \frac 12 m_3(v v_2))\alpha_{v_2}$.
Since $m_3(u) + m_3(v) - m_3(u v) = m_3(u u_1) + 2 + m_3(v v_2) + 2 - 2 =  m_3(u u_1) + 2 + m_3(v v_2)$, we have
\begin{equation}
\label{eq10}
\begin{aligned}
	&\sum_{x\in \{u, v\}\cup (F_3(u)\cup F_3(v))}\omega'(x)
 = m_3(u) + m_3(v) -2 - \sum_{x\in \{u_1, v_2\}}\Gamma(u \rightarrow x) - \sum_{y\in \{u_2, u_3\}}\Gamma(u/v \rightarrow y)\\
= \	&(1-\frac 12 \alpha_{u_1})\times m_3(u u_1) + (1-\frac 12 \alpha_{v_2})\times m_3(v v_2) + \alpha_{u_1} + \alpha_{v_2}
	+ 2 - \sum_{y\in \{u_2, u_3\}}\Gamma(u/v \rightarrow y)
\end{aligned}
\end{equation}

When $m_3(u u_1) + m_3(v v_2)\le 2$, Eq.~(\ref{eq10}) becomes
$\sum_{x\in \{u, v\}\cup (F_3(u)\cup F_3(v))}\omega'(x)
 \le (1-\frac 12 \alpha_{u_1})\times m_3(u u_1) + (1-\frac 12 \alpha_{v_2})\times m_3(v v_2) + 2 - \alpha_{u_1} - \alpha_{v_2}- 2\alpha_{u_2}  - 2\alpha_{u_3}
 \le (1 - \frac 12 \times \frac 35)\times 2 + 2 - 6\times \frac 35 = - \frac 15< 0$.

In the other case, $3\le m_3(u u_1) + m_3(v v_2) \le 4$.
It follows that $\max\{m_3(u u_1), m_3(v v_2)\} \ge 2$.
Assume w.l.o.g. $m_3(u u_1) = 2$, $m_3(v v_2) \ge 1$ with $d(h_1) = 3$.
Then $\Gamma(u/v \rightarrow u_2) = 3\alpha_{u_2}$.
One sees that if $d(h_2)\ge 4$, $\Gamma(u/v \rightarrow u_3) = \frac 52 \alpha_{u_3}$;
if $d(h_2)= 3$, $\Gamma(u/v \rightarrow u_3) = 3\alpha_{u_3}$, i.e., $\Gamma(u/v \rightarrow u_3) = (\frac 52 + \frac 12 \times (m_3(v v_2) - 1))\alpha_{u_3}$.
Hence, Eq.~(\ref{eq10}) becomes
\begin{equation*}
\begin{aligned}
		&\ \sum_{x\in \{u, v\}\cup (F_3(u)\cup F_3(v))}\omega'(x)\\
\le 	&\ (1-\frac 12 \alpha_{u_1})\times 2 + (1-\frac 12 \alpha_{v_2})\times m_3(v v_2) + 2 - \alpha_{u_1} - \alpha_{v_2}
			- 3\alpha_{u_2} - (\frac{5}{2} + \frac{1}{2}\times (m_3(vv_2) - 1))\times \alpha_{u_3}\\
= 		&\ 4 - 2\alpha_{u_1}- \alpha_{v_2} - 3\alpha_{u_2} - 2\alpha_{u_3} + (1 - \frac 12 \alpha_{v_2} - \frac 12 \alpha_{u_3}) m_3(v v_2)\\
\le 	&\ 4 - 8\times \frac 35 + (1 - \frac 35)m_3(v v_2)= - \frac 45 + (1 - \frac 35)m_3(v v_2)\\
=		&\ - \frac 45 + \frac 25 \times 2 = 0.
\end{aligned}
\end{equation*}

To prove (3), we have
\[
\sum_{x\in \{u\}\cup F_3(u)}\omega(x) = \omega(u) + \sum_{x\in F_3(u)}\omega(x) = \sum_{x\in F_3(u)} 1 = m_3(u) = m_3(uu_1) + m_3(uv)= m_3(uu_2) + m_3(uu_3).
\]
Using (R1), (R3.1), (R3.2), (R4.1) and (R4.2), we discuss the following two cases for the minimum degree of vertices in $N_H(u)$.

Case 1. For each $x\in N_H(u)$, $d_H(x)\ge 6$.
In this case, for each $x\in \{u_1, u_2, u_3\}$, $\alpha_x\ge \frac 13$,
and by (1), $\Gamma(u \rightarrow x) = (1 + \frac 12 m_3(ux))\alpha_{x}$,
$\Gamma(u \rightarrow v) = (1 + \frac 12 m_3(uv))\alpha_{v}$,
and thus
\begin{equation}
\label{eq11}
\sum_{x\in \{u\}\cup F_3(u)}\omega'(x) = m_3(u)- \sum_{x\in N_H(x)}\Gamma(u \rightarrow x) = m_3(u)- \sum_{x\in N_H(x)}(1 + \frac 12 m_3(ux))\alpha_{x}.
\end{equation}
One sees from Lemma~\ref{lemma04properties}(1) and (4) that if $\alpha_x < \frac 12$, then $d(x) = d_H(x)\in [6, 7]$;
furthermore, by the definition of $\alpha_x$ in Eq.~(\ref{eq02}), either $d(x) = d_H(x) = 6$, or $x$ is $(7, 4^-)_{H G}$.
We separate the discussion into the following three scenarios:

{Case 1.1.} For each $x\in N_H(x)$, $\alpha_x \ge \frac 12$.
When $m_3(u) = m_3(u u_1) + m_3(u v) \le 2$, Eq.~(\ref{eq11}) becomes
$\sum_{x\in \{u\}\cup F_3(u)}\omega'(x)
 \le m_3(u u_1) + m_3(u v) - (1 + \frac 12 m_3(u u_1))\alpha_{u_1} - \alpha_{u_2}- \alpha_{u_3} - (1 + \frac 12 m_3(u v))\alpha_{v} \le -\frac 12< 0$.
In the other case,  $3\le m_3(u) = m_3(u u_1) + m_3(u v)\le 4$, and we assume w.l.o.g. $m_3(u u_1) = 2$;
it follows that $m_3(u v)\ge 1$ and $d(f_2) = 3$.
Furthermore, we have $m_3(u u_2)= 2$, $m_3(u u_3)\ge 1$, and thus $m_3(u) = 3 + m_3(u u_3) - 1 = m_3(u u_3) + 2$ and
$(1 + \frac 12 m_3(u v))\alpha_{v} = (1 + \frac 12 m_3(u u_3))\alpha_{v}$.
This way, Eq.~(\ref{eq11}) becomes
$\sum_{x\in \{u\}\cup F_3(u)}\omega'(x)
 = m_3(u u_3) + 2 - 2\alpha_{u_1} - 2\alpha_{u_2}- (1 + \frac 12 m_3(u u_3))\alpha_{u_3} - (1 + \frac 12 m_3(u u_3))\alpha_{v} \le 0$.

{Case 1.2.} $d(v) = d_H(v) = 6$ and $d_H(x)\ge 6$ for each $x\in \{u_1, u_2, u_3\}$.
Recall that for each $x\in N_H(x)$, $\Gamma(u \rightarrow x)= (1 + \frac 12 m_3(u x))\alpha_{x}$ and $\alpha_x\ge \frac 13$.
When $m_3(u)\le 1$, Eq.~(\ref{eq11}) becomes
$\sum_{x\in \{u\}\cup F_3(u)}\omega'(x) \le 1- \sum_{x\in N_H(x)}\alpha_x = 1 - 4\times \frac 13 = - \frac 13< 0$.

In the other case, $m_3(u) \ge 2$.
If $m_3(u v) = 0$, i.e., $d(f_1) = d(f_4) = 3$, then $m_3(u) = 2$, and thus Eq.~(\ref{eq11}) becomes
$\sum_{x\in \{u\}\cup F_3(u)}\omega'(x) \le 2 - (2\alpha_{u_1} +  \frac 32 \alpha_{u_2} + \frac 32 \alpha_{u_3} + \alpha_v) = 2 - 6\times \frac 13 = 0$.
If $m_3(uv) \ge 1$ (and thus $d(f_2)= 3$), it follows from Lemma~\ref{lemma02}(2.3) that for each $x\in \{u_1, u_2, u_3\}$, $d_H(x) \ge 8$ and $d(x)\ge 8$;
and thus, by Lemma~\ref{lemma04properties}(1) and (4), $\alpha_x\ge \frac 12$.
Therefore, when $m_3(u) = 2$, Eq.~(\ref{eq11}) becomes
$\sum_{x\in \{u\}\cup F_3(u)}\omega'(x)
 \le 2 - (\alpha_{u_1} + \frac 32 \alpha_{u_2} + \alpha_{u_3} + \frac 32 \alpha_v) = 2 - \frac 72\times \frac 12 - \frac 32 \times \frac 13 = -\frac 14< 0$.
In the other case, $m_3(u) \ge 3$.
If $m_3(u v) = 1$, i.e., $d(f_1) = d(f_4) = 3$, then $m_3(u) = 3$, and Eq.~(\ref{eq11}) becomes
$\sum_{x\in \{u\}\cup F_3(u)}\omega'(x)
 \le 3 - (2\alpha_{u_1} +  2 \alpha_{u_2} + \frac 32 \alpha_{u_3} + \frac 32 \alpha_v) = 3 - \frac {11}2 \times \frac 12 - \frac 32 \times \frac 13  = - \frac 12< 0$.
If $m_3(u v) = 2$, then it follows from Lemma~\ref{lemma02}(2.5) that for each $x\in \{u_1, u_2, u_3\}$, $d_H(x)\ge 8$ and $d(x)\ge 9$;
and thus, by Lemma~\ref{lemma04properties}(1) and (4), $\alpha_x\ge \frac 59$.
Since $m_3(u) \ge 3$, assuming w.l.o.g. $d(f_1) = 3$, we have $m_3(u) = m_3(u u_3) + 2 = m_3(u u_1) + 2$, and Eq.~(\ref{eq11}) becomes
$\sum_{x\in \{u\}\cup F_3(u)}\omega'(x) = m_3(u u_1) + 2 - ((1 + \frac 12 m_3(u u_1))\alpha_{u_1} + 2\alpha_{u_2} + (1 + \frac 12 m_3(u u_1))\alpha_{u_3} + 2\alpha_v)
\le 2 +  \frac 49 \times m_3(u u_1) - \frac{26}{9}
\le 2 +  \frac 49 \times 2 - \frac{26}{9} = 0$.

{Case 1.3.} $v$ is $(7, 4^-)$ and $d_H(x)\ge 7$ for each $x\in \{u_1, u_2, u_3\}$.
Recall that for each $x\in N_H(x)$, $\Gamma(u \rightarrow x)= (1 + \frac 12 m_3(u x))\alpha_{x}$ and $\alpha_x\ge \frac 37$.
When $m_3(u)\le 1$, Eq.~(\ref{eq11}) becomes
$\sum_{x\in \{u\}\cup F_3(u)}\omega'(x) \le 1- \sum_{x\in N_H(x)}\alpha_x = 1 - 4\times \frac 37 = - \frac 57< 0$.

In the other case, $m_3(u)\ge 2$.
If $m_3(uv) = 0$, i.e., $d(f_1) = d(f_4) = 3$, then $m_3(u) = 2$, and Eq.~(\ref{eq11}) becomes
$\sum_{x\in \{u\}\cup F_3(u)}\omega'(x) \le 2 - (2\alpha_{u_1} +  \frac 32 \alpha_{u_2} + \frac 32 \alpha_{u_3} + \alpha_v) = 2 - 6\times \frac 37 = -\frac 47< 0$.
If $m_3(uv) \ge 1$ (and thus $d(f_2) = 3$), then when $m_3(u) \le 2$, Eq.~(\ref{eq11}) becomes
$\sum_{x\in \{u\}\cup F_3(u)}\omega'(x) \le 2 - (\alpha_{u_1} + \frac 32 \alpha_{u_2} + \alpha_{u_3} + \frac 32 \alpha_v) = 2 - 5\times \frac 37 = -\frac 17< 0$;
when $m_3(u) = 3$, no matter which two of $\{f_1, f_4, f_3\}$ are $3$-faces, Eq.~(\ref{eq11}) becomes $\sum_{x\in \{u\}\cup F_3(u)}\omega'(x) \le 3 - 7\times \frac 37 = 0$;
when $m_3(u) = 4$, it follows from Lemma~\ref{lemma02}(2.6) that for each $x\in \{u_1, u_2, u_3\}$, $d_H(x)\ge 8$ and $d(x)\ge 9$,
and thus, by Lemma~\ref{lemma04properties}(1) and (4), $\alpha_x\ge \frac 59$, and Eq.~(\ref{eq11}) becomes
$\sum_{x\in \{u\}\cup F_3(u)}\omega'(x) = 4 - (2\alpha_{u_1} + 2\alpha_{u_2} + 2\alpha_{u_3} + 2\alpha_v) \le 4 - 6\times \frac 59 - 2\times \frac 37 = \frac 4 {21}< 0$.

Case 2. For some $x\in N_H(u)$, $d(x) = d_H(x)= 5$, assuming w.l.o.g. $d(v) = d_H(v) = 5$.
In this case, for each $x\in \{u_1, u_2, u_3\}$, it follows from Lemma~\ref{lemma02} (2.2) that $d_H(x) \ge 8$ and $d_(x) \ge 9$;
and then by Lemma~\ref{lemma04properties}(1) and (4), $\alpha_x\ge \frac 59$.
Recall that for each $x\in \{u_1, u_2, u_3\}$, $\Gamma(u \rightarrow x) = (1 + \frac 12 m_3(u x))\alpha_{x}$,
and if $m_3(u) \ne 4$ then $\Gamma(u \rightarrow v)= 0$.
When $m_3(u) \le 1$, Eq.~(\ref{eq11}) becomes
$\sum_{x\in \{u\}\cup F_3(u)}\omega'(x) \le 1- \sum_{x\in N_H(x)}\alpha_x = 1 - 3\times \frac 59 = - \frac 23< 0$.

In the other case, $m_3(u) \ge 2$.
If $m_3(uv) = 0$, assuming w.l.o.g. $d(f_1) = 3$,
then $m_3(u) = m_3(u u_1) = m_3(u u_3) + 1$, $\Gamma(u \rightarrow u_2) = \frac 32 \alpha_{u_2}$,
$\Gamma(u \rightarrow u_1)= (1 + \frac 12 m_3(uu_1))\alpha_{u_1}$,
$\Gamma(u \rightarrow u_3)= (1 + \frac 12 (m_3(uu_1) - 1))\alpha_{u_3}$,
and thus Eq.~(\ref{eq11}) becomes
$\sum_{x\in \{u\}\cup F_3(u)}\omega'(x) = m_3(uu_1) - (\frac 32 \alpha_{u_2} +  (1 + \frac 12 m_3(uu_1))\alpha_{u_1} + (1 + \frac 12 (m_3(uu_1) - 1))\alpha_{u_3})
 = (1 - \frac 12\alpha_{u_1} - \frac 12 \alpha_{u_3})m_3(uu_1) - (\frac 32 \alpha_{u_2} + \alpha_{u_1} + \frac 12 \alpha_{u_3})
 \le (1 - \frac 59)\times m_3(uu_1) - 3\times \frac 59 = -\frac 79 < 0$.

If $m_3(uv) \ge 1$ (and thus $d(f_2)= 3$), then when $m_3(u v)\le 1$, at least one of $\{f_1, f_4\}$ is a $3$-face,
and we discuss the following three subcases for three possible values of $m_3(u)$:
\begin{itemize}
\parskip=0pt
\item
	$d(f_1) = d(f_4) = 3$.
	Then Eq.~(\ref{eq11}) becomes $\sum_{x\in \{u\}\cup F_3(u)}\omega'(x)
	 = 3 - (2\alpha_{u_1} + 2\alpha_{u_2} + \frac 32\alpha_{u_3}) \le 3 - \frac {11} 2 \times \frac 59 = - \frac 1{18}< 0$.
\item
	$d(f_1) = 3$ and $d(f_4)\ge 4$.
	Then Eq.~(\ref{eq11}) becomes $\sum_{x\in \{u\}\cup F_3(u)}\omega'(x)
	 = 2 - (\frac 32 \alpha_{u_1} + 2\alpha_{u_2} + \alpha_{u_3}) \le 2 - \frac {9} 2\times \frac 59 = - \frac 12< 0$.
\item
	$d(f_1) \ge 4$ and $d(f_4) = 3$.
	Then Eq.~(\ref{eq11}) becomes $\sum_{x\in \{u\}\cup F_3(u)}\omega'(x)
	 = 2 - (\frac 32 \alpha_{u_1} + \frac 32 \alpha_{u_2} + \frac 32 \alpha_{u_3}) \le 2 - \frac {9} 2\times \frac 59 < 0$.
\end{itemize}
When $m_3(u v) = 2$, i.e., $d(f_3) = d(f_2) = 3$, it follows from Lemma~\ref{lemma02}(2.4) that for each $x\in \{u_1, u_2, u_3\}$, $d_H(x)\ge 8$ and $d(x)\ge 10$;
and thus, by Lemma~\ref{lemma04properties}(1) and (4), $\alpha_x\ge \frac 35$.
We discuss the following three subcases for three possible values of $m_3(u)$.
\begin{itemize}
\parskip=0pt
\item
	$m_3(u) = 2$, i.e., $d(f_1)\ge 4$ and $d(f_4)\ge 4$.
	Then Eq.~(\ref{eq11}) becomes $\sum_{x\in \{u\}\cup F_3(u)}\omega'(x)
	 = 2 - (\alpha_{u_1} + \frac 32\alpha_{u_2} + \frac 32\alpha_{u_3}) \le 2 - 4\times \frac 35 = - \frac 25< 0$.
\item
	$m_3(u) = 3$, and assuming w.l.o.g. $d(f_1) = 3$ and $d(f_4)\ge 4$.
	Then Eq.~(\ref{eq11}) becomes $\sum_{x\in \{u\}\cup F_3(u)}\omega'(x)
	 = 3 - (\frac 32 \alpha_{u_1} + 2\alpha_{u_2} + \frac 32 \alpha_{u_3}) \le 3 - 5\times \frac 35 = 0$.
\item
	$m_3(u) = 4$, i.e., $d(f_1) = d(f_4) = 3$.
	We can have the following three scenarios for three possible values of $d((f_2)_{u_2v})$ and $d((f_3)_{u_3v})$:
	\begin{itemize}
	\parskip=0pt
	\item
		$d((f_2)_{u_2 v}) \ge 4$ and $d((f_3)_{u_3 v}) \ge 4$.
		Then Eq.~(\ref{eq11}) becomes $\sum_{x\in \{u\}\cup F_3(u)}\omega'(x)
		 = 4 - (2 \alpha_{u_1} + \frac 52 \alpha_{u_2} + \frac 52 \alpha_{u_3}) \le 4 - 7 \times \frac 35 = - \frac 15< 0$.
	\item
		$\min\{d((f_2)_{u_2 v}), d((f_3)_{u_3 v})\} = 3$ and $\max\{d(f_{v w}), d((f_{u v})_{v w^*})\}\ge 4$,
		and assuming w.l.o.g. that $d((f_2)_{u_2 v}) = 3$ and $d((f_3)_{u_3 v}) \ge 4$.
		Then by (R4.1), $\tau(F(uv) \rightarrow v) = \frac 1{10}$, and thus Eq.~(\ref{eq11}) becomes
		$\sum_{x\in \{u\}\cup F_3(u)}\omega'(x)
		 = 4 - (2 \alpha_{u_1} + 2\alpha_{u_2} + \frac 52 \alpha_{u_3} + \frac 1{10}) \le 4 - \frac {13}2 \times \frac 35 - \frac 1{10}= 0$.
	\item
		$d((f_2)_{u_2 v}) = 3$ and $d((f_3)_{u_3 v}) = 3$.
		Then by (R4.2), $\tau(F(u v) \rightarrow v) = \frac 25$, and thus Eq.~(\ref{eq11}) becomes
		$\sum_{x\in \{u\}\cup F_3(u)}\omega'(x) = 4 - (2 \alpha_{u_1} + 2\alpha_{u_2} + 2\alpha_{u_3} + \frac 2{5}) \le 4 - 6 \times \frac 35 - \frac 25 = 0$.
	\end{itemize}
\end{itemize}
This finishes the proof of the lemma.
\end{proof}

\begin{demma}
\label{demma05}
Let $V_{34} = V_3(H)\cup V_4(H)$ and $F_{0}= \bigcup_{v\in V_{34}(H)}F_3(v)$. Then
$\sum_{x\in V_{34}\cup F_{0}}\omega'(x)\le 0. $
\end{demma}
\begin{proof}
It holds naturally by Lemma~\ref{lemma03}, Discharging-Lemmas~\ref{demma03} and \ref{demma04}.
\end{proof}

\begin{demma}
\label{demma06}
For every $3$-face $f = [u v w]\in F(H)$ with $\delta(f)\ge 5$, $\omega'(f) \le 0$.
\end{demma}
\begin{proof}
When $\delta(f)\ge 6$, by (R2), we have for each $x\in V(f)$, $\tau(f \rightarrow v)_f = \frac 13$, and thus
$\omega'(f) = \omega(f) - \sum_{x\in \{u, v, w\}} \tau(f \rightarrow x)_f = 4 - 3 - 3\times \frac 13 = 0$.

When $\delta(f) = 5$, i.e, there exists $x\in \{u, v, w\}$ such that $d(x) = d_H(x) = 5$, and assuming w.l.o.g. $d(u) = d_H(u) = 5$.
If $d_H(v)\ge 6$, $d_H(w)\ge 6$, and $\alpha_v + \alpha_w\ge 1$,
then by (R3.1), (R3.2), and (R5.2), $\tau(f \rightarrow v)_f = \alpha_v$, $\tau(f \rightarrow w)_f = \alpha_w$,
$\tau(u \rightarrow f) = \alpha_v + \alpha_w - 1$,
and thus
$\omega'(f) = \omega(f) + \tau(u \rightarrow f) - \sum_{x\in \{v, w\}} \tau(f \rightarrow x)_f \le 4 - 3 + (\alpha_v + \alpha_w - 1) - (\alpha_v + \alpha_w) = 0$.
Otherwise, by (R3.1), (R3.2) and (R5.1), the total weight transferred from $f$ to $u$, $v$, $w$ is at least $1$,
and thus $\omega'(f) \le \omega(f) - 1 = 4 - 3 - 1 = 0$.
\end{proof}

One sees that for each $5$-vertex $x$ in $H$, $\omega(x) = 4 - 5 = -1$, and $d(x) = d_H(x) = 5$.
It remains to validate that $\omega'(u)\le 0$ for every $u\in V(H)$ with $d_H(u) = 5$.

In the sequel we assume that $d(u) = d_H(u) = 5$.
Before validating $\omega'(u)$, using (R3.1), (R3.2), and (R5.1),
we characterize of the weight $\tau(f \rightarrow u)$ transferred from an incident $3$-face $f = [u v w]$ to $u$ in Discharging-Lemma~\ref{demma07},
also depicted in Figure~\ref{fig04}.

\begin{figure}[]
\begin{center}
\includegraphics[width=5.5 in]{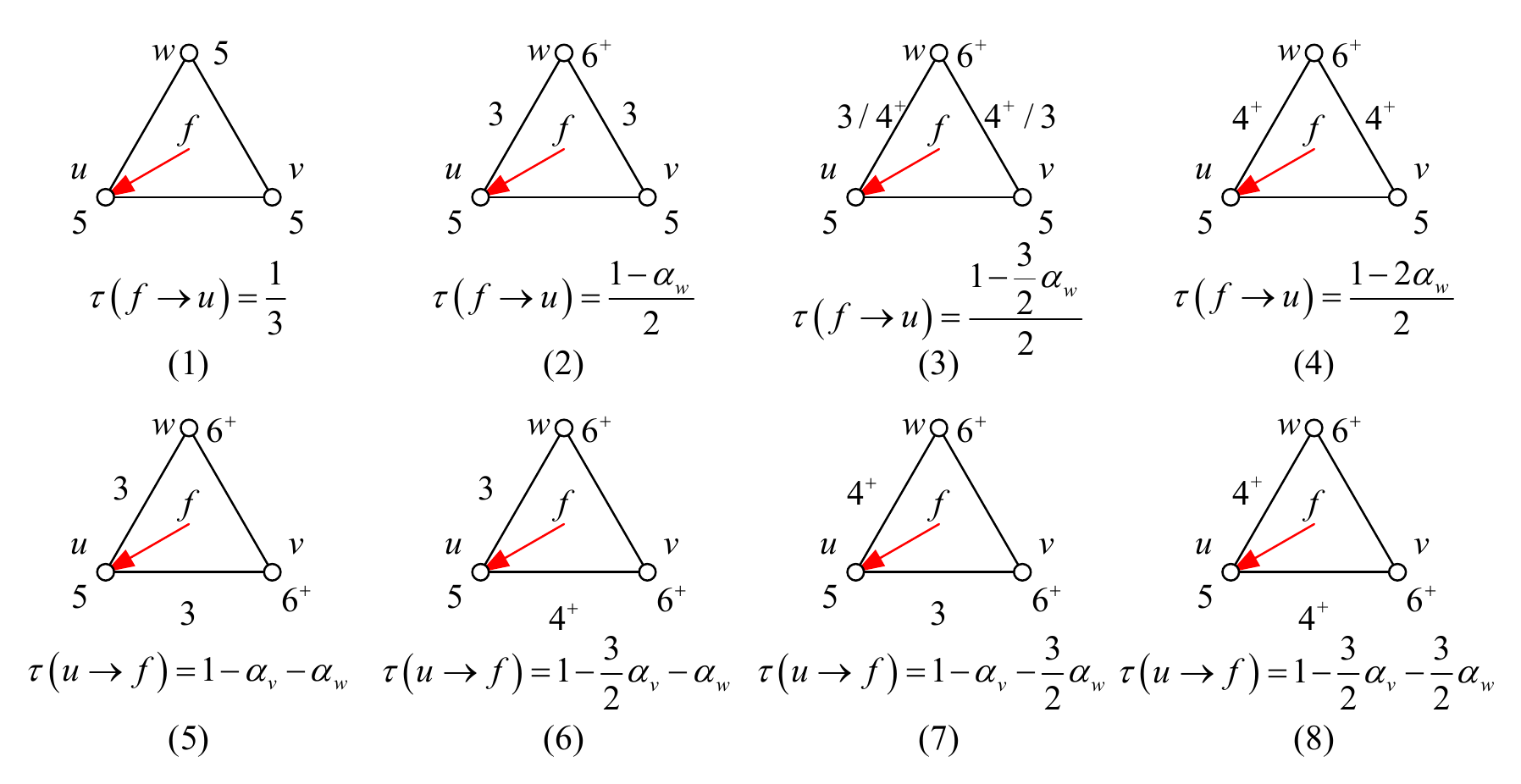}
\caption{The weight $\tau(f \rightarrow u)$ transferred from an incident $3$-face $f = [u v w]$ to $u$.\label{fig04}}
\end{center}
\end{figure}

\begin{demma}
\label{demma07}
For every $5$-vertex $u$ in $V(H)$,
let $f = [u v w]$ be a $3$-face,
if $d_H(v)\ge 6$, $d_H(w)\ge 6$, and $\alpha_v + \alpha_w\ge 1$,
then $\tau(u \rightarrow f) = \alpha_v + \alpha_w - 1$;
otherwise, $\tau(f \rightarrow u)\le \frac 13$, with the following details:
\begin{itemize}
\parskip=0pt
\item[{\rm (1)}]
	if $d_H(v) = d_H(w) = 5$, then $\tau(f \rightarrow u) = \frac 13$;
\item[{\rm (2)}]
	if $d_H(v)= 5$ and $d_H(w)\ge 6$, then $$\tau(f \rightarrow u) =\begin{cases}
		\frac {1 - \alpha_w}2, 				&\mbox{ if } d(f_{uw}) = d(f_{vw}) = 3, \mbox{ and } \alpha_w\le 1;\\
		\frac {1 - \frac 32 \alpha_w}2,	&\mbox{ if } 3\in \{d(f_{uw}), d(f_{vw})\}, \max\{d(f_{uw}), d(f_{vw})\}\ge 4,  \mbox{ and } \alpha_w\le \frac 23;\\
		\frac {1 - 2\alpha_w}2,				&\mbox{ if } d(f_{uw})\ge 4, d(f_{vw})\ge 4, \mbox{ and } \alpha_w\le \frac 12;\end{cases}$$
\item[{\rm (3)}]
	if $d_H(v)\ge 6$, and $d_H(w)\ge 6$, then $$\tau(f \rightarrow u) =\begin{cases}
		1 - \alpha_v - \alpha_w,				&\mbox{ if } d(f_{uv}) = d(f_{uw}) = 3, \mbox{ and } \alpha_v + \alpha_w\le 1;\\
		1 - \frac 32\alpha_v - \alpha_w, &\mbox{ if } d(f_{uv})\ge 4, d(f_{uw}) = 3, \mbox{ and } \frac 32\alpha_v + \alpha_w\le 1;\\
		1 - \alpha_v - \frac 32\alpha_w, &\mbox{ if } d(f_{uv})= 3, d(f_{uw})\ge 4, \mbox{ and } \alpha_v + \frac 32 \alpha_w\le 1;\\
		1 - \frac 32\alpha_v - \frac 32\alpha_w = 0, &\mbox{ if } d(f_{uv})\ge 4, d(f_{uw})\ge 4, \mbox{ and } \alpha_v = \alpha_w = \frac 13.\end{cases}$$
\end{itemize}
\end{demma}
\begin{proof}
When $d_H(v)\ge 6$, $d_H(w)\ge 6$, and $\alpha_v + \alpha_w\ge 1$,
then it follows from (R5.2) clearly that $\tau(u \rightarrow f) = \alpha_v + \alpha_w - 1$.
In the other case, either $\min \{d_H(v), d_H(w)\} = 5$,
or $d_H(v)\ge 6$, $d_H(w)\ge 6$ and $\alpha_v + \alpha_w < 1$.
We assume that after applying (R3), the face $f$ still has charge $\beta \ge 0$.
Using (R3.1), (R3.2), and (R5.1), we discuss the following three subcases for three possible values of $n_5(f)$.

Case 1. $n_5(f) = 3$, i.e, $d_H(v) = d_H(w) = 5$.
In this case, no weight transferred from $f$ to $u$, $v$, or $w$ by (R3.1) and (R3.2).
It follows that $\beta = 1$, and $\tau(f \rightarrow u) = \frac \beta 3 = \frac 13$.

Case 2. $n_5(f) = 2$, and assuming w.l.o.g. $d_H(v)= 5$ and $d_H(w)\ge 6$.
In this case, $\tau(f \rightarrow w)_f = \alpha_w$ and $\tau(f \rightarrow w)_{f_{uw}} \le \frac 12\alpha_w$,
$\tau(f \rightarrow w)_{f_{vw}} \le \frac 12\alpha_w$.
Recall that $\alpha_w\ge \frac 13$.
\begin{itemize}
\parskip=0pt
\item
	$d(f_{uw}) = d(f_{vw}) = 3$.
	It follows that $\beta = 1 - \alpha_w$.
	If $\alpha_w\le 1$, then $\beta\ge 0$ and thus $\tau(f \rightarrow u) = \frac \beta 2 = \frac {1 - \alpha_w}2 \le \frac {1 - \frac 13}2 = \frac 13$.
\item
	$3\in \{d(f_{uw}), d(f_{vw})\}$ and $\max\{d(f_{uw}), d(f_{vw})\}\ge 4$.
	It follows that $\beta = 1 - \alpha_w - \frac 12 \alpha_w = 1 - \frac 32 \alpha_w$.
	If $\frac 32 \alpha_w \le 1$, i.e., $\alpha_w\le \frac 23$,
	then $\beta\ge 0$ and thus $\tau(f \rightarrow u) = \frac \beta 2 = \frac {1 - \frac 32 \alpha_w}2 \le \frac {1 - \frac 32 \times \frac 13}2 = \frac 14$.
\item
	$d(f_{uw})\ge 4$, and $d(f_{vw})\ge 4$.
	It follows that $\beta = 1 - \alpha_w - \frac 12 \alpha_w - \frac 12 \alpha_w = 1 - 2 \alpha_w$.
	If $2 \alpha_w \le 1$, i.e., $\alpha_w\le \frac 12$,
	then $\beta\ge 0$ and thus $\tau(f \rightarrow u) = \frac \beta 2 = \frac {1 - 2 \alpha_w}2 \le \frac {1 - 2\times \frac 13}2 = \frac 16$.
\end{itemize}

Case 3. $n_5(f) = 1$, i.e., $d_H(v)\ge 6$ and $d_H(w)\ge 6$.
In this case, $\tau(f \rightarrow v)_f = \alpha_v$, $\tau(f \rightarrow w)_f = \alpha_w$,
and $\tau(f \rightarrow w)_{f_{uv}} \le \frac 12\alpha_v$,
$\tau(f \rightarrow w)_{f_{uw}} \le \frac 12\alpha_w$.
Recall that $\alpha_v\ge \frac 13$ and $\alpha_w\ge \frac 13$.
\begin{itemize}
\parskip=0pt
\item
	$d(f_{uw}) = d(f_{vw}) = 3$.
	It follows that $\beta = 1 - \alpha_v - \alpha_w$.
	If $\alpha_v + \alpha_w\le 1$, then $\beta\ge 0$ and thus $\tau(f \rightarrow u) = \beta  = 1 - \alpha_v - \alpha_w \le 1 - \frac 13 - \frac 13 = \frac 13$.
\item
	$d(f_{uv})\ge 4$, $d(f_{uw}) = 3$.
	It follows that $\beta = 1 - \frac 32 \alpha_v - \alpha_w$.
	If $\frac 32\alpha_v + \alpha_w \le 1$,
	then $\beta\ge 0$ and thus $\tau(f \rightarrow u) = \beta  = 1 - \frac 32 \alpha_v - \alpha_w \le 1 - \frac 32 \times \frac 13 - \frac 13 = \frac 16$.
\item
	$d(f_{uv}) = 3$, $d(f_{uw}) \ge 4$.
	It follows that $\beta = 1 - \alpha_v - \frac 32 \alpha_w$.
	If $\alpha_v + \frac 32 \alpha_w \le 1$,
	then $\beta\ge 0$ and thus $\tau(f \rightarrow u) = \beta  = 1 - \alpha_v - \frac 32 \alpha_w \le 1 - \frac 13 - \frac 32 \times \frac 13 = \frac 16$.
\item
	$d(f_{uv}) \ge $, $d(f_{uw}) \ge 4$.
	It follows that $\beta = 1 - \frac 32 \alpha_v - \frac 32 \alpha_w$.
	If $\frac 32 \alpha_v - \frac 32 \alpha_w\le 1$, i.e., $\alpha_v = \alpha_w = \frac 13$,
	then $\beta\ge 0$ and thus $\tau(f \rightarrow u) = \beta  = 1 - \frac 32 \alpha_v - \frac 32 \alpha_w = 1 -\frac 32 \times \frac 13 - \frac 32 \times \frac 13 = 0$.
\end{itemize}
This finishes the proof of the lemma that $\tau(f \rightarrow u) \le \frac 13$.
\end{proof}

Discharging-Lemma~\ref{demma07} tells that the value of $\tau(f \rightarrow u)$ depends on $n_5(f)$ and some other factors.
One sees that for any $x$ with $d_H(x) = k\ge 6$, $\alpha_x \ge 1 - \frac 4k$.
Hence, we have the following lemma to upper bound $\tau(f \rightarrow u)$.

\begin{demma}
\label{demma08}
For every $5$-vertex $u$ in $V(H)$, let $f = [u v w]$ be a $3$-face, and then we have the following:
\begin{itemize}
\parskip=0pt
\item[{\rm (1)}]
	If $d_H(v)= 5$ and $d_H(w) = k \ge 6$, then $$\tau(f \rightarrow u) =\begin{cases}
		\frac {1 - \alpha_w}2\le \frac 2k,									&\mbox{ if } d(f_{uw}) = d(f_{vw}) = 3, \mbox{ and } \alpha_w\le 1;\\
		\frac {1 - \frac 32 \alpha_w}2\le \frac 3k - \frac 14,	&\mbox{ if } 3\in \{d(f_{uw}), d(f_{vw})\}, \max\{d(f_{uw}), d(f_{vw})\}\ge 4,  \mbox{ and } k\le 12;\\
		\frac {1 - 2\alpha_w}2\le \frac 4k - \frac 12,				&\mbox{ if } d(f_{uw})\ge 4, d(f_{vw})\ge 4, \mbox{ and } k\le 8.\end{cases}$$
\item[{\rm (2)}]
	If $d_H(v) = k_1\ge 6$ and $d_H(w) = k _2\ge 6$, then $$\tau(f \rightarrow u) =\begin{cases}
		1 - \alpha_v - \alpha_w\le \frac 4{k_1} + \frac 4{k_2} - 1,
				&\mbox{ if } d(f_{uv}) = d(f_{uw}) = 3, \mbox{ and } \frac 4{k_1} + \frac 4{k_2} - 1\ge 0;\\
		1 - \frac 32\alpha_v - \alpha_w\le \frac 6{k_1} + \frac 4{k_2} - \frac 32,
				&\mbox{ if } d(f_{uv})\ge 4, d(f_{uw}) = 3, \mbox{ and } \frac 6{k_1} + \frac 4{k_2} - \frac 32\ge 0;\\
		1 - \alpha_v - \frac 32\alpha_w\le \frac 4{k_1} + \frac 6{k_2} - \frac 32,
				&\mbox{ if } d(f_{uv})= 3, d(f_{uw})\ge 4, \mbox{ and } \frac 4{k_1} + \frac 6{k_2} - \frac 32\ge 0;\\
		1 - \frac 32\alpha_v - \frac 32\alpha_w\le \frac 6{k_1} + \frac 6{k_2} - 2 = 0,
				&\mbox{ if } d(f_{uv})\ge 4, d(f_{uw})\ge 4, \mbox{ and } k_1 = k_2 = 6.\end{cases}$$
\end{itemize}
\end{demma}

\begin{figure}[h]
\begin{center}
\includegraphics[width=0.8\textwidth]{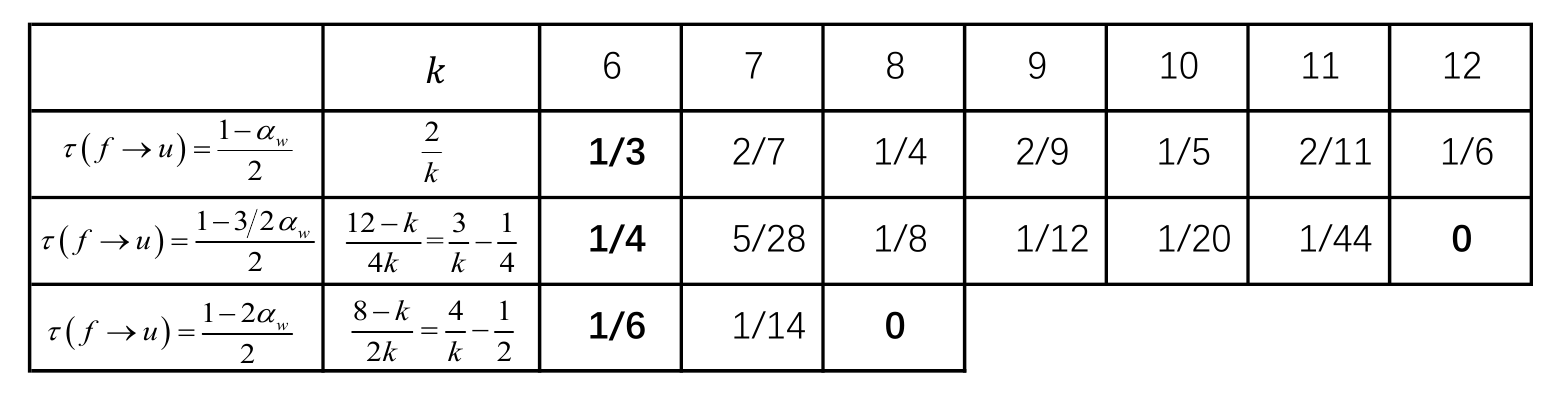}\medskip
\caption{Some special values for the upper bounds $\frac 2k$, $\frac 3k - \frac 14$, and $\frac 4k - \frac 12$, where $k \ge 6$.\label{fig05}}
\end{center}
\end{figure}

\begin{figure}[h]
\begin{center}
\includegraphics[width=0.8\textwidth]{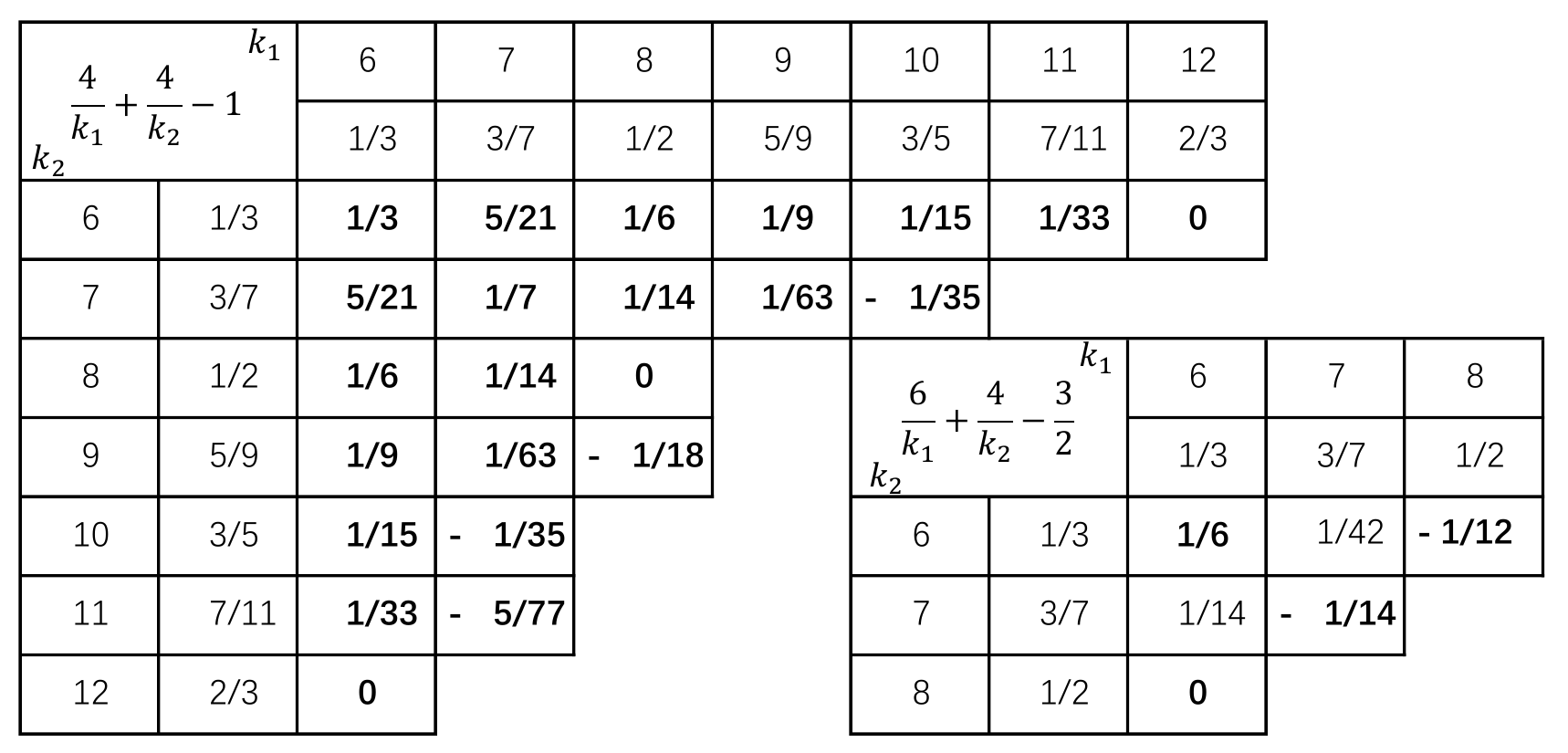}
\caption{Some special values for the upper bounds $\frac 4{k_1} + \frac 4{k_2} - 1$ and $\frac 6{k_1} + \frac 4{k_2} - 1$, where $k_1, k_2 \ge 6$.\label{fig06}}
\end{center}
\end{figure}

For convenience, for some small values of $k, k_1, k_2 \ge 6$,
we list the upper bounds stated in Discharging-Lemma~\ref{demma08} in the two tables depicted in Figure~\ref{fig05} and Figure~\ref{fig06}.
In the sequel, we will directly refer to these values.
Let $u_1, u_2, u_3, u_4, u_5$ denote all the neighbors of $u$ in clockwise order,
and denote the face containing $u_1 u u_2$ ($u_2 u u_3$, $u_3 u u_4$, $u_4 u u_5$, $u_5 u u_1$, respectively) as $f_1$ ($f_2, f_3, f_4$, $f_5$ respectively).
Recall that for each $x\in N_H(u)$, $d_H(x)\ge 4$.

For a $3$-face $f = [v u w]$ incident to $u$ with $d_H(v)\ge 6$, $d_H(w)\ge 6$, and $\alpha_v + \alpha_w\ge 1$,
by (R5.2), $\tau(u \rightarrow f) = \alpha_v + \alpha_w - 1\ge 0$.
Then $\tau(f \rightarrow u) = - \tau(u \rightarrow f)\le 0$.
Hence,
\begin{equation}
\label{eq12}
\begin{aligned}
\omega'(u) = \omega(u) + \sum_{f\in F_3(v)} \tau(f \rightarrow u)
           = 4 - 5 + \sum_{f\in F_3(v)} \tau(f \rightarrow u)
           = -1 + \sum_{f\in F_3(v)} \tau(f \rightarrow u).
\end{aligned}
\end{equation}
We validate that $\omega'(u) \le 0$ for each $5$-vertex $u$, using Eq.~(\ref{eq12}), in the following two lemmas.

\begin{demma}
\label{demma09}
For every $5$-vertex $u$ adjacent with a $4$-vertex $u_1$ in $V(H)$, $\omega'(u) \le 0$.
\end{demma}
\begin{proof}
One sees that for any $4$-vertex $x\in N_H(u)$, $d(x) = d_H(x) = 4$ and
the average weight transferred from each $3$-face in $F(ux)$ to $u$ is at most $\max\{\frac 12\times \frac 1{10}, \frac 12\times \frac 25\} = \frac 15$.
Thus, together with Discharging-Lemma~\ref{demma07},
the average weight transferred from each $3$-face in $F_3(u)$ to $u$ is at most $\frac 13$.

When $m_3(u_1) \ne 4$, or $m_3(u_1) = 4$ and $\min\{d(f_2), d(f_4)\}\ge 4$, no weight is transferred from any face in $F(uu_1)$ to $u$;
since the average weight transferred from each $3$-face in $\{f_2, f_3, f_4\}$ is at most $\frac 13$, Eq.~(\ref{eq12}) becomes
$\omega'(u) = -1 + \sum_{f\in \{f_2, f_3, f_4\}} \tau(f \rightarrow u) \le -1 + 3\times \frac 13 = 0$.

When $m_3(u_1) = 4$ and $\min\{d(f_2), d(f_4)\} = 3$, it follows from Lemma~\ref{lemma02} (2.4) that for each $x\in \{u_2, u_5\}$, $d_H(x)\ge 8$ and $d(x)\ge 10$;
thus, by Lemma~\ref{lemma04properties}(1) and (4), $\alpha_x\ge \frac 35$.
Using (R3), (R4), (R5) and Discharging-Lemma~\ref{demma07}, we discuss the following three subcases on $d(f_i)$, $i\in \{2, 3, 4\}$.
\begin{itemize}
\parskip=0pt
\item[{\rm (1)}]
	$\max\{d(f_2), d(f_4)\}\ge 4$, and assuming w.l.o.g. $d(f_2) = 3$ and $d(f_4)\ge 4$.
	In this case, $\tau(F(uu_1) \rightarrow u) = \frac 1{10}$.
	Since the average weight transferred from each $3$-face in $\{f_2, f_3\}$ is at most $\frac 13$, Eq.~(\ref{eq12}) becomes
	$\omega'(u) = -1 + \sum_{f\in \{f_1, f_5\}} \tau(f \rightarrow u) \le -1 + \frac 1{10} + 2\times \frac 13 = - \frac 7{30}< 0$.
\item[{\rm (2)}]
	$d(f_2) = d(f_4) = 3$ and $d(f_3)\ge 4$.
	In this case, $\tau(F(uu_1) \rightarrow u) = \frac 25$.
	One sees that if $d_H(u_3) = 4$, no weight will be transferred from $f_2$ to $u$;
	if $d_H(u_3)= 5$, then $\tau(f_2 \rightarrow u)\le \frac {1 - \alpha_{u_2}} 2 \le \frac {1 - \frac 35} 2 = \frac 15$;
	if $d_H(u_3)\ge 6$, since $\frac 32 \times \alpha_{u_3} + \alpha_{u_2}\ge \frac 32 \times \frac 13 + \alpha_{u_2}\ge  \frac 12 + \frac 35 = \frac {11}{10}$,
	then we have $\tau(f_2 \rightarrow u)\le 0$.
	It follows that $\tau(f_2 \rightarrow u)\le \frac 15$; and likewise, $\tau(f_4 \rightarrow u)\le \frac 15$.
	Hence, Eq.~(\ref{eq12}) becomes
	$\omega'(u) = -1 + \sum_{f\in \{f_1, f_5\}} \tau(f \rightarrow u) + \sum_{f\in \{f_2, f_4\}} \tau(f \rightarrow u) \le -1 + \frac 25 + 2\times \frac 15= - \frac 15< 0$.
\item[{\rm (3)}]
	$d(f_3) = d(f_2) = d(f_4) = 3$.
	In this case, $\tau(F(uu_1) \rightarrow u) = \frac 25$.
	\begin{itemize}
	\parskip=0pt
	\item
		$\min\{d_H(u_3), d_H(u_4)\} = 4$, and assuming w.l.o.g. $d(u_3) = d_H(u_3) = 4$.
		It follows that $d_H(u_4)\ge 8$ and $d(u_4)\ge 10$,
		and then $\tau(F(uu_3) \rightarrow u) \le \frac 25$ and $\tau(u \rightarrow f_4) = \alpha_{u_5} + \alpha_{u_4} - 1\ge 2\times \frac 35 - 1 = \frac 15$,
		i.e., $\tau(f_4 \rightarrow u) \le -\frac 15$.
     Hence, Eq.~(\ref{eq12}) becomes
		$\omega'(u) = -1 + \sum_{f\in \{f_1, f_5\}} \tau(f \rightarrow u) + \sum_{f\in \{f_2, f_3\}} \tau(f \rightarrow u) + \tau(f_4 \rightarrow u)
			\le - 1 + 2\times \frac 25  - \frac 15 = -\frac 25 < 0$.
	\item
		$\min\{d_H(u_3), d_H(u_4)\} = 5$, and assuming w.l.o.g. $d(u_3) = d_H(u_3) = 5$.
		The non-existence of ($A_{7.1}$) in $G$ states that $d_H(u_4)\ge 7$.
		Then we have $\tau(f_2 \rightarrow u)\le \frac {1 - \alpha_{u_2}}2 \le  \frac {1 - \frac 35}2 = \frac 15$,
		$\tau(f_3 \rightarrow u)\le \frac {1 - \alpha_{u_4}}2 \le  \frac {1 - \frac 37}2 = \frac 27$,
		and since $\alpha_{u_5} + \alpha_{u_4}\ge \frac 35 + \frac 37 = \frac {36}{35}$,
		$\tau(u \rightarrow f_4) = \alpha_{u_5} + \alpha_{u_4} - 1 \ge \frac {36}{35} - 1 = \frac 1{35}$, i.e., $\tau(f_4 \rightarrow u)\le - \frac 1{35}$.
     Hence, Eq.~(\ref{eq12}) becomes
		$\omega'(u) = -1 + \sum_{f\in \{f_1, f_5\}} \tau(f \rightarrow u) + \tau(f_2 \rightarrow u) + \tau(f_3 \rightarrow u) + \tau(f_4 \rightarrow u)
			\le - 1 + \frac 25  + \frac 15 + \frac 27 - \frac 1{35}= -\frac 17 < 0$.
	\item
		$d_H(u_3) \ge 6$ and $d_H(u_4) \ge 6$.
		Then we have $\tau(f_2 \rightarrow u)= 1 - \alpha_{u_2} - \alpha_{u_3} \le 1 - \frac 35 - \frac 13 = \frac 1{15}$,
		$\tau(f_3 \rightarrow u)\le \frac 13$, and $\tau(f_4 \rightarrow u)= 1 - \alpha_{u_5} - \alpha_{u_4} \le 1 - \frac 35 - \frac 13 = \frac 1{15}$.
     Hence, Eq.~(\ref{eq12}) becomes
		$\omega'(u) = -1 + \sum_{f\in \{f_1, f_5\}} \tau(f \rightarrow u) + \tau(f_2 \rightarrow u) + \tau(f_3 \rightarrow u) + \tau(f_4 \rightarrow u)
			\le - 1 + \frac 25  + \frac 1{15} + \frac 13 + \frac 1{15} = -\frac 2{15} < 0$.
	\end{itemize}
\end{itemize}
This proves the lemma.
\end{proof}

\begin{demma}
\label{demma10}
For every $5$-vertex $u$ in $V(H)$, if $d_H(x)\ge 5$ for each $x\in N_H(u)$, then $\omega'(u) \le 0$.
\end{demma}
\begin{proof}
It follows from Discharging-Lemma~\ref{demma07} that for each face $f\in F_3(u)$, $\tau(f \rightarrow u)\le \frac 13$.
When $m_3(f)\le 3$, Eq.~(\ref{eq12}) becomes $\omega'(u) \le -1 + 3\times \frac 13 = 0$.

When there exist two faces $f = [\ldots v_1 u w_1 \ldots\}$ and $h = [\ldots v_2 u w_2]$ such that for any $x\in \{v_1, w_1\}\cup \{v_2, w_2\}$, $\alpha_x\ge \frac 12$,
one sees that if $d(f)$ = 3, then $\tau(u \rightarrow f) = \alpha_{v_1} + \alpha_{w_1} - 1 \ge 2\times \frac 12 - 1\ge 0$, i.e., $\tau(f \rightarrow u)\le 0$.
It follows that $\tau(f \rightarrow u)\le 0$, and likewise $\tau(g \rightarrow u)\le 0$, and Eq.~(\ref{eq12}) becomes
$\omega'(u) = -1 + \sum_{f\in F_3(v)\setminus \{f, h\}} \tau(f \rightarrow u)\le - 1 + 3\times \frac 13 = 0$.

If there exist $v, w\in N_H(u)$ such that $d_H(v)= 5$ and $d_H(w)\le 6$, i.e., $d(v) = 5$ and $d(w)\le 6$,
$f = [vuw]$ is a $3$-face and $m_3(uv) = 2$,
then the non-existence of ($2.1$), ($6.3$) and ($7.1$) in $G$ imply that for each $x\in N_H(u)\setminus \{v, w\}$, $d(x)\ge 8$,
and since $d(u)= 5$, we have $d_H(x)\ge 8$ and thus $\alpha_x\ge \frac 12$.

Recall that if $d_H(x)\ge 8$ or $x$ is $(7, 5^+)_H$, then $\alpha_x\ge \frac 12$;
if $d_H(x)\le 6$, then $d(x) = d_H(x)\le 6$;
and if $x$ is $(7, 4^-)_H$, then $x$ is $(7, 4^-)_{HG}$.
In the other case, $4\le m_3(f)\le 5$.
In the sequel, we continue the proof with the following propositions which summarize the above argument:
\begin{itemize}
\parskip=0pt
\item[{\rm (I)}]
    For each $v, w\in N_H(x)$ with $vw\in E(H)$ and $m_3(uv) = 2$, if $d_H(v)= 5$, i.e., $d(v) = 5$, then $d_H(w)\ge 7$ and $\alpha_{w}\ge \frac 37$; if $d_H(w)\le 6$, i.e., $d(w) \le 6$, then $d_H(v)\ge 6$ and $\alpha_{v}\ge \frac 13$;
\item[{\rm (II)}]
    If $\alpha_x < \frac 12$ for some $x$, then $d(x) = d_H(x)\le 6$ or $x$ is $(7, 4^-)_{HG}$.
\end{itemize}

Assume w.l.o.g. $d(f_i) = 3$, for any $i = 1, 2, 3, 4$, and $d(f_5)\ge 3$.
Using (R3.1), (R3.2), (R5.1), and (R5.2), we discuss the following two cases for two possible values of $m_3(u)$.

{Case 1.}  $m_3(u) = 4$, i.e., $d(f_5)\ge 4$.
In this case, Eq.~(\ref{eq12}) becomes
\begin{equation}
\label{eq13}
\omega'(u) = -1 + \sum_{f\in \{f_1, f_2, f_3, f_4\}} \tau(f \rightarrow u).
\end{equation}
We can have the following three scenarios for three possible values of $d_H(u_1)$ and $d_H(u_5)$.

{Case 1.1.} $d_H(u_1)\ge 6$ and $d_H(u_5)\ge 6$.

(Case 1.1.1.) $d_H(u_2)\ge 6$ and $d_H(u_4)\ge 6$.
In this case, by Discharging-Lemmas~\ref{demma07} and \ref{demma08}%
\footnote{We do not calculate $\tau(f \rightarrow u)$ in details but refer to its upper bounds in the two tables in Figures~\ref{fig05} and \ref{fig06}},
$\tau(f_1 \rightarrow u) = 1 - \frac 32\alpha_{u_1} - \alpha_{u_2} = 1 - \frac 32 \times \frac 13 - \frac 13 = 1 - \frac 56 = \frac 16 $,
and likewise $\tau(f_4\rightarrow u)\le \frac 16$,
and thus Eq.~(\ref{eq13}) becomes $\omega'(u) \le - 1 + 2\times \frac 16 + 2\times \frac 13 = 0$.

(Case 1.1.2.) $d_H(u_2) = d_H(u_4) = 5$, i.e.,  $d(u_2) = d(u_4) = 5$.
It follows from (I) that for each $i\in 1, 3, 5$, $\alpha_{u_i}\ge \frac 37$.
Then $\tau(f_1 \rightarrow u) = \frac {1 - \frac 32 \alpha_{u_1}} 2\le \frac {1 - \frac 32 \times \frac 37} 2\le \frac 5{28}$,
and $\tau(f_2\rightarrow u)\le \frac {1 - \alpha_{u_3}}2\le \frac {1 - \frac 37}2 = \frac 27$.
Likewise, $\tau(f_4 \rightarrow u) \le \frac 5{28}$, and $\tau(f_3\rightarrow u)\le \frac 27$.
Thus, Eq.~(\ref{eq13}) becomes $\omega'(u) \le - 1 + 2\times \frac 5{28} + 2\times \frac 27 = -1 + \frac {13}{14} = -\frac 1{14}$.

(Case 1.1.3.) Assuming w.l.o.g. $d_H(u_2) = 5$ and $d_H(u_4)\ge 6$, i.e.,  $d(u_2)= 5$.
It follows from (I) that for each $i\in \{1, 3\}$, $\alpha_{u_i}\ge \frac 37$.
Then $\tau(f_1 \rightarrow u)\le \frac 5{28}$,
$\tau(f_2\rightarrow u)\le \frac {1 - \alpha_{u_3}}2\le \frac {1 - \frac 37}2 = \frac 27$,
$\tau(f_3\rightarrow u) = 1 - \alpha_{u_3} - \alpha_{u_4}\le 1 - \frac 13 - \frac 37 = \frac 5{21}$,
and $\tau(f_4\rightarrow u)\le \frac 16$.
Thus, Eq.~(\ref{eq13}) becomes $\omega'(u) \le - 1 + \frac 5{28} + \frac 27 + \frac 5{21} + \frac 16 = - 1 + \frac {73}{84} = - \frac {11}{84}< 0$.

{Case 1.2.} $d_H(u_1) = d_H(u_5) = 5$, i.e., $d(u_1) = d(u_5) = 5$.
In this case, by (I), for each $i\in \{2, 4\}$, $d(u_i)\ge 6$ and $d_H(u_i)\ge 6$, $\alpha_{u_i}\ge \frac 13$.
When $\alpha_{u_3}\ge \frac 12$, $\tau(f_2\rightarrow u) = 1 - \alpha_{u_2} - \alpha_{u_3} \le \frac 16$,
and likewise $\tau(f_3\rightarrow u)\le \frac 12$, and thus Eq.~(\ref{eq13}) becomes
$\omega'(u) \le -1 + 2\times \frac 16 + 2\times \frac 13 = - 1 + \frac 56 = -\frac 16< 0$.

In the other case, $\alpha_{u_3}< \frac 12$.
It follows from (II) that $d(u_3) = d_H(u_3)\le 6$ or $u_3$ is $(7,4^-)_{HG}$.

(Case 1.2.1.) $d(u_3)= 5$.
For each $i\in \{2, 4\}$, it follows from (I) that $d_H(u_i)\ge 7$;
the non-existence of ($A_{7.2}$) in $G$ implies that $d(u_i)\ge 8$;
and since $n_{5}(u_i)\ge 3$, we have $d_H(u_i)\ge 8$.
It follows that $\alpha_{u_2}\ge \frac 12$ and $\alpha_{u_4}\ge \frac 12$.
Then for each $i\in \{1, 2, 3, 4\}$, $\tau(f_i \rightarrow u)\le \frac 14$, and Eq.~(\ref{eq13}) becomes $\omega'(u) \le - 1 + 4\times \frac 14 = 0$.

(Case 1.2.2.) $d(u_3)= 6$.
When for each $i\in \{2, 4\}$, $\alpha_{u_i}\ge \frac 12$,
we have $\tau(f_1 \rightarrow u)\le \frac 14$, $\tau(f_2\rightarrow u)\le \frac 16$,
$\tau(f_3\rightarrow u)\le \frac 16$, and $\tau(f_4\rightarrow u)\le \frac 14$.
Thus, Eq.~(\ref{eq13}) becomes $\omega'(u) \le - 1 + 2\times \frac 14 + 2\times \frac 16 = - \frac 16< 0$.

In the other case, there exists an $i\in \{2, 4\}$ such that $\alpha_{u_i}< \frac 12$.
It follows that $d(u_i) = d_H(u_i) = 6$, or $u_i$ is $(7, 4^-)_{HG}$.
Assuming w.l.o.g. $i = 2$, one sees that if $d(u_2) = 6$ or $d(u_4) = 6$, then ($A_{8.1}$) will be obtained, a contradiction.
Hence, $u_1$ is $(7, 4^-)_{HG}$ and $d_H(u_4)\ge 7$.
If $\alpha_{u_4}\ge \frac 12$, then $\tau(f_1 \rightarrow u)\le \frac 27$, $\tau(f_2\rightarrow u)\le \frac 5{21}$,
$\tau(f_3\rightarrow u)\le \frac 16$, and $\tau(f_4\rightarrow u)\le \frac 14$.
Thus, Eq.~(\ref{eq13}) becomes $\omega'(u) \le - 1 + \frac 27 + \frac 5{21} + \frac 16 + \frac 14 = - 1 + \frac{79}{84}  = - \frac 5{84} < 0$.
Otherwise, $\alpha_{u_4}< \frac 12$.
It follows from (II) that $d(u_4) = d_H(u_4)\le 6$ or $u_4$ is $(7, 4^-)_{HG}$, which leads to ($A_{8.1}$) or ($A_{8.4}$), a contradiction.

(Case 1.2.3.) $u_3$ is $(7, 4^-)_{HG}$.
When there exists an $i\in \{2, 4\}$ such that $d_H(u_i)\ge 7$, and assuming w.l.o.g. $i = 4$,
then $\tau(f_1 \rightarrow u)\le \frac 13$, $\tau(f_2\rightarrow u)\le \frac 5{21}$,
$\tau(f_3\rightarrow u) \frac 17$, and $\tau(f_4\rightarrow u)\le \frac 27$.
Thus, Eq.~(\ref{eq13}) becomes $\omega'(u) \le - 1 + \frac 13 + \frac 5{21} + \frac 17  + \frac 27 = 0$.
In the other case, for each $i\in \{2, 4\}$, $d(u_i) = d_H(u_i) = 6$, which leads to ($A_{8.2}$), a contradiction.

{Case 1.3.} Assuming w.l.o.g. $d(u_1) = d_H(u_1) = 5$ and $d_H(u_5)\ge 6$.
In this case, by (I), $d(u_2)\ge 6$ and $\tau(f_1\rightarrow u)\le \frac 13$.

(Case 1.3.1.) $d_H(u_4)= 5$, i.e., $d(u_4) = 5$.
It follows from (I) that for each $i\in \{3, 5\}$, $d_H(u_i)\ge 7$ and thus $\alpha_{u_i}\ge \frac 37$.
When $\alpha_{u_3}\ge \frac 12$,
$\tau(f_2\rightarrow u)\le \frac 16$,
$\tau(f_3\rightarrow u)\le \frac 14$,
$\tau(f_4\rightarrow u)\le \frac 5{28}$,
and thus Eq.~(\ref{eq13}) becomes $\omega'(u) \le - 1 + \frac 13 + \frac 16 + \frac 14 + \frac 5{28} = -\frac {13}{84}< 0$.

When $\alpha_{u_5}\ge \frac 12$,
$\tau(f_2\rightarrow u)\le \frac 5{21}$,
$\tau(f_3\rightarrow u)\le \frac 27$,
$\tau(f_4\rightarrow u)\le \frac 18$,
and thus Eq.~(\ref{eq13}) becomes $\omega'(u) \le - 1 + \frac 13 + \frac 5{21} + \frac 27 + \frac 18 = -\frac 3{168}< 0$.

In the other case, for each $i\in \{3, 5\}$, $\alpha_{u_i} < \frac 12$;
and it follows from (2) that $u_i$ is $(7, 4^-)_{HG}$, which leads to ($A_{8.4}$), a contradiction.

(Case 1.3.2.) $d_H(u_4)\ge 6$ and $d_H(u_2)\ge 6$.
When $d_H(u_3)\ge 7$,
$\tau(f_2\rightarrow u)\le \frac 5{21}$,
$\tau(f_3\rightarrow u)\le \frac 5{21}$,
$\tau(f_4\rightarrow u)\le \frac 16$,
and thus Eq.~(\ref{eq13}) becomes $\omega'(u) \le - 1 + \frac 13 + 2\times \frac 5{21} + \frac 16 = - \frac 1{42}< 0$.

In the other case, $d(u_3) = d_H(u_3)\le 6$.
First, assume that $d(u_3) = 5$.
It follows from (I) that for each $i\in \{2, 4\}$, $\alpha_{u_i}\ge \frac 37$.
Then for each $i\in \{1, 2, 3\}$, $\tau(f_i\rightarrow u)\le \frac 27$,
$\tau(f_4\rightarrow u)\le \frac 1{14}$,
and thus Eq.~(\ref{eq13}) becomes $\omega'(u) \le - 1 + 3\times \frac 27 + \frac 1{14} = - \frac 1{14}< 0$.
Next, assume that $d(u_3) = 6$.
When $\alpha_{u_5}\ge \frac 12$,
$\tau(f_4\rightarrow u_4)_{f_4}\ge \frac 13$,
$\tau(f_4\rightarrow u_5)_{f_4}\ge \frac 12$,
$\tau(f_4\rightarrow u_5)_{f_1}\ge \frac 14$,
it follows that $\tau(f_4\rightarrow u)= 0$,
and thus Eq.~(\ref{eq13}) becomes $\omega'(u) \le - 1 + 3\times \frac 13 = 0$.
In the other case, $\alpha_{u_5}<\frac 12$.
It follows from (II) that $d(u_5) = d_H(u_5) = 6$, or $u_5$ is $(7, 4^-)_{HG}$.
\begin{itemize}
\parskip=0pt
\item
	$d(u_5) = 6$.
	It follows from (I) that $d_H(u_4)\ge 6$.
	If $d_H(u_4) = 6$, then $d(u_4) = 6$ and ($A_{8.1}$) is obtained, a contradiction.
	Otherwise, $d_H(u_4) \ge 7$ and $\alpha_{u_4}\ge \frac 37$.
	Then, $\tau(f_3\rightarrow u)\le \frac 5{21}$, $\tau(f_4\rightarrow u)\le \frac 1{14}$,
	and thus Eq.~(\ref{eq13}) becomes $\omega'(u) \le - 1 + 2\times \frac 13 + \frac 5{21} + \frac 1{14}  = - \frac 1{42}< 0$.
\item
	$u_5$ is $(7, 4^-)$.
	When $d_H(u_4)\ge 7$,
	$\tau(f_4\rightarrow u_4)_{f_4}\ge \frac 37$,
	$\tau(f_4\rightarrow u_5)_{f_4}\ge \frac 37$,
	$\tau(f_4\rightarrow u_5)_{f_1}\ge \frac 3{14}$,
	it follows that $\tau(f_4\rightarrow u)= 0$, and thus Eq.~(\ref{eq13}) becomes $\omega'(u) \le - 1 + 3\times \frac 13 = 0$.
	In the other case, $d_H(u_4)= 6$, i.e., $d(u_4) = d_H(u_4) = 6$, which leads to ($A_{8.2}$), a contradiction.
\end{itemize}

{Case 2.}  $m_3(u) = 5$.
In this case, Eq.~(\ref{eq12}) becomes
\begin{equation}
\label{eq14}
\omega'(u) = -1 + \sum_{f\in \{f_1, f_2, f_3, f_4, f_5\}} \tau(f \rightarrow u).
\end{equation}

When $d_H(x)\ge 7$ for each $x\in N_H(u)$, $\alpha_x\ge \frac 37$,
for each $i\in [1, 5]$, $\tau(f_i\rightarrow u)\le \frac 17$,
and thus Eq.~(\ref{eq14}) becomes $\omega'(u) \le - 1 + 5\times \frac 17 = -\frac 27$.

In the other case, there exists an $x\in N_H(x)$ such that $d_H(x)\le 6$ and $d(x) = d_H(x)\le 6$.
We have the following two subcases for two possible values of $\min_{x\in N_H(x)} d_H(x)$.

{Case 2.1.} $d(u_3)= d_H(u_3) = 5$.
In this case, it follows from (I) that $d_H(u_2)\ge 7$, $d_H(u_4)\ge 7$, and $n_5(u)\le 2$.
We can have the following two scenarios:

(Case 2.1.1.) $n_5(u) = 2$, and assuming w.l.o.g. $d(u_1) = d_H(u_1)= 5$.
It follows from (I) that $d_H(u_5)\ge 7$.
When $\alpha_{u_i}\ge \frac 12$ for each $i\in \{2, 4, 5\}$,
$\tau(f_i \rightarrow u)\le \frac 14$, $i\in \{1, 2, 3, 5\}$,
and $\tau(f_4 \rightarrow u)\le 0$ since $\tau(f_4 \rightarrow u_4)_{f_4}\ge \frac 12$ and $\tau(f_4 \rightarrow u_5)_{f_4}\ge \frac 12$.
Thus, Eq.~(\ref{eq14}) becomes $\omega'(u) \le - 1 + 4\times \frac 14 = 0$.

In the other case, there exists a $w\in \{u_2, u_4, u_5\}$ such that $\alpha_{w}< \frac 12$.
It follows from (II) that $w$ is $(7, 4^-)_{HG}$.
The non-existence of ($A_{7.2}$) in $G$ implies that for each $x\in \{u_2, u_4, u_5\}\setminus \{w\}$, $d(x)\ge 9$,
and since $n'_{5}(x)\ge 2$, $d_H(x)\ge 9$ and $\alpha_x\ge \frac 59$.
By symmetry, we assume w.l.o.g. $w\in \{u_2, u_4\}$.
\begin{itemize}
\parskip=0pt
\item
	$u_2$ is $(7, 4^-)_{HG}$, $d_H(u_4)\ge 9$ and $d_H(u_5)\ge 9$.
	Then $\tau(f_i \rightarrow u)\le \frac 27$, $i = 1, 2$;
	$\tau(f_j \rightarrow u)\le \frac 29$, $j = 3, 5$;
	since $\alpha_{u_4}\ge \frac 59$ and $\alpha_{u_5}\ge \frac 59$,
	we have $\tau(u \rightarrow f_4) = \alpha_{u_4} + \alpha_{u_5} - 1 \ge 2\times 59 -1 = \frac 19$, i.e., $\tau(f_4 \rightarrow u)\le -\frac 19$.
	Thus, Eq.~(\ref{eq14}) becomes $\omega'(u) \le - 1 + 2\times \frac 27 + 2\times \frac 29 - \frac 19 = -\frac 2{21}< 0$.
\item
	$u_4$ is $(7, 4^-)_{HG}$, $d_H(u_2)\ge 9$ and $d_H(u_5)\ge 9$.
	Then $\tau(f_i \rightarrow u)\le \frac 29$, $i \in \{1, 2, 5\}$;
	$\tau(f_3 \rightarrow u)\le \frac 72$, and $\tau(f_4 \rightarrow u)\le \frac 1{63}$.
	Thus, Eq.~(\ref{eq14}) becomes $\omega'(u) \le - 1 + 3\times \frac 29 + \frac 27 + \frac 1{63} = -\frac 2{63}< 0$.
\end{itemize}

(Case 2.1.2.) $n_5(u) = 1$, i.e., $d_H(u_1)\ge 6$ and $d_H(u_5)\ge 6$.
When $n'_{7^+}(u) = 4$,
$\tau(f_i \rightarrow u)\le \frac 17$, $i\in \{1, 4, 5\}$,
$\tau(f_j \rightarrow u)\le \frac 27$, $j\in \{2, 3\}$,
and thus Eq.~(\ref{eq14}) becomes $\omega'(u) \le - 1 + 3\times \frac 17 + 2\times \frac 27 = 0$.

In the other case, $n'_{7^+}(u) \le 3$, i.e., there exists a $w\in \{u_1, u_5\}$ such that $d_H(w) = 6$ and then $d(w) = 6$.
Assume w.l.o.g. that $w = u_1$.
If $\alpha_{u_5}\ge \frac 59$,
then for each $i\in \{2, 3\}$, $\tau(f_1 \rightarrow u)\le \frac 5{21}$,
$\tau(f_i \rightarrow u)\le \frac 27$,
$\tau(f_4 \rightarrow u)\le \frac 1{63}$,
$\tau(f_5 \rightarrow u)\le \frac 19$,
and thus Eq.~(\ref{eq14}) becomes $\omega'(u) \le - 1 + \frac 5{21} + 2\times \frac 27 + \frac 1{63} + \frac 19 = -\frac 4{63}< 0$.
Otherwise, $\alpha_{u_5}< \frac 59$.
Hence, $d(u_5) = d_H(u_5) = 6$, or $u_5$ is $(7, 5^-)_{HG}$, or $u_5$ is $(8, 4^-)_{HG}$.
\begin{itemize}
\parskip=0pt
\item
	$d(u_5) = d_H(u_5) = 6$.
	The non-existence of ($A_{7.2}$) and ($A_{7.3}$) in $G$ imply that for each $i\in \{2, 4\}$, $d(u_i)\ge 9$,
	and since $n'_{5}(u_i) = n_{5}(u_i)\ge 2$, $d_H(u_i)\ge 9$.
	It follows that $\alpha_{u_2}\ge \frac 59$ and $\alpha_{u_4}\ge \frac 59$.
	Then $\tau(f_i \rightarrow u)\le \frac 19$, $i = 1, 4$,
	$\tau(f_j \rightarrow u)\le \frac 29$, $j = 2, 3$,
	$\tau(f_5 \rightarrow u)\le \frac 13$,
	and thus Eq.~(\ref{eq14}) becomes $\omega'(u) \le - 1 + 2\times \frac 19 + 2\times \frac 29 + \frac 13 = 0$.
\item
	$u_5$ is $(7, 5^-)_{HG}$.
	The non-existence of ($A_{7.2}$) in $G$ implies that for each $i\in \{2, 4\}$,
	$d(u_i)\ge 8$ and since $n'_{5}(u_i) = n_{5}(u_i)\ge 2$,  $d_H(u_i)\ge 8$.
	It follows that $\alpha_{u_2}\ge \frac 12$ and $\alpha_{u_4}\ge \frac 12$.
	Then $\tau(f_1 \rightarrow u)\le \frac 16$,
	$\tau(f_i \rightarrow u)\le \frac 14$, $i = 2, 3$,
	$\tau(f_4 \rightarrow u)\le \frac 1{14}$,
	$\tau(f_5 \rightarrow u)\le \frac 5{21}$,
	and thus Eq.~(\ref{eq14}) becomes $\omega'(u) \le - 1 + \frac 16 + 2\times \frac 14 + \frac 1{14}+ \frac 5{21} = \frac 1{42}< 0$.
\item
	$u_5$ is $(8, 4^-)_{HG}$.
	Then $\tau(f_5 \rightarrow u)\le \frac 16$.
	When $\alpha_{u_2}\ge \frac 12$,
	$\tau(f_1 \rightarrow u)\le \frac 16$,
	$\tau(f_2 \rightarrow u)\le \frac 14$,
	$\tau(f_3 \rightarrow u)\le \frac 27$,
	$\tau(f_4 \rightarrow u)\le \frac 1{14}$,
	and thus Eq.~(\ref{eq14}) becomes $\omega'(u) \le - 1 + 2\times \frac 16 + \frac 14 + \frac 27+ \frac 1{14} = -\frac 5{84}< 0$.

	When $\alpha_{u_4}\ge \frac 12$,
	$\tau(f_1 \rightarrow u)\le \frac 5{21}$,
	$\tau(f_2 \rightarrow u)\le \frac 27$,
	$\tau(f_3 \rightarrow u)\le \frac 14$,
	$\tau(f_4 \rightarrow u)\le 0$ since $\tau(f_4 \rightarrow u)_{f_4}\ge \frac 12$ and $\tau(f_4 \rightarrow u)_{f_5}= \frac 12$,
	and thus Eq.~(\ref{eq14}) becomes $\omega'(u) \le - 1 + \frac 5{21} +  \frac 27 + \frac 14 + \frac 16 = -\frac 5{84}< 0$.

	In the other case, $\alpha_{u_2}< \frac 12$ and $\alpha_{u_4}< \frac 12$.
	It follows from (II) that $u_2$ and $u_4$ are $(7, 4^-)_{HG}$, which leads to ($A_{7.2}$), a contradiction.
\end{itemize}

{Case 2.2.} $d(u_3)= d_H(u_3) = 6$ and $n'_{6^+}(u)= 5$.
In this case, Eq.~(\ref{eq14}) becomes
\begin{equation}
\label{eq15}
\omega'(u) \le - 1 + \sum_{i\in [1, 5]}(1 - \alpha_{u_i} - \alpha_{u_{(i \mod 5) + 1}}) = 4 - 2\sum_{i\in [1, 5]} \alpha_{u_i}.
\end{equation}
When $n'_{7^+}(u) = 4$, $\alpha_{u_i}\ge \frac 37$ for each $i\in \{1, 2, 4, 5\}$,
and thus Eq.~(\ref{eq15}) becomes $\omega'(u) \le 4 - 2\times (4\times \frac 37 + \frac 13) =  - \frac 2{21}< 0$.

In the other case, $n'_{7^+}(u)\le 3$, i.e., $n_{6}(u) = n'_{6}(u)\ge 2$.
The non-existence of ($A_{8.1}$) in $G$ implies that $n_{6}(u)\le 3$.
By symmetry, we can have the following two scenarios:

(Case 2.2.1.) $n_{6}(u) = 3$.
In this case, $d(u_1) = d(u_4) = 6$ or $d(u_2) = d(u_4) = 6$.
The non-existence of ($A_{8.1}$) and ($A_{8.2}$) in $G$ imply that for each $7^+$-vertex $x\in N(u)$, $d(x)\ge 8$ or $x$ is $(7, 5^+)$;
it follows that $x$ is $(7, 5^+)_H$ or $d_H(u_i)\ge 8$, which implies that $\alpha_{x}\ge \frac 12$.
Thus Eq.~(\ref{eq15}) becomes $\omega'(u) \le 4 - 2\times (2\times \frac 12 + 3\times \frac 13) = 0$.

(Case 2.2.2.) $n_{6}(u) = 3$.
When there exists a $7^+$-vertex $w\in N_H(u)$ such that $\alpha_w\ge \frac 12$, Eq.~(\ref{eq15}) becomes
$\omega'(u) \le 4 - 2\times (\frac 12 + 2\times \frac 37 + 2\times \frac 13) = -\frac 1{21} < 0$.
In the other case, for each $7^+$-vertex $x\in N_H(u)$, $\alpha_x < \frac 12$.
It follows from (II) that $x$ is $(7, 4^-)_{HG}$.
Since $d(u_2) = 6$ or $d(u_1) = 6$, ($A_{8.3}$) or ($A_{8.4}$) is obtained in $G$, a contradiction.
\end{proof}

The above Discharging-Lemmas~\ref{demma05}, \ref{demma06}, \ref{demma09} and \ref{demma10} together
contradict the positive total weight of $8$ stated in Eq.~(\ref{eq01}),
and thus prove Theorem~\ref{thm02}.

\section{Acyclic edge $(\Delta+5)$-coloring}
In this section, we show how to derive an acyclic edge $(\Delta+5)$-coloring for a simple $2$-connected planar graph $G$,
by an induction on the number of edges $|E(G)|$.
The following lemma gives the starting point.

\begin{lemma}{\rm (\cite{Sku04,AMM12,BC09,SWMW19,WMSW19})}
\label{lemma05}
If $\Delta \in \{3, 4\}$, then $a'(G) \le \Delta + 2$, and an acyclic edge $(\Delta + 2)$-coloring can be obtained in polynomial time.%
\footnote{\cite{Sku04} gives an $O(n)$-time algorithm for $\Delta = 3$;
	\cite{BC09} gives a polynomial-time algorithm for $\Delta = 4$ except $4$-regular;
	\cite{SWMW19,WMSW19} gives a polynomial-time algorithm for $4$-regular graphs.}
\end{lemma}

Given a partial acyclic edge $k$-coloring $c(\cdot)$ of the graph $G$ using the color set $C = \{1, 2, \ldots, k\}$,
for a vertex $v\in V(G)$, let $C(v)$ denote the set of colors assigned to the edges incident at $v$ under $c(\cdot)$.
If the edges of a path $P= ux\ldots v$ are alternatively colored $i$ and $j$, we call it an {\em $(i, j)_{(u, v)}$-path}.
Furthermore, if $uv\in E(G)$ is also colored $i$ or $j$, we call $ux\ldots vu$ an {\em $(i, j)_{(u, v)}$-cycle}
(which implies $c(\cdot)$ is invalid and needs to be revised).

For simplicity, we use {\em $\{e_1, e_2, \ldots, e_m\} \to a$} to state that all the edges $e_1$, $e_2$, $\ldots$, $e_m$ are colored $a$,
use simply $e_1 \to a$ to state that $e_1$ is colored $a$,
use $e_1 \to S$ ($S \ne\emptyset$) to state that $e_1$ is colored with a color in $S$,
and use $(e_1, e_2, \ldots, e_m) \to (a_1, a_2, \ldots, a_m)$ to state that $e_j$ is colored $a_j$, for $j = 1, 2, \ldots, m$.
We also use $(e_1, e_2, \ldots, e_m)_c = (a_1, a_2, \ldots, a_m)$ to denote that $c(e_j) = a_j$, for $j = 1, 2, \ldots, m$;
that is, to emphasize the coloring is by $c(\cdot)$.
Moreover, we use {\em $i_1/i_2/\ldots/i_k\in S$} to state that at least one of $i_1, i_2, \ldots, i_k$ is in $S$, i.e., $S\cap \{i_1, i_2, \ldots, i_k\}\ne \emptyset$.
By a graph $H$ containing an $(S_1, S_2)_{(u, v)}$-path
(an $(a, S)_{(u, v)}$-path, an $(a, i_1/i_2/\ldots/i_k)_{(u, v)}$-path, respectively),
it means that $H$ contains an $(i, j)_{(u, v)}$-path for each pair $i\in S_1$, $j\in S_2$
(an $(a, j)_{(u, v)}$-path for each $j\in S$, an $(a, j)_{(u, v)}$-path for some $j\in \{i_1, i_2, \ldots, i_k\}$, respectively).

The following lemma is obvious (via a simple contradiction):
\begin{lemma}{\rm \cite{SWMW19}}
\label{lemma06}
Suppose $G$ has an acyclic edge coloring $c(\cdot)$,
and $P = uv_1v_2$-$\ldots$-$v_kv_{k+1}$ is a maximal $(a, b)_{(u, v_{k+1})}$-path with $c(uv_1) = a$ and $b \not\in C(u)$.
Then there is no $(a, b)_{(u, w)}$-path for any vertex $w \not\in V(P)$.
\end{lemma}

The rest of the section is devoted to the proof of Theorem~\ref{thm01}, by an induction on the number of edges $|E(G)|$.
The flow of the proof is depicted in Figure~\ref{fig07}.
\begin{figure}[h]
\begin{center}
\includegraphics[width=0.9\textwidth]{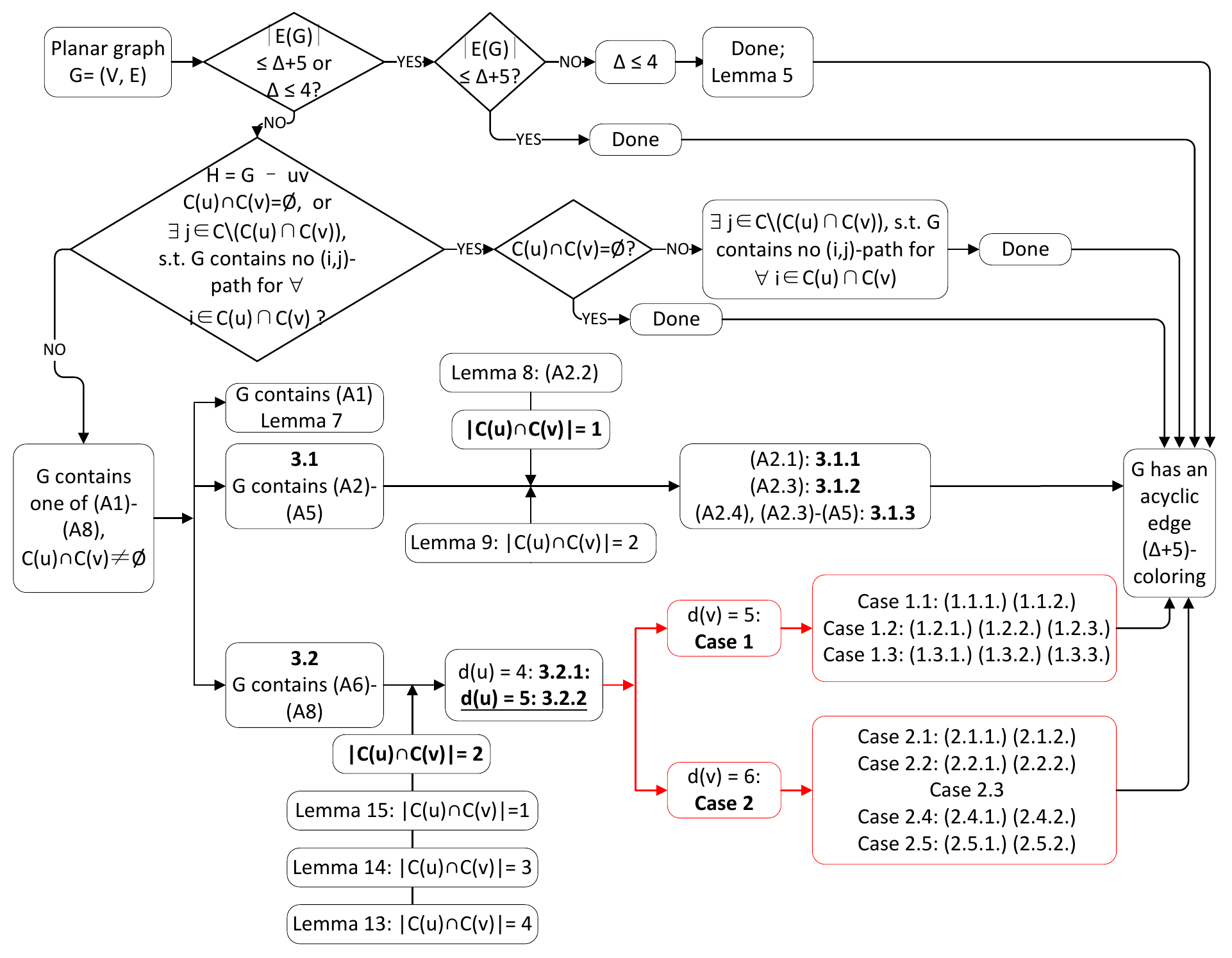}
\caption{The flow of the proof of Theorem~\ref{thm01} by an induction on $|E(G)|$.\label{fig07}}
\end{center}
\end{figure}

For the base cases in the induction,
one sees that when $|E(G)| \le \Delta + 5$, coloring each edge by a distinct color gives a valid acyclic edge $(\Delta + 5)$-coloring;
when $\Delta \le 4$, an acyclic edge $(\Delta + 5)$-coloring is guaranteed by Lemma~\ref{lemma05}.

\subsection*{3.0 \ Configuration ($A_1$), the simplest cases of the inductive step}
In the sequel we consider a simple $2$-connected planar graph $G$ such that $\Delta \ge 5$ and $|E(G)| \ge \Delta + 6$.
We first use some simplest cases to convince the readers that the inductive step works,
then starting Subsection 3.1 the inductive step becomes more complex yet using the similar ideas.
Theorem~\ref{thm02} states that $G$ contains at least one of the configurations ($A_1$)--($A_8$),
which are the cases to be discussed inside the inductive step.

As the most simplest case of the inductive step, consider the configuration ($A_{1.1}$) in $G$.
Let $G' = G - v + u w$, that is, by merging the two edges $u v$ and $v w$ into a new edge $u w$ while eliminating the vertex $v$.
The graph $G'$ is planar, $\Delta(G') = \Delta$, and contains one less edge than $G$.
Using the inductive hypothesis, $G'$ has an acyclic edge $(\Delta + 5)$-coloring $c(\cdot)$ using the color set $C = \{1, 2, \ldots, \Delta + 5\}$,
in which we assume that $c(u w) = 1$.
To construct an edge coloring for $G$,
we keep the color for every edge of $E(G') \setminus \{u w\}$, $v w \to 1$, and $u v \to C \setminus C(u)$.
One may trivially verify that no bichromatic cycle would be introduced, that is, the achieved edge $(\Delta + 5)$-coloring for $G$ is acyclic.

One sees that for the edge $uv$ in each of the other configurations,
we have the degree $d(v)\in \{2, 3\}$ (in ($A_{1.2}$), ($A_{1.3}$), ($A_{2}$)--($A_5$) respectively),
or the degree $d(u)\in \{4, 5\}$ (in ($A_6$)--($A_8$) respectively).
We pick the edge $uv$ and let $H = G - uv$ (that is, $H$ is obtained by removing the edge $uv$ from $G$, or precisely $H = (V(G), E(G) \setminus\{uv\})$).
The graph $H$ is planar, and $4 \le \Delta - 1 \le \Delta(H) \le \Delta$.%
\footnote{Note that $H$ might not be $2$-connected, but it does not matter.}
Using the {\em inductive hypothesis}, $H$ has an acyclic edge $(\Delta(H) + 5)$-coloring $c(\cdot)$,
and thus a $(\Delta + 5)$-coloring, using the color set $C = \{1, 2, \ldots, \Delta + 5\}$.

We have $d_H(u) + d_H(v) \le \Delta - 1 + 5 - 1 = \Delta + 3 < |C|$, or equivalently $C \setminus (C(u)\cup C(v)) \ne \emptyset$.
One sees that,
if $C(u)\cap C(v) = \emptyset$, then $uv\overset{\divideontimes}\to C\setminus (C(u)\cup C(v))$.
We define the notation $\overset{\divideontimes}\to$ to mean that $uv\to C\setminus (C(u)\cup C(v))$ gives an acyclic edge $(\Delta + 5)$-coloring for $G$ and thus we are done.
Or, if there exists a $j\in C\setminus (C(u)\cup C(v))$ such that, for all $i\in C(u)\cap C(v)$, $H$ contains no $(i, j)_{(u, v)}$-path,
then $uv\overset{\divideontimes}\to j$.
In the remaining case, for each $i\in C(u)\cap C(v)$, we assume that $c(u u_s) = c(v v_t) = i$,
and we define the color set
\[
B_i = \{j \mid \mbox{there is an } (i, j)_{(u_s, v_t)}\mbox{-path in } H\}.
\]
Therefore, for each $j \in B_i$, there is a neighbor $x$ of $u_s$ and a neighbor $y$ of $v_t$, respectively,
such that $c(u_s x) = c(v_t y) = j$, and each of $x$ and $y$ is incident with an edge colored $i$.
One sees that $B_i \subseteq C(u_s)\cap C(v_t)$.
In the sequel we continue the proof with the following Proposition~\ref{prop3001}%
\footnote{Throughout the paper, we characterize a lot of intermediate properties, but not final conclusions,
	that are useful for down-stream proofs into {\em propositions}.
	These propositions are subsection numbered in four levels.
	When citing them, we only use the 4-level labels, without mentioning ``Proposition''.},
or otherwise we are done:
\begin{proposition}
\label{prop3001}
$C(u)\cap C(v) \ne \emptyset$, and $C\setminus (C(u)\cup C(v))\subseteq \bigcup _{i\in C(u)\cap C(v)}B_i$.
\end{proposition}
For each degree $k\in \{2, 3, 4, 5\}$,
we collect the colors of the edges connecting $u$ and its $k$-neighbors ($k^-$-neighbors, respectively) into
$S_k(u) = \{c(u w) \mid w \in N(u) \setminus \{v\}, d(w) = k\}$ ($S_{k^-}(u) = \{c(u w) \mid w \in N(u) \setminus \{v\}, d(w) \le k\}$, respectively).

The second simplest cases would be the configurations ($A_{1.2}$) and ($A_{1.3}$), in which there is a path $u v w$ with $d(v)= 2$.
Let $u_1, u_2, \ldots, u_{d(u) - 2}$ be the other neighbors of $u$.
Using the acyclic edge coloring $c(\cdot)$ for $H = G - u v$,
by \ref{prop3001}, $c(vw) \in C(u)$, $H$ contains a $(c(v w), i)_{(u, w)}$-path for every $i \in C \setminus C(u)$,
and thus $(C\setminus C(u))\subseteq C(w) \setminus \{c(v w), c(u w)\}$.
Let $c(u u_i) = i$ for each $1 \le i \le d(u) - 2$ and $(v w, u w)_c= (1, d(u) - 1)$.
It follows that ($A_{1.3}$) holds and $[d(u), \Delta + 5]\subseteq C(w)$.
One sees that $|[d(u), \Delta + 5]|\ge \Delta + 5 - (d(u) - 1) = \Delta - d(u) + 6\ge 6$.
Since $n_{5^-}(u)\ge d(u) - 6$, $S_{5^-}(u)\ge d(u) - 7$ and $|S_{5^-}(u)\cup (C\setminus C(u))\cup \{1, d(u) - 1\}|\ge d(u) - 7 + (\Delta - d(u) + 6) + 2 = \Delta + 1$.
It follows that $6\in S_{5^-}(u)\setminus C(w)$.
Then $6\to v w$ and $[d(u), \Delta + 5]\subseteq C(u_6)\to u v$ give rise to an acyclic edge $(\Delta + 5)$-coloring for $G$, and we are done.

We have thus far demonstrated that the inductive step works for the case where $G$ contains the configuration ($A_1$),
i.e., one of the configurations ($A_{1.1}$)--($A_{1.3}$),
and summarize in the following lemma, where the running time refers to the time spent during the inductive step.

\begin{lemma}
\label{31}
If ($A_1$) holds, then in $O(1)$ time an acyclic edge $(\Delta + 5)$-coloring for $G$ can be obtained.
\end{lemma}

The next two subsections deal with the configurations ($A_{2}$)--($A_{5}$) and ($A_6$)--($A_8$), respectively, and arranged as follows.

In Subsection 3.1, we first deal with ($A_{2.2}$) in Lemma~\ref{322},
and then $|C(u)\cap C(v)| = 2$ in Lemma~\ref{32to1},
which says that in $O(1)$ time either an acyclic edge $(\Delta + 5)$-coloring for $G$ can be obtained,
or an acyclic edge $(\Delta + 5)$-coloring for $H$ can be obtained such that $|C(u)\cap C(v)| = 1$;
next, we deal with ($A_{2.1}$) in Subsection 3.1.1, ($A_{2.3}$) in Subsection 3.1.2, and ($A_{2.4}$), ($A_{3}$)-- ($A_{5}$) in Subsection 3.1.3, respectively.

In Subsection 3.2, we first consider the condition where $|C(u)\cap C(v)| = 4$ in Lemma~\ref{4to321},
$|C(u)\cap C(v)| = 3$ in Lemma~\ref{3to21},
and $|C(u)\cap C(v)| = 1$ in Lemma~\ref{auv1coloring},
which says that in $O(1)$ time either an acyclic edge $(\Delta + 5)$-coloring for $G$ can be obtained,
or an acyclic edge $(\Delta + 5)$-coloring for $H$ can be obtained such that $|C(u)\cap C(v)| = 1$;
next, we deal with $d(u) = 4$ in Subsection 3.2.1, $d(u) = 5$ in Subsection 3.2.2, respectively.
In particular, in Subsection 3.2.2., we distinguish two cases for the two possible values $d(v)$.
It takes a long time in Subsection 3.2.2 to verify that if $G$ contains ($A_{8.1}$),
then an acyclic edge $(\Delta + 5)$-coloring for $H$ can be obtained.
We remark that this process is the most interesting.

\subsection{Configurations ($A_{2}$)--($A_{5}$)}
In this subsection we prove the inductive step for the case where $G$ contains one of the configurations ($A_{2}$)--($A_{5}$).
We have the following revised ($A_{2}$)--($A_{5}$):
An edge $u v$ with $d(v) = 3$ and $d(u)= k\le 11$ such that at least one of the configurations ($A_{2}$)--($A_{5}$) occurs.

\begin{lemma}
\label{322}
If ($A_{2.2}$) holds, then in $O(1)$ time an acyclic edge $(\Delta + 5)$-coloring for $G$ can be obtained.
\end{lemma}
\begin{proof}
Let $G' = G - v + v_1 v_2$, that is, merging the two edges $v_1 v$ and $v v_2$ into a new edge $v_1 v_2$ while eliminating the vertex $v$.
The graph $G'$ is a planar graph, $\Delta(G') = \Delta$, and contains two less edges than $G$.
By the inductive hypothesis, $G'$ has an acyclic edge $(\Delta + 5)$-coloring $c(\cdot)$ using the color set $C = \{1, 2, \ldots, \Delta + 5\}$,
in which $C(x)$ is the color set of the edges incident at $x\in \{u, v_1, v_2\}$,
and we assume that $(u v_1, u v_2)_c = (2, 3)$ and $c(u u_i) = i$, $i\in \{1, 4, 5, \ldots, k - 1\}$,
where $u_1, u_4, \ldots, u_{k - 1}$ are the neighbors of $u$ other that $v, v_1, v_2$.

To transform back an edge coloring for $G$,
we keep the color for every edge of $E(G') \setminus \{v_1 v_2\}$, and color $\{u v, v v_1, v v_2\}$ as follows by recoloring some edges incident to $u$ if needed.

{\bf Case 1.} $c(v_1 v_2)\not\in C(u)$, assuming w.l.o.g. $c(v_1 v_2) = k$.

When $[k + 1, \Delta + 5]\setminus C(v_2) \ne \emptyset$, $v v_1\to k$,
$v v_2\to \alpha\in [k + 1, \Delta + 5]\setminus C(v_2)$,
and $u v\overset{\divideontimes}\to [k + 1, \Delta + 5]\setminus \{\alpha\}$.

In the other case, $[k + 1, \Delta + 5]\subseteq C(v_2)$ and likewise, $[k + 1, \Delta + 5]\subseteq C(v_1)$.
Since $u$ is $([8, 11], 5^-)$, $|S_{5^-}(u)|\ge k - 6$.
One sees that $|S_{5^-}(u)\cup [k + 1, \Delta + 5]|\ge k - 6 + \Delta + 5 - k = \Delta - k -1$.
It follows that there exists an $i_0\in S_{5^-}(u)\setminus C(v_1)$.
Then $(v v_1, v v_2)\to (i_0, k)$ and $u v\overset{\divideontimes}\to [k + 1, \Delta + 5]\setminus C(u_{i_0})$.

{\bf Case 2.} $c(v_1 v_2)\in C(u)$, assuming w.l.o.g. $c(v_1 v_2) = 1$.

When $[k, \Delta + 5]\setminus C(v_i)\ne \emptyset$ for some $i\in \{1, 2\}$, assuming w.l.o.g. $k\not\in C(v_2)$,
if there is no $(1, j_0)_{(u_1, v_1)}$-path for some $j_0\in [k + 1, \Delta + 5]$,
then $(v v_1, v v_2, u v)\overset{\divideontimes}\to (1, k, j_0)$;
otherwise, $G$ contains a $(1, j)_{(u_1, v_1)}$-path for each $j\in [k + 1, \Delta + 5]\subseteq C(v_1)$.
Since $|[k + 1, \Delta + 5]| = \Delta - k + 5 \ge 5$ and $S_{5^-}(u)\ge k - 6$,
there exists an $j_1\in S_{5^-}(u)\setminus C(v_1)$.
Then $(v v_1, v v_2)\to (j_1, 1)$ and $u v\overset{\divideontimes}\to [k + 1, \Delta + 5]\setminus C(u_{j_1})$.

In the other case, $[k, \Delta + 5]\subseteq C(v_1)\cap C(v_2)$.
Note that $|[k, \Delta + 5]| = \Delta - k + 6\ge 6$.
When there exists $j_2\in S_{5^-}(u)\setminus (C(v_1)\cup C(v_2))$,
let $\beta\in [k, \Delta + 5]\setminus C(u_{j_2})$,
by Lemma~\ref{lemma06}, $(v v_1, v v_2, u v)\to (1, j_2, \beta)$ or $(v v_1, v v_2, u v)\to (j_2, 1, \beta)$ gives rise to an acyclic edge $(\Delta + 5)$-coloring for $G$.

Now we assume that $S_{5^-}(u)\subseteq (C(v_1)\cup C(v_2))$.
Since $2 + |[k, \Delta + 5]| + 2 + |[k, \Delta + 5]| + |S_{5^-}(u)| = 2\Delta + 10 - k \le d(v_1) + d(v_2)\le 2\Delta$, we have $k\ge 10$.
Since $2 + |[k, \Delta + 5]| + |S_{5^-}(u)| = 2 + \Delta + 6 - k + k - 6 = \Delta + 2$,
we have $|S_{5^-}(u)\setminus C(v_1)|\ge 2$ and $|S_{5^-}(u)\setminus C(v_2)|\ge 2$.

For each $i_1\in S_{5^-}(u)\setminus C(v_1)$ and $i\in C\setminus C(v_2)$,
if $G$ contains no $(1, j)_{(v_2, u_1)}$-path for some $j\in [k, \Delta + 5]\setminus C(u_{i_1})$,
then $(v v_1, v v_2, u v)\overset{\divideontimes}\to (i_1, 1, j)$;
otherwise, $G$ contains a $(1, j)_{(v_2, u_1)}$-path for each $j\in [k, \Delta + 5]\setminus C(u_{i_1})$;
if $G$ contains no $(i, j)_{(v_2, u)}$-path for some $j\in [k, \Delta + 5]\setminus C(u_{i_1})$,
then $(v v_1, v v_2, u v)\overset{\divideontimes}\to (1, i, j)$.
In the other case, we proceed with the following proposition:

\begin{proposition}
\label{prop3102}
\begin{itemize}
\parskip=0pt
\item[{\rm (1)}]
	for each $i_1\in S_{5^-}(u)\setminus C(v_1)$ and $i\in C\setminus C(v_2)$,
    $G$ contains a $(1, j)_{(v_2, u_1)}$-path and a $(i, j)_{(v_2, u_{i_2})}$-path for each $j\in [k, \Delta + 5]\setminus C(u_{i_1})$;
\item[{\rm (2)}]
	for each $i_2\in S_{5^-}(u)\setminus C(v_2)$ and $i\in C\setminus C(v_1)$,
    $G$ contains a $(1, j)_{(v_1, u_1)}$-path and a $(i, j)_{(v_1, u_{i_1})}$-path for each $j\in [k, \Delta + 5]\setminus C(u_{i_2})$;
\item[{\rm (3)}]
	for each $i_1\in S_{5^-}(u)\setminus C(v_1)$ and for each $i_2\in S_{5^-}(u)\setminus C(v_2)$,
    $[k, \Delta + 5]\setminus C(u_{i_1})\subseteq C(u_{i_2})$,
    $[k, \Delta + 5]\setminus C(u_{i_2})\subseteq C(u_{i_1})$,
    and $[k, \Delta + 5]\subseteq C(u_{i_1})\cup C(u_{i_2})$.
\end{itemize}
\end{proposition}

For each $i_2\in S_{5^-}(u)\setminus C(v_2)$,
when for some $j\in [k, \Delta + 5]\setminus C(u_{i_2})$, $G$ contains no $(i, j)_{(u_{i_2}, u)}$-path for any $i\in C(u)$,
it follows from \ref{prop3102} that $G$ contains a $(1, j)_{(v_1, u_1)}$-path.
If $G$ contains no $(1, l)_{(u_1, v_1)}$-path for some $[k, \Delta + 5]\setminus \{j\}$,
then $(v v_1, v v_2, u u_{i_2}, u v)\overset{\divideontimes}\to (1, i_2, j, l)$.
Otherwise, $G$ contains a $(1, l)_{(u_1, v_1)}$-path for any $l\in [k, \Delta + 5]$.
Then $(v v_1, v v_2, u v)\overset{\divideontimes}\to (i_1, 1, l)$.
In the other case, we proceed with the following proposition:

\begin{proposition}
\label{prop3103}
\begin{itemize}
\parskip=0pt
\item[{\rm (1)}]
    For each $j\in [k, \Delta + 5]\setminus C(u_{i_2})$, where $i_2\in S_{5^-}(u)\setminus C(v_2)$, $G$ contains a $(i, j)_{(u_{i_2}, u)}$-path for some $i\in C(u)$.
\item[{\rm (2)}]
    For each $j\in [k, \Delta + 5]\setminus C(u_{i_1})$, where $i_1\in S_{5^-}(u)\setminus C(v_1)$, $G$ contains a $(i, j)_{(u_{i_1}, u)}$-path for some $i\in C(u)$.
\end{itemize}
\end{proposition}

It follows that $|C(u_{i_1})\cap [k, \Delta + 5]|\le 3$ and $|C(u_{i_2})\cap [k, \Delta + 5]|\le 3$.
Since $|[k, \Delta + 5]|\ge 6$, we have $|C(u_{i_1})\cap [k, \Delta + 5]|= 3$, $|C(u_{i_2})\cap [k, \Delta + 5]|= 3$, and $\Delta = k$.
Further, by Propositions~\ref{prop3102} and \ref{prop3103},
we assume that for each $i_2\in S_{5^-}(u)\setminus C(v_2)$, $C(u_{i_2}) = \{a_{i_2}, k, k + 1, k + 2\}$;
and for each $i_1\in S_{5^-}(u)\setminus C(v_1)$, $C(u_{i_1}) = \{a_{i_1}, k + 3, k + 4, k + 5\}$.

We consider below the case where $k = 10$ and $k = 11$, respectively.

{Case 2.1.} $k = 10$, and $C(v_1) = \{2, 6, 7\}\cup [10, 16]$,
$C(v_2) = \{3, 8, 9\}\cup [10, 16]$.
Then $C(u_6) = \{a_6, 10, 11, 12\}$, $C(u_7) = \{a_7, 10, 11, 12\}$,
and by Propositions~\ref{prop3102}, $G$ contains a $(i, 15)_{(u, v_1)}$-path for any $i\in \{1, 3, 4, 5, 8, 9\}$.
However, by Propositions~\ref{prop3103}, $G$ contains a $(2, 15)_{(v_1, u_6)}$-path and a $(2, 15)_{(v_1, u_7)}$-path, a contradiction.

{Case 2.2.} $k = 11$.
Assume w.l.o.g. $C(v_1) = \{2, 6, 7, 10\}\cup [11, \Delta + 5]$,
$\{3, 8, 9\}\cup [11, \Delta + 5]\subseteq C(v_2)$.
Then $C(u_6) = \{a_6, 11, 12, 13\}$, $C(u_7) = \{a_7, 11, 12, 13\}$,
and $C(u_8) = \{a_8, 14, 15, 16\}$, $C(u_9) = \{a_9, 14, 15, 16\}$,
where $a_i\in C(u)$ for $i \in [6, 9]$.
It follows from Propositions~\ref{prop3102} that $G$ contains a $(i, 15)_{(u, v_1)}$-path for any $i\in \{1, 3, 4, 5, 8, 9\}$ and a $(1, 11)_{(u_1, v_2)}$-path.

Hence, by Propositions~\ref{prop3103}, $G$ contains a $(i, 15)_{(u, u_6)}$-path and a $(j, 15)_{(u, u_7)}$-path for some $i, j\in \{2, 10\}$.
By Lemma~\ref{lemma03}, we assume w.l.o.g. $G$ contains a $(2, 15)_{(v_1, u_6)}$-path and a $(10, 15)_{(u_{10}, u_7)}$-path,
and $C(u_6) = \{2, 6, 11, 12, 13\}$, $C(u_7) = \{7, 10, 11, 12, 13\}$.
It follows from Propositions~\ref{prop3103} that $G$ contains a $(10, j)_{(u_{10}, u_7)}$-path for each $j\in \{14, 15, 16\}$, and assume w.l.o.g. $11\not\in C(u_{10})$.

If $10\not\in C(v_2)$, then $(v v_1, v v_2, u v)\overset{\divideontimes}\to (1, 10, 11)$.
Otherwise, 
$C(v_2) = \{3\}\cup [8, 16]$.
It follows from Propositions~\ref{prop3102} that $G$ contains a $(i, 11)_{(v_2, u)}$-path for each $i\in \{1, 2, 4, 5, 6, 7\}$.
However, by Propositions~\ref{prop3103}, $G$ contains a $(3, 11)_{(v_2, u_8)}$-path and a $(3, 11)_{(v_2, u_8)}$-path, a contradiction.
\end{proof}

By Lemma \ref{322}, we consider below the case where $G$ contains the configuration ($A_{2}$)--($A_{5}$) other that ($A_{2.2}$).
Let $u_1$, $u_2$, $\ldots$, $u_{k - 1}$ be the other neighbors of $u$ other than $v$,
and assume that $c(u u_i)= i$ for $i = 1, 2, \ldots, k - 1$, i.e., $C(u) = [1, k - 1]$;
let $v_1$, $v_2$ be the neighbors of $v$ other than $u$.

Before we discuss the detailed configurations,
we do not distinguish $v_1$ and $v_2$ but use $x_1$ to refer to any one of $u_1$ and $u_2$ while $x_2$ to refer to the other ($\{x_1, x_2\} = \{u_1, u_2\}$).
In particular, $u_1 = x_1$ in ($A_{2.3}$), ($A_{2.4}$), and ($A_{3}$)--($A_{5}$), and $u_2 = x_2$ in ($A_{2.4}$), and ($A_{3}$)--($A_{5}$).
It follows from \ref{prop3001} that $1\le|C(u)\cap C(v)|\le 2$,
and $C\setminus (C(u)\cup C(v))\subseteq \bigcup _{i\in C(u)\cap C(v)}B_i$.

\begin{lemma}
\label{32to1}
If $|C(u)\cap C(v)| = 2$, then in $O(1)$ time either an acyclic edge $(\Delta + 5)$-coloring for $G$ can be obtained,
or an acyclic edge $(\Delta + 5)$-coloring for $H$ can be obtained such that $|C(u)\cap C(v)| = 1$.
\end{lemma}
\begin{proof}
Assume w.l.o.g. that $(v x_1, v x_2)_c= (a, b)$, where $a, b\in [1, k - 1]$.
When there exists a color $j\in [k, \Delta + 5]\setminus C(x_i)$ for some $i \in \{1, 2\}$,
$v x_i\to j$ gives rise to an acyclic edge $(\Delta + 5)$-coloring for $H$ such that $|C(u)\cap C(v)| = 1$.
In the other case, we proceed with the following proposition:
\begin{proposition}
\label{prop3104}
{\rm (1)}\ $[k, \Delta + 5]\subseteq C(u_a)\cup C(u_b)$; {\rm (2)}\ $[k, \Delta + 5]\subseteq C(x_1)\cap C(x_2)$.
\end{proposition}

{\bf Case 1.} ($A_{2.1}$) holds, i.e., $d(u) = 7$, and $C(u)= [1, 6]$.

Assuming w.l.o.g. $\{a, b\} = \{1, 2\}$.
It follows from \ref{prop3104} that $C(x_1) = \{1\}\cup [7, \Delta + 5]$ and $C(x_2) = \{2\}\cup [7, \Delta + 5]$.
Since $7\in B_1\cup B_2$, assume w.l.o.g. $7\in B_1$.

If $G$ contains no $(i, 7)_{(x_1, u_i)}$-path for some $i\in \{2, 3\}$,
then $(v x_1, v x_2, u v)\overset{\divideontimes}\to (i, 1, 7)$.
Otherwise, for each $i\in \{2, 3\}$, $G$ contains a $(i, 7)_{(x_1, u_i)}$-path and $7\in C(u_i)$.
Likewise, we have $[7, \Delta + 5]\subseteq C(u_1)\cap C(u_2)\cap C(u_3)$ and thus $C(u_i) = \{i\}\cup [7, \Delta + 5]$ for every $i\in \{1, 2, 3\}$.
By Lemma~\ref{lemma06}, there is neither a $(2, 7)_{(u_2, x_2)}$-path nor a $(3, 7)_{(u_1, x_1)}$-path,
and $(u u_1, u u_3, v x_1, u v)\overset{\divideontimes}\to (3, 1, 3, 7)$.

{\bf Case 2.} ($A_{2.3}$) holds, i.e., $u$ is $(8, 4^-)$, $v_2$ is $([8, 11], 5^-)$, and $u v_1, v_1 v_2\in E(G)$, $u v_2\not\in E(G)$.

Assuming w.l.o.g. $c(vx_1) = 2$ and $c(v x_2)\in C(u)\setminus \{2\}$.
It follows from \ref{prop3104} that $C(x_1) = \{1, 2\}\cup [8, \Delta + 5]$ with $c(x_1 x_2) = 8$,
and $\{c(v x_2)\}\cup [8, \Delta + 5]\subseteq C(x_2)$.

When $c(v x_2)\in C(u)\setminus \{1, 2\}$, assuming w.l.o.g. $c(v x_2) = 3$,
since $8\in B_2\cup B_3$ and $3\not\in C(x_1)$, we have $2\in C(x_2)$ and $C(x_2) = \{2, 3\}\cup [8, \Delta + 5]$;
then $(v x_1, u v)\overset{\divideontimes} \to (4, 8)$.

In the other case, $c(v x_2) = 1$.
One sees from $|[3, 7]\cap C(x_2)|\le 1$ that there exists an $i\in [3, 7]\setminus C(x_2)$.
Then $v x_2\to i$ reduces to the case where $c(v x_2)\in C(u)\setminus \{1, 2\}$.

{\bf Case 3.} ($A_{2.4}$) holds, i.e., $d(u) = 8$, and $u v_1, u v_2, v_1 v_2\in E(G)$, $C(u) = [1, 7]$.

It follows from \ref{prop3102} that $C(x_1) = \{1, a\}\cup [8, \Delta + 5]$ and $C(x_2) = \{2, b\}\cup [8, \Delta + 5]$ with $c(x_1 x_2) = 8$.
Since $8\in B_a\cup B_b$, we assume w.l.o.g. $a = 2$ and $b = 3$.
Then $(v x_1, u v)\overset{\divideontimes} \to (4, 8)$.

{\bf Case 4.} ($A_{3}$)--($A_{5}$) holds, i.e., $u$ is $([9, 11], 5^-)$, and $u v_1, u v_2, v_1 v_2\in E(G)$.

Assuming w.l.o.g. $u_i$ be a $6^+$-vertex, $i\in [3, 5]$, and $u_i$ be a $5^-$-vertex, $i\in [6, k - 1]$.
It follows from \ref{prop3102} that
$[k, \Delta + 5] \subseteq C(x_1)$ and $[k, \Delta + 5] \subseteq C(x_2)$.

For any $i\in \{a, b\}\cap S_{5^-}(u)$, say $6\in \{a, b\}$,
when $C(u_6)\setminus\{6\}\subseteq [k, \Delta + 5]$, $u u_6\to [k, \Delta + 5]\setminus C(u_6)$ and an acyclic edge $(\Delta + 5)$-coloring for $H$ such that $|C(u)\cap C(v)| = 1$ is obtained.
In the other case, we proceed with the following proposition:
\begin{proposition}
\label{prop3105}
For any $i\in \{a, b\}\cap S_{5^-}(u)$, $|C(u_i)\setminus [k, \Delta + 5]|\ge 2$ and $|C(u_i)\cap [k, \Delta + 5]|\le 3$.
\end{proposition}

Since $|\{a, b\}\cap \{1, 2\}|\le 1$, $|[3, k - 1]\cap \{a, b\}|\ge 1$ and we discuss the following three
subcases for three possible values of $\{a, b\}$, respectively.

{Case 4.1.} $a, b\in S_{5^-}(u)$, assuming w.l.o.g. $a = 6$, $b = 7$.
Since $|[k, \Delta + 5]|\ge 6$,
it follows from \ref{prop3105} that $\Delta = k$ and assume $C(u_6) = \{a_6, k, k + 1, k + 2\}$,
$C(u_7) = \{a_7, k + 3, k + 4, k + 5\}$.

When $7\not\in C(x_1)$, first $v x_1\to 7$,
if $6\not\in C(x_2)$, then $(v x_2, uv)\overset{\divideontimes}\to (6, k)$;
otherwise, $6\in C(x_2)$.
Since $[k, k + 5]\cup \{2\}\subseteq C(x_2)$, there exists $\beta\in [8, k - 1]\setminus C(x_2)$.
Then $v x_2\to \beta$ and $[k, k + 5]\setminus C(u_\beta)\overset{\divideontimes}\to uv$.

In the other case, $7\in C(x_1)$ and likewise, $6\in C(x_2)$.
For some $i\in (\{3, 4, 5\}\cup [8, k - 1])\setminus (C(x_1)\cup C(x_2))$,
if $G$ contains no $(i, j)_{(u_i, x_1)}$-path for some $j\in [k, k + 2]$,
then $(v x_1, u v)\overset{\divideontimes}\to (i, j)$;
if $G$ contains no $(i, j)_{(u_i, x_2)}$-path for some $j\in [k + 3, k + 5]$,
then $(v x_2, u v)\overset{\divideontimes}\to (i, j)$.
Otherwise, $G$ contains a $(i, j)_{(u_i, x_1)}$-path for each $j\in [k, k + 2]$ and a $(i, j)_{(u_i, x_2)}$-path for each $j\in [k + 3, k + 5]$.
Thus, $[k, k + 5]\subseteq C(u_i)$ and $d(u_i)\ge 7$.
It follows that $[8, k - 1]\subseteq C(x_1)\cup C(x_2)$.

(4.1.1.)\ $c(x_1 x_2)\in \{3, 4, 5\}$, say $c(x_1x_2) = 3$.
Since $2\times |[k, k + 5]| + |\{1, 3, 6, 7\}| + |\{2, 3, 6, 7\}| + |[8, k - 1]| \le 2\Delta = 2k$,
we have $k\ge 12$, a contradiction.

(4.1.2.)\ $c(x_1 x_2)\in [8, k - 1]$, say $c(x_1 x_2) = 8$.
Since $2\times |[k, k + 5]| + |\{1, 6, 7, 8\}| + |\{2, 6, 7, 8\}| + |[9, k - 1]| \le 2\Delta = 2k$,
we have $k\ge 11$.
It follows that $k = 11$, and $C(x_1) = \{1, 6, 7, 8, 9\}\cup [k, k + 5]$,
$C(x_2) = \{2, 6, 7, 8, 10\}\cup [k, k + 5]$.
Since $G$ contains a $(3, j)_{(u_3, x_1)}$-path for each $j\in [k, k + 2]$ and a $(3, j)_{(u_3, x_2)}$-path for each $j\in [k + 3, k + 5]$,
thus $x_1 x_2\to 3$ reduces to the case where $8\not\in C(x_1)\cup C(x_2)$.

(4.1.3.)\ $c(x_1 x_2)\in [k, k + 5]$, say $c(x_1x_2) = k$.
Since $k\le 11$ and $2\times |[k, k + 5]| + |\{1, 6, 7\}| + |\{2, 6, 7\}| + |[8, k - 1]| \le 2\Delta = 2k$, i.e., $10 + k\le 2k$,
we have $|\{3, 4, 5\}\cap (C(x_1)\cup C(x_2))|\le 1$.
Assuming w.l.o.g. $3, 4\not\in C(x_1)\cup C(x_2)$.
Then $(v x_1, v x_2, u v)\overset{\divideontimes}\to (3, 4, k)$.

{Case 4.2.} $\{a, b\}\cap S_{5^-}(u)\ne \emptyset$, assuming w.l.o.g. $a = 6$, $b \in \{1, 3\}$.

When $6\not\in C(x_2)$,
since $[k, \Delta + 5]\cup \{1, 6\}\subseteq C(x_1)$ and $|[k, \Delta + 5]\cup \{1, 6\}| + |[7, k - 1]| = \Delta + 1$, there exists $\gamma\in [7, k - 1]\setminus C(x_1)$.
Then $(v x_2, v x_1)\to (6, \gamma)$ reduces the proof to Case 4.1.

In the other case, $6\in C(x_2)$.
Assume $[7, k - 1]\setminus (C(x_1)\cup C(x_2))\ne \emptyset$, say $7\not\in C(x_1)\cup C(x_2)$,
then $v x_2\to 7$ reduces the proof to Case 4.1.
Now assume that $[7, k - 1]\subseteq C(x_1)\cup C(x_2)$.

One sees that $k\le 11$ and $2\times |[k, \Delta + 5]| + |\{1, 6\}| + |\{2, b, 6\}| + |[7, k - 1]| = 2\Delta + 10 - k \le 2\Delta$.

When $b = 1$, it follows that $\{3, 4, 5\}\setminus (C(x_1)\cup C(x_2))\ne \emptyset$,
say $3\not\in C(x_1)\cup C(x_2)$,
thus $v x_2\to 3$ reduces the proof to the case where $b\in \{3, 4, 5\}$.
Hence, it suffices to assume that $b = 3$.

(4.2.1.)\ $c(x_1 x_2)\in \{4, 5\}$, say $c(x_1 x_2) = 4$.
Since $2\times |[k, \Delta + 5]| + |\{1, 4, 6\}| + |\{2, 3, 4, 6\}| + |[7, k - 1]| = 2\Delta + 12 - k \le 2\Delta$,
we have $k\ge 12$, a contradiction.

(4.2.2.)\ $c(x_1 x_2)\in [7, k - 1]$, say $c(x_1 x_2) = 7$.
Since $2\times |[k, \Delta + 5]| + |\{1, 6, 7\}| + |\{2, 3, 6, 7\}| + |[8, k - 1]| = 2\Delta + 11 - k \le 2\Delta$,
we have $k\ge 11$.
It follows that $k = 11$, and $C(x_1) = \{1, 6, 7, 8, 9\}\cup [11, \Delta  + 5]$,
$C(x_2) = \{2, 3, 6, 7\}\cup [10, \Delta  + 5]$.
It follows from \ref{prop3105} that $|C(u_6)\cap [11, \Delta + 5]|\le 3$.
Let $[14, \Delta + 5]\cap C(u_6) = \emptyset$.
Then $[14, \Delta + 5]\subseteq B_3$.
Let $(v x_1, v x_2)\to (3, 8)$.

If there exists an $i\in [14, \Delta + 5]\setminus C(u_8)$,
then $u v\overset{\divideontimes}\to i$.
Otherwise, $[14, \Delta + 5]\subseteq C(u_8)$ and it follows from \ref{prop3105} that $\Delta = 11$ and
$|C(u_8)\cap [11, \Delta + 5]| = [14, 16]$.

Hence, $G$ contains a $(3, i)_{(x_1, u_3)}$-path for any $i\in [11, 13]$,
i.e., $G$ contains no $(3, i)_{(x_1, x_2)}$-path for any $i\in [11, 16]$.
Then $(v x_1, x_1 x_2, v x_2)\to (6, 3, 7)$ reduces the proof to Case 4.1.

(4.2.3.)\ $c(x_1 x_2)\in [k, \Delta + 5]$, say $c(x_1 x_2) = k$.
When $4, 5\not\in C(x_1)\cup C(x_2)$, $(v x_1, v x_2, u v)\overset{\divideontimes}\to (4, 5, k)$.
In the other case, assume that $4\in C(x_1)\cup C(x_2)$.

Since $k\le 11$ and $2\times |[k, \Delta + 5]| + |\{1, 6\}| + |\{2, 3, 6\}| + |[7, k - 1]| + |\{4\}| =
2\Delta  + 11 - k \le 2\Delta$,
we have $k = 11$ and $3\not\in C(x_1)$, $5\not\in C(x_1)\cup C(x_2)$.
Then $(v x_1, v x_2, u v)\overset{\divideontimes}\to (3, 5, k)$.

{Case 4.3.} $\{a, b\}\cap S_{5^-}(u)= \emptyset$, assuming w.l.o.g. $a = 3$, $b \in \{1, 4\}$.

When $6\not\in C(x_1)\cup C(x_2)$, $v x_2\to 6$ reduces the proof to to Case 4.2.
In the other case, $[6, k - 1]\subseteq C(x_1)\cup C(x_2)$.
One sees that $k\le 11$ and $2\times |[k, \Delta + 5]| + |\{1, 3\}| + |\{2, b\}| + |[6, k - 1]| = 2\Delta + 10 - k \le 2\Delta$.
When $b = 1$, it follows that $\{4, 5\}\setminus (C(x_1)\cup C(x_2))\ne \emptyset$,
say $4\not\in C(x_1)\cup C(x_2)$,
thus $v x_2\to 4$ reduces the proof to the case where $b\in \{4, 5\}$.
Hence, it suffices to assume that $b = 4$.

When $3\not\in C(x_2)$, since $\{1, 3\}\cup [k, \Delta + 5]\subseteq C(x_1)$,
we have $[6, k - 1]\setminus C(x_1)\ne \emptyset$,
and thus $v x_2 \to 3$, $v x_1\to [6, k - 1]\setminus C(x_1)$ reduces the proof to Case 4.2.
In the other case, $3\in C(x_2)$.
Since $k\le 11$ and $2\times |[k, \Delta + 5]| + |\{1, 3\}| + |\{2, 4, 3\}| + |[6, k - 1]| = 2\Delta + 11 - k \le 2\Delta$,
we have $k = 11$, $4\not\in C(x_1)$, $5\not\in C(x_1)\cup C(x_2)$, and $c(x_1 x_2)\in [k, \Delta + 5]$.
Then $(v x_1, u v)\overset{\divideontimes}\to (5, c(x_1 x_2))$.
\end{proof}

By Lemma~\ref{32to1}, hereafter we assume that $|C(u)\cap C(v)| = 1$,
and w.l.o.g. that $(v x_1, v x_2)_c = (a, k)$,
and further from \ref{prop3001} that $[k + 1, \Delta + 5] = (C\setminus \{1, 2, \ldots, k\})\subseteq B_1 \subseteq C(u_1)\cap C(x_1)$.
We first characterize a useful particularity as follows.

\begin{lemma}
\label{3akcycle}
For ($A_{2}$)--($A_{5}$) other than ($A_{2.2}$),
either an acyclic edge $(\Delta + 5)$-coloring for $G$ can be obtained from $H$ in $O(1)$ time,
or $k\in B_a$ and the $(a, k)_{(u_1, x_1)}$-path in $H$ cannot pass through $x_2$.
\end{lemma}
\begin{proof}
Assume $a = 1$ and $k\not\in B_1$.
One sees that $|[k + 1, \Delta + 5]| = \Delta + 5 - k$.
When $[k + 1, \Delta + 5]\setminus C(x_2) \ne \emptyset$,
let $S = [k + 1, \Delta + 5]\setminus C(x_2)$;
when $1\not\in C(x_2)$, let $S = C\setminus C(x_2)$;
when there exists $C(u)\setminus (C(x_1)\cup C(x_2))$, let $S = \{\alpha\}$.
Then $v x_2\to S$ and $u v\overset{\divideontimes} \to k$.

In the other case, $[k + 1, \Delta + 5]\cup \{1\} \subseteq C(x_1)\cap C(x_2)$ and $[2, k - 1]\subseteq C(x_1)\cup C(x_2)$.
Note that $2\times |[k + 1, \Delta + 5]\cup \{1\}| + |\{k\}| + |[2, k - 1]| = 2\Delta - k + 11\le d(x_1) + d(x_2)\le 2\Delta.$
It follows that $k = 11$, and assume w.l.o.g. $C(x_1) = [1, 6]\cup [12, \Delta + 5]$,
$C(x_2) = \{1\} \cup [7, \Delta + 5]$.
Since $[k + 1, \Delta + 5] = [12, \Delta + 5]\subseteq B_1 \subseteq C(x_1)$, $d(x_1)\ge 6$.
One sees that the same argument applies if let $v x_1\to [7, 10]$ and thus
$d(x_i)\ge 6$ for each $i\in [7, 10]$.
Since from ($A_{5}$) that $u$ is $(11, 5^-)$,
we have $u$ is $(11, 5)$ and $d(x_2)\le 5$.
Then $(v x_1, v x_2)\to (11, 2)$ and $u u_2\overset{\divideontimes}\to [12, \Delta + 5]\setminus C(u_2)$.
\end{proof}

By Lemma~\ref{3akcycle}, we assume below $C\setminus C(u) = [k, \Delta + 5]\subseteq C(u_a)\cap C(x_1)$.
Since $|[k, \Delta  + 5]| = \Delta + 6 - k\ge 6$,
we have $a\in C(u)\setminus S_{5^-}(u)$.
That is, if $a\in S_{5^-}(u)$, then we can obtain an acyclic edge $(\Delta + 5)$-coloring for $G$.

\subsubsection{Configuration ($A_{2.1}$)}
In this case, $C(u) = [1, 7]$, $(v x_1, v x_2)_c = (1, 7)$,
and $C(u_1) = C(x_1) = \{1\}\cup B_1 = \{1\}\cup [7, \Delta + 5]$.

When there exists $b$ such that $C(u_b) = \{b\}\cup [7, \Delta + 5]$, $(u u_1, u u_b, u v)\overset{\divideontimes}\to (b, 1, 8)$.
In the other case, $[7, \Delta + 5]\setminus C(u_i) \ne \emptyset$ for any $i\in [2, 6]$.
When $1\not\in C(x_2)$, $(v x_1, v x_2)\to (2, 1)$ and $u v\overset{\divideontimes} \to [7, \Delta + 5]\setminus C(u_2)$.
One sees that when $G$ contains no $(2, j)_{(x_1, x_2)}$-path for some $j\in \{7\}\cup ([7, \Delta + 5]\setminus C(x_2))$,
it follows from Lemma~\ref{3akcycle} that $C(u_2) = \{2\}\cup B_2 = \{2\}\cup [7, \Delta + 5]$ if $(v x_1, v x_2)\to (2, j)$.

In the other case, $G$ contains an $(i, j)_{(x_1, x_2)}$-path for each $i\in [2, 6]$,
$j\in  \{7\}\cup ([7, \Delta + 5]\setminus C(x_2))$, and $[1, 7]\subseteq C(x_2)$.
It follows that $x_1, x_2\not\in N(u)$, i.e., $u v$ is not contained in a $3$-cycle.

When $\Delta = 7$, let $(x_1 w_7, x_1 w_8)_c = (7, 8)$,
it follows that $C(w_7)\setminus \{7\} = C(w_8)\setminus \{8\} = [1, 6]$,
thus by Lemma~\ref{lemma06}, $(x_1 w_7, x_1 w_8, u v)\overset{\divideontimes} \to (8, 7, 8)$.

In the other case, $\Delta \ge 8$ and $u x_1, u x_2\not\in E(G)$.
Let $G' = G - v + \{u x_1, u x_2\}$, that is, merging the two edges $u v$ and $v x_i$ into a new edge $u x_i$, $i = 1, 2$, while eliminating the vertex $v$.
The graph $G'$ is a planar graph, $\Delta(G') = \Delta$, and contains one less edge than $G$.
By the inductive hypothesis, $G'$ has an acyclic edge $(\Delta + 5)$-coloring $c(\cdot)$ using the color set $C = \{1, 2, \ldots, \Delta + 5\}$ and we assume w.l.o.g. that $c(u u_i) = i$, $i\in [1, 6]$,
$(u x_1, u x_2)_c = (7, 8)$.
To transform back an edge coloring for $G$,
we keep the color for every edge of $E(G') \setminus \{u x_1, u x_2\}$, $(v x_1, v x_2)\to (7, 8)$,
and $u v \overset{\divideontimes}\to C\setminus [1, 8]$.

\subsubsection{Configuration ($A_{2.3}$)}
In this case, $C(u) = [1, 8]$ with $c(u v_1) = 1$.
Let $u_i$ be a $6^+$-vertex, $i\in [2, 4]$ and $u_j$ be a $5^-$-vertex, $j\in [5, 7]$.
When $a = 1$ and assume w.l.o.g. $(v v_1, v v_2)_c = (8, 1)$,
since $[9, \Delta + 5]\subseteq C(v_1)\cap C(v_2)$, there exists an $i\in [2, 7]\setminus (C(v_1)\cup C(v_2))$,
thus $(v v_1, u v)\overset{\divideontimes}\to (i, 8)$\footnote{It is clearly that $G$ contains no $(1, 8)_{(v_1, v_2)}$-path that cannot pass through $v_1$,
a contradiction to Lemma~\ref{3akcycle}. }.
In the other case, $a\in [2, 4]$, say $a = 2$.
Below we distinguish two cases where $c(v v_2) = 2$ and $c(v v_1) = 2$.

{\bf Case 1.} $(v v_1, v v_2)_c = (8, 2)$.

In this case, $[8, \Delta + 5]\subseteq C(u_2)\cap C(v_2)$.
If $1\not\in C(v_2)$, then $v v_2\to 1$ reduces the proof to the case where $a = 1$.
Otherwise, $1\in C(v_2)$ and $C(v_2) = \{1, 2\}\cup [8, \Delta + 5]$.
It follows that $c(v_1 v_2)\in [9, \Delta + 5]$, say $c(v_1 v_2) = 9$.
Since $9\in B_2$, we have $2\in C(v_1)$.

When there exists an $i\in [3, 7]\setminus C(v_1)$, $(v v_2, u v)\overset{\divideontimes} \to (i, 9)$.
In the other case, $[3, 7]\subseteq C(v_1)$ and then $[1, 9]\subseteq C(v_1)$.
It follows that $d(v_1)\ge 9$ and let $T_1 = ([10, \Delta + 5]\setminus C(v_1))\cup \{8\}$.
One sees that $|T_1| = \Delta + 5 - 9 - (d(v_1) - 9) + 1 = \Delta - d(v_1) + 6 \ge 6$.

Let $c(v_2 w_j) = j \in T_1$.
If $G$ contains no $(i, j)_{(v_1, v_2)}$-path for some $i\in [5, 7]$ and $j\in T_1$,
then $(v v_1, v v_2)\to (j, i)$ reduces the proof to the case where $a\in S_{5^-}(u)$.
Otherwise, $G$ contains a $(i, j)_{(v_1, v_2)}$-path for each $i\in [5, 7]$ and $j\in T_1$.
It follows that for any $j\in T_1$, $[5, 7]\cup \{2\}\subseteq C(w_j)$ and $d(w_j)\ge 5$.
Since $v_2$ is $([8, 11], 5^-)$ and $d(v_1)\ge 9$,
there exists $j, k\in T_1$ such that $d(w_j) = d(w_k) = 5$,
i.e., $C(w_j)\setminus \{j\} = C(w_k)\setminus \{k\} = \{2\}\cup [5, 7]$.
Then $(v_2 w_j, v_2 w_k, v v_2, v v_1)\to (k, j, 4, k)$ reduces the proof to the case where $a\in S_{5^-}(u)$.

{\bf Case 2.} $(v v_1, vv_2)_c = (2, 8)$.

In this case, $[8, \Delta + 5]\subseteq C(u_2)\cap C(v_1)$.
It follows that $C(v_1) = \{1, 2\}\cup [8, \Delta + 5]$ with $c(v_1 v_2)\in [9, \Delta + 5]$, say $c(v_1 v_2) = 9$.
Since $9\in B_2$, we have $2\in C(v_2)$.

When there exists an $i\in [3, 7]\setminus C(v_2)$, $(v v_1, u v)\overset{\divideontimes} \to (i, 9)$.
In the other case, $[3, 7]\subseteq C(v_2)$ and then $[2, 9]\subseteq C(v_2)$.
It follows that $d(v_2)\ge 8$ and let $T_2 = ([10, \Delta + 5]\setminus C(v_2))\cup \{8\}$.
One sees that $|T_2| = \Delta + 5 - 9 - (d(v_1) - 8) + 1 = \Delta - d(v_1) + 5 \ge 5$.

Let $c(v_2 w_i) = i \in [2, 7]$.
If $G$ contains no $(i, j)_{(v_1, v_2)}$-path for some $i\in [5, 7]$ and $j\in T_2$,
then $(vv_1, vv_2)\to (i, j)$ reduces the proof to the case where $a\in S_{5^-}(u)$.
Otherwise, $G$ contains a $(i, j)_{(v_1, v_2)}$-path for each $i\in [5, 7]$ and $j\in T_2$.
It follows that for any $i\in [5, 7]$, $T_2\subseteq C(w_i)$ and $d(w_i)\ge 6$.
Since $v_2$ is $([8, 11], 5^-)$ and $d(v_1)\ge 9$,
there exists $i\in \{3, 4\}$ such that $d(w_i)\le 5$, assuming $d(w_3)\le 5$.
One sees that an acyclic edge coloring of $H$ will be produced if let $v v_1\to 3$ and $v v_2\to T_2\setminus C(w_3)$.
Hence, we may assume that $G$ contains a $(3, i)_{(u_3, v_1)}$-path for any $i\in [8, \Delta + 5]$.

For brief description, we introduce a few symbols used in \cite{BLSNHT11}.
A {\em multiset} is a generalized set where a member can appear multiple times.
If an element $x$ appears $t$ times in the multiset $MS$, then we say that the {\em multiplicity} of $x$ in $MS$ is $t$, and write mult$_{MS}(x) = t$.
The {\em cardinality} of a finite multiset $M S$, denoted by $\|M S\|$, is defined as $\|M S\|= \sum_{x\in M S}$mult$_{M S}(x) = t$.
Let $M S_1$ and $M S_2$ be two multisets.
The {\em join} of $M S_1$ and $M S_2$, denoted $M S_1\biguplus M S_2$, is a multiset that has all the members of $M S_1$ and of $M S_2$.
For $x\in M S_1 \biguplus MS_2$, mult$_{M S_1 \biguplus M S_2}(x) = $ mult$_{M S_1}(x) + $ mult$_{M S_2}(x)$.
Clearly, $\|M S_1 \biguplus M S_2\| = \|M S_1\| + \|M S_2\|$.

For $i\in [5, 7]$, one sees from $d(u_i)\le 5$ that $[8, \Delta + 5]\setminus C(u_i)\ne \emptyset$,
and let $T_i = [8, \Delta + 5]\setminus C(u_i)$.

Let
\[
S_u = \biguplus_{i\in [5, 7]}(C(u_i)\setminus \{c(u u_i)\}).
\]

For each $j\in [8, \Delta + 5]$, when $j\not\in S_u$,
there exists an $i\in [5, 7]$ such that $G$ contains neither a $(4, j)_{(u_4, u_i)}$-path nor a $(1, j)_{(v_1, u_i)}$-path,
thus $(u u_i, u v)\overset{\divideontimes}\to (j, i)$.
For each $j\in T_5$, if $G$ contain no $(j, k)_{(u_5, u_k)}$-path for each $k\in \{4, 6, 7\}$,
then $(u u_5, u v_1, u v)\overset{\divideontimes}\to (j, 5, 1)$.

If $1\not\in C(v_2)$ and $1\not\in C(u_i)$ for some $i\in [5, 7]$,
then $(v v_1, v v_2)\to (i, 1)$ and $uv \overset{\divideontimes} \to T_2$.

In the other case, we proceed with the following proposition:
\begin{proposition}
\label{prop3121}
\begin{itemize}
\parskip=0pt
\item[{\rm (1)}] For any $j\in [8, \Delta + 5]$, mult$_{S_u}(j)\ge 1$.
	
\item[{\rm (2)}] For each $j\in T_i$, $i\in [5, 7]$, $G$ contain a $(j, a_i)_{(u_i, u_{a_i})}$-path for some $a_i\in ([4, 7]\setminus \{i\})\cap C(u_i)$.

\item[{\rm (3)}] If $1\not\in C(v_2)$, then $1\in C(u_5)\cup C(u_6)\cup C(u_7)$.
\end{itemize}
\end{proposition}

For each $i \in [5, 7]$, let $s_i = |[8, \Delta + 5]\cap C(u_i)|$,
it follows from \ref{prop3121}(2) that $s_i \le 3$
and from \ref{prop3121}(1) that $[8, \Delta + 5]\subseteq S_u$.
Thus, there exists an $i_0\in [8, \Delta + 5]$ such that mult$_{S_u}(i_0) = 1$,
with $i_0\in C(u_7)\setminus (C(u_5)\cup C(u_6))$.
Then it follows from \ref{prop3121}(2) that $G$ contains a $(4, i_0)_{(u_4, u_{j_1})}$-path and a $(1, i_0)_{(u_7, u_{j_2})}$-path, where $\{j_1, j_2\} = \{5, 6\}$.

One sees that if $G$ contains no $(4, i)_{(v_1, v_2)}$-path for some $i\in T_2$,
then we can obtain that $G$ contains a $(4, j)_{(u_4, v_1)}$-path for any $j\in [8, \Delta + 5]$ by letting $(v v_1, v v_2)\to (4, i)$.

Hence, $G$ contains a $(4, i)_{(v_1, v_2)}$-path for each $i\in T_2$ and $T_2\subseteq C(w_4)$.
It is clearly that $d(w_4)\ge 6$.
Since $v_2$ is $([8, 11], 5^-)$ and $d(v_1) = \Delta \ge 8$,
we have $d(x)\le 5$ for any $x\in N(v_2)\setminus \{v_1, w_4, w_5, w_6, w_7\}$.

Now we discuss the following two subcases for two possible values of mult$_{S_u}(1)$, respectively.

{Case 2.1.} mult$_{S_u}(1)\ge 2$.
In this case, assume w.l.o.g. $1\in C(u_5)\cap C(u_6)$.
One sees that $s_5\le 2$, $s_6\le 2$ and $s_7\le 3$.
It follows from $|[8, \Delta + 5]| = \Delta - 2\ge 6$ that there exists an $i\in \{5, 6\}$ such that $s_i = 2$,
assuming w.l.o.g. $s_5 = 2$ and $C(u_5) = \{5, a_5, 1, \alpha_8, \alpha_9\}$, where $\{\alpha_8, \alpha_9, \ldots, \alpha_{\Delta + 5}\} = [8, \Delta + 5]$.

Since $s_5 + s_6\le 2 + 2 = 4 < 6$ and $s_5 + s_7\le 2 + 3 = 5 < 6$,
it follows from \ref{prop3121}(2) that $a_5 = 4$ and $G$ contains a $(4, i)_{(u_4, u_5)}$-path for any $i\in [8, \Delta + 5]\setminus \{\alpha_8, \alpha_9\}$,
$a_6 = 7$ and $G$ contains a $(7, j)_{(u_6, u_7)}$-path for any $j\in [8, \Delta + 5]\setminus (C(u_5)\cup C(u_6))$.
Thus, since $s_6 + s_7\le 2 + 3 = 5 < 6$, it follows from \ref{prop3121}(2) that $4/5\in C(u_6)$.

Hence, we have $\Delta = 8$ and $C(u_6)= \{1, 7, 4/5, \alpha_{10}\}$\footnote{It means that $C(u_6)= \{1, 7, x, \alpha_{10}\}$ where $x\in \{4, 5\}$},
$C(u_7) = \{a_7, \alpha_{11}, \alpha_{12}, \alpha_{13}\}$.
Since $G$ contains a $(4, \alpha_{10})_{(u_4, u_5)}$-path, we have $a_7 = 6$.
One sees clearly that $G$ contains no $(\alpha_8, i)_{(u_i, u_7)}$-path for any $i\in \{4, 5, 6\}$,
a contradiction to \ref{prop3121}(2).

{Case 2.2.} mult$_{S_u}(1)\le 1$.
In this case, it follows from \ref{prop3121}(3) that $1\in C(v_2)$ with $c(v_2w_1) = 1$,
and $\Delta \ge 9$, $|[8, \Delta + 5]| = \Delta - 2 \ge 7$.
Recall that $d(w_1)\le 5$ and $T_2\setminus C(w_1) \ne \emptyset$.

When $1\in C(u_2)\cap C(u_3)$, $(u u_2, u u_3, u v)\overset{\divideontimes}
\to (3, 2, 9)$.
In the other case, $1\not\in C(u_i)$ for some $i\in \{2, 3\}$, assuming w.l.o.g. $1\not\in C(u_3)$.
When mult$_{S_u}(1)= 0$, there exists an $i\in [5, 7]$, assuming w.l.o.g. $i = 5$,
such that $G$ contains no $(1, j)_{(u_i, u)}$-path for any $j\in \{2, 4, 6, 7\}$,
thus $(u u_5, u v_1, v v_1)\to (1, 5, 1)$ and $v v_2\to T_2\setminus C(w_1)$ reduces the proof to the case where $a\in S_{5^-}(u)$.

In the other case, mult$_{S_u}(1) = 1$ and assume w.l.o.g. $1\in C(u_5)\setminus (C(u_6)\cup C(u_7))$.
It follows that $s_5\le 2$.
Since $|[8, \Delta + 5]|\ge 7$, there exists an $i\in \{6, 7\}$ such that $s_i = 3$.
Assume w.l.o.g. $s_6 = 3$ and $C(u_6) = \{6, a_6, \alpha_8, \alpha_9, \alpha_{10}\}$.

Since $s_5 + s_6\le 2 + 3 = 5 < 7$ and $s_6 + s_7\le 3 + 3 = 6 < 7$,
it follows from \ref{prop3121}(2) that $a_6 = 4$ and $G$ contains a $(4, i)_{(u_4, u_6)}$-path for any $i\in [8, \Delta + 5]\setminus \{\alpha_8, \alpha_9, \alpha_{10}\}$,
$a_5 = 7$ and $G$ contains a $(7, j)_{(u_5, u_7)}$-path for any $j\in [8, \Delta + 5]\setminus (C(u_5)\cup C(u_6))$.
Thus, since $s_5 + s_7\le 2 + 3 = 5 < 7$, it follows from \ref{prop3121}(2) that $4/6\in C(u_5)$.

Hence, we have $\Delta = 9$ and $C(u_5)= \{1, 7, 4/6, \alpha_{11}\}$,
$C(u_7) = \{a_7, \alpha_{12}, \alpha_{13}, \alpha_{14}\}$.
Since $G$ contains a $(4, \alpha_{11})_{(u_4, u_6)}$-path, we have $a_7 = 5$.
One sees clearly that $G$ contains no $(\alpha_8, i)_{(u_i, u_7)}$-path for any $i\in \{4, 5, 6\}$,
a contradiction to \ref{prop3121}(2).


\subsubsection{Configuration ($A_{2.4}$), ($A_{3}$)--($A_{5}$)}
In this case, $C(u) = [1, k]$ with $(u x_1, u x_2)_c = (1, 2)$.
When $a\in \{1, 2\}$, say $(v x_1, v x_2)_c = (2, k)$,
since $G$ contains no $(2, k)_{(x_1, x_2)}$-path,
by Lemma~\ref{3akcycle}, we can get an acyclic edge $(\Delta + 5)$-coloring for $G$.

In the other case, $a\in [3, k - 1]\setminus S_{5^-}(u)$,
assuming w.l.o.g. $a = 3$.
It follows that $[k, \Delta + 5]\subseteq C(x_1)\cap C(u_3)$.
We first summarize the special cases in the following Lemma~\ref{lemma17}.
\begin{lemma}
\label{lemma17}
If re-coloring some edges $E'\subseteq E(H)$ incident to $u$ or $v$ does not produce any new bichromatic cycles,
but at least one of the following holds:
\begin{itemize}
\parskip=0pt
\item[{\rm (1)}]
	$C(u)\cap C(v) = \emptyset$,
\item[{\rm (2)}]
	$C(u)\cap C(v) = \{x\}$ with $x\in \{1, 2\}\cup S_{5^-}(u)$.
\end{itemize}
then an acyclic edge $(\Delta + 5)$-coloring for $G$ can be obtained in $O(1)$ time.
\end{lemma}

When $2\not\in C(x_1)$, $v v_1\to 2$ and by Lemma~\ref{lemma17}, we are done.
In the other case, $2\in C(x_1)$, and $d(x_1) \ge |[k, \Delta + 5]\cup \{1, 3, 2\}| = \Delta - k + 9\ge \Delta - 2$.
It follows that $k\in [9, 11]$ and ($A_{3}$)--($A_{5}$) holds.

Hereafter, since $u$ is $([9, 11], 5^-)$, let $u_i$ be a $5^-$-vertex, $i\in [6, k - 1]$.
One sees clearly that $|S_{5^-}(u)|\ge k - 1 - 5 = k - 6$,
and $|[k, \Delta + 5]| = \Delta + 6 - k$.
Thus $|S_{5^-}(u)\setminus C(x_1)| \ge |[k, \Delta + 5]\cup \{1, 2, 3\}\cup S_{5^-}(u)| - d(v_1) \ge \Delta + 6 - k + 3 + k - 6 - d(v_1) = \Delta - d(v_1) + 3 \ge 3$.

Let $S_1 = S_{5^-}(u)\setminus C(x_1)$ and assume w.l.o.g. $6, 7, 8\in S_1$.
For any $i\in [3, k - 1]$, $T_i = [k, \Delta + 5]\setminus C(u_i)$ and $T_2 = ([k, \Delta + 5]\setminus C(x_2))\cup \{k\}$.
Let $y_i$ be a neighbor of $x_1$ with $c(x_1 y_i) = i$,
$z_j$ be a neighbor of $x_2$ with $c(x_2 z_j) = j$.
In particular, $c(x_1 y_i) = i$, $i\in \{1, 2\}\cup [k, \Delta + 5]$,
and $c(x_2 z_2) = 2$, where $y_1 = z_2 = u$.

If $3\not\in C(x_2)$, then $(vx_1, vx_2)\to (8, 3)$ and $uv \to [k, \Delta + 5]\setminus C(u_8)$.
One sees that if there exists an $i\in C(u)\setminus C(x_1)\cup C(x_2)$,
then we can let $v x_1\to i$ and the same argument will apply,
thus an acyclic edge coloring of $G$ will be obtained.
For each $i\in S_1$, if $G$ contains no $(i, j)_{(x_1, x_2)}$-path for some $j\in T_2$,
then $(v x_1, v x_2)\to (i, j)$ and by Lemma~\ref{lemma17}, we are done.

In the other case, we proceed with the following proposition:
\begin{proposition}
\label{ijT2}
\begin{itemize}
\parskip=0pt
\item[{\rm (1)}]
	$3\in C(x_2)$ and $C(u)\subseteq C\subseteq C(x_1)\cup C(x_2)$;
\item[{\rm (2)}]
	For each $i\in S_1$, $j\in T_2$, $G$ contains a $(i, j)_{(x_1, x_2)}$-path,
    and $S_1\subseteq C(x_2)$, $T_2\subseteq C(z_i)$, $[6, 8]\subseteq S_1\subseteq C(y_j)$.
\end{itemize}
\end{proposition}

One sees that $|T_2|\ge (2\times (|[k, \Delta + 5]|)) + |\{1, 3, 2\}| + |\{2, 3\} + |[4, k - 1]| - (d(x_1) + d(x_2)) + |\{k\}| = 2\Delta + 14 - k - (d(x_1) + d(x_2))
	= (\Delta - d(x_1)) + (\Delta - d(x_2)) + 14 - k \ge \Delta - d(x_2) + 14 - k\ge 14 - k \ge 3$,
and if $1\in C(x_2)$, then $|T_2|\ge (2\times (|[k, \Delta + 5]|)) + |\{1, 3, 2\}| + |\{2, 3, 1\} + |[4, k - 1]| - (d(x_1) + d(x_2)) + |\{k\}| \ge \Delta - d(x_2) + 15 - k$.
Thus, $T_2\setminus\{k\} \ne \emptyset$.

When $d(y_i) = 2$ for some $i\in [k, \Delta + 5]$, i.e., $C(y_i) = \{i, 3\}$,
$x_1y_i\to C\setminus C(x_1)$, $uv \to j$, and if $i = k$, $v x_2\overset{\divideontimes}\to T_2\setminus \{k\}$.
When $d(y_j)\le 5$ for some $j\in T_2$, i.e., $C(y_j) = \{i, 3, 6, 7, 8\}$,
$x_1 y_j\to C\setminus C(x_1)$, $u v \to j$, and if $j = k$, $v x_2\overset{\divideontimes}\to T_2\setminus \{k\}$.
For each $i\in S_1$, if $G$ contains no $(i, j)_{(x_1, x_2)}$-path for some $j\in (C\setminus C(x_2))\setminus T_2$,
then $(v x_1, v x_2)\to (i, j)$ and $uv\overset{\divideontimes}\to T_2$.

In the other case, we proceed with the following proposition:
\begin{proposition}
\label{ijC6+x1x2}
\begin{itemize}
\parskip=0pt
\item[{\rm (1)}]
	For each $i\in S_1$, $j\in C\setminus C(x_2)$, $G$ contains a $(i, j)_{(x_1, x_2)}$-path,
    and $(C\setminus C(x_2))\subseteq C(z_i)$, $[6, 8]\subseteq S_1\subseteq C(y_j)$, $d(y_j)\ge |S_1| + 1$;
\item[{\rm (2)}]
	For any $j\in [k, \Delta + 5]$, $d(y_j)\ge 3$, and if $j\in T_2$, $d(y_j)\ge 6$;
\item[{\rm (3)}]
    $n_{6^+}(x_1)\ge |T_2| + |\{u, x_2\}|= |T_2| + 2\ge \Delta - d(x_2) + 14 - k + 2 = \Delta - d(x_2) + 16 - k$;
\item[{\rm (4)}]
    If $1\in C(x_1)$, then $n_{6^+}(x_1)\ge |T_2| + |\{u, x_2\}|= |T_2| + 2\ge \Delta - d(x_2) + 17 - k$;
\item[{\rm (5)}]
	For each $i\in S_1$, $d(z_i)\ge |C\setminus C(x_2)| + |\{i\}| = \Delta + 5 - d(x_2) + 1 = \Delta - d(x_2) + 6\ge 6$;
\item[{\rm (6)}]
	$n_{6^+}(x_2)\ge |S_1| + |\{u, x_1\}|\ge |S_1| + 2\ge 5$.
\end{itemize}
\end{proposition}

When there exists an $j\in (S_{5^-}(u)\cap C(x_1))\setminus C(x_2)$,
it follows that $G$ contains a $(i, j)_{(x_1, x_2)}$-path for each $i\in S_1$, and $S_1\subseteq C(y_j)$,
if $d(y_j) \le |S_1| + 2$ with $3\in C(y_j)$, i.e., $C(y_j) = \{3, j\}\cup S_1$,
then $x_1 y_j\to C\setminus C(x_1)$, $v x_1 \to j$, and by Lemma~\ref{lemma17}, we are done;
if $3\not\in C(y_j)$ and $G$ contains no $(3, i)_{(y_i, x_1)}$-path for any $i\in C(x_1)\setminus [k, \Delta + 5]$,
$x_1 y_j\to 3$, $v x_1 \to j$, and by Lemma~\ref{lemma17}, we are done.
In the other case, we proceed with the following proposition:
\begin{proposition}
\label{jx1notx2}
For each $j\in (S_{5^-}(u)\cap C(x_1))\setminus C(x_2)$,
\begin{itemize}
\parskip=0pt
\item[{\rm (1)}]
	$d(y_j)\ge |S_1| + 3\ge 6$, or
\item[{\rm (2)}]
	$d(y_j) = |S_1| + 2$, $3\not\in C(y_j)$, and $G$ contains a $(3, i)_{(y_j, x_1)}$-path for some $i\in C(x_1)\setminus [k, \Delta + 5]$.
\end{itemize}
\end{proposition}

For $i\in (C(u)\setminus (S_{5^-}(u)\cup C(x_1)))\cap C(x_2)$,
if $G$ contains no $(i, j)_{(x_1, x_2)}$-path for some $j\in T_2$, then $(v x_1, v x_2)\to (i, j)$,
and thus an acyclic edge coloring of $G$ can be obtained if $G$ contains no $(i, k)_{(u_i, x_1)}$-path for some $k\in [k, \Delta + 5]$.
Hence, we proceed with the following proposition:
\begin{proposition}
\label{9ij}
For each $i\in (C(u)\setminus (S_{5^-}(u)\cup C(x_1)))\cap C(x_2)$,
\begin{itemize}
\parskip=0pt
\item[{\rm (1)}]
	$G$ contains a $(i, j)_{(x_1, x_2)}$-path for each $j\in T_2$, $T_2\subseteq C(z_i)$,
    and $d(z_i)\ge |T_2| + 1$, or
\item[{\rm (2)}]
	if $G$ contains no $(i, j)_{(x_1, x_2)}$-path for some $j\in T_2$ or $d(z_i)\le |T_2|$,
    then $G$ contains a $(i, j)_{(u_i, v_1)}$-path for each $j\in [k, \Delta + 5]$.
\end{itemize}
\end{proposition}

For each $i \in S_{5^-}(u)$, let $s_i = |[k, \Delta + 5]\cap C(u_i)|$.
Let $S_2$ be the set of $\{i \in S_{5^-}(u)\cap C(x_1)\}$ such that $G$ contains no $(3, i)_{(u_3, x_1)}$-path,
and
\[
S_u = \biguplus_{i\in S_1\cup S_2}(C(u_i)\setminus \{c(u u_i)\}).
\]

When for some $i\in S_1\cup S_2$, $G$ contain no $(j, k)_{(u_i, u_k)}$-path for each $j\in T_i$, $k\in C(u)\setminus \{3, i\}$,
then $(uu_i, uv)\overset{\divideontimes} \to (j, i)$.
In the other case, we proceed with the following proposition:
\begin{proposition}
\label{S1S2}
Let $i\in S_1\cup S_2$.
For each $j\in T_i$,
$G$ contain a $(j, a_i)_{(u_i, u_{a_i})}$-path for some $a_i\in (C(u)\setminus \{3, i\})\cap C(u_i)$ and $s_i\le 3$.
\end{proposition}

{\bf Case 1.} $u$ is $(9, 5^-)$ and at least one of ($A_{3.1}$)--($A_{3.8}$) holds.

In this case, $C(u) = [1, 8]$.
It follows from Propositions~\ref{ijT2} and \ref{ijC6+x1x2} that
\begin{proposition}
\label{91}
\begin{itemize}
\parskip=0pt
\item[{\rm (1)}]
	$C(x_1) = [1, 3]\cup [9, \Delta + 5]$ and $[2, 9]\subseteq C(x_2)$ with $c(x_1 x_2)\in [10, \Delta + 5]$.
\item[{\rm (2)}]
	$d(x_1) = \Delta$ and $n_{6^+}(x_1)\ge |T_2| + 2 \ge \Delta - d(x_2) + 16 - k = \Delta - d(x_2) + 16 - 9 = \Delta - d(x_2) + 7\ge 7$;
\item[{\rm (3)}]
	If $1\in C(x_2)$, then $n_{6^+}(x_1)\ge \Delta - d(x_2) + 17 - k = \Delta - d(x_2) + 17 - 9 = \Delta - d(x_2) + 8\ge 8$;
\item[{\rm (4)}]
	For any $j\ne 2$, $d(y_j)\ge 3$;
\item[{\rm (5)}]
    $n_{6^+}(x_2)\ge |S_1| + |\{u, x_1\}|\ge |S_1| + 2 = |S_{5^-}(u)| + 2 = n_{5^-}(u) - 1 + 2 = n_{5^-}(u) + 1\ge 4 + 1 = 5$.
\end{itemize}
\end{proposition}

When $d(y_2) = 2$, let $y'_2$ be the neighbor of $y_2$ other than $x_1$,
if $c(y_2 y'_2) = 3$, then $(x_1 y_2, v x_1)\to (6, 2)$,
if $c(y_2 y'_2)\not\in \{1, 3\}$, then $(x_1 y_2, v x_1)\to (3, 2)$,
otherwise, $y_2 y'_2\to C\setminus (C(y'_2)\cup \{3\})$, and $(x_1 y_2, v x_1)\to (3, 2)$,
thus, by Lemma~\ref{lemma17}, we are done.

In the other case, we proceed with the following proposition:
\begin{proposition}
\label{92}
For any $y_j\in N(x_1)$, $d(y_j)\ge 3$, i.e., $n_2(x_1) = 0$, and $d(x_1) - n_2(x_1) = \Delta$.
\end{proposition}

It follows from Propositions~\ref{91} and \ref{92} that $v_1 = x_1$, $v_2 = x_2$ and ($A_{3.4}$),
($A_{3.5}$) where $v_2$ is $(14, 6/7)$, ($A_{3.7}$), or ($A_{3.8}$) holds.
Since $v_1$ is $(\Delta, 7)$, $\Delta\in [10, 14]$, by Propositions~\ref{91}(3),
we have $1\not\in C(v_2)$, $|T_2| = 5$, and for each $i\in S_1 = S_{5^-}(u)$, $G$ contains a $(1, i)_{(v_1, v_2)}$-path, $1\in C(u_i)$.

Now assume w.l.o.g. $c(x_1 x_2) = 14$, $T_2 = [9, 13]$ and $C(x_2) = [2, 9]\cup [14, \Delta + 5]$.
It follows from $1\in C(u_i)$ that $s_i \le 3$ and
\[
S_u = \biguplus_{i\in S_{5^-}(u)}(C(u_i)\setminus \{c(u u_i)\}).
\]

{Case 1.1.} $u$ is $(9, 4^-)$ with $[5, 8]\subseteq S_{5^-}(u)$ and $\Delta \in [11, 14]$.
When there exists a $j\in [9, \Delta + 5]\setminus S_u$,
one sees clearly that there exists an $i\in [5, 8]$ such that $G$ contains no $(j, k)_{(u_j, u)}$-path for any $k\in C(u)$, $(u u_i, u v)\overset{\divideontimes}\to (j, i)$.
In the other case, for $j\in [9, \Delta + 5]$, mult$_{S_u}(j)\ge 1$, i.e., $[9, \Delta + 5]\subseteq S_u$.
Note that $|[9, \Delta + 5]| = \Delta - 3$,
and if $\Delta \ge 13$, then $|[9, \Delta + 5]| \ge 10 > 3\times 2 + 1\times 3 = 9$.

(1.1.1.)\ $\Delta \ge 13$.
Then there exists $i, j\in [5, 8]$ such that $s_i = s_j = 3$.
Assuming w.l.o.g. $C(u_{i_1}) = \{1, i_1, \alpha_9, \alpha_{10}, \alpha_{11}\}$,
$C(u_{i_2}) = \{1, i_2, \alpha_{12}, \alpha_{13}, \alpha_{14}\}$,
where $\alpha_i\in [9, \Delta + 5]$, $i\in [9, 14]$.

Let $\alpha\in [9, \Delta + 5]\setminus (C(u_{i_1})\cup C(u_{i_2}))$.
By Lemma~\ref{lemma06}, there exists an $i\in \{i_1, i_2\}$ such that $G$ contains no $(1, \alpha)_{(u_{i}, v_1)}$-path, say $G$ contains no $(1, \alpha)_{(u_{i_1}, v_1)}$-path.
Then $(u u_{i_1}, uv)\overset{\divideontimes}\to (\alpha, i_1)$.

(1.1.2.)\ ($A_{3.4}$) holds and $v_2$ is $(\Delta, 6)$ with $\Delta \in [11, 12]$.
Then $d(z_4)\le 5$ and it follows from \ref{9ij}(2) that $G$ contains a $(4, i)_{(u_4, v_1)}$-path for each $i\in [9, \Delta + 5]$.

Since $|[14, \Delta + 5]| + 2\times |[9, 13]| = 
\Delta + 2\ge 11 + 2 = 13 > 12$,
there exists a $j\in [9, 13]$ such that mult$_{S_u}(j) = 1$, assuming w.l.o.g. mult$_{S_u}(10) = 1$\footnote{If mult$_{S_u}(9) = 1$, $vv_2\to 10$} and $10\in C(u_5)$.
By Lemma~\ref{lemma06}, there exists an $i\in [6, 7, 8]$ such that $G$ contains no $(10, j)_{(u_{i}, u)}$-path for any $j\in C(u)$.
Then $(u u_i, u v)\overset{\divideontimes}\to (10, i)$.

{Case 1.2.} $u$ is $(9, 5)$, and $v_2$ is $(\Delta, 5)$, $\Delta \in [10, 13]$,
or $v_2$ is $(\Delta, 6)$, $\Delta \in [12, 13]$.
Since $n_{6^+}(v_2)\le 6$ and $d(z_i)\ge 6$, $i\in [6, 8]$,
there exists an $i\in \{4, 5\}$ such that $d(z_i)\le 5$,
say $d(z_4)\le 5$ and it follows from \ref{9ij}(2) that $G$ contains a $(4, j)_{(u_4, v_1)}$-path for each $j\in [9, \Delta + 5]$.

When there exists a $j\in [9, 14]\setminus S_u$,
one sees clearly that there exists an $i\in [6, 8]$ such that $G$ contains no $(j, k)_{(u_j, u)}$-path for any $k\in C(u)$, $(u u_i, u v)\overset{\divideontimes}\to (j, i)$.

In the other case, we proceed with the following proposition:
\begin{proposition}
\label{95678j}
For $j\in [9, 14]$, mult$_{S_u}(j)\ge 1$, i.e., $[9, 14]\subseteq S_u$.
\end{proposition}

(1.2.1.)\ $[15, \Delta + 5]\subseteq S_u$.
Note that $|[9, \Delta + 5]| = \Delta - 3\ge 10 - 3 = 7$.
Here we do not distinguish the colors of $\{9, \ldots, \Delta + 5\}$, and refer them to as $\alpha_9, \ldots, \alpha_{(\Delta + 5)}$, respectively.%
\footnote{That is, $(\alpha_9, \ldots, \alpha_{(\Delta + 5)})$ is an arbitrary permutation of $(9, \ldots, \Delta + 5)$.}
For each $i\in [6, 8]$, one sees from $1\in C(u_i)$ that $s_i\le 3$.
It follows from \ref{95678j} that there exists an $i\in [6, 8]$ such that $s_3 = 3$,
assuming w.l.o.g. $s_6 = 3$.
Then by \ref{S1S2}, $G$ contains a $(1, j)_{(u_6, v_1)}$-path for each $j\in T_6$,
and since $1\not\in C(v_2)$, we have $14\in C(u_6)$.

Since $s_6 + s_7\le 6 < 7$ and $s_6 + s_8\le 6 < 7$,
it follows from \ref{S1S2} that $s_7 = s_8 = 2$,
$\Delta = 10$ and $v_2$ is $(\Delta, 5)$.
Thus, $d(z_5)\le 5$ and it follows from \ref{9ij} that $G$ contains a $(5, i)_{(u_5, v_1)}$-path for each $i\in [9, \Delta + 5]$.
Then $C(u_6)= \{6, 1, 14, \alpha_9, \alpha_{10}\}$, $C(u_7) = \{7, 1, a_7, \alpha_{11}, \alpha_{12}\}$,
$C(u_8) = \{8, 1, a_8, \alpha_{13}, \alpha_{15}\}$, where $a_7, a_8\in \{2, 6, 7, 8\}$.

Since for each $i\in [7, 8]$,
$[9, 13]\setminus (C(u_6)\cup C(u_i)))\ne \emptyset$ and $[9, 15]\setminus (C(u_6)\cup C(u_i)))\ne \emptyset$,
it follows from \ref{S1S2} that $a_7 = 8$, $a_8 = 7$.
However, there exists an $i\in \{7, 8\}$ such that
$G$ contains no $(\alpha_9, j)_{(u_i, u)}$-path for each $j\in C(u)$,
a contradiction to \ref{S1S2}.

(1.2.2.)\ $[15, \Delta + 5]\setminus S_u \ne \emptyset$, assuming w.l.o.g. $15\not\in S_u$.
It follows from \ref{S1S2} that for each $i\in [6, 8]$,
$G$ contains a $(15, j)_{(u_i, u)}$-path for some $j\in \{1, 2, 5\}$.
By Lemma~\ref{lemma06}, there exists an $i\in [6, 8]$ such that $5\in C(u_i)$, and $G$ contains a $(5, 15)_{(u_5, u_i)}$-path.
It follows from \ref{9ij} that $G$ contains a $(5, j)_{(v_1, v_2)}$-path for each $j\in T_2$ and $d(z_5)\ge |T_2| + 1\ge 6$.
One sees clearly from Lemma~\ref{lemma06} that $G$ contains no $(i, j)_{(u, v_2)}$-path for each $i\in \{1\}\cup [3, 8]$ and $j\in T_2$.

Let $uv_2\to 10$.
Since $v_1$ is $(\Delta, 7)$, i.e., $d(y_2)\le 5$,
by the similar discussion to show \ref{ijC6+x1x2}(2),
an acyclic edge $(\Delta + 5)$-coloring for $G$ can be obtained.

{\bf Case 2.} $u$ is $(10, 5^-)$ and at least one of ($A_{4.1}$)--($A_{4.4}$) holds.

In this case, $C(u) = [1, 9]$.
Let $u_i$ be a $6^+$-vertex, $i\in [3, 5]$ and $u_j$ be a $5^-$-vertex, $j\in [6, 9]$.
It follows from \ref{ijC6+x1x2} that $n_{6^+}(x_1)\ge 6$ and $n_{6^+}(x_2)\ge 5$.

Hereafter, we assume ($A_{4.2}$), ($A_{4.3}$), or ($A_{4.4}$) holds.
One sees that $|T_2| = |[10, \Delta + 5]\setminus C(x_2)| + |\{10\}| = |[10, \Delta + 5]\setminus C(x_2)| + 1$.
Since $[10, \Delta + 5]\cup [1, 3]\subseteq C(x_1)$, $d(x_1)\ge 3 + \Delta + 5 - 9 = \Delta - 1$.

When $d(x_1) = \Delta - 1$, it follows that $C(x_1) = [1, 3]\cup [10, \Delta + 5]$,
$c(x_1 x_2)\in [11, \Delta + 5]$, assuming w.l.o.g. $c(x_1x_2) = 11$.
Since $[2, 11]\subseteq C(x_2)$,
$n_{6^+}(x_1)\ge |T_2| + 2\ge (\Delta + 5 -11) - (d(v_2) - 10) + 1 + 2 = \Delta -d(v_2) + 7$.
Since $S_1 = S_{5^-}(u)$,
it follows from \ref{ijC6+x1x2}(6) that $n_{6^+}(x_2)\ge |S_1| + 2 = |S_{5^-}(u)| + 2 \ge 4 + 2 = 6$.
Thus, none of ($A_{4.2}$), ($A_{4.3}$), and ($A_{4.4}$) holds.

When $d(x_1) = \Delta$,
we discuss the following two subcases where $S_{5^-}(u)\cap C(x_1)= 0, 1$, respectively.

{Case 2.1.} $S_{5^-}(u)\cap C(x_1)= 0$, assuming w.l.o.g. $4\in C(x_1)$.

In this case, $C(x_1)= [1, 4]\cup [10, \Delta + 5]$,
and $\{2, 3, c(x_1 x_2)\}\cup [5, 10]\subseteq C(x_2)$.
One sees clearly that $S_1 = [6, 9]\subseteq C(x_2)$.
It follows from \ref{ijC6+x1x2} (6) that $n_{6^+}(x_2)\ge |S_1| + 2 = 4 + 2 = 6$.

(2.1.1.) $c(x_1 x_2) = 4$.
Since $[2, 10]\subseteq C(x_2)$,
$|T_2|\ge \Delta + 5 -10 - (d(x_2) - 9) + 1 = \Delta - d(x_2) + 5$,
and thus $n_{6^+}(x_1)\ge |T_2| + 2\ge \Delta - d(x_2) + 7$.

(2.1.2.) $c(x_1x_2) \in [11, \Delta + 5]$, assuming w.l.o.g. $c(x_1x_2) = 11$.
Since $([2, 3]\cup [5, 11])\subseteq C(x_2)$,
$|T_2|\ge \Delta + 5 -11 - (d(x_2) - 9) + 1 = \Delta - d(x_2) + 4$,
and if $1/4\in C(x_2)$, $|T_2|\ge \Delta + 5 -11 - (d(x_2) - 10) + 1 = \Delta - d(x_2) + 5$.
Thus, when $1/4\in C(x_2)$, $n_{6^+}(x_1)\ge |T_2| + 2\ge \Delta - d(x_2) + 7$.

In the other case, $1, 4\not\in C(x_2)$.
It follows from \ref{ijC6+x1x2}(1) that $S_1\subseteq C(y_4)$ and $d(y_4)\ge |S_1| + 1 = 5$.
If $d(y_4) = 5$, i.e., $C(y_4)= \{4\}\cup [6, 9]$,
then $(x_1y_4, vx_1, vx_2, uv)\overset{\divideontimes}\to (5, 6, 4, 10)$.
Otherwise, $d(y_4)\ge 6$.
Then $n_{6^+}(x_1)\ge |T_2| + |\{u, x_2, y_4\}| = |T_2| + 3 \ge \Delta - d(x_2) + 4 + 3 = \Delta - d(x_2) + 7$.

By the two subcases discussed above,
one sees that $n_{6^+}(x_2)\ge 6$ and $n_{6^+}(x_1)\ge \Delta - d(x_2) + 7\ge 7$.
It follows that ($A_{4.2}$) holds.
However, $v_1 = x_1$, $v_2 = x_2$, and $n_{6^+}(v_1)\ge \Delta - d(v_2) + 7$, a contradiction.

{Case 2.2.} $S_{5^-}(u)\cap C(x_1) = 1$, assuming w.l.o.g. $9\in C(x_1)$.

In this case, $C(x_1) = [1, 3]\cup [9, \Delta + 5]$,
and $([2, 8]\cup \{10, c(x_1 x_2)\})\subseteq C(x_2)$.
One sees clearly that $S_1 = S_{5^-}(u)\setminus C(x_1) = [6, 8]$.
It follows from \ref{ijC6+x1x2} (6) that $n_{6^+}(x_2)\ge |S_1| + 2 = |S_{5^-}(u)| + 2 \ge 3 + 2 = 5$,
and from \ref{ijC6+x1x2}(2) that $d(y_j) \ge 3$, $j\in [10, \Delta + 5]$.

When $d(y_2) = 2$, let $y'_2$ be the neighbor of $y_2$ other than $x_1$,
$y_2y'_2\to ([10, \Delta + 5]\cup [4, 8])\setminus C(y'_2)$, $(x_1y_2, vx_1)\to (3, 2)$,
thus, by Lemma~\ref{lemma17}, we are done.
In the other case, $d(y_2)\ge 3$.

(2.2.1.) $c(x_1x_2) = 9$. Thus, $n_2(x_1) = 0$ and $d(x_1) - n_2(x_1) = \Delta$.

Since $[2, 10]\subseteq C(x_2)$,
$|T_2|\ge \Delta + 5 - 10 - (d(x_2) - 9) + 1 = \Delta - d(x_2) + 5$,
and thus $n_{6^+}(x_1)\ge |T_2| + 2\ge \Delta - d(x_2) + 7\ge 7$.
Further, if $1\in C(x_2)$,
then $|T_2|\ge \Delta + 5 - 10 - (d(x_2) - 10) + 1 = \Delta - d(x_2) + 6$,
and thus $n_{6^+}(x_1)\ge |T_2| + 2\ge \Delta - d(x_2) + 8\ge 8$.

It follows that ($A_{4.4}$) holds, $v_1 = x_1$, $v_2 = x_2$ and $C(v_2) = [2, 10]\cup [15, \Delta + 5]$, i.e., $T_2 = [10, 14]$.

Since $v_2$ is $(\Delta, 5)$, for any $i\in [4, 5]$, $d(z_i)\le 5$,
and it follows from \ref{9ij}(2) that $G$ contains a $(i, j)_{(u_i, v_1)}$-path for each $j\in [9, \Delta + 5]$.
One sees from Lemma~\ref{lemma06} that $G$ contains no $(i, j)_{(u, v_2)}$-path for each $i\in [3, 8]$, $j\in T_2$,
and from $|T_2| = 5$ assuming w.l.o.g. that $11\not\in T_2\setminus C(u_9)$\footnote{If $C(u_9) = \{9\}\cup [11, 14]$, then $vv_2\to 11$ reduces to the case where $(T_2\setminus C(u_9))\setminus \{c(vv_2)\} \ne \emptyset$. }.

Let $uv_2\to 11$.
Since $v_1$ is $(\Delta, 7)$, i.e., $d(y_2)\le 5$,
by the similar discussion to show \ref{ijC6+x1x2}(2),
an acyclic edge $(\Delta + 5)$-coloring for $G$ can be obtained.

(2.2.2.) $c(x_1 x_2) \in [11, \Delta + 5]$, assuming w.l.o.g. $c(x_1 x_2) = 11$.
Since $([2, 8]\cup [10, 11])\subseteq C(x_2)$,
$|T_2|\ge \Delta + 5 - 11 - (d(x_2) - 9) + 1 = \Delta - d(x_2) + 4$,
$n_{6^+}(x_1)\ge |T_2| + 2\ge \Delta - d(x_2) + 6\ge 6$.

One sees that if $n_{6^+}(x_2) = 5$, then, for any $i\in [4, 5]$, $d(z_i)\le 5$,
and it follows from \ref{9ij}(2) that $G$ contains a $(i, j)_{(u_i, v_1)}$-path for each $j\in [10, \Delta + 5]$.

When $n_{6^+}(x_2) = 5$ and $d(y_9) = 2$, let $y'_9$ be the neighbor of $y_9$ other than $x_1$,
$y_9 y'_9\to ([10, \Delta + 5]\cup [5, 8]\cup \{3\})\setminus C(y'_9)$ and $x_1 y_9\to 4$
reduces to the Case 2.1.

In the other case, if $n_{6^+}(x_2) = 5$, then $d(y_9)\ge 3$;
together with $d(y_j)\ge 3$, $j\in [1, 3]\cup [10, \Delta + 5]$,
we have $n_2(x_1) = 0$ and $d(x_1) - n_2(x_1) = \Delta$.
It follows that ($A_{4.2}$) or ($A_{4.4}$) holds.

When $1, 9\in C(x_2)$, $|T_2|\ge \Delta + 5 -11 - (d(x_2) - 11) + 1 = \Delta - d(x_2) + 6$,
$n_{6^+}(x_1)\ge |T_2| + 2\ge \Delta - d(x_2) + 8\ge 8$.
Thus, none of ($A_{4.2}$) and ($A_{4.4}$) holds.
In the other case, $|\{1, 9\}\cap C(x_2)|\le 1$.

One sees that if $1/9\in C(x_2)$,
then $|T_2|\ge \Delta + 5 -11 - (d(x_2) - 10) + 1 = \Delta - d(x_2) + 5$,
$n_{6^+}(x_1)\ge |T_2| + 2\ge \Delta - d(x_2) + 7\ge 7$ and ($A_{4.4}$) holds.

When $C(x_2)\cap \{1, 9\} = \beta$, since ($A_{4.4}$) holds, we have $v_1 = x_1$, $v_2 = x_2$.
It follows that $T_2 = \{10\}\cup [12, 15]$,
$C(v_2) = [2, 8]\cup \{10, \beta\}\cup ([10, \Delta + 5]\setminus T_2)$, and $d(y_2)\le 5$, $d(z_\beta)\le 5$.
One sees clearly from Lemma~\ref{lemma06} that $G$ contains no $(i, j)_{(u, v_2)}$-path for each $i\in [3, 8]$, $j\in T_2$.
Assume w.l.o.g. $12\not\in T_2\setminus C(z_\beta)$\footnote{If $C(z_\beta) = \{\beta\}\cup [12, 15]$,
then $v v_2\to 12$ reduces to the case where $(T_2\setminus C(z_\beta))\setminus \{c(vv_2)\} \ne \emptyset$. }.
Let $u v_2\to 12$.
Since $d(y_2)\le 5$,
by the similar discussion to show \ref{ijC6+x1x2}(2),
an acyclic edge $(\Delta + 5)$-coloring for $G$ can be obtained.

In the other case, $1, 9\not\in C(x_2)$.
It follows from \ref{ijT2}(2) and \ref{ijC6+x1x2}(1) that for each $i\in [6, 8]$,
$G$ contains a $(i, j)_{(x_1, x_2)}$-path for each $j\in T_2\cup \{1, 9\}$, $1\in C(u_i)$, $6, 7, 8\in C(y_9)$.
\begin{itemize}
\parskip=0pt
\item
    Assume that there exists an $j\in T_2$ such that $G$ contains neither a $(4, j)_{(u, x_2)}$-path nor a $(5, j)_{(u, x_2)}$-path, say $G$ contains neither a $(4, 12)_{(u, x_2)}$-path nor a $(5, 12)_{(u, x_2)}$-path.

    When $2\not\in B_3$ or $d(y_2)\le 6$, $ux_2\to 12$, and $uv\overset{\divideontimes}\to 2$,
    or by the similar discussion to show \ref{ijC6+x1x2}(2), an acyclic edge $(\Delta + 5)$-coloring for $G$ can be obtained.

    In the other case, $2\in B_3$ and $d(y_2)\ge 6$.
    Since $n_{6^+}(x_1)\ge |T_2| + |\{u, x_2, y_2\}|= |T_2| + 3\ge \Delta - d(x_2) + 7\ge 7$,
    thus ($A_{4.4}$) holds.
    Hence, $v_1 = x_1$, $v_2 = x_2$, $T_2 = \{10, 12, 13, 14\}$, $C(v_2) = C\setminus \{1, 9, 12, 13, 14\}$,
    and $d(z_4)\le 5$, $d(z_5)\le 5$, $d(y_9)\le 5$.
    It follows from \ref{9ij}(2) that for any $i\in [4, 5]$,
    $G$ contains a $(i, j)_{(u_i, v_1)}$-path for each $j\in [10, \Delta + 5]$,
    and from \ref{jx1notx2} that $C(y_9) = \{9, 6, 7, 8, 1\}$ and $G$ contains a $(3, 1)_{(y_9, u_1)}$-path.

    If $G$ contains no $(i, 1)_{(u_i, x_1)}$-path, then $v_1 y_9\to i$ reduces to the Case 2.1.
    Otherwise, for each $i\in [4, 5]$, $G$ contains a $(i, 1)_{(u_i, x_1)}$-path.
    If $1\not\in C(u_9)$, then $(u u_9, u v_1, v_1 y_9, v v_1)\to (1, 9, 4, 1)$ and by Lemma~\ref{lemma17},
    we are done.
    Otherwise, $1\in C(u_9)$.

    When mult$_{S_u}(j)\le 1$ for some $j\in [10, 14]$,
    one sees that there exists an $i\in [6, 9]$ such that $G$ contains no $(j, k)_{(u_j, u)}$-path for any $k\in C(u)$, $(u u_i, u v)\overset{\divideontimes}\to (j, i)$.
    When there exists a $j\in [15, \Delta + 5]\setminus S_u$,
    one sees that there exists an $i\in [6, 9]$ such that $G$ contains no $(j, k)_{(u_j, u)}$-path for any $k\in C(u)$, $(u u_i, u v)\overset{\divideontimes}\to (j, i)$.
    In the other case, for each $j\in [10, \Delta + 5]$, mult$_{S_u}(j)\ge 1$, i.e., $[10, \Delta + 5]\subseteq S_u$, and if $j\in [10, 14]$, mult$_{S_u}(j)\ge 2$.

    Note that $2\times |[10, 14]| + |[15, \Delta + 5]| = \Delta + 1\ge 11$.
    Here we do not distinguish the colors of $\{10, 12, \ldots, \Delta + 5\}$, and refer them to as $\alpha_{10}, \alpha_{11}, \ldots, \alpha_{(\Delta + 5)}$, respectively.%
    \footnote{That is, $(\alpha_{10}, \alpha_{12}, \ldots, \alpha_{(\Delta + 5)})$ is an arbitrary permutation of $(10, 12, \ldots, \Delta + 5)$.}
    For each $i\in [6, 9]$, one sees from $1\in C(u_i)$ that $s_i\le 3$.
    It follows from \ref{S1S2} that there exists $i_1, i_2\in [6, 9]$ such that $s_{i_1} = s_{i_2} = 3$;
    for each $i\in \{i_1, i_2\}$,
    $G$ contains a $(1, j)_{(u_i, v_1)}$-path for each $j\in T_i$,
    and since $1\not\in C(v_2)$, we have $11\in C(u_i)$.
    It follows that $C(u_{i_1}) = \{i_1, 1, 11, \alpha_{10}, \alpha_{12}\}$,
    and $C(u_{i_2}) = \{i_2, 1, 11, \alpha_{13}, \alpha_{14}\}$.
    However, there exists an $i\in \{i_1, i_2\}$ such that $G$ contains no $(\alpha_{15}, j)_{(u_i, u)}$-path for each $j\in C(u)$,
    a contradiction to \ref{S1S2}.
    \item
    Assume that for each $j\in T_2$, $G$ contains a $(4, j)_{(u, x_2)}$-path or a $(5, j)_{(u, x_2)}$-path,
    say $G$ contains a $(5, 12)_{(u, x_2)}$-path.
    It follows from \ref{9ij} that $G$ contains a $(5, j)_{(x_1, x_2)}$-path for each $j\in T_2$, $T_2\subseteq C(z_5)$.
    When $G$ contains no $(9, 5)_{(x_1, x_2)}$-path, $(v x_2, v x_1, u v)\overset{\divideontimes}\to (9, 5, 12)$.
    In the other case, $G$ contains a $(9, 5)_{(x_1, x_2)}$-path and $9\in C(z_5)$, $5\in C(y_9)$.

    When $d(y_9) = 5$, i.e., $C(y_9) = [5, 9]$, $(x_1 y_9, v x_1)\to (3, 9)$,
    and by Lemma~\ref{lemma17}, we are done.
    In the other case, $d(y_9)\ge 6$ and $n_{6^+}(x_1)\ge |T_2| + |\{u, x_2, y_9\}|\ge \Delta - d(x_2) + 7\ge 7$.

    Since $d(z_5)\ge |T_2| + |\{5, 9\}|\ge 4 + 2 = 6$, it follows that $n_{6^+}(x_2)\ge 6$, and ($A_{4.2}$) holds.

    Thus, we have $v_1 = x_1$, $v_2 = x_2$ with $n_{6^+}(v_1)\ge \Delta - d(v_2) + 7$, a contradiction.
\end{itemize}

{\bf Case 3.} $u$ is $(11, 5^-)$ and ($A_{5.1}$) or ($A_{5.2}$) holds.

In this case, $C(u) = [1, 10]$.
Let $u_i$ be a $6^+$-vertex, $i\in [3, 5]$ and $u_j$ be a $5^-$-vertex, $j\in [6, 10]$.
One sees that $|T_2| = |[11, \Delta + 5]\setminus C(x_2)| + |\{11\}| = |[11, \Delta + 5]\setminus C(x_2)| + 1$,
and $|[11, \Delta + 5]| = \Delta + 5 - 10 = \Delta - 5$.
It follows from \ref{ijC6+x1x2} that $n_{6^+}(x_1)\ge 6$ and $n_{6^+}(x_2)\ge 5$.

Since $[11, \Delta + 5]\cup [1, 3]\subseteq C(x_1)$, $d(x_1)\ge 3 + \Delta + 5 - 10 = \Delta - 2$.
We discuss the following three subcases where $c(x_1 x_2)\in [4, 5]$, $c(x_1 x_2)\in S_{5^-}(u)$,
and $c(x_1 x_2)\in [12, \Delta + 5]$, respectively.

{Case 3.1.} $c(x_1 x_2)\in [4, 5]$, assuming w.l.o.g. $c(x_1 x_2) = 4$.

One sees that $[1, 4]\cup [11, \Delta + 5]\subseteq C(x_1)$.
It follows that $|[5, 10]\cap C(x_1)|\le 1$.

(3.1.1.) $C(x_1)\cap S_{6^-}(u) = \emptyset$, i.e., $S_1 = S_{5^-}(u) = [6, 10]$.
It follows from \ref{ijC6+x1x2} (6) that $n_{6^+}(x_2)\ge |S_1| + 2 = |S_{5^-}(u)| + 2 = 5 + 2 = 7$.
Since $[2, 4]\cup [6, 11]\subseteq C(x_2)$,
$|T_2|\ge \Delta + 5 - 11 - (d(x_2) - 9) + 1 = \Delta - d(x_2) + 4$,
and it follows from \ref{ijC6+x1x2} (3) that $n_{6^+}(x_1)\ge |T_2| + |\{u, x_2\}|\ge \Delta - d(x_2) + 6\ge 6$.
Thus, none of ($A_{5.1}$) and ($A_{5.2}$) holds.

(3.1.2.) $|C(x_1)\cap S_{6^-}(u)|= 1$,
assuming w.l.o.g. $C(x_1) = [1, 4]\cup [11, \Delta + 5]\cup \{10\}$, i.e., $S_1 = [6, 9]$.
It follows from \ref{ijC6+x1x2}(6) that $n_{6^+}(x_2)\ge |S_1| + 2 = |S_{5^-}(u)| + 2 = 4 + 2 = 6$.
Since $[2, 9]\cup \{11\}\subseteq C(x_2)$,
$|T_2|\ge \Delta + 5 - 11 - (d(x_2) - 9) + 1 = \Delta - d(x_2) + 4$,
and if $1/10\in C(x_2)$, $|T_2|\ge \Delta + 5 - 11 - (d(x_2) - 10) + 1 = \Delta - d(x_2) + 5$.
It follows from \ref{ijC6+x1x2}(3) that $n_{6^+}(x_1)\ge |T_2| + |\{u, x_2\}|\ge \Delta - d(x_2) + 6\ge 6$, and if $1/10\in C(x_2)$, then $n_{6^+}(x_1)\ge \Delta - d(x_2) + 7\ge 7$.

Hence, $n_{6^+}(x_1) = n_{6^+}(x_2) = 6$ and $d(y_{10})\le 5$, with $10\not\in C(x_2)$.
It follows from \ref{ijC6+x1x2}(1) that $C(y_{10}) = [6, 10]$.
Then $(x_1 y_{10}, v x_1)\to (3, 6)$ and $uv\overset{\divideontimes}\to T_2\setminus \{11\}$.

{Case 3.2.} $c(x_1 x_2)\in S_{5^-}(u)$, assuming w.l.o.g. $c(x_1 x_2) = 10$.

One sees that $[1, 3]\cup [10, \Delta + 5]\subseteq C(x_1)$.
It follows that $|[4, 5]\cup [6, 9]\cap C(x_1)|\le 1$.

(3.2.1.) $C(x_1)\cap S_{5^-}(u) = \{10\}$, i.e., $S_1 = S_{5^-}(u) = [6, 9]$.
It follows from \ref{ijC6+x1x2} (6) that $n_{6^+}(x_2)\ge |S_1| + 2 = |S_{5^-}(u)| + 2 = 4 + 2 = 6$.
Since $|\{4, 5\}\cap C(x_2)|\ge 1$ and $[2, 3]\cup [6, 11]\subseteq C(x_2)$,
$|T_2|\ge \Delta + 5 - 11 - (d(x_2) - 9) + 1 = \Delta - d(x_2) + 4$,
and if $4, 5\in C(x_2)$, then $|T_2|\ge \Delta + 5 - 11 - (d(x_2) - 10) + 1 = \Delta - d(x_2) + 5$.
It follows from \ref{ijC6+x1x2} (3) that $n_{6^+}(x_1)\ge |T_2| + |\{u, x_2\}|\ge \Delta - d(x_2) + 6\ge 6$, and if $4, 5\in C(x_2)$, then $n_{6^+}(x_1)\ge \Delta - d(x_2) + 7\ge 7$.

Hence, $n_{6^+}(x_1) = n_{6^+}(x_2) = 6$,
and there exists an $i\in (\{4, 5\}\cap C(x_1))\setminus C(x_2)$ with $d(y_i)\le 5$.
Assuming w.l.o.g. $4\in C(x_1)\setminus C(x_2)$ and $d(y_4)\le 5$.
It follows from \ref{ijC6+x1x2}(1) that $C(y_{4}) = \{4\}\cup S_1 = \{4\}\cup [6, 9]$.
Then $(x_1 y_{4}, v x_1, v x_2, u v)\overset{\divideontimes}\to (5, 7, 4, 11)$.

(3.2.2.) $|C(x_1)\cap S_{5^-}(u)|= 2$,
assuming w.l.o.g. $C(x_1) = [1, 3]\cup [9, \Delta + 5]$, i.e., $S_1 = [6, 8]$.
It follows from \ref{ijC6+x1x2}(6) that $n_{6^+}(x_2)\ge |S_1| + 2 = |S_{5^-}(u)| + 2 = 3 + 2 = 5$.

Since $[2, 8]\cup \{10, 11\}\subseteq C(x_2)$,
$|T_2|\ge \Delta + 5 - 11 - (d(x_2) - 9) + 1 = \Delta - d(x_2) + 4$,
It follows from \ref{ijC6+x1x2}(3) that $n_{6^+}(x_1)\ge |T_2| + |\{u, x_2\}|\ge \Delta - d(x_2) + 6\ge 6$.

When $1, 9\in C(x_2)$, $|T_2|\ge \Delta + 5 - 11 - (d(x_2) - 11) + 1 = \Delta - d(x_2) + 6$,
and it follows from \ref{ijC6+x1x2}(3) that $n_{6^+}(x_1)\ge |T_2| + |\{u, x_2\}|\ge \Delta - d(x_2) + 8\ge 8$.
Thus, none of ($A_{5.1}$) and ($A_{5.2}$) holds.
When $1/9\in C(x_2)$, $|T_2|\ge \Delta + 5 - 11 - (d(x_2) - 10) + 1 = \Delta - d(x_2) + 5$,
and it follows from \ref{ijC6+x1x2}(3) that $n_{6^+}(x_1)\ge |T_2| + |\{u, x_2\}|\ge \Delta - d(x_2) + 7\ge 7$.
It follows that ($A_{5.2}$) holds and $v_1 = x_1$, $v_2 = x_2$.
Since $d(v_1) = 12$ and $\Delta \ge d(v_1)\ge 13$, $n_{6^+}(v_1)\ge \Delta - d(v_2) + 7\ge 13 - 11 + 7 = 9$, a contradiction.
In the other case, $1, 9\not\in C(x_2)$.

If $n_{6^+}(x_1)\ge 7$, i.e., ($A_{5.2}$) holds and $v_1 = x_1$, $v_2 = x_2$,
then, since $d(v_1) = 12$ and $\Delta \ge d(v_1)\ge 13$,
we have $n_{6^+}(v_1)\ge \Delta - d(x_2) + 6\ge 13 - 11 + 6\ge 8$, a contradiction.
Otherwise, $n_{6^+}(x_1)= 6$.
It follows that ($A_{5.1}$) holds and $d(x_1) = d(x_2) = \Delta\in [11, 12]$, $T_2 = [11, 14]$,
$C(x_2) = [2, 8]\cup [10, 11]\cup [15, \Delta + 5]$.

Since $n_{6^+}(x_2)\le 6$, there exists an $i\in [4, 5]$ such that $d(z_i)\le 5$.
Assume w.l.o.g. $d(z_4)\le 5$ and it follows from \ref{9ij}(2) that $G$ contains a $(4, i)_{(u_4, v_1)}$-path for each $i\in [11, \Delta + 5]$.

One sees that $G$ contains no $(i, )_{(u, x_2)}$-path for each $i\in \{1\}\cup [3, 4]\cup [6, 8]$, $j\in T_2$.
If $T_2\setminus C(u_{10})\ne \emptyset$, then let $S = T_2\setminus C(u_{10})$;
if $C(u_{10}) = [10, 14]$, then $uu_{10}\to 15$ and let $S = T_2$.
\begin{itemize}
\parskip=0pt
\item
    Assume that $G$ contains no $(5, j)_{(u, x_2)}$-path for some $j\in S$.
    First, $u x_2\to j$ and if $j = 11$, then $v x_2\to 12$.
    Since $d(y_2)\le 5$, by the similar discussion to show \ref{ijC6+x1x2}(2),
    an acyclic edge $(\Delta + 5)$-coloring for $G$ can be obtained.
\item
    Assume that $G$ contains a $(5, j)_{(u, x_2)}$-path for each $j\in S$.
    If $5\not\in C(y_9)$, then $(v x_2, v x_1)\to (9, 5)$ and $uv\overset{\divideontimes}\to S$.
    Otherwise, $5\in C(y_9)$ and since $d(y_9)\le 5$,
    it follows from \ref{ijC6+x1x2}(1) that $C(y_9) = \{5, 9\}\cup S_1 = [5, 9]$.
    Then $(x_1 y_9, v x_1)\to (3, 9)$, and by Lemma~\ref{lemma17}, we are done.
\end{itemize}

{Case 3.3.} $c(x_1 x_2)\in [12, \Delta + 5]$, assuming w.l.o.g. $c(x_1 x_2) = 12$.

One sees that $[1, 3]\cup [11, \Delta + 5]\subseteq C(x_1)$.
It follows that $|[4, 5]\cup [6, 10]\cap C(x_1)|\le 2$.

(3.3.1.) $C(x_1)\cap S_{5^-}(u) = \emptyset$, i.e., $S_1 = S_{5^-}(u) = [6, 10]$.
It follows from \ref{ijC6+x1x2} (6) that $n_{6^+}(x_2)\ge |S_1| + 2 = 5 + 2 = 7$.
Since $[2, 3]\cup [6, 12]$,
$|T_2|\ge \Delta + 5 - 12 - (d(x_2) - 9) + 1 = \Delta - d(x_2) + 3$,
and if $1/4/5\in C(x_2)$, then $|T_2|\ge \Delta + 5 - 12 - (d(x_2) - 10) + 1 = \Delta - d(x_2) + 4$.
It follows from \ref{ijC6+x1x2} (3) that $n_{6^+}(x_1)\ge |T_2| + |\{u, x_2\}|\ge \Delta - d(x_2) + 5\ge 5$,
and if $1/4/5\in C(x_2)$, then $n_{6^+}(x_1)\ge \Delta - d(x_2) + 6\ge 6$.

Hence, $n_{6^+}(x_1) = 5$, $n_{6^+}(x_2) = 7$.
It follows that for each $i\in [4, 5]$, $i\in C(x_1)\setminus C(x_2)$ with $d(y_i)\le 5$,
and thus by \ref{ijC6+x1x2}, $d(y_i)\ge |S_1| + 1\ge 5 + 1 = 6$, a contradiction.

(3.3.2.) $|C(x_1)\cap S_{5^-}(u)| = 1$,
assuming w.l.o.g. $C(x_1) \cap S_{5^-}(u) = \{10\}$, i.e., $S_1 = [6, 9]$.
It follows from \ref{ijC6+x1x2} (6) that $n_{6^+}(x_2)\ge |S_1| + 2 = 4 + 2 = 6$.
It follows that if $n_{6^+}(x_1)\ge 7$, then none of ($A_{5.1}$) and ($A_{5.2}$) holds, and we are done.

One sees from $|\{4, 5\}\cap C(x_1)|\le 1$ that $|\{4, 5\}\cap C(x_2)|\ge 1$,
and from \ref{jx1notx2} that if $10\not\in C(x_2)$, then $d(y_{10}) \ge |S_1| + 2 = 4 + 2 = 6$.
Assuming w.l.o.g. $4\in C(x_2)\setminus C(x_1)$.

Note that $[2, 4]\cup [6, 9]\cup [11, 12]\subseteq C(x_2)$.
When $|\{1, 5, 10\}\cap C(x_2)|\ge 2$, $|T_2|\ge \Delta + 5 - 12 - (d(x_2) - 11) + 1 = \Delta - d(x_2) + 5$,
and thus it follows from \ref{ijC6+x1x2}(3) that $n_{6^+}(x_1)\ge |T_2| + |\{u, x_2\}|\ge \Delta - d(x_2) + 7\ge 7$.
In the other case, $|\{1, 5, 10\}\cap C(x_2)|\le 1$.

When $1/5\in C(x_2)$, $|T_2|\ge \Delta + 5 - 12 - (d(x_2) - 10) + 1 = \Delta - d(x_2) + 4$,
and thus since $10\not\in C(x_2)$,
it follows from \ref{ijC6+x1x2}(3) that $n_{6^+}(x_1)\ge |T_2| + |\{u, x_2, y_{10}\}|\ge \Delta - d(x_2) + 7\ge 7$.

Hence, it suffices to assume that $1, 5\not\in C(x_2)$ and $5\in C(x_1)$.
It follows from \ref{ijC6+x1x2}(1) that $S_1\subseteq C(y_5)$.
If $d(y_5) = 5$, i.e., $C(y_5) = [5, 9]$,
then $(x_1 y_5, v x_1, v x_2, u v)\overset{\divideontimes}\to (4, 6, 5, 11)$.
Otherwise, $d(y_5)\ge 6$.
One sees that $|T_2|\ge \Delta + 5 - 12 - (d(x_2) - 9) + 1 = \Delta - d(x_2) + 3$,
and thus it follows from \ref{ijC6+x1x2}(3) that $n_{6^+}(x_1)\ge |T_2| + |\{u, x_2, y_5\}|\ge \Delta - d(x_2) + 6$.
\begin{itemize}
\parskip=0pt
\item
    $10\not\in C(x_2)$. Then $n_{6^+}(x_1)\ge |T_2| + |\{u, x_2, y_5, y_{10}\}|\ge \Delta - d(x_2) + 7\ge 7$.
\item
    $10\in C(x_2)$. Then, $|T_2|\ge \Delta + 5 - 12 - (d(x_2) - 10) + 1 = \Delta - d(x_2) + 4$,
    and thus it follows from \ref{ijC6+x1x2}(3) that $n_{6^+}(x_1)\ge |T_2| + |\{u, x_2, y_5\}|\ge \Delta - d(x_2) + 7\ge 7$.
\end{itemize}

(3.3.3.) $|C(x_1)\cap S_{5^-}(u)| = 2$,
assuming w.l.o.g. $C(x_1) = [1, 3]\cup [9, \Delta + 5]$, i.e., $S_1 = [6, 8]$.
It follows that $[2, 8]\cup [11, 12]\subseteq C(x_2)$, and from \ref{ijC6+x1x2}(6) that $n_{6^+}(x_2)\ge |S_1| + 2 = 3 + 2 = 5$.

When $|\{1, 9, 10\}\cap C(x_2)|\ge 2$, $|T_2|\ge \Delta + 5 - 12 - (d(x_2) - 11) + 1 = \Delta - d(x_2) + 5$,
and it follows from \ref{ijC6+x1x2}(3) that $n_{6^+}(x_1)\ge |T_2| + |\{u, x_2\}|\ge \Delta - d(x_2) + 7\ge 7$.
It follows that ($A_{5.2}$) holds and $v_1 = x_1$, $v_2 = x_2$.
Since $d(v_2) = 11$ and $\Delta \ge d(v_1)\ge 13$,
$n_{6^+}(v_1)\ge \Delta - d(v_2) + 7\ge 13 - 11 + 7 = 9$, a contradiction.
In the other case, $|\{1, 9, 10\}\cap C(x_2)|\le 1$.

First, assume that $\{1, 9, 10\}\cap C(x_2) = \{i_0\}$.
Then $|T_2|\ge \Delta + 5 - 12 - (d(x_2) - 10) + 1 = \Delta - d(x_2) + 4$,
and thus it follows from \ref{ijC6+x1x2}(3) that $n_{6^+}(x_1)\ge |T_2| + |\{u, x_2\}|\ge \Delta - d(x_2) + 6$.
When $n_{6^+}(x_1)\ge 7$, it follows that ($A_{5.2}$) holds and $v_1 = x_1$, $v_2 = x_2$,
since $d(v_2) = 11$ and $\Delta \ge d(v_1)\ge 13$,
$n_{6^+}(v_1)\ge \Delta - d(v_2) + 6\ge 13 - 11 + 6 = 8$, a contradiction.
In the other case, $n_{6^+}(x_1) = 6$ and $d(x_2) = \Delta \in [11, 12]$, $|T_2| = 4$ with $d(z_i)\le 5$, $i\in \{4, 5, i_0\}$.

When $C(u_i) = \{i\}\cup T_2$ for some $i\in [9, 10]$,
it follows from \ref{jx1notx2}(2) that $3\not\in C(y_i)$,
thus, $(u u_i, u v)\overset{\divideontimes}\to (12, i)$.
In other case, for each $i\in [9, 10]$, $T_2\setminus C(u_i)\ne \emptyset$.

For each $i\in [4, 5]$,
if $C(z_i) = \{i\}\cup T_2$, then $v x_1\to i$, $x_2 z_i\to \{9, 10\}\setminus C(x_2)$,
and thus an acyclic edge coloring of $G$ can be obtained if $G$ contains no $(i, k)_{(u_i, x_1)}$-path for some $k\in [11, \Delta + 5]$.
Together with \ref{9ij}(2),
for each $i\in [4, 5]$, $G$ contains an $(i, j)_{(u_i, v_1)}$-path for each $j\in [k, \Delta + 5]$.
It follows that $G$ contains no $(i, j)_{(u_i, v_1)}$-path for each $i\in [3, 8]$, $j\in T_2$.

If $i_0 \in [9, 10]$, then let $S = T_2\setminus C(u_{i_0})$;
if $T_2\setminus C(z_{1})\ne \emptyset$, then let $S = T_2\setminus C(z_{i_0})$;
if $C(z_{1}) = \{1\}\cup T_2$, then $x_2 z_{1}\to 9$ and let $S = T_2\setminus C(u_9)$.

Then $u x_2\to j\in S$ and if $j = 11$, then $v x_2\to 12$.
Since $d(y_2)\le 5$, by the similar discussion to show \ref{ijC6+x1x2}(2),
an acyclic edge $(\Delta + 5)$-coloring for $G$ can be obtained.

Next, assume that $1, 9, 10\not\in C(x_2)$.
Ones sees that $|T_2|\ge \Delta + 5 - 12 - (d(x_2) - 9) + 1 = \Delta - d(x_2) + 3$.
It follows from \ref{jx1notx2} that for each $j\in [9, 10]$,
if $d(y_j) = 5$, then $3\not\in C(y_j)$,
and $G$ contains a $(3, i)_{(y_j, x_1)}$-path for some $i\in C(x_1)\setminus [k, \Delta + 5]$.

When there exists $i_1, i_2\in \{2, 9, 10\}$ such that $d(y_{i_1})\ge 6$ and $d(y_{i_2})\ge 6$,
it follows from \ref{ijC6+x1x2}(3) that $n_{6^+}(x_1)\ge |T_2| + |\{u, x_2, y_{i_1}, y_{i_2}\}|\ge \Delta - d(x_2) + 7$.
Hence, $n_{6^+}(x_1)= 7$, ($A_{5.2}$) holds and $v_1 = x_1$, $v_2 = x_2$.
Thus, since $d(v_2) = 11$ and $\Delta \ge d(v_1)\ge 13$,
$n_{6^+}(v_1)\ge \Delta - d(v_2) + 7\ge 13 - 11 + 7 = 9$, a contradiction.

In the other case, at most one of $\{y_2, y_9, y_{10}\}$ is a $6^+$-vertex,
and at most one of $\{y_9, y_{10}\}$ is a $6^+$-vertex.
Assuming w.l.o.g. $d(y_{9})\le 5$.
It follows that $C(y_9) = [6, 9]\cup \{i_9\}$ and $G$ contains a $(3, i_9)_{(y_9, x_1)}$-path,
where $i_9\in \{1, 2, 10\}$.

One sees clearly that $4, 5\not\in C(y_9)$.
For $i\in [4, 5]$,
if $G$ contains no $(i, 11)_{(u_i, x_1)}$-path,
then $(v x_1, v x_2, u v)\overset{\divideontimes}\to (i, 6, 11)$;
otherwise, $G$ contains a $(i, 11)_{(u_i, x_1)}$-path.
Thus an acyclic edge coloring of $G$ can be obtained if $v x_1\to i$ and $G$ contains no $(i, j)_{(u_i, x_1)}$-path for some $j\in [12, \Delta + 5]$.

Hence, for each $i\in [4, 5]$, $G$ contains an $(i, j)_{(u_i, v_1)}$-path for each $j\in [11, \Delta + 5]$.
It follows that $G$ contains no $(i, j)_{(u_i, v_1)}$-path for each $i\in [3, 8]$, $j\in T_2$.
When $G$ contains no $(i, 2)_{(u_i, y_2)}$-path for some $i\in [3, 5]$,
$(u x_2, v x_1, u v)\overset{\divideontimes}\to (13, i, 2)$.
When $d(y_2)\le 5$, $u x_2\to 13$, and by the similar discussion to show \ref{ijC6+x1x2}(2),
an acyclic edge $(\Delta + 5)$-coloring for $G$ can be obtained.

In the other case, $d(y_2)\ge 6$ and $G$ contains a $(i, 2)_{(u_i, y_2)}$-path for each $i\in [3, 5]$.
It follows that $d(y_{10})\le 5$, $C(y_{10}) = [6, 8]\cup \{10, 1\}$, $G$ contains a $(3, 1)_{(y_{10}, x_1)}$-path,
and $C(y_9) = [6, 8]\cup \{9, 1\}$, $G$ contains a $(3, 1)_{(y_9, x_1)}$-path, a contradiction.

This finishes the inductive step for the case where $G$ contains the configuration ($A_{2}$)--($A_5$).


\subsection{Configurations ($A_{6}$)--($A_{8}$)}
In this subsection we prove the inductive step for the case where $G$ contains one of the configurations ($A_{6}$)--($A_{8}$).
We have the following revised ($A_{6}$)--($A_{8}$):
An edge $u v$ with $d(u) = k\in [4, 5]$ such that at least one of the configurations ($A_{6}$)--($A_{8}$) occurs.

Let $u_1$, $u_2$, $\ldots$, $u_{k - 1}$ be the neighbors of $u$ other than $v$.
Recall that $H$ has an acyclic edge $(\Delta + 5)$-coloring $c(\cdot)$ using the color set $C = \{1, 2, \ldots, \Delta + 5\}$.
One sees that if re-coloring some edges $E'\subseteq E(H)$ incident to $u$ or $v$ does not produce any new bichromatic cycles, then a new acyclic edge $(\Delta + 5)$-coloring $c'(\cdot)$ for $H$ will be obtained.

To introduce some commonalities of any acyclic edge $(\Delta + 5)$-coloring for $H$,
for simplicity and clarity, for some acyclic edge $(\Delta + 5)$-coloring $c(\cdot)$ for $H$,
we introduce the following symbols.

Let $C_{u v} = C(u)\cup C(v)$, $T_{u v} = C\setminus (C(u)\cup C(v))$,
$A_{u v} = C(u)\cap C(v)$,
and $a_{u v} = |A_{u v}|$, $t_{u v} = |T_{u v}|$.

Let $A_{u v} = \{i_1, \ldots, i_{a_{u v}}\}$,
$U_{u v} = C(u)\setminus A_{u v} = \{i_{a_{u v} + 1}, \ldots, i_{k - 1}\}$,
and if $d(v)\ge a_{u v} +2$,
then $V_{u v} = C(v)\setminus A_{u v} = [i_{d(u)}, \ldots, i_{d(u) + d(v) - a_{u v} - 2}]$.

Assume that $c(u u_{i_j}) = c(v v_{i_j}) = i_j$ for any $j\in [1, a_{u v}]$,
and $c(u u_i)= i$ for $i \in U_{u v}$, $c(v v_j) = j$ for $j\in V_{u v}$.
It follows from \ref{prop3001} that $1\le a_{uv} \le 4$,
and $T_{u v}\subseteq \bigcup _{i\in C(u)\cap C(v)}B_i$.

For $i\in C(u)$, let $T_i = T_{u v}\setminus C(u_i)$, $t_i = |T_i|$;
$C_i = C(u_i)\cap C(u)$, $W_i = C(u_i)\cap C_{u v}$,
and $c_i = |C_i|$, $w_i = |W_i|$.
For $j\in C(v)$, let $R_j = T_{u v}\setminus C(v_j)$, $r_j = |R_j|$,
$D_j = C(v_j)\cap C(v)$, and $D_j = |D_j|$.
Let $$S_u = \biguplus_{i\in [1, k - 1]}(C(u_i)\setminus \{c(u u_i)\}),
S_v = \biguplus_{x\in N(v)\setminus \{v\}}(C(x)\setminus \{c(vx)\}).  $$

One sees clearly that if ($A_7$) holds, then
\begin{equation}
\label{eq16}
\sum_{x\in N(u)\setminus \{v\}}d(x) = \sum_{i\in [1, 4]}d(u_i)\le \begin{cases}
                \Delta + 20 < 2\Delta + 12, \mbox{ if } \Delta\ge 9;\\
				\Delta + 20 = 2\Delta + 12, \mbox{ if } \Delta = 8;\\
				2\Delta + 13 < 2\Delta + 14, \mbox{ if } \Delta = 7;\\
				2\Delta + 12 < 2\Delta + 13, \mbox{ if } \Delta = 6,
				\end{cases}
\end{equation}
and if ($A_8$) holds, then
\[
\sum_{i\in [1, 4]}d(u_i)\le \begin{cases}
				\Delta + 19 < \Delta + 20 \le 2\Delta + 12, \mbox{ if } \Delta\ge 8;\\
				2\Delta + 13 < 2\Delta + 14, \mbox{ if } \Delta = 7;\\
				2\Delta + 12 < 2\Delta + 13, \mbox{ if } \Delta = 6.
             \end{cases}
\]

Since $T_{u v}\subseteq \bigcup _{i\in A_{u v}}B_i$,
mult$_{S_u}(j)\ge 1$ for each $j\in T_{u v}$,
and in the sequel we continue the proof with the following \ref{prop3209}, or otherwise we are done:
\begin{proposition}
\label{prop3209}
For each $i\in A_{u v}$,
\begin{itemize}
\parskip=0pt
\item[{\rm (1)}]
	if $a_{u v} = 1$, then $T_{u v}\subseteq B_i\subseteq C(u_i)\cap C(v_i)$ and $T_i = \emptyset$, $t_i = 0$;
\item[{\rm (2)}]
    if $t_i\ge 1$, i.e., $T_i \ne\emptyset$, then $a_{u v}\ge 2$, and for each $j\in T_i$,
    $j\in B_{i_j}$ for some $i_j\in A_{u v}\setminus \{i\}$;
\item[{\rm (3)}]
    if $r_i\ge 1$, i.e., $R_i \ne\emptyset$, then $a_{u v}\ge 2$, and for each $j\in R_i$,
    $j\in B_{i_j}$ for some $i_j\in A_{u v}\setminus \{i\}$;
\item[{\rm (4)}]
    $T_i\subseteq \bigcup _{j\in A_{u v}\setminus\{i\}}B_j$ and $T_i\subseteq \bigcup _{j\in A_{u v}\setminus\{i\}}C(u_j)$.
\end{itemize}
\end{proposition}

\begin{lemma}
\label{auvge2}
If $a_{u v}\ge 2$,
then an acyclic edge $(\Delta + 5)$-coloring for $H$ can be obtained such that $|C(u)\cap C(v)| \le a_{u v} - 1$,
or the following holds:
\begin{itemize}
\parskip=0pt
\item[{\rm (1)}]
    for each $i\in A_{u v}$,
    \begin{itemize}
    \parskip=0pt
    \item[{\rm (1.1)}]
        for each $j\in T_i\ne \emptyset$,
        $H$ contains a $(i_j, j)_{(u_{i_j}, v_{i_j})}$-path and a $(a_i, j)_{(u_i, u_{a_i})}$-path for some $i_j\in A_{u v}$, $a_i\in C_i\setminus\{i, i_j\}$, and $j\in C(u_{i_j})\cap C(u_{a_i})$;
   \item[{\rm (1.2)}]
        for each $j\in R_i\ne \emptyset$,
        $H$ contains a $(i_j, j)_{(u_{i_j}, v_{i_j})}$-path and a $(b_i, j)_{(v_i, v_{b_i})}$-path for some $i_j\in A_{u v}$, $b_i\in D_i\setminus \{i, i_j\}$, and $j\in C(v_{i_j})\cap C(v_{b_i})$;
   \end{itemize}
\item[{\rm (2)}]
    for each $j\in T_{u v}$, mult$_{S_u}(j)\ge 2$ and mult$_{S_v}(j)\ge 2$.
\item[{\rm (3)}]
   For each $\alpha\in U_{u v}$,
   if there exists an $j\in T_\alpha$ such that recoloring $u u_{\alpha}\to j$ cannot obtain a bichromatic cycle,
   then $\alpha\in \bigcup_{i\in A_{u v}}B_{i}$,
   mult$_{S_u}(\alpha)\ge 2$,
   and for each $i\in A_{u v}$, $\alpha\in C(u_i)$, or there exists a $a_i\in C(u_i)$ such that $H$ contains a $(\alpha, a_i)_{(u_{a_i}, u_i)}$-path.
\item[{\rm (4)}]
   For each $\beta\in V_{u v}$, if there exists an $j\in R_{\beta}$ such that recoloring $v v_\beta\to j$ cannot obtain a bichromatic cycle,
   then $\beta\in \bigcup_{i\in A_{u v}}B_{i}$,
   mult$_{S_u}(\beta)\ge 2$,
   and for each $i\in A_{u v}$, $\beta\in C(u_i)$, or there exists a $a_i\in C(u_i)$ such that $H$ contains a $(\beta, a_i)_{(u_{a_i}, u_i)}$-path.
\end{itemize}
\end{lemma}
\begin{proof}
For each $j\in T_{u v}$,
it follows from $j\in \bigcup _{i\in A_{u v}}B_i$ that there exists an $i_j\in A_{u v}$ such that $j\in B_{i_j}\subseteq C(u_{i_j})\cap C(v_{i_j})$.

When, for some $j\in T_i\ne \emptyset$, $G$ contains no $(j, l)_{(u_i, u_l)}$-path for each $l\in C(u)$,
$u u_i\to j$ reduces to the case where $|C(u)\cap C(v)| \le a_{u v} - 1$.
In the other case, for each $j\in T_i\ne \emptyset$,
$H$ contains a $(a_i, j)_{(u_i, u_{a_i})}$-path for some $a_i\in C_i$,
and $j\in C(u_{a_i})$.
By Lemma~\ref{lemma06}, $a_i\ne i_j$ and $j\in C(u_{i_j})\cap C(u_{a_i})$.

When, for some $j\in R_i\ne \emptyset$, $G$ contains no $(j, l)_{(v_i, v_l)}$-path for each $l\in C(v)$,
$v v_i\to j$ reduces to the case where $|C(u)\cap C(v)| \le a_{u v} - 1$.
In the other case, for each $j\in R_i\ne \emptyset$,
$H$ contains a $(b_i, j)_{(v_i, v_{b_i})}$-path for some $b_i\in D_i$,
and $j\in C(v_{b_i})$.
By Lemma~\ref{lemma06}, $b_i\ne i_j$ and thus $j\in C(v_{i_j})\cap C(v_{b_i})$.

For each $j\in T_{u v}$,
one sees clearly that if $j\in C(u_i)$ for each $i\in A_{u v}$, then mult$_{S_u}(j)\ge 2$;
if $j\not\in C(u_i)$ for some $i\in A_{u v}$, i.e., $j\in T_i$, then by (1.1),
there exists $i_j\in A_{u v}$, $a_i\in C(u)\setminus\{i, i_j\}$ such that $j\in C(u_{i_j})\cap C(u_{a_i})$.
Hence, mult$_{S_u}(j)\ge 2$.
Since recoloring $u u_{\alpha}\to j$ reduces to another acyclic edge $(\Delta + 5)$-coloring for $H$ such that $a_{u v}\ge 2$,
by the similar discussion to show (1) and (2), (3) holds.
Likewise, (4) holds.
\end{proof}

Let $T_0 = \{j\in T_{u v} \mid \mbox{mult}_{S_u}(j) \ge 3\}$, and $t_0 = |T_0|$.
If $a_{u v}\ge 2$,
Then
\begin{equation}\label{eq17}
\begin{aligned}
\sum_{i\in [1, k - 1]}d(u_i) = \sum_{i\in [1, k - 1]}w_i + \sum_{j\in T_{uv}}mult_{S_u}(j)\ge \sum_{i\in [1, k - 1]}w_i + 2t_{uv} + \sum_{j\in T_0}(mult_{S_u}(j) - 2).
\end{aligned}
\end{equation}

\begin{corollary}
\label{auv4}
If $a_{u v} = 4$, i.e., $d(u) = 5$, and $A_{u v} = [1, 4]$,
then an acyclic edge $(\Delta + 5)$-coloring for $H$ can be obtained such that $|C(u)\cap C(v)| \le 3$,
or the following holds:
\begin{itemize}
\parskip=0pt
\item[{\rm (1)}]
    for each $i\in A_{u v}$, $T_i\subseteq \bigcap_{j\in A_{u v}\setminus \{i\}} C(u_j)$,
    and $\bigcup_{j\in A_{u v}\setminus \{i\}} T_j\subseteq C(u_i)$;
\item[{\rm (2)}]
   for each $i, j\in A_{u v}$, $T_i\cap T_j = \emptyset$.
\end{itemize}
\end{corollary}
\begin{proof}
Assuming $j\in T_1$ and w.l.o.g. that $H$ contains a $(2, j)_{(u_2, v)}$-path.
One sees that if $j\in T_i$ for some $i\in [3, 4]$, say $j\in T_3$,
then it follows from Lemma~\ref{auvge2} that $H$ contains a $(4, j)_{(u_1, u_4)}$-path and a $(4, j)_{(u_3, u_4)}$-path, a contradiction.
Hence, $j\in C(u_3)$ and (1), (2) hold.
\end{proof}
\begin{corollary}
\label{auv3}
If $a_{u v} = 3$, i.e., $A_{u v} = \{i_1, i_2, i_3\}$,
then either an acyclic edge $(\Delta + 5)$-coloring for $H$ can be obtained such that $|C(u)\cap C(v)| \le 2$,
or the following holds:
\begin{itemize}
\parskip=0pt
\item[{\rm (1)}]
   for each $j\in T_{i_1}\ne \emptyset$, mult$_{S_u}(j)\ge 2$,
   $G$ contains a $(i_j, j)_{(u_{i_j}, v_{i_j})}$-path for some $i_j\in \{i_2, i_3\}$, say $i_j = i_2$, and
   \begin{itemize}
   \parskip=0pt
   \item[{\rm (1.1)}]
        if $d(u) = 4$, then $i_3\in C(u_{i_1})$, $G$ contains a $(i_3, j)_{(u_{i_1}, u_{i_3})}$-path;
   \item[{\rm (1.2)}]
        if $d(u) = 5$, then $U_{u v} = \{i_4\}$, there exists a $a_{i_1}\in \{i_3, i_4\}\cap C(u_{i_1})$ such that $G$ contains a $(a_{i_1}, j)_{(u_{i_1}, u_{a_{i_1}})}$-path, and $j\in C(u_{a_{i_1}})$;
   \end{itemize}
\item[{\rm (2)}]
    for each $i\in A_{u v}$, $T_i\subseteq \bigcap_{j\in A_{u v}\setminus \{i\}} C(u_j)$ and $\bigcup_{j\in A_{u v}\setminus \{i\}} T_j\subseteq C(u_i)$;
\item[{\rm (3)}]
   for each $i, j\in A_{u v}$, $T_i\cap T_j = \emptyset$.
\item[{\rm (4)}]
   If $d(u) = 5$ and there exists an $j\in T_{i_4}$ such that recoloring $u u_{i_4}\to j$ cannot obtain a bichromatic cycle,
   then $i_4\in \bigcup_{i\in \{i_1, i_2, i_3\}}B_{i}$,
   mult$_{\bigcup_{i\in A_{u v}}C(u_i)}(i_4)\ge 2$,
   and for each $i\in \{i_1, i_2, i_3\}$, $i_4\in C(u_i)$, or there exists a $a_i\in C(u_i)$ such that $H$ contains a $(i_4, a_i)_{(u_{a_i}, u_i)}$-path.
\item[{\rm (5)}]
   For each $\beta\in V_{u v}$,
   if there exists an $j\in R_{\beta}$ such that recoloring $vv_\beta\to j$ cannot obtain a bichromatic cycle,
   then $\beta\in \bigcup_{i\in \{i_1, i_2, i_3\}}B_{i}$,
   mult$_{\bigcup_{i\in A_{u v}}C(u_i)}(\beta)\ge 2$,
   and for each $i\in \{i_1, i_2, i_3\}$, $\beta\in C(u_i)$, or there exists a $a_i\in C(u_i)$ such that $H$ contains a $(\beta, a_i)_{(u_{a_i}, u_i)}$-path.
\item[{\rm (6)}]
   for each $j\in T_{u v}$, if mult$_{S_u}(j) = 2$ with $j\in T_{i_1}$, then
   \begin{itemize}
   \parskip=0pt
   \item[{\rm (6.1)}]
        $i_1\in S_u$, $V_{u v}\subseteq S_u$;
   \item[{\rm (6.2)}]
        if $d(u) = 5$ with $U_{u v} = \{i_4\}$,
        then $H$ contains no $(j, i)_{(u_{i_4}, u_i)}$-path for each $i\in C(u)$ and mult$_{S_u}(i_4)\ge 2$.
   \end{itemize}
\item[{\rm (7)}]
   mult$_{S_u}(j)\ge 3$ for each $j\in T_{u v}$,
   or if there exists a $j\in T_{u v}$ such that mult$_{S_u}(j)= 2$, then $V_{u v}\subseteq S_u$.
\end{itemize}
\end{corollary}
\begin{proof}
It follows from Lemma~\ref{auvge2}(1.1) that (1) holds.

One sees clearly that if $d(u) = 4$, then (2) holds naturally.
Hence, assume $d(u) = 5$, $j\in T_{i_1}$ and w.l.o.g. that $H$ contains a $(i_2, j)_{(u_{i_2}, v_{i_2})}$-path.
When $j\in T_{i_3}$, it follows from Lemma~\ref{auvge2}(1.1) that $H$ contains a $(i_4, j)_{(u_{i_1}, u_{i_4})}$-path and a $(i_4, j)_{(u_{i_3}, u_{i_4})}$-path.
Hence, $j\in C(u_{i_3})$ and (2--3) hold.

Together with Lemma~\ref{auvge2}(3--4), (4--5) hold.

To prove (6), assume mult$_{S_u}(j) = 2$.
Assuming w.l.o.g. $j\in T_{i_1}$,
and it follows from (1) that $H$ contains a $(i_2, j)_{(u_{i_2}, v_{i_2})}$-path and a $(i_3, j)_{(u_{i_3}, u_{i_1})}$-path.
When $i_1\not\in S_u$, $(u u_{i_3}, u u_{i_1})\to (i_1, j)$ reduces to an acyclic edge $(\Delta + 5)$-coloring for $H$ such that $C(u)\cap C(v) = \{i_2, i_1\}$.
When there exists an $i\in V_{u v}\setminus S_(u)$, $(u u_{i_3}, u u_{i_1})\to (i, j)$ reduces to an acyclic edge $(\Delta + 5)$-coloring for $H$ such that $C(u)\cap C(v) = \{i_2, i\}$.
In the other case, $i_1\in S_u$ and $V_{u v}\subseteq S_u$.
Since recoloring $u u_{i_4}\to j$ reduces to another acyclic edge $(\Delta + 5)$-coloring for $H$ such that $A_{u v} = \{i_1, i_2, i_3\}$,
by the similar discussion to show Lemma~\ref{auvge2}(2), $\beta\in C(u_j)$ for each $j\in A_{u v}\setminus \{i\}$.
Hence, (6) holds and then (7) holds.
\end{proof}

\begin{lemma}
\label{4to321}
If $a_{u v} = 4$, then in $O(1)$ time either an acyclic edge $(\Delta + 5)$-coloring for $G$ can be obtained,
or an acyclic edge $(\Delta + 5)$-coloring for $H$ can be obtained such that $|C(u)\cap C(v)| \le 3$.
\end{lemma}
\begin{proof}
It suffices to assume $d(u) = 5$ and $c(u u_i) = i$, $i\in [1, 4]$.

Since $a_{u v} = 4$ and $|C(v)\setminus \{c(v u_3), c(v u_4)\}\le 3$,
assume w.l.o.g. $c(v u_3) = i_1\in \{1, 2, 4\}$.
Since $d(v)\in [5, 6]$, $t_{u v}\ge \Delta + 5 - 5 = \Delta$.
One sees clearly that $T_i\ne \emptyset$, $i\in [1, 4]$, and $R_{i_1}\ne \emptyset$.

It follows from Lemma~\ref{auvge2}(1.1) that for each $i\in [1, 4]$ and $j\in T_i$,
there exists a $a_i\in C_i\setminus\{i\}$ such that $H$ contains a $(a_i, j)_{(u_i, u_{a_i})}$-path,
and from Lemma~\ref{auvge2}(1.2) that for each $j\in R_{i_1} = T_3$, there exists a $b_{i_1}\in C(u_3)\setminus\{i_1\}$ such that $H$ contains a $(b_i, j)_{(u_3, v_{b_i})}$-path.
Thus, $a_3\ne b_{i_1}$ by Lemma~\ref{lemma06}.

Since $|T_{u v}|\ge \Delta$ and $1, a_1\in C(u_1)$, we have $t_1\ge 2$.
Likewise, $t_2\ge 2$, $t_3\ge 4$, and $t_4\ge 3$.

It follows from Corollary~\ref{auv4} that for each $i\in [1, 4]$,
$\bigcup_{j\in A_{u v}\setminus \{i\}} T_j\subseteq C(u_i)$,
and for each $i, j, l\in A_{u v}$, $|T_i\cup T_j\cup T_l| = |T_i| + |T_j| + |T_l| = t_i + t_j + t_l$.

Then, $d(u_1) \ge |\{1, a_1\}| + |T_2\cup T_3\cup T_4| = 2 + t_2 + t_3 + t_4\ge 2 + 2 + 4 + 3 = 11$,
$d(u_2) \ge |\{1, a_2\}| + |T_3\cup T_4\cup T_1| = 2 + t_3 + t_4 + t_1\ge 2 + 4 + 3 + 2 = 11$,
$d(u_3) \ge |\{3, i_1, a_3, b_{i_1}\}| + |T_4\cup T_1\cup T_2| = 4 + t_4 + t_1 + t_2 \ge 4 + 3 + 2 + 2= 11$,
$d(u_4) \ge |\{4, i_2, a_4\}| +|T_1\cup T_2\cup T_3| = 3 + t_1 + t_2 + t_3\ge 3 + 2 + 2 + 4 = 11$,
and thus none of ($A_7$) and ($A_8$) holds.
\end{proof}
\begin{lemma}
\label{3to21}
If $a_{u v} = 3$, then in $O(1)$ time either an acyclic edge $(\Delta + 5)$-coloring for $G$ can be obtained,
or an acyclic edge $(\Delta + 5)$-coloring for $H$ can be obtained such that $|C(u)\cap C(v)| \le 2$.
\end{lemma}
\begin{proof}
We discuss the following two cases for $d(u) = 4$ and $d(u) = 5$, respectively.

{\bf Case 1.} $d(u) = 4$ and $c(u u_i) = i$, $i\in [1, 3] = A_{uv}$.

In this case, assume $c(v v_i) = i$, $i\in [1, d(v) - 1]$,
and $t_{u v} = \Delta + 5 - (d(v) - 1) = \Delta + 6 - d(v)$.

Assume that $v u_2\in E(G)$.
Since $t_{u v}\ge \Delta - 1$ and $2, c(v u_2)\in C(u_2)$,
we have $T_2 \ne \emptyset$.
It follows from Corollary~\ref{auv3}(1.1) that for each $j\in T_2$,
there exists a $a_2\in \{1, 3\}$ such that $H$ contains a $(a_i, j)_{(u_i, u_{a_i})}$-path.
When $c(v u_2) = \alpha\in \{1, 3\}$, let $\beta = \{1, 3\}\setminus\{\alpha\}$,
thus it follows from Corollary~\ref{auv3}(1.1) that for each $j\in T_2$,
$H$ contains a $(\beta, j)_{(u_2, u_\beta)}$-path and a $(\beta, j)_{(u_\beta, v_\beta)}$-path,
a contradiction.
In the other case, if $v u_2\in E(G)$, then $T_2 \ne \emptyset$, for each $j\in T_2$,
there exists a $a_2\in \{1, 3\}$ such that $H$ contains a $(a_i, j)_{(u_i, u_{a_i})}$-path and $c(v u_2)\not\in A_{u v}$.

{Case 1.1.} $d(v)\le 6$.
Then $t_{u v}\ge \Delta$.
It follows that for each $i\in [1, 3]$, $T_i\ne\emptyset$,
and from Corollary~\ref{auv3}(1.1) that for each $i\in [1, 3]$ and $j\in T_i$,
there exists a $a_i\in C_i\setminus \{i\}$ such that $H$ contains a $(a_i, j)_{(u_i, u_{a_i})}$-path.
Further, by Corollary~\ref{auv3}(7), mult$_{S_u}(j)\ge 3$ for each $j\in T_{u v}$,
or if there exists a $j\in T_{u v}$ such that mult$_{S_u}(j)= 2$, then $V_{u v}\subseteq S_u$.
One sees that $t_{u v}  = \Delta + 6 - d(v)$, $|V_{u v}| = d(v) - 1 - 3 = d(v) - 4$,
and $t_{u v} - |V_{u v}| = \Delta + 10 - 2d(v)\ge 4 + 10 -2\times 6 = 2 > 0$.
Hence,
\begin{equation*}
\begin{aligned}
\sum_{x\in N(u)}d(x) &\ge 2\times t_{u v} + \sum_{i\in [1, 3]}|\{1, a_i\}| + |V_{u v}| + d(v)\\
                     & = 2\times (\Delta + 6 - d(v)) + 2\times 3 + d(v) - 4 + d(v) = 2\Delta + 14.
\end{aligned}
\end{equation*}
and thus ($A_{6.3}$) holds, i.e., $u v$ is contained in two $3$-cycles $u v u_2$, $u v u_3$ and $\Sigma_{x\in N(u)}d(x) = 2\Delta + 14$.

Since $v u_2, v u_3\in E(G)$, $c(v u_2), c(v u_2)\not\in [1, 3]$,
i.e., $1, 2, 3\in C(v)\setminus \{c(v u_2), c(v u_3)\}$.
It follows that $d(v)= 6$ and assume $(v u_2, v u_3)_c = (4, 5)$.
Further, we have $C_{uv}\cap C(u_1) = \{1, a_1\}$, $a_1\in \{2, 3\}$;
$C_{u v}\cap C(u_2) = \{2, 4, a_2\}$, $a_2\in \{1, 3\}$;
and $C_{u v}\cap C(u_3) = \{3, 5, a_3\}$, $a_3\in \{1, 2\}$.

Let $\gamma = \{1, 3\}\setminus \{a_2\}$.
One sees clearly that $4\not\in S_u\setminus C(u_2)$,
$5\not\in C(u_2)$,
and for each $j\in T_2$,
$H$ contains a $(a_2, j)_{(u_2, u_{a_2})}$-path and a $(\gamma, j)_{u_\gamma, v_\gamma}$-path.
Thus, $v u_2\to T_2$ and $uv\overset{\divideontimes}\to 4$.

{Case 1.2.} $d(v) = 7$ and ($A_{6.4}$) holds.
Since $v u_2, v u_3\in E(G)$, $1, 2, 3\in C(v)\setminus \{c(v u_2), c(v u_3)\}$.
Assume $(v u_2, v u_3)_c = (4, 5)$.
Further, we have $a_2\in C(u_2)$, $a_2\in \{1, 3\}$;
and $a_3\in C(u_3)$, $a_3\in \{1, 2\}$.

When $4\not\in C(u_1)$ and $H$ contains neither a $(2, 4)_{(u_1, u_2)}$-path nor a $(3, 4)_{(u_1, u_2)}$-path,
$u u_1\to 4$ reduces to the case where $c(v u_2)\in A_{u v}$.
In the other case, $4\in C(u_1)$, or $G$ contains a $(4, i)_{(u_1, u_i)}$-path for some $i\in \{2, 3\}$.
i.e., $[2, 4]\cap C(u_1) \ne \emptyset$.
It follows that $T_1\ne \emptyset$ and for each $j\in T_1$,
there exists a $a_1\in \{2, 3\}$ such that $H$ contains a $(a_1, j)_{(u_1, u_{a_1})}$-path.

Since $|\{1, a_1\}| + |\{2, 4, a_2\}| + |\{3, 5, a_3\}|+ 2t_{u v} = 2(\Delta - 1) + 8 = 2\Delta + 6\le \sum_{i\in [1, 3]}d(u_i)\le 2\Delta + 8$, there exists a $j\in T_{u v}$ such that mult$_{S_u}(j)= 2$.
It follows from Corollary~\ref{auv3}(4) that $6\in S_u$.

Since $|\{1, a_1\}| + |\{2, 4, a_2\}| + |\{3, 5, a_3\}|+ 2t_{u v} + |\{6\}|= 2\Delta + 7\le \sum_{i\in [1, 3]}d(u_i)\le 2\Delta + 8$, mult$_{S_u}(i)\le 2$ for each $i\in [4, 6]$,
and there exists $i, j\in [4, 6]$ such that mult$_{S_u}(i) = $mult$_{S_u}(j) = 1$.

Assuming w.l.o.g. $5\not\in C(u_1)\cup C(u_2))$.
It follows that $3\in C(u_1)$.
When $4/6\not\in C(u_3)$, $v u_3\to T_3$ and $uv\overset{\divideontimes}\to 5$.
When $6\not\in C(u_2)$ and $4\not\in C(u_1)\cup C(u_3)$, $v u_2\to T_2$ and $u v\overset{\divideontimes}\to 4$.
In the other case, $4/6\in C(u_3)$, and if $4\not\in C(u_1)\cup C(u_3)$, then $6\in C(u_2)\cap C(u_3)$.

It follows that mult$_{S_u}(5) = 1$, mult$_{S_u}(4) + $mult$_{S_u}(6) = 3$,
and for any $j\in T_{uv}$, mult$_{S_u}(j) = 2$.
Hence, $1, 2, 3\in S_u$, $C_1 = \{1, 3\}$, $C_2 = \{2, 1\}$, $C_3 = \{3, 2\}$,
for each $i\in [1, 3]$, mult$_{S_u}(i) = 1$,
and $c(u_1 u_2), c(u_1 u_3)\in T_{uv}$.
Assuming w.l.o.g. $(u_1 u_2, u_1 u_3)_c = (7, 8)$.
Since mult$_{S_u}(7) = 2$, $7\not\in C(u_3)$.
Since $2\not\in C(u_1)$, $H$ contains a $(1, 7)_{(u_1, v}$-path,
and thus there is no $(2, 7)_{(u_2, u_3)}$-path, a contradiction.

{\bf Case 2.} $d(u) = 5$ and $c(uu_i) = i$, $i\in [1, 4]$.

In this case, assume $c(v v_i) = i$, $i\in A_{u v}\cup [5, d(v)] = \{i_1, i_2, i_3\}\cup [5, d(v)]$,
and $t_{u v} = \Delta + 5 - d(v)\ge \Delta - 1$.
One sees that $c(u u_3), c(v u_3)\in C(u_3)$ and $c(u u_4), c(v u_4)\in C(u_4)$.
Thus, $T_3\ne\emptyset$, $T_4\ne\emptyset$,
and it follows that for each $i\in [3, 4]$, if $i\in A_{u v}$,
then for each $j\in T_i$, there exists a $a_i\in C(u_i)\setminus \{c(v u_i)\}$ such that $H$ contains a $(a_i, j)_{(u_i, u_{a_i})}$-path.
By Corollary~\ref{auv3}(2), for each $i\in A_{u v}$,
$\bigcup_{j\in A_{u v}\setminus \{i\}} T_j\subseteq C(u_i)$,
and from Corollary~\ref{auv3}(3) that for each $i, j\in A_{u v}$,
$T_i\cap T_j = \emptyset$,
and then $|T_i\cup T_j| = |T_i| + |T_j| = t_i + t_j$.

By Corollary~\ref{auv3}(1) and Lemma~\ref{lemma06}, we proceed with the following proposition:
\begin{proposition}
\label{prop32010}
For each $i\in [3, 4]$,
if $c(u u_i)\in A_{u v}$ and $c(v u_i) = i_1\in A_{u v}$,
let $A_{u v}\setminus \{i, i_1\} = \{i_2\}$, $U_{u v} = \{a_i\}$,
then for each $j\in T_i$,
$H$ contains a $(i_2, j)_{(u_{i_2}, v_{i_2})}$-path, a $(a_i, j)_{(u_i, u_{a_i})}$-path,
and there exists a $b_{i_1}\in (C(u_i)\cap V_{u v})\setminus \{i, i_1, a_i\}$ such that $H$ contains a $(b_{i_1}, j)_{(u_i, v_{b_{i_1}})}$-path. \hfill $\Box$
\end{proposition}

{Case 2.1.} $d(v) = 5$.
Then $C(v) = A_{u v}\cup \{5\}$, $T_{u v} = [6, \Delta + 5]$ and $t_{u v} = \Delta$.

It follows that $T_i\ne\emptyset$ for each $i\in A_{u v}$,
and from Corollary~\ref{auv3}(1.2) that for each $i\in A_{u v}$, $j\in T_i$,
there exists a $a_i\in C_i\setminus \{i\}$ such that $H$ contains a $(a_i, j)_{(u_i, u_{a_i})}$-path.
By considering the following three scenarios, none of ($A_7$) and ($A_8$) holds.
\begin{itemize}
\parskip=0pt
\item
    $A_{u v} = [2, 4]$ and $c(v u_3) = i_1\in \{2, 4\}$.
    It follows that there exists $a_i$, $i\in [2, 4]$ and $b_{i_1}$ such that $a_2\in C(u_2)$,
    $a_4\in C(u_4)\setminus \{4, c(v u_4)\}$, $a_3, b_{i_1}\in C(u_3)\setminus \{3, i_1\}$.
    Then $t_2\ge 2$, $t_3\ge 4$ and $t_4\ge 3$.
    Since $T_2\cup T_4\subseteq C(u_3)$ and $T_2\cup T_3\subseteq C(u_4)$,
    thus, $d(u_3)\ge |\{3, i_1, a_3, b_{i_1}\}| + |T_2\cup T_4| = 4 + |T_2| + |T_4|\ge 4 + 2 + 3 = 9$,
    and $d(u_4)\ge |\{4, c(v u_4), a_4\}| + |T_2\cup T_3| = 3 + |T_2| + |T_3|\ge 3 + 2 + 4 = 9$.
\item
    $A_{uv} = [1, 3]$ and $c(v u_3) \in [1, 2]$, say $c(v u_3) = 1$.
    It follows that there exists $a_i$, $i\in [1, 3]$ and $b_1$ such that $a_1\in C(u_1)$, $a_2\in C(u_2)$,
    $4, 5\in C(u_3)\setminus \{3, 1\}$,
    and $G$ contains a $(4, j)_{(u_3, u_4)}$-path for each $j\in T_3$,
    which implies that $T_3\subseteq C(u_4)\setminus \{c(v u_4)\}$.
    Then $t_1\ge 2$, $t_2\ge 2$, and $t_3\ge 4$.
    Since $T_2\cup T_3\subseteq C(u_1)$, $T_3\cup T_1\subseteq C(u_2)$ and $T_1\cup T_2\subseteq C(u_3)$,
    thus, $d(u_1)\ge |\{1, a_1\}| + |T_2\cup T_3| = 2 + |T_2| + |T_3|\ge 2 + 2 + 4 = 8$,
    $d(u_2)\ge |\{2, a_2\}| + |T_1\cup T_3| = 2 + |T_1| + |T_3|\ge 2 + 2 + 4 = 8$,
    $d(u_3)\ge |\{3, 1, 4, 5\}| + |T_1\cup T_2| = 4 + |T_1| + |T_2|\ge 4 + 2 + 2 = 8$,
    and $d(u_4)\ge |\{4, c(vu_4)\}| + |T_3| = 2 + |T_3|\ge 2 + 4 = 6$.
\item
    $A_{u v} = [1, 3]$ and $(v u_3, v u_4)_c = (5, i_1)$, where $\{i_1, i_2, i_3\} = [1, 3]$.
    It follows that there exists $a_i$, $i\in [1, 3]$ and $b_{i_1}$ such that $a_1\in C(u_1)$, $a_2\in C(u_2)$,
    $a_3\in C(u_3)\setminus \{c(v u_3)\}$, $b_{i_1}\in C(u_4)\setminus \{4, i_1\}$.
    Then $t_1\ge 2$, $t_2\ge 2$, $t_3\ge 3$, and likewise, $r_{i_2}\ge 2$, $r_{i_3}\ge 2$.
    Since $T_2\cup T_3\subseteq C(u_1)$, $T_3\cup T_1\subseteq C(u_2)$, $T_1\cup T_2\subseteq C(u_3)$,
    and $R_{i_2}\cup R_{i_3}\subseteq C(u_4)$ with $R_{i_2}\cap R_{i_3} = \emptyset$,
    thus, $d(u_1)\ge |\{1, a_1\}| + |T_2\cup T_3| = 2 + |T_2| + |T_3|\ge 2 + 2 + 3 = 7$,
    $d(u_2)\ge |\{2, a_2\}| + |T_1\cup T_3| = 2 + |T_1| + |T_3|\ge 2 + 2 + 3 = 7$,
    $d(u_3)\ge |\{3, c(vu_3), a_3\}| + |T_1\cup T_2| = 3 + |T_1| + |T_2|\ge 3 + 2 + 2 = 7$,
    and $d(u_4)\ge |\{4, i_1, b_{i_1}\}| + |R_{i_2}\cup R_{i_3}| = 3 + |R_{i_2}| + |R_{i_3}|\ge 3 + 2 + 2 = 7$.
\end{itemize}

{Case 2.2.} $d(v) = 6$.
Then $C(v) = \{i_1, i_2, i_3\}\cup \{5, 6\}$, $T_{u v} = [7, \Delta + 5]$ and $t_{u v} = \Delta - 1\ge 5$.

We consider the following two subcases due to the value of $|\{c(u u_3), c(u u_4)\}\cap A_{u v}|$.

(2.2.1.) $|\{c(u u_3), c(u u_4)\}\cap A_{u v}| = 2$, assuming w.l.o.g. $A_{u v} = [2, 4]$.
It follows from \ref{prop32010} that for each $i\in [3, 4]$ and for each $j\in T_i$,
there exists a $a_i \in C_i\setminus \{i\}$ such that $H$ contains a $(a_i, j)_{(u_i, u_{a_i})}$-path,
and if $c(v u_i) = i_1\in A_{u_i}$, then $a_i = 1$ and $b_{i_1}\in \{5, 6\}$.

Since $T_3\subseteq C(u_4)$, we have $T_{u v}\subseteq C(u_3)\cup C(u_4)$.
One sees that if $|C(u_3)\cup C_{u v}|\ge 4$ and $|C(u_4)\cup C_{u v}|\ge 4$,
then $d(u_2) = d(u_3) = d(u_4) = \Delta = 7$ and $C(u_2) = \{2\}\cup T_3\cup T_4 = \{2\}\cup [7, \Delta + 5]$.

We consider the following three scenarios due to the value of $|\{c(v u_3), c(v u_4)\}\cap A_{u v}|$.
\begin{itemize}
\parskip=0pt
\item
     $|\{c(v u_3), c(v u_4)\}\cap A_{u v}| = 2$.
     Assume w.l.o.g. $(v u_3, v u_4)_c = (2, 3)$.
     It follows that $1, 5/6\in C(u_3)$, $1, 5/6\in C(u_4)$ and $T_3\cup T_4 = T_{u v}\subseteq C(u_1)$.
     Then $\Delta = 7$ and for each $i\in [1, 4]$, $d(u_i)= 7$.
     Thus, none of ($A_7$) and ($A_8$) holds.
\item
     $|\{c(v u_3), c(v u_4)\}\cap A_{u v}|\le 1$.
     Assume w.l.o.g. $c(v u_4) = 6$.
     It follows that $a_3\in C(u_3)\setminus \{c(v u_3)\}$ and $a_4\in C(u_4)\setminus \{c(v u_4)\}$.
     When $W_i = \{i, c(v u_i), a_i\}$ for some $i\in \{3, 4\}$,
     assuming w.l.o.g. $W_4 = \{4, 6, a_4\}$,
     it follows from \ref{prop32011}(1) that $6\in C(u_3)\cap C(u_2)$ and $\Delta\in [7, 8]$.
     Then $t_3\ge 3$ and $t_4\ge 2$.
     Since $T_3\cup T_4\subseteq C(u_2)$, $d(u_2)\ge |\{2, 6\}| + t_3 + t_4\ge 7$.
     One sees that if $\Delta = 8$,
     then $T_3\cup T_4 = T_{u v}$ and $d(u_2)\ge |\{2, 6\}| + t_3 + t_4\ge \Delta + 1$.
     Hence, $\Delta = 7$ and $C(u_2) = \{2, 6\}\cup T_3\cup T_4$, where $t_3 = 3$, $t_4 = 2$.
     However, for each $j\in T_2$, $H$ contains no $(a_2, j)_{(u_2, u)}$-path, a contradiction
     \footnote{Then $uu_2\to T_2$ reduces to an acyclic edge $(\Delta + 5)$-coloring for $H$ such that $|C(u)\cap C(v)| = 1$. }.
     In the other case, $w_3\ge 4$ and $w_4\ge 4$.
     Recall that $C(u_2) = \{2\}\cup T_3\cup T_4 = \{2\}\cup [7, \Delta + 5]$.
     \begin{itemize}
     \parskip=0pt
     \item
         $c(vu_3)= i_1$.
         It follows that $1\in C(u_3)$ and $T_3\subseteq C(u_1)$.
         Hence, ($A_{8.4}$) holds and $d(u_1)\le 6$.
         It follows that $[10, 12]\setminus C(u_1)\ne \emptyset$, say $10\not\in C(u_1)$.
         If $10\not\in B_3$, then $(u u_2, u v)\overset{\divideontimes}\to (6, 10)$.
         Otherwise, $10\in B_3$.
         Then $(u u_2, u u_4)\to (6, 10)$ reduces to an acyclic edge $(\Delta + 5)$-coloring for $H$ such that $|C(u)\cap C(v)| = 2$.
     \item
         $c(v u_3) = 5$.
         Then $(u u_2)\to 5$ reduces to the above case where $c(v u_3)\in A_{u v}$.
     \end{itemize}
\end{itemize}

(2.2.2.) $|\{c(u u_3), c(u u_4)\}\cap A_{u v}| = 1$, assuming w.l.o.g. $A_{u v} = [1, 3]$.

Assume that there exists a $j^*\in T_4$ such that mult$_{S_u}(j^*) = 2$ or $H$ contains no $(j^*, i)_{(u_4, u_i)}$-path for each $i\in C(u)$.
It follows from Corollary~\ref{auv3}(4) that $4\in \bigcup_{i\in [1, 3]}B_{i}$,
mult$_{\bigcup_{i\in A_{u v}}C(u_i)}(4)\ge 2$,
and for each $i\in [1, 3]$, $4\in C(u_i)$,
or there exists a $a_i\in C(u_i)$ such that $H$ contains a $(4, a_i)_{(u_{a_i}, u_i)}$-path.
If $C(u_i) = \{i\}\cup T_{u v}$ for some $i\in [1, 3]$,
then $(u u_i, u u_4)\to (4, j^*)$ reduces to an acyclic edge $(\Delta + 5)$-coloring for $H$ such that $a_{u v} = 2$.
Otherwise, for each $i\in [1, 2]$, $T_i\ne \emptyset$ and there exists a $a_i\in C(u_i)$.
Together with Corollary~\ref{auv3}(1),
we proceed with the following proposition:
\begin{proposition}
\label{prop32011}
If there exists a $j^*\in T_4$ such that mult$_{S_u}(j^*) = 2$ or $H$ contains no $(j^*, i)_{(u_4, u_i)}$-path for each $i\in C(u)$, then
\begin{itemize}
\parskip=0pt
\item[{\rm (1)}]
	$4\in \bigcup_{i\in [1, 3]}B_{i}$, mult$_{\bigcup_{i\in A_{u v}}C(u_i)}(4)\ge 2$,
\item[{\rm (2)}]
    for each $i\in [1, 3]$,
    \begin{itemize}
    \parskip=0pt
    \item[{\rm (2.1)}]
	$T_i\ne \emptyset$, $4\in C(u_i)$,
    or there exists a $a_i\in C(u_i)$ such that $H$ contains a $(4, a_i)_{(u_{a_i}, u_i)}$-path;
    \item[{\rm (2.2)}]
    $2/3/4\in C(u_1)$, and if $(T_4\cup \{4\})\setminus C(u_1)\ne \emptyset$, then $2/3\in C(u_1)$;
    \item[{\rm (2.3)}]
    $1/3/4\in C(u_2)$, and if $(T_4\cup \{4\})\setminus C(u_2)\ne \emptyset$, then $1/3\in C(u_2)$;
    \item[{\rm (2.4)}]
    $1/2/4\in C(u_3)$, and if $(T_4\cup \{4\})\setminus C(u_3)\ne \emptyset$, then $1/2\in C(u_3)$. 		
    \end{itemize}
\end{itemize}
\end{proposition}

Now  we can have the following three scenarios due to $c(v u_3)$ and $c(v u_4)$:

(2.2.2.1.) $c(v u_3)\in [1, 2]$, say $c(v u_3) = 2$.
It follows that $4\in C(u_3)$ and for each $j\in T_3$,
$H$ contains a $(1, j)_{(u_1, v_1)}$-path, a $(4, j)_{(u_3, u_4)}$-path and
there exists a $b_2\in \{5, 6\}\cap C(u_3)$ such that $H$ contains a $(b_2, j)_{(u_3, v_{b_2})}$-path.
Then $T_3\subseteq C(u_4)$, $T_4\subseteq C(u_3)$,
$w_3\ge 4$ and $t_3 = \Delta - 1 -(d(u_3) - w_3) = \Delta - d(u_3) + w_3 - 1\ge 3$,
$t_4 = \Delta - d(u_4) + w_4 - 1\ge 1$.

One sees that if $\Delta \le 7$, then $\Delta - 1 + w_3 + w_4\le 2\Delta$ and $w_3 + w_4\le \Delta + 1\le 8$;
if $\Delta \ge 8$, then $\Delta - 1 + w_3 + w_4\le 14$, and $w_3 + w_4\le 15 - \Delta\le 7$.
It follows that $\Delta\le 9$, $6\le w_3 + w_4\le8 $, and $w_3 + w_4 = 8$ if and only if $d(u_3) = d(u_4) = \Delta = 7$ and $T_3\cup T_4 = T_{u v} = [7, 12]$.

When $w_4 = 4$, it follows that $T_{u v} = [7, 12]$, $t_3 = t_4 = 3$.
Since $T_3\subseteq C(u_1)\cap C(u_2)\cap C(u_4)$,
$2\Delta + 11\le \sum_{i\in [1, 4]}w_i + 2t_{u v} + t_3\le \sum_{i\in [1, 4]}d(u_i)\le 2\Delta + 13$,
there exists a $j^*\in T_{4}$ such that mult$_{S_u}(j^*) = 2$.
One sees clearly that for each $i\in [1, 2]$, $(T_4\cup \{4\})\setminus C(u_i) \ne \emptyset$.
It follows from \ref{prop32011} that $4\in C(u_1)\cup C(u_2)$, and $2/3\in C(u_1)$, $1/3\in C(u_2)$.
Hence, $w_1 + w_2 \ge |\{1, 2/3\}| + |\{2, 1/3\}| + |\{4\}| = 5$.
Thus, $\sum_{i\in [1, 4]}d(u_i)\ge 2\Delta + 14$, a contradiction.

In the other case, $w_4\le 3$.
Note that
\[
\begin{aligned}
\sum_{i\in [1, 4]}d(u_i)
& \ge \sum_{i\in [1, 4]}w_i + 2t_{uv} + t_3 = \sum_{i\in [1, 4]}w_i + 2(\Delta - 1) + \Delta - d(u_3) + w_3 - 1 \\
& = 2\Delta + (\Delta - d(u_3)) - 3 + w_1 + w_2 + w_4 + 2w_3,
\end{aligned}
\]
\begin{equation}
\label{eq18}
\begin{aligned}
t_3 + t_4 + 2
& = \Delta - d(u_3) + w_3 - 1 + \Delta - d(u_4) + w_4 - 1 + 2 \\
& = (\Delta - d(u_3)) + (\Delta - d(u_4)) + w_3 + w_4 \ge w_3 + w_4.
\end{aligned}
\end{equation}

\begin{itemize}
\parskip=0pt
\item
    $c(vu_4) = i_1\in \{1, 3\}$.
    Then for each $j\in T_4$, there exists a $b_{i_1}\in C(u_4)\setminus \{4, i_1\}$ such that $H$ contains a $(j, b_{i_1})_{(u_4, v)}$-path.
    Then $w_4 = 3$, i.e., $W_4 = \{4, i_1, b_{i_1}\}$ and $H$ contains no $(i, j)_{(u_4, u_i)}$-path for each $i\in C(u)$, $j\in T_4$.
    It follows from \ref{prop32011} that $2/3/4\in C(u_1)$, $1/3/4\in C(u_2)$, and $4\in C(u_1)\cup C(u_2)$.
    Then $w_1 + w_2\ge |\{1, 2/3/4\}| + |\{2, 1/3/4\}| = 4$,
    and $\sum_{i\in [1, 4]}d(u_i)\ge 2\Delta + (\Delta - d(u_3)) - 3 + w_1 + w_2 + w_4 + 2w_3 = 2\Delta + (\Delta - d(u_3)) + 4 + 2w_3$.

    When $\Delta\ge 8$, $\sum_{i\in [1, 4]}d(u_i)\ge 2\Delta + (\Delta - d(u_3)) + 4 + 2w_3\ge 2\Delta + 1 + 4 + 2\times 4 = 2\Delta + 13\ge \Delta + 21$.
    When $w_3 = 5$, $\sum_{i\in [1, 4]}d(u_i)\ge 2\Delta + (\Delta - d(u_3)) + 4 + 2w_3\ge 2\Delta + 4 + 2\times 5 = 2\Delta + 14$.

    In the other case, $\Delta \in [6, 7]$, $w_3 = 4$, i.e., $W_3 = \{2, 3, 4, 5\}$ and for each $j\in T_3$, $H$ contains a $(5, j)_{(u_3, v_5)}$-path.
    Then $\sum_{i\in [1, 4]}d(u_i)\ge 2\Delta + (\Delta - d(u_3)) + 4 + 2w_3 \ge 2\Delta + (\Delta - d(u_3)) + 12$.
    Since $t_4\ge 2$, there exists a $j^*\in T_{4}$ such that mult$_{S_u}(j^*) = 2$.
    It follows from \ref{prop32011} (2.2--2.3) that $2/3/4\in C(u_1)$, $1/3/4\in C(u_2)$,
    and by Corollary~\ref{auv3}(7), $6\in S_u$.

    One sees that $t_3 + t_4 + 2 = (\Delta - d(u_3)) + (\Delta - d(u_4)) + w_3 + w_4\ge w_3 + w_4 \ge 7$.
    It follows that if $6\in \{d(u_3), d(u_4)\}$,
    then for each $i\in [1, 2]$, $T_4\setminus C(u_i)\ne \emptyset$, $2/3\in C(u_1)$, $1/3\in C(u_2)$,
    and $4\in C(u_1)\cup C(u_2)$, i.e., $w_1 + w_2\ge 5$.

    If $\Delta = 6$,
    then $\sum_{i\in [1, 4]}d(u_i)\ge 2\Delta + (\Delta - d(u_3)) - 3 + w_1 + w_2 + w_4 + 2w_3 \ge 2\Delta - 3 +5 + 3 + 2\times 4 = 2\Delta + 13 > 2\Delta + 12$.
    Otherwise, $\Delta = 7$.

    When $d(u_3) = 6$,
    $\sum_{i\in [1, 4]}d(u_i) \ge  2\Delta + (\Delta - d(u_3)) - 3 + w_1 + w_2 + w_4 + 2w_3 = 2\Delta + 1 - 3 + 5 + 3 + 2\times 4 = 2\Delta + 14 > 2\Delta + 13$.
    In the other case, $d(u_3) = 7$ and $T_3 = [7, 9]$.
    \begin{itemize}
    \parskip=0pt
    \item[(i)]
         $d(u_4) = 6$.
         Since $2/3\in C(u_1)$, $1/3\in C(u_2)$, $4\in C(u_1)\cup C(u_2)$, and $\sum_{i\in [1, 4]}d(u_i) \ge 2\Delta + 5 + \Delta - d(u_3) + 2w_3 = 2\Delta + 13$,
         it follows that $w_1 + w_2 = 5$,
         and $W_1\biguplus W_2 = \{1, 2/3, 2, 1/3, 4\}$.
    \item[(ii)]
         $d(u_3) = d(u_4) = 7$,
         say $C(u_3)= [2, 5]\cup [10, 11, 12]$, $C(u_4) = \{4, i_1, b, 12\}\cup [7, 9]$.
         Since $12\in C(u_1)\cup C(u_2)$, mult$_{S_u}(12) = 3$.
         Since $2(\Delta - 1) + |[7, 9]\cup \{12\}| + \sum_{i\in [1, 4]}w_i = 2\Delta + 9 + w_1 + w_2$,
         we have $w_1 + w_2 = 4$, and $W_1\biguplus W_2 = \{1, 2/3/4, 2, 1/3/4\}$.
         If $3\not\in C(u_1)$, then $(uu_1, uv)\overset{\divideontimes}\to (5, 7)$.
         Otherwise, $3\in C(u_1)$ and $W_1\biguplus W_2 = \{1, 3, 2, 4\}$.
    \end{itemize}
    Hence, for (i) and (ii), mult$_{C(u_1)\cup C(u_2)}(4) = 1$.
    Since $6\in S_u$, we have $b = 6$ and then $5\not\in C(u_1)\cup C(u_2)\cup C(u_4)$.
    Then $(u u_4, u u_3)\to (5, 7)$ reduces to an acyclic edge $(\Delta + 5)$-coloring for $H$ such that $C(u)\cap C(v) = \{1, 2, 5\}$.
    One sees that mult$_{C(u_1)\cup C(u_2)}(4) = 1$.
    Then by the similar discussion to show Corollary~\ref{auv3}(2),
    an acyclic edge $(\Delta + 5)$-coloring for $H$ can be obtained such that $|C(u)\cap C(v)|\le 2$.
\item
    $c(vu_4) = 6$.
    Note that $t_3 = \Delta - d(u_3) + w_3 - 1\ge 4 - 1 = 3$,
    $t_4 = \Delta - d(u_4) + w_4 - 1\ge 2 - 1 = 1$,
    and
    \[
    \begin{aligned}
     \sum_{i\in [1, 4]}w_i + 2t_{uv} + t_3 + t_4 & = \sum_{i\in [1, 4]}w_i + 2(\Delta - 1) + t_3 + t_4 = w_1 + w_2 \\ & + 2(w_3 + w_4) + 2\Delta - 4 + (\Delta - d(u_3)) + (\Delta - d(u_4))
    \end{aligned}
    \]

    \quad If $6\not\in C(u_1)$ and $H$ contains no $(6, i)_{(u_1, u_i)}$-path for each $i\in [2, 3]$,
    then $u u_1\to 6$ reduces to the above case where $C(u)\cap C(v) = \{2, 6, 3\}$ and $c(v u_4)\in A_{u v}$.
    Otherwise, we proceed with the following proposition:
    \begin{proposition}
    \label{6/2/3inCu1}
      $6/2/3\in C(u_1)$, and if $6\in C(u_1)$, then $H$ contains a $(6, i)_{(u_1, u_i)}$-path for some $i\in [2, 3]$, and $6\in C(u_1)\cup C(u_i)$.
    \end{proposition}

\quad First, assume that $w_4 = 3$ with $b\in C(u_4)$.
    Since for each $j\in T_4$, $H$ contains at most one of $(j, b)_{(u_4, u_b)}$-path and $(j, b)_{(u_4, v_b)}$-path,
    it follows from Corollary~\ref{auv3}(4) and (5) that mult$_{\bigcup_{i\in [1, 3]}C(u_i)}(i) \ge 2$ for some $i\in \{4, 6\}$.
    It follows that $4/6\in C(u_1)\cup C(u_2)$, $w_1 + w_2\ge 3$.

    \quad Since $\sum_{i\in [1, 4]}w_i + 2t_{u v} + t_3 + t_4 \ge w_1 + w_2 + 2(w_3 + w_4) + 2\Delta - 4 + (\Delta - d(u_3)) + (\Delta - d(u_4)) \ge 2\Delta - 4 + 3 + 2\times 7 = 2\Delta + 13$,
    there exists a $j^*\in T_4$ such that mult$_{S_u}(j^*) = 2$.
    It follows from \ref{prop32011} that $2/3/4\in C(u_1)$, $1/3/4\in C(u_2)$, $4\in C(u_1)\cup C(u_2)$,
    and by Corollary~\ref{auv3}(7), $5, 6\in S_u$.
    By Eq.~(\ref{eq18}), $t_3 + t_4 + 2 = (\Delta - d(u_3)) + (\Delta - d(u_4)) + w_3 + w_4 \ge (\Delta - d(u_3)) + (\Delta - d(u_4)) + 7\ge 7 > 6$.
    Further, if $\Delta\ge 8$, then $t_3 + t_4 + 2 \ge (\Delta - 7) + (\Delta - 7) + 7 \ge \Delta + 1$,
    if $\Delta = 7$, $6\in \{d(u_3), d(u_4)\}$ or $w_3 = 5$, then $t_3 + t_4 + 2\ge 8 > 7$.

    \quad One sees that if $\Delta\ne 7$, or $\Delta = 7$ and $6\in \{d(u_3), d(u_4)\}$ or $w_3 = 5$,
    then for each $i\in [1, 2]$, $T_4\setminus C(u_i)\ne \emptyset$,
    $2/3\in C(u_1)$, $1/3\in C(u_2)$, $4\in C(u_1)\cup C(u_2)$.
    Hence, $\sum_{i\in [1, 4]}d(u_i)\ge 2\Delta + (\Delta - d(u_3)) - 3 + w_1 + w_2 + w_4 + 2w_3 \ge 2\Delta + (\Delta - d(u_3)) - 3 + 5 + 3 + 2\times 4 = 2\Delta + (\Delta - d(u_3)) + 13$,
    and if $\Delta = 7$, $d(u_3) = 6$ or $w_3= 5$, then $\sum_{i\in [1, 4]}d(u_i)\ge 2\Delta + 14$.

    \quad Hence, it suffices to assume $d(u_3) = \Delta = 7$, $w_3 = 4$,
    and if $d(u_4) = 6$, $W_1\biguplus W_2 = \{1, 2/3\}\cup \{2, 1/3\}\cup \{4\}$.
    If $d(u_4) = 7$, say $12\in C(u_3)\cap C(u_4)$, since $12\in C(u_1)\cup C(u_2)$,
    $12\in T_0$, and $\sum_{i\in [1, 4]}d(u_i)\ge \sum_{i\in [1, 4]}w_i + 2t_{uv} + t_3 + |\{12\}| = 2\Delta + (\Delta - d(u_3)) - 2 + w_1 + w_2 + w_4 + 2w_3 \ge 2\Delta + 13$,
    thus, $W_1 = \{1, 2/3/4\}$, $W_2 = \{2, 1/3/4\}$.

    \quad One sees clearly that $5, 6\not\in C(u_1)\cup C(u_2)$.
    By \ref{6/2/3inCu1}, $3\in C(u_1)$ and $6\in C(u_3)$.
    Then $b = 5$, and since $6\not\in C(u_1)\cup C(u_2)$,
    by Corollary~\ref{auv3} (5), for each $j\in T_4$, $H$ contains a $(5, j)_{(u_4, v_5)}$-path.
    \begin{itemize}
    \parskip=0pt
    \item
        $d(u_4) = 6$.
        Since for each $i\in [1, 2]$, $T_4\setminus C(u_i)\ne \emptyset$,
        it follows from Corollary~\ref{auv3} (6.1) that $1, 2\in S_u$ and then $1\in C(u_2)$.
        Then $(u u_2, u u_3, u u_4)\to (5, 7, 3)$ and $uv \overset{\divideontimes}\to T_4\setminus C(u_1)$.
    \item
        $d(u_4) = 7$. 
        It follows that $W_1 = \{1, 3\}$, $W_2 = \{2, 4\}$.
        Since $1\not\in S_u$ and $2, 3\not\in C(u_1)$,
        it follows from Corollary~\ref{auv3} (1.2) and (6.1) that $T_4\subseteq C(u_1)\cap C(u_2)$, a contradiction.
    \end{itemize}

    \quad Next, assume that $w_4 = 2$.
    One sees clearly that for each $j\in T_4$, $H$ contains neither a $(i, j)_{(u_4, u_i)}$-path nor a $(l, j)_{(u_4, v_l)}$-path, where $i\in C(u)$, $l\in C(v)$.
    It follows from Corollary~\ref{auv3} (4--5) that $4/2/3, 6/2/3\in C(u_1)$, $4/1/3, 6/1/3\in C(u_2)$,
    and 
    $4, 6\in C(u_1)\cup C(u_2)$.

    \quad If $2/3\in C(u_1)$ and $1/3\in C(u_2)$,
    then $w_1 + w_2\ge 6$,
    and thus $\sum_{i\in [1, 4]}w_i + 2t_{uv} + t_3 + t_4 \ge w_1 + w_2 + 2(w_3 + w_4) + 2\Delta - 4 + (\Delta - d(u_3)) + (\Delta - d(u_4))\ge w_1 + w_2 + 2\times 6 + 2\Delta - 4 \ge 6 + 2\Delta + 8 = 2\Delta + 14$.
    Otherwise, $2, 3\not\in C(u_1)$ or $1, 3\in C(u_2)$.
    Then $u u_1\to T_1$ if $2, 3\not\in C(u_1)$,
    or $u u_2\to T_2$ if $1, 3\in C(u_2)$,
    and $u u_4\to 5$ reduces to the case (2.2.1.).
\end{itemize}

(2.2.2.2.) $c(vu_3) = 5$.
It follows from Lemma~\ref{auvge2} (1.1) that for each $j\in T_3$,
there exists a $a_3\in \{1, 2, 4\}$ such that $H$ contains a $(a_3, j)_{(u_3, u_{a_3})}$-path.
If $5\not\in C(u_2)$ and $H$ contains no $(i, 5)_{(u_2, u_i)}$-path for each $i\in C(u)$,
then $u u_2\to 5$ reduces to the case (2.2.2.1.).
Otherwise, $5\in C(u_2)$ or $H$ contains a $(a_2, 5)_{(u_2, u_{a_2})}$-path for some $a_2\in C(u)$.
Hence, since $|T_{uv}\cup \{5\}| = \Delta$,
together with Lemma~\ref{auvge2} (1.1),
we proceed with the following proposition:
\begin{proposition}
\label{234/134}
\begin{itemize}
\parskip=0pt
\item[{\rm (1)}]
	$1/3/4\in C(u_2)$, $2/3/4\in C(u_1)$;
\item[{\rm (2)}]
    for each $j\in T_1\cup \{5\}$,
    $H$ contains a $(a_1, j)_{(u_1, u_{a_1})}$-path for some $a_1\in \{2, 3, 4\}$;
\item[{\rm (3)}]
    for each $j\in T_2\cup \{5\}$,
    $H$ contains a $(a_2, j)_{(u_2, u_{a_2})}$-path for some $a_2\in \{1, 3, 4\}$.
\end{itemize}
\end{proposition}

Note that $T_1\cup T_2\subseteq C(u_3)$,
$T_3\cup T_1\subseteq C(u_2)$,
$T_3\cup T_2\subseteq C(u_1)$,
and for each $i\in [1, 3]$, $t_i = \Delta - d(u_i) + w_i - 1$.
It follows from Corollary~\ref{auv3} (5) that $5\in C(u_1)\cap C(u_2)$,
or for each $j\in T_3$,
there exists a $b_5\in C(u_3)\setminus \{a_3, 3\}$ such that $H$ contains a $(b_5, j)_{(u_3, v_{b_5})}$-path,
and by Lemma~\ref{lemma06}, $a_3\ne b_5\in C(u_3)$ and $w_3\ge 4$.

Assume that $5\not\in C(u_1)\cup C(u_2)$.
Since for each $i\in [1, 2]$, $H$ contains a $(5, j)_{(u_i, u_j)}$-path for some $j\in [3, 4]$,
assume w.l.o.g. $3\in C(u_2)$, $4\in C(u_1)$, $5\in C(u_4)$,
and $H$ contains a $(4, 5)_{(u_1, u_4)}$-path, $(3, 5)_{(u_2, u_3)}$-path.
If there exists an $\gamma\in (T_1\cup T_2\cup T_3)\setminus C(u_4)$,
then $(u u_1, u u_4)\to (5, \gamma)$ reduces to the case (2.2.2.1.).
Otherwise, $T_1\cup T_2\cup T_3\subseteq C(u_4)$.
Thus, $d(u_4)\ge |\{4, c(v u_4), 5\}| + t_1 + t_2 + t_3 = 3 + t_1 + t_2 + t_3$.
we proceed with the following proposition:
\begin{proposition}
\label{5notin}
If $5\not\in C(u_1)\cup C(u_2)$, then
$d(u_4)\ge 3 + t_1 + t_2 + t_3$.
\end{proposition}

\begin{itemize}
\parskip=0pt
\item
     $c(vu_4) = i_1$.
     It follows from Lemma~\ref{auvge2} (1.2) that for each $j\in T_4$,
     there exists a $b_{i_1}\in C(v)\setminus\{4, i_1\}$ such that $H$ contains a $(b_{i_1}, j)_{(v_{b_{i_1}}, u_4)}$-path.

     \quad When $5\in C(u_1)\cap C(u_2)$,
     one sees clearly that $t_1\ge 2$, $t_2\ge 2$, and $w_3\ge 3$,
     we have $d(u_3)= 7$, and $d(u_1) = d(u_2) = \Delta$.
     It follows that for each $i\in [1, 3]$, $d(u_i) = \Delta = 7$, $t_3 = 2$,
     $C(u_3) = \{3, 5, a_3\}\cup T_1\cup T_2$,
     $C(u_1) = \{1, a_1, 5\}\cup T_2\cup T_3$,
     $C(u_2) = \{2, a_2, 5\}\cup T_1\cup T_3$,
     and $d(u_4)\le 6$.
     One sees clearly that $(T_1\cup T_2\cup T_3)\setminus C(u_4)\ne \emptyset$.
     It follows from Corollary~\ref{auv3} (6.2) that mult$_{\bigcup_{i\in [1, 3]}C(u_i)}(4)\ge 2$,
     say $a_i = a_j = 4$, $i, j\in [1, 3]$.
     Since $T_i\cup T_j\subseteq C(u_4)$ by \ref{234/134},
     $d(u_4)\ge |\{4, i_1, b_{i_1}\}\cup T_i\cup T_j| = |\{4, i_1, b_{i_1}\}| + |T_i| + |T_j|\ge 3 + t_i + t_j\ge 7> 6$, a contradiction.

     \quad In the other case, since $a_3\ne b_5\in C(u_3)$, $w_3\ge 4$.
     We can have the following two scenarios:
     \begin{itemize}
     \parskip=0pt
     \item
         $5\in C(u_1)\cup C(u_2)$, say $5\in C(u_2)$.
         Since $t_1\ge 1$, $t_2\ge 2$, and $w_3\ge 3$,
         we have $d(u_3)= 7$, and $d(u_1) = d(u_2) = \Delta$.
         It follows that for each $i\in [1, 3]$, $d(u_i) = \Delta = 7$, $t_3 = 3$,
         $C(u_3) = \{3, 5, a_3, b_5\}\cup T_1\cup T_2$,
         $C(u_1) = \{1, a_1\}\cup T_2\cup T_3$,
         $C(u_2) = \{2, a_2, 5\}\cup T_1\cup T_3$,
         and $d(u_4)\le 6$.
         One sees clearly that $(T_1\cup T_2\cup T_3)\setminus C(u_4)\ne \emptyset$.
         It follows from Corollary~\ref{auv3} (6.2) that mult$_{\bigcup_{i\in [1, 3]C(u_i)}}(4)\ge 2$.
         Assuming w.l.o.g. $a_i = a_j = 4$, where $i, j\in [1, 3]$.

         \quad Since $T_i\cup T_j\subseteq C(u_4)$,
         $d(u_4)\ge |\{4, i_1, b_{i_1}\}\cup T_i\cup T_j| = |\{4, i_1, b_{i_1}\}| + |T_i| + |T_j| = 3 + t_i + t_j$.
         Recall that $d(u_4)\le 6$.
         It follows that $4\in C(u_1)\cap C(u_2)$ and by \ref{234/134},
         $T_1\cup T_2\cup \{5\}\subseteq C(u_4)$.
         Hence, $C(u_4) = \{4, i_1, 5\}\cup T_1\cup T_2$.
         Then $v u_4\to T_3$ reduces to an acyclic edge $(\Delta + 5)$-coloring for $H$ such that $a_{u v} = 2$
         \footnote{For each $j\in T_4 = T_3$, $H$ contain no $(5, j)_{(u_3, u_4)}$-path, a contradiction}.
     \item
          $5\not\in C(u_1)\cup C(u_2)$.
          Since $t_1\ge 1$, $t_2\ge 1$, $t_3\ge 3$,
          it follows from \ref{5notin} that $d(u_4)\ge 3 + t_1 + t_2 + t_3\ge 3 + 1 + 1 + 3 = 8 > 7$, a contradiction.
     \end{itemize}
 \item
     $c(v u_4) = 6$.
     When $6\not\in C(u_1)\cup C(u_2)$,
     if $H$ contains no $(3, 6)_{(u_1, u_3)}$-path, then $u u_1\to 6$;
     otherwise, $H$ contains a $(3, 6)_{(u_1, u_3)}$-path, then $u u_2\to 6$,
     this reduces to the above case where $c(v u_4)\in A_{u v}$.
     In the other case, $6\in C(u_1)\cup C(u_2)$.
     Together with $1, 2/3/4\in C(u_1)$ and $2, 1/3/4\in C(u_2)$,
     we have $w_1 + w_2\ge 5$, $t_1 + t_2\ge 3$.
     Note that $d(u_3)\ge |W_3\cup T_1\cup T_2| = |W_3| + |T_1| + |T_2| = w_3 + t_1 + t_2$.
     It follows that $w_3\le 4$ and if $w_3 = 4$, then $t_1 + t_2 = 3$.
     \begin{itemize}
     \parskip=0pt
     \item
          $w_3 = 3$. Then $5\in C(u_1)\cap C(u_2)$.
          Since $w_1 + w_2\ge 7$,
          we have $t_1 + t_2\ge 5$,
          and thus $d(u_3)\ge w_3 + t_1 + t_2 = 3 + 5\ge 8> 7$.
     \item
          $w_3 = 4$. Then $t_1 + t_2 = 3$, $t_3\ge 3$.
          Hence, $5\not\in C(u_1)\cup C(u_2)$.
          It follows from \ref{5notin} that $d(u_4)\ge 3 + t_1 + t_2 + t_3\ge 3 + 3 + 3 = 9 > 7$, a contradiction.
     \end{itemize}
\end{itemize}

This finishes the inductive step for the case where $G$ contains the configuration ($A_6$)--($A_8$),
and thus completes the proof of Lemma~\ref{3to21}.
\end{proof}

By \ref{prop3001}, we proceed with the following proposition,
or otherwise we are done:
\begin{proposition}
\label{auv1}
If $c(\cdot)$ is an acyclic edge $(\Delta(H) + 5)$-coloring for $H$ such that $C(u)\cap C(v) = \{i_1\}$,
then $T_{uv}\subseteq B_{i_1}\subseteq C(u_{i_1})\cap C(v_{i_1})$, and $t_{uv}\le \Delta - 1$.
\end{proposition}

\begin{lemma}
\label{auv1coloring}
If $a_{u v} = 1$, then in $O(1)$ time an acyclic edge $(\Delta + 5)$-coloring for $G$ can be obtained.
\end{lemma}
\begin{proof}
Let $c(u u_i) = i$, $i\in C(u) = [1, d(u) - 1]$, $c(vv_j) = j$, $j\in \{i_1\}\cup [d(u), d(u) + d(v) - 3]$,
where $i_1\in C(u)$.
One sees clearly that $t_{uv} = \Delta + 5 - (d(u) + d(v) - 3) = \Delta - (d(u) + d(v)) + 8$.
It follows from \ref{auv1} that $d(u) + d(v) \ge 9$.
Hence, $d(u) = 4$ and $d(v)\ge 5$, or $d(u) = 5$.

We discuss the following two cases for $d(u) = 4$ and $d(u) = 5$, respectively.

{\bf Case 1.} $d(u) = 4$ and $c(u u_i) = i$, $i\in [1, 3]$, $c(v v_j) = j$,
$j\in \{i_1\}\cup [4, d(v) + 1]$, where $i_1\in [1, 3]$.

In this case, $5\le d(v)\le 7$,
and $t_{u v} = \Delta - (d(u) + d(v)) + 8 = \Delta + 4 - d(v)$.

When $d(v) = 5$, $t_{u v} = \Delta - 1$,
assuming w.l.o.g. $i_1 = 1$,
it follows that $C(u_1) = C(v_1) = \{1\}\cup T_{u v}$,
then $(u u_1, v v_1, u v)\overset{\divideontimes}\to (4, 1, 2)$.
In the other case, $d(v)\in [6, 7]$.

{Case 1.1.} $d(v) = 6$.
Then $t_{u v} = \Delta - 2$, $V_{u v} = [4, 7]$.
Since $u v$ is contained in a $3$-cycle, assume that $v u_2\in E(G)$.
By symmetry, we consider the following three cases.

(1.1.1.) $c(v u_2) = 1$.
It follows from \ref{auv1} that $C(u_2) = [1, 2]\cup T_{u v}$,
$T_{u v}\subseteq C(u_1)$, and there exists an $i\in [4, 7]\setminus C(u_1)$.
Then $(u u_2, u v)\overset{\divideontimes} \to (i, 2)$.

(1.1.2.) $i_1 = 2$.
Assume w.l.o.g. $c(v u_2) = 4$.
It follows from \ref{auv1} that $C(u_2) = [2, 4]\cup T_{u v}$,
$T_{u v}\subseteq C(v_2)$, and there exists an $i\in \{1, 3\}\setminus C(v_2)$.
Then $(v u_2, u v)\overset{\divideontimes} \to (i, 4)$.

(1.1.3.) Assume w.l.o.g. $i_1 = 1$ and $c(v u_2) = 4$.
If $4\not\in C(u_1)$ and $H$ contains no $(3, 4)_{(u_1, u_3)}$-path,
then $u u_1\to 4$ reduces the proof to (1.1.1.).
If $2\not\in C(v_1)$ and $H$ contains no $(2, i)_{(v_1, v_i)}$-path for some $i\in [5, 7]$,
then $v v_1\to 2$ reduces the proof to (1.1.2.).
Otherwise, $C(v_1) = T_{u v}\cup \{1, 2/5/6/7\}$, $C(u_1) = T_{u v}\cup \{1, 3/4\}$,
and if $4\not\in C(u_1)$,
then $H$ contains a $(3, 4)_{(u_1, u_3)}$-path.

One sees clearly that $2, 4\not \in B_1$.
If $H$ contains no $(3, i)_{(u_2, u_3)}$-path for some $i\in (T_{u v}\cup [5, 7])\setminus C(u_2)$,
then $(u u_2, u v)\overset{\divideontimes}\to (i, 2)$.
Otherwise, $3\in C(u_2)$ and $H$ contains a $(3, i)_{(u_2, u_3)}$-path for each $i\in (T_{u v}\cup [5, 7])\setminus C(u_2)$.
It follows that $(T_{u v}\cup [5, 7])\subseteq C(u_2)\cup C(u_3)$ and $T_{u v}\setminus C(u_2) \ne \emptyset$.

Let $8\not\in C(u_2)$.
If $H$ contains no $(8, i)_{(u_2, v_i)}$-path for some $i\in [5, 7]$,
then $(v u_2, u v)\overset{\divideontimes}\to (8, 4)$.
Otherwise, assume w.l.o.g. $5\in C(u_2)$ and $H$ contains a $(5, 8)_{(u_2, v_5)}$-path.

If $H$ contains no $(3, 5)_{(u_1, u_3)}$-path,
$(u u_1, u v)\overset{\divideontimes}\to (5, 8)$.
Otherwise, $C(u_1) = \{1, 3\}\cup T_{u v}$,
and $H$ contains a $(3, i)_{(u_1, u_3)}$-path for each $i\in [4, 5]$.
It follows that $u_1\not\in N(v)$, $u_3\not\in \{v_1, v_4, v_5\}$.

If $u_3 = v_i$ for some $i\in [6, 7]$, since $3, 4, 5, i\in C(u_3) = C(v_i)$,
one sees that $H$ contains no $(3, i)_{(u_1, u_3)}$-path,
then $u u_1\to i$, $u v\overset{\divideontimes}\to T_{uv}\setminus C(u_3) (T_{u v}\setminus C(v_i))$.
Otherwise, $u_3\not\in N(v)$.
Hence, $u v$ is contained in no $3$-cycles other that $u v u_2$ and $\sum_{x\in N(u)}d(x)\le 2\Delta + 13$.

Since $d(u_1) = \Delta$, $2, 3, 4, 5\in C(u_2)$, $3, 4, 5\in C(u_3)$,
and $T_{u v}\cup [6, 7]\subseteq C(u_2)\cup C(u_3)$,
$\sum_{x\in N(u)}d(x)\ge \Delta + 4 + 3 + \Delta - 2 + 2 + 6 = 2\Delta + 13$.

It follows that $\sum_{x\in N(u)}d(x) = 2\Delta + 13$ and $1, 2\not\in C(u_3)$.
Then $(u u_1, u u_3, u v)\overset{\divideontimes}\to (5, 1, 8)$.

{Case 1.2.} $v$ is $(7, 4^-)$, $u_1 u_2$, $u_1 u_3$, $v u_2$,
$v u_3\in E(G)$ and $\min \{d(u_1), d(u_2), d(u_3)\}\le 8$.
Then $t_{u v} = \Delta - 3\ge 4$, $V_{u v} = [4, 8]$,
and assume w.l.o.g. $d(v_7)\le 5$, $d(v_8)\le 5$, $d(v_6)\ge 6$.

Assume $d(v_{i_1})\le 5$, say $i_1 = 1$.
Since $T_{u v} = [9, 12] \subseteq C(u_1)\cap C(v_1)$,
it follows that $\Delta = 7$ and $C(v_1) = \{1\}\cup T_{u v}$.
If $2, 3\not\in C(u_1)$, then $u u_1\to [4, 8]\setminus C(u_1)$, $(v v_1, u v)\to (2, 1)$.
Otherwise, $2/3\in C(u_1)$.
One sees that the same argument applies if $v v_1\to [2, 3]$,
and thus $1/3\in C(u_2)$, $1/2\in C(u_3)$.
Hence, for each $i\in [1, 3]$, $|[4, 8]\cap C(u_i)|\le 1$.
Then $u u_1\to [4, 8]\setminus S_u$, $u v\overset{\divideontimes}\to 2$.
In the sequel we continue the proof with the following Proposition, or otherwise we are done:
\begin{proposition}
\label{auv176+}
$d(v_{i_1})\ge 6$,
and if re-coloring some edges $E'\subseteq E(H)$ incident to $u$ or $v$ does not produce any new bichromatic cycles,
but $C(u)\cap C(v) = \{c(v v_{j})\} = \{\alpha\}$ with $d(v_j)\le 5$,
then an acyclic edge $(\Delta + 2)$-coloring for $G$ can be obtained in $O(1)$ time.
\end{proposition}

(1.2.1.) $i_1\in \{c(v u_2), c(v u_3)\}$ and $i_1\in \{2, 3\}$, say $(v u_2, v u_3)_c = (3, 4)$.
It follows from \ref{auv1} that $T_{u v}\subseteq C(u_2)\cap C(u_3)$,
$|[5, 8]\cap C(u_2)|\le 1$, $|[5, 8]\cap C(u_3)|\le 1$.
If $H$ contains no $(1, i)_{(u_2, u_1)}$-path for some $[5, 8]\setminus (C(u_2)\cup C(u_3))$,
then $u u_2\to [5, 8]\setminus (C(u_2)\cup C(u_3))$ and $uv\overset{\divideontimes}\to 2$.
Otherwise, $C(u_2) = T_{u v}\cup [1, 3]$,
and $H$ contains a $(3, i)_{(u_2, u_1)}$-path for each $i\in [5, 8]\setminus C(u_3)$.
One sees that $[7, 8]\setminus C(u_3)\ne \emptyset$.
Then $u u_3\to [7, 8]\setminus C(u_3)$,
and by \ref{auv176+}, we are done.

(1.2.2.) $i_1\in \{c(v u_2), c(v u_3)\}$ and $i_1 = 1$, say $(v u_2, v u_3)_c = (1, 4)$.
It follows from \ref{auv1} that $T_{u v}\subseteq C(u_1)\cap C(u_2)$,
$|[5, 8]\cap C(u_1)|\le 2$, $|[5, 8]\cap C(u_2)|\le 1$.
When $3\not\in C(u_2)$, then $[5, 8]\setminus (C(u_1)\cup C(u_2))$ and $u v\overset{\divideontimes}\to 2$.
In the other case, $C(u_2) = T_{u v}\cup [1, 3]$.
If $c(u_1 u_2) = 3$, then $[5, 8]\setminus C(u_1)$ and $uv\overset{\divideontimes}\to 2$.
Otherwise, $(u u_2, v u_2)\to (4, 2)$,
and by \ref{auv176+}, we are done.

(1.2.3.) $i_1\not\in \{c(v u_2), c(v u_3)\}$ and $i_1\in \{2, 3\}$,
say $i_1 = 2$, $(v u_2, v u_3)_c = (5, 4)$.
It follows from \ref{auv1} that $T_{u v}\subseteq C(u_2)\cap C(v_2)$,
$|[5, 8]\cap C(u_2)|\le 1$, $|[5, 8]\cap C(v_2)|\le 2$.

If $4\not\in C(u_2)$ and $H$ contains no $(4, 1)_{(u_2, u_1)}$-path,
then $u u_2\to 4$ reduces the proof to (1.2.1.).
Otherwise, $4\in C(u_2)$ or $H$ contain a $(4, 1)_{(u_1, u_2)}$-path.
When $4\in C(u_2)$, i.e., $C(u_2) = \{2, 5, 4\}\cup T_{u v}$,
$(u u_2, v u_2)\to (5, 3)$ reduces the proof to (1.2.1.).
In the other case, $C(u_2) = \{2, 5, 1\}\cup T_{uv}$ and $H$ contains a $(4, 1)_{(u_1, u_2)}$-path.
Assume that $4\in C(u_3)\setminus \{c(u_1 u_3)\}$.

If $H$ contains no $(2, 3)_{(u_3, v)}$-path,
then $(v u_2, u v)\overset{\divideontimes}\to (3, 5)$.
If $H$ contains no $(1, 5)_{(u_1, u_3)}$-path,
then $(v u_2, u u_2)\to (3, 5)$ reduces the proof to (1.2.1.).
Otherwise,  $H$ contains a $(2, 3)_{(u_3, v)}$-path and a $(1, 5)_{(u_1, u_3)}$-path.
It follows that $3\in C(v_2)$, $2\in C(u_3)$, and $5\in C(u_1)$.

If $1\in C(v_2)$, i.e., $C(v_2) = [1, 3]\cup T_{u v}$,
then $(v v_2, v u_2)\to (5, 3)$ reduces the proof to (1.2.1.).
Otherwise, $1\not\in C(v_2)$.
One sees clearly that $1, 3\not\in B_2$.

If $H$ contains no $(1, i)_{(u_1, u_2)}$-path for some $i\in [7, 8]$,
then $u u_2\to i$ and by \ref{auv176+}, we are done.
Otherwise, $H$ contains a $(1, i)_{(u_1, u_2)}$-path for each $i\in [7, 8]$.
If there exits an $i\in \{5, 7, 8\}\setminus C(u_3)$,
then $(u u_3, u v)\overset{\divideontimes}\to (i, 3)$.
Otherwise, $\{5, 7, 8\}\subseteq C(u_1)\cap C(u_3)$.

If $H$ contains no $(1, i)_{(u_1, u_3)}$-path for some $i\in (\{6\}\cup T_{u v})\setminus C(u_3)$,
then $(u u_3, u v)\overset{\divideontimes}\to (j, 3)$.
If $H$ contains no $(3, j)_{(u_1, u_3)}$-path for some $j\in (\{6\}\cup T_{u v})\setminus C(u_1)$,
then $(u u_1, u v)\overset{\divideontimes}\to (j, 1)$.
Otherwise, $H$ contains a $(1, i)_{(u_1, u_3)}$-path for each $i\in (\{6\}\cup T_{u v})\setminus C(u_3)$,
and a $(3, j)_{(u_1, u_3)}$-path for each $j\in (\{6\}\cup T_{u v})\setminus C(u_1)$.
It follows that $3\in C(u_1)$, $1\in C(u_3)$ and $\{6\}\cup T_{u v}\subseteq C(u_1)\cup C(u_3)$.

Together with $\{1\}\cup [3, 5]\cup [7, 8]\subseteq C(u_1)$ and $[1, 5]\cup [7, 8]\subseteq C(u_3)$,
we have $d(u_1) + d(u_3)\ge t_{uv} + 2\times 2 + 1 + 4 + 5 = \Delta + 11> \Delta + 8$, a contradiction.

(1.2.4.) $i_1\not\in \{c(v u_2), c(v u_3)\}$ and $i_1 = 1$, say $(v u_2, v u_3)_c = (5, 4)$.
It follows from \ref{auv1} that $T_{u v}\subseteq C(u_1)\cap C(v_1)$,
$|[5, 8]\cap C(u_1)|\le 2$, $|[5, 8]\cap C(v_1)|\le 2$.

If $5\not\in C(u_1)$ and $H$ contains a $(3, 5)_{(u_1, u_3)}$-path,
then $u u_1\to 5$ reduces the proof to (1.2.2.).
If $4\not\in C(u_1)$ and $H$ contains a $(2, 4)_{(u_1, u_2)}$-path,
then $u u_1\to 4$ reduces the proof to (1.2.2.).
Otherwise, $5\in C(u_1)$ or $H$ contains a $(3, 5)_{(u_1, u_3)}$-path,
$4\in C(u_1)$ or $H$ contains a $(2, 4)_{(u_1, u_2)}$-path.
It follows that $3/5, 2/4\in C(u_1)$ and $C(u_1) = \{1, 3/5, 2/4\}\cup T_{u v}$.

If for some $i\in [7, 8]$, $H$ contains no $(i, j)_{(u_1, u)}$-path for each $j\in [2, 3]$,
then $u u_1\to i$, and by \ref{auv176+}, we are done.
Otherwise, for each $i\in [7, 8]$, $H$ contains a $(i, j)_{(u_1, u)}$-path for some $j\in [2, 3]$,
and $7, 8\in C(u_2)\cup C(u_3)$.
Assuming w.l.o.g. $2\in C(u_1)$, $H$ contains a $(2, 7)_{(u_1, u)}$-path,
and $2\in C(u_1)\setminus \{c(u_1 u_3)\}$.
Since $c(u_1u_2)\not\in \{3, 5\}$, assume w.l.o.g. $c(u_1 u_2) = 9$ and $1\in C(u_2)$.

If $6\not\in C(u_2)\cup C(u_3)$, then $(u u_1, u v)\overset{\divideontimes}\to (6, 9)$;
if $7\not\in C(u_3)$ and $3\not\in B_1$,
then $(u u_3, u v)\overset{\divideontimes}\to (7, 3)$.
Otherwise, $6\in C(u_2)\cup C(u_3)$, $7\in C(u_3)$ or $3\in B_1$.
\begin{itemize}
\parskip=0pt
\item
    $3\in C(u_1)$, i.e., $C(u_1) = [1, 3]\cup T_{u v}$ and assume w.l.o.g. $c(u_1 u_3) = 10$.
    Then $1, 5\in C(u_3)$.

  \quad  If $2, 9\not\in C(u_3)$, then $(u u_1, u u_3)\to (5, 9)$;
    if $3, 10\not\in C(u_2)$, then $(u u_1, u u_2)\to (4, 10)$,
    this reduces the proof to (1.2.2.).
    Otherwise, $2/9\in C(u_3)$, $3/10\in C(u_2)$.
    One sees that $1$, $4$, $2$, $5$, $9$, $3/10\in C(u_2)$,
    $1$, $5$, $3$, $4$, $10$, $2/9\in C(u_3)$, $6, 7, 8\in C(u_2)\cup C(u_3)$,
    and $d(u_2) + d(u_3)\le \Delta + 8$.
    Since $6 + 6 + 3 + |[11, \Delta + 5]| = \Delta + 10$,
    there exists a $j\in T_{u v}\setminus (C(u_2)\cup C(u_3))$,
    say $j = 11$.

    \quad If $7\not\in C(u_3)$, then $(u u_3, u u_1, u v)\overset{\divideontimes}\to (7, 5, 11)$.
    Otherwise, $7\in C(u_2)\cap C(u_3)$.
    Likewise, $8\in C(u_2)\cap C(u_3)$.
    If $3\not\in B_1$, then $(u u_3, u v)\overset{\divideontimes}\to (11, 3)$,
    if $2\not\in B_1$, then $(u u_2, u v)\overset{\divideontimes}\to (11, 2)$,
    otherwise, $2, 3\in B_1$.
    It follows that $C(v_1) = [1, 3]\cup T_{uv}$.

    \quad When $H$ contains neither a $(2, 6)_{(u_2, u_1)}$-path nor a $(3, 6)_{(u_3, u_1)}$-path,
    first, $u u_1\to 6$.
    By the same argument, we have $C(v_6) = \{6, 2, 3\}\cup T_{u v}$ and $11\in B_6$.
    Then $(v v_1, v v_6, u v)\overset{\divideontimes}\to (6, 1, 11)$.
    In the other case, assume w.l.o.g. $6\in C(u_2)$ and $H$ contains a $(2, 6)_{(u_2, u_1)}$-path.

    \quad If $6\not\in C(u_3)$, then $(u u_3, u u_1, u v)\overset{\divideontimes}\to (6, 5, 11)$.
    Otherwise, $6\in C(u_3)\cap C(u_2)$.
    It follows that $d(u_2)\ge 9$ and $d(u_3)\ge 9$ and $d(u_1) = \Delta \ge 9$, a contradiction.
\item
    $3\not\in C(u_1)$, i.e., $C(u_1) = \{1, 2, 5\}\cup T_{u v}$ and assume w.l.o.g. $c(u_1 u_3)\in \{5, 10\}$.
    It follows that $7, 8\in C(u_2)\cap C(u_3)$.
    If $6\not\in C(u_2)$, then $(u u_1, u v)\overset{\divideontimes}\to (6, 9)$;
    if $2\not\in C(u_3)$, then $T_{u v}\setminus C(u_3)$ and $uv\overset{\divideontimes}\to 3$.
    Otherwise, $6\in C(u_2)$ and $2\in C(u_3)$.
    One sees that $[1, 2]$ $\cup$ $[4, 9]\subseteq C(u_2)$,
    $[2, 4]$ $\cup$ $\{7, 8, 5/10\}\subseteq C(u_3)$, and $d(u_2) + d(u_3)\le \Delta + 8$.
    Since $8 + 6 + |[11, \Delta + 5]| = \Delta + 9$,
    there exists a $j\in T_{uv}\setminus (C(u_2)\cup C(u_3))$,
    say $j = 11$.
    Then $(u u_3, u v)\overset{\divideontimes}\to (11, 3)$.
\end{itemize}

{\bf Case 2.} $d(u) = 5$ and $c(u u_i) = i$, $i\in [1, 4]$, $c(v v_j) = j$,
$j\in \{i_1\}\cup [5, d(v) + 2]$, where $i_1\in [1, 4]$.

In this case, $5\le d(v)\le 6$,
and $t_{u v} = \Delta - (d(u) + d(v)) + 8 = \Delta + 3 - d(v)$.

{Case 2.1.} $d(v) = 5$.
Then $t_{u v} = \Delta - 2$, $V_{u v} = [5, 7]$.

By symmetry, we consider the following three cases.

(2.1.1.) $i_1\in \{c(u u_3), c(u u_4)\}$.
Assuming w.l.o.g. $i_1 = 3$ and $c(v u_3) = 5$.
It follows from \ref{auv1} that $C(u_3) = \{3, 5\}\cup T_{uv}$,
$T_{u v}\subseteq C(v_3)$.
One sees clearly from $t_{u v} = \Delta - 2$ that there exists an $i\in \{1, 2\}\setminus C(v_3)$.
Then $v u_3\to i$ and $u v\overset{\divideontimes}\to 5$.

(2.1.2.) $i_1\in C(u)\setminus \{c(u u_3), c(u u_4)\}$ and $i_1\in \{c(v u_3), c(v u_4)\}$.
Assuming w.l.o.g. $(v u_3, v u_4)_c = (1, 6)$.
It follows from \ref{auv1} that $C(u_3) = \{1, 3\}\cup T_{u v}$,
$T_{u v}\subseteq C(u_1)$.
One sees clearly from $t_{u v} = \Delta - 2$ that there exists an $i\in [5, 7]\setminus C(u_1)$.
Then $u u_3\to i$ and $u v\overset{\divideontimes}\to 3$.

(2.1.3.) $i_1\in C(u)\setminus \{c(u u_3), c(u u_4)\}$ and $i_1\in C(v)\setminus \{c(v u_3), c(v u_4)\}$.
Assuming w.l.o.g. $i_1 = 1$ and $(v u_3, v u_4)_c = (5, 6)$.
It follows from \ref{auv1} that $T_{u v}\subseteq C(u_1)\cap C(v_1)$,
and then $w_1\le 2$, $|C_{uv}\cap C(v_1)|\le 2$.
If $2, 4, 5\not\in C(u_1)$, then $u u_1\to 5$,
if $2, 3, 6\not\in C(u_1)$, then $u u_1\to 6$,
this reduces the proof to (2.1.2.).
If $3, 6, 7\not\in C(v_1)$, then $v v_1\to 3$,
if $4, 5, 7\not\in C(v_1)$, then $v v_1\to 4$,
this reduces the proof to (2.1.1.).
Otherwise,  $2/4/5, 2/3/6\in C(v_1)$, $3/6/7, 4/5/7\in C(v_1)$.
Hence, $2\in C(u_1)$, $7\in C(v_1)$,
and $C(u_1) = [1, 2]\cup T_{u v}$, $C(v_1) = \{1, 7\}\cup T_{u v}$.

Further, assume that $H$ contains a $(2, i)_{(u_1, u_2)}$-path for each $i\in [5, 6]$,
a $(7, i)_{(v_1, v_2)}$-path for each $v\in [3, 4]$ and $5, 6\in C(u_2)$, $3, 4\in C(v_7)$.
One sees clearly that there exists a $\alpha\in T_2$, $\beta\in R_7$.
If $H$ contains no $(2, 7)_{(u_1, u_2)}$-path, then $(u u_1, u v)\to (7, \beta)$,
if $H$ contains no $(2, 7)_{(v_1, v_7)}$-path, then $(v v_1, u v)\to (2, \alpha)$.
Otherwise, $H$ contains a $(2, 7)_{(u_1, u_2)}$-path and a $(2, 7)_{(v_1, v_7)}$-path.

If $1\not\in S_u$, then $(u u_2, u u_1)\to (1, 5)$ and $u v\to T_{u v}\setminus C(u_3)$,
if $3\not\in C(u_4)$, then $(v u_4, u v)\overset{\divideontimes}\to (3, 6)$,
if $5\not\in C(u_4)$, then $(u u_4, u v)\overset{\divideontimes}\to (5, 4)$.
Otherwise, $1\in S_u$, $3, 5\in C(u_4)$ and likewise $4, 6\in C(u_3)$.

If for some $i\in T_3$, $H$ contains no $(i, j)_{(u_3, u_j)}$-path for each $j\in \{2, 4\}$,
then $(u u_3, u v)\overset{\divideontimes}\to (i, 3)$.
Otherwise, for each $i\in T_3$, $H$ contains a $(i, j)_{(u_3, u_j)}$-path for some $j\in \{2, 4\}$.

If there exists an $i\in T_3\setminus C(u_4)$, since $H$ contains a $(2, i)_{(u_2, u_3)}$-path,
then $(u u_4, u v)\overset{\divideontimes}\to (i, 4)$.
If there exists an $i\in T_3\setminus C(u_2)$, since $H$ contains a $(4, i)_{(u_4, u_3)}$-path,
then $(u u_2, u v)\overset{\divideontimes}\to (i, 2)$.
Otherwise, $T_3\subseteq C(u_2)\cap C(u_4)$.
Likewise, $T_4\subseteq C(u_2)\cap C(u_3)$.
It follows that $T_{u v}\subseteq C(u_2)\cup C(u_3)\cup C(u_4)$,
and from $t_3\ge 2$, $t_4\ge 2$ that $|T_0|\ge 4$.

Since $W_1 = \{1, 2\}$, $\{2\}\cup [5, 7]\subseteq W_2$, $[3, 6]\cup C(u_3)\cap C(u_4)$,
and $1\in S_u$, $T_{uv}\subseteq C(u_1)$,
we have $\sum_{x\in N(u)\setminus \{v\}}d(x) = \sum_{i\in [1, 4]}d(u_i) = \sum_{i\in [1, 4]}w_i + \sum_{i\in T_{uv}}\mbox{mult}_{S_u}(i)
	\ge 2 + 4 + 4 + 4 + 1 + 2t_{uv} + |T_0| \ge 15 + 2(\Delta - 2) + |T_0|\ge 2\Delta + 15 > \max\{2\Delta + 13, \Delta + 20\}$,
and thus none of ($A_{7.1}$)--($A_{7.3}$) holds.

{Case 2.2.} $d(v) = 6$.
Then $t_{u v} = \Delta - 3$, $V_{u v} = [5, 8]$.

By symmetry, we consider the following three cases.

(2.2.1.) $i_1\in \{c(u u_3), c(u u_4)\}$ and $i_1\in \{c(v u_3), c(v u_4)\}$.
Assuming w.l.o.g. $(v u_3, v u_4)_c = (5, 3)$.
It follows from \ref{auv1} that $T_{u v}\subseteq C(u_3)\cap C(u_4)$.
One sees clearly from $t_{u v} = \Delta - 3$ that $w_3\le 3$, $w_4\le 3$.
When $[6, 8]\cap C(u_3) = \emptyset$ and there exists an $i\in [1, 2]\setminus (W_3\cup W_4)$,
then $(v u_3, u v)\overset{\divideontimes}\to (i, 5)$.
In the other case, $[6, 8]\cap C(u_3) \ne \emptyset$ or $1, 2\in W_3\cup W_4$.
We can have the following two scenarios:
\begin{itemize}
\parskip=0pt
\item
     $[6, 8]\cap C(u_3) = \emptyset$.
     It follows that $1\in W_1$, $2\in W_2$, i.e., $C(u_3) = \{1, 3, 5\}\cup T_{u v}$, $C(u_4) = [2, 4]\cup T_{u v}$.
     If there exists an $i\in [6, 8]\setminus C(u_2)$,
     then $(v u_4, u v)\overset{\divideontimes}\to (i, 4)$.
     Otherwise, $[6, 8]\subseteq C(u_2)$ and there exists a $j\in T_{u v}\setminus C(u_2)$.
     Then $(v u_4, v u_3, u v)\overset{\divideontimes}\to (5, 2, j)$.
\item
     $6\in W_1$, i.e., $C(u_3) = \{3, 5, 6\}\cup T_{uv}$.
     If $1, 2\not\in C(u_4)$,
     then $vu_4\to \{7, 8\}\setminus C(u_4)$ and $u v\overset{\divideontimes}\to 4$.
     Otherwise, assume w.l.o.g. $1\in C(u_4)$.
     Then $(v u_4, u u_3, u v)\overset{\divideontimes}\to (2, 7, 3)$.
\end{itemize}

(2.2.2.) $i_1\in \{c(u u_3), c(u u_4)\}$ and $i_1\in C(v)\setminus \{c(v u_3), c(v u_4)\}$.
Assuming w.l.o.g. $(v u_3, v u_4)_c = (5, 6)$.
It follows from \ref{auv1} that $T_{u v}\subseteq C(u_3)\cap C(v_3)$,
and then $w_3\le 3$, $|C(v_3)\cap C_{u v}|\le 3$.
When $1, 2, 6\not\in W_3$, 
$u u_3\to 6$ reduces the proof to (2.2.1.)
In the other case, $6/1/2\in W_3$. 

If $6\in W_3$, i.e., $C(u_3) = \{3, 5, 6\}\cup T_{u v}$,
then $(v u_3, u u_3)\to (4, 5)$ reduces the proof to (2.2.1.).
Otherwise, assume w.l.o.g. $C(u_3) = \{3, 5, 1\}\cup T_{u v}$.
If there exists an $i\in \{2, 4\}\setminus C(v_3)$,
then $(v u_3, u v)\overset{\divideontimes}\to (i, 5)$.
Otherwise, $C(v_3) = [2, 4]\cup T_{u v}$,
$(v v_3, v u_3)\to (5, 4)$ reduces the proof to (2.2.1.).

(2.2.3.) $i_1\in C(u)\setminus \{c(u u_3), c(u u_4)\}$ and $i_1\in \{c(v u_3), c(v u_4)\}$.
Assuming w.l.o.g. $(v u_3, v u_4)_c = (1, 6)$.
It follows from \ref{auv1} that $T_{u v}\subseteq C(u_1)\cap C(u_3)$,
and then $w_3\le 3$.
When $4, 5, 7, 8\not\in W_3$, 
$v u_3\to 4$ reduces the proof to (2.2.1.)
In the other case, $4/5/7/8\in W_3$. 

If $4\in W_3$, i.e., $C(u_3) = \{3, 1, 4\}\cup T_{u v}$,
then $(v u_3, u u_3)\to (3, 6)$ reduces the proof to (2.2.1.).
Otherwise, assume w.l.o.g. $C(u_3) = \{3, 1, 7\}\cup T_{u v}$.
If there exists an $i\in \{5, 8\}\setminus C(u_1)$,
then $(u u_3, u v)\overset{\divideontimes}\to (i, 3)$.
Otherwise, $C(u_1) = \{1, 5, 8\}\cup T_{u v}$,
$(u u_1, u u_3, u v)\overset{\divideontimes}\to (3, 8, 9)$.

(2.2.4.) $i_1\in C(u)\setminus \{c(u u_3), c(u u_4)\}$ and $i_1\in C(v)\setminus \{c(v u_3), c(v u_4)\}$.
Assuming w.l.o.g. $i_1 = 1$ and $(v u_3, v u_4)_c = (5, 6)$.
Here let $\{j_1, j_2, j_3\} = [2, 4]$.
It follows from \ref{auv1} that $T_{u v}\subseteq C(u_1)\cap C(v_1)$,
and then $w_1\le 3$, $|C_{u v}\cap C(v_1)|\le 3$.
Note that if $d(u_3) = d(u_4) = 7$, then $\min\{d(u_1), d(u_2)\}\le 6$;
if $d(u_1) = d(u_2) = \Delta$, then $\min\{d(u_3), d(u_4)\}\le 6$.

If $5\not\in C(u_1)$ and $H$ contains no $(5, i)_{(u_1, u_i)}$-path for each $i\in \{2, 4\}$,
then $u u_1\to 5$ and reduces the proof to (2.2.3.).
Otherwise, $5\in C(u_1)$ or $H$ contains a $(5, i)_{(u_1, u_i)}$-path for some $i\in \{2, 4\}$.
If $5\not\in C(u_1)\cup C(u_4)$ and $4\not\in B_1$,
since $H$ contains a $(5, 2)_{(u_1, u_2)}$-path,
then $(u u_4, u v)\overset{\divideontimes}\to (5, 4)$.
Otherwise, $5\in C(u_1)\cup C(u_4)$ or $4\in B_1$.
Likewise, we proceed with the following proposition:
\begin{proposition}
\label{256u3u4}
\begin{itemize}
\parskip=0pt
\item[{\rm (1)}]
	$5\in C(u_1)$ or $H$ contains a $(5, i)_{(u_1, u_i)}$-path for some $i\in \{2, 4\}$;
\item[{\rm (2)}]
	$5\in C(u_1)\cup C(u_4)$ or $4\in B_1\subseteq C(u_1)\cap C(v_1)$.
\item[{\rm (3)}]
	$6\in C(u_1)$ or $H$ contains a $(6, i)_{(u_1, u_i)}$-path for some $i\in \{2, 3\}$;
\item[{\rm (4)}]
	$6\in C(u_1)\cup C(u_3)$ or $3\in B_1\subseteq C(u_1)\cap C(v_1)$.
\item[{\rm (5)}]
	$5/2/4, 6/2/3\in C(u_1)$.
\end{itemize}
\end{proposition}

If $2\not\in B_1$, and for some $j\in T_2\cup ([5, 8]\setminus C(u_2))$,
$H$ contains no $(i, j)_{(u_2, u_i)}$-path for each $i\in \{1, 3, 4\}$,
then $(u u_2, u v)\overset{\divideontimes}\to (j, 2)$.
If $3\not\in B_1$, and for some $j\in T_3\cup ([6, 8]\setminus C(u_3))$,
$H$ contains no $(i, j)_{(u_3, u_i)}$-path for each $i\in \{1, 2, 4\}$,
then $(u u_3, u v)\overset{\divideontimes}\to (j, 3)$.
If $4\not\in B_1$, and for some $j\in T_4\cup (\{5, 7, 8\}\setminus C(u_4))$,
$H$ contains no $(i, j)_{(u_4, u_i)}$-path for each $i\in \{1, 2, 3\}$,
then $(u u_4, u v)\overset{\divideontimes}\to (j, 4)$.
Otherwise, we proceed with the following proposition:
\begin{proposition}
\label{a2a3a4}
\begin{itemize}
\parskip=0pt
\item[{\rm (1)}]
	$2\in B_1$; or $5/1/4, 6/1/3\in C(u_2)$,
    and for each $j\in T_2\cup ([5, 8]\setminus C(u_2)) = (T_3\cup [5, 8])\setminus C(u_3)$,
    $H$ contains a $(a_2, j)_{(u_2, u_{a_2})}$-path for some $a_2\in \{1, 3, 4\}\cap C(u_2)$;
\item[{\rm (2)}]
	$3\in B_1$; or $6/1/2\in C(u_3)$,
    and for each $j\in T_3\cup ([6, 8]\setminus C(u_3)) = (T_3\cup [6, 8])\setminus C(u_3)$,
    $H$ contains a $(a_3, j)_{(u_3, u_{a_3})}$-path for some $a_3\in \{1, 2, 4\}\cap C(u_3)$;
\item[{\rm (3)}]
	$4\in B_1$; or $5/1/2\in C(u_4)$,
    and for each $j\in T_4\cup (\{5, 7, 8\}\setminus C(u_4)) = (T_4\cup \{5, 7, 8\})\setminus C(u_4)$,
    $H$ contains a $(a_4, j)_{(u_4, u_{a_4})}$-path for some $a_4\in \{1, 2, 3\}\cap C(u_4)$.
\end{itemize}
\end{proposition}

If for some $j_1, j_2\in [2, 4]\setminus B_1$,
there exists a $l\in [(T_{j_1}\cup [5, 8])\setminus (C(u_{j_1})\cup C(u_1))]\setminus C(u_{j_2})$,
since $H$ contains a $(j_3, l)_{(u_{j_1}), u_{j_3}}$-path by \ref{a2a3a4},
then $(uu_{j_2}, uv)\overset{\divideontimes}\to (j_2, l)$.
Otherwise, we proceed with the following proposition:
\begin{proposition}
\label{T2T3T4}
For some $i, j\in [2, 4]\setminus B_1$, $[(T_i\cup [5, 8])\setminus (C(u_i)\cup C(u_1))]\subseteq C(u_j)$.
\end{proposition}

If for some $i\in [2, 4] \setminus C(u_1)$,
$1\not\in C(u_i)$, $H$ contains neither a $(1, j)_{(u_i, u_j)}$-path nor a $(i, j)_{(u_1, u_j)}$-path, for each $j\in [2, 4]\setminus \{i\}$,
then $(uu_i, uu_1, uv)\overset{\divideontimes}\to (1, i, 9)$.
For some $j_1, j_2\in [2, 4]\setminus C(u_1)$,
if $1\not\in C(j_1)\cup C(j_2)$ and $H$ contains no $(j_3, j)_{(u_{j_3}, u_1)}$-path for each $j\in \{j_1, j_2\}$,
then $(uu_{j_1}, uu_1, uv)\to (1, j_1, 9)$ or $(uu_{j_2}, uu_1, uv)\to (1, j_2 9)$, and we are done.
Otherwise, we proceed with the following proposition:
\begin{proposition}
\label{1a2a3a4}
\begin{itemize}
\parskip=0pt
\item[{\rm (1)}]
	For some $i\in [2, 4] \setminus C(u_1)$, $1\in C(u_i)$,
    or for some $j\in [2, 4]\setminus \{i\}$,
    $H$ contains a $(1, j)_{(u_i, u_j)}$-path or a $(i, j)_{(u_1, u_j)}$-path.
\item[{\rm (2)}]
	For some $j_1, j_2\in [2, 4]\setminus C(u_1)$,
    $1\in C(j_1)\cup C(j_2)$ or $H$ contains a $(j_3, j)_{(u_{j_3}, u_1)}$-path for some $j\in \{j_1, j_2\}$.
\end{itemize}
\end{proposition}

When $w_1 = 2$,
it follows from \ref{256u3u4} that $C(u_1) = [1, 2]\cup T_{u v}$,
and $5\in C(u_2)\cap C(u_4)$, $6\in C(u_2)\cap C(u_3)$.
One sees clearly that $3, 4\not\in B_1$,
and for each $i\in [3, 4]$, $(T_i\cup [7, 8])\setminus C(u_i)\ne \emptyset$.
It follows from \ref{a2a3a4} that for each $i\in [3, 4]$, there exists a $a_i\in C(u_)$,
and from \ref{T2T3T4} that $((T_3\cup [7, 8])\setminus C(u_3))\subseteq C(u_4)$.
Hence, $T_{u v}\cup [7, 8]\subseteq C(u_3)\cup C(u_4)$.
One sees that $|T_{u v}\cup [7, 8]| = \Delta - 1$, $\{3, 5, 6, a_3\}\subseteq C(u_3)$ and $\{4, 6, 5, a_4\}\subseteq C(u_4)$.
It follows from $d(u_3) + d(u_4)\ge \Delta - 1 + 8 = \Delta + 7$,
$d(u_3)\le 7$ and $d(u_4)\le 7$ that $d(u_3) = d(u_4) = \Delta = 7$.

Assuming w.l.o.g. $C(u_3) = \{3, 5, 6, a_3\}\cup \{\alpha_1, \alpha_2, \alpha_3\}$,
$C(u_4) = \{4, 6, 5, a_4\}\cup \{\beta_1, \beta_2, \beta_3\}$,
where $a_3\in \{2, 4\}$, $a_4\in \{2, 3\}$ and $\{\alpha_1, \alpha_2, \alpha_3, \beta_1, \beta_2, \beta_3\} = T_{u v}\cup [7, 8]$.
One sees clearly that $T_3\ne \emptyset$, $T_4\ne \emptyset$.
Since $|T_{u v}\cup [7, 8]| = 6$, assume w.l.o.g. $a_4 = 3$,
and $H$ contains a $(3, j)_{(u_4, u_3)}$-path for each $j\in \{\alpha_i\}$.
If $4\not\in C(u_2)$,
then $(u u_4, u u_1, u v)\overset{\divideontimes}\to (1, 4, 9)$.
If $1\not\in C(u_2)$,
then $(u u_2, u u_1)\to (1, 5)$ and $uv\overset{\divideontimes}\to T_3$.
Otherwise, $4, 1\in C(u_2)$.
Together with $2, 5, 6\in C(u_2)$ and $d(u_2)\le 7$, it follows that $a_3 = 4$.
If $3\not\in C(u_2)$,
then $(u u_3, u u_1, u v)\overset{\divideontimes}\to (1, 3, 9)$.
Otherwise, $3\in C(u_2)$.
Then $u u_2\to j\in T_2$ and $(u u_4, u v)\overset{\divideontimes}\to (2, 4)$.
One can clarify that $G$ contains no $(3, j)_{(u_2, u_3)}$-path even if $j\in C(u_3)$.

In the other case, $w_1 = 3$ and $d(u_1) = \Delta$.
One sees from $|T_{u v}\cup [7, 8]| = \Delta - 1$ that
for each $i\in [3, 4]$, $T_i\cup ([7, 8]\setminus C(u_i)) \ne \emptyset$,
and if $|[1, 6]\cap C(u_2)|\ge 2$, then $T_2\cup ([7, 8]\setminus C(u_2)) \ne \emptyset$.

Let $C(u_1) = \{1, x, y\}\cup T_{uv}$,
and
\[
S'_u = \biguplus_{i\in [2, 4]}(C(u_i)\setminus \{c(u u_i)\}).
\]

When $2, 3, 4\not\in B_1$,
it follows from \ref{a2a3a4}(1) that $5/1/4\in C(u_2)$,
$([7, 8]\cup T_{uv})\setminus C(u_2)\ne \emptyset$ and then $1/3/4\in C(u_2)$,
from \ref{a2a3a4}(2) that $1/2/4\in C(u_3)$,
from \ref{a2a3a4}(3) that there exists a $1/2/3\in C(u_4)$.
One sees from \ref{T2T3T4} that for each $i, j\in [2, 4]$,
$[(T_i\cup [5, 8])\setminus (C(u_i)\cup C(u_1))]\subseteq C(u_j)$.
If follows that for each $j\in T_{u v}\cup ([7, 8]\setminus C(u_1))$, mult$_{S'_u}(j)\ge 2$.
One sees that for each $j\in T_{u v}$,  mult$_{S_u}(j)\ge 3$,
and if $j\in [7, 8]\setminus C(u_1)$, then mult$_{S'_u}(j)\ge 2$ and mult$_{S_u}(j)\ge 2$.

If there exists an $j\in ([7, 8]\cap C(u_1))\setminus S'_u$,
since there exists an $i\in [2, 4]$ such that $H$ contains no $(1, j)_{u_1, u_i}$-path,
then $(u u_i, u v)\overset{\divideontimes}\to (j, i)$.
Otherwise, for some $j\in [7, 8]\cap C(u_1)$, $j\in S'(u)$, i.e., mult$_{S_u}(j)\ge 2$.
It follows from \ref{256u3u4} and \ref{a2a3a4} that if $5\not\in C(u_1)$, then $5\in C(u_2)\cap C(u_4)$,
and if $5\in C(u_1)$, then $5\in C(u_2)\cup C(u_4)$.
Hence, mult$_{S_u\setminus \{c(v u_3)\}}(5)\ge 2$ and likewise mult$_{S_u\setminus \{c(v u_3)\}}(6)\ge 2$.
It follows that $\sum_{i\in [1, 4]}d(u_i) \ge 3t_{uv} + 2|[5, 8]| + |\{1\}| + |\{2, 1/3/4\}| + |\{3, 5, 1/2/4\}| + |\{4, 6, 1/2/3\}|\ge 3(\Delta - 3) + 8 + 1 + 2 + 3 + 3 = 3\Delta + 8\ge 2\Delta + 6 + 8\ge 2\Delta + 14$.
Thus none of ($A_{8.1}$)--($A_{8.4}$) holds.

In the other case,
we proceed with the following proposition, or otherwise we are done:
\begin{proposition}
\label{234B1}
$\{2, 3, 4\}\cap B_1\ne \emptyset$, $2/3/4\in C(u_1)$ and if $C(u_1)\cap [2, 4] = \{a_1\}$, then $a_1\in B_1$.
\end{proposition}

Together with \ref{256u3u4}(5) and by symmetry, it suffices to consider the following five scenarios:
\begin{itemize}
\parskip=0pt
\item {\rm (i)}
    $\{x, y\} = \{3, 5\}$.
    One sees clearly that $2, 4\not\in B_1$, and from \ref{234B1} that $3\in B_1\cap C(v_1)$.
    If $2/6\in C(v_1)$, i.e., $C(v_1) = \{1, 3, 2/6\}\cup T_{u v}$,
    then $v v_1\to 4$ reduces the proof to (2.2.2.).
    Otherwise, $2, 6\not\in C(v_1)$.
    If $T_4 = \emptyset$, i.e., $C(u_4) = T_{u v}\cup \{2/3, 4, 6\}$,
    then $v u_4\to [2, 3]\setminus C(u_4)$ and $u v\overset{\divideontimes}\to 6$.
    Otherwise, $T_4\ne \emptyset$ and assume w.l.o.g. $9\in T_4$.

    \quad
    It follows from \ref{256u3u4}(3) that $6\in C(u_3)$, $H$ contains a $(3, 6)_{(u_1, u_3)}$-path,
    and from \ref{a2a3a4}(1) that $6\in C(u_2)$, there exists a $a_2\in \{3, 4\}\cap C(u_2)$,
    from \ref{a2a3a4}(3) that there exists a $a_4\in \{2, 3\}\cap C(u_4)$,
    and $5\in C(u_2)\cup C(u_4)$.
    Since from \ref{T2T3T4} that for each $i, j\in \{2, 4\}$, $T_i\cup ([7, 8]\cup C(u_i))\subseteq C(u_j)$,
    it follows that $T_{u v}\cup \{7, 8\}\subseteq C(u_2)\cup C(u_4)$.

    \quad If $1\not\in C(u_3)$ and $H$ contains no $(1, i)_{(u_3, u_i)}$-path for each $i\in \{2, 4\}$,
    then $(u u_3, u u_1, u v)\overset{\divideontimes}\to (1, 6, 9)$.
    If $7\not\in C(u_3)$ and $H$ contains no $(7, i)_{(u_3, u_i)}$-path for each $i\in \{2, 4\}$,
    then $(u u_1, u v)\to (7, 9)$ or $(u u_3, u u_1, u v)\to (7, 6, 9)$, and we are done.
    Otherwise, $1\in S_u$,
    and for each $i\in \{1, 7, 8\}$,
    $i\in C(u_3)$ or $H$ contains a $(i, j)_{(u_3, u_1)}$-path for some $j\in \{2, 4\}$.

    \quad One sees that if $7\not\in C(u_3)\cup C(u_4)$, then $H$ contains a $(2, 7)_{(u_2, u_3)}$-path and a $(2, 7)_{(u_2, u_4)}$-path, a contradiction.
    Hence, for some $i\in [7, 8]\setminus C(u_3)$, $i\in C(u_2)\cap C(u_4)$.

    \quad Together with \ref{1a2a3a4},
    it follows that if $1\not\in C(u_3)\cup C(u_4)$,
    then $H$ contains a $(1, 2)_{(u_3, u_2)}$-path and a $(3, 4)_{(u_1, u_3)}$-path, i.e., $2, 4\in C(u_3)$;
    if $1\not\in C(u_3)\cup C(u_2)$,
    then $H$ contains a $(1, 4)_{(u_3, u_4)}$-path and a $(3, 2)_{(u_1, u_3)}$-path, i.e., $2, 4\in C(u_3)$;
    and if $1\not\in C(u_2)\cup C(u_4)$, then $2, 4\in C(u_3)$,
    or for some $i\in [2, 4]\setminus C(u_3)$, $3\in C(u_i)$.

    \quad Let $\eta_{78} = |[7, 8]\setminus C(u_3)|$, $[7, 8]\cap C(u_3) = 2 - \eta_{78}$.
    Hence, $\sum_{i\in [7, 8]}$mult$_{S'_u}(i) = 2 + \eta_{78} + 2 - \eta_{78} = 4$,
    and $\sum_{i\in [1, 4]}d(u_i) = \Delta + \sum_{i\in [2, 4]}d(u_i)\ge \Delta + |\{2, 3/4, 6\}| + |\{3, 5, 2/4, 6\}| + |\{4, 6, 2/3\}| + |\{1\}| + \sum_{i\in T_{uv}\cup [5, 8]}$
	  mult$_{S'_u}(i) \ge \Delta + 11 + \sum_{i\in T_{uv}\cup \{5\}}$ mult$_{C(u_2)\cup C(u_4)}(i) + \sum_{i\in [7, 8]}$ mult$_{S'_u}(i) \ge \Delta + 11 + \Delta - 3 + 1 + 4 = 2\Delta + 13$.
    One sees clearly that if $3, 4\in C(u_2)$, or $2, 3\in C(u_4)$, or $2, 4\in C(u_3)$, or for some $j\in \{5\}\cup T_{uv}$, mult$_{S'_u}(j) \ge 2$,
    then $\sum_{i\in [1, 4]}d(u_i) \ge 2\Delta + 14$.

    \quad It follows that $\sum_{i\in [1, 4]}d(u_i) = 2\Delta + 13$,
    $\Delta = 7$, mult$_{S'_u}(1) = 1$,
    for each $j\in \{5\}\cup T_{uv}$, mult$_{S'_u}(j) = 1$,
    and $1\in C(u_3)\setminus (C(u_2)\cup C(u_3))$, $T_3 = T_{u v}$,
    $|\{3, 4\}\cap C(u_2)| = 1$,  $|\{2, 3\}\cap C(u_4)| = 1$,
    $|\{2, 4\}\cap C(u_3)| = 1$.
    Further, since $9\not\in C(u_3)$, $a_4 = 2$;
    since $5\cup T_2\ne\emptyset$, $a_3 = 4$.
    However, $3\not\in C(u_2)\cup C(u_4)$, a contradiction to \ref{1a2a3a4}(1).
\item {\rm (ii)}
    $\{x, y\} = \{3, 4\}$.
    It follows from \ref{256u3u4}(3) that $6\in C(u_3)$, $H$ contains a $(3, 6)_{(u_1, u_3)}$-path,
    and from \ref{a2a3a4}(1) that $6\in C(u_2)$, there exists a $a_2\in \{3, 4\}\cap C(u_2)$.
    It follows from \ref{256u3u4}(1) that $5\in C(u_4)$, $H$ contains a $(4, 5)_{(u_1, u_4)}$-path,
    and from \ref{a2a3a4}(1) that $5\in C(u_2)$.
    One sees from $2, 5, 6, 3/4\in C(u_2)$ that $T_2\ne\emptyset$.
    If $3, 4\in C(v_1)$, i.e., $C(v_1) = \{1, 3, 4\}\cup T_{uv}$,
    then $v v_1\to 2$ and $u v\overset{\divideontimes}\to T_2$.
    Otherwise, assume w.l.o.g. $4\not\in B_1$.
    It follows from \ref{a2a3a4}(3) that there exists a $a_4\in \{2, 3\}\cap C(u_4)$.
    Since from \ref{T2T3T4} that for each $i, j\in \{2, 4\}$, $T_i\cup ([7, 8]\cup C(u_i))\subseteq C(u_j)$,
    it follows that $T_{u v}\cup \{7, 8\}\subseteq C(u_2)\cup C(u_4)$.

    \quad Since $d(u_2) + d(u_4) \ge |\{2, 5, 6, 3/4\}| + |\{2/3, 4, 5, 6\}| + |T_{u v}\cup [7, 8]| = \Delta + 7$,
    $d(u_1) = \Delta$, and $d(u_4)\le 7$,
    it follows that $d(u_2) = d(u_4) = \Delta = 7$, $d(u_3)\le 6$,
    and $C(u_2) = \{2, 5, 6, 3/4, \beta_1, \beta_2, \beta_3\}$, $C(u_4) = \{2/3, 4, 5, 6, \alpha_1, \alpha_2, \alpha_3\}$, where $\{\alpha_i\}\cup \{\beta_j\} = [7, 8]\cup T_{uv}$.

    \quad One sees clearly that $T_4\ne \emptyset$ and assume w.l.o.g. $9\in T_4$.
    If $1\not\in C(u_3)$, then $(u u_3, u u_1, u v)\overset{\divideontimes}\to (1, 6, 9)$.
    Otherwise, $1\in C(u_3)$ and $T_3\ne \emptyset$.
    Since $\{\beta_1, \beta_2, \beta_3\}\setminus C(u_3)\ne \emptyset$,
    we have $a_4 = 2$ and $3\not\in C(u_3)$.
    Then $(u u_4, u u_1)\to (1, 5)$ and $uv\overset{\divideontimes}\to T_3$.
\item {\rm (iii)}
    $\{x, y\} = \{2, 7\}$.
    One sees that $|\{8\}\cup T_{u v}| = \Delta - 2$.
    It follows from \ref{256u3u4}(3) that
    $6\in C(u_2)$, $H$ contains a $(2, 6)_{(u_1, u_2)}$-path,
    and from \ref{a2a3a4}(2) that $6\in C(u_3)$, there exists a $a_3\in \{2, 4\}\cap C(u_3)$.
    It follows from \ref{256u3u4}(3) that
    $5\in C(u_2)$, $H$ contains a $(2, 5)_{(u_1, u_2)}$-path,
    and from \ref{a2a3a4}(3) that $5\in C(u_4)$, there exists a $a_4\in \{2, 3\}\cap C(u_4)$.
    For each $i, j\in \{3, 4\}$, by \ref{T2T3T4} (2), $T_i\cup (\{8\}\cup C(u_i))\subseteq C(u_j)$,
    and then $T_{u v}\cup \{8\}\subseteq C(u_3)\cup C(u_4)$;
    and if $1\not\in C(u_3)\cup C(u_4)$, $T_i\cup (\{7, 8\}\cup C(u_i))\subseteq C(u_j)$ and $T_{u v}\cup \{7, 8\}\subseteq C(u_3)\cup C(u_4)$.
    Hence, $d(u_3) + d(u_4)\ge |\{3, 5, 6, 2/4\}| + |\{4, 5, 6, 2/3\}| + |\{1\}| + |T_{u v}\cup \{8\}|= \Delta + 7$,
    or $d(u_3) + d(u_4) \ge |\{3, 5, 6, 2/4\}| + |\{4, 5, 6, 2/3\}| + |T_{u v}\cup \{7, 8\}| = \Delta + 7$.

    \quad It follows that $d(u_3) = d(u_4) = \Delta = 7$, $d(u_2)\le 6$,
    and for each $j\in T_{u v}\cup \{8\}$, mult$_{C(u_3)\cup C(u_4)}(j) = 1$,
    mult$_{C(u_3)\cup C(u_4)}(1) + $mult$_{C(u_3)\cup C(u_4)}(7) = 1$.

    \quad When $1\in C(u_3)\cup C(u_4)$,
    assume w.l.o.g. $C(u_3) = \{1, 3, 5, 6, 2/4, \alpha_1, \alpha_2\}$,
    $C(u_4) = \{4, 5, 6, 2/3, \beta_1, \beta_2, \beta_3\}$,
    where $\{\alpha_1, \alpha_2, \beta_1, \beta_2, \beta_3\} = [8, \Delta + 5]$.
    Since $7\not\in C(u_3)\cup C(u_4)$ and $1\not\in C(u_4)$,
    it follows from \ref{a2a3a4}(3) that $2\in C(u_4)$ and $C(u_2) = \{2\}\cup [5, 7]\cup \{\alpha_1, \alpha_2\}$.
    Then $(u u_4, u u_1, u v)\overset{\divideontimes}\to (1, 4, 9)$.
    In the other case, $1\not\in C(u_3)\cup C(u_4)$.

    \quad Assume w.l.o.g. $C(u_3) = \{3, 5, 6, 2/4, \alpha_1, \alpha_2, \alpha_3\}$,
    $C(u_4) = \{4, 5, 6, 2/3, \beta_1, \beta_2, \beta_3\}$,
    where $\{\alpha_1, \alpha_2, \alpha_3, \beta_1, \beta_2, \beta_3\} = [7, \Delta + 5]$.
    One sees that $T_3\ne \emptyset$, $T_4\ne\emptyset$.
    If $1\not\in C(u_2)$, then $(u u_2, u u_1)\to (1, 5)$ and $u v\overset{\divideontimes}\to T_3$.
    Otherwise, $1\in C(u_2)$.
    Since $\{\alpha_i\}\setminus C(u_2)\ne \emptyset$, $\{\beta_i\}\setminus C(u_2)\ne \emptyset$,
    it follows from \ref{a2a3a4} (2--3) that $a_4 = 3$, $a_3 = 4$,
    and from \ref{1a2a3a4}(1) that $3, 4\in C(u_2)$ and $C(u_2) = [1, 6]$.
    Then $u u_2\to T_4$, $(u u_3, u v)\overset{\divideontimes}\to (2, 3)$.

\item {\rm (iv)}
    $\{x, y\} = \{2, 5\}$.
    It follows from \ref{256u3u4}(3) that
    $6\in C(u_2)$, $H$ contains a $(2, 6)_{(u_1, u_2)}$-path,
    from \ref{a2a3a4}(2) that $6\in C(u_3)$, there exists a $a_3\in \{2, 4\}\cap C(u_3)$,
    and \ref{a2a3a4}(3) that there exists a $a_4\in \{2, 3\}\cap C(u_4)$ and $1/2/5\in C(u_4)$.
    For each $i, j\in \{3, 4\}$, by \ref{T2T3T4}, $T_i\cup ([7, 8]\cup C(u_i))\subseteq C(u_j)$,
    and then $T_{u v}\cup [7, 8]\subseteq C(u_3)\cup C(u_4)$.

    \quad It follows from \ref{234B1} that $2\in B_1\cap C(v_1)$.
    If $4, 5, 7, 8\not\in C(v_1)$,
    then $v v_1\to 4$ reduces to (2.2.2.);
    if $3, 6, 7, 8\not\in C(v_1)$,
    then $v v_1\to 3$ reduces to (2.2.2.).
    Otherwise, assume w.l.o.g. $7\in C(v_1)$, i.e., $C(v_1) = \{1, 2, 7\}\cup T_{u v}$,
    and $H$ contains a $(7, i)_{(v_1, v_7)}$-path for each $i\in \{3, 4\}$.

    \quad If $C(v_7) = \{7, 3, 4\}\cup T_{u v}$,
    then $(v v_1, v v_7)\to (3, 1)$ and $u v\overset{\divideontimes}\to T_3$.
    Otherwise, $R_7\ne \emptyset$.
    If $H$ contains no $(2, 7)_{(u_1, u_2)}$-path, then $u u_1\to 7$ and $uv\overset{\divideontimes}\to R_7$.
    Otherwise, $H$ contains a $(2, 7)_{(u_1, u_2)}$-path and $7\in C(u_2)$.
    It follows from \ref{a2a3a4} that (2) that if $7\not\in C(u_3)$,
    then $H$ contains a $(4, 7)_{(u_3, u_4)}$-path and $4\in C(u_3)$,
    and from \ref{a2a3a4} that (3) that if $7\not\in C(u_4)$,
    then $H$ contains a $(3, 7)_{(u_3, u_4)}$-path and $3\in C(u_4)$.

    \quad Since $2/3, 7\in C(u_4)$ or $3, 1/2/5\in C(u_4)$,
    $T_4\ne \emptyset$ and assume w.l.o.g. $9\in T_4$.
    If $1\not\in C(u_2)$ and $H$ contains no $(1, i)_{(u_2, u_i)}$-path for each $i\in [3, 4]$,
    then $(u u_2, u u_1, u v)\overset{\divideontimes}\to (1, 6, 9)$.
    Otherwise, $1/3/4\in C(u_2)$,
    and if $1\not\in C(u_2)$, then $H$ contains a $(1, i)_{(u_2, u_i)}$-path for some $i\in [3, 4]$.

    \quad When $2, 4\in C(u_3)$, or $2, 3\in C(u_4)$,
    or for some $j\in [7, 8]\cup T_{u v}$, mult$_{C(u_3)\cup C(u_4)}(j)\ge 2$,
    one see that $d(u_3) + d(u_4)\ge |\{3, 5, 6, 2/4\}| + |\{4, 6, 2/3\}| + \Delta - 1 + 1 = \Delta + 7$,
    it follows that $d(u_3)= d(u_4)= \Delta = 7$, $1, 5\not\in C(u_3)\cup C(u_4)$ and $d(u_2)\le 6$.
    Hence, $1, 5\in C(u_2)$, $2\in C(u_4)$,
    and it follows from \ref{1a2a3a4}(2) that $C(u_2) = [1, 2]\cup [5, 7]\cup \{a_2\}$ and $H$ contains a $(2, a_2)_{(u_1, u_2)}$-path, where $a_2\in \{3, 4\}$.
    If $H$ contains no $(8, a_2)_{(u_2, u_{a_2})}$-path,
    then $(u u_2, u u_1, u v)\overset{\divideontimes}\to (8, 6, 9)$.
    Otherwise, $H$ contains a $(8, a_2)_{(u_2, u_{a_2})}$-path and from \ref{a2a3a4}(2--3) that $8\in C(u_3)\cap C(u_4)$.
    It follows that $3\not\in C(u_4)$ and $7\in C(u_4)\setminus C(u_3)$.
    Recall that $H$ contains a $(2, 3)_{(u_1, u_2)}$-path even if $3\in C(u_2)$.
    Then $(u u_3, u u_4, u v)\overset{\divideontimes}\to (7, 3, 4)$.

    \quad In the other case, $|C(u_3)\cap \{2, 4\}| = 1$, $|C(u_4)\cap \{2, 3\}| = 1$,
    and for each $j\in [7, 8]\cup T_{u v}$, mult$_{C(u_3)\cup C(u_4)}(j) = 1$.
    One sees that if $4\in C(u_3)$,
    then for each $j\in T_{u v}$, $j\not\in C(u_4)$ or $H$ contains a $(4, j)_{(u_3, u_4)}$-path.
    If $2\not\in C(u_3)\cup C(u_4)$ and $T_2\ne\emptyset$,
    then $u u_2\to T_2$, and $(u u_3, u v)\overset{\divideontimes}\to (2, 3)$.
    Otherwise, $2\in C(u_3)\cup C(u_4)$ or $T_2 = \emptyset$.
    Hence, there are two scenarios.

    \quad Assume that $7\in C(u_3)\setminus C(u_4)$.
    Then $3, 1/5\in C(u_4)$.
    Since $d(u_3) + d(u_4)\ge |\{3, 5, 6, 2/4\}| + |\{4, 6, 3, 1/5\}| + \Delta - 1 = \Delta + 7$,
    it follows that $d(u_3)= d(u_4)= \Delta = 7$, and $d(u_2)\le 6$.
    Since $1/3/4\in C(u_2)$, $|T_3\cup [8]\setminus C(u_3)|\ge 3$, $a_3 = 4$.
    Then $2\not\in C(u_3)\cup C(u_4)$ and $T_2\ne \emptyset$.

    \quad Assume that $7\in C(u_4)\setminus C(u_3)$ and then $4\in C(u_3)$.
    One sees from \ref{a2a3a4}(3) that $5\in C(u_2)\cup C(u_4)$ or $1\in C(u_4)$.
    Since $d(u_2) + d(u_3) + d(u_4)\ge |\{3, 5, 6, 4\}| + |\{4, 6, 2/3\}| + \Delta - 1 + \{1/5, 2, 6, 7\}= \Delta + 10$,
    $T_2\ne \emptyset$ and then $2\in C(u_4)$.
    If $3\not\in C(u_2)$, then $(u u_3, u u_4, u v)\overset{\divideontimes}\to (7, 3, 4)$.
    Otherwise, $3\in C(u_2)$.
    It follows from $d(u_2) + d(u_4)\ge |\{2, 3, 6, 7\}| + |\{4, 6, 2, 7\}| + \Delta - 2 + \{1/5\}= \Delta + 7$
    that $d(u_2) = d(u_4) = \Delta = 7$, $1, 4\not\in C(u_2)$, and $1\in C(u_3)$.
    However, $d(u_3) + d(u_4)\ge |\{3, 5, 6, 4, 1\}| + |\{4, 6, 2\}| + \Delta - 1 = \Delta + 7$,
    and thus $d(u_3) = 7 > 6$.
\item {\rm (v)}
    $\{x, y\} = \{2, 3\}$.
    It follows from \ref{256u3u4}(1) that
    $5\in C(u_2)$, $H$ contains a $(2, 5)_{(u_1, u_2)}$-path,
    and from \ref{a2a3a4}(3) that $5\in C(u_4)$, there exists a $a_4\in \{2, 3\}\cap C(u_4)$.
    It follows from \ref{256u3u4}(3) that $H$ contains a $(i, 6)_{(u_1, u_i)}$-path for some $i\in [2, 3]$,
    and from \ref{256u3u4}(4) that if $6\not\in C(u_3)$, then $3\in B_1$;
    if $6\not\in C(u_2)$, then $2\in B_1$.
    If $2, 3\in C(v_1)$, then $v v_1\to 4$ reduces to (2.2.2.).
    Together with \ref{234B1}, $[2, 3]\cap C(v_1) = \{a\}$ and $a\in B_1$.

    \quad If $3\not\in C(v_1)$, and $H$ contains no $(3, i)_{(v_1, v)}$-path for each $i\in [6, 8]$,
    then $v v_1\to 3$ reduces the proof to (2.2.2.).
    If $4\not\in C(v_1)$, and $H$ contains no $(4, i)_{(v_1, v)}$-path for each $i\in \{5, 7, 8\}$,
    then $v v_1\to 4$ reduces the proof to (2.2.2.).
    If $2\not\in C(v_1)$, and $H$ contains no $(2, i)_{(v_1, v)}$-path for each $i\in [5, 8]$,
    then $v v_1\to 2$ reduces to (i) or (iv) if $w_2 = 3$,
    or if $w_2 = 2$, we can get an acyclic edge $(\Delta + 5)$-coloring for $H$ easily.
    Hence, by symmetry, there are two scenarios.

    \quad When $C(v_1) = \{1, 3, 5/7\}\cup T_{u v}$ and then $6\in C(u_2)$,
    it follows from \ref{a2a3a4} that there exists a $a_2\in \{3, 4\}\cap C(u_2)$,
    and from \ref{T2T3T4} that $T_{u v}\cup [7, 8]\subseteq C(u_2)\cup C(u_4)$.
    Since $d(u_2) + d(u_4)\ge \{2, 5, 6, 3/4\} + \{4, 5, 6, 2/3\} + \Delta - 1 = \Delta + 7$,
    it follows that $d(u_2) = d(u_4) = \Delta = 7$, and $d(u_3)\le 6$.
    Further, assume w.l.o.g. $C(u_2) = \{2, 5, 6, a_2\}\cup \{\beta_1, \beta_2, \beta_3\}$,
    $C(u_4) = \{4, 5, 6, a_4\}\cup \{\alpha_1, \alpha_2, \alpha_3\}$, where $\{\alpha_1, \alpha_2, \alpha_3, \beta_1, \beta_2, \beta_3\} = [7, 12]$.

    \quad If $H$ contains no $(1, 3)_{(u_2, u_3)}$-path,
    then $(u u_1, u u_2)\to (5, 1)$ and $u v\to (\{2\}\cup T_{u v})\setminus C(u_3)$.
    Otherwise, $H$ contains a $(1, 3)_{(u_2, u_3)}$-path.
    It follows that $a_2 = 3$, $C(u_3) = \{3, 5, 1\}\cup \{\beta_1, \beta_2, \beta_3\}$, and $a_4 = 2$.
    Then $(u u_2, u v)\to (4, 2)$ and $u u_4\overset{\divideontimes}\to T_4$.

    \quad In the other case, $C(v_1) = \{1, 2, 7\}\cup T_{uv}$ and $H$ contains a $(7, i)_{(v_1, v_7)}$-path,
    $6\in C(u_3)$.
    It follows from \ref{a2a3a4} that there exists a $a_3\in \{2, 4\}\cap C(u_3)$,
    and from \ref{T2T3T4} that $T_{u v}\cup [7, 8]\subseteq C(u_3)\cup C(u_4)$.
    Since $d(u_3) + d(u_4)\ge \{3, 5, 6, 2/4\} + \{4, 5, 6, 2/3\} + \Delta - 1 = \Delta + 7$,
    it follows that $d(u_3) = d(u_4) = \Delta = 7$, and $d(u_2)\le 6$.
    Further, assume w.l.o.g. $C(u_3) = \{3, 5, 6, a_3\}\cup \{\beta_1, \beta_2, \beta_3\}$,
    $C(u_4) = \{4, 5, 6, a_4\}\cup \{\alpha_1, \alpha_2, \alpha_3\}$, where $\{\alpha_1, \alpha_2, \alpha_3, \beta_1, \beta_2, \beta_3\} = [7, 12]$.
    If $C(v_7) = \{7, 3, 4\}\cup T_{uv}$,
    then $(v v_1, v v_7)\to (3, 1)$ and $uv\overset{\divideontimes}\to T_3$.
    Otherwise, $R_7\ne \emptyset$.

    \quad If $1\not\in C(u_2)$, then $(u u_1, u u_2)\to (5, 1)$ and $u v\to T_4$.
    Otherwise, $1\in C(u_2)$.
    If $a_3 = 2$, then $C(u_2) = \{2, 1, 5, \alpha_1, \alpha_2, \alpha_3\}$, $a_4 = 3$,
    and $(u u_4, u u_1, u v)\to (1, 4, 9)$.
    Otherwise, $a_3 = 4$.
    If $6\not\in C(u_2)$, then $(u u_3, u u_1)\to (1, 6)$ and $uv\to T_4$;
    if $7\not\in C(u_2)$, then $(u u_3, u u_1)\to (1, 7)$ and $uv\to R_7$,
    otherwise, $6, 7\in C(u_2)$ and then $a_4 = 3$.
    Recall that for each $j\in T_2$, $j\in C(u_4)$ or $H$ contains a $(4, j)_{(u_3, u_4)}$-path.
    Then $(u u_3, u v)\to (2, 3)$ and $u u_2\overset{\divideontimes}\to T_2$.
\end{itemize}

This finishes the inductive step for the case where $G$ contains the configuration ($A_8$),
and thus completes the proof of Lemma~\ref{auv1coloring}.
\end{proof}

By Lemmas~\ref{4to321}, \ref{3to21} and \ref{auv1coloring},
hereafter we assume that $a_{u v} = 2$.
It suffices to show that in $O(1)$ time an acyclic edge $(\Delta + 5)$-coloring for $G$ can be obtained,
or an acyclic edge $(\Delta + 5)$-coloring for $H$ can be obtained such that $|C(u)\cap C(v)| \le 1$.

\begin{corollary}
\label{auv2}
If $a_{u v} = 2$, i.e., $A_{u v} = \{i_1, i_2\}$ and $U_{u v} = \{i_3, \ldots, i_{k - 1}\}$,
then either an acyclic edge $(\Delta + 5)$-coloring for $H$ can be obtained such that $|C(u)\cap C(v)| \le 1$,
or the following holds:
\begin{itemize}
\parskip=0pt
\item[{\rm (1)}]
   for each $j\in T_{i_1}\ne \emptyset$, mult$_{S_u}(j)\ge 2$, $G$ contains a $(i_2, j)_{(u_{i_2}, v)}$-path, and
   \begin{itemize}
   \parskip=0pt
   \item[{\rm (1.1)}]
        if $d(u) = 4$, then $i_3\in C(u_{i_1})$, $G$ contains a $(i_3, j)_{(u_{i_1}, u_{i_3})}$-path,
        and $T_{i_1}\subseteq C(u_{i_3})$;
   \item[{\rm (1.2)}]
        if $d(u) = 5$, then $U_{u v} = \{i_3, i_4\}$, there exists a $a_{i_1}\in \{i_3, i_4\}\cap C(u_{i_1})$ such that
			$G$ contains a $(a_{i_1}, j)_{(u_{i_1}, u_{a_{i_1}})}$-path, $j\in C(u_{a_{i_1}})$, and $T_{i_1}\subseteq C(u_{i_3})\cup C(u_{i_4})$;
   \end{itemize}
\item[{\rm (2)}]
    $T_{i_1}\subseteq C(u_{i_2})$, $T_{i_2}\subseteq C(u_{i_1})$,
    and $T_{i_1}\cap T_{i_2} = \emptyset$.
\item[{\rm (3)}]
   For each $\alpha\in U_{u v}$,
   if there exists an $j\in T_{\alpha}$ such that $u u_{\alpha}\to j$ cannot obtain a bichromatic cycle,
   then $\alpha\in \bigcup_{i\in \{i_1, i_2\}}B_{i}$,
   mult$_{S_u}(\alpha)\ge 2$,
   and for each $i\in \{i_1, i_2\}$, $\alpha\in C(u_i)$, or there exists a $a_i\in C(u_i)\cap \{i_3, i_4\}$ such that $H$ contains a $(\alpha, a_i)_{(u_{\alpha}, u_i)}$-path.
\item[{\rm (4)}]
   For each $\beta\in V_{u v}$,
   if there exists an $j\in R_{\beta}$ such that $vv_\beta\to j$ cannot obtain a bichromatic cycle,
   then $\beta\in \bigcup_{i\in \{i_1, i_2\}}B_{i}$,
   mult$_{S_u}(\beta)\ge 2$,
   and for each $i\in \{i_1, i_2\}$, $\beta\in C(u_i)$, or there exists a $a_i\in C(u_i)\cap \{i_3, i_4\}$ such that $H$ contains a $(\beta, a_i)_{(u_{a_i}, u_i)}$-path.
\item[{\rm (5)}]
   for each $j\in T_{u v}$, if mult$_{S_u}(j) = 2$ with $j\in T_{i_1}\cap C(u_{i_2})\cap C(u_{i_3})$, and
   \begin{itemize}
   \parskip=0pt
   \item[{\rm (5.1)}]
        if $(T_{i_2}\cup \{i_3\})\setminus C(u_{i_2})\ne \emptyset$, then $i_1\in S_u$, $V_{u v}\subseteq S_u$,
        and {\rm (i)} $i_1\in C(u_{i_3})$ or there exists an $i\in C(u)$ such that $H$ contains a $(i, i_1)_{(u_{i_3}, u_i)}$-path; {\rm (ii)} for each $l\in V_{uv}$, $l\in C(u_{i_1})\cup C(u_{i_3})$,
        or there exists an $i\in C(u)$ such that $H$ contains a $(i, l)_{(u_{i_1}, u_i)}$-path,
        or there exists an $j\in C(u)$ such that $H$ contains a $(j, l)_{(u_{i_3}, u_j)}$-path;
   \item[{\rm (5.2)}]
        if $d(u) = 5$ with $U_{u v} = \{i_4\}$,
        then $H$ contains no $(j, i)_{(u_{i_4}, u_i)}$-path for each $i\in C(u)$, mult$_{S_u}(i_4)\ge 2$,
        and $H$ contains a $(i_3, l)_{(u_{i_1}, u_{i_3})}$-path for each $l\in T_{i_1}$, $T_{i_1}\subseteq C(u_{i_3})$.
   \end{itemize}
\end{itemize}
\end{corollary}
\begin{proof}
It follows from Lemma~\ref{auvge2}(1.1) that (1--2) holds,
and from Lemma~\ref{auvge2}(3--4) that (3--4) hold.

For (5), assume that $T_{i_2}\cup (\{i_3\}\setminus C(u_{i_2}))\ne \emptyset$.
If $i_1\not\in C(u_{i_3})$ and $H$ contains no $(i, i_1)_{(u_{i_3}, u_i)}$-path for each $i\in C(u)$,
then $(u u_{i_1}, u u_{i_3})\to (j, i_1)$ and $u v\to (T_2\cup \{i_3\})\setminus C(u_2)$.
Otherwise, $i_1\in C(u_{i_3})$ or $H$ contains a $(i, i_1)_{(u_{i_3}, u_i)}$-path for some $i\in C(u)$.
It follows that $i_1\in S_u$.

For each $l\in V_{uv}$, if $l\not\in C(u_{i_1})\cup C(u_{i_3})$,
and $H$ contains no $(i, l)_{(u_{i_1}, u_i)}$-path for each $i\in C(u)$,
one sees that the same argument applies if $uu_{i_1}\to l$ and it follows that there exists an $j\in C(u)$ such that $H$ contains a $(j, l)_{(u_{i_3}, u_j)}$-path.
Otherwise, $l\in C(u_{i_1})\cup C(u_{i_3})$,
or there exists an $i\in C(u)$ such that $H$ contains a $(i, l)_{(u_{i_1}, u_i)}$-path,
or there exists an $j\in C(u)$ such that $H$ contains a $(j, l)_{(u_{i_3}, u_j)}$-path,
and it follows that $l\in S_u$ and then $V_{uv}\subseteq S_u$.
Hence, (5.1) holds.

Assume $d(u) = 5$ with $U_{u v} = \{i_4\}$.
Since $H$ contains a $(i_2, j)_{(u_{i_2}, v)}$-path and a $(i_3, j)_{(u_{i_1}, u_{i_3})}$-path,
by Lemma~\ref{lemma06}, $H$ contains no $(j, i)_{(u_{i_4}, u_i)}$-path for each $i\in C(u)$.
One sees that same argument to show Lemma~\ref{auvge2}(2) applies if $uu_{i_4}\to j$.
Hence, mult$_{S_u}(i_4)\ge 2$.
If $H$ contains no $(i_3, l)_{(u_{i_1}, u_{i_2})}$-path for some $l\in T_{i_1}\setminus \{j\}$,
then $(u_{i_1}, u_{i_4})\to (l, j)$ reduces to an acyclic edge $(\Delta + 5)$-coloring for $H$ such that $C(u)\cap C(v) = \{i_2\}$.
Hence, (5.2) holds.
\end{proof}

Recall that Eq.(\ref{eq17})
\begin{equation*}
\begin{aligned}
\sum_{i\in [1, k - 1]}d(u_i) = \sum_{i\in [1, k - 1]}w_i + \sum_{j\in T_{uv}}mult_{S_u}(j)\ge \sum_{i\in [1, k - 1]}w_i + 2t_{uv} + \sum_{j\in T_0}(mult_{S_u}(j) - 2).
\end{aligned}
\end{equation*}

Let $c(u u_i) = i$, $i\in C(u) = [1, d(u) - 1]$, $c(v v_j) = j$, $j\in \{i_1, i_2\}\cup [d(u), d(u) + d(v) - 4]$,
where $A_{u v} = \{i_1, i_2\}$.
One sees clearly that $t_{u v} = \Delta + 5 - (d(u) + d(v) - 4) = \Delta - (d(u) + d(v)) + 9$.

We discuss the following two cases for $d(u) = 4$ and $d(u) = 5$, respectively.

\subsubsection{$d(u) = 4$ and $c(u u_i) = i\in [1, 3]$, $c(v v_j) = j\in \{i_1, i_2\}\cup [4, d(v)]$, where $i_1, i_2\in [1, 3]$.}


In this case, $5\le d(v)\le 7$,
and $t_{u v} = \Delta - (d(u) + d(v)) + 9 = \Delta + 5 - d(v)$.

{\bf Case 1.}  $d(v)\le 5$ and $t_{u v} = \Delta + 5 - d(v)\ge \Delta$ that $T_{i_1}\ne \emptyset$, $T_{i_2}\ne \emptyset$.

It follows from Corollary~\ref{auv2} (6.1) that $i_3\in C(u_{i_1})\cap C(u_{i_2})$, ${i_1}, {i_2}\in S_u$,
and $V_{u v} = [4, d(v)]\subseteq S_u$,
thus $\sum_{x\in N(u)}d(x) \ge 2t_{uv} + $mult$_{S_u}(i_3) + |C(u)| + |[i_1, i_2]\cup [4, d(v)]| + d(v) = 2(\Delta + 5 - d(v)) + 2 + 3 + d(v) -1 + d(v) = 2\Delta + 14$.
Hence, ($A_{6.3}$) holds where $d(v) = 5$, $\sum_{x\in N(u)}d(x) = 2\Delta + 14$,
and for each $i\in \{i_1, i_2\}\cup \{4, 5\}$, mult$_{S_u}(i) = 1$.
By symmetry, we assume w.l.o.g. that $\{i_1, i_2\} = \{2, 3\}$ or $\{i_1, i_2\} = \{1, 2\}$.
Note that $c(v u_2)\not\in A_{u v}$.
Assume that $c(v u_2) = 4$ and then $4\not\in C(u_1)\cup C(u_3)$.
It follows from Corollary~\ref{auv2}(4) that $5\in C(u_2)$ and then $5\not\in C(u_1)\cup C(u_3)$.
One sees that if $\{i_1, i_2\} = \{2, 3\}$, then $(v u_2, v u_3)_c = (4, 5)$.
Hence, it suffices to assume that $\{i_1, i_2\} = \{1, 2\}$ and $c(v u_3)\in A_{u v}$.
It follows from Lemma~\ref{auvge2}(1.2) that there exists a $a_i\in \{4, 5\}\cap C(u_3)$.
Thus, mult$_{S_u}(a_i) \ge 2$, a contradiction.


{\bf Case 2.} $d(v) = 6$.
Then $t_{u v} = \Delta - 1$, $V_{u v} = [4, 6]$.

Since $u v$ is contained in a $3$-cycle, assume that $v u_2\in E(G)$.
One sees that if $2\in A_{u v}$, then $c(v u_2)\not\in A_{u v}$.
We will first show that $T_{i_1}\ne\emptyset$, $T_{i_2}\ne \emptyset$,
$\sum_{x\in N(u)}d(x) = 2\Delta + 14$,
and for each $i\in \{i_1, i_2\}\cup [4, 6]$, mult$_{S_u}(i) = 1$.
By symmetry, assume w.l.o.g.$\{i_1, i_2\} = \{1, 2\}$ or $\{i_1, i_2\} = \{1, 3\}$.
\begin{itemize}
\parskip=0pt
\item
    $\{i_1, i_2\} = \{1, 2\}$. Assume w.l.o.g. $c(v u_2) = 4$.
    One sees that $T_2\ne \emptyset$.
    When $C(u_1) = \{1\}\cup T_{u v}$, $u u_1\to 4$ and $u v\to T_2$.
    In the other case, $T_1\ne\emptyset$.
\item
    $\{i_1, i_2\} = \{1, 3\}$ and $c(v u_2) = 1$.
    One sees that $T_2\ne \emptyset$ and $H$ contains a $(3, j)_{(u_3, v_3)}$-path.
    When $C(u_i) = \{i\}\cup T_{u v}$ for some $i\in \{1, 3\}$, $u u_i\to 2$ and $u u_2\to T_2$ reduces to an acyclic edge $(\Delta + 5)$-coloring for $H$ such that $|C(u)\cap C(v)| = 1$.
    In the other case, $T_1\ne\emptyset$ and $T_3\ne \emptyset$.
\item
    $\{i_1, i_2\} = \{1, 3\}$ and $c(v u_2) = 4$.
    When $C(u_i) = \{i\}\cup T_{u v}$ for some $i\in \{1, 3\}$, $u u_i\to 4$ reduces to the above case where $c(v u_2)\in A_{u v}$.
    In the other case, $T_1\ne\emptyset$ and $T_3\ne \emptyset$.
\end{itemize}

It follows from Corollary~\ref{auv2} (6.1) that $i_3\in C(u_{i_1})\cap C(u_{i_2})$, $i_1, i_2, 4, 5, 6\in S_u$,
and $\sum_{x\in N(u)}d(x) \ge 2t_{u v} + $mult$_{S_u}(i_3) + |C(u)| + |\{i_1, i_2, 4, 5, 6\}| + d(v) = 2(\Delta - 1) + 2 + 3 + 5 + 6 = 2\Delta + 14$.
Hence, ($A_{6.3}$) holds where $d(v) = 6$, $\sum_{x\in N(u)}d(x) = 2\Delta + 14$,
and for each $i\in \{i_1, i_2\}\cup [4, 6]$, mult$_{S_u}(i) = 1$.

By symmetry, we assume w.l.o.g. that $\{i_1, i_2\} = \{2, 3\}$ or $\{i_1, i_2\} = \{1, 2\}$.
Since $c(v u_2)\not\in A_{u v}$, assume w.l.o.g. $c(v u_2) = 4$.
Then $4\not\in C(u_1)\cup C(u_3)$.
It follows from Corollary~\ref{auv2}(4) that there exists a $a_i\in \{5, 6\}\cap C(u_2)$, say $a_i = 6$,
and then $6\not\in C(u_1)\cup C(u_3)$.

When $c(v u_3) = 5$, since $4, 6\not\in C(u_3)$ and by Corollary (4),
we can get an acyclic edge $(\Delta + 5)$-coloring for $H$.
In the other case, $c(v u_3) \ne 5$.
One sees that if $\{i_1, i_2\} = \{2, 3\}$, then $c(v u_3) = 5$.
Hence, it suffices to assume that $\{i_1, i_2\} = \{1, 2\}$ and $c(v u_3)\in A_{u v}$.

Recall that $4\not\in C(u_1)\cup C(u_3)$.
Let $7\not\in C(u_2)$.
If $H$ contains no $(2, 4)_{(u_1, v)}$-path,
then $(u u_1, u v)\overset{\divideontimes}\to (4, 7)$.
Otherwise, $H$ contains a $(2, 4)_{(u_1, v)}$-path and then $c(v u_3) = 1$.
Then $(u u_1, u u_2, u v)\overset{\divideontimes}\to (4, 1, 7)$.


{\bf Case 3.} $v$ is $(7, 4^-)$, $u_1 u_2$, $u_1 u_3$, $v u_2$,
$v u_3\in E(G)$ and $\min \{d(u_1), d(u_2), d(u_3)\}\le 8$.

Then $t_{u v} = \Delta - 2\ge 5$, $V_{u v} = [4, 7]$.
It follows from \ref{prop3001} that $T_{u v}\subseteq B_{i_1}\cup B_{i_2}$,
and from Lemma~\ref{auvge2}(2) that for each $j\in T_{u v}$, mult$_{S_u}(j)\ge 2$.
One sees that for each $i\in \{i_1, i_2\}$, if $d(v_i)\le 5$, then $R_i\ne \emptyset$.
By symmetry, we discuss the following seven subcases, respectively.

(3.1.) $\{i_1, i_2\} = \{1, 2\}$ and $(v u_2, v u_3)_c = (1, 2)$.
It follows that $C(u_2) = [1, 2]\cup T_{u v}$.
If there exists a $j\in T_3$, then $u u_3\to T_3$;
otherwise, $C(u_3) = [2, 3]\cup T_{u v}$ and $u u_3\to 4$.
Then $u v\overset{\divideontimes}\to 3$.

(3.2.) $\{i_1, i_2\} = \{2, 3\}$ and $(v u_2, v u_3)_c = (3, 4)$.
It follows that $C(u_2) = [2, 3]\cup T_{u v}$, $c(u_1 u_2) = 8$ with $2/3\in C(u_1)$.
If $4\not\in C(u_1)$ or $1\not\in C(u_3)$, then $(u u_2, v u_2)\to (4, 1)$ reduces the proof to (1.3.1.).
Otherwise, $4\in C(u_1)$ and $1\in C(u_3)$.
One sees clearly that $T_3\ne \emptyset$ with $9\in T_3$ and $T_1\ne \emptyset$.
It follows from Corollary~\ref{auv2} (1.1) that $8\in C(u_3)$ and $H$ contains a $(1, 9)_{(u_1, u_3)}$-path.

If $3\not\in C(u_1)$, then $u u_1\to T_3$ and $u u_2\to 1$ reduces to an acyclic edge $(\Delta + 5)$-coloring for $H$ such that $a_{u v} = 1$.
If there exists an $i\in [5, 7]\setminus S_u$,
then $(u u_1, u v)\overset{\divideontimes}\to (i, 1)$.
Otherwise, $3\in C(u_1)$ and $[5, 7]\subseteq S_u$.
It follows that $d(u_1) + d(u_3)\ge 2\times |\{1, 3, 4, 8\}| + |[9, \Delta + 5]| + |[5, 7]| = 8 + \Delta - 3 + 3 =\Delta + 8$ and then $d(u_1) + d(u_3) = 8$.
One sees that $2\not\in S_u$.
Then $(u u_3, u u_2, v u_2, u v)\overset{\divideontimes} \to (2, 3, 1, 9)$.

(3.3.) $\{i_1, i_2\} = \{1, 2\}$ and $(v u_2, v u_3)_c = (1, 4)$.
It follows that $C(u_2) = [1, 2]\cup T_{u v}$ and $c(u_1 u_2) = 8$.
If $4\not\in C(u_1)$ or $1\not\in C(u_3)$, then $u u_2\to 4$ reduces the proof to (3.1.).
Otherwise, $4\in C(u_1)$ and $1\in C(u_3)$.
If $2\not\in C(u_3)$, then $v u_2\to 3$ reduces the proof to (3.2.).
Otherwise, $2\in C(u_3)$.
One sees clearly that $T_3\ne \emptyset$ with $9\in T_3$.
If $3\not\in C(u_1)$, then $(u u_3, u u_1, u u_2)\to (9, 3, 4)$ reduces to an acyclic edge $(\Delta + 5)$-coloring for $H$ such that $a_{u v} = 1$.
If there exists an $i\in [5, 7]\setminus S_u$, then $(u u_3, u v)\overset{\divideontimes}\to (i, 3)$.
Otherwise, $3\in C(u_1)$ and $[5, 7]\subseteq S_u$.

It follows that $d(u_1) + d(u_3)\ge |\{1, 3, 4, 8\}| + |[1, 4]| + |[9, \Delta + 5]| + |[5, 7]| = 8 + \Delta - 3 + 3 =\Delta + 8$ and then $d(u_1) + d(u_3) = 8$.
One sees that $2\not\in C(u_1)$ and $8\not\in C(u_3)$.
Then $(u u_3, u v)\overset{\divideontimes} \to (8, 3)$.

(3.4.) $\{i_1, i_2\} = \{1, 2\}$ and $(v u_2, v u_3)_c = (4, 1)$.
If $C(u_3) = \{1, 3\}\cup T_{u v}$,
then $(v u_3, u u_3)\to (3, 4)$ reduces the proof to (3.2.).
Otherwise, $T_3\ne \emptyset$ and assume w.l.o.g. $8\in T_3$.
One sees that $H$ contains a $(2, 8)_{(u_2, v_2)}$-path.
If $H$ contains no $(2, 3)_{(u_2, v_2)}$-path,
then $(u u_3, u v)\overset{\divideontimes}\to (8, 3)$.
Otherwise, $H$ contains a $(2, 3)_{(u_2, v_2)}$-path.
It follows that $3\in C(u_2)$, $T_2\ne\emptyset$, and assume w.l.o.g. $9\in T_2$.

If $3\not\in C(u_1)$, then $(u u_3, u u_1)\to (8, 3)$ reduces to an acyclic edge $(\Delta + 5)$-coloring for $H$ such that $a_{u v} = 1$;
if $2\not\in S_u$, then $(u u_3, u u_2, u v)\overset{\divideontimes}\to (2, 9, 3)$.
Otherwise, $3\in C(u_1)$ and $2\in S_u$.
One sees that if there exists an $i\in [5, 7]\setminus S_u$,
then similar argument applies if $u u_2\to i$.
Hence, $[5, 7]\subseteq S_u$.

If $4\not\in C(u_1)$ and $H$ contains no $(4, i)_{(u_1, u_i)}$-path for each $i\in [2, 3]$,
then $u u_1\to 4$ reduces the proof to (3.1.).
If $5\not\in C(u_3)$ and $H$ contains no $(4, i)_{(u_3, i_i)}$-path for each $i\in [1, 2]$,
then $(u u_3, u v)\overset{\divideontimes}\to (4, 3)$.
Otherwise, $4\in C(u_1)$ or $H$ contains a $(4, i)_{(u_1, u_i)}$-path for some $i\in [2, 3]$,
and $4\in C(u_3)$ and $H$ contains a $(4, i)_{(u_3, i_i)}$-path for some $i\in [1, 2]$.
It follows that $4\in C(u_1)\cup C(u_3)$.

Since $d(u_1) + d(u_2) + d(u_3)\ge |\{1, 3\}| + |[2, 4]| + |\{1, 3\}| + 2\times |[8, \Delta + 5]| + |\{2\}\cup [4, 7]\}| = 7 + 2(\Delta - 2) + 5 = 2\Delta + 8$,
we have $d(u_1) + d(u_2) + d(u_3) = 2\Delta + 8$, $1\not\in C(u_2)$,
for each $i\in \{2\}\cup [4, 7]\}$, mult$_{S_u}(i) = 1$,
and for each $j\in T_{u v}$, mult$_{S_u}(j) = 2$.
It follows that $c(u_1 u_2), c(u_1 u_3)\not\in \{2\}\cup [4, 7]\}$.
Assuming w.l.o.g. $(u_1 u_2, u_1 u_3)_c = \{\alpha, \beta\}\subseteq T_{uv}$.

Since mult$_{S_u}(\alpha) = 2$, $\alpha\in T_3$.
Then $(u u_3, u u_1, u u_2)\to (\alpha, 4, 1)$ reduces the proof to (1.3.1.).

(3.5.) $\{i_1, i_2\} = \{1, 2\}$ and $(vu_2, vu_3)_c = (4, 2)$.
When $H$ contains no $(1, 3)_{(u_1, v_1)}$-path,
if there exists a $j\in T_3$, then $u u_3\overset{\divideontimes}\to j$.
otherwise, $C(u_3) = [2, 3]\cup T_{u v}$;
if there exists an $i\in [5, 7]$, then $u u_3\to i$.
otherwise, $[5, 7]\subseteq C(u_2)$, $T_2\ne\emptyset$, and $u u_3\to 4$, $u u_2\to T_2$.
Then $u v\overset{\divideontimes}\to 3$.
In the other case, $H$ contains a $(1, 3)_{(u_1, v_1)}$-path and $3\in C(u_1)$.

When $3\not\in C(u_2)$, it follows that $C(u_2) = \{2, 4\}\cup T_{u v}$,
assume w.l.o.g. $c(u_1 u_2) = 8$.
Since $8\in B_1\cup B_2$, $2\in C(u_1)$ and then $T_3\ne\emptyset$.
If there exists a $j\in T_3$, then $u u_2\to 3$ and $u u_3\to j$ reduces to an acyclic edge $(\Delta + 5)$-coloring for $H$ such that $a_{u v} = 1$.
Otherwise, $C(u_3) = [2, 3]\cup T_{u v}$.
Then $u u_3\to 1$ and $u u_1\to T_1$ reduces the proof to (3.1.).
In the other case, $3\in C(u_2)$, $T_2\ne\emptyset$ and assume w.l.o.g. $8\in T_2$.
One sees that $H$ contains a $(1, 8)_{(u_1, v_1)}$-path.

When $1\not\in S_u$, if there exists a $j\in T_1$, then $(u u_3, u u_1)\to (1, j)$ reduces the proof to (3.1.);
otherwise, $C(u_1) = \{1, 3\}\cup T_{u v}$, then $(u u_3, u u_1, u v)\to (1, 4, 8)$.
In the other case, $1\in S_u$.
One sees that if there exists an $i\in [5, 7]\setminus S_u$,
then similar argument applies if $u u_1\to i$.
Hence, $[5, 7]\subseteq S_u$.

If $4\not\in C(u_1)\cup C(u_3)$, then $u u_1\to 4$ reduces the proof to (3.1.).
Otherwise, $4\in C(u_1)\cup C(u_3)$.

Since $d(u_1) + d(u_2) + d(u_3)\ge |\{1, 3\}| + |[2, 4]| + |\{2, 3\}| + 2\times |[8, \Delta + 5]| + |\{1\}\cup [4, 7]\}| = 7 + 2(\Delta - 2) + 5 = 2\Delta + 8$,
we have $d(u_1) + d(u_2) + d(u_3) = 2\Delta + 8$, $2\not\in C(u_1)$,
for each $i\in \{1\}\cup [4, 7]\}$, mult$_{S_u}(i) = 1$,
and for each $j\in T_{u v}$, mult$_{S_u}(j) = 2$.
It follows that $c(u_1 u_2), c(u_1 u_3)\not\in \{1\}\cup [4, 7]\}$.

Assume w.l.o.g. $c(u_1 u_3) \in T_{u v}$.
It follows from $2\not\in C(u_1)$ that $1\in C(u_3)\setminus C(u_2)$.
Recall that $H$ contains a $(1, 3)_{(u_1, v_1)}$-path.
We have $c(u_1 u_2)\ne 3$ and assume w.l.o.g. $c(u_1 u_2) = 9$.
Thus $u v\overset{\divideontimes}\to 9$\footnote{Since $1\not\in C(u_2)$ and $2\not\in C(u_1)$, $9\not\in B_1\cup B_2$, a contradiction to \ref{prop3001}. }.

(3.6.) $\{i_1, i_2\} = \{2, 3\}$ and $(v u_2, v u_3)_c = (5, 4)$.
If $5\not\in C(u_3)$ and $H$ contains no $(5, i)_{(u_3, u_i)}$-path for each $i\in [1, 2]$,
then $u u_3\to 5$ reduces the proof to (3.2.).
If $4\not\in C(u_3)$ and $H$ contains no $(4, i)_{(u_2, u_i)}$-path for each $i\in \{1, 3\}$,
then $u u_2\to 4$ reduces the proof to (3.2.).
Otherwise, $5\in C(u_3)$ or $H$ contains a $(5, i)_{(u_3, u_i)}$-path for some $i\in [1, 2]$,
and $4\in C(u_3)$ and $H$ contains a $(4, i)_{(u_2, u_i)}$-path for some $i\in \{1, 3\}$.
It follows that $1/2/5\in C(u_3)$, $1/3/4\in C(u_2)$ and $T_2\ne\emptyset$, $T_3\ne\emptyset$.
Hence, it follows from Corollary~\ref{auv2} (1.1) that $1\in C(u_2)\cap C(u_3)$,
and from (6.1) that $2, 3, 6, 7\in S_u$.

Since $d(u_1) + d(u_2) + d(u_3)\ge |\{1\}| + |\{1, 2, 5\}| + |\{1, 3, 4\}| + 2\times |[8, \Delta + 5]| + |[2, 3]\cup [6, 7]|= 7 + 2(\Delta - 2) + 4 = 2\Delta + 7$,
assume w.l.o.g. $5\not\in C(u_1)\cup C(u_3)$.
It follows that $H$ contains a $(2, 5)_{(u_3, u_2)}$-path.
Then $u u_1\to 5$, $u u_3\to T_3$, and $u v\overset{\divideontimes}\to T_2$.

(3.7.) $\{i_1, i_2\} = \{1, 2\}$ and $(v u_2, v u_3)_c = (5, 4)$.
If $4\not\in C(u_2)$ and $H$ contains no $(1, 4)_{(u_1, u_2)}$-path,
then $u u_2\to 4$ reduces the proof to (3.5.).
If $4\not\in C(u_1)$ and $H$ contains no $(2, 4)_{(u_1, u_2)}$-path,
then $u u_1\to 4$ reduces the proof to (3.4.).
Otherwise, $4\in C(u_2)$ or $H$ contains no $(1, 4)_{(u_1, u_2)}$-path;
and $4\in C(u_1)$ or $H$ contains no $(2, 4)_{(u_1, u_2)}$-path.
It follows that $1/4\in C(u_2)$, $T_2\ne\emptyset$, $3\in C(u_2)$,
and $2/4\in C(u_1)$, $4\in C(u_1)\cup C(u_2)$.

If $5\not\in C(u_1)\cup C(u_3)$,
then $u u_3\to 5$ and $u u_2\to T_2$ reduces the proof to (3.5.).
If $5, 2, 3\not\in C(u_1)$, then $u u_1\to 5$ reduces the proof to (3.3.).
Otherwise, $5\in C(u_1)\cup C(u_3)$ and $2/3/5\in C(u_1)$.

One sees that $(T_1\cup \{3\})\setminus C(u_1)\ne \emptyset$.
If follows from Corollary~\ref{auv2} (5.1) that $2\in S_u$.
One sees that if there exists an $i\in [6, 7]\setminus S_u$,
then similar argument applies if $u u_2\to i$.
Hence, $6, 7\in S_u$.

When $1\not\in S_u$, it follows that $4\in C(u_2)$.
If $4\not\in C(u_1)$, then $(u u_1, u u_2)\to (4, 1)$ reduces the proof to (3.4.).
Otherwise, $4\in C(u_1)$.
It follows from $2/3/5\in C(u_1)$ that $T_1\ne \emptyset$ and then $3\in C(u_1)$.
Then $d(u_1) + d(u_2) + d(u_3)\ge |\{1, 3, 4\}| + |[2, 5]| + |\{3, 4\}| + 2\times |[8, \Delta + 5]| + |\{2\}\cup [5, 7]|= 9 + 2(\Delta - 2) + 4 = 2\Delta + 9> 2\Delta + 8$, a contradiction.

In the other case, $1\in S_u$.
Since $d(u_1) + d(u_2) + d(u_3)\ge |\{1\}| + |\{2, 3, 5\}| + |\{3, 4\}| + 2\times |[8, \Delta + 5]| + |\{1, 2\}\cup [4, 7]|= 6 + 2(\Delta - 2) + 6 = 2\Delta + 8$,
we have for each $i\in [1, 7]$, mult$_{S_u}(i) = 1$,
and for each $j\in T_{uv}$,  mult$_{S_u}(i) = 2$.
It follows that $3\not\in C(u_1)$, $C(u_1) = [1, 2]\cup T_{uv}$, and $c(u_1u_2)\not\in [1, 7]$.
Assume w.l.o.g. $c(u_1 u_2) = 8$.
Since $8\in T_3$, $(u u_1, u u_2, u u_3)\to (4, 1, 8)$  reduces the proof to (3.4.).

\subsubsection{$d(u) = 5$ and $c(u u_i) = i\in [1, 4]$, $c(v v_j) = j\in \{i_1, i_2\}\cup [5, d(v) + 1]$, where $i_1, i_2\in [1, 4]$. }


In this subsection, $5\le d(v)\le 6$,
and $t_{u v} = \Delta - (d(u) + d(v)) + 9 = \Delta + 4 - d(v)$.
One sees clearly from \ref{prop3001} that $T_{u v}\subseteq B_{i_1}\cup B_{i_2}\subseteq C(u_{i_1})\cup C(u_{i_2})$.
Besides Corollary~\ref{auv2}, we additionally have the following Corollary~\ref{auv2d(u)5},
which plays a very important role in the proof.

\begin{corollary}
\label{auv2d(u)5}
If $d(u) = 5$, $a_{u v} = 2$, i.e., $A_{u v} = \{i_1, i_2\}$ and $U_{u v} = \{i_3, i_4\}$,
then either an acyclic edge $(\Delta + 5)$-coloring for $H$ can be obtained such that $|C(u)\cap C(v)| \le 1$,
or the following holds:
\begin{itemize}
\parskip=0pt
   \item[{\rm (1)}]
        assume $d(u) = 5$ and $C(u_{i_3}) = \{i_3\}\cup T_{i_1}\cup T_{i_2}$ with $T_{i_1}\ne\emptyset$.
        \begin{itemize}
        \parskip=0pt
        \item[{\rm (1.1)}]
             if $(T_{i_2}\cup \{i_3\})\setminus C(u_{i_2})\ne \emptyset$,
             then $H$ contains a $(j, i_4)_{(u_{i_1}, u_{i_4})}$-path,
             and $i_4\in C(u_{i_1})$, $T_{i_1}\subseteq C(u_{i_4})$;
        \item[{\rm (1.2)}]
             if $T_{i_2}\ne \emptyset$,
             then for each $i\in \{i_1, i_2\}$, $H$ contains a $(j, i_4)_{(u_i, u_{i_4})}$-path,
             and $i_4\in C(u_i)$, $T_i\subseteq C(u_{i_4})$;
        \item[{\rm (1.3)}]
             $i_3\in B_{i_1}\cup B_{i_2}$;
        \item[{\rm (1.4)}]
             if $T_{uv}\setminus (T_{i_1}\cup T_{i_2})\ne \emptyset$,
             then mult$_{S_u}(i_3)\ge 2$,
             and for each $i\in \{i_1, i_2\}$, $i_3\in C(u_i)$ or $H$ contains a $(i_3, a_i)_{(u_{i}, u_{a_i})}$-path for some $a_i\in C(u)$.
        \end{itemize}
   \item[{\rm (2)}]
        \begin{itemize}
        \parskip=0pt
        \item[{\rm (2.1)}]
             if recoloring some edges incident to $u$ reduces to an acyclic edge $(\Delta + 5)$-coloring $c'$ for $H$ such that $C(u)\cap C(v) = \{i_1, i_2\}$, where $(uu_a, uu_b)_{c'} = (i_1, i_2)$,
             and $a\ne i_1$, $b\ne i_2$,
             then $T_{u v}\subseteq B_i\subseteq C(v_i)$, and $d(v_i)\ge t_{u v} + 1$, $i = i_1, i_2$;
        \item[{\rm (2.2)}]
             if $(v v_{i_1}, v v_{i_2})\to (i_2, i_1)$ reduces to an acyclic edge $(\Delta + 5)$-coloring for $H$, then $T_{u v}\subseteq B_i\subseteq C(u_i)$, and $d(u_i)\ge t_{u v} + 1$, $i = i_1, i_2$.
         \end{itemize}
   \item[{\rm (3)}]
        if $R_{i}\ne \emptyset$ for some $i\in \{i_1, i_2\}$,
        then $i_1/i_2\in S_u$, and if $i_1\not\in C(u_3)$ and $H$ contains no $(i_1, i)_{(u_{i_3}, u_i)}$-path for each $i\in \{i_2, i_4\}$, then $i_3\in C(u_1)$
	or $H$ contains a $(i_1, i)_{(u_{i_3}, u_i)}$-path for some $i\in \{i_2, i_4\}$.
   \item[{\rm (4)}]
        for each $j\in (C(u_{i_1})\cap C(u_{i_2}))\setminus (C(u_{i_3})\cup C(u_{i_4}))$,
        \begin{itemize}
        \parskip=0pt
        \item[{\rm (4.1)}]
             if for each $i\in \{i_3, i_4\}$, $H$ contains no $(j, l)_{(u_i, u_l)}$-path,
             then $i\in B_{i_1}\cup B_{i_2}$,
             and for each $j\in \{i_1, i_2\}$, $i\in C(u_j)$, or there exists a $a_j\in C(u_j)$ such that $H$ contains a $(a_j, i)_{u_i, u_{a_j}}$-path;
        \item[{\rm (4.2)}]
             there exists an $i\in \{i_3, i_4\}$ such that $H$ contains no $(l, j)_{(u_i, u)}$-path for $l\in C(u)$, $i\in B_{i_1}\cup B_{i_2}$,
             and for each $j\in \{i_1, i_2\}$, $i\in C(u_j)$, or there exists a $a_j\in C(u_j)$ such that $H$ contains a $(a_j, i)_{u_i, u_{a_j}}$-path.
         \end{itemize}
   \item[{\rm (5)}]
        for each $j\in T_{i_1}\cup T_{i_2}$, mult$_{S_u}(j) \ge 3$;
        or if mult$_{S_u}(j) = 2$, and $j\in T_i$ for some $i\in \{i_1, i_2\}$,
        then there exists an $a_i\in \{i_3, i_4\}$,
        such that $H$ contains a $(a_i, l)_{(u_{i_1}, u_{a_i})}$-path for each $l\in T_i$ and $j\in T_i\subseteq C(u_{a_i})$;
   \item[{\rm (6)}]
        for each $j\in T_{i_2}\ne \emptyset$,
        if $H$ contains a $(a_{i_2j}, j)_{(u_{i_2}, u_{a_{i_2j}})}$-path,
        then $a_{i_2j}\in C(v_{i_1})$ or $H$ contains a $(a_{i_2j}, b_{i_1})$-path for some $b_{i_1}\in C(v)$;
   \item[{\rm (7)}]
        assume $T_{i_1}\ne\emptyset$, $(T_{i_2}\cup \{i_3\})\setminus C(u_{i_2})\ne \emptyset$.
        \begin{itemize}
        \parskip=0pt
        \item[{\rm (7.1)}]
             If $i_3\in C(u_{i_1})$, then $i_4\in C(u_{i_1})$, or $\{i_1\}\cup V_{uv}\subseteq S_u$;
        \item[{\rm (7.2)}]
             If $i_1\not\in C(u_{i_3})$ and $H$ contains no $(i, i_1)_{(u_{i_3}, u_i)}$-path for each $i\in C(u)$, then $i_4\in C(u_{i_1})$ and $H$ contains a $(i_4, j)_{(u_{i_1}, u_{i_4})}$-path for each $j\in T_{i_1}$.
         \end{itemize}
   \item[{\rm (8)}]
        if $T_i\ne \emptyset$ and $H$ contains a $(l, j)(v_{i_1}, v_l)$-path for some $j\in T_{i_1}\cap R_{i_1}$,
        $l\in V_{uv}$, then $l\in C(u_{i_2})$ or $H$ contains a $(l, i)_{(u_{i_2}, u_i)}$-path for some $i\in C(u)$.
   \item[{\rm (9)}]
        If $T_{i_1}\setminus T_0 \ne \emptyset$,
        then $i_3\in C(u_{i_1})$, or $i_3\in C(u_{i_2})\cap C(u_{i_4})$.
\end{itemize}
\end{corollary}
\begin{proof}
For (1), if $H$ contains no $(i_4, j)_{(u_{i_1}, u_{i_4})}$-path for some $j\in T_{i_1}$,
then $(uu_{i_3}, uu_{i_1})\to (i_1, j)$ and $uv\overset{\divideontimes}\to (T_{i_2}\cup \{i_3\})\setminus C(u_{i_2})$. Otherwise, (1.1) holds.
If $T_{i_2}\ne \emptyset$, then (7.2) holds naturally
by (1.1), $H$ contains a $(j, i_4)_{(u_{i_1}, u_{i_4})}$-path for each $j\in T_{i_1}$.

If $i_3\not\in B_{i_1}\cup B_{i_2}$, then $(uu_{i_3}, uv)\to (5, i_3)$.
Otherwise, (1.3) holds.
One sees that $u u_{i_3}\to T_{u v}\setminus (T_{i_1}\cup T_{i_2})$ cannot obtain a bichromatic cycle.
It follows from Corollary~\ref{auv2}(3) that (1.4) holds.

For (2), let $j\in T_{uv}$.
Since $j\in B_{i_1}\cup B_{i_2}$, assume w.l.o.g. $j\in B_{i_1}$,
i.e., $H$ contains a $(j, i_1)_{(u_{i_1}, v_{i_1})}$-path and then $j\in C(v_{i_1})$.
One sees from Lemma~\ref{lemma06} that $H$ contains no $(j, i_1)_{(u_a, v_{i_1})}$-path.
If $H$ contains no $(i_2, j)_{(u_b, v_{i_2})}$-path,
then on the basic coloring of $c'$, $uv\overset{\divideontimes}\to j$.
Otherwise, $H$ contains a $(i_2, j)_{(u_b, v_{i_2})}$-path, $j\in C(v_{i_2})$.
Hence, $T_{u v}\subseteq C(v_i)$, $i = i_1, i_2$, and (2.1) holds.
Likewise, (2.2) holds.

For (3), if $i_1\not\in C(u_3)$, $H$ contains no $(i_1, i)_{(u_{i_3}, u_i)}$-path for each $i\in \{i_2, i_4\}$,
and $i_3\not\in C(u_1)$, $H$ contains no $(i_1, i)_{(u_{i_3}, u_i)}$-path for each $i\in \{i_2, i_4\}$,
then $(u u_{i_1}, u u_{i_2})\to (i_2, i_1)$ is still an acyclic edge $(\Delta + 5)$-coloring for $H$,
and by (2.1), for each $i\in \{i_1, i_2\}$, $R_i = \emptyset$, a contradiction.

It follows Corollary~\ref{auv2} (3) that (4.1) holds.
Since $j\in B_{i_1}\cup C(u_{i_2})$, $H$ contains a $(i, j)_{(u_i, v_i)}$-path for some $i\in \{i_1, i_2\}$.
By Lemma~\ref{lemma06}, there exists an $i\in \{i_3, i_4\}$ such that $H$ contains no $(l, j)_{(u_i, u)}$-path for $l\in C(u)$,
and by Corollary~\ref{auv2} (3), (4.2) holds naturally.

It follows from Corollary~\ref{auv2} (1.2), (5.2) that (5) holds.

For (6), if for some $j\in T_{i_2}\ne \emptyset$,
$a_{i_2j}\not\in C(v_{i_1})$ and $H$ contains no $(a_{i_2j},l)$-path for each $l\in C(v)$,
then $(v v_{i_1}, u v)\overset{\divideontimes}\to (a_{i_2j}, j)$.
Otherwise, (6) holds.

For (7), if $i_4\not\in C(u_{i_1})$ and $i_1\not\in S_u$,
then $(u u_{i_3}, u u_{i_1})\to (i_1, T_{i_1})$,
and $u v\overset{\divideontimes}\to (T_{i_2}\cup \{i_3\})\setminus C(u_{i_2})$.
If there exist an $i\in V_{u v}\setminus S_u$,
one sees that the same argument applies if $uu_{i_1}\to i$,
then an acyclic edge $(\Delta + 5)$-coloring for $H$ can be obtained since $i\not\in S_u$.
Otherwise, $V_{u v}\subseteq S_u$ and (7.1) holds.

If $H$ contains no $(i_4, j)_{(u_{i_1}, u_{i_4})}$-path for some $j\in T_{i_1}$,
then $(u u_3, u u_1)\to (i_1, j)$ and $uv\overset{\divideontimes}\to (T_{i_2}\cup \{i_3\})\setminus C(u_{i_2})$.
Otherwise, (7.2) holds.

For (8),
if $l\not\in C(u_{i_2})$ and $H$ contains no $(l, i)_{(u_{i_2}, u_i)}$-path for each $i\in C(u)$,
since $H$ contains a $(i_2, j)_{(u_{i_2}, v_{i_2})}$-path and a $(l, j)_{(v_{i_1}, v_{l})}$-path,
$(u u_{i_2}, u v)\overset{\divideontimes}\to (l, j)$.
Otherwise, (8) holds.

For (9), if $i_4\not\in C(u_{i_1})$,
one sees that $i_3\in C(u_{i_1})$ and $H$ contains a $(i_3, T_{i_1})_{(u_{i_1}, u_{i_3})}$-path.
It follows from Corollary~\ref{auv2}(5.2) and (4) that $i_4\in B_{i_2}$, mult$_{S_u}(i_4)\ge 2$, and (9) holds.
\end{proof}

By Corollary~\ref{auv2d(u)5}(6), we proceed with the following proposition:
\begin{proposition}
\label{auv2Cvi1vi2}
If $T_{i_2}\ne \emptyset$ and $C(v_{i_1})\setminus \{i_1\}\subseteq T_{u v}$,
or $T_{i_1}\ne \emptyset$ and $C(v_{i_2})\setminus \{i_2\}\subseteq T_{u v}$,
then an acyclic edge $(\Delta + 2)$-coloring for $G$ can be obtained in $O(1)$ time.
\end{proposition}

\begin{corollary}
\label{Ti1i4Ti2i3}
If $d(u) = 5$, $a_{u v} = 2$, and $T_{i_1}\setminus C(u_{i_4})\ne \emptyset$, $T_{i_2}\setminus C(u_{i_3})\ne \emptyset$,
then either an acyclic edge $(\Delta + 5)$-coloring for $H$ can be obtained such that $|C(u)\cap C(v)| \le 1$,
or $H$ contains a $(T_{i_1}\cap T_{i_4}, T_{i_2}\cap T_{i_3})_{(u_{i_1}, u_{i_3})}$-path,
and a $(T_{i_2}\cap T_{i_3}, T_{i_1}\cap T_{i_4})_{(u_{i_2}, u_{i_4})}$-path.
\end{corollary}
\begin{proof}
One sees that $H$ contains a $(i_3, T_{i_1})_{(u_{i_1}, u_{i_3})}$-path and a $(i_4, T_{i_2})_{(u_{i_2}, u_{i_4})}$-path.
If $H$ contains no $(\alpha, \beta)_{(u_{i_1}, u_{i_3})}$-path for some $\alpha\in T_{i_1}\cap T_{i_4}$,
$\beta\in T_{i_2}\cap T_{i_3}$, then $(uu_{i_1}, uu_{i_3})\to (\alpha, \beta)$ reduces to an acyclic edge $(\Delta + 5)$-coloring for $H$ such that $C(u)\cap C(v) = \{i_2\}$.
Otherwise, $H$ contains a $(T_{i_1}\cap T_{i_4}, T_{i_2}\cap T_{i_3})_{(u_{i_1}, u_{i_3})}$-path,
and a $(T_{i_2}\cap T_{i_3}, T_{i_1}\cap T_{i_4})_{(u_{i_2}, u_{i_4})}$-path.
\end{proof}

We prove the inductive step by first treating $u_3, u_4$ indistinguishably to obtain an acyclic edge $(\Delta + 2)$-coloring for $G$,
until impossible, by then (after we discuss the detailed configurations) to distinguish $u_3$ and $u_4$.
Hereafter, for each $i, j\in [1, 4]$, let $C_{ij} = T_{u v}\cap C(u_i)\cap C(u_j)$, $c_{i j} = |C_{i j}|$.

Below we distinguish two subcases where $d(v) = 5$ and $6$, respectively.


{\bf Case 1.} $d(v) = 5$.
Then $t_{u v} = \Delta - 1$, $V_{u v} = [5, 6]$.

Since $T_{u v}\subseteq C(u_{i_1})\cup C(u_{i_2})$,
one sees that if $\{i, c(v u_i)\} = \{i_1, i_2\}$, then $T_i\ne \emptyset$,
thus, for each $i\in [3, 4]$,
if $i\in A_{u v}$, then $c(v u_i)\not\in A_{uv}$,
and it follows from Corollary~\ref{auv2}(1.2) that there exists a $a_i\in \{i_3, i_4\}\cap C(u_i)$ such that $H$ contains a $(a_i, j)_{(u_i, u)}$-path for some $j\in T_i$;
if $c(v u_i) = l\in A_{u v}$, then $i\not\in A_{u v}$,
and it follows from Lemma~\ref{auvge2}(1.2) that there exists a $b_l\in V_{u v}\cap C(u_i)$ such that $H$ contains a $(b_l, j)_{(u_i, v)}$-path for some $j\in T_i$.

One sees from Eq.~(\ref{eq16}) that
if $\sum_{x\in N(u)\setminus \{v\}}d(x)\ge 2\Delta + 13$,
then $\min\{d(u_3), d(u_4)\} = 6$, or $\min\{d(u_3), d(u_4)\} = \Delta = 7$;
and if $\Delta\ge 8$ and $\sum_{x\in N(u)\setminus \{v\}}d(x)\ge \sum_{i\in [1, 4]}w_i + 2t_{u v}\ge 2\Delta + 8$,
then $\Delta = 8$, $\sum_{x\in N(u)\setminus \{v\}}d(x) = 2\Delta + 8$,
$d(u_3) = d(u_4) = 8$, $d(u_1) = d(u_2) = 6$,
and for each $j\in T_{u v}$, mult$_{S_u}(j) = 2$, i.e., $|T_0| = 0$.
By symmetry, we consider the following three cases.

\begin{figure}[h]
\begin{center}
\includegraphics[width=5.5in]{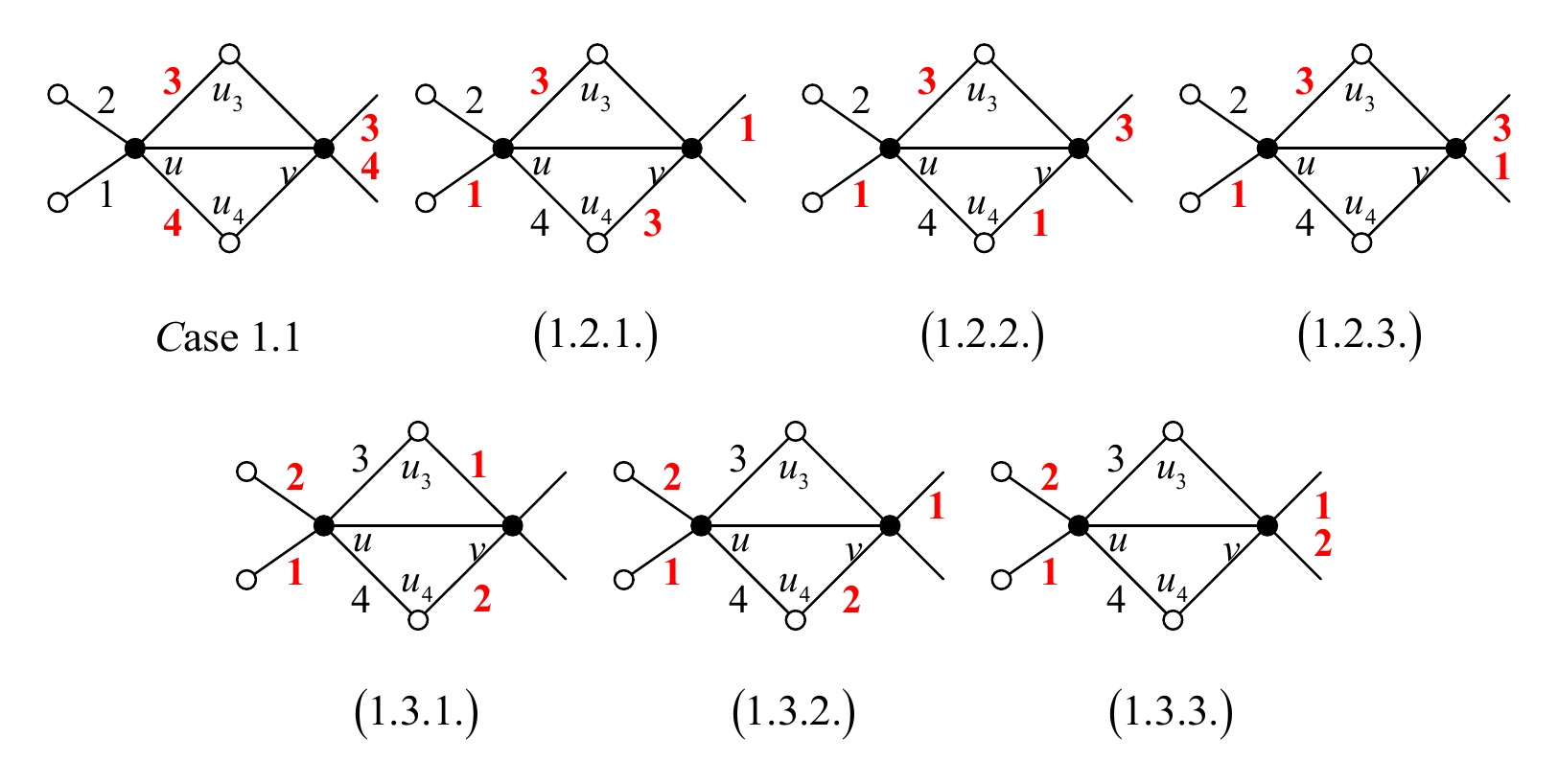}
\caption{Coloring of ($A_7$): $|C(u)\cap C(v)| = 2$ and $C(v) = \{i_1, i_2, 5, 6\}$.\label{fig08}}
\end{center}
\end{figure}

Case 1.1. $\{i_1, i_2\} = \{3, 4\}$.
It follows that $(v u_3, v u_4)_c = (5, 6)$,
$T_{u v}\subseteq C(u_3)\cup C(u_4)$,
and there exists a $a_3\in \{1, 2\}\cap C(u_3)$ and a $a_4\in \{1, 2\}\cap C(u_4)$.
Let $C(u_3) = W_3\cup T_4\cup C_{3 4}$, $C(u_4) = W_4\cup T_3\cup C_{3 4}$.

Note that $H$ contains a $(4, j)_{(u_4, v_4)}$-path for each $j\in T_3$ and a $(3, l)_{(u_3, v_3)}$-path for each $l\in T_4$.
It follows from Corollary~\ref{auv2}(4) that if $H$ contains no $(6, j)_{(u_3, u_4)}$-path for some $j\in T_3$,
then $5\in B_4$ and mult$_{S_u\setminus C(u_3)}(5)\ge 2$,
and if $H$ contains no $(5, j)_{(u_3, u_4)}$-path for some $j\in T_4$,
then $6\in B_3$ and mult$_{S_u\setminus C(u_4)}(6)\ge 2$.
Hence, by symmetry, assume w.l.o.g. $6\in C(u_3)$.

Note that $t_3\ge 3$, $t_4\ge 2$, and $5\in C(u_4)$ or $6\in C(u_1)\cup C(u_2)$.
For each $i = i_3\in [1, 2]$ with $d(u_i)\le 6$,
we first show that $(T_3\cup T_4)\setminus C(u_i)\ne \emptyset$.
One sees that if $t_3 + t_4\ge 6$, then $(T_3\cup T_4)\setminus C(u_i)\ne \emptyset$.
Now assume $t_3 = 3$, $t_4 = 2$.
It follows $W_3 = \{3, 5, 6, a_3\}$ and $W_4 = \{4, 6, a_4\}$.
If $C(u_i) = \{i\}\cup T_3\cup T_4$,
then it follows from Corollary~\ref{auv2d(u)5}(1.2) that $i_4\in C(u_3)\cap C(u_4)$,
and from Corollary~\ref{auv2d(u)5}(1.3) that $i\in C(u_3)\cup C(u_4)$.
Otherwise, $(T_3\cup T_4)\setminus C(u_i)\ne \emptyset$.
Next, it follows from Corollary~\ref{auv2} (5.1) that
there exist a $l_i\in [3, 4]$ such that $i_4\in C(u_l)$, $l\in S_u$,
and from Corollary~\ref{auv2}(5.2) that $H$ contains a $(i_4, j)_{(u_l, u_{i_4})}$-path for each $j\in T_l$, $T_l\subseteq C(u_{i_4})$.

For each $i\in [1, 2]$, let $\{i, j\} = [1, 2]$,
if $j, 3, 4, 5\not\in C(u_i)$, then $(u u_i, u v)\overset{\divideontimes}\to (5, T_3)$;
otherwise, $j/3/4/5\in C(u_i)$ and then $T_i\ne \emptyset$.
One sees that $d(u_3) + d(u_4)\ge t_{uv} + |\{3, 5, 6, a_3\}| + |\{4, 6, a_4\}| = \Delta + 6$.
It follows that $\min\{d(u_3), d(u_4)\}\ge 6$.

(1.1.1) $\min\{d(u_3), d(u_4)\} = 6$.
Then $C(u_3) = \{3, 5, 6, a_3\}\cup T_4$, $C(u_4) = \{4, 6, a_4\}\cup T_3$,
and $T_4\cup T_3 = T_{u v}$, $6\in C(u_1)\cup C(u_2)$.
It follows from Corollary~\ref{auv2}(5.1) that $3, 4\in S_u$,
and from Corollary~\ref{auv2}(5.2) that for each $i\in [1, 2]$, mult$_{S_u}(i)\ge 2$.
Assume w.l.o.g. $C(u_3) = \{3, 5, 6, 2\}\cup T_4$, $C(u_4) = \{4, 6, 1\}\cup T_3$, and $1, 2, 3, 4, 6\in C(u_1)\cup C(u_2)$.
Since $\sum_{x\in N(u)\setminus \{v\}}d(x)\ge \sum_{i\in [1, 4]}w_i + 2t_{u v} + t_0\ge |\{1\}| + |\{2\}| + |\{3, 5, 2, 6\}| + |\{4, 6, 1\}| + |\{1, 2, 3, 4, 6\}| + 2t_{u v} + t_0 = 2\Delta + 12 + t_0$,
we have mult$_{S_u\setminus C(u_3)}(5)\le 1$.

(1.1.2) $\min \{d(u_3), d(u_4)\}\in [7, 8]$.
One sees that $\min\{d(u_1), d(u_2)\}\le 6$.
Assume w.l.o.g. $d(u_{i_3})\le 6$ and $i_1\in [3, 4]$ with $i_4\in C(u_{i_1})$, $T_{i_1}\subseteq C(u_{i_4})$.
It follows that $1, 2\in C(u_3)\cup C(u_4)$ and mult$_{S_u}(i_3)\ge 2$, $i_1\in S_u$.
One sees that $T_{i_2}\subseteq C(u_1)\cap C(u_2)$, or $i_2\in S_u$ and mult$_{S_u}(i_3)\ge 2$.
Then $\sum_{x\in N(u)\setminus \{v\}}d(x)\ge \sum_{i\in [1, 4]}w_i + 2t_{u v} + t_0\ge |\{1\}| + |\{2\}| + |\{3, 5, 6\}| + |\{4, 6\}| + |\{5/6\}| + |\{i_3, i_1, 1, 2\}| + 2t_{u v} + t_0 = 2\Delta + 10 + t_0$,
and if $T_{i_2}\subseteq C(u_1)\cap C(u_2)$,
then $\sum_{x\in N(u)\setminus \{v\}}d(x)\ge 2\Delta + 10 + t_0\ge 2\Delta + 10 + t_{i_2}\ge 2\Delta + 12$,
or otherwise, $i_2\in S_u$ and mult$_{S_u}(i_3)\ge 2$,
then $\sum_{x\in N(u)\setminus \{v\}}d(x)\ge 2\Delta + 10 + |\{i_2, i_3\}|\ge 2\Delta + 12$.
It follows that if $5\in C(u_4)$, then $\{5, 6\}\setminus (C(u_1)\cup C(u_2))\ne \emptyset$ and assume w.l.o.g. $5\not\in C(u_1)\cup C(u_2)$,
or if $5\not\in C(u_4)$, then $6\in C(u_1)\cup C(u_2)$ and mult$_{S_u\setminus C(u_3)}(5)\le 1$.

Hence, for (1.1.1) or (1.1.2), assume w.l.o.g. mult$_{S_u\setminus C(u_3)}(5)\le 1$.
It follows that for each $j\in T_3$, $H$ contains a $(6, j)_{(u_3, u_4)}$-path,
and from Lemma~\ref{auvge2}(1.2) that $T_3\subseteq C(v_3)$.
Hence, $\{3\}\cup T_3\cup T_4\subseteq C(v_3)$.
It follows from \ref{auv2Cvi1vi2} that $T_{u v}\setminus (T_3\cup T_4)\ne \emptyset$ and $C_{3 4}\ne \emptyset$.
One sees that if $\min\{d(u_3), d(u_4)\} = 6$, then $T_{u v} = T_3\cup T_4$.
It suffices to consider the following two scenarios:
\begin{itemize}
\parskip=0pt
\item
    $5\in C(u_4)$ or $c_{3 4}\ge 2$.
    One sees that if $5\in C(u_4)$,
    $d(u_3) + d(u_4)\ge |\{3, 5, 6, a_3\}| + |[4, 6]\cup \{a_4\}| + t_{u v} + c_{3 4} = \Delta + 7 + c_{3 4}\ge \Delta + 8$;
    or if $c_{3 4}\ge 2$, then $d(u_3) + d(u_4)\ge |\{3, 5, 6, a_3\}| + |\{4, 6, a_4\}| + t_{u v} + c_{3 4} = \Delta + 6 + c_{3 4}\ge 2\Delta + 8$.
    Hence, $\min\{d(u_3), d(u_4)\} = 8$, $|C(u_3)\cap [1, 2]| = 1$,
    $|C(u_4)\cap [1, 2]| = 1$, $W_3 = \{3, 5, 6, a_3\}$, $W_4\setminus \{5\} = \{4, 6, a_4\}$,
    and $d(u_1) = d(u_2) = 6$.
    Since $1, 2\in C(u_3)\cup C(u_4)$, assume w.l.o.g. $a_3 = 2$, $a_4 = 1$.
    Then $1\in C(u_2)$, $2\in C(u_1)$.
    It follows from $3, 4\in C(u_1)\cup C(u_2)$ that $5\not\in C(u_1)\cup C(u_2)$ and $T_3\cap C(u_1) = \emptyset$.
    If $4\not\in C(u_2)$, then $(u u_2, u u_3, u v)\overset{\divideontimes}\to (5, T_3, 2)$\footnote{$(u u_2, u u_3, u v)\overset{\divideontimes}\to (5, T_3, 2)$ means $u u_2\to 5$, $u u_3\to T_3$, and $u v\to 2$. }.
    Otherwise, $4\in C(u_2)\setminus C(u_1)$ and likewise $3\in C(u_1)\setminus C(u_2)$.
    It follows from Corollary~\ref{auv2}(5.1) that $H$ contains a $(1, 3)_{(u_1, u_2)}$-path.
    Then $(u u_2, u u_3, u u_4, u v)\overset{\divideontimes}\to (5, T_3, 3, 2)$.

\item
    $5\not\in C(u_4)$, and $c_{3 4} = 1$.
    It follows that $C(v_3) = \{3, 6\}\cup T_3\cup T_4$.
    By \ref{auv2Cvi1vi2}, $H$ contains a $(6, a_4)_{(v_3, u_4)}$-path,
    and a $(a_4, j^*)_{(u_4, u_{a_4})}$-path for some $j^*\in T_4$.
    Then $(v v_3, v u_4)\to (a_4, j^*)$ reduces to an acyclic edge $(\Delta + 5)$-coloring for $H$ such that $C(u)\cap C(v) = \{4\}$.
\end{itemize}

Case 1.2. $\{i_1, i_2\} = \{1, 3\}$.
Assume w.l.o.g. $c(v u_3) = 5$.
If $5\not\in C(u_1)$ and $H$ contains no $(5, i)_{(u_1, u_i)}$-path for each $i\in [2, 4]$,
$(u u_1, u v)\overset{\divideontimes}\to (5, T_3)$.
Otherwise, $5/2/3/4\in C(u_1)$, and if $5\not\in C(u_1)$,
then $H$ contains a $(5, i)_{(u_1, u_i)}$-path for some $i\in [2, 4]$.
It follows that $T_1\ne\emptyset$, and then $2/4\in C(u_1)$.

(1.2.1) $c(v u_4) = 3$. When there exists a $j\in T_{u v}\setminus (C(u_3)\cup C(u_4))$,
say $7\not\in C(u_3)\cup C(u_4)$,
it follows that $2\in C(u_3)$, $6\in C(u_4)$,
and $H$ contains $(2, 7)_{(u_2, u_3)}$-path, a $(6, 7)_{(u_4, v_6)}$-path.
It follows from Corollary~\ref{auv2}(4) that $5\in C(u_1)$ and mult$_{S_u\setminus C(u_3)}(5)\ge 2$,
and from Corollary~\ref{auv2}(3) or (5.2) that $4\in C(u_1)$ and mult$_{S_u\setminus C(u_3)}(4)\ge 2$.
If $3\not\in C(u_1)\cup C(u_2)$, then $(u u_2, u u_3, u v)\overset{\divideontimes}\to (3, 7, T_1)$.
Otherwise, $3\in C(u_1)\cup C(u_2)$.
Since from Corollary~\ref{auv2d(u)5} (7.1) that $2\in C(u_1)$ or $1\in S_u$,
$\sum_{x\in N(u)\setminus \{v\}}d(x)\ge \sum_{i\in [1, 4]}w_i + 2t_{u v} + t_0\ge |\{1, 4, 5\}| + |\{2\}| + |\{3, 5, 2\}| + |\{3, 4, 6\}| + |\{3, 4, 5, 1/2\}| + 2t_{u v} + t_0 = 2\Delta + 12 + t_0$.
One sees from $1, 4, 5\in C(u_1)$ that $t_1\ge 2$.
Hence, for some $j\in T_1$, mult$_{S_u\setminus C(u_3)}(j) = 2$ and $H$ contains a $(4, j)_{(u_1, u_4)}$-path.
It follows from Corollary~\ref{auv2}(5.2) that $2\in C(u_1)\cup C(u_4)$ and from Corollary~\ref{auv2}(5.1) that $1\in S_u$.
Thus, $\sum_{x\in N(u)\setminus \{v\}}d(x)\ge \sum_{i\in [1, 4]}w_i + 2t_{u v} + t_0\ge |\{1, 4, 5\}| + |\{2\}| + |\{3, 5, 2\}| + |\{3, 4, 6\}| + |\{3, 4, 5, 1, 2\}| + 2t_{u v} + t_0 = 2\Delta + 13 + t_0$.
It follows that $\sum_{x\in N(u)\setminus \{v\}}d(x) = 2\Delta + 13$ and $6\not\in S_u\setminus C(u_4)$.
If $H$ contains no $(4, 6)_{(u_1, u_4)}$-path, then $(u u_1, u v)\overset{\divideontimes}\to (6, 7)$.
Otherwise, $H$ contains a $(4, 6)_{(u_1, u_4)}$-path.
Then $(u u_3, u v)\to (6, T_1)$ or $(u u_2, u u_3, u v)\overset{\divideontimes}\to (6, 7, T_1)$%
\footnote{If $H$ contains no $(6, j)_{(u_3, v)}$-path for some $j\in T_1$, then $(u u_3, u v)\overset{\divideontimes}\to (6, j)$. Otherwise, $H$ contains a $(6, j)_{(u_3, v)}$-path for each $j\in T_1$,
then $(u u_2, u u_3, uv)\overset{\divideontimes}\to (6, 7, T_1)$. }.

In the other case, $T_{u v} \subseteq C(u_3)\cup C(u_4)$.
Since $T_1\subseteq C(u_3)\cap C(u_4)$, $c_{3 4}\ge t_1$ and $d(u_3) + d(u_4)\ge w_3 + w_4 + t_{uv} + c_{3 4}\ge w_3 + w_4 + t_{u v} + t_1$.
Recall that $T_{u v} \subseteq B_1\cup B_3$.

When $w_3 + w_4\ge 8$, $d(u_3) + d(u_4)\ge w_3 + w_4 + t_{uv} + t_1 = 8 + \Delta - 1 + t_1\ge \Delta + 7 + 1 = \Delta + 8$.
It follows that $c_{34} = t_1 = 1$, $\{d(u_3), d(u_4)\} = \{8, \Delta\}$,
and $C(u_1) = \{1, a_1\}\cup (T_{u v}\setminus T_1)$.
However, $d(u_1) = \Delta - 1 - 1 + 2\ge 8> 6$, a contradiction.

When $w_3 + w_4 = 7$ and $t_1\ge 2$,
$d(u_3) + d(u_4)\ge 7 + t_{u v} + t_1 = 7 + \Delta - 1 + t_1\ge \Delta + 6 + 2 = \Delta + 8$.
It follows that $c_{34} = t_1 = 2$, $\{d(u_3), d(u_4)\} = \{8, \Delta\}$,
and $(T_{u v}\setminus T_1)\cup \{1, a_1\}\subseteq C(u_1)$.
However, $d(u_1)\ge 2 + \Delta - 1 - 2 = \Delta - 1\ge 8 - 1 = 7 > 6$, a contradiction.

Hence, $w_3 + w_4\le 7$, and if $w_3 + w_4 = 7$, then $t_1 = 1$.
Since $6\not\in C(u_3)$ or $2\not\in C(u_4)$,
One sees that if $2\not\in C(u_4)$, then it follows from Corollary~\ref{auv2}(3) that $4\in C(u_1)$,
and if $6\not\in C(u_3)$, then it follows from Corollary~\ref{auv2}(4) that $5\in C(u_1)$.
\begin{itemize}
\parskip=0pt
\item
     $W_3 = \{3, 5, 2/4, 6\}$ and $W_4 = \{3, 4, 5/6\}$.
     Since $1, 4\in C(u_1)$ and $t_1 = 1$, $C(u_1) = \{1, 4\}\cup (T_{u v}\setminus \{7\})$.
     Then $(u u_4, u u_1, u v)\overset{\divideontimes}\to (1, 7, T_3)$.
\item
     $W_3 = \{3, 5, 2/4\}$ and $W_4 = \{3, 4, 5/6, 2\}$.
     Since $1, 2/4, 5\in C(u_1)$, $t_1\ge 2$, a contradiction.
\item
     $W_3 = \{3, 5, 2/4\}$ and $W_4 = \{3, 4, 5/6\}$.
     Then $4, 5\in C(u_1)$.
     If $t_1\ge 3$, since $d(u_3) + d(u_4)\ge 6 + t_{u v} + t_1 \ge \Delta + 8$,
     it follows that $c_{3 4} = t_1 = 3$, $\{d(u_3), d(u_4)\} = \{8, \Delta\}$, and $(T_{uv}\setminus T_1)\cup \{1, a_1\}\subseteq C(u_1)$,
     however, $d(u_1)\ge 3 + \Delta - 1 - 3 = \Delta - 1\ge 8 - 1 = 7 > 6$, a contradiction.
     Otherwise, $t_1 = 2$ and $C(u_1) = \{1, 4, 5\}\cup (T_{u v}\setminus \{7, 8\})$.
     Then $(u u_4, u u_1, u v)\overset{\divideontimes}\to (1, 7, T_3)$.
\end{itemize}

(1.2.2.) $c(v u_4) = 1$.
One sees that $T_{u v}\subseteq C(u_3)\cup C(u_4)$.
Let $C(u_3) = W_3\cup T_4\cup C_{3 4}$, $C(u_4) = W_4\cup T_3\cup C_{3 4}$.
If $4\not\in C(v_3)$ and $H$ contains no $(4, i)_{(v_3, v_i)}$-path for each $i\in \{1, 5, 6\}$,
$(v v_3, u v)\overset{\divideontimes}\to (4, T_4)$.
Otherwise, $4/1/5/6\in C(v_3)$, and if $4\not\in C(v_3)$,
then $H$ contains a $(4, i)_{(v_3, v_i)}$-path for some $i\in \{1, 5, 6\}$.
It follows that $R_3\ne\emptyset$, and then $5/6\in C(v_3)$.

When $5\not\in B_1$, it follows that from Corollary~\ref{auv2} (4) that for each $j\in T_3$,
$H$ contains a $(6, j)_{(u_3, v_6)}$-path, $6\in C(u_3)$,
and from Lemma~\ref{auvge2}(1.2) that $T_3\subseteq C(v_3)$.
Hence, $\{3\}\cup T_3\cup T_4\subseteq C(v_3)$.
It follows from $R_3\ne \emptyset$ that $C_{3 4}\ne \emptyset$.
Since $d(u_3) + d(u_4) = w_3 + w_4 + t_{u v} + c_{3 4}\ge |\{3, 5, 6, a_3\}| + |\{4, 1, b_1\}| + \Delta - 1 + c_{3 4} = \Delta + 6 + c_{3 4}$,
$1\le c_{34}\le 2$ and if $c_{3 4} = 2$, then $d(u_1) = d(u_2) = 6$.
Assume w.l.o.g. $7\in C_{3 4}\setminus C(v_3)$.
Then $7\in C(u_1)$.
It suffices to consider the following two scenarios:
\begin{itemize}
\parskip=0pt
\item
   $c_{3 4} = 1$, i.e., $T_{u v} = T_3\cup T_4\cup \{7\}$.
   It follows that $C(v_3) = \{3, 5/6\}\cup T_3\cup T_4$ and $T_3\cup \{7\}\subseteq C(u_1)$.
   One sees that $\emptyset\ne T_1\subseteq T_4$.
   Then $(u u_4, u v)\overset{\divideontimes}\to (T_1, 4)$.
\item
   $c_{3 4} = 2$, i.e., $T_{u v} = T_3\cup T_4\cup \{7, 8\}$.
   Since $d(u_1) = 6$,
   it follows that $C(u_1) = T_3\cup \{1, 2/4, 7\}$ and $C(v_3) = \{3, 5/6, 8\}\cup T_3\cup T_4$.
   One sees that $T_1 = T_4\cup \{8\}$.
   Then $(u u_4, u v)\overset{\divideontimes}\to (T_4, 4)$.
\end{itemize}

In the other case, $5\in B_1$ and $5\in C(u_1)\cap C(u_4)$.
When $T_1\setminus C(u_4)\ne\emptyset$,
i.e., there exists a $j^*\in T_1\setminus C(u_4)$ such that $H$ contains a $(2, j^*){(u_1, u_2)}$-path,
it follows from Corollary~\ref{auv2}(5.2) that $4\in B_3$, $4\in C(u_2)\cup C(u_1)$,
from Corollary~\ref{auv2}(5.1) that $6\in S_u$,
and from Corollary~\ref{auv2d(u)5} (7.1) that $2\in C(u_3)$ or $3\in S_u$.
If $1\not\in C(u_2)\cup C(u_3)$, then $(u u_2, u u_1, u v)\overset{\divideontimes}\to (1, j^*, T_3)$.
Otherwise, $1\in C(u_2)\cup C(u_3)$.
Then $\sum_{x\in N(u)\setminus \{v\}}d(x)\ge \sum_{i\in [1, 4]}w_i + 2t_{u v} + t_0\ge |\{1, 2, 5\}| + |\{2\}| + |\{3, 4, 5\}| + |\{1, 4, 5\}| + |\{4, 1, 6, 2/3\}| + 2t_{u v} + t_0 = 2\Delta + 12 + t_0$.
One sees from $t_3\ge 2$ that $T_3\setminus C(u_2)\ne\emptyset$.
It follows from Corollary~\ref{auv2}(5.2) that $2\in B_1\cup B_2$, $2\in C(u_3)\cup C(u_4)$,
from Corollary~\ref{auv2}(5.1) that $3\in S_u$,
and $\sum_{x\in N(u)\setminus \{v\}}d(x)\ge \sum_{i\in [1, 4]}w_i + 2t_{u v} + t_0\ge |\{1, 2, 5\}| + |\{2\}| + |\{3, 4, 5\}| + |\{1, 4, 5\}| + |\{4, 1, 6, 2, 3\}| + 2t_{u v} + t_0 = 2\Delta + 13 + t_0$.
Hence, $\sum_{x\in N(u)\setminus \{v\}}d(x) = 2\Delta + 13$,
for each $i\in [1, 6]\setminus \{5\}$, mult$_{S_u}(i) = 1$;
for each $j\in T_{u v}$, mult$_{S_u}(j) = 2$,
and $\min\{d(u_3), d(u_4)\} = 6$, or $\Delta = \min\{d(u_3), d(u_4)\} = 7$.

Recall that $4\in C(u_3)\cap C(v_3)$ and if $2\in B_3$, then $2\in C(u_3)\cap C(v_3)$,
or if $2\in B_1$, then $2\in C(u_4)\setminus C(u_3)$.
Since $d(u_3) + d(u_4)\ge |[3, 5]| + |\{1, 4, 5\}| + |\{2\}| + t_{uv} + c_{3 4} = \Delta + 6 + c_{3 4}$,
$c_{3 4}\le 1$.
One sees from $3, 4, 5/6\in C(v_3)$ that $R_3\ge 2$.
It follows that $T_3\setminus C(v_3)\ne \emptyset$.
Assume w.l.o.g. $7\in T_3\cap R_3$ and $H$ contains a $(6, 7)_{(v_3, v_6)}$-path.
It follows from Corollary~\ref{auv2}(5.2) that $H$ contains a $(6, j)_{(v_3, v_6)}$-path for each $j\in R_3$,
$R_3\in C(v_6)$ and then $T_{u v}\subseteq C(v_1)\cup C(v_6)$.
One sees from $H$ contains no $(6, 7)_{(u_3, v_6)}$-path that $5\in C(v_3)\cup C(v_6)$.
If $2\not\in C(v_3)\cup C(v_6)$, then $(v v_3, u v)\overset{\divideontimes}\to (2, j^*)$.
If $1, 3\not\in C(v_6)$, then $(v v_6, v v_3, u v)\overset{\divideontimes}\to (3, 7, T_4)$.
Otherwise, $2\in C(v_3)\cup C(v_6)$ and $1/3\in C(v_6)$.
It follows from Corollary~\ref{auv2}(4) that $T_4\subseteq C(v_6)$.
Then $d(v_3) + d(v_6)\ge |\{3, 4, 5/6\}| + |\{6, 1/3\}| + |\{2, 5\}| + t_{uv} + t_4 = \Delta + 6 + t_4\ge \Delta + 8$.
Thus, $\min\{d(u_3), d(u_4)\} = 6$.
Since $\Delta + 6 + t_4\le 2\Delta$, $t_4\le \Delta - 6$ and then $t_3\ge 5$, $d(u_4)\ge |\{1, 4, 5\}| + t_3 \ge 8$.
It follows that $d(u_3) = 6$, $W_3 \biguplus W_4 = [3, 4, 5]\cup \{1, 4, 5\}\cup \{2\}$.

One sees from Corollary~\ref{auv2} (5.1) that $H$ contains a $(3, i)_{(u_4, u_i)}$-path for some $i\in [1, 2]$.
If $3\not\in C(u_1)$, one sees that $2\in C(u_4)\setminus C(u_3)$ and $H$ contains a $(2, 3)_{(u_4, u_2)}$-path,
then $(u u_1, u u_3, u v)\overset{\divideontimes}\to (3, 1, 7)$.
If $6\not\in C(u_1)$,
then $(u u_4, u u_3, u v)\overset{\divideontimes}\to (6, T_3, T_4)$ if $H$ contains no $(6, 2)_{(u_4, u_2)}$-path,
or else, $(u u_1, u v)\overset{\divideontimes}\to (6, 7)$.
Otherwise, $3, 6\in C(u_1)$.
Then $t_1\ge 4$ and $d(u_3)\ge |[3, 5]| + t_1 \ge 7$, a contradiction.

Now assume that $5\in B_1$ and $T_1\subseteq C(u_4)$.
Since $d(u_3) + d(u_4)\ge |\{3, 5, a_3\}| + |\{1, 4, 5\}| + t_{u v} + c_{3 4} \ge \Delta + 5 + c_{3 4}\ge \Delta + 5 + t_1$,
$t_1\le 3$ and $d(u_1)\ge |\{1, 5, 2/4\}| + |T_{u v}\setminus T_1| = 3 + \Delta - 1 - t_1\ge \Delta - 1$.
One sees that if $d(u_3) + d(u_4)\ge \Delta + 8$, then $\min\{d(u_3), d(u_4)\} = 8$,
$d(u_1) = d(u_2) = 6$ and $d(u_1)\ge \Delta - 1\ge 8 - 1 = 7 > 6$,
a contradiction.
Hence, $d(u_3) + d(u_4) = \Delta = 7$.
It follows that $t_1 = c_{3 4} = 2$, $W_3 = \{3, 5, a_3\}$, and $W_4 = \{1, 5, b_1\}$.
One sees from $2\not\in C(u_4)$ that $a_3 = 4$, $b_1 = 5$.
It follows from  Corollary~\ref{auv2} (5.2) that $T_1\cup T_3\cup \{5\}\subseteq C(u_2)$,
and from  Corollary~\ref{auv2d(u)5} (7.1) that $3, 6\in S_u$, and $2, 4\in C(u_1)$ or $2/4\in C(u_1)$, $1\in S_u$.
One sees that $t_0\ge t_1 + t_3\ge 4$.

Thus, $\sum_{x\in N(u)\setminus \{v\}}d(x)\ge \sum_{i\in [1, 4]}w_i + 2t_{u v} + t_0
\ge |\{1, 5, 2/4\}| + |\{2, 5\}| + |\{3, 4, 5\}| + |\{1, 4, 5\}| + |\{3, 6\}| + 2t_{u v} + t_0 = 2\Delta + 11 + t_0\ge 2\Delta + 11 + t_1 + t_3\ge 2\Delta + 15$.

(1.2.3.) $c(vu_4) = 6$.
If $6\not\in C(u_1)$ and $H$ contains no $(6, i)_{(u_1, u_i)}$-path for each $i\in [2, 3]$,
then $u u_1\to 6$ reduces the proof to (2.1.2.2.).
If $6\not\in C(u_3)$ and $H$ contains no $(6, i)_{(u_3, u_i)}$-path for each $i\in [1, 2]$,
then $u u_3\to 6$ reduces the proof to (2.1.2.1.).
Otherwise, $6\in C(u_1)$ or $H$ contains a $(6, i)_{(u_1, u_i)}$-path for some $i\in [2, 3]$,
and $6\in C(u_3)$ or $H$ contains a $(6, i)_{(u_3, u_i)}$-path for some $i\in [1, 2]$.
It follows that $6\in C(u_1)\cup C(u_3)$.

When there exists a $j\in T_{u v}\setminus (C(u_3)\cup T_{u_4})$, say $7\in T_3\cap T_4$,
it follows that $2\in C(u_3)$ and $H$ contains a $(2, j)_{(u_2, u_3)}$-path for each $j\in T_3$, $T_3\subseteq C(u_2)$,
and from Corollary~\ref{auv2} (5.1) that $3\in S_u$.
It follows from Corollary~\ref{auv2} (3) that $4\in B_1\cup B_3$, mult$_{S_u}(4)\ge 2$,
and from Corollary~\ref{auv2} (4) that $5\in C(u_1)$ and mult$_{S_u}(5)\ge 2$.
One sees from Corollary~\ref{auv2d(u)5} (7.1) that $2\in C(u_1)$ or $1\in S_u$,
and from $5, 2/4\in C(u_1)$ that $t_1\ge 2$.
Then $\sum_{x\in N(u)\setminus \{v\}}d(x)\ge \sum_{i\in [1, 4]}w_i + 2t_{uv} + t_0\ge |\{1\}| + |\{2\}| + |\{3, 5, 2\}| + |\{4, 6\}| + |\{4, 4, 5, 5, 6, 3, 2/1\}| + 2t_{uv} + t_0 = 2\Delta + 12 + t_0$.
It follows that there exists a $8\in T_1\setminus (C(u_2)\cap C(u_4))$, and then follows from Corollary~\ref{auv2} (5.1) that $1\in S_u$.

If mult$_{S_u\setminus C(u_4)}(i) = 2$ for some $i\in \{3, 6\}$,
one sees that $\sum_{x\in N(u)\setminus \{v\}}d(x)\ge \sum_{i\in [1, 4]}w_i + 2t_{uv} + t_0\ge |\{1\}| + |\{2\}| + |\{3, 5, 2\}| + |\{4, 6\}| + |\{4, 4, 5, 5, 6, 3, 1, 6/3\}| + 2t_{uv} + t_0 = 2\Delta + 13 + t_0$,
then it follows $2\not\in C(u_1)\cup C(u_4)$, $T_1\subseteq C(u_4)\setminus C(u_2)$, a contradiction to Corollary~\ref{auv2} (5.2).
Otherwise, mult$_{S_u\setminus C(u_4)}(6) = 1$ and $6\in (C(u_1)\cup C(u_3))\setminus C(u_2)$.
One sees that $3/6\in C(u_1)$ and if $3\not\in C(u_4)$,
then it follows from Corollary~\ref{auv2} (4) that mult$_{S_u\setminus C(u_4)}(6)\ge 2$.
Hence, $3\in C(u_4)\setminus C(u_1)$ and then $6\in C(u_1)$.
Since $H$ contains a $(1, 6)_{(u_3, u_1)}$-path, $(u u_2, u u_3, u v)\overset{\divideontimes}\to (6, 7, 8)$.

In the other case, $T_{uv}\subseteq C(u_3)\cup C(u_4)$.
Let $C(u_3) = W_3\cup T_4\cup C_{34}$, $C(u_4) = W_4\cup T_3\cup C_{34}$.
One sees from Corollary~\ref{auv2}(4) that if mult$_{S_u\setminus C(u_3)}(5)\le 1$,
then $6\in C(u_3)$.

First, assume that $5\not\in S_u\setminus C(u_3)$.
Then $6\in T_3$.
It follows that $3\in C(u_1)$ and $H$ contains a $(3, 5)_{(u_1, u_3)}$-path.
If $H$ contains no $(2, j)_{(u_3, u_2)}$-path for some $j\in T_3$,
then $(u u_4, u u_3)\to (5, j)$ reduces the proof to (1.2.1.)
If $H$ contains no $(4, j)_{(u_4, u_j)}$-path for some $j\in T_4$,
then $(u u_2, u u_1, u v)\overset{\divideontimes}\to (5, j, T_3)$.
Otherwise, $2\in C(u_3)$ and $H$ contains a $(2, j)_{(u_3, u_2)}$-path for each $j\in T_3$,
$4\in C(u_1)$ and $H$ contains a $(4, j)_{(u_4, u_j)}$-path for each $j\in T_4$,
and $T_3\subseteq C(u_1)$, $T_1\subseteq C(u_4)$.
It follows that $T_4\subseteq C(u_1)$ and $T_1\subseteq C_{34}$.
If $4\not\in C(u_2)\cup C(u_3)$,
then $(u u_4, u v)\overset{\divideontimes}\to (5, 4)$, or $(u u_4, u u_3)\to (5, 4)$ reduces the proof to (1.2.1.).
Otherwise, $4\in C(u_2)\cup C(u_3)$.
One sees from Corollary~\ref{auvge2}(7.1) that $2\in C(u_1)$ or $1\in S_u$.
Since $t_0\ge t_3\ge 3$, $\sum_{x\in N(u)\setminus \{v\}}d(x)\ge \sum_{i\in [1, 4]}w_i + 2t_{u v} + t_0
\ge |\{1, 3, 4\}| + |\{2\}| + |\{3, 5, 2, 6\}| + |\{4, 6\}| + |\{4, 2/1\}| + 2t_{u v} + t_0 = 2\Delta + 10 + t_0\ge 2\Delta + 10 + t_3\ge 2\Delta + 13$.
Then $t_0 = t_3 = 3$ and $T_1\setminus C(u_2)\ne \emptyset$.
It follows from Corollary~\ref{auv2}(5.2) that $2\in C(u_1)\cup C(u_4)$,
and from Corollary~\ref{auv2}(5.1) that $1\in S_u$.
Thus, $\sum_{x\in N(u)\setminus \{v\}}d(x)\ge \sum_{i\in [1, 4]}w_i + 2t_{u v} + t_0\ge |\{1, 3, 4\}| + |\{2\}| + |\{3, 5, 2, 6\}| + |\{4, 6\}| + |\{4, 2, 1\}| + 2t_{u v} + t_0 \ge 2\Delta + 14$, a contradiction.

Next, assume that $5\in S_u\setminus C(u_3)$.
Recall that $6\in C(u_1)\cup C(u_3)$.
\begin{itemize}
\parskip=0pt
\item
     $H$ contains no $(6, j)_{(u_3, v)}$-path for some $j\in T_3$.
     It follows from Corollary~\ref{auv2}(4) that $5\in B_1\cap C(u_1)$ and $5\in C(u_2)\cup C(u_4)$.
     By Corollary~\ref{auv2d(u)5} (7.1), $2, 4\in C(u_1)$ or $1\in S_u$, $2, 4\in C(u_3)$,
		or $3\in S_u$, and $\sum_{x\in N(u)\setminus \{v\}}d(x)\ge \sum_{i\in [1, 4]}w_i + 2t_{u v} + t_0
		\ge |\{1, 5, a_1\}| + |\{2\}| + |\{3, 5, a_3\}| + |\{4, 6\}| + |\{5, 6\}| + |\{2, 4\}\setminus \{a_1\}/1, \{2, 4\}\setminus\{a_3\}/3\}|+ 2t_{u v} + t_0 \ge 2\Delta + 11 + t_0$.

     \quad When $T_3\subseteq C(u_2)$.
     Then $t_0\ge t_3\ge 2$ and $\sum_{x\in N(u)\setminus \{v\}}d(x) = 2\Delta + 13$.
     If follows that $t_0 = t_3 = 2$, $W_3 = \{3, 5, a_3\}$, and $6\not\in C(u_2)\cup C(u_3)$.
     Since $6/1/2, 2/4\in C(u_3)$, $a_3 = 2$.
     Then $u u_3\to 6$ reduces to (1.2.1.).
     In the other case, $T_3\setminus C(u_2)\ne \emptyset$. Then $4\in C(u_3)$.
     It follows from Corollary~\ref{auv2}(5.2) that mult$_{S_u}(2)\ge 2$,
     and Corollary~\ref{auv2}(5.1) that $3\in S_u$.
     Then $\sum_{x\in N(u)\setminus \{v\}}d(x)\ge\sum_{i\in [1, 4]}w_i + 2t_{u v} + t_0\ge |\{1, 5\}| + |\{2\}| + |\{3, 5, 4\}| + |\{4, 6\}| + |\{5, 6, 2, 2, 3, 1/4\}| + 2t_{u v} + t_0
		\ge 2\Delta + 12 + t_0$. It follows that there exists a $j\in T_1$ such that mult$_{S_u}(j) = 2$,
     and follows from Corollary~\ref{auv2}(5.1) that $1\in S_u$.
     One sees that if $4\not\in C(u_1)$, then $T_1\subseteq C(u_1)$, $T_1\setminus C(u_4)\ne \emptyset$,
     and it follows from Corollary~\ref{auv2}(5.2) that mult$_{S_u}(4)\ge 2$.
     Hence, $\sum_{x\in N(u)\setminus \{v\}}d(x)\ge\sum_{i\in [1, 4]}w_i + 2t_{u v} + t_0\ge |\{1, 5\}| + |\{2\}| + |\{3, 5, 4\}| + |\{4, 6\}| + |\{5, 6, 2, 2, 3, 1, 4\}| + 2t_{u v} + t_0 \ge 2\Delta + 13 + t_0$.
     Then $t_0 = 0$ and $\sum_{x\in N(u)\setminus \{v\}}d(x) = 2\Delta + 13$.
     If follows that mult$_{S_u\setminus C(u_4)}(6) = 1$ and $1/6\in C(u_3)$, $3/6\in C(u_1)$.
     One sees clearly from $d(u_1)\ge t_3 + |\{1, a_1, 5, 1/6\}|\ge 7$ that $d(u_3) + d(u_4)\le \Delta + 7$,
     and if $T_1\subseteq C(u_4)$, then $c_{3 4}\ge t_1\ge 3$ and $d(u_3) + d(u_4)\ge |\{3, 5, 4, 1/6\}| + |\{4, 6\}| + t_{u v} + c_{3 4}\ge \Delta + 8$
     Hence, $T_1\setminus C(u_4)\ne \emptyset$ and $2\in C(u_1)$, $H$ contains a $(2, j)_{(u_1, u_2)}$-path for each $j\in T_1$.
     If $6\not\in C(u_1)$, then $H$ contains a $(3, 6)_{(u_1, u_3)}$-path,
     and $u u_2\to 6$, $u u_1\to T_1\setminus C(u_4)$ reduces the proof to (1.2.2.).
     Otherwise, $6\in C(u_1)$ and $1\in C(u_3)$.
     It follows from Corollary~\ref{auv2}(5.1) that $3\in C(u_2)\setminus (C(u_1)\cup C(u_4))$, and $H$ contains a $(2, 3)_{(u_2, u_4)}$-path.
     Then $(u u_1, u u_3)\to (3, 6)$ reduces the proof to (1.2.1.).
\item
     $H$ contains a $(6, j)_{(u_3, v)}$-path for each $j\in T_3$.
     It follows from Lemma~\ref{auvge2} (1.2) that $T_3\subseteq C(v_3)$.
     Hence, $\{3\}\cup T_3\cup T_4\subseteq C(v_3)$.
     It follows from \ref{auv2Cvi1vi2} that $5/6\in C(v_3)$ and there exists a $j\in C_{3 4}\cap R_3\cap B_1$,
     say $7\in C_{3 4}\cap R_3\cap B_1$.
     One sees clearly that mult$_{S_u}(7)\ge 3$ and $t_0\ge 1$.

     \quad When $6\not\in C(u_1)\cup C(u_2)$,
     one sees that $3\in C(u_1)$ and $H$ contains a $(3, 6)_{(u_1, u_3)}$-path.
     If $H$ contains no $(4, j)_{(u_1, u_4)}$-path for some $j\in T_1$,
     then $(u u_2, u u_1)\to (6, j)$ reduces the proof to (1.2.2.).
     Otherwise, $4\in C(u_1)$ and $H$ contains a $(4, j)_{(u_1, u_4)}$-path.
     Thus, $T_4\subseteq C(u_1)$, $T_1\subseteq C_{3 4}$ and $c_{3 4}\ge t_1\ge 2$.
     Then $d(u_3) + d(u_4) = |\{3, 5, 6, a_3\}| + |\{4, 6\}| + t_{u v} + c_{3 4} = \Delta + 5 + c_{3 4}\ge \Delta + 5 + t_1$.
     Since $d(u_1)\ge t_3 + t_4 + |\{1, 3, 4\}|\ge 3 + 1 + 3 = 7$,
     $d(u_3) + d(u_4) = \Delta + 7$.
     It follows that $c_{3 4} = t_1 = 2$ and $W_4 = \{4, 6\}$.
     One sees that $H$ contains no $(6, j)_{(u_4, v)}$-path for each $j\in T_4$,
     and it follows from Corollary~\ref{auv2}(4) that mult$_{S_u}(6)\ge 6$, a contradiction.

     \quad In the other case, $6\in C(u_1)\cup C(u_2)$.
     Then $\sum_{x\in N(u)\setminus \{v\}}d(x)\ge \sum_{i\in [1, 4]}w_i + 2t_{uv} + t_0
		\ge |\{1, a_1\}| + |\{2\}| + |\{3, 5, 6, a_3\}| + |\{4, 6\}| + |\{5, 6\}| + |\{2, 4\}\setminus \{a_1\}/1, \{2, 4\}\setminus\{a_3\}/3\}|+ 2t_{u v} + t_0 \ge 2\Delta + 11 + t_0$.
     Since $t_3\ge 3$, there exists a $j\in T_3\setminus C(u_2)$ and $4\in C(u_3)$.
     It follows from Corollary~\ref{auv2}(5.2) that mult$_{S_u}(2)\ge 2$,
     and from Corollary~\ref{auv2}(5.1) that $3\in S_u$.
     By Corollary~\ref{auv2d(u)5} (7.1), $4\in C(u_1)$ or $1\in S_u$,
     and then $\sum_{x\in N(u)\setminus \{v\}}d(x)\ge\sum_{i\in [1, 4]}w_i + 2t_{u v} + t_0\ge |\{1\}| + |\{2\}| + |[3, 6]| + |\{4, 6\}| + |\{5, 6, 2, 2, 3, 1/4\}| + 2t_{u v} + t_0 \ge 2\Delta + 12 + t_0\ge 2\Delta + 13$.
     Hence, there exists a $j\in T_1$ such that mult$_{S_u}(j) = 2$.
     It follows from Corollary~\ref{auv2}(5.1) that $1\in S_u$,
     and from Corollary~\ref{auv2}(5.2) that $4\in S_u\setminus C(u_3)$.
     Thus, $\sum_{x\in N(u)\setminus \{v\}}d(x)\ge\sum_{i\in [1, 4]}w_i + 2t_{u v} + t_0\ge |\{1\}| + |\{2\}| + |[3, 6]| + |\{4, 6\}| + |\{5, 6, 2, 2, 3, 1, 4\}| + 2t_{u v} + t_0\ge 2\Delta + 14$,
     a contradiction.
\end{itemize}

Case 1.3. $\{i_1, i_2\} = \{1, 2\}$.
Then $T_{u v}\subseteq C(u_1)\cup C(u_2)$.
One sees clearly that if $v_1\in \{u_3, u_4\}$, then $R_1\ne \emptyset$.
When $v_1\not\in \{u_3, u_4\}$,
if for some $j\in [3, 4]$, $j\not\in C(v_1)$ and $H$ contains no $(j, i)_{(v_1, v_i)}$-path for each $i\in \{2, 5, 6\}$,
then then $v v_1\to j$ reduces to (1.2.).
Otherwise, We proceed with the following proposition, or otherwise we are done:
\begin{proposition}
\label{R12notempty}
For each $i\in [1, 2]$, $R_i\ne \emptyset$, there exists a $b_i\in [5, 6]\cap C(v_i)$,
and if $v_{i_1}\not\in \{u_3, u_4\}$, where $[1, 2] = \{i_1, i_2\}$, then
\begin{itemize}
\parskip=0pt
\item[{\rm (1)}]
	for each $j\in [3, 4]$, $j\in C(v_{i_1})$ or $H$ contains a $(j, i)_{(v_{i_1}, v_i)}$-path for some $i\in \{i_2, 5, 6\}$;
\item[{\rm (2)}]
	$3/i_2/5/6, 4/i_2/5/6\in C(v_{i_1})$.
\end{itemize}
\end{proposition}

For $j\in R_1$ and $H$ contains a $(b_{1j}, j)_{(v_1, v_{b_{1 j}})}$-path,
if $b_{1 j}\not\in C(u_2)$, and $H$ contains no $(b_{1 j}, i)_{(u_2, u_i)}$-path for each $i\in \{1, 3, 4\}$,
then $(u u_2, u v)\overset{\divideontimes}\to (b_{1 j}, j)$.
Hence, We proceed with the following proposition, or otherwise we are done:
\begin{proposition}
\label{T12notempty}
Let $j\in R_{i_1}$ and $H$ contains a $(b_{i_1 j}, j)_{(v_{i_1}, v_{b_{i_1 j}})}$-path. Then
\begin{itemize}
\parskip=0pt
\item[{\rm (1)}]
	$b_{i_1j}\in C(u_{i_2})$, or $H$ contains a $(b_{i_1 j}, i)_{(u_{i_2}, u_i)}$-path for some $i\in \{i_1, 3, 4\}$;
\item[{\rm (2)}]
	$b_{i_1j}/i_1/3/4\in C(u_{i_2})$ and $T_{i_2}\ne \emptyset$;
\item[{\rm (3)}]
	there exists a $a_{i_2}\in [3, 4]\cap C(u_{i_2})$.
\end{itemize}
\end{proposition}

For each $i\in [1, 2]$, it follows from \ref{R12notempty} that there exists a $b_i\in [5, 6]\cap C(v_i)$,
and follows from \ref{T12notempty} that there exists a $a_i\in [3, 4]\cap C(u_i)$.

(1.3.1) $(vu_3, vu_4)_c = (1, 2)$.
It follows that $T_{uv}\subseteq C(u_3)\cup C(u_4)$.
Let $C(u_3) = W_3\cup T_4\cup C_{34}$, $C(u_4) = W_4\cup T_3\cup C_{34}$,
where $\{3, 1, b_1\}\subseteq W_3$, $\{4, 2, b_2\}\subseteq W_4$.
Then $T_4\subseteq C(u_1)$ and $T_3\subseteq C(u_2)$.

Note that $H$ contains a $(2, j)_{(u_2, u_3)}$-path for each $j\in T_3$ and a $(1, l)_{(u_1, u_4)}$-path for each $l\in T_4$.
It follows from Corollary~\ref{auv2}(4) that if $H$ contains no $(4, j)_{(u_3, u_4)}$-path for some $j\in T_3$,
then $3\in B_2\subseteq C(u_2)\cap C(u_4)$,
and if $H$ contains no $(3, j)_{(u_3, u_4)}$-path for some $j\in T_4$,
then $4\in B_1\subseteq C(u_1)\cup C(u_3)$.
Hence, by symmetry, assume w.l.o.g. $4\in C(u_3)$.
Note that $t_3\ge 3$, $t_4\ge 2$.

When $3\not\in C(u_4)$, then $4\in C(u_1)$.
If $H$ contains no $(4, j)_{(u_4, u_3)}$-path for some $j\in T_3$,
then $(u u_3, u v)\overset{\divideontimes}\to (j, 3)$.
Otherwise, $H$ contains a $(4, j)_{(u_4, u_3)}$-path for each $j\in T_3$.
It follows from Corollary~\ref{auv2}(1.2) that $T_3\subseteq C(u_1)$.
Since $\{4\}\cup T_4\cup T_3\subseteq T_1$, we have $C_{3 4}\ne \emptyset$ and $d(u_1)\ge 7$.
It follows that $d(u_3) + d(u_4)\le \Delta + 7$.
Since $d(u_3) + d(u_4)\ge |\{1, 3, 4, b_1\}| + |\{2, 4, b_2\}| + t_{u v} + c_{3 4}\ge \Delta + 6 + c_{3 4}\ge \Delta + 7$, we have $d(u_3) + d(u_4) = \Delta + 7$, $c_{3 4} = 1$, and $W_3 = \{1, 3, 4, b_1\}$, $W_4 = \{2, 4, b_2\}$.
If there exists an $i\in [5, 6]\setminus (C(u_3)\cup C(u_4))$,
then $(u u_3, u v)\overset{\divideontimes}\to (i, 3)$.
Otherwise, assume w.l.o.g. $b_1 = 5$, $b_2 = 6$.
One sees from \ref{T12notempty}(1) that $H$ contains a $(4, 6)_{(u_1, u_4)}$-path.
Then $(u u_3, u v)\overset{\divideontimes}\to (6, 3)$.

In the other case, $3\in C(u_4)$.
Then $t_3\ge 3$,
$d(u_3) + d(u_4)\ge |\{1, 3, 4, b_1\}| + |\{2, 4, 3, b_2\}| + t_{u v} + c_{3 4}\ge \Delta + 7 + c_{3 4}\ge \Delta + 7$.
Since $\min\{d(u_1), d(u_2)\}\le 6$, assume w.l.o.g. $d(u_1)\le 6$.
One sees clearly that $T_3\ne \emptyset$.
Then, $4\in C(u_1)$ and $H$ contains a $(4, j)_{(u_4, u_1)}$-path for each $j\in T_3\cap T_1$.
It follows from Corollary~\ref{auv2}(3) that $3\in B_2\cap C(u_2)$,
and from Corollary~\ref{auv2}(5.1) that $5, 6\in S_u$, $1\in C(u_4)$ or $H$ contains a $(1, 2)_{(u_2, u_3)}$-path.
One sees that if $1\not\in C(u_4)$, then $2\in C(u_3)$ and $1\in C(u_2)$.
Since $d(u_3) + d(u_4)\ge |\{1, 3, 4, b_1\}| + |\{2, 4, 3, b_2\}| + |\{1/2\}| + t_{uv} + c_{34}\ge \Delta + 8 + c_{34}\ge \Delta + 8$,
we have $d(u_3) + d(u_4) = \Delta + 8$, $c_{34} = 0$, and $W_3 \biguplus W_4 = \{1, 3, 4, b_1\}\cup \{2, 4, b_2\}\cup \{1/2\}$.

It follows that $d(u_1) = d(u_2) = 6$.
One sees that if $1\not\in C(u_4)$, then $2\in C(u_3)$, $1\in C(u_2)$,
and $d(u_1)\ge |[1, 3]| + t_3\ge 3 + 4 = 7>6$.
Hence, $1\in C(u_4)$, $t_4\ge 4$ and $C(u_1) = \{1, 4\}\cup T_4$.

If $1, 4\not\in C(u_1)$, then $(u u_1, u u_2, u v)\overset{\divideontimes}\to (1, 2, T_3)$.
Otherwise, $1/4\in C(u_1)$ and $C(u_1) = \{2, 3, 1/4\}\cup T_3$.
It follows that $\{b_1, b_2\} = \{5, 6\}$.
Assume w.l.o.g. $b_1 = 5$, $b_2 = 6$.
One sees from \ref{T12notempty}(1) that $H$ contains a $(4, 6)_{(u_1, u_4)}$-path.
Then $(u u_3, u u_2, u v)\overset{\divideontimes}\to (6, T_4, 3)$.

(1.3.2) $(v u_3, v u_4)_c = (5, 2)$.
It follows from Corollary~\ref{auv2d(u)5}(2.1) that $2\in S_u\setminus C(u_4)$, or $1\in C(u_2)$,
or $H$ contains a $(1, i)_{(u_2, u_)}$-path for some $i\in [3, 4]$.

First, assume that there exists $j\in T_{uv}\setminus (C(u_3)\cup C(u_4))$, say $7\not\in C(u_3)\cup C(u_4)$,
then $7\in B_1\subseteq C(u_1)$.
It follows from Lemma~\ref{auvge2}(1.2) that $6\in C(u_4)$ and $H$ contains a $(6, 7)_{(u_4, v_6)}$-path,
from Corollary~\ref{auv2}(1.2) that $7\in C(u_2)$,
from Corollary~\ref{auv2}(4) that $5\in B_1\cup B_2$ and mult$_{S_u}(5)\ge 2$,
and from Corollary~\ref{auv2}(3) that $4\in B_1$, mult$_{S_u}(4)\ge 2$.
If $3\not\in S_u$, then $(u u_2, u u_3, u v)\overset{\divideontimes} \to (3, 7, T_1)$.
Otherwise, $3\in S_u$.
One sees from Corollary~\ref{auv2}(3) that if $2\not\in C(u_3)$, then $3\in B_1\cup B_2$, mult$_{S_u}(3)\ge 2$.

When $6\not\in S_u\setminus C(u_4)$, it follows from \ref{T12notempty}(1) that $H$ contains a $(4, 6)_{(u_1, u_4)}$-path.
If $2\not\in C(u_1)$ and $H$ contains no $(2, 3)_{(u_1, u_3)}$-path,
then $(u u_2, u u_1, u v)\overset{\divideontimes} \to (6, 2, 7)$.
Otherwise, $2\in C(u_1)$ or $H$ contains a $(2, 3)_{(u_1, u_3)}$-path.
It follows that $2\in S_u\setminus C(u_4)$ and $2/3\in C(u_1)$, $t_1\ge 2$.
Since $(u u_3, u u_2)\to (7, 6)$ reduces to an acyclic edge $(\Delta + 5)$-coloring for $H$ such that $C(u)\cap C(v) = \{1, 6\}$, it follows from Lemma~\ref{auvge2}(2) that mult$_{S_u}(3)\ge 2$.
One sees that $|\{1, 2, 3, 4, 5, 2, 6\}| + \{5, 5, 4, 4\} = 11$.
Then $\sum_{x\in N(u)\setminus \{v\}}d(x)\ge\sum_{i\in [1, 4]}w_i + 2t_{u v} + t_0\ge 11 + |\{3, 2, 3\}| + 2t_{u v} + t_0 \ge 2\Delta + 12 + t_0$.
It follows that $T_1\setminus T_0\ne \emptyset$ and $1\in S_u$.
Hence, $\sum_{x\in N(u)\setminus \{v\}}d(x)\ge 2\Delta + 13 + t_0$.
It follows that $\sum_{x\in N(u)\setminus \{v\}}d(x) = 2\Delta + 13$, and $t_0 = 0$.
If $H$ contains no $(4, j)_{(u_2, u_4)}$-path for some $j\in T_2$,
then $(u u_3, u u_2)\to (6, T_2)$ reduces the proof to (2.1.2.3).
Otherwise, $4\in C(u_2)$, and $T_2\subseteq C(u_4)$.
Then $d(u_1) + d(u_2) + d(u_3)\ge |\{1, 2, 4, 2, 6\}| + 2t_{uv} + 2\times |\{3, 4, 5\}| = 2\Delta + 9$
and $d(u_3)\le 4$, a contradiction.

In the other case, $6\in S_u\setminus C(u_4)$.
If $1\not\in S_u$, then $2\in S_u\setminus C(u_4)$,
and it follows from Corollary~\ref{auv2d(u)5}(7.2) that $3\in C(u_1)$ and $T_1\subseteq C(u_3)$, i.e., $t_0\ge t_1\ge 2$, thus $\sum_{x\in N(u)\setminus \{v\}}d(x)
\ge\sum_{i\in [1, 4]}w_i + 2t_{u v} + t_0\ge 11 + + |\{3, 6, 2\}| + 2t_{u v} + t_0 \ge 2\Delta + 12 + t_0\ge 2\Delta + 12 + 2\ge 2\Delta + 14$.

Otherwise, $1\in S_u$. If $2\not\in S_u\setminus C(u_4)$, then mult$_{S_u}(3)\ge 2$,
and $\sum_{x\in N(u)\setminus \{v\}}d(x)\ge\sum_{i\in [1, 4]}w_i + 2t_{u v} + t_0\ge 11 + |\{3, 6, 1, 3\}| + 2t_{u v} + t_0 \ge 2\Delta + 13 + t_0$.
It follows that $\sum_{x\in N(u)\setminus \{v\}}d(x) = 2\Delta + 13$ and $t_0 = 0$.
Since it follows from Corollary~\ref{auv2d(u)5}(7.2) that $4\in C(u_2)$ and $T_2\subseteq C(u_4)$,
thus, $d(u_1) + d(u_2) + d(u_3)\ge |\{1, 2, 4, 2, 6\}| + 2t_{u v} + 2\times |\{3, 4, 5\}| = 2\Delta + 9$
and $d(u_3)\le 4$, a contradiction.
Now, assume $2\in S_u\setminus C(u_4)$.
One sees from Corollary~\ref{auv2}(5.2) that $t_0\ge t_1\ge 1$ or mult$_{S_u}(3)\ge 2$.
Let $t^* = \max\{t_0, \mbox{mult}_{S_u}(3) - 1\}$.
Then $t^* \ge 1$ and $\sum_{x\in N(u)\setminus \{v\}}d(x)\ge\sum_{i\in [1, 4]}w_i + 2t_{u v} + t_0\ge 11 + |\{3, 6, 1, 2\}| + 2t_{uv} + t^* \ge 2\Delta + 13 + t^* \ge 2\Delta + 14$, a contradiction.

Next, assume that $T_{u v}\subseteq C(u_3)\cup C(u_4)$.
Let $C(u_3) = W_3\cup T_4\cup C_{3 4}$, $C(u_4) = W_4\cup T_3\cup C_{34}$,
where $\{3, 5\}\subseteq W_3$, $\{4, 2, b_2\}\subseteq W_4$.
Then $T_4\subseteq C(u_1)$.

If $5\not\in C(u_1)$ and $H$ contains no $(5, i)_{(u_1, u_i)}$-path for each $i\in \{2, 4\}$,
then $u u_1\to 5$ reduces the proof to (2.1.3.1).
If $5\not\in C(u_2)$, $H$ contains no $(4, 5)_{(u_2, u_4)}$-path,
and $2\not\in C(u_1)$, $H$ contain no $(3, 2)_{(u_2, u_3)}$-path,
then $(u u_1, u u_2)\to (5, 2)$ reduces the proof to (2.1.3.1).
If $H$ contain neither a $(4, i)_{(u_4, v_i)}$-path for each $i\in \{1, 5, 6\}$,
nor a $(3, j)_{(u_3, u_4)}$-path for some $j\in T_4$,
then $(v u_4, u u_4)\to (4, j)$ reduces to acyclic edge $(\Delta + 5)$-coloring for $H$ such that $C(u)\cap C(v) = \{1\}$.

Otherwise, we proceed with the following proposition, or otherwise we are done:
\begin{proposition}
\label{245524}
\begin{itemize}
\parskip=0pt
\item[{\rm (1)}]
    $5\in C(u_1)$ or $H$ contains a $(5, i)_{(u_1, u_i)}$-path for some $i\in \{2, 4\}$;
\item[{\rm (2)}]
	$5\in C(u_2)$ or $H$ contains a $(4, 5)_{(u_2, u_4)}$-path,
    or $2\in C(u_1)$ or $H$ contain a $(3, 2)_{(u_2, u_3)}$-path;
\item[{\rm (3)}]
	$H$ contain a $(4, i)_{(u_4, v_i)}$-path for some $i\in \{1, 5, 6\}$,
    or a $(3, j)_{(u_3, u_4)}$-path for each $j\in T_4$.
\end{itemize}
\end{proposition}

We distinguish the following two scenarios on whether $T_4\subseteq C(u_2)$ or not.
\begin{itemize}
\parskip=0pt
\item
     $T_4\subseteq C(u_4)$.
     When $T_3\subseteq C(u_1)\cap C(u_2)$, since $3, 4\in C(u_1)$ or $1\in S_u$ by Corollary~\ref{auv2d(u)5}(5.1), $\sum_{x\in N(u)\setminus \{v\}}d(x)
		\ge\sum_{i\in [1, 4]}w_i + 3t_{u v}\ge 11 + |\{2, a_2, 1, a_1\}| + \{3, 5, 2, 4, b_2\} + \{3, 4\}\setminus \{a_1\}/1 +3t_{u v} = 3(\Delta - 1) + 10 = 3\Delta + 7$,
     it follows that $\Delta \ge 7$ and $\sum_{x\in N(u)\setminus \{v\}}d(x)\ge 2\Delta + 14$, a contradiction.
     Hence, $T_3\setminus (C(u_1)\cap C(u_2))\ne \emptyset$.
     One sees that if $T_3\setminus C(u_1)\ne \emptyset$,
     then $4\in C(u_1)$,
     and it follows from Corollary~\ref{auv2}(5.2) that $H$ contains a $(4, j)_{(u_1, u_4)}$-path for $j\in T_1$, mult$_{S_u}(3)\ge 2$, from Corollary~\ref{auv2}(5.1) that $1, 6\in S_u$;
     if $T_3\setminus C(u_2)\ne \emptyset$,
     then $4\in C(u_2)$,
     and it follows from Corollary~\ref{auv2}(5.2) that $H$ contains a $(4, j)_{(u_2, u_4)}$-path for $j\in T_2$, mult$_{S_u}(3)\ge 2$, from Corollary~\ref{auv2}(5.1) that $6\in S_u$.
     One sees that  $t_0\ge t_4\ge 2$.
     Since $5\in S_u\setminus C(u_3)$, $6\in S_u$ and mult$_{S_u}(3)\ge 2$,
     $|C(u)\cup \{c(v u_3), c(v u_4)\}| + $mult$_{S_u}(3) + $mult$_{S_u\setminus C(u_3)}(5) + |\{6\}|\ge 10$.
     If $T_3\cup T_4\subseteq C(u_i)$ for some $i\in [1, 2]$, then $t_{3 4}\ge w_i - 1$.
     \begin{itemize}
     \parskip=0pt
     \item
         $T_3\subseteq C(u_1)$ and $T_3\cap C(u_2)\ne \emptyset$.
         It follows that $4\in C(u_2)$, $6\in S_u$, $5\in S_u\setminus C(u_3)$.
         Then $\sum_{x\in N(u)\setminus \{v\}}d(x)\ge\sum_{i\in [1, 4]}w_i + 2t_{u v} + t_0\ge 10 + |\{4\}| + |\{1/4\}| + 2t_{u v} + t_0 \ge 2\Delta + 10 + t_4 + c_{3 4}\ge 2\Delta + 2 + 1 =2\Delta + 13$.
         Hence, $c_{3 4}= 1$.

         \quad Since $3/4, 5/2/4\in C(u_1)$, $C(u_1) = \{1, 4\}\cup T_3\cup T_4$.
         It follows from \ref{245524}(1) that $H$ contains a $(4, 5)_{(u_1, u_4)}$-path, $5\in C(u_4)$,
         from \ref{245524}(2) that $5\in C(u_2)$,
         and from Corollary~\ref{auv2d(u)5}(7.1) that $1\in S_u$.
         Then $\sum_{x\in N(u)\setminus \{v\}}d(x)\ge\sum_{i\in [1, 4]}w_i + 2t_{u v} + t_0\ge 10 + |\{4\}| + |\{5, 4, 1\}| + 2t_{u v} + t_0 \ge 2\Delta + 12 + t_4 + t_{3 4}\ge 2\Delta + 15$.
     \item
         $T_3\subseteq C(u_2)$ and $T_3\cap C(u_1)\ne \emptyset$.
         It follows that $4\in C(u_1)$, $1\in S_u$,
         and from Corollary~\ref{auv2d(u)5}(7.1) that if $2\not\in S_u$, then $4\in C(u_2)$.
         Then $\sum_{x\in N(u)\setminus \{v\}}d(x)\ge\sum_{i\in [1, 4]}w_i + 2t_{u v} + t_0\ge 10 + |\{4\}| + |\{1\}| + |\{2/4\}| + 2t_{u v} + t_0 \ge 2\Delta + 11 + t_4 + c_{3 4}\ge 2\Delta + 2 + 1 = 2\Delta + 14$.
     \item
         $T_3\setminus C(u_1)\ne \emptyset$ and $T_3\cap C(u_1)\ne \emptyset$.
         It follows that $4\in C(u_1)\cap C(u_2)$, $1, 6\in S_u$, $5\in S_u\setminus C(u_3)$.
         Then $\sum_{x\in N(u)\setminus \{v\}}d(x)\ge\sum_{i\in [1, 4]}w_i + 2t_{u v} + t_0\ge 10 + |\{4\}| + |\{4\}| + |\{1\}| + 2t_{u v} + t_0 \ge 2\Delta + 11 + t_0\ge 2\Delta + 13$.
         It follows that $2\not\in S_u\setminus C(u_4)$, mult$_{S_u}(6)= $ mult$_{S_u\setminus C(u_3)}(5) = 1$.
         One sees that for each $i\in [1, 2]$, $B_1\cap T_3\ne \emptyset$.
         It follows from Corollary~\ref{auv2}(4) that $6\in C(u_3)$ and then $b_2 = 5$.
         Hence, by \ref{245524} (1) and (2), we can get an acyclic edge $(\Delta + 5)$-coloring for $H$.

     \end{itemize}

\item
     $T_4\setminus C(u_4)\ne \emptyset$.
     Then $3\in C(u_2)$ and $H$ contains a $(3, j)_{(u_3, u_2)}$-path for each $j\in T_4\cap T_2$.
     It follows from Corollary~\ref{auv2}(5) that $H$ contains a $(3, j)_{(u_3, u_2)}$-path for each $j\in T_2$.
     Thus, one sees clearly that $T_2\subseteq C(u_3)$ and $T_3\subseteq C(u_2)$.
     It follows from Corollary~\ref{auv2}(7.2) that $2\in S_u\setminus C(u_4)$,
     from Corollary~\ref{auv2}(5.1) that $6\in S_u$,
     from Corollary~\ref{auv2}(5.2) or (3) that $4\in B_1\subseteq C(u_1)\cap C(v_1)$ and mult$_{S_u}(4)\ge 2$.
     One sees from Corollary~\ref{auv2d(u)5}(7.1) that if $3\not\in C(u_1)$, then $1\in S_u$.
     Then $\sum_{x\in N(u)\setminus \{v\}}d(x)\ge\sum_{i\in [1, 4]}w_i + 2t_{u v} + t_0\ge |\{1, 4\}| + |\{2, 3\}| + |\{3, 5\}| + |\{2, 4\}| + |\{5, 4, 2, 6, 1/3\}| + 2t_{u v} + t_0 \ge 2\Delta + 11 + t_0$.

     \quad When $T_3\subseteq C(u_1)$, one sees that $c_{3 4}\ge 1$, and $t_0\ge t_3 + t_{3 4}$,
     Thus, $\sum_{x\in N(u)\setminus \{v\}}d(x)\ge 2\Delta + 11 + t_0\ge 2\Delta + 11 + t_3 + t_{34}\ge 2\Delta + 11 + 1 + 1 = 2\Delta + 13$.
     It follows that $W_3 = \{3, 5\}$, and mult$_{S_u\setminus C(u_3)}(5) = 1$.
     However, we are done by Corollary~\ref{auv2}(4).

     \quad In the other case, $t_3\setminus C(u_1)\ne\emptyset$.
     Then $H$ contains a $(4, j)_{(u_1, u_4)}$-path,
     and it follows from Corollary~\ref{auv2}(5.2) mult$_{S_u}(3)\ge 2$,
     from Corollary~\ref{auv2}(5.1)$1\in S_u$.
     Then $\sum_{x\in N(u)\setminus \{v\}}d(x)\ge\sum_{i\in [1, 4]}w_i + 2t_{u v} + t_0\ge |\{1, 4\}| + |\{2, 3\}| + |\{3, 5\}| + |\{2, 4\}| + |\{5, 4, 2, 6, 3, 1\}| + 2t_{uv} + t_0 \ge 2\Delta + 12 + t_0\ge 2\Delta + 12 + c_{34}$.
     It follows that $c_{3 4}\le 1$.
     Together with $\{1, 4, 5/6\}\cup T_4\subseteq C(v_1)$, $R_1\ge 2$ and $R_1\cap T_3\ne \emptyset$.
     It follows that for some $\gamma\in R_1\cap T_3$,
     $H$ contains a $(2, \gamma)_{(u_2, u_4)}$-path and a $(6, \gamma)_{(v_1, v_6)}$-path.
     Hence, it follows from Corollary~\ref{auv2}(4) that $5\in B_1\cup B_2$,
     mult$_{S_u\setminus C(u_3)}(5)\ge 2$,
     and $\sum_{x\in N(u)\setminus \{v\}}d(x)\ge 2\Delta + 12 + c_{3 4}$.
     Then $c_{34} = 0$, $t_0= 0$, $T_4\cap C(u_2) = \emptyset$, and mult$_{S_u\setminus C(u_3)}(5)= 2$, mult$_{S_u}(6) = 1$.
     Since $4\in B_1$ and $H$ contains a $(3, j)_{(u_2, u_3)}$-path for each $j\in T_4$,
     it follows from \ref{245524}(3) that $H$ contains a $(4, i)_{(u_4, v_i)}$-path for some $i\in [5, 6]$.
     \begin{itemize}
     \parskip=0pt
     \item
         $5\in C(u_1)\cap C(u_2)$.
         It follows that $6\in C(u_4)$ and $H$ contains a $(j, 6)_{(u_4, v_6)}$-path for each $j\in \{4\}\cup T_4$.
         Then $(u u_1, u v)\overset{\divideontimes} \to (6, T_4)$\footnote{By \ref{T12notempty}(1) and \ref{245524}(3), we are done. }.
     \item
         $5\in C(u_2)\cap C(u_4)$.
         It follows that $H$ contains a $(4, 5)_{(u_4, u_1)}$-path,
         and from \ref{245524}(3) that $H$ contains a $(4, 6)_{(u_4, v_6)}$-path.
         Then $(u u_2, u v)\overset{\divideontimes}\to (6, \gamma)$.
     \item
         $5\in C(u_1)\cap C(u_4)$.
         Then $5\in B_1$.
         If $H$ contains no $(4, 5)_{(u_2, u_4)}$-path, then $(u u_2, u v)\overset{\divideontimes}\to (5, T_3)$.
         Otherwise, $H$ contains a $(4, 5)_{(u_2, u_4)}$-path and from \ref{245524}(3) that $H$ contains a $(4, 6)_{(u_4, v_6)}$-path.
         Then $(u u_2, u v)\overset{\divideontimes}\to (6, \gamma)$.
     \end{itemize}
\end{itemize}

(1.3.3.) $(v u_3, v u_4)_c = (5, 6)$.
If $5\not\in C(u_1)$ and $H$ contains no $(5, i)_{(u_1, u_i)}$-path for each $i\in \{2, 4\}$,
then $u u_1\to 5$ reduces the proof to (1.3.2).
If $6\not\in C(u_1)$ and $H$ contains no $(6, i)_{(u_1, u_i)}$-path for each $i\in \{2, 3\}$,
then $u u_1\to 6$ reduces the proof to (1.3.2).
Hereafter, we proceed with the following proposition, or otherwise we are done:
\begin{proposition}
\label{5656}
\begin{itemize}
\parskip=0pt
\item[{\rm (1)}]
	$5\in C(u_1)$ or $H$ contains a $(5, i)_{(u_1, u_i)}$-path for some $i\in \{2, 4\}$;
\item[{\rm (2)}]
	$6\in C(u_1)$ and $H$ contains no $(6, i)_{(u_1, u_i)}$-path for some $i\in \{2, 3\}$,
\item[{\rm (3)}]
    $5\in C(u_2)$ or $H$ contains a $(5, i)_{(u_1, u_i)}$-path for some $i\in \{1, 4\}$;
\item[{\rm (4)}]
	$6\in C(u_2)$ and $H$ contains no $(6, i)_{(u_1, u_i)}$-path for some $i\in \{1, 3\}$.
\end{itemize}
\end{proposition}

It follows that $5, 6\in C(u_1)\cup C(u_2)$, $5/2/4, 6/2/3\in C(u_1)$, $5/1/4, 6/2/3\in C(u_2)$,
and $T_1\ne \emptyset$, $T_2\ne\emptyset$.
If for some $i\in [1, 2]$, $C(v_i) =\{i\}\cup T_{uv}$,
then $v v_i\to 3$ reduces the proof to (1.3.2).
Otherwise, $R_i\ne \emptyset$.
Hence, for each $i\in [1, 2]$, there exists a $a_i\in [3, 4]\cap C(u_i)$, and $b_i\in [5, 6]\cap C(v_i)$.

One sees that $w_1 + w_2\ge 6$, $d(u_1) + d(u_2)\ge w_1 + w_2 + t_{u v}\ge |\{1, a_1\}| + |\{2, a_2\}| + |\{5, 6\}| + \Delta - 1 = \Delta + 5$, $t_1 = \Delta - 1 - (d(u_1) - w_1)$ and $t_2 = \Delta - 1 - (d(u_2) - w_2)$.
Then $t_1 + t_2 = 2\Delta - 2 - (d(u_1) + d(u_2)) + (w_1 + w_2)\ge 2\Delta - 2 - 2\Delta + w_1 + w_2 = w_1 + w_2 - 2\ge 4$.

When $\min \{d(u_3), d(u_4)\} = 8$, it follows that $d(u_1) = d(u_2) = 6$,
thus, $d(u_1) + d(u_2)\ge \Delta + 5\ge 8 + 5 = 13 > 12$.
When $\min \{d(u_3), d(u_4)\} = 7$, it follows that $d(u_1) + d(u_2)\le \Delta + 6$, $6\le w_1 + w_2 \le 7$,
and $t_1 + t_2 = 2\Delta - 2 - (d(u_1) + d(u_2)) + (w_1 + w_2)\ge 2\Delta - 2 - (\Delta + 6) + 6\ge \Delta - 2$.

One sees that if $c_1\ge 3$ and $c_2\ge 3$, then $d(u_1) + d(u_2)\ge c_1 + c_2 + |\{5, 6\}| + t_{uv} = \Delta + 7$.
Hence, assume w.l.o.g. $c_1 = 2$, i.e., $4\in C(u_1)$ and $2, 3\not\in C(u_1)$.
Then $H$ contains a $(4, j)_{(u_1, u_4)}$-path for each $j\in T_1$ and $T_1\subseteq C(u_4)$.
It follows from \ref{5656} (2) that $6\in C(u_1)$,
and from Corollary~\ref{auv2}(5.1) that $1\in S_u$.
By Corollary~\ref{auv2d(u)5}(7.1), $1/2/3\in C(u_4)$.
Since $\{4, 5, 1/2/3\}\cup T_1\subseteq C(u_4)$, $T_2\setminus C(u_4)\ne \emptyset$.
Then $3\in C(u_2)$, and it follows from Corollary~\ref{auv2}(5.1) that $2\in S_u$,
from Corollary~\ref{auv2}(5.2) that mult$_{S_u}(4)\ge 2$,
from Corollary~\ref{auv2d(u)5}(7.2) that $2/1/4\in C(u_3)$.
Since $\{3, 5, 2/1/4\}\cup T_2\subseteq C(u_3)$, $T_1\setminus C(u_3)\ne \emptyset$.
It follows from Corollary~\ref{auv2}(5.2) that mult$_{S_u}(3)\ge 2$ and $3\in C(u_4)$.

If $6\not\in C(u_2)\cup C(u_3)$,
it follows from \ref{5656}(4) that $W_2 = [1, 3]$, $H$ contains a $(1, 6)_{(u_1, u_2)}$-path,
then $(u u_3, u u_2)\to (6, T_2)$ reduces the proof to (1.2.1.).
Otherwise, $6\in C(u_2)\cup C(u_3)$.
Then $\sum_{x\in N(u)\setminus \{v\}}d(x)\ge |\{1, 4, 6\}| + |\{2, 3\}| + |\{3, 5\}| + |\{3, 4, 6\}| + |\{4, 1, 2, 5, 6\}| + 2t_{u v} + t_0\ge 2\Delta + 13 + t_0$.
It follows that mult$_{S_u\setminus C(u_3)}(2) = 1$.
One sees from Corollary~\ref{auv2d(u)5}(7.2) that if $2\in C(u_4)$, then $H$ contains a $(4, 2)_{(u_3, u_4)}$-path.
Then $(u u_1, u u_2)\to (2, 5)$ reduces the proof to (1.3.2.).

In the other case, ($A_{7.1}$), i.e., $\min\{d(u_3), d(u_4)\}\in [5, 6]$,
and $\sum_{x\in N(u)\setminus \{v\}}d(x)\le 2\Delta + 13$.
Assume w.l.o.g. $d(u_3)\le 6$.
If $T_1\cup T_2\subseteq C(u_3)$, then $C(u_3) = \{3, 5\}\cup T_1\cup T_2$, $t_1 = t_2 = 2$,
it follows from Corollary~\ref{auv2d(u)5}(7.2) that $4\in C(u_1)\cap C(u_2)$,
thus $W_1 = \{1, 4, 5\}$, $W_2 = \{2, 4, 6\}$, and by \ref{5656} (2), we are done.
Otherwise, $(T_1\cup T_2)\setminus C(u_3)\ne \emptyset$.
It follows from Corollary~\ref{auv2}(5.2) that mult$_{S_u}(3)\ge 2$.
Likewise, $(R_1\cup R_2)\setminus C(u_3)\ne \emptyset$ and mult$_{S_u}(5)\ge 2$.

Assume w.l.o.g. $T_1\setminus C(u_3)\ne \emptyset$.
Then $4\in C(u_1)$ and $T_1\subseteq C(u_3)$.
By Corollary~\ref{auv2d(u)5}(7.1), $4\in C(u_2)$ or $2\in S_u$.
One sees that $\{1, 4\}| + |\{2\}| + |\{3, 5\}| + |\{4, 6\}| + |\{3, 3, 5, 5, 1, 6\}| + 2t_{u v} + t_0 = 2\Delta + 11 + t_0$.
Then $\sum_{x\in N(u)\setminus \{v\}}d(x)\ge 2\Delta + 11 + t_0 + |\{2/4\}| = 2\Delta + 12 + t_0$.
One sees that $t_2\ge 2$.
It follows that $T_2\setminus T_0\ne \emptyset$.
It follows that $2\in S_u$,
and from Corollary~\ref{auv2}(5.2) that if $4\not\in C(u_2)$, then mult$_{S_u}(4)\ge 2$, i.e., $4\in C(u_2)\cup C(u_3)$.
Hence, $\sum_{x\in N(u)\setminus \{v\}}d(x)\ge  2\Delta + 11 + t_0 + |\{2, 4\}| = 2\Delta + 13 + t_0$.
It follows that for each $i\in \{1, 2, 6\}$, mult$_{S_u}(i) = 1$,
for each $j\in [3, 5]\cup T_{u v}$, mult$_{S_u}(j) = 2$.
Since $4\in C(u_1)$, $6\in C(u_1)\cup C(u_2)$, $C(u_3)\subseteq [1, 5]\cup T_2$.

One sees from Corollary~\ref{auv2d(u)5} (5) that $T_2\subseteq C(u_i)$ for some $i\in [3, 4]$.
If $T_2\subseteq C(u_4)$, then $4\in C(u_2)$, and $C(u_3)\subseteq \{1, 2, 3, 5\}$, i.e., $d(u_3)\le 4$, a contradiction.
Otherwise, $T_2\subseteq C(u_3)$.
Then $3\in C(u_2)$.
If $6\in C(u_1)$, one sees from \ref{5656} (4) that $H$ contains a $(1, 6)_{(u_1, u_2)}$-path,
then $(u u_3, u u_2)\to (6, T_2)$ reduces the proof to (2.1.2.1.).
Otherwise, $6\in C(u_2)$ and it follows from \ref{5656} (4) that $2\in C(u_1)$.
Then it follows from Corollary~\ref{auv2d(u)5} (7.2) that $1\in C(u_3)$.
Since $d(u_2) + d(u_3)\ge t_{u v} + \{2, 3, 6\} + \{3, 5, 1\} + \{4\} = \Delta + 6$,
$5\not\in C(u_1)$ and then it follows from \ref{5656}(1) that $4\in C(u_2)$.
Then $W_2 = \{2, 3, 4, 6\}$, $W_3 = \{3, 5, 1\}$.
It follows from Corollary~\ref{auv2} (5.1) that $H$ contains a $(1, 3)_{(u_3, u_4)}$-path.
Then $(u u_1, u u_2)\to (1, 6)$ reduces the proof to (1.3.2.).



{\bf Case 2.} $d(v) = 6$.
Then $t_{u v} = \Delta - 2$, $V_{u v} = [5, 7]$.

One sees that $T_i = \Delta - 2 - (d(u_i) - w_i) = \Delta - d(u_i) + w_i - 2\ge w_i - 2$, $i\in [1, 4]$;
$R_j = \Delta - d(v_i) + |C(v_i)\cap C_{u v}| - 2\ge |C(v_i)\cap C_{u v}| - 2$, $j\in C(v)$,
and if for some $i\in [3, 4]$, $i, c(v u_i)\in \{i_1, i_2\}$,
then $C(u_i) = \{i, c(v u_i)\}\cup T_{u v}$ and $d(u_i) = \Delta$.
By symmetry, we consider the following five cases.

\begin{figure}[h]
\begin{center}
\includegraphics[width=6.0in]{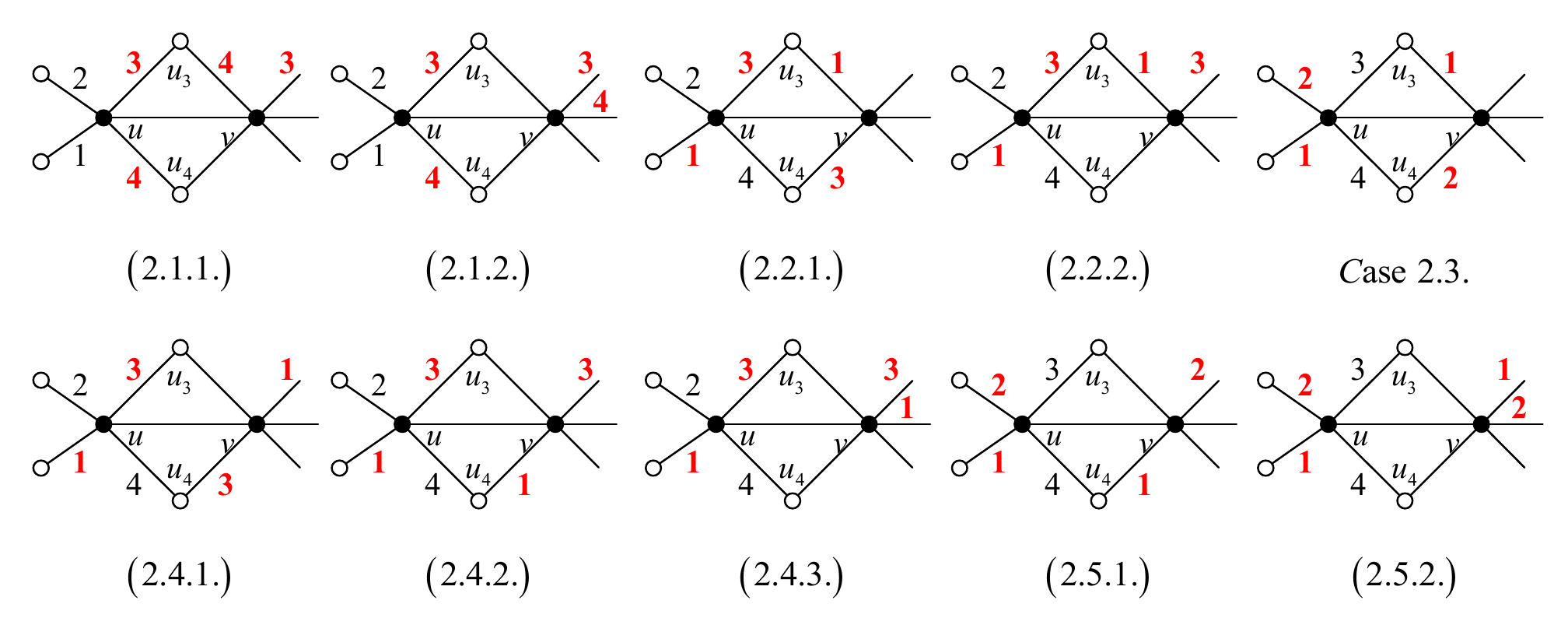}
\caption{Coloring of ($A_8$): $|C(u)\cap C(v)| = 2$ and $C(v) = \{i_1, i_2\}\cup [5, 7]$.\label{fig09}}
\end{center}
\end{figure}

{Case 2.1.} $\{i_1, i_2\} = \{3, 4\}$.
Then $T_{u v}\subseteq C(u_3)\cup C(u_4)$.
Let $C(u_3) = W_3\cup T_4\cup C_{3 4}$, $C(u_4) = W_4\cup T_3\cup C_{3 4}$.
One sees from $t_3 + t_4 + c_{3 4} = t_{u v}$, $w_3 + w_4 + t_{u v} + c_{3 4}\le \Delta + 7$ that $w_3 + w_4\le 9$, $c_{3 4}\le \Delta + 7 - (w_3 + w_4 + t_{u v}) = 9 - (w_3 + w_4)$
and $t_3 + t_4\ge 2(\Delta - 2) + w_3 + w_4 - (\Delta + 7) = \Delta + (w_3 + w_4) - 11$.
By symmetry, there are two subcases on whether $\{c(v u_3), c(v u_4)\}\cap A_{u v} = \emptyset$ or not.

(2.1.1.) $\{c(v u_3), c(v u_4)\}\cap A_{u v} \ne \emptyset$, say $c(v u_3) = 4$.
Then $C(u_3) = [3, 4]\cup [8, \Delta + 5]$.
If $H$ contains no $(i, j)_{(u_4, v_i)}$-path for each $i\in [6, 7]$, $j\in T_{u v}\cup (\{1, 2\}\setminus C(u_4))$,
then $(v u_4, u v)\overset{\divideontimes} \to (i, 5)$.
Otherwise, assume w.l.o.g. $8\in T_4$ and $H$ contains a $(6, 8)_{(u_4, v_6)}$-path, $6\in C(u_4)$.
By Corollary~\ref{auv2}(1.2), assume w.l.o.g. $H$ contains a $(1, 8)_{(u_1, u_4)}$-path and $1\in C(u_4)$.
If $8\not\in C(u_2)$, then $(u u_2, u v)\overset{\divideontimes}\to (8, 2)$.
Otherwise, $8\in C(u_2)$.
It follows that $T_4\subseteq C(u_1)\cap C(u_2)$.
One sees that $1, 2\not\in B_3\cup B_4$ and from Corollary~\ref{auv2d(u)5}(4.2) that $T_{u v}\cap T_4\subseteq C(u_1)\cup C(u_2)$.
Hence, $T_0 = T_{u v}$.
If $H$ contains no $(4, 6)_{(u_4, v)}$-path, then $(u u_3, u v)\overset{\divideontimes}\to (6, 8)$.
If $H$ contains no $(1, i)_{(u_1, v)}$-path for some $i\in [5, 6]$,
then $(v u_3, u u_3, u v)\overset{\divideontimes}\to (1, i, 8)$.
Otherwise, $H$ contains a $(4, 6)_{(u_4, v)}$-path and a $(1, i)_{(u_1, v)}$-path for each $i\in [5, 6]$ and $5, 6\in C(u_1)$.
If there exists a $i\in [5, 6]\setminus C(u_2)$, then $(u u_2, u v)\overset{\divideontimes}\to (i, 2)$.
If $7\not\in S_u$, then $(u u_2, u v)\overset{\divideontimes}\to (7, 2)$.
Otherwise, $5, 6\in C(u_2)$ and $7\in S_u$.
Then $\sum_{x\in N(u)\setminus \{v\}}d(x)\ge\sum_{i\in [1, 4]}w_i + 2t_{u v} + t_0\ge |\{1, 5, 6\}| + |\{2, 5, 6\}| + |\{3, 4\}| + |\{4, 5, 1, 6\}| + |\{7\}| + 2t_{u v} + t_{u v} \ge 13 + 3t_{u v} = 3(\Delta - 2) + 13 = 3\Delta + 7$.
Ones sees clearly that $\sum_{x\in N(u)\setminus \{v\}}d(x)\ge 3\Delta + 7\ge 2\Delta + 14$ and none of ($A_{8.1}$)-- ($A_{8.4}$) holds.

(2.1.2.) $\{c(v u_3), c(v u_4)\}\cap A_{u v} = \emptyset$, i.e., $(v u_3, v u_4)_c = (5, 6)$.
If $6\not\in C(u_3)$ and $H$ contains no $(6, i)_{(u_3, u_i)}$-path for each $i\in C(u)$,
then $u u_3\to 6$ reduces the proof to (2.1.1.).
If $5\not\in C(u_4)$ and $H$ contains no $(5, i)_{(u_4, u_i)}$-path for each $i\in C(u)$,
then $u u_4\to 5$ reduces the proof to (2.1.1.).
Otherwise, $6/1/2/4\in C(u_3)$, $5/1/2/3\in C(u_4)$, and we proceed with the following proposition:
\begin{proposition}
\label{6u35u4}
\begin{itemize}
\parskip=0pt
\item[{\rm (1)}]
	$6\in C(u_3)$ or $H$ contains a $(6, i)_{(u_3, u_i)}$-path for some $i\in \{1, 2, 4\}$.
\item[{\rm (2)}]
	$5\in C(u_4)$ or $H$ contains a $(5, i)_{(u_4, u_i)}$-path for some $i\in \{1, 2, 3\}$.
\end{itemize}
\end{proposition}

Then $T_3\ne\emptyset$ and likewise, $T_4\ne \emptyset$.
It follows from Corollary~\ref{auv2}(1.2) that for each $i\in [3, 4]$, there exists a $a_i\in \{1, 2\}\cap C(u_i)$.
One sees from Corollary~\ref{auv2}(4) that $H$ contains a $(j, 6/7)_{(u_3, v)}$-path for some $j\in T_3$ or $5\in B_4$,
and $H$ contains a $(j, 5/7)_{(u_4, v)}$-path for some $j\in T_4$ or $6\in B_3$.

If $5\not\in S_u\setminus C(u_3)$,
since $H$ contains a $(3, 5)_{(u_3, u_4)}$-path by \ref{6u35u4},
then $(u u_1, u u_4, u v)\overset{\divideontimes}\to (5, T_4\cap T_2, T_3)$.
Otherwise, $5\in S_u\setminus C(u_3)$.
One sees that if $T_4\setminus C(u_2) \ne \emptyset$,
then it follows from Corollary~\ref{auv2}(5.2) that mult$_{S_u}(2)\ge 2$, $a_4 = 1$,
$H$ contains a $(1, T_4)_{(u_1, u_3)}$-path,
and from Corollary~\ref{auv2}(5.1) that $4\in C(u_1)$ or $H$ contains a $(4, 2)_{(u_3, u_2)}$-path, i.e., $2/4\in C(u_1)$.

When $6\not\in C(u_3)$ and $5\not\in C(u_4)$, it follows that $7\in C(u_3)\cap C(u_4)$,
and $H$ contains a $(7, T_3)_{(u_3, v)}$-path and a $(7, T_4)_{(u_4, v)}$-path.
Since $w_3 + w_4\le 9$, assume w.l.o.g. $w_3 = 4$, i.e., $W_3 = \{2, 3, 5, 7\}$.
It follows that $T_3\subseteq C(u_2)$ and from \ref{6u35u4}(1) that $H$ contains a $(2, 6)_{(u_3, u_2)}$-path,
$6\in C(u_2)$, from Corollary~\ref{auv2}(5.1) that $3\in C(u_2)$ or $H$ contains a $(3, 1/4)_{(u_3, u)}$-path, i.e., $3/1/4\in C(u_2)$.
One sees that $c_{3 4}\le 9 - (w_3 + w_4)\le 9 - 8 = 1$, $t_3 + t_4\ge \Delta + (w_3 + w_4) - 11 \ge \Delta + 8 - 11 = \Delta - 3$, and $t_3 + t_4 = \Delta - 3$ if and only if $c_{3 4} = 1$, $w_3 = w_4 = 4$.
\begin{itemize}
\parskip=0pt
\item
    $T_4\subseteq C(u_2)$.
    Since $1, 6, 3/1/4\in C(u_2)$,
    we have $t_3 + t_4 = \Delta - 3$, $c_{3 4} = 1$, and then $W_4 = \{4, 6, a_4, 7\}$.
    One sees that $5\not\in C(u_2)$ and $3\not\in C(u_4)$.
    It follows from \ref{6u35u4}(2) that $a_4 = 1$ and $H$ contains a $(1, 5)_{(u_4, u_1)}$-path.
    Then $(u u_2, u u_3, u v)\overset{\divideontimes}\to (5, T_3, 2)$.
\item
    $T_4\setminus C(u_2) \ne \emptyset$.
    If $6\not\in C(u_1)$, then $(u u_1, u u_4, u v)\to (6, T_4\cap T_2, 1)$.
    Otherwise, $6\in C(u_2)$.
    One sees that $|\{1, 6\}| + |\{2, 6\}| + |\{3, 5, 2, 7\}| + |\{4, 6, 1, 7\}| + |\{3, 2, 4, 5\}| + 2t_{uv} + t_0\ge 16 + 2(\Delta - 2) + t_0 = 2\Delta + 12 + t_0$.
    It follows that $T_3\setminus T_0\ne\emptyset$ and from Corollary~\ref{auv2}(5.2) and (3) that $1\in B_4$, mult$_{S_u}(1)\ge 2$ and then $1\in C(u_2)$.
    Then $\sum_{x\in N(u)\setminus \{v\}}d(x)\ge 2\Delta + 12 + t_0 + |\{1\}| = 2\Delta + 13 + t_0$.
    It follows that $t_0 = 0$, mult$_{S_u}(i) = 1$, $i\in [3, 4]$ and $7\not\in C(u_1)\cup C(u_2)$.
    If $4\not\in C(u_3)$, then $(u u_1, u u_3, u v)\to (7, T_3, T_4)$.
    If $3\not\in C(u_4)$, then $(u u_1, u u_4, u v)\to (7, T_4, T_3)$.
    Otherwise, $4\in C(u_3)$, $3\in C(u_4)$ and then $H$ contains a $(3, 1)_{(u_1, u_2)}$-path and a $(2, 4)_{(u_1, u_2)}$-path.
    One sees that $(u u_3, u u_4)\to (4, 3)$ reduces to an acyclic edge $(\Delta + 5)$-coloring for $H$.
    Since from $1\in B_4$ that $1\in C(v_1)$, it follows from Corollary~\ref{auv2d(u)5}(6) that $R_4\ne \emptyset$. Thus, by Corollary~\ref{auv2d(u)5}(2.1), we are done.
\end{itemize}

In the other case, assume w.l.o.g. $6\in C(u_3)$.
When $w_4 = 3$, say $W_4 = \{4, 6, 1\}$,
then $T_4\subseteq C(u_1)$.
It follows from \ref{6u35u4}(2) that $5\in C(u_1)$, $H$ contains a $(1, 5)_{(u_1, u_4)}$-path,
from Corollary~\ref{auv2}(4) that $6\in B_3$, $6\in C(u_1)$, and $H$ contains a $(1, 5)_{(u_1, u_4)}$-path,
from Corollary~\ref{auv2}(5.1) that for each $i\in \{4, 7\}$, $i\in C(u_1)$ or $H$ contains a $(i, 2/3)_{(u_1, u)}$-path and $i/2/3\in C(u_1)$.
If $5\not\in C(u_2)$, then $(u u_2, u u_3, u v)\overset{\divideontimes}\to (5, T_3, 2)$.
Otherwise, $5\in C(u_2)$.
One sees that $|\{1, 5, 6\}| + |\{2, 5\}| + |\{3, 5, a_3, 6\}| + |\{4, 6, 1\}| + |\{4, 7\}| + 2t_{uv} + t_0\ge 14 + 2(\Delta - 2) + t_0 = 2\Delta + 10 + t_0$.
Since from Corollary~\ref{auv2d(u)5}(7.1) that $1, 2\in C(u_3)$ or $3\in S_u$,
$\sum_{x\in N(u)\setminus \{v\}}d(x)\ge 2\Delta + 10 + t_0 + |\{\{1, 2\}\setminus \{a_3\}/3\}| = 2\Delta + 11 + t_0$.

Since $t_{u v}\ge \Delta + w_3 + w_4 -11\ge \Delta + 7 -11 = \Delta - 4$,
one sees clearly that $t_{u v} = \Delta - 4$ if and only if $c_{3 4} = 2$, $w_3 = 4$, $w_4 = 3$, and $d(u_3) = d(u_4) = \Delta = 7$.
\begin{itemize}
\parskip=0pt
\item
    $T_3\subseteq C(u_1)$.
    Since $1, 5, 6, 4/2/3, 7/2/3\in C(u_2)$,
    we have $t_3 + t_4 = \Delta - 4 = 7 - 4 = 3$, $c_{3 4} = 2$, and then $W_3 = \{3, 5, a_3, 6\}$,
    $W_4 = \{4, 6, 1\}$.
    It follows that $(T_3\cup T_4)\setminus C(u_2)\ne \emptyset$ and it follows from Corollary~\ref{auv2}(5.2) and (3) that $2\in B_3$, $T_3\subseteq C(u_3)$, $a_3 = 2$ and $2\in C(u_1)$.
    Thus, $t_0\ge t_3$, and $\sum_{x\in N(u)\setminus \{v\}}d(x)\ge 2\Delta + 10 + t_0 + |\{3, 2\}| = 2\Delta + 12 + t_0\ge 2\Delta + 12 + 2 = 2\Delta + 14$.
\item
    $T_3\setminus C(u_1) \ne \emptyset$.
    It follows from Corollary~\ref{auv2}(5.2) that $1\in C(u_2)\cup C(u_3)$, $2\in C(u_3)$,
    from Corollary~\ref{auv2}(5.1) that $3\in S_u$.
    Since from Corollary~\ref{auv2}(5.2) that $T_4\subseteq T_0$ or $2\in C(u_1)$,
    $\sum_{x\in N(u)\setminus \{v\}}d(x)\ge 2\Delta + 10 + |\{1, 3, T_4/2\}| = 2\Delta + 13$.
    It follows that $6\not\in C(u_2)$ and mult$_{S_u}(7) = 1$.
    If $H$ contains no $(6, j)_{(u_2, u_3)}$-path for some $j\in T_3$,
    then $(u u_2, u u_3, v u_4, u v)\overset{\divideontimes}\to (6, j, T_4, 2)$.
    Otherwise, $H$ contains a $(6, T_3)_{(u_2, u_3)}$-path and then a $(7, T_3)_{(u_4, v_7)}$-path.
    It follows that $7\in C(u_3)\setminus (C(u_1)\cup C(u_2))$.
    Then $(u u_2, u u_3, u v)\overset{\divideontimes}\to (7, T_3, 2)$.
\end{itemize}

Hence, $w_4\ge 4$, $8\le w_3 + w_4\le 9$, and for each $i\in [3, 4]$, $w_i\le 5$.
It follows that $c_{3 4}\le 9 - (w_3 + w_4)\le 9 - 8 = 1$, $t_3 + t_4\ge \Delta + (w_3 + w_4) - 11 \ge \Delta + 8 - 11 = \Delta - 3$.
One sees that $t_3 + t_4 = \Delta - 3$ if and only if $c_{3 4} = 1$, $w_3 = w_4 = 4$ and $d(u_3) = d(u_4) = \Delta = 7$.

\begin{lemma}
\label{6T33425}
$G$ admits an acyclic edge $(\Delta + 2)$-coloring, or the following holds:
\begin{itemize}
\parskip=0pt
\item[{\rm (1)}]
	If $7\not\in C(u_3)$ and $T_4\cap R_4\ne \emptyset$,
    then {\rm (i)} $H$ contains a $(7, j)_{(v_4, v_7)}$-path for some $j\in T_4\cap R_4$,
    or {\rm (ii)} $H$ contains a $(3, 4)_{(u_4, v_3)}$-path, a $(i, 5)_{(u_3, u_i)}$-path for some $i\in \{1, 2, 4\}$,
    and $3\in C(u_4)$, $i\in C(u_3)$, $5\in C(u_i)$.
\item[{\rm (2)}]
	If $7\not\in C(u_4)$ and $T_3\cap R_3\ne \emptyset$,
    then {\rm (i)} $H$ contains a $(7, j)_{(v_3, v_7)}$-path for some $j\in T_3\cap R_3$,
    or {\rm (ii)} $H$ contains a $(3, 4)_{(u_3, v_4)}$-path, a $(i, 6)_{(u_4, u_i)}$-path for some $i\in \{1, 2, 3\}$,
    and $4\in C(u_3)$, $i\in C(u_4)$, $6\in C(u_i)$.
\item[{\rm (3)}]
	If $W_3 = \{3, 5, 6, 2\}$, then $H$ contains a $(6, T_3)_{(u_3, u_4)}$-path.
\item[{\rm (4)}]
	If $5\in C(u_4)$ and $w_4 = 4$, then $H$ contains a $(5, T_4)_{(u_3, u_4)}$-path.
\item[{\rm (5)}]
	If $W_3 = \{3, 5, 6, a_3, a\}$, then for each $j\in T_3$, $H$ contains a $(i, j)_{(u_3, v_i)}$-path for some $i\in [6, 7]$.
\item[{\rm (6)}]
	If $W_4 = \{3, 5, 6, a_4, b\}$,, then for each $j\in T_4$, $H$ contains a $(i, j)_{(u_4, v_i)}$-path for some $i\in [5, 7]$.
\item[{\rm (7)}]
	If $W_3 = \{3, 5, a_3, 6, 7\}$, then $H$ contains no $(7, j)_{(u_3, v_7)}$-path for each $j\in T_3$ and $H$ contains a $(6, T_3)_{(u_3, u_4)}$-path.
\item[{\rm (8)}]
	If $W_4 = \{4, 6, a_4, 5, 7\}$, then $H$ contains no $(7, j)_{(u_4, v_7)}$-path for each $j\in T_4$ and $H$ contains a $(5, T_4)_{(u_3, u_4)}$-path.
\end{itemize}
\end{lemma}
\begin{proof}
For each $j\in T_4\cap R_4$, assume $H$ contains no $(7, j)_{(v_4, v_7)}$-path.
It follows from Lemma~\ref{auvge2}(1.2) that $H$ contains a $(5, j)_{(v_4, u_3)}$-path.
If $H$ contains no $(3, 4)_{(u_4, v_3)}$-path, then $(v u_3, v v_4)\to (4, T_4\cap R_4)$ reduces the proof to (2.1.1.).
If $H$ contains no $(i, 5)_{(u_3, u_i)}$-path for each $i\in \{1, 2, 4\}$,
then $(u u_3, v u_3, v v_4, u v)\overset{\divideontimes}\to (5, 4, T_4\cap R_4, T_3)$.
Otherwise, (1) holds and likewise, (2) holds.

For (3) and (5), assume that $H$ contains no $(i, 8)_{(u_3, v_i)}$-path for each $i\in [6, 7]$, where $8\in T_3$.
Then it follows from Corollary~\ref{auv2}(4) that $5\in B_4\subseteq C(u_4)\cap C(v_4)$,
$H$ contains a $(i, 5)_{(u_i, u_3)}$-path for some $i\in [1, 2]$ and $5\in C(u_i)$.

For (3), $H$ contains a $(2, T_3)_{(u_2, u_3)}$-path, $T_3\subseteq C(u_2)$,
and $H$ contains a $(2, 5)_{(u_2, u_3)}$-path, $5\in C(u_2)$.
One sees from Corollary~\ref{auv2}(5.1) that $3, 7\in S_u$.

{\bf Case 1.} $T_4\subseteq C(u_2)$.
Since $2, 5, 3/1/4, 7/1/4\in C(u_2)$,
it follows that $t_3 + t_4 = \Delta - 3$, $c_{34} = 1$, and $w_3 = w_4 = 4$,
i.e., $W_4 = \{4, 6, a_4, 5\}$.
Then $C(u_2) = \{2, 5, 1\}\cup T_3\cup T_4$ and $3, 7\in C(u_1)$.
Since $|\{1, 3, 7\}| + |\{2, 5, 1\}| + |\{3, 5, 2, 6\}| + |\{4, 6, a_4, 5\}| + 2t_{u v} + t_0\ge 14 + 2(\Delta - 2) + t_0 = 2\Delta + 10 + t_0$, $(T_3\cup T_4)\setminus C(u_1)\ne \emptyset$.
It follows from Corollary~\ref{auv2}(5.2) and (3) that $1\in B_4\cap C(u_4)\cap C(v_4)$.
Then $W_4 = \{4, 6, 1, 5\}$, $T_4\subseteq C(u_1)$ and $T_4\setminus C(v_4)\ne \emptyset$.
Since $3\not\in C(u_4)$, by (1), $H$ contains a $(7, T_4\cap R_4)_{(v_4, v_7)}$-path.
Then $(u u_3, u v)\overset{\divideontimes}\to (7, T_4\cap R_4)$.

{\bf Case 2.} $T_4\setminus C(u_2)\ne \emptyset$.
It follows from Corollary~\ref{auv2}(5.2) that $2\in (C(u_1)\cup C(u_4))\cap (B_3\cup B_4)$,
$1\in C(u_4)$, $H$ contains a $(1, T_4)_{(u_1, u_4)}$-path,
and from Corollary~\ref{auv2}(5.1) that $4\in S_u$.
One sees that $|\{1\}| + |\{2, 5\}| + |\{3, 5, 2, 6\}| + |\{4, 6, 1, 5\}| + |\{3, 7, 2, 4\}| + 2t_{uv} + t_0 = 15 + 2(\Delta - 2) + t_0 = 2\Delta + 11 + t_0$.

{Case 2.1.} $1\not\in C(v_4)$.
It follows from Corollary~\ref{auv2}(5.2) and (3) that $T_3\subseteq C(u_1)$ and $t_0\ge t_3\ge 2$.
If $5\not\in C(u_1)$, then $(u u_1, v u_3, u v)\overset{\divideontimes}\to (5, 8, 1)$.
Otherwise, $5\in C(u_1)$.
Thus, $\sum_{x\in N(u)\setminus \{v\}}d(x)\ge 2\Delta + 11 + |\{5\}| + t_0\ge 2\Delta + 12 + 2 = 2\Delta + 14$, a contradiction.

{Case 2.2.} $1\in C(v_4)$.
It follows from Corollary~\ref{auv2d(u)5}(6) that $T_4\setminus C(v_4)\ne\emptyset$.
We may show that mult$_{S_u}(i)\ge 2$ for some $i\in \{3, 7\}$.
\begin{itemize}
\parskip=0pt
\item
    $H$ contains a $(7, j)_{(v_4, v_7)}$-path for some $j\in T_4\cap R_4$.
    If $H$ contains no $(2, 7)_{(u_3, u_2)}$-path, then $(u u_3, u v)\to (7, T_4\cap R_4)$.
    Otherwise, $7\in C(u_2)$ and $H$ contains a $(2, 7)_{(u_2, u_3)}$-path.
    It follows from Corollary~\ref{auv2}(5.1) that $7\in C(u_1)\cup C(u_4)$ and then mult$_{S_u}(7)\ge 2$.
\item
    $H$ contains no $(7, j)_{(v_4, v_7)}$-path for each $j\in T_4\cap R_4$.
    It follows from (1) that $H$ contains a $(3, 4)_{(u_4, v_3)}$-path and $3\in C(u_4)$,
    and follows from Corollary~\ref{auv2}(5.1) that $3\in C(u_2)$ or $H$ contains a $(2, 3)_{(u_1, u_2)}$-path, i.e., $3\in C(u_1)\cup C(u_2)$ and then mult$_{S_u}(3)\ge 2$.
\end{itemize}

Then $\sum_{x\in N(u)\setminus \{v\}}d(x)\ge 2\Delta + 11 + |\{7/3\}| + t_0\ge 2\Delta + 12 + t_0$.
One sees from $t_3\ge 2$ that $T_3\setminus C(u_1)\ne \emptyset$,
and from Corollary~\ref{auv2}(5.2) that that $1\in C(u_2)$.
Then $\sum_{x\in N(u)\setminus \{v\}}d(x)\ge 2\Delta + 12 + |\{1\}| + t_0 = 2\Delta + 13 +t_0$.
It follows that $\sum_{x\in N(u)\setminus \{v\}}d(x) = 2\Delta + 13$,
$6\not\in C(u_1)\cup C(u_2)$ and $5\not\in C(u_2)$.
Since $W_4 = \{4, 6, 1, 5, 3\}$ or $H$ contains a $(7, T_4\cap R_4)_{(v_4, v_7)}$-path,
$H$ contains no $(7, j)_{(u_4, v_7)}$-path for each $j\in T_4\cap R_4$.
One sees that mult$_{S_u}(6) = 1$.
It follows from Corollary~\ref{auv2}(4) that $H$ contains a $(5, T_4\cap R_4)_{(u_3, u_4)}$-path.
Then $(u u_1, u u_4, v u_3, u v)\overset{\divideontimes}\to (5, T_4\cap R_4, 8, 1)$.

Hence, (3) holds and likewise, (4) holds.

For (5), $W_3 = \{3, 5, a_3, 6, a\}$ and $W_4 = \{4, 6, 1, 5\}$.
Then $H$ contains a $(1, T_4)_{(u_1, u_4)}$-path and $T_4\subseteq C(u_1)$.
It follows from Corollary~\ref{auv2}(5.1) that $4, 7\in S_u$.
By (4), $H$ contains a $(5, T_4)_{(u_3, u_4)}$-path.
If $5\not\in C(u_1)$, one sees that $H$ contains a $(2, 5)_{(u_2, u_3)}$-path,
then $(u u_1, u u_4, v u_3, u v)\overset{\divideontimes}\to (5, T_4, 8, T_3)$.
If $5\not\in C(u_1)$, one sees that $H$ contains a $(1, 5)_{(u_1, u_3)}$-path,
then $(u u_2, u u_3, u v)\overset{\divideontimes}\to (5, T_3, 2)$.
Otherwise, $5\in C(u_1)\cap C(u_2)$.
Since $1, 5, 4/2/3, 7/2/3\in C(u_1)$, $T_3\setminus C(u_1)\ne\emptyset$.
It follows from Corollary~\ref{auv2}(5.2) that $a_3 = 2$ and $H$ contains a $(2, T_3)_{(u_2, u_3)}$-path,
and $1\in B_3\cup B_4$ and $1\in C(u_2)\cup C(u_3)$.
Since $2, 5, 3/1/4\in C(u_2)$, $T_4\setminus C(u_2)\ne\emptyset$.
It follows from Corollary~\ref{auv2}(5.2) that $2\in B_3$, $2\in C(u_2)$,
$3\in S_u$,
and from  Corollary~\ref{auv2d(u)5}(6) that $T_3\setminus C(v_3)\ne\emptyset$.
One sees that $\sum_{x\in N(u)\setminus \{v\}}d(x)\ge |\{1, 2, 5\}| + |\{2, 5\}| + |\{3, 5, 2, 6\}| + |\{4, 6, 1, 5\}| + |\{3, 4, 7, 1\}| + 2t_{uv} + t_0 = 17 + 2(\Delta - 2) + t_0 = 2\Delta + 13 + t_0$.
\begin{itemize}
\parskip=0pt
\item
    $H$ contains a $(7, j)_{(v_3, v_7)}$-path for some $j\in T_3\cap R_3$.
    If $H$ contains no $(1, 7)_{(u_1, u_4)}$-path, then $(u u_4, u v)\to (7, T_3\cap R_3)$.
    Otherwise, $7\in C(u_1)$ and $H$ contains a $(1, 7)_{(u_1, u_4)}$-path.
    It follows from Corollary~\ref{auv2}(5.1) that $7\in C(u_2)\cup C(u_3)$ and then mult$_{S_u}(7)\ge 2$.
\item
    $H$ contains no $(7, j)_{(v_3, v_7)}$-path for each $j\in T_3\cap R_3$.
    It follows from (2) that $H$ contains a $(3, 4)_{(u_3, v_4)}$-path and $4\in C(u_3)$,
    and follows from Corollary~\ref{auv2}(5.1) that $4\in C(u_1)$ or $H$ contains a $(2, 4)_{(u_1, u_2)}$-path, i.e., $4\in C(u_1)\cup C(u_2)$ and then mult$_{S_u}(4)\ge 2$.
\end{itemize}

Then for some $i\in \{4, 7\}$, mult$_{S_u}(i)\ge 2$ and $\sum_{x\in N(u)\setminus \{v\}}d(x)\ge 2\Delta + 13 + t_0 + |\{4/7\}| = 2\Delta + 14$, a contradiction.
Hence, (5) holds and likewise, (6) holds.

For (7), let $2, 7\in W_3$ and assume that $H$ contains a $(7, 8)_{(u_3, v_7)}$-path, where $8\in T_3$.
Then $t_4 = 4$, $c_{34} = 0$ and $d(u_3) = d(u_4) = \Delta = 7$.
If $7\not\in C(u_4)$ and $H$ contains no $(7, 1/2/3)_{(u_4, u)}$-path,
then $(u u_4, u v)\overset{\divideontimes}\to (7, 8)$.
Otherwise, $7\in C(u_4)$ and $H$ contains a $(7, 1/2/3)_{(u_4, u)}$-path.
If $7\not\in S_u\setminus C(u_3)$, since $H$ contains a $(7, 3)_{(u_4, u_3)}$-path,
then $(u u_1, u u_4, u v)\overset{\divideontimes}\to (7, T_4, 8)$.
Otherwise, $7\in S_u\setminus C(u_3)$.
\begin{itemize}
\parskip=0pt
\item
    $T_4\subseteq C(u_2)$.
    Then $C(u_2) = \{2, 3/1/4\}\cup T_3\cup T_4$.
    If $1\not\in B_4$, it follows from Corollary~\ref{auv2}(3) that $T_3\cup T_4\subseteq C(u_1)$,
    then $C(u_1) = \{1\}\cup T_3\cup T_4$, and $(u u_1, u v)\overset{\divideontimes}\to (5, 1)$.
    Otherwise, $1\in B_4\subseteq C(u_4)\cap C(v_4)$.
    When $3\in C(u_2)$,
    if $2\not\in C(u_4)$, then $(u u_2, u u_3, u v)\overset{\divideontimes}\to (5, T_3, 2)$.
    Otherwise, $2\in C(u_4)$ and it follows from \ref{6u35u4} (2) that $5\in C(u_1)$.
    It follows that $(T_3\cup T_4)\setminus C(u_1)\ne \emptyset$,
    and by Corollary~\ref{auv2}(5.2), we are done.
    If $4\in C(u_2)$ and then $W_4 = \{4, 6, 1, 3\}$, then $(u u_2, u u_3, u v)\overset{\divideontimes}\to (5, T_3, 2)$.
    Otherwise, $1\in C(u_2)$ and $H$ contains a $(1, 3)_{(u_1, u_2)}$-path.
    Since $1\in C(v_4)$, by Corollary~\ref{auv2d(u)5}(6), $T_4\setminus C(v_4)\ne \emptyset$.
    It follows from Corollary~\ref{auv2d(u)5}(2.1) that $3\in C(u_4)$ and then $W_4 = \{4, 6, 1, 3\}$.
    It follows that $5, 7\in C(u_1)$ and $C(u_1) = \{1, 3, 5, 7\}\cup T_4$.
    One sees that $H$ contains no $(1/3/4, T_4)_{(u_4, v)}$-path.
    It follows from Corollary~\ref{auv2}(4) that we are done.
\item
    $T_4\setminus C(u_2) \ne \emptyset$.
    One sees that $t_3\ge 3$ and $|\{1\}| + |\{2\}| + |\{3, 5, 2, 6, 7\}| + |\{4, 6, 1\}| + |\{3, 5, 2, 4, 7\}| + 2t_{uv} + t_0\ge 15 + 2(\Delta - 2) + t_0 = 2\Delta + 11 + t_0$.
    It follows that $T_3\setminus T_0\ne\emptyset$ and from Corollary~\ref{auv2}(5.2) and (3) that $1\in B_4$, and $1\in C(u_2)$.
    Then $\sum_{x\in N(u)\setminus \{v\}}d(x)\ge 2\Delta + 11 + t_0 + |\{1\}| = 2\Delta + 12 + t_0$ and it follows from Corollary~\ref{auv2d(u)5}(6) that $T_4\setminus C(v_4)\ne \emptyset$.
    \begin{itemize}
    \parskip=0pt
    \item
        $3, 5\not\in C(u_4)$.
        One sees from \ref{6u35u4}(2) that $H$ contains a $(1, 5)_{(u_1, u_4)}$-path,
        if $5\not\in C(u_2)$, then $(u u_2, u u_3, u v)\overset{\divideontimes}\to (5, T_3, T_4\cap T_2)$.
        Otherwise, $5\in C(u_2)$.
        Then $\sum_{x\in N(u)\setminus \{v\}}d(x)\ge 2\Delta + 12 + t_0 + |\{5\}| = 2\Delta +13 +t_0$.
        It follows that mult$_{S_u}(i) = 1$, $i\in [3, 4]$.
        Note that $4\not\in C(u_2)$ or $H$ contains a $(2, 4)_{(u_1, u_2)}$-path,
        and $3\not\in C(u_1)$ or $H$ contains a $(1, 3)_{(u_1, u_2)}$-path.
        Then by Corollary~\ref{auv2d(u)5}(2.1), we are done.
    \item
        $W_4 = \{4, 6, 1, 3\}$. Then $5\in C(u_1)$ and it follows from Corollary~\ref{auv2}(4) that $6\in C(u_1)$. Then $\sum_{x\in N(u)\setminus \{v\}}d(x)\ge 2\Delta + 12 + t_0 + |\{6\}| = 2\Delta +13 + t_0$.
        It follows that $C(u_1) = \{1, 2, 5, 6, 7\}\cup T_4$ and $C(u_2) = \{2, 1, 4\}\cup T_3$.
        Then $(u u_2, u u_1, u u_3, u v)\overset{\divideontimes} \to (5, 3, T_3, T_4\cap T_2)$.
    \item
        $W_4 = \{4, 6, 1, 5\}$. Then $7\in C(u_1)$ and $H$ contains a $(1, 7)_{(u_1, u_4)}$-path.
        If $7\not\in C(u_2)$, then $(u u_2, u u_3, u v)\to (7, 8, T_4\cap T_2)$.
        Otherwise, $7\in C(u_2)$.
        Then $\sum_{x\in N(u)\setminus \{v\}}d(x)\ge 2\Delta + 12 + t_0 + |\{7\}| = 2\Delta +13 + t_0$.
        It follows that mult$_{S_u}(i) = 1$, $i\in [3, 4]$.
        Note that $4\not\in C(u_2)$ or $H$ contains a $(2, 4)_{(u_1, u_2)}$-path,
        and $3\not\in C(u_1)$ or $H$ contains a $(1, 3)_{(u_1, u_2)}$-path.
        Then by Corollary~\ref{auv2d(u)5}(2.1), we are done.
    \end{itemize}
\end{itemize}

Hence, (7) holds and likewise, (8) holds.
\end{proof}

By symmetry, we can have the following two scenarios.

(2.1.2.1) $W_3 = \{3, 5, 2, 6\}$.
We distinguish whether or not $5\in C(u_4)$:
\begin{itemize}
\parskip=0pt
\item
     $5\in C(u_4)$.
     First, assume that $W_4 = \{4, 6, 5, 1, 2\}$.
     Then by Lemma~\ref{6T33425} (4) and (6),
     $H$ contains a $(6, T_3)_{(u_3, u_4)}$-path and a $(5, T_3)_{(u_3, u_4)}$-path.
     \begin{itemize}
     \parskip=0pt
     \item
          $1\not\in B_4$, it follows from Corollary~\ref{auv2}(5.2) and (3) that $T_3\cup T_4\subseteq C(u_1)$.
          One sees that $|\{1\}| + |\{2\}| + |\{3, 5, 2, 6\}| + |\{4, 6, 1, 2, 5\}| + |\{3, 7\}| + 2t_{uv} + t_0\ge 13 + 2(\Delta - 2) + 2 = 2\Delta + 9 + 2 = 2\Delta + 11$.
          Since $t_4\ge 3$, $T_4\setminus C(u_2)\ne \emptyset$.
          It follows from  Corollary~\ref{auv2}(5.1) that $4\in C(u_1)$ or $H$ contains a $(2, 4)_{(u_1, u_2)}$-path.
          If $H$ contains no $(2, 5)_{(u_1, u_2)}$-path,
          then $(u u_1, u u_4, u v)\overset{\divideontimes}\to (5, T_4\cap T_3, 1)$.
          Otherwise, $H$ contains a $(2, 5)_{(u_1, u_2)}$-path and $2\in C(u_1)$.
          It follows that $4, 3, 7, 5\in C(u_2)$ and $d(u_2)\ge 7$, a contradiction.
     \item
         $1\in B_4\subseteq C(u_4)\cap C(v_4)$.
         It follows from  Corollary~\ref{auv2d(u)5}(6) that $T_4\setminus C(v_4)\ne \emptyset$,
         and then from  Lemma~\ref{6T33425}(1) that $H$ contains a $(7, T_4\cap R_4)_{(v_4, v_7)}$-path.
         If $H$ contains no $(2, 7)_{(u_2, u_3)}$-path,
         then $(u u_3, u v)\overset{\divideontimes}\to (7, T_4\cap R_4)$.
         Otherwise, $H$ contains a $(2, 7)_{(u_2, u_3)}$-path and $7\in C(u_2)$.
         It follows from  Corollary~\ref{auv2}(5.1) that $7\in C(u_1)$.
         Since $2, 7, 3/1\in C(u_2)$, $T_4\setminus C(u_2)\ne\emptyset$.
         It follows that $a_4 = 1$, $H$ contains a $(1, T_4)_{(u_1, u_4)}$-path,
         and from Corollary~\ref{auv2}(5.1) that $4\in C(u_1)$ or $H$ contains a $(2, 4)_{(u_1, u_2)}$-path.
         Since $1, 7, 4/2\in C(u_1)$, $T_3\setminus C(u_1)\ne\emptyset$.
         It follows from Corollary~\ref{Ti1i4Ti2i3} that $H$ contains a $(T_4\cap T_2, T_3\cap T_1)_{(u_2, u_3)}$-path and a $(T_3\cap T_1, T_4\cap T_2)_{(u_1, u_4)}$-path.
         If $5\not\in C(u_2)$, then $vu_4\in \alpha\in T_4\cap T_2$, and $(vu_3, uu_3, uv)\overset{\divideontimes}\to (T_3\cap T_1, 5, T_4\setminus \{\alpha\})$.
         If $6\not\in C(u_2)\cup C(u_1)$, then $vu_3\in \beta\in T_3\cap T_1$,
         and $(v u_4, u u_4, uv)\overset{\divideontimes}\to (T_4\cap T_2, 6, T_3\setminus \{\beta\})$.
         Otherwise, $5\in C(u_2)$ and $6\in C(u_1)\cup C(u_2)$.
         Thus, $\sum_{x\in N(u)\setminus \{v\}}d(x)\ge |\{1, 7\}| + |\{1, 2, 5, 7\}| + |\{3, 5, 2, 6\}| + |\{4, 6, 1, 2, 5\}| + |\{4, 5, 6\}| + 2t_{u v} + t_0\ge 18 + 2(\Delta - 2)= 2\Delta + 14$.
     \end{itemize}

     \quad Next, assume that $|W_4\cap [1, 2]| = 1$.
     When $W_4\cap [1, 2] = \{2\}$, then $T_3\cup T_4\subseteq C(u_2)$.
     It follows from Corollary~\ref{auv2}(5.2) and (3) that $T_3\cup T_4\subseteq C(u_1)$,
     from Corollary~\ref{auv2}(5.1) that $3, 4, 7\in S_u$ and for each $i\in \{3, 4, 7\}$, $i\in C(u_2)$ or $H$ contains a $(i, 1)_{(u_1, u_2)}$-path.
     One sees that $\sum_{x\in N(u)\setminus \{v\}}d(x)\ge |\{1\}| + |\{2\}| + |\{3, 5, 2, 6\}| + |\{4, 6, 2, 5\}| + |\{3, 4, 7\}| + 2t_{u v} + t_0\ge 13 + 2(\Delta - 2) + 4 = 2\Delta + 13$.
     It follows that $\sum_{x\in N(u)\setminus \{v\}}d(x) = 2\Delta + 13$, $t_3 = t_4 = 2$, $\Delta = 7$,
     and $c_{34} = 1$, $1\not\in S_u$.
     Then $3, 4, 7\in C(u_2)$ and $d(u_2)\ge t_3 + t_4 + 4 =8 > 7$.
     In the other case, $W_4\cap [1, 2] = \{1\}$.
     Then $H$ contains a $(1, T_4)_{(u_1, u_4)}$-path,
     and it follows from Corollary~\ref{auv2}(5.1) that $4\in C(u_1)$ or $H$ contains a $(2, 4)_{(u_1, u_2)}$-path.
     For $j\in \{5, 6\}$, if $j\not\in C(u_2)$ and $H$ contains no $(j, i)_{(u_2, u_i)}$-path for each $i\in \{1, 4\}$,
     then $(u u_2, u u_3, u v)\to (j, T_3, 2)$;
     if $j\not\in C(u_1)$ and $H$ contains no $(j, i)_{(u_1, u_i)}$-path for each $i\in [2, 3]$,
     then $(u u_1, u u_4, u v)\to (j, T_4, T_3)$.
     Otherwise, $j\in C(u_2)$ or $H$ contains a $(j, i)_{(u_2, u_i)}$-path for some $i\in \{1, 4\}$,
     and  $j\in C(u_1)$ or $H$ contains a $(j, i)_{(u_1, u_i)}$-path for some $i\in [2, 3]$.
     Then $6/1/4, 5/1/4\in C(u_2)$ and $5/2/3, 6/2/3\in C(u_1)$.
     \begin{itemize}
     \parskip=0pt
     \item
          $C_{34} = \{12\}$.
          It follows from Corollary~\ref{auv2}(5.1) that $3/1, 7/1\in C(u_2)$ and $4/2, 7/2\in C(u_1)$.
          If $12\not\in C(u_1)\cup C(u_2)$, then $u u_1\to 12$ or $u u_2\to 12$ reduces the proof to the above case where $1, 2\in C(u_4)$.
          Otherwise, $12\in C(u_1)\cup C(u_2)$.

          \quad When $1\not\in C(u_1)$, then $3, 7\in C(u_1)$ and it follows from Corollary~\ref{auv2}(5.2) that $T_3\subseteq C(u_1)$.
          One sees that $\sum_{x\in N(u)\setminus \{v\}}d(x)\ge |\{1\}| + |\{2, 3, 7\}| + |\{3, 5, 2, 6\}| + |\{4, 6, 1, 5\}| + |\{4, 12\}| + 2t_{u v} + t_0\ge 14 + 2(\Delta - 2) + 2 = 2\Delta + 12$.
          It follows that $T_4\setminus C(u_2)\ne \emptyset$ and it follows from Corollary~\ref{auv2}(5.2) that $2\in C(u_1)$.
          Thus, $\sum_{x\in N(u)\setminus \{v\}}d(x)\ge2\Delta + 12 + |\{2\}| = 2\Delta + 13$.
          It follows that $5\not\in C(u_1)\cup C(u_2)$.
          One sees from  Lemma~\ref{Ti1i4Ti2i3} that $H$ contains a $(T_2\cap T_4, T_3)_{(u_2, u_3)}$-path.
          Then $vu_4\to \alpha\in T_4\cap T_2$, $(v u_3, u u_3, u v)\overset{\divideontimes}\to (T_3, 5, T_4\setminus\{\alpha\})$.

          \quad In the other case, $1\in C(u_2)$ and likewise $2\in C(u_1)$.
          One sees that $\sum_{x\in N(u)\setminus \{v\}}d(x)\ge |\{1, 2\}| + |\{2, 1\}| + |\{3, 5, 2, 6\}| + |\{4, 6, 1, 5\}| + |\{4, 12, 3, 7\}| + 2t_{uv} + t_0\ge 16 + 2(\Delta - 2) = 2\Delta + 12$.
          It follows that $T_3\setminus C(u_1)\ne \emptyset$ and $T_4\setminus C(u_2)\ne \emptyset$.
          It follows from Corollary~\ref{Ti1i4Ti2i3} that $H$ contains a $(T_4\cap T_2, T_3\cap T_1)_{(u_2, u_3)}$-path and a $(T_3\cap T_1, T_4\cap T_2)_{(u_1, u_4)}$-path.
          If $5\not\in C(u_2)$, then $v u_4\in \alpha\in T_4\cap T_2$,
          and $(vu_3, uu_3, u v)\overset{\divideontimes}\to (T_3\cap T_1, 5, T_4\setminus \{\alpha\})$.
          If $6\not\in C(u_1)$, then $vu_3\in \beta\in T_3\cap T_1$, and $(vu_4, uu_4, uv)\overset{\divideontimes}\to (T_4\cap T_2, 6, T_3\setminus \{\beta\})$.
          Otherwise, $5\in C(u_2)$ and $6\in C(u_1)$.
          Thus, $\sum_{x\in N(u)\setminus \{v\}}d(x)\ge 2\Delta + 12 + |\{5, 6\}| = 2\Delta + 14$.
     \item
         $c_{34} = 0$, i.e., $t_3 + t_4 = \Delta - 2$.
         If $T_4\subseteq C(u_2)$,
         then $C(u_2) = \{2, 1\}\cup T_4\cup T_3$ and $C(u_1) = \{1, 4, 3, 5, 6, 7\}\cup T_4$, i.e.,
         $d(u_2) = \Delta$, $d(u_1)\ge 8$, a contradiction.
         Otherwise, $T_4\setminus C(u_2)\ne \emptyset$,
         and it follows from Corollary~\ref{auv2}(5.2) that $2\in C(u_1)$.
         If $T_3\subseteq C(u_1)$,
         then $C(u_1) = \{2, 1\}\cup T_4\cup T_3$, $4, 5, 6\in C(u_2)$, $3, 7\in C(u_2)\cup C(u_4)$,
         and $\sum_{x\in N(u)\setminus \{v\}}d(x)\ge |\{1, 2\}| + |\{2, 4, 5, 6\}| + |\{3, 5, 2, 6\}| + |\{4, 6, 1, 5\}| + |\{3, 7\}| + 2t_{u v} + t_0\ge 16 + 2(\Delta - 2) + 2 = 2\Delta + 14$, a contradiction.
         Otherwise, $T_3\setminus C(u_1)\ne \emptyset$ and it follows from Corollary~\ref{auv2}(5.2) that $1\in C(u_2)$, from Corollary~\ref{auv2d(u)5}(6) that $T_4\setminus C(v_4)\ne \emptyset$.

         \quad One sees from Corollary~\ref{Ti1i4Ti2i3} that $H$ contains a $(T_4\cap T_2, T_3\cap T_1)_{(u_2, u_3)}$-path and a $(T_3\cap T_1, T_4\cap T_2)_{(u_1, u_4)}$-path.
         If $H$ contains no $(2, 5)_{(u_2, u_3)}$-path, then $vu_4\in \alpha\in T_4\cap T_2$,
         and $(v u_3, u u_3, u v)\overset{\divideontimes}\to (T_3\cap T_1, 5, T_4\setminus \{\alpha\})$.
         If $H$ contains no $(1, 6)_{(u_1, u_4)}$-path, then $v u_3\in \beta\in T_3\cap T_1$,
         and $(v u_4, u u_4, u v)\overset{\divideontimes}\to (T_4\cap T_2, 6, T_3\setminus \{\beta\})$.
         Otherwise, $5\in C(u_2)$, $6\in C(u_1)$, and $H$ contains a $(2, 5)_{(u_2, u_3)}$-path, a $(1, 6)_{(u_1, u_4)}$-path.
         Then $\sum_{x\in N(u)\setminus \{v\}}d(x)\ge |\{1, 2, 6\}| + |\{2, 1, 5\}| + |\{3, 5, 2, 6\}| + |\{4, 6, 1, 5\}| + |\{3, 4, 7\}| + 2t_{uv} + t_0\ge 17 + 2(\Delta - 2) = 2\Delta + 13$.
         It follows that for each $i\in [3, 4]$, mult$_{S_u}(i)  = 1$, and $i\in [5, 6]$, mult$_{S_u}(i)  = 2$.
         Hence, $3\in C(u_1)$, $4\in C(u_1)$.
         Then by Corollary~\ref{auv2d(u)5}(2.1), we are done.
     \end{itemize}
\item
     $5\not\in C(u_4)$.
     \begin{itemize}
     \parskip=0pt
     \item
          $C_{34} = 12$. Then $d(u_3) = d(u_4) = \Delta = 7$ and $t_3 = t_4 = 2$.
          When for each $j\in T_4$, $H$ contains no $(7, j)_{(v_4, v_7)}$-path,
          it follows from Corollary~\ref{auv2}(4) that $6\in B_3$ $H$ contains a $(6, i)_{(u_4, u_i)}$-path for some $i\in [1, 2]$.
          If $6\not\in C(u_1)$, one sees that $H$ contains a $(2, 6)_{(u_2, u_4)}$-path,
          then $(u u_1, v u_4, u v)\overset{\divideontimes}\to (6, T_4, 1)$.
          If $6\not\in C(u_2)$, one sees that $H$ contains a $(1, 6)_{(u_1, u_4)}$-path,
          then $v u_4\to \alpha\in T_4$ and $(u u_2, u u_3, u v)\overset{\divideontimes}\to (6, T_3, T_4\setminus \{\alpha\})$.
          Otherwise, $6\in C(u_1)\cap C(u_2)$.
          If $H$ contains no $(i, j)_{(u_3, u_4)}$-path for some $i\in T_3$ and $j\in T_4$,
          then $(v u_3, v u_4, u v)\overset{\divideontimes}\to (i, j, 5)$.
          Otherwise, $H$ contains a $(T_3, T_4)_{(u_3, u_4)}$-path.
          It follows from Corollary~\ref{Ti1i4Ti2i3} that $T_4\subseteq C(u_2)$ and then $C(u_1) = \{1, 2, 6\}\cup T_3\cup T_4$, $3, 7\in C(u_2)$.
          One sees that $\sum_{x\in N(u)\setminus \{v\}}d(x)\ge |\{1, 3, 7, 6\}| + |\{2, 1, 6\}| + |\{3, 5, 2, 6\}| + |\{4, 6, a_1, b\}| + 2t_{u v} + t_0\ge 15 + 2(\Delta - 2) = 2\Delta + 11$.
          It follows that $(T_3\cup T_4)\setminus C(u_1)\ne \emptyset$,
          and from Corollary~\ref{auv2}(4) that $1\in B_4$.
          Hence, $1\in C(u_4)\cap C(v_4)$ and it follows from Corollary~\ref{auv2d(u)5}(6) that $T_4\setminus C(v_4)\ne \emptyset$.
          One sees that $4\not\in C(u_2)$ and $H$ contains a $(1, 3)_{(u_1, u_2)}$-path.
          Then by Corollary~\ref{auv2d(u)5}(2.1), $3\in C(u_4)$ and then $T_4\subseteq C(u_1)$, $C(u_1) = \{1, 6, 3, 7\}\cup T_4$.
          Then $(u u_1, u u_4)\overset{\divideontimes}\to (4, T_4, T_3)$.

          \quad In the other case, $7\in C(u_4)$ and $H$ contains a $(7, 10)_{(u_4, v_7)}$-path, where $10\in T_4$. If $H$ contains no $(2, 7)_{(u_2, u_3)}$-path, then $(u u_3, u v)\overset{\divideontimes}\to (7, 10)$.
          Otherwise, $H$ contains a $(2, 7)_{(u_2, u_3)}$-path.
          When $2\in C(u_4)$, then $T_4\subseteq C(u_4)$.
          It follows from Corollary~\ref{auv2}(5.1) that $3/1, 4/1\in C(u_2)$ and from Corollary~\ref{auv2}(5.2) that $T_3\cup T_4\subseteq C(u_1)$.
          Thus, $C(u_1) = \{2, 7, 1\}\cup T_3\cup T_4$, $\{1, 3, 4\}\cup T_3\cup T_4\subseteq C(u_1)$, a contradiction.
          In the other case, $1\in C(u_4)$, $H$ contains a $(1, T_4)_{(u_1, u_4)}$-path,
          and it follows from Corollary~\ref{auv2}(5.1) that $4\in C(u_1)$ or $H$ contains a $(2, 4)_{(u_1, u_2)}$-path,
          from \ref{6u35u4}(2) that $H$ contains a $(1, 5)_{(u_1, u_4)}$-path, $5\in C(u_1)$.
          If $7\not\in C(u_1)$, then $(u u_1, u u_4, u v)\overset{\divideontimes}\to (7, 10, 1)$.
          If $5\not\in C(u_2)$, then $(u u_2, u u_3, u v)\overset{\divideontimes}\to (5, T_3, 2)$.
          Otherwise, $5\in C(u_2)$, $7\in C(u_1)$,
          and then $T_3\setminus C(u_1)\ne \emptyset$, $T_4\setminus C(u_2)\ne \emptyset$.
          It follows that $1\in B_4$ and from Corollary~\ref{auv2}(5.2) that $1\in C(u_2)$, $2\in C(u_1)$,
          from Corollary~\ref{auv2d(u)5}(6) that $T_{uv}\setminus C(v_4)\ne \emptyset$.
          Then $\sum_{x\in N(u)\setminus \{v\}}d(x)\ge |\{1, 2, 5, 7\}| + |\{2, 7, 1, 5\}| + |\{3, 5, 2, 6\}| + |\{4, 6, 1, 7\}| + |\{3, 4\}| + 2t_{uv} + t_0\ge 18 + 2(\Delta - 2) = 2\Delta + 14$.
     \item
         $c_{34} = 0$, i.e., $t_3 + t_4 = \Delta - 2$.
         When $T_4\subseteq C(u_2)$, it follows that $C(u_2) = \{2, 1\}\cup T_4\cup T_3$ and $3, 7, 5\in C(u_1)$.
         One sees that $\sum_{x\in N(u)\setminus \{v\}}d(x)\ge |\{1, 3, 5, 7\}| + |\{2, 1\}| + |\{3, 5, 2, 6\}| + |\{4, 6, a_1, b\}| + 2t_{u v} + t_0\ge 14 + 2(\Delta - 2) = 2\Delta + 10$.
         Then $(T_3\cup T_4)\setminus C(u_1)\ne \emptyset$.
         It follows from Corollary~\ref{auv2}(5.2) and (3) that $1\in B_4$,
         and from Corollary~\ref{auv2d(u)5}(6) that $T_4\setminus C(v_4)\ne \emptyset$.
         However, by Lemma~\ref{6T33425}(1), we are done since $5, 7\not\in C(u_2)$.

         \quad In the other case, $T_4\setminus C(u_2)\ne \emptyset$.
         It follows that $1\in C(u_4)$, $H$ contains a $(1, T_4)_{(u_1, u_4)}$-path,
         and from Corollary~\ref{auv2}(5.2) that $2\in C(u_1)\cup C(u_4)$,
         from Corollary~\ref{auv2}(5.1) that $4\in C(u_1)$ or $H$ contains a $(2, 4)_{(u_1, u_2)}$-path.
         When $T_3\subseteq C(u_1)$, i.e., $C(u_1) = \{1, 2/4\}\cup T_3\cup T_4$,
         if $H$ contains no $(2, 5)_{(u_1, u_2)}$-path, then $(u u_1, u u_4, u v)\overset{\divideontimes}\to (5, T_4\cap T_2, 1)$.
         if $H$ contains no $(2, 6)_{(u_1, u_2)}$-path, then $(u u_1, u u_4, u v)\overset{\divideontimes}\to (6, T_4\cap T_2, 1)$.
         Otherwise, $2\in C(u_1)$ and $H$ contains a $(2, [5, 6])_{(u_1, u_2)}$-path, $4, 5, 6\in C(u_2)$.
         It follows from \ref{6u35u4}(2) that $3\in C(u_4)$ and $H$ contains a $(3, 5)_{(u_4, u_3)}$-path.
         One sees that $\sum_{x\in N(u)\setminus \{v\}}d(x)\ge |\{1, 2\}| + |\{2, 4, 5, 6\}| + |\{3, 5, 2, 6\}| + |\{4, 6, 1, 3\}| + |\{7\}| +  2t_{u v} + t_0\ge 15 + 2(\Delta - 2) + 2 = 2\Delta + 13$.
         It follows that $3\not\in C(u_2)$.
         Then $(u u_2, u u_3, u u_4, u v)\overset{\divideontimes}\to (3, T_3, 5, 2)$.

         \quad Now assume that $T_3\setminus C(u_1)\ne \emptyset$.
         It follows from Corollary~\ref{auv2}(5.2) that $1\in C(u_2)$,
         from Corollary~\ref{auv2}(5.2) and (3) that $1\in B_4$.
         Further, it follows from Corollary~\ref{auv2d(u)5}(6) that $T_4\setminus C(v_4)\ne \emptyset$.
         If $H$ contains no $(7, j)_{(u_4, v_7)}$-path for each $j\in T_4\cap T_2$,
         one sees from Corollary~\ref{Ti1i4Ti2i3} that $H$ contains a $(T_1\cap T_3, T_4\cap T_2)_{(u_1, u_4)}$-path,
         then $(v u_3, v u_4, u v)\overset{\divideontimes}\to (T_3\cap T_1, T_4\cap T_2, 5)$.
         Otherwise, $7\in C(u_4)$ and $H$ contains a $(7, 10)_{(u_4, v_7)}$-path, where $10\in T_4\cap T_2$.
         If $H$ contains no $(2, 7)_{(u_2, u_3)}$-path, then $(u u_3, u v)\overset{\divideontimes}\to (7, 10)$.
         Otherwise, $H$ contains a $(2, 7)_{(u_2, u_3)}$-path, $7\in C(u_2)$.
         If $7\not\in C(u_1)$, then $(u u_1, u u_4, u v)\overset{\divideontimes}\to (7, 10, 1)$.
         Otherwise, $7\in C(u_1)$.
         One sees that $\sum_{x\in N(u)\setminus \{v\}}d(x)\ge |\{1, 7\}| + |\{2, 1, 7\}| + |\{3, 5, 2, 6\}| + |\{4, 6, 1, 7\}| + |\{3, 4, 2, 5\}| +  2t_{u v} + t_0\ge 17 + 2(\Delta - 2) = 2\Delta + 13$.
         It follows that $t_0 = 0$, for each $i\in [2, 5]$, mult$_{S_u}(i) = 1$,
         and $6\not\in S_u\setminus C(u_3)$.
         If $4\not\in C(u_2)$, then $(u u_2, u u_3, u v)\overset{\divideontimes}\to (6, T_3, T_4)$.
         If $3\not\in C(u_1)$, then $(u u_1, u u_4, u v)\overset{\divideontimes}\to (6, T_4, 1)$.
         Otherwise, $4\in C(u_2)$, $3\in C(u_1)$, and $H$ contains a $(2, 4)_{(u_1, u_2)}$-path, a $(1, 3)_{(u_1, u_2)}$-path.
         Recall that $T_4\setminus C(v_4)\ne \emptyset$.
         Then by Corollary~\ref{auv2d(u)5}(2.1), we are done.
     \end{itemize}
\end{itemize}

(2.1.2.2) $W_3 = \{3, 5, a_3, 6, a\}$ and $W_4 = \{4, 6, 1, b\}$ with $5\not\in C(u_4)$.
Then $t_3 = 3$, $t_4 = 2$, $c_{34} = 0$ and $d(u_3) = d(u_4) = \Delta = 7$.
One sees that $|\{1\}| + |\{2\}| + |\{3, 5, a_3, 6, a\}| + |\{4, 6, 1, b\}| + |\{5\}| + 2t_{u v} + t_0\ge 12 + 2(\Delta - 2) + t_0 = 2\Delta + 8 + t_0$.
One sees from Lemma~\ref{6T33425} (7) that if $H$ contains no $(7, j)_{(u_4, v_7)}$-path,
then $H$ contains a $(T_3, T_4)_{(u_3, u_4)}$-path.
\begin{itemize}
\parskip=0pt
\item
     $W_4 = \{4, 6, 1, 2\}$.
     It follows from \ref{6u35u4}(2) that $H$ contains a $(5, i)_{(u_4, u_i)}$-path for some $i\in [1, 2]$,
     $5\in C(u_1)\cup C(u_2)$,
     and from Corollary~\ref{auv2}(4) that $H$ contains a $(6, i)_{(u_4, u_i)}$-path for some $i\in [1, 2]$,
     $6\in C(u_1)\cup C(u_2)$.
     When $T_4\subseteq C(u_1)\cap C(u_2)$, we have $t_0\ge t_4\ge 2$,
     $\sum_{x\in N(u)\setminus \{v\}}d(x)\ge 2\Delta + 8 + t_0 + |\{6\}| = 2\Delta + 9 + t_0\ge 2\Delta + 9 + 2 = 2\Delta + 11$.
     One sees from $t_3\ge 3$ that $T_3\setminus T_0 \ne \emptyset$.
     Assume w.l.o.g. $T_3\setminus C(u_1)\ne \emptyset$.
     Then $H$ contains a $(2, T_3)_{(u_2, u_3)}$-path.
     It follows from Corollary~\ref{auv2}(5.1) that $3/1/4\in C(u_2)$ and then $C(u_2) = \{2, 3/1/4\}\cup T_3\cup T_4$.
     Then $u u_1\to \gamma\in T_3\cap T_1$ and $(u u_4, u v)\overset{\divideontimes}\to (5, T_3\setminus \{\gamma\})$.
     In the other case, assume w.l.o.g. $T_4\subseteq C(u_1)$ and $T_4\setminus C(u_2)\ne \emptyset$.
     It follows from Corollary~\ref{auv2}(5.1) that $4/2/3\in C(u_1)$.
     If $T_3\cap T_1\ne \emptyset$, then $uu_1\to \gamma\in T_3\cap T_1$,
     $(u u_4, u v)\overset{\divideontimes}\to (T_4\cap T_2, T_3\setminus \{\gamma\})$.
     Otherwise, $T_3\subseteq C(u_1)$ and then $C(u_1) = \{1, 4/2/3\}\cup T_3\cup T_4$.
     Then $(u u_2, u u_4, u v)\to (T_4\cap T_2, 5, T_3)$.
\item
     $W_4 = \{4, 6, 1, 3/7\}$.
     Then $T_4\subseteq C(u_1)$ and it follows from Corollary~\ref{auv2}(5.1) that $4/2/3\in C(u_1)$.
     When $H$ contains no $(7, j)_{(u_4, v_7)}$-path, then $H$ contains a $(T_3, T_4)_{(u_3, u_4)}$-path.
     If $T_3\cap T_1\ne \emptyset$, then $uu_1\to \gamma\in T_3\cap T_1$, $(u u_4, u v)\overset{\divideontimes}\to (T_4, T_3\setminus \{\gamma\})$.
     Otherwise, $T_3\subseteq C(u_1)$ and then $C(u_1) = \{1, 4/2/3\}\cup T_3\cup T_4$.
     Then $(u u_4, u v)\overset{\divideontimes}\to (5, T_3)$ if $7\in C(u_4)$,
     or $(u u_4, v u_4, u v)\overset{\divideontimes}\to (6, T_4, T_3)$.
     In the other case, $W_4 = \{4, 6, 1, 7\}$ and $H$ contains a $(7, 8)_{(u_4, v_7)}$-path, where $8\in T_4$.
     Then it follows from \ref{6u35u4}(2) that $H$ contains a $(1, 5)_{(u_1, u_4)}$-path and $5\in C(u_1)$.
     Since $1, 4/2/3, 5\in C(u_1)$, $T_3\setminus C(u_1)\ne\emptyset$.
     It follows that $a_3 = 2$, $H$ contains a $(2, T_3)_{(u_2, u_3)}$-path,
     and from Corollary~\ref{auv2}(5.1) that $3\in C(u_2)$ or $H$ contains a $(3, 1)_{(u_1, u_2)}$-path,
     from Corollary~\ref{auv2}(5.2) and (3) that $1\in B_3\cup B_4$, $1\in C(u_2)\cup C(u_3)$.
     Further, it follows Corollary~\ref{auv2d(u)5}(6) that $T_4\setminus C(v_4)\ne \emptyset$.
     If $7\not\in C(u_1)$ and $H$ contains no $(7, i)_{(u_1, u_i)}$-path for each $i\in [2, 3]$,
     then $(u u_1, u u_4, u v)\overset{\divideontimes}\to (7, 8, T_3\cap T_1)$.
     Otherwise, $7\in C(u_1)$ or $H$ contains a $(7, i)_{(u_1, u_i)}$-path for some $i\in [2, 3]$.
     If $5\not\in C(u_2)$, then $(u u_2, u u_3, u v)\to (5, T_3, 2)$.
     Otherwise, $5\in C(u_2)$ and then $T_4\setminus C(u_2)\ne\emptyset$.
     It follows from Corollary~\ref{auv2}(5.2) that $2\in B_3\cap C(u_1)$.
     One sees that $\sum_{x\in N(u)\setminus \{v\}}d(x)\ge |\{1, 2, 5\}| + |\{2, 5\}| + |\{3, 5, 2, 6\}| + |\{4, 6, 1, 7\}| + |\{1, 3, 4, 7\}| + 2t_{u v} + t_0\ge 17 + 2(\Delta - 2) + t_0 = 2\Delta + 13 + t_0$.
     It follows that $t_0 = 0$, for each $i\in \{1, 3, 4, 7\}$, mult$_{S_u}(i) = 1$, and $6\not\in C(u_1)$.
     If $H$ contains a $(7, j)_{(v_4, v_7)}$-path for some $j\in T_4\cap R_4$,
     then $(v u_4, u u_4, u v)\overset{\divideontimes}\to (j, 6, T_3)$.
     Otherwise, $H$ contains no $(7, j)_{(v_4, v_7)}$-path.
     Since $3\not\in C(u_4)$, it follows from Lemma~\ref{6T33425}(1) that $7\in C(u_3)$.
     By Corollary~\ref{auv2d(u)5}(2.1), we are done.
\end{itemize}

{Case 2.2.} $\{i_1, i_2\} = \{1, 3\}$ and $c(v u_3) = 1$.
Then $C(u_3) = \{1, 3\}\cup T_{u v}$.
There are two subcases on whether $c(v u_4) = 3$ or not.

(2.2.1.) $c(v u_4) = 3$.
If $C(u_4) = \{3, 4\}\cup T_{u v}$, then $(u u_4, u v)\overset{\divideontimes}\to (5, 4)$.
Otherwise, $T_4\ne\emptyset$ and assume w.l.o.g. $8\in T_4$.
It follows from Lemma~\ref{auvge2}(1.2) that w.l.o.g. $H$ contains a $(8, 5)_{(u_4, v_5)}$-path and $5\in C(u_4)$.
If $H$ contains no $(2, 8)_{(u_2, u_4)}$-path, then $(u u_4, u v)\overset{\divideontimes}\to (8, 4)$.
Otherwise, $H$ contains a $(2, 8)_{(u_2, u_4)}$-path and $8\in C(u_2)$, $2\in C(u_4)$.
It follows that $T_4\subseteq C(u_2)$.
If there exists a $i\in \{3, 5\}\setminus C(u_2)$,
then $(v u_3, u u_3, u v)\overset{\divideontimes}\to (2, i, 8)$.
Otherwise, $3, 5\in C(u_2)$.
If there exists a $i\in [6, 7]\setminus (C(u_2)\cup C(u_4))$,
then $(u u_2, u v)\to (i, 2)$ or $(u u_4, u v)\to (i, 4)$.
Otherwise, $6, 7\in C(u_2)\cup C(u_4)$.
One sees that $d(u_2) + d(u_4)\ge |\{2, 3, 5\}| + |\{4, 3, 2, 5\}| + |\{6, 7\}| + t_{u v} = \Delta + 7$.
It follows that $d(u_2) = d(u_4) = \Delta = 7$ and $1, 4\not\in C(u_2)$.
Then $(u u_2, u u_4, u v)\overset{\divideontimes}\to (4, 8, 2)$.

(2.2.2.) $c(v u_4) = 5$.
If $H$ contains no $(i, 5)_{(u_1, u_4)}$-path for some $i\in [1, 2]$,
then $(u u_3, v u_3)\to (5, i)$ reduces the proof to (2.2.1.).
Otherwise, $H$ contains a $(i, 5)_{(u_1, u_4)}$-path for each $i\in [1, 2]$.
It follows that $5\in C(u_1)\cap C(u_2)$ and $1, 2\in C(u_4)$.
If $H$ contains no $(i, 4)_{(u_4, v_i)}$-path for some $i\in \{3, 6, 7\}$,
then $(v u_3, u u_3)\to (4, i)$ reduces the proof to (2.1.1.).
Otherwise, $H$ contains a $(i, 4)_{(u_4, v_i)}$-path for each $i\in \{3, 6, 7\}$.
It follows that $3, 6, 7\in C(u_4)$ and then $C(u_4) = [1, 7]$.
Since $2, 4\not\in C(u_3)$, it follows from Corollary~\ref{auv2d(u)5}(4.2) that $T_{u v}\subseteq C(u_2)$.
Hence, $C(u_2) = \{2, 5\}\cup T_{u v}$ and $d(u_1)\le 6$.
One sees clearly that $T_1\ne \emptyset$.
Assume w.l.o.g. $8\in T_1$.
It follows that $H$ contains a $(3, 8)_{(u_3, v_3)}$-path and from Corollary~\ref{auv2}(1.2) that $H$ contains a $(2, 8)_{(u_1, u_2)}$-path.
Then $(u u_4, u v)\overset{\divideontimes}\to (8, 4)$\footnote{Or by Corollary~\ref{auv2d(u)5}(4.2),
for each $j\in T_4$, $H$ contains a $(i, j)_{(u_4, u_i)}$-path for some $i\in [1, 3]$, a contradiction. }.

{Case 2.3.} $\{i_1, i_2\} = \{1, 2\}$ and $(v u_3, v u_4)_c = (1, 2)$.
Then $T_{u v}\subseteq C(u_3)\cup C(u_4)$.
Let $C(u_3) = W_3\cup T_4\cup C_{3 4}$, $C(u_4) = W_4\cup T_3\cup C_{3 4}$.
Then $T_3\subseteq B_2\subseteq C(u_2)$, $T_4\subseteq B_1\subseteq C(u_1)$.
One sees from $t_3 + t_4 + c_{3 4} = t_{u v}$, $w_3 + w_4 + t_{u v} + c_{3 4}\le \Delta + 7$ that $w_3 + w_4\le 9$,
$c_{3 4}\le \Delta + 7 - (w_3 + w_4 + t_{u v}) = 9 - (w_3 + w_4)$ and $t_3 + t_4\ge 2(\Delta - 2) + w_3 + w_4 - (\Delta + 7) = \Delta + (w_3 + w_4) - 11$.
If $C(u_3) = \{1, 3\}\cup T_{uv}$, then $vu_3\to 4$ reduces the proof to (2.2.1.).
Otherwise, $T_3\ne \emptyset$ and likewise, $T_4\ne\emptyset$.
It follows from Lemma~\ref{auvge2}(1.2) that there exists a $b_1\in [5, 7]\cap C(u_3)$, $b_2\in [5, 7]\cap C(u_4)$, say $b_1 = 5$,
and $H$ contains a $(5, \alpha)_{(u_3, v_5)}$-path for some $\alpha\in T_3$, and a $(b_2, \beta)_{(u_4, v_{b_2})}$-path for some $\beta\in T_4$.

If $5\not\in C(u_2)$ and $H$ contains no $(5, i)_{(u_2, u_i)}$-path for each $i\in \{1, 3, 4\}$,
then $(u u_2, u v)\overset{\divideontimes}\to (5, \alpha)$.
If $b_2\not\in C(u_1)$ and $H$ contains no $(b_2, i)_{(u_1, u_i)}$-path for each $i\in \{2, 3, 4\}$,
then $(u u_1, u v)\overset{\divideontimes}\to (b_2, \beta)$.

Otherwise, we proceed with the following proposition:
\begin{proposition}
\label{5134}
\begin{itemize}
\parskip=0pt
\item[{\rm (1)}]
	$5\in C(u_2)$ or $H$ contains a $(5, i)_{(u_2, u_i)}$-path for some $i\in \{1, 3, 4\}$;
\item[{\rm (2)}]
	$b_2\in C(u_1)$ or $H$ contains a $(b_2, i)_{(u_1, u_i)}$-path for some $i\in \{2, 3, 4\}$;
\end{itemize}
\end{proposition}

One sees that if $4\not\in C(u_3)$,
then it follows from Corollary~\ref{auv2}(3) that $H$ contains a $(2, 3)_{(u_2, u_4)}$-path and $3\in C(u_4)\cap C(u_2)$.
Hereafter, we assume w.l.o.g. $3\in C(u_4)$.

Assume that $4\not\in C(u_3)$ or $T_3\setminus C(u_1)\ne \emptyset$.
One sees that if $4\not\in C(u_3)$, then $3\in B_2$,
and if $T_3\setminus C(u_1)\ne \emptyset$,
then $H$ contains a $(4, T_3\cap T_1)_{(u_1, u_4)}$-path and from Corollary~\ref{auv2}(5.2) and (3) that $3\in B_2$.
When $H$ contains no $(i, 3)_{(u_3, v_i)}$-path for each $i\in [5, 7]\cap C(u_3)$,
first, if $4\not\in C(u_3)$, then $u u_3\to T_3$,
if $T_3\setminus C(u_1)\ne \emptyset$, then $u u_3\to T_3\cap T_1$,
next, $(v u_3, u v)\overset{\divideontimes}\to (3, T_4)$,

Assume that $T_4\setminus C(u_2)\ne \emptyset$.
One sees that $H$ contains a $(3, T_4\cap T_2)_{(u_2, u_3)}$-path and from Corollary~\ref{auv2}(5.2) and (3) that $4\in B_1$.
If $H$ contains no $(4, i)_{(u_4, v_i)}$-path for each $i\in [5, 7]\cap C(u_4)$,
then $(v u_4, u u_4, u v)\overset{\divideontimes}\to (4, T_4\cap T_2, T_3)$.
Otherwise, we proceed with the following proposition:
\begin{proposition}
\label{53536464}
\begin{itemize}
\parskip=0pt
\item[{\rm (1)}]
	If $4\not\in C(u_3)$, or $T_3\setminus C(u_1)\ne \emptyset$,
    then $3\in B_2$ and $H$ contains a $(i, 3)_{(u_3, v_i)}$-path for some $i\in [5, 7]\cap C(u_3)$;
\item[{\rm (2)}]
	If $T_4\setminus C(u_2)\ne \emptyset$, then $4\in B_1$ and $H$ contains a $(4, i)_{(u_4, v_i)}$-path for some $i\in [5, 7]\cap C(u_4)$.
\end{itemize}
\end{proposition}

Assume that $1\not\in C(u_4)$ and $H$ contains no $(1, 2)_{(u_2, u_3)}$-path.
If $T_3\setminus C(u_1)\ne \emptyset$, then $(u u_4, u u_1, u v)\overset{\divideontimes}\to (1, T_3\cap T_1, T_4)$.
If $H$ contains no $(3, j)_{(u_1, u_3)}$-path for some $j\in C_{34}\cap T_1$,
then $(u u_4, u u_1, u v)\to (1, C_{34}\cap T_1, T_4)$.
Assume $2\not\in C(u_2)$ and $H$ contains no $(1, 2)_{(u_1, u_4)}$-path.
If $T_4\setminus C(u_2)\ne \emptyset$, then $(u u_3, u u_2, u v)\overset{\divideontimes}\to (2, T_4\cap T_2, T_3)$.
If $H$ contains no $(4, j)_{(u_2, u_4)}$-path for some $j\in C_{34}\cap T_2$,
then $(u u_3, u u_2, u v)\to (2, C_{34}\cap T_2, T_3)$.

Otherwise, we proceed with the following proposition:
\begin{proposition}
\label{T3Cu1T4Cu2}
\begin{itemize}
\parskip=0pt
\item[{\rm (1)}]
	If $1\not\in C(u_4)$ and $H$ contains no $(1, 2)_{(u_2, u_3)}$-path, then $T_3\subseteq C(u_1)$,
    and if $C_{34}\ne \emptyset$,
    then either $C_{34}\subseteq C(u_1)$, or $3\in C(u_1)$, $H$ contains a $(3, C_{34}\cap T_1)_{(u_1, u_3)}$-path.
\item[{\rm (2)}]
	If $2\not\in C(u_3)$ and $H$ contains no $(1, 2)_{(u_1, u_4)}$-path, then $T_4\subseteq C(u_2)$,
    and if $C_{34}\ne\emptyset$,
    then either $C_{34}\subseteq C(u_2)$, or $4\in C(u_2)$, $H$ contains a $(4, C_{34}\cap T_2)_{(u_2, u_4)}$-path.
\end{itemize}
\end{proposition}

Assume that $2\not\in C(u_3)$ and $H$ contains no $(1, 2)_{(u_1, u_4)}$-path.
If $C(u_4)\cap [5, 7] = \{b_2\}$, $5\not\in C(u_2)$ and $H$ contains no $(1, 5)_{(u_1, u_2)}$-path,
one sees that $H$ contains a $(b_2, T_4)_{(u_4, v_{b_2})}$-path,
then $(u u_2, u u_3, u u_4, u v)\overset{\divideontimes}\to (5, 2, T_4, T_3)$.
If $3\not\in C(u_2)$ and $H$ contains no $(3, i)_{(u_2, u_i)}$-path for each $i\in \{1, 4\}$,
then $(u u_2, u u_3, u v)\overset{\divideontimes}\to (3, 2, T_3)$.
Otherwise, we proceed with the following proposition:
\begin{proposition}
\label{2no1212}
If $2\not\in C(u_3)$ and $H$ contains no $(1, 2)_{(u_1, u_4)}$-path, then
\begin{itemize}
\parskip=0pt
\item[{\rm (1)}]
	 if $C(u_4)\cap [5, 7] = \{b_2\}$, then $5\in C(u_2)$ or $H$ contains a $(1, 5)_{(u_1, u_2)}$-path;
\item[{\rm (2)}]
	 $3\in C(u_2)$ or $H$ contains a $(3, i)_{(u_2, u_i)}$-path for some $i\in \{1, 4\}$.
\end{itemize}
\end{proposition}

Hence, by symmetry, we can have the following two subcase.

(2.3.1.) $4\in C(u_3)$ and $3\in C(u_4)$.
One sees from $w_3 + w_4\le 9$ that $\min\{w_3, w_4\} = 4$.
Assume w.l.o.g. $W_3 = \{1, 3, 4, 5\}$.
Then it follows from Lemma~\ref{auvge2}(1.2) that $H$ contains a $(5, T_3)_{(u_3, v_5)}$-path.

It follows from \ref{5134} (1) that $5/1/4/3\in C(u_2)$,
and from \ref{5134} (2) that $b_2/2/3/4\in C(u_1)$.
We distinguish the following two cases on whether $1\in C(u_4)$ or not.

(2.3.1.1.) $1\not\in W_4$.
It follows from \ref{T3Cu1T4Cu2} (1) that $T_3\subseteq C(u_1)$ and from \ref{T3Cu1T4Cu2} (2) that $T_4\subseteq C(u_2)$.
One sees that $C_{34}\subseteq C(u_1)\cup C(u_2)$.
Then $d(u_1) + d(u_2)\ge t_{uv} + t_3 + t_4 + |\{1, b_2/2/3/4\}| + |\{2, 5/1/3/4\}| = \Delta + 2 + t_3 + t_4$.
It follows that if $t_3 + t_4 = \Delta - 2$, then $d(u_1) = d(u_2) = \Delta = 7$ and then $d(u_3) + d(u_4)\le \Delta + 6$, i.e., $w_3 + w_4 = 8$, and $W_4 = \{4, 2, b_2, 3\}$,
and if $t_3 + t_4 = \Delta - 3$, then $c_{34} = 1$, $\Delta = 7$, $w_3 + w_4 = 8$,
$W_4 = \{4, 2, b_2, 3\}$, and $d(u_1) + d(u_2)\le \Delta + 6$.

It follows from \ref{2no1212} (1) that $5\in C(u_1)\cup C(u_2)$,
and from \ref{2no1212} (2) that $3\in C(u_1)\cup C(u_2)$, or $4\in C(u_2)$.
Then $d(u_1) + d(u_2)\ge t_{u v} + t_3 + t_4 + |\{1\}| + |\{2\}| + |\{5, 3/4\}| = \Delta + 2 + t_3 + t_4$.
It follows that $d(u_1) + d(u_2) = \Delta + 2 + t_3 + t_4$, $w_1 + w_2 = 4$,
i.e., $W_1\cup W_2 = \{1, 2, 5, 3/4\}$.
However, by \ref{2no1212}(1--2), $5, 3/4\in C(u_2)$.
Recall that $b_2/2/3/4\in C(u_1)$, a contradiction.

(2.3.1.2.) $W_4 = \{2, 4, b_2, 3, 1\}$.
Then $H$ contains a $(b_2, T_4)_{(u_4, v_{b_2})}$-path and $d(u_3) = d(u_4) = \Delta = 7$, $t_3 = 2$, $t_4 = 3$,
$d(u_1) + d(u_2)\le \Delta + 6$.

When $H$ contains no $(1, 2)_{(u_1, u_4)}$-path,
it follows from \ref{T3Cu1T4Cu2}(2) that $T_4\subseteq C(u_2)$,
and from \ref{2no1212}(1--2) that $C(u_2) = \{1, 2\}\cup T_{u v}$ and $C(u_1) = \{1, 3, 5\}\cup T_4$.
Then $(u u_1, u v)\overset{\divideontimes}\to (T_3, T_4)$\footnote{Or since $T_3\setminus C(u_1)\ne \emptyset$,
it follows from Corollary~\ref{auv2}(1.2) that $4\in C(u_1)$ and $d(u_1) = d(u_2) = \Delta$, a contradiction. }.

In the other case, $H$ contains a $(1, 2)_{(u_1, u_4)}$-path.
It follows that $2\in C(u_1)$.
\begin{itemize}
\parskip=0pt
\item
    $T_3\subseteq C(u_1)$, i.e., $C(u_1) = \{1, 2\}\cup T_{u v}$ and $d(u_1) = \Delta$.
    One sees from \ref{5134}(2) that $b_2\in C(u_2)$.
    Since $d(u_2)\le 6$, $T_4\setminus C(u_2) \ne \emptyset$.
    It follows that $3\in C(u_2)$ and $H$ contains a $(3, T_4\cap T_2)_{(u_2, u_3)}$-path.
    Then $(u u_3, u v)\overset{\divideontimes}\to (T_4\cap T_2, 4)$.
\item
    $T_3\setminus C(u_1)\ne \emptyset$.
    It follows that $4\in C(u_1)$ and $H$ contains a $(4, T_3\cap T_1)_{(u_1, u_4)}$-path,
    and from \ref{53536464}(1) that $3\in B_2$, $H$ contains a $(5, 3)_{(u_3, v_5)}$-path and $5\in C(v_5)$.
    If $b_2 = 5$, one sees that $C(v_5) = \{3, 5\}\cup T_{uv}$,
    then $(v v_5, v u_4, u v)\to (1, T_3, T_4)$.
    Otherwise, $b_2\ne 5$ and assume w.l.o.g. $b_2 = 6$.
    It follows from \ref{5134} (1) that $5/1\in C(u_2)$ and then $T_4\setminus C(u_2)\ne \emptyset$.
    Further, it follows from \ref{53536464}(2) that $H$ contains a $(4, 6)_{(u_4, v_6)}$-path,
    and then from \ref{5134} (2) that $6\in C(u_1)$ or $H$ contains a $(2, 6)_{(u_1, u_2)}$-path, i.e., $6\in C(u_1)\cup C(u_2)$.
    Since $T_3\setminus T_0\ne \emptyset$, it follows from Corollary~\ref{auv2}(5.1) that $7\in S_u$ and then $7\in C(u_1)\cup C(u_2)$.
    Then $d(u_1) + d(u_2)\ge t_{uv} + |\{1, 2, 4\}| + |\{2, 3\}| + |\{5, 6, 7\}| = 13$.
    It follows that $1, 4\not\in C(u_2)$ and $5\not\in C(u_1)$.
    Then $(u u_1, u u_2, u v)\overset{\divideontimes}\to (5, 1, T_3)$.
\end{itemize}

(2.3.2.) $4\not\in C(u_3)$ and $3\in B_2\subseteq C(u_4)\cap C(u_2)$.
Since $4\not\in C(u_3)$, it follows from \ref{53536464}(1) that $H$ contains a $(3, i)_{(u_3, v_i)}$-path for some $i\in [5, 7]\cap C(u_3)$.
If $H$ contains no $(3, j)_{(u_3, u_4)}$-path for some $j\in T_4$,
then $(u u_4, u v)\overset{\divideontimes}\to (j, 4)$.
Otherwise, $H$ contains a $(3, T_4)_{(u_3, u_4)}$-path.
It follows from Corollary~\ref{auv2}(1.2) that $T_4\subseteq C(u_2)$ and then $\{2, 3\}\cup T_3\cup T_4\subseteq C(u_2)$.
Since $T_4\subseteq C(u_1)\cap C(u_2)\cap C(u_3)$ and $C_{34}\subseteq (C(u_1)\cup C(u_2))\cap C(u_3)\cap C(u_4)$, $t_0\ge t_4 + c_{34}\ge t_4\ge 2$.

Assume that $W_3 = \{1, 3, 5\}$ and $W_4 = [2, 4]\cup \{b_2\}$.
One sees that $H$ contains a $(5, T_3)_{(u_3, v_5)}$-path and a $(b_2, T_4)_{(u_4, v_{b_2})}$-path.
If there exists an $i\in [5, 7]\setminus (C(u_2)\cup C(u_4))$, then $(u u_4, u v)\overset{\divideontimes}\to (i, 4)$.
If there exists a $j\in [6, 7]\setminus C(u_1)$, then $(u u_3, u u_4, u v)\overset{\divideontimes}\to (i, T_4, 4)$.
Otherwise, $[5, 7]\setminus C(u_4)\subseteq C(u_2)$ and $6, 7\in C(u_1)$.
If $3\not\in C(u_1)$ and $H$ contains no $(3, 4)_{(u_1, u_4)}$-path,
then $(u u_1, u u_3, u v)\overset{\divideontimes}\to (3, T_3, T_4)$.
Otherwise, $3\in C(u_1)$ or $H$ contains a $(3, 4)_{(u_1, u_4)}$-path.
It follows that $6, 7, 1, 3/4\in C(u_1)$.
One sees from \ref{2no1212}(1) that $1/5\in C(u_2)$.
If $b_2 = 5$, then $6, 7, 5/1, 3\in C(u_2)$,
and $\sum_{x\in N(u)\setminus \{v\}}d(x)\ge \sum_{i\in [1, 4]}w_i + 2t_{u v} + t_0\ge |\{1, 6, 7, 3/4\}| + |\{2, 6, 7, 3, 1/5\}| + |\{1, 3, 5\}| + |[2, 5]| + 2t_{u v} + t_0 = 16 + 2t_{u v} + t_0
\ge 2\Delta + 12 + 2 = 2\Delta + 14$, a contradiction.
Otherwise, $b_2 = 6$.
Then $5, 7, 3\in C(u_2)$,
and $\sum_{x\in N(u)\setminus \{v\}}d(x)\ge \sum_{i\in [1, 4]}w_i + 2t_{u v} + t_0\ge |\{1, 6, 7, 3/4\}| + |\{2, 5, 7, 3\}| + |\{1, 3, 5\}| + |\{2, 3, 4, 6\}| + 2t_{u v} + t_0 = 16 + 2t_{u v} + t_0\ge 2\Delta + 11 + 2 = 2\Delta + 13$.
It follows that $\sum_{x\in N(u)\setminus \{v\}}d(x) = 2\Delta + 13$, $t_0 = t_4 = 2$, and then $\Delta = 7$.
One sees from $t_3\ge 1$ and $C(u_2) = \{2, 3, 5, 7\}\cup T_3\cup T_4$ that $t_3 = 1$.
Thus, $c_{3 4} = 2$ and $t_0\ge t_4 + c_{3 4} = 2 + 2 = 4$, a contradiction.

In the other case, $w_3\ge 4$ or $w_4\ge 5$.
One sees clearly that $w_3 + w_4\ge 8$.
It follows that $c_{3 4}\le 1$ and $t_3 + t_4\ge \Delta - 3$.
we can have the following two subcase.

(2.3.2.1.) $t_3 + t_4 = \Delta - 2$.
Then $C(u_2) = [2, 3]\cup T_3\cup T_4 = [2, 3]\cup T_{u v}$.
It follows from \ref{5134}(1) that $H$ contains a $(5, 3)_{(u_2, u_3)}$-path,
and from \ref{53536464}(1) that $H$ contains a $(i, 3)_{(u_3, u_i)}$-path for some $i\in [5, 7]\cap C(u_3)$.
By Lemma~\ref{lemma06}, assume w.l.o.g. $H$ contains a $(3, 6)_{(u_3, v_6)}$-path.
When $1\not\in C(u_4)$,
if there exists an $i\in [5, 6]\setminus C(u_4)$, then $(u u_4, u v)\overset{\divideontimes}\to (i, 4)$.
Otherwise, $W_4 = [2, 6]$ and $W_3 = \{1, 3, 5, 6\}$.
Then $(u u_4, u v)\overset{\divideontimes}\to (7, 4)$.
In the other case, $1\in C(u_4)$.
Then $W_4 = [1, 4]\cup \{b_2\}$, $W_3 = \{1, 3, 5, 6\}$ and $d(u_3) = d(u_4) = \Delta = 7$.
It follows that $d(u_1)\le 6$.
If there exists an $i\in [5, 7]\setminus (C(u_4)\cup C(u_1))$, then $(u u_4, u v)\overset{\divideontimes}\to (i, 4)$.
Otherwise, $\{5, 6, 7\} = \{a_1, a_2, a_3\}\subseteq C(u_1)\cup C(u_4)$.
It follows that $C(u_1) = T_4\cup \{1, a_1, a_2\}$.
Then $(u u_1, u v)\overset{\divideontimes}\to (T_3, T_4)$.

(2.3.2.2.) $t_3 + t_4 = \Delta - 3$.
It follows that $c_{3 4} = 1$, $w_3 + w_4 = 8$ and $d(u_3) = d(u_4) = \Delta = 7$,
$\min\{d(u_1), d(u_2)\}\le 6$.
\begin{itemize}
\parskip=0pt
\item
    $W_3 = \{3, 1, 5, a\}$, $W_4 = \{4, 2, b_2, 3\}$.
    When $b_2 = 5$, if there exists an $i\in [6, 7]\setminus (C(u_2)\cup C(u_3))$,
    then $(u u_4, u v)\overset{\divideontimes}\to (i, 4)$.
    Otherwise, $6, 7\in C(u_2)\cup C(u_3)$ and then assume w.l.o.g. $W_2 = \{2, 3, 6\}$ and $a = 7$.
    It follows from \ref{5134}(1) that $H$ contains a $(5, 3)_{(u_2, u_3)}$-path and from \ref{53536464} that $H$ contains a $(7, 3)_{(u_3, v_7)}$-path.
    Then $(u u_4, u v)\overset{\divideontimes}\to (7, 4)$.

    \quad In the other case, assume w.l.o.g. $b_2 = 6$.
    If $C(u_2) = \{2, 3, 7\}\cup T_3\cup T_4$,
    one sees from \ref{5134}(1) that $H$ contains a $(5, 3)_{(u_2, u_3)}$-path,
    then $(u u_4, u v)\overset{\divideontimes}\to (5, 4)$.
    Otherwise, $7\not\in C(u_2)$.
    If $H$ contains no $(3, 7)_{(u_3, u_4)}$-path, then $(u u_4, u v)\to (7, 4)$;
    otherwise, $H$ contains a $(3, 7)_{(u_3, u_4)}$-path.
    It follows from \ref{53536464}(1) that $H$ contains a $(5, 3)_{(u_3, v_5)}$-path,
    and from \ref{5134}(1) that $5/1\in C(u_2)$.
    Thus $C(u_2) = \{2, 3, 5/1\}\cup T_3\cup T_4$ and $(u u_2, u u_3, u v)\overset{\divideontimes}\to (12, 2, T_3)$.
\item
    $W_3 = \{3, 1, 5\}$, $W_4 = \{4, 2, b_2, b_3, 3\}$.
    It follows from \ref{53536464} that $H$ contains a $(5, 3)_{(u_3, v_5)}$-path and from \ref{5134}(1) that $5\in C(u_2)$ or $H$ contains a $(5, i)_{(u_2, u_i)}$-path for some $i\in \{1, 4\}$.
    Hence, $C(u_2) = \{2, 3, 1/5\}\cup T_3\cup T_4$ and $T_4\cup T_{3 4}\subseteq C(u_1)$.
    One sees that $t_4 = 3$ and $c_{3 4} = 1$ and from \ref{5134} (2) that $b_2/2/3/4\in C(u_1)$.
    It follows that $T_3\setminus C(u_1)\ne \emptyset$ and then $C(u_1) = \{1, 4\}\cup T_4\cup C_{3 4}$.
    Further, it follows from Corollary~\ref{auv2}(5.1) that $6, 7\in S_u$.
    Hence, $W_4 = \{4, 2, 6, 7, 3\}$ and $W_2 = \{2, 3, 5\}$.
    Then $(u u_2, u u_3, u v)\overset{\divideontimes}\to (12, 2, T_3)$.
\end{itemize}

{Case 2.4.} $\{i_1, i_2\} = \{1, 3\}$ assume w.l.o.g. $c(v u_3) = 5$.

If $5\not\in C(u_1)$ and $H$ contains no $(5, i)_{(u_1, u_i)}$-path for each $i\in [2, 4]$,
$uu_1\to 5$ reduces the proof to the Case 2.2.
Hence, we proceed with the following proposition, or otherwise we are done:
\begin{proposition}
\label{315234}
$5/2/3/4\in C(u_1)$, $5\in S_u\setminus C(u_3)$,
and if $5\not\in C(u_1)$, then $H$ contains a $(5, i)_{(u_1, u_i)}$-path for some $i\in [2, 4]$.
\end{proposition}

Assume that $H$ contains a $(\rho, 8)_{(u_3, v_\rho)}$-path for some $\rho\in [6, 7]$ and $8\in T_3$.
If $\rho\not\in C(u_1)$ and $H$ contains no $(\rho, i)_{(u_1, u_i)}$-path for each $i\in [2, 4]$,
then $(u u_1, u v)\overset{\divideontimes}\to (\rho, 8)$.
Hence, we proceed with the following proposition, or otherwise we are done:
\begin{proposition}
\label{31rho8234}
If $H$ contains a $(\rho, 8)_{(u_3, v_\rho)}$-path for some $\rho\in [6, 7]$ and $8\in T_3$,
then $\rho\in C(u_1)$ or $H$ contains a $(\rho, i)_{(u_1, u_i)}$-path for some $i\in [2, 4]$.
\end{proposition}

(2.4.1.) $c(v u_4) = 3$. One sees that $T_1\subseteq C(u_3)\cap C(u_4)$, $T_4\subseteq C(u_1)$.

\begin{lemma}
\label{3T3T4neemptyset}
$G$ admits an acyclic edge $(\Delta + 2)$-coloring, or $T_3\ne \emptyset$ and $T_4\ne \emptyset$.
\end{lemma}
\begin{proof}
We consider the following two cases on whether $T_3 = \emptyset$ or not.

{\bf Case 1.} $T_3 = \emptyset$, i.e., $C(u_3) = \{3, 5\}\cup T_{uv}$.
There are two subcase on whether $T_4 = \emptyset$ or not.

{Case 1.1.} $C(u_4) = \{3, 4\}\cup T_{u v}$.
If $4\not\in C(u_1)$, then $(u u_4, u v)\overset{\divideontimes}\to (6, 4)$.
If $5\not\in C(u_1)$, then $(v u_3, u v)\overset{\divideontimes}\to (2, 5)$.
Otherwise, $4, 5\in C(u_1)$ and then $T_1\ne \emptyset$.
If $H$ contains no $(6, j)_{(u_3, v_6)}$-path for some $j\in T_1$,
then $(u u_3, u v)\overset{\divideontimes}\to (6, j)$.
Otherwise, $H$ contains a $(6, T_1)_{(u_3, v_6)}$-path.
Then $(u u_3, u u_4, u v)\overset{\divideontimes}\to (4, 6, T_1)$.

{Case 1.2.} $T_4 \ne \emptyset$.
One sees that $H$ contains a $(1, T_4)_{(u_1, v_1)}$-path.
For some $j\in T_4$, if $H$ contains no $(2, j)_{(u_2, u_4)}$-path, then $(u u_3, u u_4)\overset{\divideontimes}\to (4, j)$;
if $H$ contains no $(i, j)_{(u_4, v_i)}$-path for each $i\in [6, 7]$, then $(v u_4, v u_3, u u_3)\to (j, 3, 5)$, these reduce to an acyclic edge $(\Delta + 5)$-coloring for $H$ such that $C(u)\cap C(v) = \{1\}$.
Otherwise, $H$ contains a $(2, T_4)_{(u_2, u_4)}$-path and a $(6/7, T_4)_{(u_4, v)}$-path.
It follows $2\in C(u_2)$, $T_4\subseteq C(u_2)$, and there exists a $a_3\in [6, 7]\cap C(u_4)$.
Assume w.l.o.g. $6\in C(u_4)$ and $H$ contains a $(6, 8)_{(u_4, v_6)}$-path, where $8\in T_4$.
One sees from \ref{315234} that $5\in S_u\setminus C(u_3)$.

If $3\not\in C(u_1)$ and $H$ contains no $(2, 3)_{(u_1, u_2)}$-path,
then $(u u_1, u u_3, u v)\overset{\divideontimes}\to (3, 1, 8)$.
If $6\not\in C(u_1)$ and $H$ contains no $(i, 6)_{(u_1, u_i)}$-path-path for each $i\in \{2, 4\}$,
then $(u u_1, u v)\overset{\divideontimes}\to (6, 8)$.
If $4\not\in C(u_2)$ and $H$ contains no $(1, 4)_{(u_1, u_2)}$-path,
then $(u u_2, u u_3, u u_4)\to (4, 2, 8)$ reduces to an acyclic edge $(\Delta + 5)$-coloring for $H$ such that $C(u)\cap C(v) = \{1\}$.
Otherwise, $3\in C(u_1)$ or $H$ contains a $(2, 3)_{(u_1, u_2)}$-path;
$6\in C(u_1)$ or $H$ contains a $(6, 2/4)_{(u_1, u)}$-path;
and $4\in C(u_2)$ or $H$ contains a $(1, 4)_{(u_1, u_2)}$-path.
Since $(v u_3, u u_3)\to (2, 5)$ reduces to an acyclic edge $(\Delta + 5)$-coloring for $H$ such that $C(u)\cap C(v) = \{1, 2\}$, $T_1\subseteq C(u_2)$ and it follows that for each $j\in T_{u v}$,
mult$_{C(u_1)\cup C(u_2)\cup C(u_4)}(j)\ge 2$.

One sees that $2t_{uv} + |\{1\}| + |\{2\}| + |\{4, 3, 2, 6\}| + |\{3, 5\}| = 2\Delta - 4 + 8 = 2\Delta + 4$.
When $T_1 = \emptyset$, i.e., $w_1\le 2$,
it follows that $C(u_1) = [1, 2]\cup T_{u v}$ and $[3, 6]\subseteq C(u_2)$,
thus, $d(u_1) + d(u_2) + d(u_4)\ge 2\Delta + 4 + |\{2\}| + |\{4, 6\}| = 2\Delta + 7$.

In the other case, $T_1 \ne \emptyset$.
It follows that $2/4\in C(u_1)$.
If $6\in C(u_1)\cup C(u_2)$,
then $d(u_1) + d(u_2) + d(u_4)\ge 2\Delta + 4 + |\{2/4\}| + |\{1/4\}| + |\{6\}| = 2\Delta + 7$.
Otherwise, $6\not\in C(u_1)\cup C(u_2)$.
It follows that $4\in C(u_1)$, $H$ contains a $(4, 6)_{(u_1, u_4)}$-path.
If $2\not\in C(u_1)$, then $(u u_2, u u_4, u v)\overset{\divideontimes}\to (6, 8, 2)$.
Otherwise, $2\in C(u_1)$ and thus $d(u_1) + d(u_2) + d(u_4)\ge 2\Delta + 4 + |\{2, 4\}| + |\{1/4\}|= 2\Delta + 7$.

One sees clearly that none of ($A_{8.1}$)--($A_{8.4}$) holds.

{\bf Case 2.} $C(u_4) = [3, 4]\cup T_{uv}$.
Then $(v u_4, u u_4)\to (4, 5)$ reduces the proof to Case 1.
\end{proof}

By Lemma~\ref{3T3T4neemptyset}, hereafter we assume $T_3\ne\emptyset$ and $T_4\ne \emptyset$.

(2.4.1.1.) When there exists a $j\in T_{u v}\setminus (C(u_3)\cup C(u_4))$,
say $8\not\in C(u_3)\cup C(u_4)$,
then $2\in C(u_3)$, $6\in C(u_4)$,
and $H$ contains $(2, 8)_{(u_2, u_3)}$-path, a $(6, 8)_{(u_4, v_6)}$-path.
It follows from Corollary~\ref{auv2}(5.2) that $H$ contains a $(2, T_3)_{(u_2, u_3)}$-path,
from Corollary~\ref{auv2}(5.1) that $7\in S_u$,
and from Corollary~\ref{auv2}(5.2) and (3) that $4\in B_1$, and $4\in C(u_2)$, or $H$ contains a $(2, 4)_{(u_2, u_3)}$-path.

If $3\not\in C(u_2)$ and $H$ contains a $(3, 1)_{(u_1, u_2)}$-path,
then $(u u_2, u u_3)\to (3, 8)$, and if $2\not\in C(u_1)$, $u v\to 2$,
or otherwise, $u v\overset{\divideontimes}\to T_1$.
If $6\not\in C(u_1)$ and $H$ contains a $(6, i)_{(u_1, u_i)}$-path for each $i\in [2, 3]$,
then $(u u_1, u v)\overset{\divideontimes}\to (6, 8)$ if $H$ contains no $(4, 6)_{(u_1, u_4)}$-path,
or otherwise, $(u u_1, u u_4, u v)\overset{\divideontimes}\to (6, 8, 4)$.
In the other case, we proceed with the following proposition, or otherwise we are done:
\begin{proposition}
\label{331623}
\begin{itemize}
\parskip=0pt
\item[{\rm (1)}]
	$3\in C(u_2)$ or $H$ contains a $(3, 1)_{(u_1, u_2)}$-path.
\item[{\rm (2)}]
	$6\in C(u_1)$ or $H$ contains a $(6, i)_{(u_1, u_i)}$-path for some $i\in [2, 3]$.
\end{itemize}
\end{proposition}

One sees from Corollary~\ref{auv2d(u)5}(7.1) that $2\in C(u_1)$ or $1\in S_u$.
Then $|\{1, 4\}| + |\{2\}| + |\{2, 3, 5\}| + |\{4, 3, 6\}| + |\{4, 3, 6, 7, 5, 1/2\}| + 2t_{u v} + t_0 = 2\Delta - 4 + 15 + t_0 = 2\Delta + 11 + t_0$.
\begin{itemize}
\parskip=0pt
\item
    $H$ contains no $(7, 8)_{(u_3, v_7)}$-path.
    It follows from Corollary~\ref{auv2}(4) that $5\in B_1\subseteq C(u_1)$, and $H$ contains a $(5, i)_{(u_3, u_i)}$-path for some $i\in \{2, 4\}$.
    Then $\sum_{x\in N(u)\setminus \{v\}}d(x)\ge 2\Delta + 11 + t_0 + |\{5\}| = 2\Delta + 12 + t_0$.
    One sees from $1, 4, 5, 6/2/3\in C(u_1)$ that $t_1\ge 2$.
    Hence, $T_1\setminus T_0\ne \emptyset$ and then $H$ contains a $(4, T_1)_{(u_1, u_4)}$-path.
    It follows from Corollary~\ref{auv2}(5.2) that $2\in C(u_1)\cup C(u_4)$,
    and from Corollary~\ref{auv2}(5.1) that $1\in C(u_4)$, or $H$ contains  a $(1, i)_{(u_4, u_i)}$-path for some $i\in \{2, 3\}$.
    Thus, $\sum_{x\in N(u)\setminus \{v\}}d(x)\ge \sum_{i\in [1, 4]}w_i + 2t_{u v} + t_0\ge |\{1, 4, 5\}| + |\{2\}| + |\{3, 5, 2\}| + |\{3, 4, 6\}| + |\{4, 3, 6, 7, 5, 1, 2\}| + 2t_{u v} + t_0= 2\Delta + 13 + t_0$.
    It follows that $\sum_{x\in N(u)\setminus \{v\}}d(x) = 2\Delta + 13$, $t_0 = 0$,
    mult$_{S_u}(i) = 1$, $i = 1, 7$, mult$_{S_u}(j) = 2$, $i = 2, 3, 6$, and mult$_{S_u}(5) = 3$.

    \quad If $5\not\in C(u_2)$, one sees that $H$ contains a $(4, 5)_{(u_3, u_4)}$-path,
    then $(u u_2, u u_3, u v)\overset{\divideontimes}\to (5, 8, T_1)$.
    Otherwise, $5\in C(u_2)\setminus C(u_4)$.
    When $H$ contains a $(7, 4)_{(u_4, v_7)}$-path, one sees that $7\not\in S_u\setminus C(u_4)$,
    if $H$ contains no $(7, j)_{(u_3, v_7)}$-path for some $j\in T_1$,
    then $(u u_3, u v)\to (7, j)$; otherwise, $(u u_2, u u_3, u v)\to (7, T_3, T_1)$.
    In the other case, $H$ contains no $(7, 4)_{(u_4, v_7)}$-path.
    If $H$ contains no $(4, 6)_{(u_4, v_6)}$-path, then $(v u_4, u u_4, u v)\to (4, 8, T_1)$.
    Otherwise, $H$ contains a $(4, 6)_{(u_4, v_6)}$-path.
    We distinguish the following three scenarios.
    \begin{itemize}
    \parskip=0pt
    \item
        $6\in C(u_3)\setminus (C(u_1)\cup C(u_2))$.
        It follows from \ref{331623}(2) that $3\in C(u_1)$,
        and then from \ref{331623}(1) that $1\in C(u_2)$.
        Hence, $1\not\in C(u_3)\cup C(u_4)$ and $H$ contains a $(1, 2)_{(u_2, u_4)}$-path.
        Then $(u u_1, u u_3, u v)\overset{\divideontimes}\to (6, 1, 8)$.
    \item
        $6\in C(u_2)\setminus (C(u_1)\cup C(u_3))$.
        One sees from \ref{331623}(2) that $2\in C(u_1)$ and $H$ contains $(2, 6)_{(u_1, u_2)}$-path.
        It follows that $1\in C(u_4)$ or $H$ contains a $(1, 3)_{(u_3, u_4)}$-path,
        i.e., $1\not\in C(u_2)$, and $1\not\in C(u_3)$ or $H$ contains a $(1, 3)_{(u_3, u_4)}$-path.
        Then $(u u_1, u u_2, u u_4, u v)\overset{\divideontimes}\to (1, 6, 2, 8)$.
    \item
        $6\in C(u_1)\setminus (C(u_2)\cup C(u_3))$.
        If $H$ contains no $(1, 6)_{(u_i, u_1)}$-path,
        then if $H$ contains no $(6, j)_{(u_3, v_6)}$-path for some $j\in T_1$,
        then $(u u_3, u v)\to (6, j)$; otherwise, $(u u_2, u u_3, u v)\to (6, T_3, T_1)$.
        Otherwise, $H$ contains a $(6, 1)_{(u_i, u_1)}$-path for some $i\in [2, 3]$.

        \quad If $1\in C(u_3)\setminus (C(u_2)\cup C(u_4))$,
        one sees that $H$ contains a $(1, 6)_{(u_1, u_3)}$-path and from \ref{331623}(1) that $3\in C(u_2)\setminus C(u_1)$,
        then $(u u_1, u u_2, u u_3, u v)\overset{\divideontimes}\to (3, 1, 6, 8)$.
        Otherwise, $1\in C(u_2)\setminus (C(u_3)\cup C(u_4))$ and $H$ contains a $(1, 6)_{(u_1, u_2)}$-path.
        It follows that $2\in C(u_4)\cap B_3$ and $H$ contains a $(1, 2)_{(u_2, u_4)}$-path.
        Then $(u u_1, u u_2, u u_4, u v)\overset{\divideontimes}\to (2, 6, 1, 8)$.
     \end{itemize}
\item
    $H$ contains a $(7, 8)_{(u_3, v_7)}$-path.
    One sees clearly that $7\in C(u_3)$.
    When $7\not\in S_u\setminus C(u_3)$,
    if $1\not\in S_u$, then $(u u_1, u u_3, u v)\overset{\divideontimes}\to (7, 1, 8)$.
    if $H$ contains no $(3, 7 )_{(u_1, u_3)}$-path, then $(u u_3, u v)\overset{\divideontimes}\to (7, 8)$.
    Otherwise, $1\in S_u$, $3\in C(u_1)$ and $H$ contains a $(3, 7)_{(u_1, u_3)}$-path.
    If $H$ contains no $(2, j)_{(u_1, u_2)}$-path for some $j\in T_1$,
    then $(u u_4, u u_1, u v)\to (7, j, 8)$.
    If $H$ contains no $(4, j)_{(u_1, u_4)}$-path for some $j\in T_1$,
    then $(u u_2, u u_1, u v)\to (7, j, 8)$.
    Otherwise, $2\in C(u_1)$ and $T_1\subseteq C(u_2)\cap C(u_4)$.
    Since, $1, 4, 3, 2\in C(u_1)$, $t_0\ge t_1\ge 2$.
    Then $\sum_{x\in N(u)\setminus \{v\}}d(x)\ge \sum_{i\in [1, 4]}w_i + 2t_{u v} + t_0\ge |\{1, 4, 2, 3\}| + |\{2\}| + |\{3, 5, 2, 7\}| + |\{3, 4, 6\}| + |\{4, 6, 5, 1\}| + 2t_{u v} + 2 = 2\Delta + 14$.

    \quad In the other case, $7\in S_u\setminus C(u_3)$.
    One sees from \ref{331623} (2) that $6/2/3\in C(u_1)$.
    It follows that $t_1\ge 1$.
    One sees that $|\{1, 4\}| + |\{2\}| + |\{2, 3, 5, 7\}| + |\{4, 3, 6\}| + |\{4, 3, 6, 7, 5, 1/2\}| + 2t_{u v} + t_0 = 2\Delta - 4 + 16 + t_0 = 2\Delta + 12 + t_0$.
    It follows that $t_0\ge t_1\ge 1$, or if $T_1\setminus T_0\ne \emptyset$,
    then $1\in S_u$ and $2\in S_u\setminus C(u_3)$.
    Hence, mult$_{S_u\setminus C(u_3)}(5) = 1$, mult$_{S_u}(3) = 1$, mult$_{S_u\setminus C(u_4)}(6) = 1$,
    and then from \ref{315234} that $3\in C(u_1)$, $5\in C(u_4)$ and $H$ contains a $(3, 5)_{(u_1, u_4)}$-path,
    from \ref{331623}(1) that $1\in C(u_2)$ and $H$ contains a $(1, 3)_{(u_1, u_2)}$-path.
    \begin{itemize}
    \parskip=0pt
    \item
        $T_1\subseteq T_0$, it follows that $t_1 = 1$, $2\not\in C(u_1)\cup C(u_4)$,
        $C(u_3) = \{1, 3, 4\}\cup T_{uv}$ and from \ref{331623}(2) that $6\in C(u_3)$ and $H$ contains a $(3, 6)_{(u_1, u_3)}$-path.
        One sees clearly that $6\not\in C(u_1)\cup C(u_2)$.
        Then $(u u_2, u u_4, u v)\overset{\divideontimes}\to (6, 8, 2)$.
    \item
        $T_1\ne \emptyset$.
        It follows from Corollary~\ref{auv2}(5.2) and (3) that $2\in B_1\cup B_3$ and $2\in C(u_1)\cup C(u_4)$,
        $1\in C(u_4)$ or $H$ contains a $(1, i)_{(u_4, u_i)}$-path for some $i\in [2, 3]$.
        Hence, $\sum_{x\in N(u)\setminus \{v\}}d(x)\ge |\{1, 3, 4\}| + |\{2, 1\}| + |\{3, 5, 2, 7\}| + |\{4, 3, 6, 5\}| + |\{4, 6, 7, 2\}| + 2t_{u v} + t_0 = 2\Delta - 4 + 17 + t_0 = 2\Delta + 13 + t_0$.
        It follows that $\sum_{x\in N(u)\setminus \{v\}}d(x) = 2\Delta + 13$, $t_0 = 0$, mult$_{S_u}(1) = 1$.
        Further, we have $H$ contains a $(1, 2)_{(u_2, u_4)}$-path.
        One sees that $T_1\cap C(u_2) = \emptyset$.
        If $H$ contains no $(4, 5)_{(u_2, u_4)}$-path,
        then $(u u_2, u u_3, u v)\overset{\divideontimes}\to (5, T_3, T_1)$. Otherwise, $H$ contains a $(4, 5)_{(u_2, u_4)}$-path.
        Then $(u u_1, u u_3, u v)\overset{\divideontimes}\to (5, 1, 8)$.
    \end{itemize}
\end{itemize}

(2.4.1.2.) $T_{u v} \subseteq C(u_3)\cup C(u_4)$.
Let $C(u_3) = W_3\cup T_4\cup C_{34}$, $C(u_4) = W_4\cup T_3\cup C_{34}$.
Then $T_3\cup T_4\subseteq B_1\subseteq C(u_1)$.
One sees that if $2\not\in C(u_4)$ or $H$ contains no $(2, j)_{(u_2, u_4)}$-path,
then $4\in B_1$ and mult$_{S_u}(4)\ge 2$,
if for some $j\in T_3$, $H$ contains no $(i, j)_{(u_3, v_i)}$-path for each $i\in [6, 7]$,
then $5\in B_1$ and mult$_{S_u}(5)\ge 2$.
If $H$ contains no $(4, i)_{(u_4, v_i)}$-path for each $i\in C(v)$ and contains no $(2, j)_{(u_2, 4_4)}$-path for some $j\in T_4$,
then $(v u_4, u u_4)\to (4, j)$ reduces to an acyclic edge $(\Delta + 5)$-coloring for $H$ such that $C(u)\cap C(v)= \{1\}$.
Otherwise, we proceed with the following proposition:
\begin{proposition}
\label{314Cv2T4}
$H$ contains a $(4, i)_{(u_4, v_i)}$-path for some $i\in C(v)$ or a $(2, T_4)_{(u_2, u_4)}$-path.
\end{proposition}

When $T_1 = \emptyset$, i.e., $W_1\le 2$. One sees from \ref{315234} that $W_1 = \{1, 5\}$,
or $W_1 = \{1, i\}$ and $H$ contains a $(5, i)_{(u_1, u_i)}$-path, where $i\in \{2, 3, 4\}$.

Assume $C(u_1) = \{1, i\}\cup T_{u v}$.
We assume w.l.o.g. $6\in C(u_3)$, $H$ contains a $(6, 8)_{(u_3, v_6)}$-path, where $8\in T_3$.
If $H$ contains no $(i, 6)_{(u_1, u_i)}$-path,
then $(u u_1, u v)\overset{\divideontimes}\to (6, 8)$.
If $7\not\in S_u$, then $(u u_1, u v)\overset{\divideontimes}\to (7, 8)$,
or $(u u_i, u u_1, u v)\overset{\divideontimes}\to (7, 5, 8)$.
If $1\not\in C(u_i)$ and $H$ contains no $(1, j)_{(u_i, u_j)}$-path for each $j\in C(u)$,
then $(u u_i, u u_1, u v)\overset{\divideontimes}\to (1, 5, 8)$.
Otherwise, we proceed with the following proposition:
\begin{proposition}
\label{i6i671j}
If $C(u_1) = \{1, i\}\cup T_{uv}$, then $H$ contains a $(i, 6)_{(u_1, u_i)}$-path, $1, 7\in S_u$,
and if $1\not\in C(u_i)$, then $H$ contains a $(1, j)_{(u_i, u_j)}$-path for some $j\in C(u)$.
\end{proposition}
\begin{itemize}
\parskip=0pt
\item
    $W_1 = \{1, 5\}$.
    Since $4\not\in C(u_1)$, $2\in C(u_4)$,
    and $H$ contains a $(2, T_4)_{(u_2, u_4)}$-path, i.e., $T_4\subseteq C(u_2)$.
    It follows that $t_0\ge t_4\ge 2$.
    Since $u u_1\to i\in [6, 7]$ reduces to an acyclic edge $(\Delta + 5)$-coloring for $H$ such that $C(u)\cap C(v) = \{3, i\}$, $T_3\cup T_4\subseteq B_6\cap B_7$.
    It follows that $b_3 = 5$ and $H$ contains a $(5, T_4)_{(u_3, u_4)}$-path.
    If $1\not\in C(u_2)$ and $H$ contains no $(1, i)_{(u_2, u_i)}$-path for each $i\in [3, 4]$,
    then $(u u_2, u u_1, u v)\overset{\divideontimes}\to (1, 2, T_3)$.
    If $1\not\in C(u_4)$ and $H$ contains no $(1, i)_{(u_4, u_i)}$-path for each $i\in [2, 3]$,
    then $(u u_4, u u_1, u v)\overset{\divideontimes}\to (1, 4, T_3)$.
    If $1\not\in C(u_3)$ and $H$ contains no $(1, i)_{(u_3, u_i)}$-path for each $i\in \{2, 4\}$,
    then $(u u_3, u u_1, u v)\overset{\divideontimes}\to (1, 3, T_4)$.
    It follows that mult$_{S_u}(1)\ge 2$.
    One sees that the same argument applies if $u u_1\to 6/7$.
    Hence, $_{S_u}(i)\ge 2$, $i = 6, 7$.
    Thus, $\sum_{x\in N(u)\setminus \{v\}}d(x)\ge |\{1, 5\}| + |\{2\}| + |\{a_3, 3, 5\}| + |\{4, 3, 2, 5\}| + 2|\{1, 6, 7\}| + 2t_{u v} + t_0 \ge 16 + 2t_{u v} + 2 = 2\Delta + 14$.
\item
     $W_1 = \{1, 2\}$ and $H$ contains a $(2, 5)_{(u_1, u_2)}$-path, $5\in C(u_2)$.
     Since $4\not\in C(u_1)$, $2\in C(u_4)$, and $H$ contains a $(2, T_4)_{(u_2, u_4)}$-path, i.e., $T_4\subseteq C(u_2)$.
     It follows that $t_0\ge t_4\ge 2$.
     It follows from \ref{i6i671j} that $6\in C(u_2)$, $1, 7\in S_u$, $1/3/4\in C(u_2)$,
     and if $1\not\in C(u_2)$, then $H$ contains a $(1, i)_{(u_2, u_i)}$-path for some $i\in [3, 4]$.
     One sees that $\sum_{x\in N(u)\setminus \{v\}}d(x)\ge\sum_{i\in [1, 4]}w_i + 2t_{u v} + t_0\ge |\{1, 2\}| + |\{2, 5, 6\}| + |\{3, 5, a_3, 6\}| + |\{4, 3, b_3, 2\}| + |\{1\}| + 2t_{u v} + 2 = 2\Delta + 12$.
     Since $t_3\ge 2$, $T_3\ne\emptyset$.
     It follows that $4\in C(u_3)$ and $H$ contains a $(4, T_3)_{(u_3, u_4)}$-path.
     If $5\not\in C(u_4)$, then $(u u_4, u u_3, u v)\overset{\divideontimes}\to (5, T_3\cap T_2, 4)$.
     Otherwise, $5\in C(u_4)$.
     Then $\sum_{x\in N(u)\setminus \{v\}}d(x)\ge \sum_{i\in [1, 4]}w_i + 2t_{u v} + t_0\ge |\{1, 2\}| + |\{2, 5, 6\}| + |\{4, 6, 3, 5\}| + |\{4, 3, 5, 2\}| + |\{1, 7\}| + 2t_{u v} + 2 = 2\Delta + 13$.
     It follows that $1\in C(u_2)\setminus (C(u_3)\cup C(u_4))$ and $4\not\in C(u_2)$.
     Then $(u u_1, u u_2, u u_4, u v)\overset{\divideontimes}\to (5, 4, 1, 8)$.
\item
     $W_1 = \{1, 3\}$ and $H$ contains a $(3, 5)_{(u_1, u_4)}$-path, $5\in C(u_4)$.
     Since $4\not\in C(u_1)$, $2\in C(u_4)$, and $H$ contains a $(2, T_4)_{(u_2, u_4)}$-path, i.e., $T_4\subseteq C(u_2)$.
     It follows that $t_0\ge t_4\ge 2$.
     If $4, 5\not\in C(u_2)$, $(u u_2, u v)\to (5, 2)$ or $(u u_2, u u_1, u v)\to (5, 2, T_3)$.
     Otherwise, $4/5\in C(u_2)$.
     It follows from \ref{i6i671j} that $1, 7\in S_u$ and $H$ contains a $(3, 6)_{(u_1, u_3)}$-path.
     If $6\not\in C(u_2)\cup C(u_4)$, then $(u u_4, u v)\overset{\divideontimes}\to (6, 4)$,
     or $(u u_4, u u_1, u v)\overset{\divideontimes}\to (6, 4, T_4)$.
     Otherwise, $6\in C(u_2)\cup C(u_4)$.
     Then $\sum_{x\in N(u)\setminus \{v\}}d(x)\ge \sum_{i\in [1, 4]}w_i + 2t_{u v} + t_0\ge |\{1, 3\}| + |\{2, 4/5\}| + |\{a_3, 6, 3, 5\}| + |\{4, 3, 5, 2\}| + |\{1, 6, 7\}| + 2t_{u v} + 2 = 2\Delta + 13$.
     It follows that $t_0 = t_4 = 2$ and $T_3\setminus C(u_2)\ne \emptyset$.
     Hence $a_3 = 4$, and $H$ contains a $(4, T_3)_{(u_3, u_4)}$-path.
     One sees that $2\not\in C(u_3)$. Then $(u u_2, u v)\overset{\divideontimes}\to (T_3, 2)$.
\item
     $W_1 = \{1, 4\}$ and $H$ contains a $(4, 5)_{(u_1, u_4)}$-path, $5\in C(u_4)$.
     It follows from \ref{i6i671j} that $1, 7\in S_u$, $6\in C(u_4)$, $H$ contains a $(4, 6)_{(u_1, u_4)}$-path,
     and if $1\not\in C(u_4)$, then $H$ contains a $(1, i)_{(u_4, u_i)}$-path for some $i\in [2, 3]$.
     It follows that $1\in C(u_3)\cup C(u_4)$ or $2\in C(u_4)$.
     One sees clearly that $7\in C(u_2)\setminus (C(u_3)\cup C(u_4))$,
     and if $2\not\in C(u_4)$, then $4\in B_1$.
     It follows from \ref{314Cv2T4} that $2\in C(u_4)$, $H$ contains a $(2, T_4)_{(u_2, u_4)}$-path and $T_4\subseteq C(u_2)$.
     One sees that $t_0\ge t_4\ge 3$.
     Then $\sum_{x\in N(u)\setminus \{v\}}d(x)\ge \sum_{i\in [1, 4]}w_i + 2t_{u v} + t_0\ge |\{1, 4\}| + |\{2, 1, 7\}| + |\{a_3, 6, 3, 5\}| + |\{4, 3, 5, 6, 2\}| + 2t_{u v} + 3 = 2\Delta + 13$.
     It follows that $t_0  = t_4 = 3$ and $a_3 = 4$, $2\not\in C(u_3)$.
     Then $(u u_2, u v)\overset{\divideontimes}\to (T_3, 2)$.
\end{itemize}

In the other case, $T_1\ne\emptyset$.
Since $T_1\subseteq C(u_3)\cap C(u_4)\subseteq C_{34}$, it follows that $c_{34}\ge t_1\ge 1$.
One sees from $\Delta + 7\ge d(u_3) + d(u_4)\ge w_3 + w_4 + t_{u v} + c_{34}$ that $w_3 + w_4\le 9 - c_{34}\le 8$ and $w_3, w_4\le 5$.
One sees from Corollary~\ref{auv2d(u)5}(7.2) that if $1\not\in C(u_4)$ and $H$ contains no $(1, i)_{(u_4, u_i)}$-path for each $i\in [2, 3]$,
then $2\in C(u_1)$ and $H$ contains a $(2, T_1)_{(u_1, u_2)}$-path, $T_1\subseteq C(u_2)$.

If $1\not\in C(u_3)$, $H$ contains no $(1, i)_{(u_i, u_i)}$-path for each $i\in \{2, 4\}$,
and $3\not\in S_u\setminus C(u_4)$,
then $(u u_1, u u_3, u v)\overset{\divideontimes}\to (3, 1, T_4)$.
Otherwise, we proceed with the following proposition:
\begin{proposition}
\label{311243}
If $1\not\in C(u_3)$ and $H$ contains no $(1, i)_{(u_i, u_i)}$-path for each $i\in \{2, 4\}$,
then $3\in C(u_1)\cup C(u_2)$.
\end{proposition}
$1\in S_u$ or $3\in S_u\setminus C(u_4)$.

Assume that $d(u_1) = \Delta$ and $C(u_1)\setminus \{1\} = B_1 = T_3\cup T_4\cup \{4, 5\}$ or $B_1 = T_3\cup T_4\cup \{4, 5, 2\}$.
One sees that $C(v_1) = C(u_1)$, and it follows from Lemma~\ref{auvge2} that $H$ contains a $(5, C_{34})_{(u_3, v_1)}$-path.
If for some $\gamma\in [6, 7]$, $H$ contains no $(j, i)_{(u_1, u_i)}$-path for each $i\in \{2, 4\}$,
one sees that the same argument applies if $u u_1\to \gamma$, and thus $C(v_1)\setminus \{1\} = C(v_\gamma)\setminus \{\gamma\} = C(u_1)\setminus \{1\}$,
then $(u u_1, v v_\gamma, u v)\overset{\divideontimes}\to (\gamma, C_{34}, T_3)$.
Hence, we proceed with the following proposition, or otherwise we are done:
\begin{proposition}
\label{26264646}
If $d(u_1) = \Delta$, $C(u_1)\setminus \{1\} = B_1 = T_3\cup T_4\cup \{4, 5\}$ or $B_1 = T_3\cup T_4\cup \{4, 5, 2\}$, then for each $j\in [6, 7]$, $H$ contains a $(i, j)_{(u_1, u_i)}$-path for some $i\in \{2, 4\}$.
\end{proposition}

We distinguish the following two scenarios:
\begin{itemize}
\parskip=0pt
\item
    For some $\alpha\in T_3$, $H$ contains no $(i, \alpha)_{(u_3, v_i)}$-path for each $i\in [6, 7]$.
    It follows from Corollary~\ref{auv2}(4) that $5\in B_1\subseteq C(u_1)\cap C(v_1)$,
    and $H$ contains a $(5, i)_{(u_3, u_i)}$-path for some $i\in \{2, 4\}$.

    \quad When $H$ contains no $(2, \beta)_{(u_2, u_4)}$-path for some $\beta\in T_4$.
    It follows from Corollary~\ref{auv2}(4) that $4\in B_1\subseteq C(u_1)\cap C(v_1)$,
    and if $4\not\in C(u_3)$, then $H$ contains a $(2, 4)_{(u_2, u_3)}$-path.
    \begin{itemize}
    \parskip=0pt
    \item
         $1\not\in C(u_3)\cup C(u_4)$ and $2\not\in C(u_4)$.
         It follows that $2\in C(u_1)$, $T_1\subseteq C(u_2)$, $c_{34}\ge t_1\ge 2$.
         One sees that $|\{1, 2, 4, 5\}| + |\{2\}| + |\{3, 5\}| + |\{3, 4\}| + 2t_{u v} = 2(\Delta - 2) + 9 = 2\Delta + 5$.
         \begin{itemize}
         \parskip=0pt
         \item
              $T_3\subseteq C(u_2)$. One sees that $t_0\ge t_1 + t_3\ge 2 + 1 = 3$.
              First assume $2, 4\in C(u_3)$. Since $\sum_{x\in N(u)\setminus \{v\}}d(x)\ge 2\Delta + 5 + |\{2, 4, 5, 1/3\}| + t_0\ge 2\Delta + 9 + t_1 + t_3\ge 2\Delta + 9 + 2 + 2 = 2\Delta + 13$,
              it follows that $t_1 = t_3 = 2$, $d(u_1) = \Delta = 7$, $C(u_1) = \{1, 5, 4, 2\}\cup T_3\cup T_4$, and $b_3 = 5$, $5\not\in C(u_2)$, $6, 7\not\in S_u$.
              If $2\not\in B_1$, then $(u u_2, v u_3, u v)\overset{\divideontimes}\to (5, \alpha, 2)$.
              Otherwise, $2\in B_1$.
              Hence, by \ref{26264646}, we are done.
              Next, assume that $|C(u_3)\cap \{2, 4\}| = 1$.
              It follows from Corollary~\ref{auv2}(5.1) that $6, 7\in S_u$.
              Since $\sum_{x\in N(u)\setminus \{v\}}d(x)\ge 2\Delta + 5 + |\{5, 1/3, 6, 7, a_3\}| + t_0\ge 2\Delta + 10 + t_1 + t_3\ge 2\Delta + 10 + 2 + 1 = 2\Delta + 13$,
              it follows that $t_1 = 2$, $t_3 = 1$,
              and since $4\in C(u_2)\cup C(u_3)$, we have $a_3 = 4$, $2\not\in C(u_3)\cup C(u_4)$,
              and $H$ contains a $(4, 5)_{(u_3, u_4)}$-path, $5\in C(u_4)\setminus C(u_2)$.
              Then $(u u_2, v u_3)\to (5, \alpha)$,
              and $u v \overset{\divideontimes}\to 2$, or $(uu_3, uv)\overset{\divideontimes}\to (2, C_{34})$.
         \item
              $T_3\setminus C(u_2)\ne \emptyset$.
              Then $a_3 = 4$ and $H$ contains a $(4, T_3)_{(u_3, u_4)}$-path.
              It follows from Corollary~\ref{auv2}(5.2) and (3) that $2\in B_1\cap C(u_3)$,
              from Corollary~\ref{auv2}(5.1) that $6, 7\in S_u$.
              Since $\sum_{x\in N(u)\setminus \{v\}}d(x)\ge 2\Delta + 5 + |\{5, 1/3, 6, 7, 2, 4\}| + t_0\ge 2\Delta + 11 + t_1 \ge 2\Delta + 11 + 2 = 2\Delta + 13$,
              it follows that $c_{34} = t_1 = 2$, $d(u_1) = \Delta = 7$, $C(u_1) = \{1, 5, 4, 2\}\cup T_3\cup T_4$, and $W_3 = \{3, 5, 2, 4\}$, $W_4 = \{4, 3, b_3\}$.
              Hence, it follows from \ref{auv2}(5.1) that $H$ contains a $(2, i)_{(u_2, u_3)}$-path for each $i\in [6, 7]\setminus \{b_3\}$.
              Thus, by \ref{26264646}, we are done.
         \end{itemize}
    \item
         $1\in C(u_3)\cup C(u_4)$ or $2\in C(u_4)$.
         One sees that $|\{1, 4, 5\}| + |\{2\}| + |\{3, 5\}| + |\{3, 4\}| + 2t_{u v} = 2(\Delta - 2) + 8 = 2\Delta + 4$.
         When $2, 4\in C(u_3)$, $w_3 + w_4 \ge |\{2, 4, 3, 5\}| + |\{4, 3, b_3\}| + |\{1/2\}| = 8$,
         it follows that $W_4\subseteq \{4, 3, 1/2, b_3\}$, $c_{34} = t_1 = 1$,
         and $d(u_1) = \Delta = 7$, $C(u_1) = \{1, 5, 4\}\cup T_3\cup T_4$.
         One sees that there exists an $i\in [6, 7]\setminus C(u_4)$.
         Thus, by \ref{26264646}, we are done.
         Hence, $C(u_3)\cap \{2, 4\} = a_3$.
         It follows from Corollary~\ref{auv2}(5.1) that $6, 7\in S_u$.
         \begin{itemize}
         \parskip=0pt
         \item
             $a_3 = 2$. Then $H$ contains a $(2, T_3\cup \{4, 5\})_{(u_2, u_3)}$-path.
             If $3\not\in C(u_1)\cup C(u_2)$,
             then $(u u_2, u u_3, u v)\overset{\divideontimes}\to (3, T_3, T_1)$.
             Otherwise, $3\in C(u_1)\cup C(u_2)$.
             First, assume that $T_1\subseteq C(u_2)$.
             Since $\sum_{x\in N(u)\setminus \{v\}}d(x)\ge 2\Delta + 4 + |\{4, 5, 2, 6, 7, 3, 1/2\}| + t_0\ge 2\Delta + 11 + t_3 + t_1\ge 2\Delta + 11 + 1 + 1 = 2\Delta + 13$,
             it follows that $t_3 = 1$, $c_{34} = t_1 = 1$,
             $d(u_1) = \Delta = 7$, $C(u_1) = \{1, 5, 4\}\cup T_3\cup T_4$.
             It follows from \ref{26264646} that $H$ contains a $(4, i)_{(u_1, u_4)}$-path for each $i\in [6, 7]$ and $W_4 = \{3, 4, 6, 7, 1/2\}$.
             Thus, by \ref{314Cv2T4}, we are done.
             Next, assume that $T_1\setminus C(u_2)\ne \emptyset$.
             It follows from Corollary~\ref{auv2}(5.2) that $2\in C(u_1)\cup C(u_4)$ and $1\in S_u$.
		Since $\sum_{x\in N(u)\setminus \{v\}}d(x)\ge 2\Delta + 4 + |\{4, 5, 2, 6, 7, 3, 1, 2\}| + t_0\ge 2\Delta + 11 + t_3\ge 2\Delta + 12 + 1 = 2\Delta + 13$,
             it follows that $t_3 = 1$, $W_3 = \{3, 5, 2\}$, $1/2\in C(u_4)$, $b_3\in [6, 7]$ and mult$_{S_u}(1) = 1$, mult$_{S_u}(i) = 2$, $i\in [2, 3]$.
             One sees that if $1\not\in C(u_4)$ and $H$ contains no $(1, 2)_{(u_2, u_4)}$-path,
             then it follows from Corollary~\ref{auv2d(u)5}(7.2) that $2\in C(u_1)$ and $T_1\subseteq C(u_2)$.
             Hence, assume that $1\in C(u_4)$ or $H$ contains a $(1, 2)_{(u_2, u_4)}$-path.
             One sees clearly that $H$ contains no $(1, 2)_{(u_2, u_3)}$-path.
             If $2, 3\not\in C(u_1)$, then $(u u_1, u u_3, u v)\overset{\divideontimes}\to (3, 1, T_4)$.
             Otherwise, $2/4\in C(u_1)$.
             It follows that $c_{34} = t_1 = 2$, $W_1 = \{1, 5, 4, 2/3\}$, $W_4 = \{3, 4, 6, 2\}$ and $6\not\in C(u_2)$.
             It follows from \ref{314Cv2T4} that $H$ contains a $(4, 6)_{(u_4, v_6)}$-path.
             Then $(u u_1, u v)\overset{\divideontimes}\to (6, T_4)$.
         \item
             $a_3 = 4$. Then $H$ contains a $(4, T_3\cup \{5\})_{(u_3, u_4)}$-path.
             First, assume $2\in C(u_1)\cap C(u_4)$.
             It follows that $c_{34} = t_1 = 2$, $W_1 = \{1, 4, 5, 2\}$, $W_4 = \{4, 3, 5, 2\}$,
             and $6, 7\in C(u_2)$.
             One sees from Corollary~\ref{auv2}(5.1) that $H$ contains a $(2, 6)_{(u_2, u_4)}$-path.
             If $3\not\in C(u_2)$, then $(u u_1, u u_3, u v)\to (3, 1, T_4)$.
             Otherwise, $3\in C(u_2)$.
             Since $\sum_{x\in N(u)\setminus \{v\}}d(x)\ge 2\Delta + 4 + |\{2, 6, 7, 3, 4, 2, 5\}| + t_0\ge 2\Delta + 11 + t_0$ and $t_1 + t_3\ge 2 + 1 = 3$,
             it follows that $(T_3\cup T_1)\setminus C(u_2)\ne \emptyset$.
             Hence, $2\in B_1$.
             Thus, by \ref{26264646}, we are done.
             Next, assume that $2\not\in C(u_1)\cap C(u_4)$.
             It follows that $T_1\cup T_3\subseteq C(u_2)$.
             If $5\not\in C(u_2)$, then $(u u_2, v u_3)\to (5, \alpha)$,
             and by Lemma~\ref{auvge2}(2), we are done.
             Otherwise, $5\in C(u_2)$.
             If $H$ contains no $(3, j)_{(u_2, u_3)}$-path for some $j\in T_4\cap T_2$,
             then $u u_2\to j$ and by Lemma~\ref{auvge2}(2), we are done.
             Otherwise, $T_4\subseteq C(u_2)$ or $3\in T_4$.
             One sees that $2\Delta + 4 + |\{5, 6, 7, 4, 5, 1/2\}| + t_0 = 2\Delta + 10 + t_0$ and $t_3 + t_4\ge 3$.
             If $T_4\subseteq C(u_2)$, then $\sum_{x\in N(u)\setminus \{v\}}d(x)\ge 2\Delta + 10 + t_0\ge 2\Delta + 10 + t_1 + t_3 + t_4\ge 2\Delta + 10 + 1 + 3 = 2\Delta + 14$.
             Otherwise, $T_4\setminus C(u_2)\ne \emptyset$ and $3\in C(u_2)$.
             Hence, $\sum_{x\in N(u)\setminus \{v\}}d(x)\ge 2\Delta + 10 + t_0 + |\{3\}|\ge 2\Delta + 10 + t_1 + t_3 \ge 2\Delta + 11 + 1 + 1 = 2\Delta + 13$.
             It follows that $t_1 = t_3 = 1$, $d(u_1) = \Delta = 7$, $C(u_1) = \{1, 5, 4\}\cup T_3\cup T_4 = B_1\cup \{1\}$, and $1/2\in C(u_4)$.
             One sees that $[6, 7]\setminus C(u_4)\ne \emptyset$.
             Thus, by \ref{26264646}, we are done.
         \end{itemize}
    \end{itemize}
    \quad In the other case, $2\in C(u_4)$ and $H$ contains a $(2, T_4)_{(u_2, u_4)}$-path.
         It follows that $T_4\subseteq C(u_2)$ and $t_0\ge t_4\ge 2$.
         If $2, 4\in C(u_1)\cap C(u_3)$, then $c_{34}\ge t_1\ge 2$ and $d(u_3) + d(u_4)\ge w_3 + w_4 + t_{uv} + c_{u v}\ge |[2, 5]| + |\{2, 3, 4, b_3\}| + \Delta - 2 + c_{34} = \Delta + 6 + c_{34}\ge \Delta + 6 + 2 = \Delta + 8$.
         Otherwise, $\{2, 4\}\cap (C(u_1)\cap C(u_3))\ne \emptyset$.
         It follows from Corollary~\ref{auv2}(5.1) that $6, 7\in S_u$.
         One sees that $|\{1, 5, a_1\}| + |\{2\}| + |\{3, 5, 2/4\}| + |\{2, 3, 4\}| + |\{5, 6, 7, 1/3\}| + 2t_{u v} = 2\Delta + 10$ and $t_1\ge 1$, $t_3\ge 1$.
         It follows that $(T_3\cup T_1)\setminus C(u_2)\ne \emptyset$ and from Corollary~\ref{auv2}(5.2) that $2\in C(u_1)\cup C(u_3)$, say $2\in C(u_{i_1})$.
         Then $4\in C(u_{i_1})$ or $T_{i_1}\subseteq C(u_2)$.
         Hence, $\sum_{x\in N(u)\setminus \{v\}}d(x)\ge 2\Delta + 10 + t_4 + |4/T_{i_1}|\ge 2\Delta + 13$.
         It follows that $t_4 = 2$, $W_4 = \{2, 3, 4, b_3\}$, $\{2, 4\}\cap C(u_{i_2}) = \{4\}$,
         mult$_{S_u}(1) +  $mult$_{C(u_2)\cup C(u_1)}(3) = 1$.
         \begin{itemize}
         \parskip=0pt
         \item
              $2\in C(u_3)$ and $\{2, 4\}\cap C(u_1) = 4$.
              It follows from Corollary~\ref{auv2}(5.1) that $1\in S_u$ and then $3\not\in C(u_1)\cup C(u_2)$.
              If $H$ contains no $(4, j)_{(u_3, u_4)}$-path,
              then $(u u_2, u u_3, u v)\overset{\divideontimes}\to (3, j, C_{34})$.
              otherwise, $4\in C(u_3)$ and $H$ contains a $(4, T_3)_{(u_3, u_4)}$-path.
              It follows that $c_{34} = t_1 = 1$, $W_1 = \{1, 4, 5\}$, $W_3 = [3, 5]$, $W_4 = \{2, 3, 4, b_3\}$.
              One sees that $H$ contains a $(1, 2)_{(u_2, u_4)}$-path.
              Then $(u u_1, u u_3, u v)\overset{\divideontimes}\to (3, 1, T_4)$.
         \item
              $2\in C(u_1)$ and $\{2, 4\}\cap C(u_3) = 4$.
              One sees from Corollary~\ref{auv2}(7.1) that it $4\not\in C(u_1)$, then $1\in S_u$,
              and from Corollary~\ref{auv2}(5.1) that if $T_1\setminus C(u_2)\ne \emptyset$, then $1\in S_u$.
              Hence, $1\in S_u$, and $3\not\in S_u\setminus C(u_4)$.
              It follows from \ref{311243} that $1\in C(u_3)$ or $H$ contains  a $(1, 4)_{(u_3, u_4)}$-path.
              One sees that $1\not\in C(u_4)$, or $H$ contains a $(1, 4)_{(u_3, u_4)}$-path.
              It follows that $c_{34} = t_1 = 1$ and $4\not\in C(u_1)$.
              Then $(u u_2, u u_1, u v)\overset{\divideontimes}\to (1, T_1, T_3)$.
         \end{itemize}
\item
     For each $j\in T_3$, $H$ contains a $(6/7, j)_{(u_3, v)}$-path.
     Assume w.l.o.g. $H$ contains a $(7, 8)_{(u_3, v_6)}$-path, where $8\in T_3$.
     If $7\not\in C(u_1)$ and $H$ contains no $(7, i)_{(u_1, u_i)}$-path for each $i\in [2, 4]$,
     then $(u u_1, u v)\to (7, 8)$.
     Otherwise, $7\in C(u_1)$ or $H$ contains a $(7, i)_{(u_1, u_i)}$-path for some $i\in [2, 4]$.

     \quad One sees that if $\{4\}\cup T_3\cup T_4\subseteq C(v_1)$ and $C_{u v} = \{\mu\}$,
     then it follows from Lemma~\ref{auvge2}(1.2) that $C(v_1) = \{1, 4, \mu/5/6/7\}\cup T_3\cup T_4$, i.e., $3\not\in C(v_1)$.
     One sees from \ref{315234} that $5\in S_u\setminus C(u_3)$,
     and if $C(u_3)\cap \{2, 4\} = \{a_3\}$, then it follows from Corollary~\ref{auv2}(5.1) that $6\in S_u$.
     \begin{itemize}
     \parskip=0pt
     \item
          $W_3 = 5$, or $c_{34} = 2$.
          Since $w_3 + w_4 + c_{34}\le 9$,
          it follows that $\Delta = 7$, $d(u_1) + d(u_2)\le 13$, $W_3 = \{3, 4, b_3\}$.
          Then $4\in B_1$ and it follows from \ref{314Cv2T4} that $H$ contains a $(b_3, 4)_{(u_4, v)}$-path.
          One sees that if $c_{34} = 2$, $H$ contains no $(1, 3)_{(u_3, u_4)}$-path,
          and if $W_3 = 5$, one sees that $t_3 + t_4 = \Delta - 3$, then $3\not\in C(v_1)$.
          Hence, it follows from Corollary~\ref{auv2}(7.2) that $2\in C(u_1)$, $H$ contains a $(2, T_1)_{(u_1, u_2)}$-path and $C_{34}\subseteq C(u_1)\cup C(u_2)$.
          \begin{itemize}
          \parskip=0pt
          \item
               $W_3 = 5$. Then $C(u_1) = \{1, 2, 4\}\cup T_3\cup T_4$ and $T_1\cup \{5, 7\}\subseteq C(u_2)$.
               One sees from $2\not\in C(u_4)\cup C(v_1)$ that $2\not\in B_1\cup B_3$.
               Hence, it follows from Corollary~\ref{auv2}(5.2) and (3) that $T_3\subseteq C(u_2)$.
               Thus, $d(u_1) = d(u_2) = \Delta$, a contradiction.
          \item
              $c_{34} = 2$.
              Then $5\in S_u\setminus C(u_3)$ and $6\in S_u$.
              One sees that if $a_3 = 2$, then $T_3\subseteq C(u_2)$,
              and if $a_3 = 4$, then it follows from Corollary~\ref{auv2}(5.2) that $T_3\subseteq C(u_2)$.
              Hence, $\sum_{x\in N(u)\setminus \{v\}}d(x)\ge |\{1, 2, 4\}| + |\{2\}| + |\{a_3, 7, 3, 5\}| + |\{4, 3\}| + |\{5, 6, 1/3\}| + 2t_{u v} + t_0 \ge 2\Delta + 9 + t_3 + c_{34}\ge 2\Delta + 9 + 2 + 2 = 2\Delta + 13$.
It follows that $7\not\in S_u\setminus C(u_3)$, and then $3\in C(u_1)$, $1\not\in S_u$.
              Then $(uu_1, uu_3, uv)\overset{\divideontimes}\to (1, 7, 8)$.
          \end{itemize}

     \item
          In the other case, $W_3 = \{3, 5, a_3, 7\}$, $T_1 = C_{34}$ and $c_{34} = t_1 = 1$.
          One sees that $5\in S_u\setminus C(u_3)$ and $6\in S_u$.
          If $7\not\in S_u\setminus C(u_3)$ and $H$ contains no $(1, a_3)_{(u_3, u_{a_3})}$-path,
          then $(u u_1, u u_3, u v)\overset{\divideontimes}\to (7, 1, 8)$.
          Otherwise, we proceed with the following proposition:
          \begin{proposition}
          \label{no71a31a3}
          If mult$_{S_u}(7) = 1$, then $H$ contains a $(1, a_3)_{(u_3, u_{a_3})}$-path and $1\in C(u_{a_3})$.
          \end{proposition}

         \quad When $H$ contains no $(2, \beta)_{(u_2, u_4)}$-path for some $\beta\in T_4$.
         It follows from Corollary~\ref{auv2}(4) that $4\in B_1\subseteq C(u_1)\cap C(v_1)$,
         and if $4\not\in C(u_3)$, then $H$ contains a $(2, 4)_{(u_2, u_3)}$-path.
         One sees that if $a_3 = 2$, or if $a_3 = 4$, since $2\not\in B_1\cup B_3$,
         then it follows from Corollary~\ref{auv2}(1.2) or Corollary~\ref{auv2}(5.2) and (3) that $T_3\subseteq C(u_2)$.
         One sees that $7\in S_u\setminus C(u_3)$ or $3\in C(u_1)$,
         and it follows from Corollary~\ref{auv2d(u)5}(7.1) that $2\in C(u_1)$ or $1\in S_u$.
         \begin{itemize}
         \parskip=0pt
         \item
              $2\in C(u_3)\cap C(u_4)$. Then $4\in C(u_2)$.
              Then $\sum_{x\in N(u)\setminus \{v\}}d(x)\ge |\{1, 4\}| + |\{2, 4\}| + |\{2, 7, 3, 5\}| + |\{4, 3, 2\}| + |\{5, 6, 7/3, 1/2\}| + 2t_{u v} + t_0\ge 2\Delta + 11 + t_0\ge 2\Delta + 11 + t_3 \ge 2\Delta + 13$.
              It follows that $T_1\setminus C(u_2)\ne \emptyset$, mult$_{S_u}(3) + $mult$_{S_u}(7) = 3$,
              and from Corollary~\ref{auv2}(5.1) that $H$ contains a $(1, 2)_{(u_2, u_4)}$-path.
              Further, by \ref{311243} that $3\in C(u_1)\cup C(u_2)$.
              Thus, $7\not\in S_u\setminus C(u_3)$, and by \ref{no71a31a3}, we are done.
         \item
              $2\not\in C(u_3)\cap C(u_4)$.
              It follows from Corollary~\ref{auv2} (5.2) and (4) that $C_{34}\subseteq C(u_2)$,
              and if $3\not\in C(u_2)$, then $T_4\subseteq C(u_2)$.
              Then $\sum_{x\in N(u)\setminus \{v\}}d(x)\ge |\{1, 4\}| + |\{2, T_4/3\}| + |\{a_3, 7, 3, 5\}| + |\{4, 3\}| + |\{5, 6, 7/3, 1/2\}| + 2t_{u v} + t_0\ge 2\Delta + 10 + t_3 + c_{34}\ge 2\Delta + 10 + 2 + 1 = 2\Delta + 13$.
              It follows that $a_3 = 4$, $2\not\in C(u_4)$, $4\not\in C(u_2)$,
              either $2\in C(u_1)$, or $1\in S_u$,
              and either $3\in C(u_1)$ or $7\in S_u\setminus C(u_1)$.
              If $2\in C(u_1)$, one sees that $1\not\in S_u$,
              then $(u u_2, u u_1, u v)\to (1, \{5, 7\}\setminus \{\rho\}, 8)$.
              Otherwise, $2\not\in C(u_1)$ and then $1\in S_u$.
              It follows from Corollary~\ref{auv2d(u)5} (7.1) that $1\in C(u_4)$.
              If $H$ contains no $(3, j)_{(u_1, u_3)}$-path for some $j\in \{5, 7\}\setminus C(u_1)$,
              then $(u u_2, u u_1, u v)\to (1, j, 8)$.
              Otherwise, $3\in C(u_1)$ and $H$ contains a $(3, j)_{(u_1, u_3)}$-path for each $j\in \{5, 7\}\setminus C(u_1)$, i.e., $b_3 = 5$.
              Then $(u u_3, u v)\overset{\divideontimes}\to (6, C_{34})$,
              or $(u u_4, u u_3, uv)\overset{\divideontimes}\to (6, T_3, C_{34})$\footnote{Or since $6\not\in C(u_1)\cup C(u_3)\cup C(u_4)$, by Corollary~\ref{auv2} (5.1), we are done. }.
         \end{itemize}
         \quad In the other case,
         $2\in C(u_4)$ and $H$ contains a $(2, T_4)_{(u_2, u_4)}$-path.
         It follows that $W_3 = \{3, 5, a_3, 7\}$, $W_4 = \{4, 3, 2, b_3\}$,
         $T_4\subseteq C(u_2)$ and $t_0\ge t_4\ge 2$.
         One sees that $|\{1, a_1\}| + |\{2\}| + |\{3, 5, 7, a_3\}| + |\{4, 3, 2\}| + |\{5, 6\}| + 2t_{u v} = 2\Delta - 4 + 12 = 2\Delta + 8$.
         \begin{itemize}
         \parskip=0pt
         \item
              $T_3\subseteq C(u_2)$.
              Since $\sum_{x\in N(u)\setminus \{v\}}d(x)\ge 2\Delta + 8 + |\{1/4\}| + t_0\ge 2\Delta + 9 + t_3 + t_4 = 2\Delta + 13$,
              we have $7\not\in S_u\setminus C(u_3)$ and then $3\in C(u_1)$, $1\not\in S_u$.
              Thus, by \ref{no71a31a3}, we are done.
         \item
              $T_3\setminus C(u_2)\ne \emptyset$.
              It follows from Corollary~\ref{auv2} (1.2) that $a_3 = 4$,
              and from Corollary~\ref{auv2}(5.2) that $2\in C(u_1)$.
              One sees clearly that $H$ contains no $(1, 4)_{(u_3, u_4)}$-path.
              It follows from~\ref{311243} that $3\in C(u_1)\cup C(u_2)$.
              One sees from Corollary~\ref{auv2d(u)5}(7.1) that $4\in C(u_1)$ or $1\in S_u$,
              and from Corollary~\ref{auv2} (1.2) and (5.1) that if $C_{34}\setminus C(u_2)\ne \emptyset$,
              then $4\in C(u_1)$ and $1\in S_u$.
              Since $\sum_{x\in N(u)\setminus \{v\}}d(x)\ge 2\Delta + 8 + |\{3, 1/4\}| + t_0\ge 2\Delta + 10 + t_4 + c_{34} = 2\Delta + 13$ if $C_{34}\subseteq C(u_2)$,
              or else, $\sum_{x\in N(u)\setminus \{v\}}d(x)\ge 2\Delta + 8 + |\{3, 1, 4\}| + t_0\ge 2\Delta + 11 + t_4 = 2\Delta + 13$,
              we have $7\not\in S_u\setminus C(u_3)$.
              Thus, by \ref{no71a31a3}, we are done.
         \end{itemize}
     \end{itemize}
\end{itemize}

(2.4.2.) $c(v u_4) = 1$.
Since $T_{u v}\subseteq C(u_3)\cup C(u_4)$,
let $C(u_3) = W_3\cup T_4\cup C_{3 4}$, $C(u_4) = W_4\cup T_3\cup C_{3 4}$.

If $C(u_3) = \{3, 5\}\cup T_{u v}$, then $(v u_3, u u_3)\to (4, 5)$ reduces the proof to (2.2.1.).
If $C(u_4) = \{1, 4\}\cup T_{u v}$, then $(v u_4, u u_4)\to (4, 5)$ reduces the proof to (2.1.1.).
Otherwise, $T_3\ne \emptyset$, $T_4\ne \emptyset$,
and there exists a $a_3\in \{2, 4\}\cap C(u_3)$, $b_1\in [5, 7]\cap C(u_4)$.
Assume w.l.o.g. $8\in T_3$.

Assume that $1\not\in C(u_3)$ and $H$ contains no $(1, 2)_{(u_2, u_3)}$-path.
If $5\not\in C(u_1)$ and $H$ contains no $(5, i)_{(u_1, u_i)}$-path for each $i\in \{2, 4\}$,
then $(u u_1, u u_3, u v)\overset{\divideontimes}\to (5, 1, T_3)$.
If for each $j\in [6, 7])\setminus C(u_1)$, $H$ contains no $(j, i)_{(u_1, u_i)}$-path for each $i\in \{2, 4\}$,
then $(u u_1, u u_3)\to (j, 1)$ reduces the proof to (2.4.1.).
One sees from Corollary~\ref{auv2d(u)5}(2.1) that $3\in C(u_1)$ or $H$ contains a $(3, i)_{(u_1, u_i)}$-path for some $i\in \{2, 4\}$.
Hence, we proceed with the following proposition:
\begin{proposition}
\label{131Cu33Cu1}
If $1\not\in C(u_3)$ and $H$ contains no $(1, 2)_{(u_2, u_3)}$-path,
then for each $j\in \{3, 5, 6, 7\}$,
$j\in C(u_1)$ or $H$ contains a $(j, i)_{(u_1, u_i)}$-path for some $i\in \{2, 4\}$.
\end{proposition}

Assume that $1\not\in C(u_3)$, $H$ contains no $(1, 2)_{(u_2, u_3)}$-path,
and $H$ contains no $(2, j)_{(u_2, u_3)}$-path for some $j\in T_3$.
One sees from \ref{131Cu33Cu1} that for each $i\in \{3, 5, 6, 7\}\setminus C(u_1)$,
$H$ contains a $(2, i)_{(u_1, u_2)}$-path.
If $3\not\in C(u_1)\cup C(u_4)$,
then $(u u_4, u u_3)\to (3, j)$ reduces the proof to (2.2.2.).
If $5\not\in C(u_1)\cup C(u_4)$,
then $(u u_4, u u_3)\to (5, j)$ reduces the proof to (2.2.1.).
If there exists a $i\in [6, 7]\setminus (C(u_1)\cup C(u_3)\cup C(u_4))$,
then $(u u_4, u u_3)\to (i, j)$ reduces the proof to (2.2.2.).
Hence, we proceed with the following proposition:
\begin{proposition}
\label{131Cu33567}
If $1\not\in C(u_3)$, $H$ contains no $(1, 2)_{(u_2, u_3)}$-path,
and $H$ contains no $(2, j)_{(u_2, u_3)}$-path for some $j\in T_3$,
then $3, 5\in C(u_1)\cup C(u_4)$ and $6, 7\in C(u_1)\cup C(u_4)\cup C(u_3)$.
\end{proposition}

Assume that $H$ contains a $(\rho, T_3)_{(u_3, v_\rho)}$-path for some $\rho\in [6, 7]\cap C(u_3)$.
If $\rho\not\in S_u\setminus C(u_3)$,
one sees from \ref{31rho8234} that $H$ contains a $(3, \rho)_{(u_1, u_3)}$-path:
if $4\not\in B_3$,
then $(u u_4, u v)\overset{\divideontimes}\to (\rho, 4)$;
if $2\not\in B_1\cup B_3$,
then $(u u_2, u v)\overset{\divideontimes}\to (\rho, 2)$;
if $H$ contains no $(2, j)_{(u_2, u_3)}$-path for some $j\in T_3$,
then $(u u_4, u u_3)\to (\rho, j)$ reduces the proof to (2.2.2.).
If $1\not\in C(u_3)$, $H$ contains no $(1, 2)_{(u_2, u_3)}$-path,
$H$ contains no $(2, j)_{(u_2, u_3)}$-path for some $j\in T_3$,
and $\rho\not\in C(u_1)\cup C(u_4)$,
one sees from \ref{131Cu33Cu1} that $H$ contains a $(2, \rho)_{(u_1, u_2)}$-path,
then $(u u_4, u u_3)\to (\rho, j)$ reduces the proof to (2.2.2.).
Hence, we proceed with the following proposition:
\begin{proposition}
\label{31rho84B32T3}
Assume that $H$ contains a $(\rho, T_3)_{(u_3, v_\rho)}$-path, $\rho\in [6, 7]\cap C(u_3)$.
\begin{itemize}
\parskip=0pt
\item[{\rm (1)}]
	if $\rho\not\in S_u\setminus C(u_3)$,
    then $4\in B_3$, $2\in B_1\cup B_3$, and $H$ contains a $(2, T_3)_{(u_2, u_3)}$-path, $2, 4\in C(u_3)$, $T_3\subseteq C(u_2)$.
\item[{\rm (2)}]
	if $1\not\in C(u_3)$, $H$ contains no $(1, 2)_{(u_2, u_3)}$-path, and $H$ contains no $(2, j)_{(u_2, u_3)}$-path for some $j\in T_3$, then $\rho\in C(u_1)\cup C(u_4)$.
\end{itemize}
\end{proposition}


If $5\not\in C(u_4)$, $H$ contains no $(5, i)_{(u_4, u_i)}$-path for each $i\in [1, 2]$,
and $H$ contains no $(2, j)_{(u_2, u_3)}$-path for some $j\in T_3$,
then $(u u_4, u u_3, u v)\overset{\divideontimes}\to (5, j, T_4)$.
If $5\not\in C(u_2)$, $H$ contains no $(5, i)_{(u_2, u_i)}$-path for each $i\in \{1, 4\}$,
and $H$ contains no $(4, j)_{(u_4, u_3)}$-path for some $j\in T_3$,
then $(u u_2, u u_3)\to (5, j)$ reduces the proof to (2.3).
One sees from \ref{315234} that if $5\not\in S_u\setminus C(u_3)$, then $H$ contains a $(3, 5)_{(u_1, u_3)}$-path.
In the other case, we proceed with the following proposition:
\begin{proposition}
\label{135Cu45Cu2}
\begin{itemize}
\parskip=0pt
\item[{\rm (1)}]
	If $5\not\in C(u_4)$, and $H$ contains no $(5, i)_{(u_4, u_i)}$-path for each $i\in [1, 2]$,
    then $H$ contains a $(2, T_3)_{(u_2, u_3)}$-path and $T_3\subseteq C(u_2)$.
\item[{\rm (2)}]
	If $5\not\in C(u_2)$, and $H$ contains no $(5, i)_{(u_2, u_i)}$-path for each $i\in \{1, 4\}$,
    then $H$ contains a $(4, T_3)_{(u_4, u_3)}$-path and $T_3\subseteq C(u_4)$.
\item[{\rm (3)}]
	If $5\not\in S_u\setminus C(u_3)$,
    then $2, 4\in C(u_3)\cap C(u_1)$, $H$ contains a $(2, T_3)_{(u_2, u_3)}$-path, a $(4, T_3)_{(u_4, u_3)}$-path,
    and $(4, T_1)_{(u_1, u_4)}$-path, a $(2, T_1)_{(u_2, u_1)}$-path.
\end{itemize}
\end{proposition}

One sees from Corollary~\ref{auv2} (1.2) that if $H$ contains a $(2, T_4)_{(u_2, u_4)}$-path,
then $T_4\subseteq C(u_1)$.

\begin{lemma}
\label{134inCu3}
$4\in C(u_3)$.
\end{lemma}
\begin{proof}
Assume that $4\not\in C(u_3)$.
It follows that $2\in C(u_3)$, $H$ contains a $(2, T_3)_{(u_2, u_3)}$-path,
$2\in C(u_4)$, $H$ contains a $(2, T_4)_{(u_2, u_4)}$-path,
and then $T_4\subseteq C(u_1)$,
and it follows from Corollary~\ref{auv2}(5.1) that $3/1/4\in C(u_2)$, $6, 7\in S_u$.
One sees from $T_3\cup T_4\subseteq C(u_1)\cap C(u_2)$ that $t_0\ge t_3 + t_4$.

When for some $j\in T_3$, $H$ contains no $(j, i)_{(u_3, v_i)}$-path for each $i\in [6, 7]$,
it follows from Corollary~\ref{auv2}(4) that $5\in B_1$ and $H$ contains a $(2, 5)_{(u_2, u_3)}$-path,
i.e., $5\in C(u_i)$, $i\in \{1, 2, 4\}$.
Since $\sum_{x\in N(u)\setminus \{v\}}d(x)\ge\sum_{i\in [1, 4]}w_i + 2t_{uv} + t_0\ge |\{1, 5\}| + |\{2, 5, 3/1/4\}| + |\{3, 5, 2\}| + |\{4, 1, 2, 5\}| + |\{6, 7\}| + 2t_{u v} + t_3 + t_4 = 2\Delta + 10 + 1 + 2 = 2\Delta + 13$,
we have $\Delta = 7$, $t_0 = t_3 + t_4$, where $t_3 = 1$, $t_4 = 2$, and $2, 4\not\in C(u_1)$.
Thus, $C_{34}\ne \emptyset$ and $T_1\ne \emptyset$, a contradiction to Corollary~\ref{auv2}(1.2).

In the other case, for each $j\in T_3$, $H$ contains a $(j, 6/7)_{(u_3, v)}$-path,
assume w.l.o.g. $H$ contains a $(7, 8)_{(u_4, v_7)}$-path.
We distinguish whether or not $H$ contains a $(7, T_3)_{(u_3, v_7)}$-path.
\begin{itemize}
\parskip=0pt
\item
     $H$ contains a $(7, T_3)_{(u_3, v_7)}$-path.
     It follows from \ref{135Cu45Cu2}(3) that $5\in S_u\setminus C(u_3)$,
     and from \ref{31rho84B32T3}(1) that $7\in S_u\setminus C(u_3)$.
     Since $\sum_{x\in N(u)\setminus \{v\}}d(x)\ge\sum_{i\in [1, 4]}w_i + 2t_{u v} + t_0\ge |\{1\}| + |\{2, 3/1/4\}| + |\{3, 5, 2, 7\}| + |\{4, 1, 2\}| + |\{5, 6, 7\}| + 2t_{u v} + t_3 + t_4 = 2\Delta + 9 + 2 + 2 = 2\Delta + 13$,
we have $\Delta = 7$, $t_0 = t_3 + t_4$, where $t_3 = t_4 = 2$, and $2, 4\not\in C(u_1)$.
     Thus, $C_{34}\ne \emptyset$ and $T_1\ne \emptyset$, a contradiction to Corollary~\ref{auv2}(1.2).
\item
     $6\in C(u_3)$. Then $w_3\ge 5$, $w_4\ge 4$ and $t_3\ge 3$, $t_4\ge 2$.
     Hence, $\sum_{x\in N(u)\setminus \{v\}}d(x)\ge\sum_{i\in [1, 4]}w_i + 2t_{u v} + t_0\ge |\{1, 5/2/3/4\}| + |\{2, 3/1/4\}| + w_3 + w_4 + 2t_{u v} + t_3 + t_4 = 2\Delta + 9 + 3 + 2 = 2\Delta + 14$.
\end{itemize}
\end{proof}

Assume $5\not\in C(u_2)$ and $H$ contains no $(5, i)_{(u_2, u_i)}$-path for each $i\in \{1, 3, 4\}$.
If $2\not\in B_1\cup B_3$, then $(u u_2, u v)\overset{\divideontimes}\to (5, 2)$.
If $2\not\in C(u_1)$ and $H$ contains no $(2, 4)_{(u_1, u_4)}$-path,
then $(u u_2, u u_1, u v)\overset{\divideontimes}\to (5, 2, T_3)$.
If $2\not\in C(u_3)$ and $H$ contains no $(2, 4)_{(u_3, u_4)}$-path,
then $(u u_2, u u_3)\to (5, 2)$ reduces the proof to (2.3.).
Hence, we proceed with the following proposition, or otherwise we are done:
\begin{proposition}
\label{135Cu22424}
If $5\not\in C(u_2)$ and $H$ contains no $(5, i)_{(u_2, u_i)}$-path for each $i\in \{1, 3, 4\}$,
then $2\in B_1\cup B_3$, for each $i\in \{1, 3\}$, $2/4\in C(u_i)$,
and if $2\not\in C(u_i)$, then $H$ contains a $(2, 4)_{(u_i, u_4)}$-path and $4\in C(u_i)$.
\end{proposition}

By Corollary~\ref{auv2}(5.2) and (3), we proceed with the following proposition:
\begin{proposition}
\label{13T3T12424}
If $(T_3\cup T_1)\setminus C(u_2)\ne \emptyset$,
then $2\in B_1\cup B_3$, for each $i\in \{1, 3\}$, $2/4\in C(u_i)$,
and if $2\not\in C(u_i)$, then $H$ contains a $(2, 4)_{(u_i, u_4)}$-path and $4\in C(u_i)$.
\end{proposition}

\begin{lemma}
\label{13non2jT4}
\begin{itemize}
\parskip=0pt
\item[{\rm (1)}]
	There exists a $j\in T_4$ such that $H$ contains no $(2, j)_{(u_2, u_4)}$-path.
\item[{\rm (2)}]
	$4\in B_3$, $2/4\in C(u_1)$, and if $4\not\in C(u_1)$, then $H$ contains $(2, 4)_{(u_1, u_2)}$-path.
\item[{\rm (3)}]
	There exists a $\gamma_0\in [5, 7]\cap C(u_4)$ such that $H$ contains a $(\gamma_0, 4)_{(u_4, v_{\gamma_0})}$-path.
\end{itemize}
\end{lemma}
\begin{proof}
For (1), assume that $H$ contains a $(2, T_4)_{(u_2, u_4)}$-path.
Then $2\in C(u_4)$, $T_4\subseteq C(u_2)\cap C(u_1)$.
One sees from \ref{315234} that $5/2/3/4\in C(u_1)$.

When $W_3 = 5$, it follows that $\Delta = 7$, $t_4 = 2$, $t_3 = 3$, and $T_3\cup T_4 = T_{uv}$.
Then $d(u_1) = \Delta$ and $d(u_2)\le 6$.
If $2, 4\not\in C(u_1)$,
one sees from~\ref{135Cu22424} that $5/1/3/4\in C(u_2)$ and then $T_3\setminus C(u_2)\ne \emptyset$,
then by \ref{13T3T12424}, we are done.
Otherwise, $2/4\in C(u_1)$.
Then assume w.l.o.g. $H$ contains a $(7, 8)_{(u_4, v_7)}$-path.
If $4\in C(u_1)$, then $(u u_1, u v)\overset{\divideontimes}\to (\{5, 7\}\setminus C(u_4), 8)$.
Otherwise, $2\in C(u_1)$ and it follows from \ref{315234} that $H$ contains a $(2, 5)_{(u_1, u_2)}$-path,
from \ref{31rho8234} that $H$ contains a $(2, 7)_{(u_1, u_2)}$-path, and $5, 7\in C(u_2)$.
One sees clearly that $T_3\setminus C(u_2)\ne \emptyset$.
It follows from Corollary~\ref{auv2d(u)5}(7.2) that $3\in C(u_2)$ and then $C(u_2) = \{2, 3, 5, 7\}\cup T_4$.
Then $(u u_4, u v)\to (\{5, 7\}\setminus C(u_4), 4)$ reduces the proof to (2.1.).

When $1, 2\not\in C(u_3)$,
it follows from \ref{131Cu33567} that $3, 5\in C(u_1)\cup C(u_4)$ and $6, 7\in C(u_1)\cup C(u_3)\cup C(u_4)$.
Let $|\{6, 7\}\cap C(u_3)| = \tau_{67}$.
One sees that $t_3 = \Delta - 2 - (d(u_3) - w_3) = \Delta - d(u_3) + w_3 - 2\ge w_3 - 2 = 3 + \tau_{67} - 2$.
Then $d(u_1) + d(u_4)\ge |\{1\}| + |\{4, 1, 2\}| + |\{3, 5\}| + |\{6, 7\}\setminus C(u_3)| + t_{u v} + t_0\ge 6 + 2 - \tau_{67} + \Delta - 2 + 3 + \tau_{67} - 2 = \Delta + 7$.
It follows that mult$_{C(u_1)\cup C(u_4)}(5) = 1$ and for each $i\in [6, 7]$, $i\not\in C(u_3)\cap (C(u_1)\cup C(u_4))$.
Thus, $H$ contains a $(7, 8)_{(u_3, v_7)}$-path and by \ref{31rho84B32T3}(2), we are done.

In the other case, $W_3 = \{3, 5, 4, 1\}$ or $W_3 = \{3, 5, 4, 2\}$.
One sees clearly from Corollary~\ref{auv2}(4)
that $5\in B_1$ and $H$ contains a $(5, i)_{(u_3, u_i)}$-path for some $i\in \{2, 4\}\cap C(u_3)$.
One sees that $|\{1, 5\}| + |\{2\}| + |\{4, 1/2, 3, 5\}| + |\{4, 1, 2, 5\}| + 2t_{u v} = 2\Delta + 7$,
and $t_0\ge t_4\ge 2$.

It follows from Corollary~\ref{auv2} (5.1) that if $2\not\in C(u_3)$ or $T_3\setminus C(u_2)\ne \emptyset$,
then $3, 6, 7\in S_u$ and $\sum_{x\in N(u)\setminus \{v\}}d(x)\ge 2\Delta + 7 + |\{3, 6, 7\}| + t_0\ge 2\Delta + 10 + t_4\ge 2\Delta + 10 + 2 = 2\Delta + 12$.

\begin{itemize}
\parskip=0pt
\item
     $W_3 = \{3, 5, 4, 1\}$.
     Then $3, 6, 7\in S_u$ and $\sum_{x\in N(u)\setminus \{v\}}d(x)\ge 2\Delta + 12$.
     It follows that $T_3\setminus C(u_2)\ne \emptyset$ and from \ref{13T3T12424} that $2\in B_1$.
     Hence, $\sum_{x\in N(u)\setminus \{v\}}d(x)\ge 2\Delta + 12 + |\{2\}| = 2\Delta + 13$.
     It follows that $t_0 = t_4$, $4\not\in C(u_1)$ and $t_3 + t_4 = \Delta - 3$.
     Then $(u u_1, u v)\overset{\divideontimes}\to (C_{34}, T_3)$\footnote{Or, since $T_1 = C_{34} \ne \emptyset$ and $4\not\in C(u_1)$, $T_1\setminus C(u_2)\ne \emptyset$, a contradiction to Corollary~\ref{auv2} (1.2). }.
\item
     $W_3 = \{3, 5, 4, 2\}$.
     When $T_3\setminus C(u_2)\ne \emptyset$,
     it follows from \ref{13T3T12424} that $2/4\in C(u_1)$.
     Hence, $\sum_{x\in N(u)\setminus \{v\}}d(x)\ge 2\Delta + 12 + |\{2/4\}| = 2\Delta + 13$.
     It follows that $t_0 = t_4$, $c_{34} = 1$, $C(u_1) = \{1, 5, 2/4\}$, $C(u_2) = \{2, 3, 6, 7\}\cup T_4$.
     Since $C_{34}\setminus (C(u_1)\cup C(u_2))\ne \emptyset$, $4\in C(u_1)$.
     Thus, by \ref{131Cu33Cu1}, we are done.
     In the other case, $T_3\subseteq C(u_2)$.
     Then $\sum_{x\in N(u)\setminus \{v\}}d(x)\ge 2\Delta + 7 + t_0\ge 2\Delta + 7 + t_3 + t_4\ge 2\Delta + 7 + 2  + 2 = 2\Delta + 11$.
     One sees clearly that $\{3, 6, 7\}\setminus S_u\ne \emptyset$.
     It follows from \ref{131Cu33Cu1} that $1\in C(u_2)$ and $H$ contains a $(1, 2)_{(u_2, u_3))}$-path.
     If $5\not\in C(u_2)$ and $2, 4\not\in C(u_1)$, then $u u_2\to 5$, $v u_3\in \alpha\in T_3$,
     and $u v\to 2$, or $(u u_1, u v)\to (2, T_3\setminus \{\alpha\})$.
     Otherwise, $5\in C(u_2)$ or $2/4\in C(u_1)$.
     Thus, $\sum_{x\in N(u)\setminus \{v\}}d(x)\ge 2\Delta + 11 + |\{1\}| + |\{5/2/4\}| = 2\Delta + 13$.
     It follows that $\Delta = 7$, $t_0 = t_3 + t_4 = 2 + 2 = 4$ and then $4\in C(u_1)$, $C(u_2) = \{1, 2\}\cup T_3\cup T_4$, $6\not\in S_u$.
     Then $(u u_4, u u_1)\to (6, C_{34})$ reduces the proof to (2.1.2.) and (1) holds.
\end{itemize}

By (1), it follows from Corollary~\ref{auv2}(5.2), (4) that (2) holds naturally.
For (3), if $H$ contains no $(4, i)_{(u_4, v_i)}$-path for each $i\in [5, 7]\cap C(u_4)$,
then $(v u_4, u u_4, u v)\overset{\divideontimes}\to (4, j, T_3)$.
Otherwise, there exists a $\gamma_0\in [5, 7]\cap C(u_4)$ such that $H$ contains a $(\gamma_0, 4)_{(u_4, v_{\gamma_0})}$-path.
\end{proof}

Assume $6\not\in C(u_1)\cup C(u_3)$ and $T_1\cup T_3 = T_{u v}$.
If $H$ contains no $(6, i)_{(u_i, u_j)}$-path for each $i\in \{2, 4\}$, $j\in \{1, 3\}$,
one sees that $u u_1\to 6$ or $u u_3\to 6$ reduces to an acyclic edge $(\Delta + 5)$-coloring for $H$ such that $|C(u)\cap C(v)| = 2$,
we may further assume that $H$ contains a $(6, T_3)_{(u_1, v_6)}$-path,
and a $(6, T_1\cup \{4\})_{u_3, v_6}$-path,
i.e., $C(v_6) = \{4, 6\}\cup T_{uv}$.
Then $(v v_6, u v)\overset{\divideontimes}\to (2, 6)$.
Hence, we proceed with the following proposition,
\begin{proposition}
\label{136B1B3}
For each $\eta\in \{6, 7\}$, if $\eta\not\in C(u_1)\cup C(u_3)$ and $T_1\cup T_3 = T_{u v}$,
then $H$ contains a $(\eta, i)_{(i, j)}$-path for some $i\in \{2, 4\}$, $j\in \{1, 3\}$.
\end{proposition}

\begin{lemma}
\label{13T42Cu1}
\begin{itemize}
\parskip=0pt
\item[{\rm (1)}]
	$T_4\setminus C(u_1)\ne \emptyset$.
\item[{\rm (2)}]
	$2\in C(u_1)$, and $H$ contains a $(2, T_1)_{(u_1, u_2)}$-path.
\item[{\rm (3)}]
	$1\in C(u_2)$ or $H$ contains a $(1, 3)_{(u_2, u_3)}$-path.
\end{itemize}
\end{lemma}
\begin{proof}
Assume that $T_4\subseteq C(u_1)$.
One sees from Lemma~\ref{13non2jT4}(2) that $2/4\in C(u_1)$.

When $1, 2\not\in C(u_3)$,
it follows from \ref{131Cu33567} that $3, 5\in C(u_1)\cup C(u_4)$ and $6, 7\in C(u_1)\cup C(u_3)\cup C(u_4)$.
Let $|\{6, 7\}\cap C(u_3)| = \tau_{67}$.
One sees that $t_3 = 3 + \tau_{67} - 2 = \tau_{67} - 1$.
Then $d(u_1) + d(u_4)\ge |\{1, 2/4\}| + |\{4, 1\}| + |\{3, 5\}| + |\{6, 7\}\setminus C(u_3)| + t_{uv} + t_0\ge 6 + 2 - \tau_{67} + \Delta - 2 + 3 + \tau_{67} - 2 = \Delta + 7$.
It follows that mult$_{C(u_1)\cup C(u_4)}(5) = 1$ and for each $i\in [6, 7]$, $i\not\in C(u_3)\cap (C(u_1)\cup C(u_4))$.
Thus, $H$ contains a $(7, 8)_{(u_3, v_7)}$-path and by \ref{31rho84B32T3}(2), we are done.
In the other case, $1/2\in C(u_3)$.
One sees from  $w_3 + w_4\ge 7$ that $c_{34}\le 2$ and $t_3 + t_4\ge \Delta - 4$, $w_1\le 4$.

{\bf Case 1.} There exists a $j\in T_3$ such that $H$ contains no $(i, j)_{(u_3, v_i)}$-path for each $i\in [6, 7]$.

It follows from Corollary~\ref{auv2}(3) that $5\in B_1$ and $H$ contains a $(5, i)_{(u_i, u_)}$-path for some $i\in \{2, 4\}\cap C(u_3)$.
Since $1, 5, 2/4\in C(u_1)$, $c_{34}\ge 1$, $\{6, 7\}\setminus C(u_1)\ne \emptyset$,
and $w_3 + w_4\le 8$, $w_3\le 5$.

Assume that $W_4 = \{1, 4, 5\}$.
If $H$ contains no $(2, j)_{(u_1, u_2)}$-path for some $j\in T_1$,
then $(u u_4, u u_1)\to (\{6, 7\}\setminus C(u_1), j)$ reduces the proof to (2.1.2.).
If $H$ contains no $(2, j)_{(u_2, u_3)}$-path for some $j\in T_3$,
then $(u u_4, u u_3)\to (\{6, 7\}\setminus C(u_1), j)$ reduces the proof to (2.2.2.).
Otherwise, $2\in C(u_1)\cap C(u_3)$ and $T_3\cup T_1 \subseteq C(u_2)$.

When $w_3 = 5$, it follows that $c_{34} = t_1 = 1$, $W_4 = \{4, 1, 5\}$,
$C(u_1) = \{1, 5, 2\}\cup T_3\cup T_4$, $d(u_2)\le 6$, $4\in C(u_2)$, $T_3\cup T_1\subseteq C(u_2)$,
and it follows from Corollary~\ref{auv2}(5.1) that $1/3\in C(u_2)$.
Thus, $d(u_2)\ge |\{2, 4, 1/3\}\cup T_3\cup T_1| = 7$, a contradiction.
In the other case, $w_3 = 4$ and $W_3 =\{3, 5, 4, 2/1\}$.

When $w_1 = 4$, it follows that $c_{34} = t_1 = 2$, $W_4 = \{4, 1, 5\}$, $2\in C(u_1)$,
and $T_3\cup T_1\subseteq C(u_2)$, $d(u_2)\le 6$.
Then $\sum_{x\in N(u)\setminus \{v\}}d(x)\ge |\{1, 2, 5\}| + |\{2\}| + |\{3, 5, 4, 2\}| + |\{1, 4, 5\}| + |\{4\}|+ 2t_{u v} + t_0 = 2\Delta + 8 + t_0\ge 2\Delta + 8 +  t_3 + c_{34}\ge 2\Delta + 8 + 2 + 2 = 2\Delta + 12$.
One sees that if $4\not\in C(u_1)$, then $1, 6, 7\in S_u$,
and if $H$ contains no $(1, 2)_{(u_2, u_3)}$-path, then $3, 6, 7\in S_u$.
It follows that $C(u_1) = \{1, 2, 4, 5\}\cup T_3\cup T_4$ and $C(u_2) = \{1, 2\}\cup T_3\cup C_{34}$.
Then $(u u_2, u u_3, u u_4)\to (5, T_3, 3)$ reduces the proof to (2.3.).

Hence, we assume $W_1 = \{1, 5, 2/4\}$.
We distinguish whether or not $2\in C(u_3)$:
\begin{itemize}
\parskip=0pt
\item
     $W_3 = \{3, 5, 4, 1\}$.
     If $H$ contains no $(2, i)_{(u_2, u_4)}$-path for some $i\in \{3, 6, 7\}\setminus C(u_4)$,
     then $(u u_4, u u_3)\to (i, T_3)$ reduces the proof to (2.2.2.).
     Otherwise, $2\in C(u_4)$.
     It follows that $W_4 = \{4, 1, 5, 2\}$, $C_{34} = T_1$, $c_{34} = t_1 = 1$, $d(u_i) = \Delta$, $i = 1, 3, 4$,
     $d(u_2)\le 6$, and $H$ contains a $(2, \{3, 6, 7\})_{(u_2, u_4)}$-path, $3, 6, 7\in C(u_2)$.
     One sees that $(T_1\cup T_3)\setminus C(u_2)\ne \emptyset$.
     It follows from \ref{13T3T12424} that $2\in B_1$.
     Then $C(u_1) = \{1, 5, 2\}\cup T_3\cup T_4$ and $C(u_2) = \{2, 3, 4, 6, 7\}\cup C_{34}$.
     Thus, by \ref{136B1B3}, we are done.
\item
     $W_3 = \{3, 5, 4, 2\}$.
     Since $|\{1, 5, 2/4\}| + |\{2\}| + |\{3, 5, 4, 2\}| + |\{4, 1, 5\}| + |\{6, 7, 3/1\}| + 2t_{uv} = 2\Delta + 10$, $(T_1\cup T_3\cup \{5\})\setminus C(u_2)\ne \emptyset$.
     It follows from \ref{135Cu22424} and \ref{13T3T12424} that $2\in C(u_1)$ or $H$ contains a $(2, 4)_{(u_1, u_4)}$-path.
     \begin{itemize}
     \parskip=0pt
     \item
          $4\in C(u_1)$. One sees that $2\in C(u_4)$.
          It follows that $C(u_1) = \{1, 4, 5\}\cup T_3\cup T_4$, $C(u_4) = \{4, 1, 2, 5\}\cup T_3\cup C_{34}$, $c_{34} = t_1 = 1$.
          If $H$ contains no $(2, 6)_{(u_1, u_4)}$-path,
          then $(u u_4, u u_1)\to (6, T_1)$ reduces the proof to (2.1.2.).
          Otherwise, $H$ contains a $(2, 6)_{(u_1, u_4)}$-path.
          Thus, by \ref{136B1B3}, we are done.
     \item
          $2\in C(u_1)$.
          It follows from Corollary~\ref{auv2}(5.1) that $1\in C(u_2)$.
          Since $|\{1, 5, 2\}| + |\{2, 1\}| + |\{3, 5, 4, 2\}| + |\{4, 1, 5\}| + |\{6, 7, 4\}| + 2t_{uv} + t_0= 2\Delta + 11 + t_0\ge 2\Delta + 11 + t_1\ge 2\Delta + 11 + 1 = 2\Delta + 12$, $T_3\setminus C(u_2)\ne \emptyset$.
          If $H$ contains no $(2, i)_{(u_2, u_4)}$-path for some $i\in \{3, 6, 7\}\setminus C(u_4)$,
          then $(u u_4, u u_3)\to (i, T_3\cap T_2)$ reduces the proof to (2.2.2.).
          Otherwise, $2\in C(u_4)$.
          It follows that $W_4 = \{4, 1, 5, 2\}$, $C_{34} = T_1$, $c_{34} = t_1 = 1$, $d(u_i) = \Delta$, $i = 1, 3, 4$, $d(u_2)\le 6$, and $H$ contains a $(2, \{3, 6, 7\})_{(u_2, u_4)}$-path, $3, 6, 7\in C(u_2)$.
          Thus, $d(u_2)\ge |[1, 4]\cup [6, 7]\cup T_1|\ge 7$, a contradiction.
     \end{itemize}
\end{itemize}

{\bf Case 2.} $H$ contains a $(6/7, T_3)_{(u_3, v)}$-path, and assume w.l.o.g. $H$ contains a $(7, 8)_{(u_3, v_7)}$-path.

When $w_3 = 6$, since $w_3 + w_4\ge 6 +3 = 9$,
$t_3 + t_4 = \Delta - 2$, $W_4 = \{4, 1, b_1\}$, $C(u_1) = \{1, a_1\}\cup T_3\cup T_4$, $d(u_i) = \Delta$, $i = 1, 3, 4$, and $d(u_2)\le 6$.
It follows from \ref{315234} and \ref{31rho8234} that $a_1 = 2$, $4\in C(u_2)$, and $5, 7\in C(u_2)$.
Since $2, 4, 5, 7\in C(u_2)$, $T_3\setminus C(u_2)\ne \emptyset$.
Then $(u u_4, u u_3, T_4)\overset{\divideontimes}\to (3, T_3\cap T_2, T_4)$.
In the other case, $W_3 =\{3, 5, 4, 7, 2/1\}$, and $H$ contains a $(7, T_3)_{(u_3, v_7)}$-path.

When $w_1 = 3$, it follows that $c_{34} = t_1 = 1$, $W_4 = \{4, 1, b_1\}$,
$d(u_i) = \Delta$, $i = 1, 3, 4$, and $d(u_2)\le 6$.
If $H$ contains no $(2, j)_{(u_2, u_3)}$-path for some $j\in T_3$,
then $(u u_4, u u_3, u v)\to (\{3, 5, 7\}\setminus (C(u_1)\cup C(u_4)), j, T_4)$.
If $H$ contains no $(2, j)_{(u_2, u_1)}$-path for some $j\in T_1$,
then $(u u_4, u u_1)\to ([5, 7]\setminus (C(u_1)\cup C(u_4)), j)$ reduces the proof to (2.1.).
Otherwise, $2\in C(u_1)\cap C(u_3)$ and $H$ contains a $(2, T_1)_{(u_2, u_1)}$-path, $(2, T_3)_{(u_2, u_3)}$-path.
If $4\in C(u_1)$, one sees that $H$ contains a $(4, \gamma_0)_{(u_4, v)}$-path,
it follows from \ref{315234} and \ref{31rho8234} that $5, 7\in C(u_2)$,
and thus $d(u_2)\ge |\{2, 5, 7\}\cup T_3\cup T_1|\ge 7$.
Otherwise, $4\in C(u_2)\setminus C(u_1)$.
It follows from Corollary~\ref{auv2}(5.1) that $1\in C(u_2)$,
and thus $d(u_2)\ge |\{2, 4, 1\}\cup T_3\cup T_1|\ge 7$.
Hence, we assume $W_1 = \{1, 2/4\}$.
We distinguish whether or not $2\in C(u_1)$:
\begin{itemize}
\parskip=0pt
\item
     $W_1 = \{1, 4\}$.
     One sees that $|[5, 7]\cap C(u_4)\le 2|$ and $H$ contains a $(4, \gamma_0)_{(u_4, v)}$-path.
     Thus, by \ref{315234} and \ref{31rho8234}, we are done.
\item
     $W_3 = \{1, 2\}$.
     Then $4\in C(u_2)$ and it follows from \ref{315234} and \ref{31rho8234} that $H$ contains a $(2, i)_{(u_1, u_2)}$-path for each $i\in \{5, 7\}\subseteq C(u_2)$.
     One sees that $|\{1, 2\}| + |\{2, 4, 5, 7\}| + |\{3, 5, 4, 7, 2\}| + 2t_{u v} = 2\Delta + 7$.
     Since $\sum_{x\in N(u)\setminus \{v\}}d(x)\ge 2\Delta + 7 + |\{4, 1, b_1\}| + |\{1/3\}| = 2\Delta + 11$,
     $T_3\setminus C(u_2)\ne \emptyset$.
     It follows Corollary~\ref{auv2}(5.1) that $3, 6\in S_u$.
     If there exists an $i\in \{5, 7\}\setminus C(u_4)$,
     then $(u u_4, u u_3, u v)\overset{\divideontimes}\to (i, T_3\cap T_2, T_4)$.
     Otherwise, $W_4 = \{4, 1, 5, 7\}$, $d(u_i) = \Delta$, $i = 1, 3, 4$.
     Then $\sum_{x\in N(u)\setminus \{v\}}d(x)\ge  2\Delta + 7 + |\{4, 1, 5, 7\}| + |\{3, 6\}| = 2\Delta + 13$,
     $1\not\in C(u_2)$ and $3\not\in C(u_1)\cup C(u_4)$.
     Thus, by \ref{131Cu33Cu1}, we are done.
\end{itemize}

Hence, $T_4\setminus C(u_1)\ne \emptyset$.
It follows from Corollary~\ref{auv2} (1.2) that $2\in C(u_1)$ and $H$ contains a $(2, T_1)_{(u_1, u_2)}$-path,
and from Corollary~\ref{auv2} (5.1) that $1\in C(u_2)$ or $H$ contains a $(1, 3)_{(u_2, u_3)}$-path.
\end{proof}

If $1\not\in C(u_3)$, $3\not\in C(u_1)\cup C(u_2)$, and $H$ contains no $(4, 3)_{(u_2, u_4)}$-path,
then $(u u_1, u u_2, u u_3)\to (T_4\cap T_1, 3, 1)$ reduces the proof to (2.4.1.).

Assume that $1, 2\not\in C(u_3)$.
It follows from \ref{131Cu33567} that $3, 5\in C(u_1)\cup C(u_4)$ and $6, 7\in C(u_1)\cup C(u_4)\cup C(u_3)$.
Let $i\in \{3\}\cup (\{6, 7\}\setminus C(u_3))$ with mult$_{S_u}(i) = 1$.
If $i\in C(u_4)$, one sees from \ref{131Cu33Cu1} that $H$ contains a $(4, i)_{(u_1, u_4)}$-path,
then $(u u_1, u u_2, u u_3)\to (T_4\cap T_1, i, 1)$ reduces the proof to (2.4.1.).
Hence, we proceed with the following proposition:
\begin{proposition}
\label{13367Cu1}
\begin{itemize}
\parskip=0pt
\item[{\rm (1)}]
	$1\in C(u_3)$, or $3\in C(u_1)\cup C(u_2)$, or $H$ contains a $(4, 3)_{(u_2, u_4)}$-path.
\item[{\rm (2)}]
	If $1, 2\not\in C(u_3)$, then for each $i\in \{3\}\cup (\{6, 7\}\setminus C(u_3))$ with mult$_{S_u}(i) = 1$,
$i\in C(u_1)$.
\end{itemize}
\end{proposition}

\begin{lemma}
\label{13T3Cu22B1B3}
\begin{itemize}
\parskip=0pt
\item[{\rm (1)}]
	$T_3\setminus C(u_2)\ne \emptyset$.
\item[{\rm (2)}]
	$2\in B_1\cup B_3$, and if $2\not\in C(u_3)$, then $H$ contains a $(2, 4)_{(u_3, u_4)}$-path.
\item[{\rm (3)}]
	$3\in C(u_4)$, or $H$ contains a $(3, i)_{(u_4, u_i)}$-path for some $i\in \{1, 2\}$.
\end{itemize}
\end{lemma}
\begin{proof}
Assume that $T_3\subseteq C(u_2)$.
By Lemma~\ref{13T42Cu1}, $T_4\setminus C(u_2)\ne\emptyset$, i.e., $T_1\setminus T_0\ne \emptyset$.
It follows from Corollary~\ref{auv2}(5.1) that $1\in C(u_2)\cup C(u_3)$, $6, 7\in S_u$,
and if $1\not\in C(u_2)$, then $H$ contains a $(1, 3)_{(u_2, u_3)}$-path,
from Corollary~\ref{auv2d(u)5}(7.1) that if $2\not\in C(u_3)$, then $3\in S_u$.
One sees that $|\{1, 2\}| + |\{2\}| + |\{4, 3, 5\}| + |\{4, 1\}| + |\{4, 1, 6, 7\}| + 2t_{uv} = 2\Delta + 8$.

{\bf Case 1.} There exists a $j\in T_3$ such that $H$ contains no $(i, j)_{(u_3, v_i)}$-path for each $i\in [6, 7]$.

It follows from Corollary~\ref{auv2}(3) that $5\in B_1$ and $H$ contains a $(5, i)_{(u_i, u_)}$-path for some $i\in \{2, 4\}\cap C(u_3)$,
from \ref{135Cu22424} that $5\in C(u_2)$ or $2\in C(u_3)\cup C(u_4)$.
\begin{itemize}
\parskip=0pt
\item
     $2\not\in C(u_3)$.
     It follows from Corollary~\ref{auv2}(5.1) that $3\in S_u$.
     Then $\sum_{x\in N(u)\setminus \{v\}}d(x)\ge 2\Delta + 8 + |\{5, 5, 5/2, 3\}| + t_0 \ge 2\Delta + 12 + 1 = 2\Delta + 13$.
     It follows that $t_0 = t_3 = 1$, $W_3 = \{3, 5, 4\}$, $1\in C(u_2)$, mult$_{S_u}(i) = 1$, $i = 3, 6, 7$,
     and from \ref{13367Cu1}(2) that $3, 6, 7\in C(u_1)$.
     Thus, $d(u_1) + d(u_2) \ge |\{1, 2, 5, 3, 6, 7\}| + |\{2, 1\}| + |\{4\}| + t_{u v} + t_0\Delta + 8$, a contradiction.
\item
     $2\in C(u_3)$.
     Then $\sum_{x\in N(u)\setminus \{v\}}d(x)\ge 2\Delta + 8 + |\{5, 5, 2\}| + t_0 \ge 2\Delta + 11 + 2 = 2\Delta + 13$.
     It follows that $t_0 = t_3 = 2$, $1\not\in C(u_3)$ and $3\not\in S_u$, $2\not\in C(u_4)$.
     Thus, by \ref{13367Cu1}(1), we are done.
\end{itemize}

{\bf Case 2.} $H$ contains a $(6/7, T_3)_{(u_3, v)}$-path, and assume w.l.o.g. $H$ contains a $(7, 8)_{(u_3, v_7)}$-path.
\begin{itemize}
\parskip=0pt
\item
     $1, 2\not\in C(u_3)$.
     It follows from \ref{131Cu33567} that $3, 5\in C(u_1)\cup C(u_4)$, $6\in C(u_1)\cup C(u_4)\cup C(u_3)$,
     from \ref{31rho84B32T3}(2) that $7\in C(u_1)\cup C(u_4)$.
     Then $\sum_{x\in N(u)\setminus \{v\}}d(x)\ge 2\Delta + 8 + |\{3, 5, 7\}| + t_0 \ge 2\Delta + 11 + 2 = 2\Delta + 13$.
     It follows that $t_0 = t_3 = 2$, $W_3 = \{3, 5, 4, 7\}$, $1\in C(u_2)$,
     and mult$_{S_u}(i) = 1$, $i = 3, 6$, mult$_{S_u}(j) = 2$, $i = 5, 7$.
     It follows from \ref{13367Cu1}(2) that $3, 6\in C(u_1)$.
     One sees that $d(u_1) + d(u_2)$ $\ge |\{1, 2, 3, 6\}| + |\{2, 1\}| + |\{4\}| + t_{uv} + t_0\ge 7 + \Delta - 2 + t_0 = \Delta + 7$.
     It follows that $5, 7\in C(u_4)$, and from \ref{131Cu33Cu1} that $H$ contains a $(4, \{5, 7\})_{(u_1, u_4)}$-path.
     Thus, $H$ contains a $(4, i)_{(u_4, v_i)}$-path for some $i\in \{5, 7\}$, a contradiction.
\item
     $1/2\in C(u_3)$.
     When $3\not\in C(u_1)$, one sees from \ref{315234} and \ref{31rho8234} that $5, 7\in C(u_1)\cup C(u_2)\cup C(u_4)$,
     $\sum_{x\in N(u)\setminus \{v\}}d(x)\ge 2\Delta + 8 + |\{5, 7\}| + t_0 \ge 2\Delta + 10 + 3 = 2\Delta + 13$.
     It follows that $t_0 = t_3 = 3$, $3\not\in S_u$,
     and then from Corollary~\ref{auv2d(u)5}(7.1) that $2\in C(u_3)$.
     Then $1\not\in C(u_3)$, and thus, by \ref{13367Cu1}(1), we are done.
     Now we assume $3\in C(u_1)$.
     Since $1/3\in C(u_1)$, $1/2\in C(u_3)$,
     $\sum_{x\in N(u)\setminus \{v\}}d(x)\ge |\{1, 2, 3\}| + |\{2, 1/3\}| + |\{4\}| + |\{3, 5, 4, 7, 1/2\}| + |\{4, 1\}| +  |\{6\}| + t_{uv} + t_0\ge 2\Delta + 10 + t_0\ge 2\Delta + 13$,
     and $d(u_1) + d(u_2)\ge |\{1, 2, 3\}| + |\{2, 1/3\}| + |\{4\}| + t_{uv} + t_0\ge \Delta + 4 + t_0\ge \Delta + 7$.
     It follows that $t_0 = t_3 = 3$, $d(u_i) = \Delta = 7$, $i\in [1, 3]$, and $C(u_4) = \{1, 4, 6\}\cup T_3$.
     One sees that $H$ contains a $(4, 6)_{(u_4, v_6)}$-path.
     Thus, by \ref{136B1B3}, we are done.
\end{itemize}

Hence, $T_3\setminus C(u_2)\ne \emptyset$.
It follows from Corollary~\ref{auv2} (5.2) and (4) that $2\in B_1\cup B_3$,
and if $2\not\in C(u_3)$, then $H$ contains a $(2, 4)_{(u_3, u_4)}$-path,
and from Corollary~\ref{auv2} (5.1) that $3\in C(u_4)$, or $H$ contains a $(3, i)_{(u_4, u_i)}$-path for some $i\in \{1, 2\}$.
\end{proof}

Hereafter, we assume w.l.o.g. $8\in T_3\cap T_2$.
If $H$ contains a $(7, 8)_{(u_3, v_\rho)}$-path,
$7\not\in C(u_4)$ and $H$ contains no $(7, i)_{(u_4, u_i)}$-path for each $i\in [1, 2]$,
then $(u u_4, u u_3, u v)\overset{\divideontimes}\to (7, T_3\cap T_2, T_4)$.
If $5\not\in C(u_4)$ and $H$ contains no $(5, i)_{(u_4, u_i)}$-path for each $i\in [1, 2]$,
then $(u u_4, u u_3, u v)\overset{\divideontimes}\to (5, T_3\cap T_2, T_4)$.
If $j\in [6, 7]\setminus C(u_3)\cup C(u_4)$ and $H$ contains no $(j, i)_{(u_4, u_i)}$-path for each $i\in [1, 2]$,
then $(u u_4, u u_3, u v)\overset{\divideontimes}\to (j, T_3\cap T_2, T_4)$.
In the other case, together with \ref{131Cu33Cu1}, we proceed with the following proposition:
\begin{proposition}
\label{13567Cu45Cu2}
\begin{itemize}
\parskip=0pt
\item[{\rm (1)}]
	$5\in C(u_4)$, or $H$ contains a $(5, i)_{(u_4, u_i)}$-path for some $i\in [1, 2]$.
\item[{\rm (2)}]
	For each $j\in [6, 7]$, $i\in C(u_3)\cup C(u_4)$, or $H$ contains a $(j, i)_{(u_4, u_i)}$-path for some $i\in [1, 2]$.
\item[{\rm (3)}]
	If $H$ contains a $(7, 8)_{(u_3, v_j)}$-path,
    then $7\in C(u_4)$ or $H$ contains a $(7, i)_{(u_4, u_i)}$-path for some $i\in [1, 2]$.
\item[{\rm (4)}]
    Assume $1, 2\not\in C(u_3)$.
    If mult$_{S_u\setminus C(u_3)}(5) = 1$, then $5\in C(u_1)\cup C(u_4)$;
    and if $H$ contains a $(7, 8)_{(u_3, v_j)}$-path and mult$_{S_u\setminus C(u_3)}(7) = 1$,
    then $7\in C(u_1)\cup C(u_4)$.
\end{itemize}
\end{proposition}

One sees that $2\in C(u_1)\cap (C(u_3)\cup C(u_4))$, $4\in C(u_3)\cap (C(u_1)\cup C(u_2))$, $1\in C(u_2)\cup C(u_3)$,
and $3, 6, 7\in S_u$.
It follows from \ref{13567Cu45Cu2}(1) that $5\in S_u\setminus C(u_3)$.
Then $|\{1, 2\}| + |\{2\}| + |\{3, 4, 5\}| + |\{4, 1\}| + |\{2, 4, 1, 3, 5, 6, 7\}| + 2t_{u v} = 2\Delta + 11$.
Since either $5\in B_1$, or $H$ contains a $(7, 8)_{(u_3, v_7)}$-path and it follows from \ref{13567Cu45Cu2}(3) that $7\in S_u\setminus C(u_3)$,
$\sum_{x\in N(u)\setminus \{v\}}d(x)\ge 2\Delta + 11 + |\{5/7\}| + t_0 \ge 2\Delta + 12 + c_{34}\ge 2\Delta + 12$.
Hence, $c_{34}\le 1$ and $\Delta \le 7$.

Since $4\in C(v_3)$, it follows from Corollary~\ref{auv2d(u)5}(6) that $R_3\ne \emptyset$ and then $5/6/7\in C(v_3)$.
Particularly, we proceed with the following proposition:
\begin{proposition}
\label{13R3neeptyset}
\begin{itemize}
\parskip=0pt
\item[{\rm (1)}]
	If $T_{uv} = T_3\cup T_4$, then $T_3\setminus C(v_3) \ne \emptyset$.
\item[{\rm (2)}]
	If $1/2\in C(v_3)$, then $r_3\ge 2$ and $T_3\setminus C(v_3)\ne \emptyset$.
\item[{\rm (3)}]
    For some $j\in [6, 7]\setminus (C(u_1)\cup C(u_3))$,
    if $uu_1\to 6$ and $(u u_1, u u_3)\to (3, 6)$ reduce to an acyclic edge $(\Delta + 5)$-coloring for $H$, then we are done by Corollary~\ref{auv2d(u)5}(2.1).
\end{itemize}
\end{proposition}

When $6, 7\in C(u_3)$, since $2\in C(u_3)\cup C(u_4)$,
it follows that $W_3\cup W_4 = [3, 7]\cup \{1, 4, b_1\}\cup \{2\}$, $T_{uv} = T_3\cup T_4$,
and mult$_{C(u_3)\cup C(u_4}(2) = 1$.
We distinguish whether or not $2\in C(u_3)$:
\begin{itemize}
\parskip=0pt
\item
     $2\in C(u_3)$.
     It follows from \ref{13T3Cu22B1B3}(3) that $3\in C(u_1)$,
     and from \ref{13567Cu45Cu2}(1) that $b_1 = 5$, i.e., $C(u_1)= [1, 3]\cup T_3$, $C(u_4) = \{1, 4, 5\}\cup T_3$. One sees that $5\not\in B_1$.
     Then $H$ contains a $(7, 8)_{(u_3, v_7)}$-path.
     Thus, by \ref{13567Cu45Cu2}(3) we are done.
\item
     $2\in C(u_4)$.
     One sees from \ref{131Cu33Cu1} that $6, 7\in S_u\setminus C(u_3)$.
     Since $\sum_{x\in N(u)\setminus \{v\}}d(x)\ge 2\Delta + 11 + |\{6, 7\}| + t_0 = 2\Delta + 13 + t_0$,
     it follows that mult$_{S_u\setminus C(u_3)}(i) = 1$, $i\in [5, 7]$, mult$_{S_u}(3) = 1$, $t_0 = 0$,
     and $H$ contains a $(7, 8)_{(u_3, v_7)}$-path.
     Further, by \ref{13367Cu1}(2), $3\in C(u_1)$, and by \ref{13567Cu45Cu2}(4), $5, 7\in C(u_1)\cup C(u_4)$.
     Then $W_1 = [1, 3]\cup \{\alpha\}$ and $W_4 = \{1, 2, 4\}\cup \{\beta\}$, where $\{\alpha, \beta\} = \{5, 7\}$. Thus, by \ref{131Cu33Cu1}, we are done.
\end{itemize}

In the other case, $\{6, 7\}\setminus C(u_3)\ne \emptyset$ and assume w.l.o.g. $6\not\in C(u_3)$.

\begin{lemma}
\label{1367124}
If $T_3\cup T_4 = T_{uv}$, then
\begin{itemize}
\parskip=0pt
\item[{\rm (1)}]
	for each $j\in [6, 7]\not\in C(u_3)$,
    $H$ contains a $(j, i)_{(u_3, u_i)}$-path for some $i\in \{1, 2, 4\}$, and mult$_{S_u}(j)\ge 2$.
\item[{\rm (2)}]
	$7\in C(u_3)$, $T_3\cap C(u_2) = \emptyset$, and mult$_{S_u}(i) = 2$, $i = 1, 2, 6$,
    mult$_{S_u}(5) + $mult$_{S_u}(7) = 4$.
\end{itemize}
\end{lemma}
\begin{proof}
For (1), assume that $H$ contains no $(6, i)_{(u_3, u_i)}$-path for each $i\in \{1, 2, 4\}$.
One sees from \ref{13R3neeptyset}(1) that $T_3\setminus C(v_3)\ne \emptyset$,
and it follows from Lemma~\ref{auvge2}(1.2) that there exists a $b_3\in [6, 7]\cap C(v_3)$ such that $H$ contains a $(b_3, j)_{(v_3, v_{b_6})}$-path for some $j\in T_3\cap R_3$.
One sees that the same argument applies if $uu_3\to 6$,
and thus $\{4\}\cup T_4\subseteq C(v_6)$, $T_3\setminus C(v_6)\ne \emptyset$,
and it follows from Lemma~\ref{auvge2}(1.2) that there exists a $b_6\in \{3, 7\}\cap C(v_6)$ such that $H$ contains a $(b_6, j)_{(v_6, v_{b_6})}$-path for some $j\in T_3\cap R_6$.
One sees from Lemma~\ref{auvge2}(1.2) and Lemma~\ref{lemma06} that $T_3\subseteq C(v_3)\cup C(v_6)$.
\begin{itemize}
\parskip=0pt
\item
     $2\not\in C(u_4)$.
     One sees that $2\in B_3$ and thus $2\in C(v_3)\cap C(v_6)$.
     Since $d(v_3) + d(v_6)\ge t_{u v} + t_4 + |\{3, 4, 6/7, 2\}| + |\{6, 4, 3/7, 2\}| = \Delta + 6 + t_4\ge \Delta + 6 + 1 = \Delta + 7$,
     it follows that $t_4 = 1$, $W_4 = \{1, 4, b_1\}$, and $1\not\in C(v_3)$.
     Thus, by Lemma~\ref{13T3Cu22B1B3} (3), we are done.
\item
     $2\in C(u_4)$.
     One sees from Corollary~\ref{auv2d(u)5} (6) and Lemma~\ref{lemma06} that $2\in C(v_3)\cup C(v_6)$ or $5\in C(v_3)\cup C(v_6)$.
     Since $d(v_3) + d(v_6)\ge t_{u v} + t_4 + |\{3, 4, 6/7\}| + |\{6, 4, 3/7\}| + |\{2/5\}| = \Delta + 5 + t_4\ge \Delta + 5 + 2 = \Delta + 7$,
     it follows that $t_4 = 2$, $W_4 = \{1, 4, b_1, 2\}$, and $1\not\in C(v_3)\cup C(v_6)$.
     Recall that $H$ contains a $(b_1, 4)_{(u_4, v_{b_1})}$-path.
     Hence, $b_1\in \{5, 7\}$ and it follows that $3, 6\in C(u_2)$ and $H$ contains a $(2, i)_{(u_2, u_4)}$-path for each $i\in \{3, 6\}$.
     One sees that $|\{3, 6\}\cap C(u_1)|\le 1$.
     Then $(u u_1, u v)\to (6, T_3\cap R_6)$ if $6\not\in C(u_1)$,
     and $(u u_1, u u_3, uv)\to (3, 6, T_3\cap R_3)$ if $3\not\in C(u_1)$.
\end{itemize}

Now assume that mult$_{S_u}(j) = 1$.
If $H$ contains a $(j, i)_{(u_3, u_i)}$-path for some $i\in [1, 2]$,
then by \ref{13567Cu45Cu2}, we are done.
Otherwise, $j\in C(u_4)$ and $H$ contains a $(4, 6)_{(u_3, u_4)}$-path.
Then $(u u_1, u v)\overset{\divideontimes}\to (6, j)$ if $H$ contains no $(6, j)_{(u_1, v_6)}$-path for some $j\in T_3$, or else, $(u u_1, u u_2, u v)\overset{\divideontimes}\to (6, T_4\cap T_1, T_3)$.

For (2), since mult$_{S_u}(6)\ge 2$, $\sum_{x\in N(u)\setminus \{v\}}d(x)\ge 2\Delta + 11 + |\{5/7\}| + |\{6\}| + t_0 = 2\Delta + 13 + t_0$.
One sees that if $7\not\in C(u_3)$, then $5\in C(u_1)\cap C(u_4)$, mult$_{S_u}(7)\ge 2$,
and $\sum_{x\in N(u)\setminus \{v\}}d(x)\ge 2\Delta + 11 + |\{5, 7, 6\}| = 2\Delta + 14$.
Hence, $7\in C(u_3)$, $t_0 = 0$, and (2) holds.
\end{proof}

(2.4.2.1.) There exists a $j\in T_3$ such that $H$ contains no $(i, j)_{(u_3, v_i)}$-path for each $i\in [6, 7]$.
Then $5\in B_1\subseteq C(u_1)\cap C(u_4)$ and $H$ contains a $(5, i)_{(u_i, u_)}$-path for some $i\in \{2, 4\}\cap C(u_3)$.

When $c_{34} = 1$, it follows that $\Delta = 7$, mult$_{S_u}(i) = 1$, $i = 3, 6, 7$,
and mult$_{S_u}(j) = 2$, $j = 1, 2, 4$.
Since $2\in C(u_3)\cup C(u_4)$ and $w_3 + w_4\le 8$,
$\{6, 7\}\setminus (C(u_3)\cup C(u_4))\ne \emptyset$ and $6\not\in C(u_3)\cup C(u_4)$.
First, assume that $H$ contains a $(2, 6)_{(u_2, u_4)}$-path.
Then $2\in C(u_4)\setminus C(u_3)$.
One sees that $H$ contains no $(i, 6)_{(u_i, u_2)}$-path for each $i\in \{1, 3\}$.
If follows from \ref{131Cu33Cu1} that $1\in C(u_3)$,
from Lemma~\ref{13T42Cu1}(3) that $3\in C(u_2)$,
and then from Lemma~\ref{13T3Cu22B1B3}(3) that $H$ contains a $(2, 3)_{(u_2, u_4)}$-path.
Thus, by \ref{13R3neeptyset}(3), we are done.
Next, by \ref{13567Cu45Cu2}(2), $6\in C(u_1)$,
and for each $i\in [6, 7]\setminus (C(u_3)\cup C(u_4))$,
assume that $i\in C(u_1)$ and $H$ contains a $(1, i)_{(u_1, u_4)}$-path.
One sees that $H$ contains no $(6, i)_{(u_i, u_3)}$-path for each $i\in \{1, 2, 4\}$,
and the same argument applies if $uu_3\to 6$.
We further assume that $3\not\in C(u_2)$ and it follows from Lemma~\ref{13T42Cu1}(3) that $1\in C(u_2)$.
\begin{itemize}
\parskip=0pt
\item
     $3\in C(u_4)$.
     It follows from \ref{13367Cu1}(1) that $4\in C(u_2)$ and $H$ contains a $(4, 3)_{(u_2, u_4)}$-path,
     and then from \ref{131Cu33Cu1} that $2\in C(u_3)$.
     Hence, $C(u_3) = [2, 5]\cup T_4\cup C_{34}$, $W_4 = \{1\}\cup [3, 5]$, and $C(u_1) = [1, 2]\cup [5, 7]\cup T_3$.
     One sees clearly that $T_4\cup T_{34}\subseteq B_3\subseteq C(v_3)$ and $T_3\setminus C(v_3) \ne \emptyset$.
     Then $(u u_1, u u_3, u v)\overset{\divideontimes}\to (3, 6, T_3\cap R_3)$.
\item
     $3\in C(u_1)$ and $H$ contains a $(1, 3)_{(u_1, u_4)}$-path.
     If $2\in C(u_3)$, then $\{2\}\cup T_4\cup C_{34}\subseteq B_3$ with $t_4 + c_{34} = 3$,
     thus $d(v_3)\ge |\{3, 1, 4, 2, 5/6/7\}\cup T_4\cup C_{34}|\ge 8$.
     Otherwise, $2\in C(u_4)\setminus C(u_3)$.
     One sees that $1, 2\not\in C(u_3)$ and from \ref{13367Cu1}(2) that $7\in C(u_1)$.
     Then $T_4\cup C_{34}\subseteq B_3$ with $t_4 + c_{34} = 4$,
     thus $d(v_3)\ge |\{3, 1, 4, 5/6/7\}\cup T_4\cup C_{34}|\ge 8$.
\end{itemize}

In the other case, $t_3 + t_4 = \Delta - 2$ and by \ref{1367124}(2) that $T_3\cap C(u_2) = \emptyset$, mult$_{S_u}(i) = 2$, $i = 1, 2, 6$, mult$_{S_u}(5) = 3$ and mult$_{S_u}(7) = 1$.
If $H$ contains no $(3, 7)_{(u_2, u_3)}$-path, then $(u u_2, u u_1, u v)\to (7, T_4\cap T_1, T_3)$.
Otherwise, $H$ contains a $(3, 7)_{(u_2, u_3)}$-path and $3\in C(u_2)$.
It follows from Lemma~\ref{13T3Cu22B1B3}(2) that $H$ contains a $(2, 3)_{(u_2, u_3)}$-path, $2\in C(u_4)$,
and from \ref{131Cu33Cu1} that $1\in C(u_3)$, i.e., $W_3 = \{1, 7\}\cup [3, 5]$, $W_4 = \{1, 2, 4, 5\}$.
Then $(u u_1, u u_3, u v)\overset{\divideontimes}\to (6, 3, T_3\cap R_3)$.

(2.4.2.2.) $H$ contains a $(7, T_3)_{(u_3, v)}$-path.
When $c_{34} = 1$, it follows that $W_3\cup W_4 = \{3, 5, 4, 7\}\cup \{1, 4, b_1\}\cup \{2\}$.
If $2\in C(u_3)$, one sees from Lemma~\ref{13T3Cu22B1B3}(3) that $3\in C(u_1)$, and $[5, 7]\setminus (C(u_1)\cup C(u_4))\ne \emptyset$,
then by \ref{13567Cu45Cu2} (1--3), we are done.
Otherwise, $2\in C(u_4)$.
One sees that $1, 2\not\in C(u_3)$.
It follows from \ref{13367Cu1}(2) that $3, 5, 6\in C(u_1)$,
and $W_1 = \{1, 2, 3, 5, 6\}$, $b_1 = 7$.
Then $(u u_1, u u_3, u v)\overset{\divideontimes}\to (7, 1, 8)$.

In the other case, $t_3 + t_4 = \Delta - 2$ and by \ref{1367124}(2) that $T_3\cap C(u_2) = \emptyset$, mult$_{S_u}(i) = 2$, $i = 1, 2, 6$, mult$_{S_u}(5) = 2$ and mult$_{S_u}(7) = 2$.
One sees clearly that $6\in C(u_1)\cup C(u_4)$ and $H$ contains a $(6, T_3\cap R_3)_{(v_3, v_6)}$-path.
If $6\not\in C(u_1)$ and $H$ contains no $(6, i)_{(u_1, u_i)}$-path for each $i\in \{2, 4\}$,
then $(u u_1, u v)\to (6, T_3\cap R_3)$.
Otherwise, $6\in C(u_1)$ or $H$ contains a $(6, i)_{(u_1, u_i)}$-path for some $i\in \{2, 4\}$.
\begin{itemize}
\parskip=0pt
\item
     There exists a $j\in \{5, 7\}\cap C(u_2)$.
     One sees from \ref{13567Cu45Cu2} (1) and (3) that $H$ contains a $(2, j)_{(u_2, u_4)}$-path, $2\in C(u_4)$.
     It follows from \ref{315234} or \ref{31rho8234} that $3\in C(u_1)$,
     and from \ref{13T42Cu1}(3) that $1\in C(u_2)$.
     One sees that $1, 2\not\in C(u_3)$ and $H$ contains no $(j, i)_{(u_1, u_i)}$-path for each $i\in \{2, 4\}$.
     Thus, by \ref{131Cu33Cu1}, we are done.
\item
     $5, 7\in C(u_1)\cup C(u_4)$.
     If $H$ contains a $(2, 3)_{(u_2, u_4)}$-path, i.e., $3\in C(u_2)$ and $2\in C(u_4)$,
     it follows from \ref{131Cu33Cu1} that $1\in C(u_3)$,
     and from \ref{315234} and \ref{31rho8234} that $W_1 = [1, 2]\cup \{5, 7\}$,
     $W_4 = \{1, 2, 4, 6\}$,
     thus by \ref{13non2jT4} (3), $H$ contains a $(4, 6)_{(u_4, v_6)}$-path,
     and by Lemma~\ref{1367124}(1), we are done.
     Otherwise, $3\in C(u_1)\cup C(u_4)$ and it follows from \ref{131Cu33Cu1} that $1\in C(u_2)$.
     \begin{itemize}
     \parskip=0pt
     \item
         $2\not\in C(u_3)$.
         It follows from \ref{13367Cu1}(2) that $3\in C(u_1)$.
         One sees from \ref{131Cu33Cu1} that for each $j\in \{5, 7\}$, $j\not\in C(u_4)$,
         or $H$ contains a $(4, j)_{(u_1, u_4)}$-path.
         It follows from \ref{13non2jT4} (3) that $H$ contains a $(4, 6)_{(u_4, v_6)}$-path.
         Thus, by Lemma~\ref{1367124}(1), we are done.
     \item
         $2\in C(u_3)$.
         It follows that $C(u_1) = \{1, 2, \alpha_1, \alpha_2\}\cup T_3$, $C(u_4) = \{1, 4, \beta_1, \beta_2\}\cup T_3$, $C(u_2) = \{1, 4, 6\}\cup T_4$,
         where $\{\alpha_1, \alpha_2, \beta_1, \beta_2\} = \{3\}\cup [5, 7]$.
         If $5, 7\in C(u_1)$, one sees that $H$ contains a $(2, 6)_{(u_1, u_2)}$-path,
         from \ref{13non2jT4} (3) that $H$ contains a $(4, 6)_{(u_4, v_6)}$-path.
         then, by Lemma~\ref{1367124}(1), we are done.
         Otherwise, there exists a $j\in \{5, 7\}\setminus C(u_1)$,
         and then $(u u_1, u u_2, u u_3, u v)\overset{\divideontimes}\to (j, T_3, 1, T_4)$.
     \end{itemize}
\end{itemize}

(2.4.3.) $c(vu_4) = 6$.
Let $C(u_1) = W_1\cup T_3\cup C_{13}$ and $C(u_3) = W_3\cup T_1\cup C_{13}$.
If $6\not\in C(u_1)$ and $H$ contains no $(6, i)_{(u_1, u_i)}$-path for each $i\in [2, 3]$,
then $u u_1\to 6$ reduces the proof to (2.4.2.).
If $6\not\in C(u_3)$ and $H$ contains no $(6, i)_{(u_3, u_i)}$-path for each $i\in [1, 2]$,
then $u u_3\to 6$ reduces the proof to (2.4.1.).
Otherwise, we proceed with the following proposition:
\begin{proposition}
\label{136236126inCu1Cu3}
\begin{itemize}
\parskip=0pt
\item[{\rm (1)}]
	$6/2/3\in C(u_1)$,
    and if $6\not\in C(u_1)$, then $H$ contains a $(6, i)_{(u_1, u_i)}$-path for some $i\in [2, 3]$.
\item[{\rm (2)}]
	$6/1/2\in C(u_3)$,
    and if $6\not\in C(u_3)$, then $H$ contains a $(6, i)_{(u_3, u_i)}$-path for some $i\in [1, 2]$.
\item[{\rm (3)}]
	$6\in C(u_1)\cup C(u_3)$, $T_3\ne \emptyset$ and there exists a $a_3\in \{2, 4\}\cap C(u_3)$.
\end{itemize}
\end{proposition}

Assume $5\not\in C(u_4)$ and $H$ contains no $(5, i)_{(u_4, u_i)}$-path for each $i\in [1, 2]$.
If $H$ contains no $(2, j)_{(u_2, u_3)}$-path for some $j\in T_3$,
then $(u u_4, u u_3)\to (5, j)$ reduces the proof to (2.4.1.).
Further, assume $H$ contains no $(5, 3)_{(u_4, u_i)}$-path.
If $4\not\in B_1\cup B_3$, then $(u u_4, u v)\overset{\divideontimes}\to (5, 4)$;
if $4\not\in C(u_3)$ and $H$ contains no $(2, 4)_{(u_2, u_3)}$-path,
then $(u u_4, u u_3)\to (5, 4)$ reduces the proof to (2.4.1.);
if $4\not\in C(u_1)$ and $H$ contains no $(2, 4)_{(u_2, u_1)}$-path,
then $(u u_4, u u_1)\to (5, 4)$ reduces the proof to (2.1.1.).
Otherwise, we proceed with the following proposition:
\begin{proposition}
\label{135Cu42T34B1B3}
Assume $5\not\in C(u_4)$ and $H$ contains no $(5, i)_{(u_4, u_i)}$-path for each $i\in [1, 2]$.
\begin{itemize}
\parskip=0pt
\item[{\rm (1)}]
	$H$ contains a $(2, T_3)_{(u_2, u_3)}$-path and $2\in C(u_3)$, $T_3\subseteq C(u_2)$.
\item[{\rm (2)}]
	If $H$ contains no $(5, 3)_{(u_4, u_3)}$-path, then $4\in B_1\cup B_3$,
    and for each $i\in \{1, 3\}$, $2/4\in C(u_i)$, and if $4\not\in C(u_i)$,
    then $H$ contains a $(2, 4)_{(u_2, u_i)}$-path.
\end{itemize}
\end{proposition}

Assume that $5\not\in C(u_2)$ and $H$ contains no $(5, i)_{(u_2, u_i)}$-path for each $i\in \{1, 3, 4\}$.
Let $u u_2\to 5$.
If $2\not\in B_1\cup B_3$, then $uv \overset{\divideontimes}\to 2$.
If $2\not\in C(u_1)$ and $H$ contains no $(2, 4)_{(u_1, u_2)}$-path,
then $(u u_1, u v)\overset{\divideontimes}\to (2, T_3)$.
If $H$ contains no $(4, j)_{(u_3, u_4)}$-path for some $j\in T_3\cup (\{2\}\setminus C(u_3))$ and there exists a $l\in T_1\setminus (C(u_2)\cap C(u_3))$,
then $(u u_3, u v)\overset{\divideontimes}\to (j, l)$.
Otherwise, we assume the following proposition:
\begin{proposition}
\label{135notCu2134}
Assume that $5\not\in C(u_2)$ and $H$ contains no $(5, i)_{(u_2, u_i)}$-path for each $i\in \{1, 3, 4\}$.
Then $2\in B_1\cup B_3$, and
\begin{itemize}
\parskip=0pt
\item[{\rm (1)}]
	$2\in C(u_1)$ or $H$ contains a $(2, 4)_{(u_1, u_2)}$-path;
\item[{\rm (2)}]
    if $H$ contains no $(j, 4)_{(u_3, u_4)}$-path for some $j\in T_3\cup (\{2\}\setminus C(u_3))$,
    then $T_1\subseteq C(u_2)\cap C(u_3)$. 	
\end{itemize}
\end{proposition}

If $5\not\in C(u_1)\cup C(u_2)$, $H$ contains no $(5, i)_{(u_2, u_i)}$-path for each $i\in [3, 4]$,
and $H$ contains no $(4, j)_{(u_1, u_4)}$-path for some $j\in T_1$,
then $(u u_2, u u_1, u v)\overset{\divideontimes}\to (5, j, T_3)$.
If $5\not\in C(u_4)\cup C(u_1)$, $H$ contains no $(5, i)_{(u_4, u_i)}$-path for each $i\in [2, 3]$,
and $H$ contains no $(2, j)_{(u_1, u_2)}$-path for some $j\in T_1$,
then $(u u_4, u u_1, u v)\overset{\divideontimes}\to (5, j, T_3)$.
Otherwise, we proceed with the following proposition:
\begin{proposition}
\label{135notCu1}
\begin{itemize}
\parskip=0pt
\item[{\rm (1)}]
	If $5\not\in C(u_2)\cup C(u_1)$ and $H$ contains no $(5, i)_{(u_2, u_i)}$-path for each $i\in [3, 4]$,
    then $4\in C(u_1)$, $H$ contains a $(4, T_1)_{(u_1, u_4)}$-path, and $T_1\subseteq C(u_4)$.
\item[{\rm (2)}]
	If $5\not\in C(u_4)\cup C(u_1)$ and $H$ contains no $(5, i)_{(u_4, u_i)}$-path for each $i\in [2, 3]$,
    then $2\in C(u_1)$, $H$ contains a $(2, T_1)_{(u_1, u_2)}$-path, and $T_1\subseteq C(u_2)$.
\end{itemize}
\end{proposition}

Assume that $5\not\in S_u\setminus C(u_3)$.
One sees from \ref{315234} that $3\in C(u_1)$ and $H$ contains a $(3, 5)_{(u_1, u_3)}$-path.
It follows from \ref{135Cu42T34B1B3}(2) that $4\in B_1\cup B_3$, $2/4\in C(u_1)$.
Since $1, 3, 2/4 \in C(u_1)$, $T_1\ne \emptyset$.
Together with \ref{135notCu2134}, \ref{135notCu1}(1--2), and \ref{135Cu42T34B1B3}(1),
\begin{proposition}
\label{135notSu}
If $5\not\in S_u\setminus C(u_3)$, then $T_1\ne \emptyset$, $2, 4\in (B_1\cup B_3)\cap C(u_1)$,
$H$ contains a $(4, T_1)_{(u_1, u_4)}$-path, a $(2, T_1)_{(u_1, u_2)}$-path, $T_1\subseteq C(u_2)\cap C(u_4)$,
and $2\in C(u_3)$, $H$ contains a $(2, T_3)_{(u_3, u_4)}$-path, $T_3\subseteq C(u_2). $
\end{proposition}

Assume that $6\not\in C(u_1)$ and $H$ contains no $(6, i)_{(u_1, u_i)}$-path for each $i\in \{2, 4\}$.
Let $u u_1\to 6$.
If $1\not\in C(u_3)$ and $H$ contains no $(1, i)_{(u_3, u_i)}$-path for each $i\in \{2, 4\}$,
then $u u_3\to 1$;
if $7\not\in C(u_3)\cup C(u_1)$ and $H$ contains no $(7, i)_{(u_3, u_i)}$-path for each $i\in \{2, 4\}$,
then $u u_3\to 7$,
these reduce the proof to (2.4.2.).
Assume that $6\not\in C(u_3)$ and $H$ contains no $(6, i)_{(u_3, u_i)}$-path for each $i\in \{2, 4\}$.
Let $u u_3\to 6$.
If $3\not\in C(u_1)$ and $H$ contains no $(1, i)_{(u_1, u_i)}$-path for each $i\in \{2, 4\}$,
then $u u_1\to 3$;
if $7\not\in C(u_3)\cup C(u_1)$ and $H$ contains no $(7, i)_{(u_1, u_i)}$-path for each $i\in \{2, 4\}$,
then $u u_1\to 7$;
these reduces the proof to (2.4.1.).
Otherwise, we proceed with the following proposition:
\begin{proposition}
\label{136notCu16notCu3}
\begin{itemize}
\parskip=0pt
\item[{\rm (1)}]
	If $6\not\in C(u_1)$, $H$ contains no $(6, i)_{(u_1, u_i)}$-path for each $i\in \{2, 4\}$, then
    \begin{itemize}
    \parskip=0pt
    \item[{\rm (1.1)}]
	      $1\in C(u_3)$ or $H$ contains a $(1, i)_{(u_3, u_i)}$-path for some $i\in \{2, 4\}$;
    \item[{\rm (1.2)}]
	      $7\in C(u_3)\cup C(u_1)$ or $H$ contains a $(7, i)_{(u_3, u_i)}$-path for some $i\in \{2, 4\}$.
     \end{itemize}
\item[{\rm (2)}]
	If $6\not\in C(u_3)$, $H$ contains no $(6, i)_{(u_3, u_i)}$-path for each $i\in \{2, 4\}$, then
    \begin{itemize}
    \parskip=0pt
    \item[{\rm (2.1)}]
	      $3\in C(u_1)$ or $H$ contains a $(1, i)_{(u_1, u_i)}$-path for some $i\in \{2, 4\}$;
    \item[{\rm (2.2)}]
	      $7\in C(u_3)\cup C(u_1)$ or $H$ contains a $(7, i)_{(u_1, u_i)}$-path for some $i\in \{2, 4\}$.
     \end{itemize}
\end{itemize}
\end{proposition}

Assume mult$_{S_u}(6) = 1$.
One sees from \ref{136236126inCu1Cu3}(1) that if $6\in C(u_3)$,
then $3\in C(u_1)$ and $H$ contains a $(3, 6)_{(u_1, u_3)}$-path,
or else, $6\in C(u_1)$, and from \ref{136236126inCu1Cu3}(2) that $1\in C(u_3)$ and $H$ contains a $(1, 6)_{(u_1, u_3)}$-path.
If $T_4 = \emptyset$, one sees that $C(u_4) = \{4, 6\}\cup T_{u v}$,
then $(u u_4, v u_4, u v)\overset{\divideontimes}\to (5, 4, 6)$.
Otherwise, $T_4\ne \emptyset$ and it follows from Corollary~\ref{auv2}(4) that for each $j\in T_4$,
$H$ contains a $(j, 1/3/5/7)_{(u_4, v)}$-path.
One sees from \ref{136notCu16notCu3}(1.2) or (2.2) that $7\in S_u$.
Let $u u_2\to 6$.
If $2\not\in B_1\cup B_3$, then $uv\overset{\divideontimes}\to 2$.
If $2\not\in C(u_1)$ and $H$ contains no $(2, 4)_{(u_2, u_4)}$-path,
then $u u_1\to 2$ reduces the proof to (2.4.2.).
If $H$ contains no $(4, j)_{(u_3, u_4)}$-path for some $j\in T_3\cup (\{2\}\setminus C(u_3))$ and there exists a $l\in T_1\setminus (C(u_2)\cap C(u_4))$,
then $(u u_3, u v)\overset{\divideontimes}\to (j, l)$.
Assume that $6\in C(u_3)$. One sees from $3, 2/4, 1\in C(u_1)$ that $T_1\ne \emptyset$.
If $H$ contains no $(4, j)_{(u_1, u_4)}$-path for some $j\in T_1$,
then $u u_1\to j$ reduces the proof to (2.4.2.).
Otherwise, we proceed with the following proposition:
\begin{proposition}
\label{136Cu34T1}
Assume mult$_{S_u}(6) = 1$. Then
\begin{itemize}
\parskip=0pt
\item[{\rm (1)}]
	$7\in S_u$, $2\in B_1\cup B_3$,
    and $2\in C(u_1)$ or $H$ contains a $(2, 4)_{(u_1, u_4)}$-path;
\item[{\rm (2)}]
	if $H$ contains no $(j, 4)_{(u_3, u_4)}$-path for some $j\in T_3\cup (\{2\}\setminus C(u_3))$,
    then $T_1\subseteq C(u_2)\cap C(u_4)$;
\item[{\rm (3)}]
	if $6\in C(u_3)$, then $\emptyset\ne T_1\subseteq C(u_4)$, $4\in C(u_1)$,
    and $H$ contains a $(4, T_1)_{(u_1, u_4)}$-path.
\item[{\rm (4)}]
	$T_4\ne \emptyset$ and for each $j\in T_4$, $H$ contains a $(j, 1/3/5/7)_{(u_4, v)}$-path.
\end{itemize}
\end{proposition}

If $4\not\in C(v_1)$ and $H$ contains no $(4, i)_{(v_1, v_i)}$-path for each $i\in \{3, 5, 6, 7\}$,
then $v v_1\to 4$ reduces the proof to (2.1.2.).
Otherwise, we proceed with the following proposition:
\begin{proposition}
\label{1343567V1}
\begin{itemize}
\parskip=0pt
\item[{\rm (1)}]
    If $2\in C(v_1)$, then $R_1\ne\emptyset$; and if $4\in C(u_1)$, then $T_1\ne \emptyset$.
\item[{\rm (2)}]
    $4/3/5/6/7\in C(v_1)$, and if $4\not\in C(v_1)$, then $H$ contains a $(4, 3/5/6/7)_{(v_1, v)}$-path.	
\end{itemize}
\end{proposition}

Assume that $4\not\in C(u_3)$ (or $T_3\cap T_4\ne \emptyset$) and $3\not\in C(u_2)$.
Let $u u_2\to 3$, $u u_3\to T_3$ (or $u u_3\to T_3\cap T_4$).
If $H$ contains no $(3, i)_{(u_2, u_i)}$-path for each $i\in \{1, 4\}$,
then $u v\to 2$ or $uv\to R_1$.
If $H$ contains no $(3, 4)_{(u_2, u_4)}$-path, $2\not\in C(u_1)$,
and $H$ contains no $(2, 4)_{(u_1, u_4)}$-path,
then $u u_1\to 2$ reduces to an acyclic edge $(\Delta + 5)$-coloring for $H$ such that $C(u)\cap C(v) = \{3\}$.
Assume that $2\not\in C(u_3)$ (or $T_3\cap T_2\ne \emptyset$) and $3\not\in C(u_4)$.
Let $u u_4\to 4$, $u u_3\to T_3$ (or $u u_3\to T_3\cap T_2$).
If $H$ contains no $(3, i)_{(u_4, u_i)}$-path for each $i\in \{1, 2\}$,
then $u v\to 4$ or $u v\to T_1$.
If $H$ contains no $(2, 4)_{(u_2, u_4)}$-path, $4\not\in C(u_1)$, $H$ contains no $(2, 4)_{(u_1, u_2)}$-path,
then $u u_1\to 4$ reduces to an acyclic edge $(\Delta + 5)$-coloring for $H$ such that $C(u)\cap C(v) = \{3\}$.
Otherwise, we proceed with the following proposition:
\begin{proposition}
\label{132or4Cu2}
\begin{itemize}
\parskip=0pt
\item[{\rm (1)}]
	 If $4\not\in C(u_3)$ (or $T_3\cap T_4\ne \emptyset$) and $3\not\in C(u_2)$, then
     \begin{itemize}
     \parskip=0pt
     \item[{\rm (1.1)}]
	      $H$ contains a $(3, i)_{(u_2, u_i)}$-path for some $i\in \{1, 4\}$,
     \item[{\rm (1.2)}]
	      if $H$ contains no $(3, 4)_{(u_2, u_4)}$-path,
          then $2\in C(u_1)$ or $H$ contains a $(2, 4)_{(u_2, u_4)}$-path.
    \end{itemize}
\item[{\rm (2)}]
	 If $2\not\in C(u_3)$ (or $T_3\cap T_2\ne \emptyset$) and $3\not\in C(u_4)$, then
     \begin{itemize}
     \parskip=0pt
     \item[{\rm (2.1)}]
	      $H$ contains a $(3, i)_{(u_4, u_i)}$-path for some $i\in \{1, 2\}$;
     \item[{\rm (2.2)}]
	      if $H$ contains no $(3, 2)_{(u_2, u_4)}$-path,
          then $4\in C(u_1)$ or $H$ contains a $(2, 4)_{(u_1, u_2)}$-path.
    \end{itemize}
\end{itemize}
\end{proposition}

If $4\not\in C(u_3)$, $H$ contains no $(4, i)_{(u_3, v_i)}$-path for each $i\in \{1, 7\}$,
and $H$ contains no $(5, j)_{(u_3, u_j)}$-path for each $j\in [1, 2]$,
then $(v u_3, u u_3)\to (4, 5)$ reduces the proof to (2.4.1.).
Otherwise, we proceed with the following proposition:
\begin{proposition}
\label{13Cu3417512}
If $4\not\in C(u_3)$ and $H$ contains no $(4, i)_{(u_3, v_i)}$-path for each $i\in \{1, 7\}$,
then $H$ contains a $(5, j)_{(u_3, u_j)}$-path for some $j\in [1, 2]$.
\end{proposition}

Assume that recoloring some edges incident to $u$ reduces to an acyclic edge $(\Delta + 5)$-coloring $c'$ for $H$ such that $C(u)\cap C(v) = \{1, 3\}$, where $(u u_a, u u_b)_{c'} = (1, 3)$, and $a\ne 1$, $b\ne 3$.
One sees from Corollary~\ref{auv2d(u)5} (2.1) that $T_{u v}\subseteq C(v_1)\cap C(v_3)$.
If $(T_3\cup T_1)\setminus C(u_2)\ne \emptyset$ and $2\not\in C'(u)$,
one sees from \ref{1343567V1} that $2\not\in C(v_1)$ and from Corollary~\ref{auv2} (3) that $2\in B_3$,
then on the basic coloring of $c'$, $u v\overset{\divideontimes}\to 2$.
If $(T_3\cup T_1)\setminus C(u_4)\ne \emptyset$ and $4\not\in C'(u)$,
one sees from Corollary~\ref{auv2} (3) that $4\in B_1\cup B_3$ and if $4\not\in C(v_i)$ for some $i\in \{1, 3\}$,
then on the basic coloring of $c'$, $u v\overset{\divideontimes}\to 4$.
thus $C(v_i) = \{i, 4\}\cup T_{u v}$, $i = 1, 3$ and by Corollary~\ref{auv2d(u)5} (2.2), we are done.
Hence, we proceed with the following propositions:
\begin{proposition}
\label{1331recoloring}
Assume that recoloring some edges incident to $u$ reduces to an acyclic edge $(\Delta + 5)$-coloring $c'$ for $H$ such that $C(u)\cap C(v) = \{1, 3\}$, where $(u u_a, u u_b)_{c'} = (1, 3)$, and $a\ne 1$, $b\ne 3$.
We are done if $(T_3\cup T_1)\setminus C(u_2)\ne \emptyset$ and $2\not\in C'(u)$,
or $(T_3\cup T_1)\setminus C(u_4)\ne \emptyset$ and $4\not\in C'(u)$.
\end{proposition}

Assume that $(u u_1, u u_3)\to (3, 1)$ reduces to an acyclic edge $(\Delta + 5)$-coloring for $H$.
One sees from Corollary~\ref{auv2d(u)5}(2.1) and Corollary~\ref{auv2}(5.2)(4) that
if $(T_3\cup T_1)\setminus C(u_i)\ne \emptyset$ for some $i\in \{2, 4\}$,
then $(u u_1, u u_3, u u_i)\to (3, 1, (T_3\cup T_1)\setminus C(u_i))$ reduces to an acyclic edge $(\Delta + 5)$-coloring $c'$ for $H$ such that $C(u)\cap C(v) = \{1, 3\}$,
where $(u u_3, u u_1)_{c'} = (1, 3)$ and $i\not\in C'(u)$.
Thus, it follows from \ref{1331recoloring} that we may proceed with the following propositions:
\begin{proposition}
\label{1331switch}
If $(u u_1, u u_3)\to (3, 1)$ reduces to an acyclic edge $(\Delta + 5)$-coloring for $H$,
then $T_3\cup T_1\subseteq C(u_2)\cap C(u_4)$,
i.e., if $(T_3\cup T_1)\setminus C(u_i)\ne \emptyset$ for some $i\in \{2, 4\}$,
then $1/3\in S_u$, and $1\in C(u_3)$ or $H$ contains a $(1, 2/4)_{(u_3, u)}$-path,
or $3\in C(u_1)$ or $H$ contains a $(3, 2/4)_{(u_1, u)}$-path.
\end{proposition}


Assume that $1\not\in C(u_4)$ and $H$ contains no $(1, i)_{(u_4, u_i)}$-path for each $i\in [2, 3]$.
Let $(u u_4, u v)\to (1, T_3)$.
If $H$ contains no $(2, j)_{(u_1, u_2)}$-path for some $j\in T_1$,
then $u u_1\overset{\divideontimes} \to j$;
if $4\not\in C(u_1)$ and $H$ contains no $(4, i)_{(u_1, u_i)}$-path for each $i\in [2, 3]$,
then $u u_1\overset{\divideontimes} \to 4$;
if $5\not\in C(u_1)$ and $H$ contains no $(5, i)_{(u_1, u_i)}$-path for each $i\in [2, 3]$,
then $u u_1\overset{\divideontimes} \to 5$.
\begin{proposition}
\label{131Cu4232T1}
Assume $1\not\in C(u_4)$, $H$ contains no $(1, i)_{(u_4, u_i)}$-path for each $i\in [2, 3]$. Then
\begin{itemize}
\parskip=0pt
\item[{\rm (1)}]
	$H$ contains a $(2, T_1)_{(u_1, u_2)}$-path, $2\in C(u_1)$, $T_1\subseteq C(u_2)$;
\item[{\rm (2)}]
	For each $j\in \{4, 5\}$, either $j\in C(u_1)$ or $H$ contains a $(j, i)_{(u_1, u_i)}$-path for some $i\in [2, 3]$.
\end{itemize}
\end{proposition}

Assume that $1\not\in C(u_2)$ and $H$ contains no $(1, i)_{(u_2, u_i)}$-path for each $i\in [3, 4]$.
Let $(u u_2, u v)\to (1, T_3)$.
If $H$ contains no $(4, j)_{(u_1, u_4)}$-path for some $j\in T_1$,
then $u u_1\overset{\divideontimes} \to j$;
if $2\not\in C(u_1)$ and $H$ contains no $(2, i)_{(u_1, u_i)}$-path for each $i\in [3, 4]$,
then $u u_1\overset{\divideontimes} \to 2$;
if $5\not\in C(u_1)$ and $H$ contains no $(5, i)_{(u_1, u_i)}$-path for each $i\in [3, 4]$,
then $u u_1\overset{\divideontimes} \to 5$.
Hence, we proceed with the following proposition:
\begin{proposition}
\label{131Cu2344T1}
If $1\not\in C(u_2)$ and $H$ contains no $(1, i)_{(u_2, u_i)}$-path for each $i\in [3, 4]$,
then
\begin{itemize}
\parskip=0pt
\item[{\rm (1)}]
	$H$ contains a $(4, T_1)_{(u_1, u_4)}$-path, $4\in C(u_1)$, $T_1\subseteq C(u_4)$;
\item[{\rm (2)}]
	for each $j\in \{2, 5\}$, either $j\in C(u_1)$ or $H$ contains a $(j, 3/4)_{(u_1, u)}$-path.
\end{itemize}
\end{proposition}

If $1\not\in C(u_3)$, $H$ contains no $(1, i)_{(u_3, u_i)}$-path for each $i\in \{2, 4\}$,
and $5\not\in C(u_1)$, $H$ contains no $(5, i)_{(u_1, u_i)}$-path for each $i\in \{2, 4\}$,
then $(u u_1, u u_3, u v)\overset{\divideontimes}\to (5, 1, T_3)$.
Otherwise, together with \ref{136notCu16notCu3}(1.1), 
we proceed with the following proposition:
\begin{proposition}
\label{131notCu3356Cu1}
If $1\not\in C(u_3)$ and $H$ contains no $(1, i)_{(u_3, u_i)}$-path for each $i\in \{2, 4\}$,
then for each $j\in [5, 6]$, $j\in C(u_1)$, or $H$ contains a $(j, i)_{(u_1, u_i)}$-path for some $i\in \{2, 4\}$.
\end{proposition}

Now we introduce several lemmas.
\begin{lemma}
\label{13t3ge2}
 $w_3\ge 4$ and $t_3\ge 2$.
\end{lemma}
\begin{proof}
One sees from \ref{136236126inCu1Cu3}(3) that $T_3\ne \emptyset$ and there exists a $a_3\in \{2, 4\}\cap C(u_3)$,
and from \ref{136236126inCu1Cu3}(2) that $6/1/2\in C(u_3)$.
Hence, it suffices to assume that $W_3 = \{3, 5, 2\}$.
Then $H$ contains a $(2, T_3\cup \{6\})_{(u_2, u_3)}$-path, $T_3\cup \{6\}\subseteq C(u_2)$,
from \ref{136236126inCu1Cu3}(3) that $6\in C(u_1)$,
and from Corollary~\ref{auv2}(5.1) that $3, 7\in S_u$.
One sees from Corollary~\ref{auv2}(4) that $5\in B_1$ and $H$ contains a $(2, 5)_{(u_2, u_3)}$-path.
Since $5, 6\in C(u_1)$, $T_1\ne\emptyset$ and there exists a $a_1\in \{2/4\}\cap C(u_1)$.
It follows that $t_1\ge 2$.
If $T_1\setminus T_0\ne \emptyset$,
one sees from Corollary~\ref{auv2}(5.2) and (3) that if there exists a $i\in \{2, 4\}\setminus C(u_1)$,
then mult$_{S_u}(i)\ge 2$,
and from Corollary~\ref{auv2}(5.1) that $1\in S_u$,
then mult$_{S_u}(1) + $ mult$_{S_u\setminus C(u_1)}(2) + $ mult$_{S_u\setminus \{a_1\}}(4)\ge 2$.
One sees that $|\{1, a_1, 5, 6\}| + |\{2, 5, 6\}| + |\{2, 3, 5\}| + |\{4, 6\}| + |\{3, 7\}| + 2t_{uv} = 2\Delta + 10$.

{\bf Case 1.} $(T_3\cup \{5\})\subseteq C(u_4)$.
Since either $t_0\ge t_1\ge 2$,
or mult$_{S_u}(1) + $ mult$_{S_u\setminus C(u_1)}(2) + $ mult$_{S_u\setminus \{a_1\}}(4)\ge 2$,
$\sum_{x\in N(u)\setminus \{v\}}d(x)\ge 2\Delta + 10 + |T_3\cup \{5\}| + 2\ge 2\Delta + 14$.

{\bf Case 2.} $(T_3\cup \{5\})\setminus C(u_4)\ne \emptyset$.
It follows from Corollary~\ref{auv2}(3) and \ref{135Cu42T34B1B3}(2) that $4\in B_1$ and $H$ contains a $(2, 4)_{(u_2, u_3)}$-path.
One sees from  Corollary~\ref{auv2d(u)5}(7.1) that $2\in C(u_1)$ or $1\in S_u$.
Thus, $\sum_{x\in N(u)\setminus \{v\}}d(x)\ge 2\Delta + 10 + |\{4\}| + |\{1/2\}| + t_0 = 2\Delta + 12 + t_0$.
It follows that $T_1\setminus T_0\ne \emptyset$ and then $1\in S_u$.
One sees from Corollary~\ref{auv2}(1.2) that if $T_1\setminus C(u_4)\ne \emptyset$,
then $2\in C(u_1)$,
and from Corollary~\ref{auv2}(5.2) that if $T_1\setminus C(u_2)\ne \emptyset$, then $2\in C(u_1)\cup C(u_4)$.
Hence, $2\in C(u_1)\cup C(u_4)$ and $\sum_{x\in N(u)\setminus \{v\}}d(x)\ge 2\Delta + 10 + |\{4\}| + |\{1, 2\}| + t_0 = 2\Delta + 13 + t_0$.
It follows that $\Delta = 7$, $t_0 = 0$, mult$_{S_u}(i) = 1$, $i = 1, 3, 7$, mult$_{S_u}(2) = 2$.
\begin{itemize}
\parskip=0pt
\item
     $T_1\setminus C(u_4)\ne \emptyset$.
     Since $2\in C(u_1)$, $T_1\subseteq C(u_2)$ and $t_1 + t_3 + |\{2, 4, 5, 6\}|\le d(u_2)\le 7$,
     we have $t_1 = 1$, $t_3 = 2$ and $c_{13} = 2$.
     Thus, $d(u_1)\ge |[1, 2]\cup [4, 6]\cup T_3\cup C_{13}|\ge 5 + 1 + 2 = 8$.
\item
     $T_1\setminus C(u_2)\ne \emptyset$.
     Since $4\in C(v_1)$, it follows from Corollary~\ref{auv2d(u)5}(6) that $R_1\ne \emptyset$.
     One sees from Corollary~\ref{auv2}(5.1) that $1\in C(u_4)$ or $H$ contains a $(1, 2)_{(u_2, u_4)}$-path.
     It follows from \ref{1331switch} that $3\in C(u_1)$ or $H$ contains a $(3, i)_{(u_1, u_i)}$-path for some $i\in \{2, 4\}$.
     One sees from Corollary~\ref{auv2}(5.1) that $3\in C(u_2)$ or $H$ contains a $(3, i)_{(u_2, u_i)}$-path for some $i\in \{1, 4\}$.
     It follows that $3\in C(u_2)$, $2\in C(u_1)$, $1\in C(u_4)$,
     or $3\in C(u_1)$, $1\in C(u_2)$, $2\in C(u_4)$.
     Then $(u u_1, u u_2, u u_3, u u_4, u v)\overset{\divideontimes}\to ([1, 3]\setminus C(u_1), [1, 3]\setminus C(u_2), T_3, [1, 2]\setminus C(u_4), R_1)$.
\end{itemize}
Hence, $w_3\ge 4$ and then $t_3\ge 2$.
\end{proof}

\begin{lemma}
\label{13T1neemptyset}
$w_1\ge 3$, $T_1\ne \emptyset$, there exists a $a_1\in \{2, 4\}\cap C(u_1)$, and $t_1\ge 1$, $t_1 + t_3\ge 3$.
\end{lemma}
\begin{proof}
Otherwise, assume that $w_1\le 2$.
One sees from \ref{315234} and \ref{136236126inCu1Cu3}(1) that $W_1 = \{1, i\}$,
$i\in [2, 3]$ and $H$ contains a $(i, [5, 6])_{(u_1, u_i)}$-path, $5, 6\in C(u_i)$.

{\bf Case 1.} $W_1 = \{1, 3\}$.
It follows from \ref{136Cu34T1}(1) that $6\in C(u_2)$,
from \ref{135Cu42T34B1B3}(2) that $5\in C(u_4)$ or $H$ contains a $(2, 5)_{(u_2, u_4)}$-path,
and from \ref{135notCu2134} that $5\in C(u_2)$ or $H$ contains a $(4, 5)_{(u_2, u_4)}$-path.
If $T_3\setminus C(u_2)\ne \emptyset$,
then $u u_2\to \alpha\in T_3\cap T_2$, $u v\to 2$ or $(u u_1, u v)\to (2, T_3\setminus \{\alpha\})$.
If $T_3\setminus C(u_4)\ne \emptyset$,
then $u u_4\to \beta\in T_3\cap T_4$, $u v\to 4$ or $(u u_1, u v)\to (4, T_3\setminus \{\beta\})$.
Otherwise, and $T_3\subseteq C(u_2)\cap C(u_4)$.

One sees from \ref{131notCu3356Cu1} that $1\in C(u_3)$ or $H$ contains a $(1, 2/4)_{(u_3, u)}$-path.
It follows that $1\in S_u$ and likewise $7\in S_u$.
One sees from \ref{132or4Cu2}(1.2) that if $4\not\in C(u_3)$,
then $3\in C(u_2)$ or $H$ contains a $(3, 4)_{(u_2, u_4)}$-path,
and from  \ref{132or4Cu2}(2.2) that if $2\not\in C(u_3)$,
then $3\in C(u_4)$ or $H$ contains a $(2, 3)_{(u_2, u_4)}$-path.
It follows that $2, 4\in C(u_3)$ or $3\in C(u_2)\cup C(u_4)$.
One sees from $\sum_{x\in N(u)\setminus \{v\}}d(x)\ge |\{1, 3\}| + |\{2, 6, 4/5\}| + |\{4, 6, 2/5\}| + |\{3, 5, 6, 2/4\}| + |\{1, 7, 3/\{2, 4\} \setminus \{a_3\}\}| + 2t_{u v}+ t_0 = 2\Delta + 11 + t_0\ge 2\Delta + 13$ that
$t_3 = 2$, $W_3 = \{3, 5, 6, a_3\}$, $1, 7\in C(u_{a_3})$ and $3\in C(u_{a_3})$,
or $H$ contains a $(3, \{2, 4\}\setminus \{a_3\})_{u_{a_3}, u}$-path.
Hence, we have the following two scenarios:
\begin{itemize}
\parskip=0pt
\item
    $W_3 = \{3, 5, 2, 6\}$ and $W_3 = \{2, 6, 1, 7, 4\}$, $W_4 = \{4, 6, 3, 5\}$.
    Then $(u u_1, u u_4, u v)\overset{\divideontimes}\to (4, 1, T_3)$.
\item
    $W_3 = \{3, 5, 4, 6\}$ and $W_3 = \{2, 6, 3, 5\}$, $W_4 = \{4, 6, 1, 7, 2\}$.
    Then $(u u_1, u u_2, u v)\overset{\divideontimes}\to (2, 1, T_3)$.
\end{itemize}

{\bf Case 2.} $W_1 = \{1, 2\}$.
It follows from \ref{136Cu34T1}(1) that $6\in C(u_3)$.
One sees from \ref{132or4Cu2}(1.1) that if $4\not\in C(u_3)$,
then $3\in C(u_2)$ or $H$ contains a $(3, 4)_{(u_2, u_4)}$-path,
and from  \ref{132or4Cu2}(2.1) that if $2\not\in C(u_3)$,
then $3\in C(u_4)$ or $H$ contains a $(2, 3)_{(u_2, u_4)}$-path.
It follows that $2, 4\in C(u_3)$ or $3\in S_u$.
One sees that $H$ contains a $(2, 5)_{(u_1, u_2)}$-path.
It follows from \ref{13Cu3417512} that $4\in C(u_3)$ or $H$ contains a $(4, 1/7)_{(u_3, v)}$-path.
One sees that the same argument applies if $7\not\in C(u_2)$ and $u u_1\to 7$.
It follows from \ref{131Cu2344T1} (2) that
we may proceed with the following proposition:
\begin{proposition}
\label{132134734}
For each $i\in \{1, 7\}$, $i\in C(u_2)$ or $H$ contains a $(i, 3/4)_{(u_2, u)}$-path.
\end{proposition}

One sees that $|\{1, 2\}| + |\{2, 5, 6\}| + |\{3, 5, 6, a_3\}| + |\{4, 6\}| + |\{1, 7\}| + 2t_{u v} = 2\Delta + 9$.
We distinguish the following on whether $T_3\subseteq C(u_2)\cap C(u_4)$ or not.
\begin{itemize}
\parskip=0pt
\item
     $T_3\subseteq C(u_2)\cap C(u_4)$.
     Then $\sum_{x\in N(u)\setminus \{v\}}d(x)\ge 2\Delta + 9 + |\{3/\{2, 4\}\setminus\{a_3\}\}| + t_0\ge 2\Delta + 10 + t_3\ge 2\Delta + 10 + 2 = 2\Delta + 12$.
     It follows that $\Delta\le 7$, $t_3\le 3$, and if $t_3 = 3$, then mult$_{S_u}(i) = 1$, $i = 1, 7$,
     and mult$_{S_u}(3) + $ mult$_{S_u}(4) = 1$.
     When $2\in C(u_3)$, one sees from $1/7/4\in C(u_)$ that $t_3 = 3$.
     Since $T_3\cup \{2, 5, 6\}\sqsubseteq C(u_2)$, it follows from \ref{132134734} that $3/4\in C(u_2)$.
     One sees that $\{1, 7\}\setminus C(u_3)\ne \emptyset$ and if $4\in C(u_2)$,
     then $\{1, 7\}\setminus C(u_4)\ne \emptyset$.
     Thus, by \ref{132134734}, we are done.
     \quad In the other case, $2\not\in C(u_3)$.
     Then $4\in C(u_3)$, $3\in S_u$ and it follows from \ref{135Cu42T34B1B3}(1) that $5\in C(u_4)$.
     $\sum_{x\in N(u)\setminus \{v\}}d(x)\ge 2\Delta + 10 + t_0 + |\{5\}|\ge 2\Delta + 11 + t_0\ge 2\Delta + 13$.
     It follows that $t_3 = 2$, $2\not\in S_u$, and follows from \ref{132or4Cu2} (2.1) that $3\in C(u_4)$,
     from \ref{132134734} that $C(u_2) = [1, 2]\cup [5, 9]$.
     Then $(u u_1, u u_2, u u_3, u v)\overset{\divideontimes}\to (5, 3, 1, 2)$\footnote{Or, by \ref{1331switch}, we are done. }.
\item
     $T_3\setminus C(u_2)\ne \emptyset$.
     Then $4\in C(u_3)$ and $H$ contains a $(4, T_3)_{(u_3, u_4)}$-path.
     It follows from Corollary~\ref{auv2} (5.2) and (3) that $2\in B_1\cup B_3$, and $2\in C(u_3)\cap C(u_4)$,
     from Corollary~\ref{132or4Cu2} (2.1) that $3\in C(u_4)$ or $H$ contains a $(3, 2)_{(u_4, u_2)}$-path,
     and from \ref{135Cu42T34B1B3}(1) that $5\in C(u_4)$.
     Then $\sum_{x\in N(u)\setminus \{v\}}d(x)\ge 2\Delta + 9 + |\{3, 5, 2\}| + t_0 = 2\Delta + 12 + t_0\ge 2\Delta + 12$.
     It follows that $\Delta\le 7$.
     \begin{itemize}
     \parskip=0pt
     \item
          $2\in C(u_3)\cap C(u_4)$.
          Then $W_2 = [1, 3]\cup [5, 7]$, $W_3 = [2, 6]$, $W_3 = \{2, 4, 5, 6\}$.
          It follows from \ref{132or4Cu2}(2.1) that $H$ contains a $(2, 3)_{(u_2, u_4)}$-path,
          and the from \ref{1331switch} that $H$ contains a $(1, 2)_{(u_2, u_3)}$-path.
          Thus, by \ref{131Cu4232T1}, we are done.
     \item
          $2\in C(u_2)\setminus C(u_4)$.
          One sees from \ref{132or4Cu2}(2.1) that $3\in C(u_4)$, $W_3 = [2, 6]$, and $W_4 = [3, 6]$.
          It follows from \ref{131Cu4232T1} that $4\in C(u_2)$, and from \ref{132134734} that $C(u_2) = [1, 2]\cup [4, 7]$.
          Then $(u u_1, u u_2, u u_3, u u_4, u v)\overset{\divideontimes}\to (4, 3, 1, 2, T_3)$.
     \item
          $2\in C(u_4)\setminus C(u_2)$.
          When $3\in C(u_4)$, one sees that $W_3 = [3, 6]$, $W_4 = [2, 6]$,
          and it follows from \ref{1331switch} that $3\in C(u_2)$,
          from \ref{132134734} that $C(u_2) = [1, 3]\cup [5, 7]$,
          then $(u u_1, u u_2, u u_3, u u_4, u v)\overset{\divideontimes}\to (5, 4, 2, 1, T_3)$.
          In the other case, $3\not\in C(u_4)$, and it follows from \ref{132or4Cu2}(2.1) that $H$ contains a $(2, 3)_{(u_2, u_4)}$-path, from \ref{1331switch} that $1\in C(u_3)$, or $H$ contains a $(1, 4)_{(u_3, u_4)}$-path.
          Hence, $7\not\in C(u_3)\cup C(u_4)$.
          If $H$ contains no $(2, 7)_{(u_2, u_4)}$-path, then $(u u_4, u u_3, u v)\to (7, T_3\cap T_2, 4)$.
          Otherwise, $H$ contains a $(2, 7)_{(u_2, u_4)}$-path.
          One sees that $u u_1\to 7$ and $(u u_1, u u_3)\to (3, 7)$ reduces to an acyclic edge $(\Delta + 5)$-coloring for $H$ such that $C(u)\cap C(v) = \{3, 7\}$.
          Thus, by \ref{1331switch}, we are done.
     \end{itemize}
\item
     $T_3\setminus C(u_4)\ne \emptyset$.
     It follows that $2\in C(u_3)$, $T_3\subseteq C(u_2)$, and
     from Corollary~\ref{auv2} (5.2) and (3) that $4\in B_3\cap C(u_2)$. 
     Since $\sum_{x\in N(u)\setminus \{v\}}d(x)\ge 2\Delta + 9 + |\{2, 4, 3\}| + t_0\ge 2\Delta + 12$ and $d(u_2)\ge |\{2\}\cup [4, 6]| + t_3\ge 7$,
     it follows that $d(u_2) = \Delta = 7$, $t_3 = 3$, $W_1 = [2, 6]$, $1, 3\not\in C(u_2)$,
     and from \ref{132134734} that $H$ contains a $(1, 4)_{(u_2, u_4)}$-path.
     Hence, by \ref{1331switch}(2), we are done.
\end{itemize}
Hence, $T_1\ne \emptyset$ and it follows from Corollary~\ref{auv2}(1.2) that there exists a $a_1\in \{2, 4\}\cap C(u_1)$.
\end{proof}

\begin{lemma}
\label{132T34T113onlyone}
If $T_1\setminus C(u_2)\ne \emptyset$, $T_3\setminus C(u_4)\ne \emptyset$,
and mult$_{S_u}(1) = $ mult$_{S_u}(3) = 1$, mult$_{S_u}(2) = $ mult$_{S_u}(4) = 2$,
then we can obtain an acyclic edge $(\Delta + 5)$-coloring for $H$ or reduces to an acyclic edge $(\Delta + 5)$-coloring for $H$ such that $a_{uv} \le 1$.
\end{lemma}
\begin{proof}
One sees that $2\in C(u_3)$, $4\in C(u_1)$,
and $H$ contains a $(2, T_3)_{(u_2, u_3)}$-path, a $(4, T_1)_{(u_1, u_4)}$-path.
It follows from Corollary~\ref{auv2} (5.2) and (3) that $2, 4\in B_1\cup B_3$,
and from Corollary~\ref{auv2} (5.1) that $3\in C(u_2)$, or $H$ contains a $(3, 1/4)_{(u_2, u)}$-path,
and $1\in C(u_4)$, or $H$ contains a $(1, 2/3)_{(u_4, u)}$-path.
We discuss the following three subcases due to $1\in C(u_i)$, $i\in [2, 4]$, respectively.

{\bf Case 1.} $1\in C(u_2)$.
Then $2\in C(u_4)\setminus C(u_1)$, and $H$ contains a $(1, 2)_{(u_2, u_4)}$-path.
One sees from \ref{1331switch} that $3\in C(u_1)$ or $H$ contains a $(3, 2/4)_{(u_1, u)}$-path.
It follows that $3\in C(u_1)$.
Since $(u u_1, u u_2, u u_3, u u_4)\to (2, 3, T_3\cap T_4, 1)$ is an acyclic edge $(\Delta + 5)$-coloring for $H$ with $4\not\in C(u)$,
we are done by \ref{1331recoloring}.

{\bf Case 2.} $1\in C(u_3)$.
Then $3\in C(u_4)$, $H$ contains a $(1, 3)_{(u_3, u_4)}$-path,
and then $4\in C(u_2)$, $H$ contains a $(4, 3)_{(u_2, u_4)}$-path.
Since $(u u_1, u u_2, u u_3, u u_4)\to (T_1\cap T_2, 3, 4, 1)$ is an acyclic edge $(\Delta + 5)$-coloring for $H$ with $2\not\in C(u)$,
we are done by \ref{1331recoloring}.

{\bf Case 3.} $1\in C(u_4)$. Then $3\in C(u_2)\cup C(u_4)$.
When $3\in C(u_4)$, it follows that $4\in C(u_2)\setminus C(u_3)$ and $H$ contains a $(4, 3)_{(u_2, u_4)}$-path.
one sees that $(u u_1, u u_3)\to (3, 1)$ reduces to an acyclic edge $(\Delta + 5)$-coloring for $H$,
then we are done by \ref{1331switch}.
In the other case, $3\in C(u_2)$.
It follows from \ref{1331switch} that $2\in C(u_1)$ or $4\in C(u_3)$.
If $2\in C(u_1)$, then $(u u_1, u u_2, u u_3, u u_4, u v)\overset{\divideontimes} \to (3, 1, T_3, 2, T_1\cap T_2)$.
Otherwise, $4\in C(u_3)$, and  $(u u_1, u u_2, u u_3, u u_4)\to (T_1, 4, 1, 3)$ reduces to (2.1.2.).
\end{proof}

Assume that there exists a $j\in T_3$ such that $H$ contains no $(j, i)_{(u_3, v)}$-path for each $i\in \{6, 7\}$.
One sees from Corollary~\ref{auv2}(4) that $5\in B_1$,
and $H$ contains a $(5, i)_{(u_3, u_i)}$-path for some $i\in \{2, 4\}\cap C(u_3)$.
If $T_3\setminus C(u_4)\ne \emptyset$ and $H$ contains no $(2, 5)_{(u_2, u_3)}$-path,
then $(u u_4, u u_3, v u_3)\to (T_3\cap T_4, 5, j)$ reduces the proof to (2.4.1.) or to an acyclic edge $(\Delta + 5)$-coloring for $H$ such that $|C(u)\cap C(v)| = 1$.
Assume that $T_3\setminus C(u_2)\ne \emptyset$ and $H$ contains no $(4, 5)_{(u_3, u_4)}$-path.
If $|T_3\cap T_2|\ge 2$, then $(u u_2, u u_3, v u_3)\to ((T_3\cap T_2)\setminus \{j\}, 5, j)$;
if $6, 7\not\in C(u_3)$, then $u u_2\to \alpha\in T_3\cap T_2$, $(u u_3, v u_3)\to (5, T_3\setminus \{\alpha\})$;
these reduce to an acyclic edge $(\Delta + 5)$-coloring for $H$ such that $|C(u)\cap C(v)| = 1$.
Hence, we proceed with the following proposition, or otherwise we are done:
\begin{proposition}
\label{1325254545}
Assume that there exists a $j\in T_3$ such that $H$ contains no $(j, i)_{(u_3, v)}$-path for each $i\in \{6, 7\}$.
Then
\begin{itemize}
\parskip=0pt
\item[{\rm (1)}]
	If $T_3\setminus C(u_4)\ne \emptyset$, then $H$ contains a $(2, 5)_{(u_2, u_3)}$-path.
\item[{\rm (2)}]
	If $T_3\setminus C(u_2)\ne \emptyset$, and $|T_3\cap T_2|\ge 2$, or $6, 7\not\in C(u_3)$,
    then $H$ contains a $(4, 5)_{(u_3, u_4)}$-path.
\end{itemize}
\end{proposition}

Assume that $H$ contains a $(7, j)_{(v, v_7)}$-path for some $j\in T_3$.
If $7\not\in C(u_1)$ and $H$ contains no $(7, i)_{(u_1, u_i)}$-path for each $i\in [2, 4]$,
then $(u u_1, u v)\overset{\divideontimes}\to (7, j)$.
Let $(u u_2, u v)\to (7, j)$.
Assume that $7\not\in C(u_2)\cup C(u_1)$ and $H$ contains no $(7, i)_{(u_2, u_i)}$-path for each $i\in [3, 4]$.
If $H$ contains no $(4, l)_{(u_1, u_4)}$-path for some $l\in T_1$,
then $u u_1\overset{\divideontimes} \to l$;
if $2\not\in C(u_1)$ and $H$ contains no $(2, i)_{(u_1, u_i)}$-path for each $i\in [3, 4]$,
then $u u_1\overset{\divideontimes} \to 2$;
if $5\not\in C(u_1)$ and $H$ contains no $(5, i)_{(u_1, u_i)}$-path for each $i\in [3, 4]$,
then $u u_1\overset{\divideontimes} \to 5$.
Hence, we proceed with the following proposition:
\begin{proposition}
\label{137Cu2344T1}
Assume that $H$ contains a $(7, j)_{(v, v_7)}$-path for some $j\in T_3$. Then
\begin{itemize}
\parskip=0pt
\item[{\rm (1)}]
	 $7\in C(u_1)$ or $H$ contains a $(7, i)_{(u_1, u_i)}$-path for some $i\in [2, 4]$.
\item[{\rm (2)}]
     If $7\not\in C(u_2)\cup C(u_1)$, and $H$ contains no $(7, i)_{(u_2, u_i)}$-path for each $i\in [3, 4]$, then
     \begin{itemize}
     \parskip=0pt
     \item[{\rm (2.1)}]
          $H$ contains a $(4, T_1)_{(u_1, u_4)}$-path, $4\in C(u_1)$, $T_1\subseteq C(u_4)$;
     \item[{\rm (2.2)}]
           for each $j\in \{2, 5\}$, either $j\in C(u_1)$ or $H$ contains a $(j, 3/4)_{(u_1, u)}$-path.
     \end{itemize}
\end{itemize}
\end{proposition}

Assume $7\not\in C(u_4)$ and $H$ contains no $(7, i)_{(u_3, u_4)}$-path for each $i\in [2, 3]$.
If $4\not\in C(u_1)$ and $H$ contains no $(4, i)_{(u_1, u_i)}$-path for each $i\in [2, 3]$,
then $(u u_1, u u_4)\to (4, 7)$ reduces the proof to (2.1.2.).
If $7\not\in C(u_1)$ and $H$ contains no $(2, j)_{(u_1, u_2)}$-path for some $j\in T_1$,
then $(u u_4, u u_1)\to (7, j)$ reduces the proof to (2.1.2.).
Hence, we proceed with the following proposition:
\begin{proposition}
\label{137Cu42T1}
Assume $7\not\in C(u_4)$, $H$ contains no $(7, i)_{(u_i, u_4)}$-path for each $i\in [2, 3]$. Then
\begin{itemize}
\parskip=0pt
\item[{\rm (1)}]
	 $4\in C(u_1)$ or $H$ contains a $(4, i)_{(u_1, u_i)}$-path for some $i\in [2, 3]$;
\item[{\rm (2)}]
     if $7\not\in C(u_1)$, then $2\in C(u_1)$, $H$ contains a $(2, T_1)_{(u_1, u_2)}$-path and $T_1\subseteq C(u_2)$.
\end{itemize}
\end{proposition}

\begin{lemma}
\label{13T3inCu4}
$T_3\subseteq C(u_4)$.
\end{lemma}
\begin{proof}
Assume that $T_{uv}\setminus (C(u_3)\cup C(u_4))\ne \emptyset$, i.e., $T_3\cap T_4\ne \emptyset$.
It follows that $2\in C(u_3)$, $H$ contains a $(2, T_3)_{(u_3, u_4)}$-path,
and from Corollary~\ref{auv2}(3) that $4\in B_1\cup B_3$ and mult$_{S_u}(4)\ge 2$.

{\bf Case 1.} There exists a $j\in T_3\cap T_4$ such that $H$ contains no $(7, j)_{(u_3, u_4)}$-path.
It follows from Corollary~\ref{auv2}(4) that $5\in B_1$, and $H$ contains a $(5, 2/4)_{(u_3, u)}$-path.
When $C(v_1) = \{1, 5\}\cup T_{uv}$,
it follows from \ref{1343567V1}(2) that $H$ contains a $(4, 5)_{(u_3, v_1)}$-path.
Then $H$ contains a $(2, 5)_{(u_2, u_3)}$-path.
Let $vv_1\to 2$.
If $H$ contains no $(2, j)_{(u_2, v_1)}$-path, then $uv\overset{\divideontimes}\to j$;
otherwise, $(vu_3, uv)\overset{\divideontimes}\to (j, 5)$.
In the other case, $R_1\ne \emptyset$.
It follows from Corollary~\ref{auv2d(u)5}(3) that $1/3\in S_u$.
One sees from \ref{136236126inCu1Cu3} that $2\in C(u_1)$ or $t_1\ge 2$,
and from Corollary~\ref{auv2d(u)5}(7.1) that $2\in C(u_1)$ or $1\in S_u$.
Further, One sees that if $2\not\in C(u_1)$, then $1\in S_u$ and $t_1\ge 2$,
and if $2\in C(u_1)$, then $t_1\ge 2$ if $4\in C(u_1)$, or else $1\in S_u$.
It follows that mult$_{S_u}(1) + $ mult$_{S_u\setminus C(u_3)}(2) + t_1\ge 3$.

{Case 1.1.} $H$ contains neither a $(3, j)_{(u_4, v)}$-path nor a $(7, j)_{(u_4, v_7)}$-path.
It follows from Corollary~\ref{auv2}(4) that $6\in B_1\cup B_3$, $6\in C(u_3)\cup C(u_4)$,
mult$_{S_u}(6)\ge 2$, and $2/6\in C(u_i)$, $i = 3, 4$.
One sees that $|\{1, 5\}| + |\{2\}| + |\{2, 3, 5\}| + |\{4, 6\}| + |\{4, 4, 6, 6, 5\}| + |\{1/3\}| + 2t_{uv} = 2\Delta + 10$.

(1.1.1.) $T_3\subseteq C(u_4)$. Then $t_0\ge t_3\ge 2$.
If $T_1\subseteq T_0$, one sees that $\sum_{x\in N(u)\setminus \{v\}}d(x)\ge 2\Delta + 10 + t_0\ge 2\Delta + 10 + t_3 + t_1\ge 2\Delta + 13$, then $t_1 = 1$ and $2\not\in C(u_1)$, a contradiction.
Otherwise, $T_1\setminus T_0\ne \emptyset$.
It follows from Corollary~\ref{auv2}(5.1) that $1, 7\in S_u$,
and from Corollary~\ref{auv2d(u)5}(9) that $2\in C(u_1)$ or $2\in C(u_4)$.
Thus, $\sum_{x\in N(u)\setminus \{v\}}d(x)\ge 2\Delta + 10 + |\{7, 2\}| + t_0\ge 2\Delta + 12 + 2 = 2\Delta + 14$.

(1.1.2.) $T_3\setminus C(u_4)\ne\emptyset$.
It follows from Corollary~\ref{auv2}(5.1) that $3, 7\in S_u$,
and from \ref{1325254545}(1) that $5\in C(u_2)$.
Then $\sum_{x\in N(u)\setminus \{v\}}d(x)\ge 2\Delta + 10 + |\{7, 1/2\}| + t_0\ge 2\Delta + 12 + t_0$.
One sees that $2\in C(u_1)$ or $6\in C(u_1)$.
When $T_1\subseteq T_0$, then $t_1 = 1$ and thus $1\in S_u$, $2\not\in C(u_1)$, a contradiction.
In the other case, $T_1\setminus T_0 \ne \emptyset$.

If $T_1\subseteq C(u_2)$, then $2\in C(u_1)$ and thus $d(u_1) + d(u_2) + d(u_3)\ge 2t_{uv} + |\{1, 2, 5\}| + |\{2, 5\}| + |\{3, 2, 5\}| + |\{4, 4, 6, 6\}| = 2\Delta + 8$, a contradiction.
Otherwise, $T_1\setminus C(u_2)\ne \emptyset$.
It follows from Corollary~\ref{auv2}(5.1) that $1\in S_u$,
and from Corollary~\ref{auv2d(u)5}(9) that mult$_{S_u}(2)\ge 2$.
Hence, $\sum_{x\in N(u)\setminus \{v\}}d(x)\ge 2\Delta + 10 + |\{1/3, 7, 2\}| + t_0 = 2\Delta + 13 + t_0$,
and mult$_{S_u}(1) = $ mult$_{S_u}(3) = 1$, mult$_{S_u}(2) = $ mult$_{S_u}(4) = 2$.
Thus, by Lemma \ref{132T34T113onlyone}, we are done.

{Case 1.2.} $3\in C(u_3)$ and $H$ contains a $(3, j)_{(u_4, v)}$-path.
One sees that $|\{1, 5\}| + |\{2\}| + |\{2, 3, 5\}| + |\{4, 6, 3\}| + |\{4, 4, 6, 5\}| + |\{1/2\}| + 2t_{uv} = 2\Delta + 10$.

(1.2.1.) $T_1\subseteq T_0$.
If $T_3\subseteq T_0$, $\sum_{x\in N(u)\setminus \{v\}}d(x)\ge 2\Delta + 10 + t_0\ge 2\Delta + t_1 + t_3\ge 2\Delta + 1 + 2 = 2\Delta + 13$,
one sees that $t_1 = 1$ and thus $2\not\in C(u_1)$, a contradiction.
Otherwise, $T_3\setminus T_0\ne \emptyset$, i.e., $T_3\setminus C(u_4)\ne \emptyset$.
It follows from \ref{1325254545}(1) that $5\in C(u_2)$,
and from Corollary~\ref{auv2}(5.1) that $7\in S_u$.
Since $\sum_{x\in N(u)\setminus \{v\}}d(x)\ge 2\Delta + 9 + |\{7\}| + $mult$_{S_u}(1) + $ mult$_{S_u\setminus C(u_3)}(2) + t_0\ge 2\Delta + 10 + $mult$_{S_u}(1) + $ mult$_{S_u\setminus C(u_3)}(2) + t_1\ge 2\Delta + 13$,
it follows that mult$_{S_u}(3) = $ mult$_{S_u}(6) = 1$,
and mult$_{S_u}(1) + $ mult$_{S_u\setminus C(u_3)}(2) + t_1 = 3$.
Hence, $H$ contains a $(4, 3)_{(u_2, u_4)}$-path.
Then it follows from \ref{136notCu16notCu3}(2.1) that $6\in C(u_3)$,
and then $uu_1\overset{\divideontimes}\to 6$\footnote{Or, by \ref{136236126inCu1Cu3}(1), we are done. }.

(1.2.2.) $T_1\setminus T_0\ne \emptyset$.
It follows from Corollary~\ref{auv2}(5.1) that $1, 7\in S_u$,
and from Corollary~\ref{auv2d(u)5}(9) that $2\in S_u\setminus C(u_3)$.
Then $\sum_{x\in N(u)\setminus \{v\}}d(x)\ge 2\Delta + 10 + |\{7, 1/2\}| +t_0 = 2\Delta + 12 + t_0$.
One sees clearly that $T_3\setminus T_0 \ne \emptyset$,
and it follows from \ref{136Cu34T1} that mult$_{S_u}(6)\ge 2$.
Hence, $\sum_{x\in N(u)\setminus \{v\}}d(x)\ge 2\Delta + 12 + |\{6\}| + t_0\ge 2\Delta + 13$.
It follows that mult$_{S_u}(1) = $ mult$_{S_u}(3) = 1$, mult$_{S_u}(2) = $ mult$_{S_u}(4) = 2$.
Thus, by Lemma \ref{132T34T113onlyone}, we are done.

{Case 1.3.} $7\in C(u_3)$ and $H$ contains a $(7, j)_{(u_4, v)}$-path.
One sees that $|\{1, 5\}| + |\{2\}| + |\{2, 3, 5\}| + |\{4, 6, 7\}| + |\{4, 4, 6, 5\}| + |\{1/3\}| + 2t_{uv} = 2\Delta + 10$.

(1.3.1.) $T_3\subseteq T_0$.
If $T_1\subseteq T_0$, $\sum_{x\in N(u)\setminus \{v\}}d(x)\ge 2\Delta + 10 + t_0\ge 2\Delta + t_1 + t_3\ge 2\Delta + 1 + 2 = 2\Delta + 13$,
one sees that $t_1 = 1$ and thus $2\not\in C(u_1)$, a contradiction.
Otherwise, $T_1\setminus T_0\ne \emptyset$.
It follows from Corollary~\ref{auv2}(5.1) that $1\in S_u$,
and Corollary~\ref{auv2d(u)5}(9) that $2\in S_u\setminus C(u_3)$.
Then $\sum_{x\in N(u)\setminus \{v\}}d(x)\ge 2\Delta + 10 + |\{2\}| + t_0\ge 2\Delta + 13 + t_0\ge 2\Delta + 13 + t_3$.
It follows that $t_3 = 2$, $3\not\in S_u$, and mult$_{S_u}(6) = 1$.
Thus, $W_3 = [2, 5]$ and then $uu_3\overset{\divideontimes}\to 6$\footnote{Or, by \ref{136236126inCu1Cu3}(2), we are done. }.

(1.3.2.) $T_3\setminus T_0\ne \emptyset$, i.e., $T_3\setminus C(u_4)\ne \emptyset$.
It follows from \ref{1325254545}(1) that $5\in C(u_2)$,
and from Corollary~\ref{auv2}(5.1) that $3\in S_u$.
When $7\not\in S_u\setminus C(u_4)$,
it follows from \ref{137Cu2344T1}(1) that $4\in C(u_1)$,
$H$ contains a $(4, 7)_{(u_1, u_4)}$-path, and then $t_1\ge 2$.
If there exists a $\alpha\in T_1\setminus (C(u_3)\cap C(u_4))$,
then $(u u_2, u u_3, u v)\overset{\divideontimes}\to (7, T_3\cap T_3, \alpha)$.
Otherwise, $T_1\subseteq C(u_2)\cap C(u_4)$.
One sees that $\sum_{x\in N(u)\setminus \{v\}}d(x)\ge 2\Delta + 10 + |\{1/2\}| + t_0\ge 2\Delta + 11 + t_1$.
It follows that $t_1 = 2$, mult$_{S_u}(6) = 1$,
and from \ref{136Cu34T1}(1) that mult$_{S_u}(2) \ge 2$.
Hence, $1\not\in S_u$ and then $W_1 = \{1, 2, 4, 5\}$.
Thus, $u u_1\overset{\divideontimes}\to 6$\footnote{Or, by \ref{136236126inCu1Cu3}(1), we are done. }.
In the other case, $7\in S_u\setminus C(u_4)$,
and $\sum_{x\in N(u)\setminus \{v\}}d(x)\ge 2\Delta + 10 + |\{7\}| + $mult$_{S_u}(1) + $ mult$_{S_u\setminus C(u_3)}(2) + t_0 = 2\Delta + 11 + $mult$_{S_u}(1) + $ mult$_{S_u\setminus C(u_3)}(2) + t_0$.
Recall that mult$_{S_u}(1) + $ mult$_{S_u\setminus C(u_3)}(2) + t_1\ge 3$.
It follows that $T_1\setminus T_0\ne \emptyset$ and then $1\in S_u$ and $2\in C(u_1)\cup C(u_4)$.
Hence, $\sum_{x\in N(u)\setminus \{v\}}d(x)\ge 2\Delta + 10 + |\{7, 1, 2\}| = 2\Delta + 13$.
One sees clearly that mult$_{S_u}(1) = $ mult$_{S_u}(3) = 1$, mult$_{S_u}(2) = $ mult$_{S_u}(4) = 2$.
If $T_1\subseteq C(u_2)$, then $2\in C(u_1)$ and thus $d(u_1) + d(u_2) + d(u_3)\ge 2t_{uv} + |\{1, 2, 5\}| + |\{2, 5\}| + |\{3, 2, 5\}| + |\{4, 4, 6, 7\}| = 2\Delta + 8$, a contradiction.
Otherwise, $T_1\setminus C(u_2)\ne \emptyset$.
Thus, by Lemma \ref{132T34T113onlyone}, we are done.

{\bf Case 2.} $H$ contains a $(7, T_3\cap T_4)_{(u_3, u_4)}$-path.
When for some $l\in \{5, 7\}$, $l\not\in S_u\setminus C(u_3)$,
it follows from \ref{315234} or \ref{137Cu2344T1}(1) that $3\in C(u_1)$,
and from \ref{135notSu}, or \ref{137Cu2344T1}(2.1) and \ref{137Cu42T1}(2) that $2, 4\in C(u_1)$ and $T_1\subseteq C(u_2)\cap C(u_4)$.
Recall that $6\in C(u_1)\cup C(u_3)$.
Then $w_1 + w_3 = 9$ and $\Delta = 7$, $t_1 + t_3 = 5$.
One sees that $d(u_1) + d(u_2) + d(u_3)\ge |[1, 4]| + |\{2\}| + |\{2, 7, 3, 5\}| + |\{4, 6\}| + 2t_{uv} = 2\Delta + 7$.
It follows that $4\in C(u_2)\setminus C(u_3)$ and $1\not\in S_u\setminus C(u_4)$.
Thus, $(uu_3, uu_1, uv)\overset{\divideontimes}\to (1, l, j)$\footnote{Or, by \ref{131notCu3356Cu1}, we are done. }.

In the other case, $5, 7\in S_u\setminus C(u_3)$.
One sees from Corollary \ref{auv2}(5.2) that $3\in C(u_4)$ or mult$_{S_u}(6)\ge 2$.
Then $\sum_{x\in N(u)\setminus \{v\}}d(x)\ge |\{1\}| + |\{2\}| + |\{2, 7, 3, 5\}| + |\{4, 6\}| + |\{4, 4, 6, 5, 7, 2/1, 3/6\}| + 2t_{u v} + t_0 = 2\Delta + 11 + t_0$.
If $T_3\subseteq T_0$, one sees that $t_3 = 2$ and $T_1\cap T_0 = \emptyset$,
it follows from Corollary \ref{auv2}(5.1) that $1\in S_u$,
and from Corollary \ref{auv2d(u)5}(9) that $2\in C(u_1)\cup C(u_4)$, a contradiction.
Otherwise, $T_3\setminus T_0 \ne \emptyset$ and from Corollary \ref{auv2}(5.1) that $3\in C(u_2)$,
or $H$ contains a $(3, 1/4)_{(u_2, u)}$-path.

{Case 2.1.} mult$_{S_u}(6)\ge 2$.
Then $\sum_{x\in N(u)\setminus \{v\}}d(x)\ge 2\Delta + 11 + |\{6\}| + t_0 = 2\Delta + 12 + t_0$.

{Case 2.1.1.} $T_1\subseteq T_0$.
It follows that $t_1 = 1$, mult$_{S_u}(i) = 1$, $i = 3, 5$,
and $1\in S_u$ or $2\in S_u\setminus C(u_3)$.
\begin{itemize}
\parskip=0pt
\item
     $2\in S_u\setminus C(u_3)$.
     Then $1\not\in S_u$.
     It follows from Corollary~\ref{auv2d(u)5}(7.1) that $W_1 = \{1, 2, 4\}$.
     One sees that $d(u_1) + d(u_2) + d(u_3)\ge |\{1, 2, 4\}| +  |\{2\}| +  |\{2, 7, 3, 5\}| +  |\{4, 6, 6\}| + 2t_{u v} = 2\Delta + 7$.
     Then $1, 3\not\in S_u\setminus C(u_4)$, $4\in C(u_2)\setminus C(u_3)$,
     and $H$ contains a $(4, 3)_{(u_2, u_4)}$-path.
     Thus, by \ref{1331switch}, we are done.
\item
     $2\not\in S_u\setminus C(u_3)$.
     It follows from \ref{136236126inCu1Cu3}(1) that $W_1 = \{1, 4, 3/6\}$.
     When $W_1 = \{1, 4, 3\}$, it follows from \ref{136236126inCu1Cu3}(1) that $H$ contains a $(3, 6)_{(u_1, u_3)}$-path.
     One sees $3\not\in C(u_4)$ and $H$ contains no $(j, i)_{(u_4, v)}$-path.
     We are done by Corollary~\ref{auv2}(3), we are done.
     In the other case, $W_1 = \{1, 4, 6\}$.
     It follows from \ref{315234} that $5\in C(u_4)$ and $H$ contains a $(4, 5)_{(u_1, u_4)}$-path,
     and then from \ref{135notCu2134}(1) that $3\in C(u_2)$.
     One sees that $d(u_1) + d(u_2) + d(u_3)\ge |\{1, 4, 6\}| +  |\{2, 3\}| +  |\{2, 7, 3, 5\}| +  |\{4, 6\}| + 2t_{u v} = 2\Delta + 7$.
     It follows from \ref{1331switch} that $1\in C(u_4)$ and $H$ contains a $(4, 1)_{(u_3, u_4)}$-path.
     Thus, $(u u_1, u u_2, u u_3, u v)\overset{\divideontimes}\to (3, 1, T_3, 2)$.
\end{itemize}

{Case 2.1.2.} $T_1\setminus T_0\ne \emptyset$.
It follows from Corollary~\ref{auv2}(5.1) that $1\in S_u$,
and from Corollary~\ref{auv2d(u)5}(9) that $2\in S_u\setminus C(u_3)$.
Then $\sum_{x\in N(u)\setminus \{v\}}d(x)\ge 2\Delta + 12 + |\{1/2\}| +t_0 = 2\Delta + 13 + t_0$.
It follows that $t_0 = 0$, mult$_{S_u}(1) = $ mult$_{S_u}(3) = 1$, mult$_{S_u}(2) = $ mult$_{S_u}(4) = 2$.
If $T_1\subseteq C(u_2)$, one sees that $d(u_1) + d(u_2) + d(u_3) = |\{1, 2\}| + |\{2\}| + |\{3, 5, 2, 7\}| + |\{4, 4, 6, 6\}| + 2t_{u v} = 2\Delta + 7$,
then $1, 3\not\in S_u\setminus C(u_4)$, $4\in C(u_2)\setminus C(u_3)$, and $H$ contains a $(4, 3)_{(u_2, u_4)}$-path.
Thus, by \ref{1331switch}, we are done.
Otherwise, $T_1\setminus C(u_2)\ne \emptyset$.
Thus, by Lemma \ref{132T34T113onlyone}, we are done.

{Case 2.2.} mult$_{S_u}(6) = 1$.
One sees from \ref{136236126inCu1Cu3}(1--2) that $1\in C(u_3)$ or $3\in C(u_1)$,
and from Corollary \ref{auv2}(5.2) that if there exists a $j\in T_3\cap T_4$ such that $H$ contains no $(3, j)_{(u_4, v_3)}$-path, then mult$_{S_u}(6)\ge 2$.
Hence, $H$ contains a $(3, T_3\cap T_4)_{(u_4, v_3)}$-path and $3\in C(u_4)$.
It follows from \ref{136Cu34T1}(1) that $2\in C(u_1)$, or $H$ contains a $(4, 2)_{(u_1, u_4)}$-path,
and from \ref{136Cu34T1}(2) that $T_1\subseteq C(u_2)\cap C(u_4)$.
Hence, $\sum_{x\in N(u)\setminus \{v\}}d(x)\ge |\{1\}| + |\{2\}| + |\{2, 7, 3, 5\}| + |\{4, 6\}| + |\{4, 4, 6, 5, 7, 2, 3, 1/3\}| + 2t_{u v} + t_0 = 2\Delta + 12 + t_0\ge 2\Delta + 12 + t_1\ge 2\Delta + 12 + 1 = 2\Delta + 13$.
It follows that $t_1 = 1$, mult$_{S_u}(i) = 1$, $i = 5, 7$, mult$_{S_u}(1) + $mult$_{S_u\setminus C(u_4)}(3) = 1$,
and $H$ contains a $(4, 3)_{(u_2, u_4)}$-path.

(2.2.1.) $6\in C(u_3)$.
Then $3\in C(u_1)$ and $1\not\in S_u$.
It follows from \ref{136Cu34T1}(3) that $4\in C(u_1)$ and then $W_1 = \{1, 3, 4\}$.
Then $(u u_1, u u_3)\to (6, 1)$ reduces the proof to (2.4.2.)
\footnote{Or, by \ref{131notCu3356Cu1}, we are done.}.

(2.2.2.) $6\in C(u_1)$.
Then $1\in C(u_3)$ and $(u u_3, u u_1)\to (6, 3)$ reduces the proof to (2.4.1.)
\footnote{Or, by \ref{136notCu16notCu3}(2.1), we are done.}.
\end{proof}

By Lemma \ref{13T3inCu4}, hereafter we assume that $T_3\subseteq C(u_4)$.
One sees that $t_3\ge 2$ and $w_1 + w_4 + t_{u v} + t_3 \le d(u_1) + d(u_4)\le \Delta + 7$.
Hence, $w_1 + w_4 + t_3 \le 9$ and $w_1 + w_4\le 7$.

\begin{lemma}
\label{135766}
\begin{itemize}
\parskip=0pt
\item[{\rm (1)}]
	$5\in S_u\setminus C(u_3)$, $7\in S_u$, and mult$_{S_u}(7) +  $mult$_{S_u\setminus C(u_3)}(5)\ge 2$.
\item[{\rm (2)}]
	If $7\in C(u_3)$ and $H$ contains a $(7, j)_{(u_3, v)}$-path for some $j\in T_3$,
    then $7\in S_u\setminus C(u_3)$.
\item[{\rm (3)}]
    mult$_{S_u}(6)\ge 2$.
\item[{\rm (4)}]
    If mult$_{S_u}(7) + $mult$_{S_u\setminus C(u_3)}(5) = 2$, then $6\in C(u_3)$ and $H$ contains a $(6, T_3)_{(u_3, u_4)}$-path.
\item[{\rm (5)}]
    $T_4\ne \emptyset$.
\end{itemize}
\end{lemma}
\begin{proof}
For (1), assume that $5\not\in S_u\setminus C(u_3)$.
It follows from \ref{315234} that $3\in C(u_1)$,
from \ref{135Cu42T34B1B3}(1) that $2\in C(u_3)$, $T_3\subseteq C(u_2)$,
from \ref{135Cu42T34B1B3}(2) that $4\in C(u_2)\cup C(u_3)$,
and from \ref{135notSu} that $2, 4\in C(u_1)$, $T_1\subseteq C(u_2)\cap C(u_4)$.
If there exists a $j\in \{1, 7\}\setminus S_u$, then $(u u_1, u u_3, u v)\overset{\divideontimes}\to (5, 7, T_3)$.
Otherwise, $1, 7\in S_u$.
Thus, $\sum_{x\in N(u)\setminus \{v\}}d(x)\ge |[1, 4]| + |\{2\}| + |\{3, 5, 2\}| + |\{4, 6\}| + |\{4, 6, 1, 7\}| + 2t_{u v} + t_0\ge 2\Delta + 10 + t_1 + t_3\ge 2\Delta + 10 + 2 + 2 = 2\Delta + 14$, a contradiction.
Hence, $5\in S_u\setminus C(u_3)$.

Assume that $7\not\in S_u$.
It follows from Corollary~\ref{auv2}(5.1) that $2, 4\in C(u_1)\cap C(u_3)$ and $T_1\cup T_3\subseteq C(u_2)\cap C(u_4)$,
and from \ref{136Cu34T1}(1) that mult$_{S_u}(6)\ge 2$.
Since $6\in C(u_1)\cup C(u_3)$, $t_1 + t_3\ge w_1 + w_3 - 4\ge 8 - 4 = 4$.
Thus, $\sum_{x\in N(u)\setminus \{v\}}d(x)\ge |\{1, 2, 4\}| + |\{2\}| + |[2, 5]| + |\{4, 6\}| + |\{5, 6, 6\}| + 2t_{u v} + t_0\ge 2\Delta + 9 + t_1 + t_3\ge 2\Delta + 9 + 4 = 2\Delta + 13$.
It follows that $1, 3\not\in S_u$, mult$_{S_u\setminus C(u_3)}(5) = 1$, and $5\in C(u_i)$ for some $i\in \{2, 4\}$.
Then $(u u_1, u u_i, u v)\overset{\divideontimes}\to (5, 1, T_3)$.

Hence, $7\in S_u$ and then mult$_{S_u}(7) +  $mult$_{S_u\setminus C(u_3)}(5)\ge 2$.

For (2), assume that $7\in C(u_3)$, $H$ contains a $(7, 8)_{(u_3, v)}$-path, and $7\not\in S_u\setminus C(u_3)$.
It follows from \ref{137Cu2344T1}(1) that $3\in C(u_1)$ and $H$ contains a $(3, 7)_{(u_1, u_3)}$-path,
and the from \ref{137Cu2344T1}(2.1) that $4\in C(u_1)$, $T_1\subseteq C(u_4)$,
from \ref{137Cu42T1}(2) that $2\in C(u_1)$, $T_1\subseteq C(u_2)$.
If $2\not\in B_1\cup B_3$, then $(u u_2, u v)\overset{\divideontimes}\to (7, 2)$;
if $4\not\in B_1\cup B_3$, then $(u u_4, u v)\overset{\divideontimes}\to (7, 4)$.
Otherwise, $2, 4\in B_1\cup B_3$.
\begin{itemize}
\parskip=0pt
\item
     $2\not\in C(u_3)\cup C(u_4)$. One sees that $2\in B_1$ and $H$ contains no $(2, 4)_{(u_3, u_4)}$-path.
     If $H$ contains no $(7, j)_{(u_2, v_7)}$-path for some $j\in T_1$,
     then $(u u_2, u u_3, uv)\overset{\divideontimes}\to (7, 2, j)$.
     otherwise, $H$ contains a $(7, T_1)_{(u_2, v_7)}$-path and $(u u_3, u u_4, u v)\overset{\divideontimes}\to (7, 8, T_1)$.
\item
     $4\not\in C(u_2)\cup C(u_3)$.
     One sees that $4\in B_1$ and $H$ contains no $(2, 4)_{(u_2, u_3)}$-path.
     If $H$ contains no $(7, j)_{(u_2, v_7)}$-path for some $j\in T_1$,
     then $(u u_2, u u_3, u v)\overset{\divideontimes}\to (7, 8, j)$.
     otherwise, $H$ contains a $(7, T_1)_{(u_2, v_7)}$-path and $(u u_3, u u_4, u v)\overset{\divideontimes}\to (4, 7, T_1)$.
\item
     $2\in C(u_3)\cup C(u_4)$ and $4\in C(u_2)\cup C(u_3)$.
     If $1\not\in C(u_1)$ and $H$ contains no $(1, i)_{(u_3, u_i)}$-path,
     then $(u u_1, u u_3, u v)\to (7, 1, 8)$.
     Otherwise, $1\in C(u_1)$ or $H$ contains a $(1, 2/4)_{(u_3, u)}$-path.
     Then $\sum_{x\in N(u)\setminus \{v\}}d(x)\ge |[1, 4]| + |\{2\}| + |\{3, 5, 7\}| + |\{4, 6\}| + |\{5, 6, 2, 4, 1\}| + 2t_{u v} + t_0\ge 2\Delta + 11 + t_1\ge 2\Delta + 11 + 2 = 2\Delta + 13$.
     It follows that $t_1 = 2$,  mult$_{S_u}(6) = 1$, $T_3\cap T_1 = \emptyset$ and then $4\in C(u_3)$.
     One sees that $2\in C(u_3)\cup C(u_4)$ and $d(u_3) + d(u_4)\ge |\{3, 5, 7, 4\}| + |\{4, 6\}| + |\{2\}| + t_{uv} + t_1\ge \Delta + 5 + 2 = \Delta + 7$.
     It follows that $1\in C(u_2)$, $H$ contains a $(1, 2)_{(u_2, u_3)}$-path and $2\in C(u_3)$.
     One sees that $W_4 = \{4, 6\}$,
     and $H$ contains no $(j, i)_{(u_4, v)}$-path for each $i\in C(v)$ and $j\in T_4$.
     Thus, by Corollary~\ref{auv2}(4), we are done.
\end{itemize}

For (3), assume that mult$_{S_u}(6) = 1$.
It follows from \ref{136Cu34T1}(4) that $T_4\ne \emptyset$, $1/3/5/7\in C(u_4)$,
and from \ref{136Cu34T1}(1) that $2\in B_1\cup B_3$,
and $2\in C(u_1)$ or $H$ contains a $(2, 4)_{(u_1, u_4)}$-path,
i.e., $2\in C(u_1)$ or $2\in C(u_3)\cap C(u_4)$.
By Corollary~\ref{auv2}(4), $H$ contains a $(7, T_3)_{(u_3, v_7)}$-path,
or $5\in B_1\cap C(u_1)$, mult$_{S_u\setminus C(u_3)}(5)\ge 2$.
Hence, by (1) and (2), we have mult$_{S_u}(7) +  $mult$_{S_u\setminus C(u_3)}(5)\ge 3$.

When $6\in C(u_3)$, then $3\in C(u_1)$, $H$ contains a $(3, 6)_{(u_1, u_3)}$-path,
and it follows from \ref{136Cu34T1}(3) that $4\in C(u_1)$ and $H$ contains a $(4, T_1)_{(u_1, u_4)}$-path.
We distinguish whether or not $2\in C(u_1)$:
\begin{itemize}
\parskip=0pt
\item
     $2\in C(u_1)$.
     One sees that $t_1\ge 2$ and $d(u_3) + d(u_4)\ge t_{uv} + t_1 + |\{3, 5, 6, 2/4\}| + |\{4, 6, 5/1/3/7\}|\ge\Delta + 5 + t_1\ge \Delta + 5 + 2 = \Delta + 7$.
     It follows that $W_1 = [1, 4]$, $t_1 = 2$, $W_3 = \{3, 5, 6, a_3\}$, $W_4 = \{4, 6, 5/1/3/7\}$, $d(u_i) = \Delta = 7$, $i = 1, 3, 4$, $d(u_2)\le 6$,
     and from \ref{136notCu16notCu3}(1.1--1.2) that $H$ contains a $(a_3, \{1, 7\})_{(u_3, u_{a_3})}$-path, $1, 7\in C(u_{a_3})$.
     Hence, $a_3 = 2$, $\{1, 7\}\cup T_3\subseteq C(u_2)$, and $T_1\setminus C(u_2)\ne \emptyset$.
     One sees from $5\not\in C(u_1)$ that $H$ contains a $(6, T_3)_{(u_3, u_4)}$-path.
     Then $(u u_2, u u_3, u v)\overset{\divideontimes}\to (6, T_3, T_1\cap T_2)$.
\item
     $2\not\in C(u_1)$ and $2\in C(u_3)\cap C(u_4)$.
     One sees that $t_1\ge 1$ and $d(u_3) + d(u_4)\ge t_{u v} + t_1 + |\{3, 5, 6, 2\}| + |\{4, 6, 2, 5/1/3/7\}|\ge\Delta + 6 + t_1\ge \Delta + 6 + 1 = \Delta + 7$.
     It follows that $W_1 = \{1, 3, 4\}$, $W_3 = \{3, 5, 6, 2\}$, $W_4 = \{4, 6, 2, 5/1/3/7\}$,
     $d(u_i) = \Delta = 7$, $i = 1, 3, 4$, $d(u_2)\le 6$, $t_1 = 1$, $t_3 = 2$, $c_{13} = 2$,
     and from \ref{136notCu16notCu3}(1.1--1.2) that $H$ contains a $(2, \{1, 7\})_{(u_3, u_{a_3})}$-path, $1, 7\in C(u_2)$.
     One sees that $1, 7\setminus C(u_4)\ne \emptyset$ and assume w.l.o.g. $1\not\in C(u_4)$\footnote{If $7\not\in C(u_4)$, then $u u_1\to 7$ reduces the proof to the case where $1\not\in C(u_4)$. }.
     Then $(u u_1, u u_4, u v)\overset{\divideontimes}\to (T_1, 1, T_3)$.
\end{itemize}

Hence, $6\in C(u_1)$, $1\in C(u_3)$, and $H$ contains a $(1, 6)_{(u_1, u_3)}$-path.
It follows from \ref{136notCu16notCu3} (2.1) that $3\in C(u_1)$, or $H$ contains a $(3, 2/4)_{(u_1, u)}$-path.
Since $7\in C(u_3)$ or $5\in C(u_1)$,
$t_1 + t_3\ge w_1 + w_3 - 4 = |\{1, 6, 2\}| + |\{3, 5, 1, 2/4\}| + |\{5/7\}| - 4 = 4$.
One sees that $\sum_{x\in N(u)\setminus \{v\}}d(x)\ge |\{1, 6, 2\}| +|\{2\}| + |\{3, 5, 1, 2/4\}| + |\{4, 6\}|+  |\{3\}| + $mult$_{S_u}(7) +  $mult$_{S_u\setminus C(u_3)}(5) + 2t_{u v} \ge 11 + 3 + 2t_{u v} + t_0 = 2\Delta + 10 + t_0$.
It follows that $(T_3\cup T_1)\setminus T_0\ne \emptyset$.
We distinguish whether or not $2\in C(u_1)$:
\begin{itemize}
\parskip=0pt
\item
    $2\in C(u_1)$.
    When $H$ contains no $(4, j)_{(u_3, u_4)}$-path for some $j\in T_3$,
    one sees from \ref{136Cu34T1}(2) that $T_1\subseteq T_0$ and then $T_3\setminus T_0\ne \emptyset$,
    i.e., $T_3\setminus C(u_2)\ne \emptyset$,
    thus, $H$ contains a $(4, T_3)_{(u_3, u_4)}$-path, a contradiction.
    In the other case, $4\in C(u_3)$ and $H$ contains a $(4, T_3)_{(u_3, u_4)}$-path.
    \begin{itemize}
    \parskip=0pt
    \item
         $2\not\in C(u_3)\cup C(u_4)$.
         It follows from Corollary~\ref{auv2}(5.1) that $3\in C(u_4)$ or $H$ contains a $(3, 1)_{u_4, u_1}$-path,
         from \ref{135Cu42T34B1B3} that $5\in C(u_4)$ or $H$ contains a $(5, 1)_{u_4, u_1}$-path,
         and from \ref{136Cu34T1}(2) that $T_1\subseteq T_0$.
         When $w_4 \ge 4$, it follows that $W_3 = \{1, 4, 3, 5\}$, $t_1 = 1$, $t_3 = 2$, $c_{13} = 2$,
         and then $W_1 = \{1, 2, 6\}$.
         Thus, $(v u_3, u v)\overset{\divideontimes}\to (T_3, 5)$
         \footnote{a contradiction to Corollary~\ref{auv2}(4). }.
         In the other case, $w_4\le 3$.
         Hence, $W_4 = \{1, 4, 6\}$, $3, 5\in C(u_1)$, and then $W_1 = \{1, 2, 3, 5, 6\}$, $W_3 = \{1, 3, 4, 5\}$. Thus, $(v u_3, u u_3, u v)\overset{\divideontimes}\to (T_3, 5, T_1)$
         \footnote{a contradiction to Corollary~\ref{auv2}(4). }.
    \item
         $2\in C(u_3)\cup C(u_4)$.
         One sees that $w_3 + w_4\ge |\{1, 3, 4, 5\}| + |\{4, 6, 5/1/3/7\}| + |\{2\}| = 8$.
         If $T_1\subseteq C(u_4)$, then $t_1 = 1$ and then $W_3 = \{1, 4, 3, 5\}$, $t_3 = 2$, $c_{13} = 2$,
         and then $W_1 = \{1, 2, 6\}$.
         Thus, $(v u_3, u v)\overset{\divideontimes}\to (T_3, 5)$
         \footnote{a contradiction to Corollary~\ref{auv2}(4). }.
         Otherwise, $T_1\setminus C(u_4)\ne \emptyset$.
         It follows from Corollary~\ref{auv2}(5.2) and (3) that $4\in B_1\cup B_3$, $4\in C(u_1)\cup C(u_2)$,
         and from \ref{136Cu34T1}(4) that $5/1/7\in C(u_4)$.
         Hence, $\sum_{x\in N(u)\setminus \{v\}}d(x)\ge 2\Delta + 10 + |\{2, 4\}| + t_0 = 2\Delta + 12 + t_0$.
         Since $t_3\ge 2$, we have $T_3\setminus C(u_2)\ne \emptyset$,
         and it follows from Corollary~\ref{auv2}(5.1) that $3\in C(u_4)$ or $H$ contains a $(3, 1/2)_{(u_4, u)}$-path,
         from \ref{135Cu42T34B1B3}(1) that $5\in C(u_4)$ or $H$ contains a $(5, 1/2)_{(u_4, u)}$-path.
         If $5\not\in C(u_1)\cup C(u_4)$, one sees that $H$ contains a $(2, 5)_{(u_2, u_4)}$-path,
         then $(u u_1, u u_3, u v)\overset{\divideontimes}\to (5, 6, T_3)$.
         Otherwise, together with \ref{136notCu16notCu3} (2.1), we have $3, 5\in C(u_1)\cup C(u_4)$.
         If mult$_{S_u}(1) + $mult$_{S_u}(3) = 2$,
         one sees from Corollary~\ref{auv2}(5.1) that $3\in C(u_2)$,
         then $(u u_1, u u_2, u u_3)\to (3, 1, 6)$, and $uv\overset{\divideontimes} \to 2$,
         or $uv\overset{\divideontimes} \to T_1\cap R_1$
         \footnote{Or, One sees that $(u u_1, u u_2, u u_3)\to (3, 1, 6)$ reduces to an acyclic edge $(\Delta + 5)$-coloring $c'$ for $H$ such that $C(u)\cap C(v) = \{1, 3\}$, $2\not\in C'(u)$,  then by \ref{1331recoloring}, we are done. }.
         Otherwise, mult$_{S_u}(1) + $mult$_{S_u}(3)\ge 3$.

         \quad It follows that $t_0 = 0$, mult$_{S_u}(7) + $mult$_{S_u\setminus C(u_3)}(5) = 3$,
         mult$_{S_u}(i) = 2$, $i = 2, 4$,
         and from \ref{1325254545}(2) that if there exists a $j\in T_3$ such that $H$ contains no $(j, 7)_{(u_3, v)}$-path, then $5\in B_1$, $H$ contains a $(4, 5)_{(u_3, u_4)}$-path, and $5\in C(u_3)\cap C(u_4)$.

         \quad If $2\in C(u_3)$, one sees from $7\in C(u_3)$ or $5\in C(u_1)$ that $3, 4\not\in C(u_1)$ and from \ref{136notCu16notCu3}(2.1) that $H$ contains a $(2, 3)_{(u_1, u_2)}$-path,
         then $(u u_1, u u_2, u u_3)\to (3, 1, 6)$, and $uv\overset{\divideontimes} \to 2$,
         or $uv\overset{\divideontimes} \to T_1\cap R_1$
         \footnote{Or, One sees that $(u u_1, u u_2, u u_3)\to (3, 1, 6)$ reduces to an acyclic edge $(\Delta + 5)$-coloring $c'$ for $H$ such that $C(u)\cap C(v) = \{1, 3\}$, $2\not\in C'(u)$,  then by \ref{1331recoloring}, we are done. }.
         Otherwise, $2\in C(u_4)\setminus C(u_3)$ and then $2\in B_1$, $T_1\cap R_1\ne \emptyset$.

         \quad If $3\in C(u_1)\cap C(u_2)$, one sees that $1\not\in S_u\setminus C(u_3)$,
         then $uv\to \alpha\in T_1\cap R_1$, and $(u u_1, u u_2, u u_3, u u_4)\to (T_1\setminus \{\alpha\}, 1, 2, 3)$.
         Otherwise, $3\not\in C(u_1)\cap C(u_2)$.
         \begin{itemize}
         \parskip=0pt
         \item
              $7\in C(u_3)$.
              It follows that $W_3 = \{1, 3, 4, 5, 7\}$, $W_4 = \{4, 6, 2, 5/1/7\}$.
              Together with $w_1 + w_3\le 9$ and $3, 5\in C(u_1)\cup C(u_3)$,
              we have $W_1 = \{1, 2, 3, 6\}$, $W_4 = \{4, 6, 2, 5\}$, and then $3\in C(u_2)$.
              Thus, are done.
         \item
              $7\not\in C(u_3)$. Then $5\in C(u_1)\cap C(u_4)$.
              It follows from \ref{136notCu16notCu3}(2.2) that $7\in C(u_1)$ or $H$ contains a $(7, 2/4)_{(u_1, u)}$-path.
              If $7\not\in C(u_4)$ and $H$ contains no $(7, i)_{(u_4, u_i)}$-path for each $i\in \{1, 2\}$,
              then $(u u_4, u u_3, u v)\overset{\divideontimes}\to (7, T_3, T_1)$.
              Otherwise, $7\in C(u_4)$ or $H$ contains a $(7, 1/2)_{(u_4, u)}$-path.
              Hence, $7\in C(u_1)\cup C(u_4)$ and $W_1\cup W_4 = \{1, 2, 5, 6\}\cup \{4, 5, 6\}\cup \{3, 7\}$.
              We will have $7\in C(u_2)$ no matter on whether or not $7\in C(u_1)$, a contradiction.
         \end{itemize}
    \end{itemize}
\item
     $2\not\in C(u_1)$ and $2\in C(u_3)\cap C(u_4)$.
     One sees that $d(u_3) + d(u_4)\ge t_{uv} - t_1 + t_3 + |\{1, 4, 6\}| + |\{4, 6, 2, 5/1/7/3\}| = \Delta + 5 + t_3 - t_1\ge \Delta + 7 - t_1$.
     If $T_1\subseteq C(u_4)$, then $W_1 = \{1, 4, 6\}$, $3\in C(u_4)$,
     and it follows from \ref{315234} that $5\in C(u_4)$, a contradiction.
     Otherwise, $T_1\setminus C(u_4)$ and it follows from \ref{136Cu34T1}(2) that $4\in C(u_3)$ and $H$ contains a $(4, T_3)_{(u_3, u_4)}$-path.
     Hence, $W_3 = [1, 5]$, $W_4 = \{4, 6, 2, 5/1/3/7\}$, $T_1\cap C(u_4) = \emptyset$, and then $W_4 = \{4, 6, 2, 5/1/7\}$, $3\in C(u_1)$.
     Thus, $w_1 + w_4\ge |\{1, 6, 4, 3, 7\}| + 5 = 10$,  a contradiction.
\end{itemize}

For (4), if mult$_{S_u}(7) + $mult$_{S_u\setminus C(u_3)}(5) = 2$, then it follows from (1) and (2) naturally that $6\in C(u_3)$ and $H$ contains a $(6, T_3)_{(u_3, u_4)}$-path.

For (5), it $C(u_4) = T_{uv}\cup \{4, 6\}$, it follows from \ref{135Cu42T34B1B3}(2) that mult$_{S_u}(4)\ge 2$,
from Corollary \ref{auv2d(u)5}(7.2) that $2\in C(u_1)\cap C(u_3)$ and $T_1\cup T_3\subseteq C(u_2)$,
thus, $d(u_1) + d(u_2) + d(u_3)\ge |\{1, 2\}| + |\{2\}| + |\{3, 5, 2\}| + |\{5, 7, 4, 4, 6, 6\}| + 2t_{u v} = 2\Delta + 8$.
Otherwise, $T_4\ne \emptyset$.
\end{proof}

\begin{lemma}
\label{13T1notinCu4T3notinCu2}
$T_1\setminus C(u_4)\ne \emptyset$ and $T_3\setminus C(u_2)\ne \emptyset$.
\end{lemma}
\begin{proof}
One sees from \ref{1325254545}(1) that $5\in S_u\setminus C(u_3)$, $7\in S_u$,
from \ref{1325254545}(3) that mult$_{S_u}(6)\ge 2$,
and from Corollary~\ref{auv2d(u)5}(7.1) that $2, 4\in C(u_1)$, or $1\in S_u$,
and $2, 4\in C(u_3)$, or $3\in S_u$.
Then $|[1, 6]| + |\{5, 6, 7, 7\}| + 2t_{u v} = 2\Delta + 6$,
and $\sum_{x\in N(u)\setminus \{v\}}d(x)\ge |[1, 6]| + |\{5, 6, 7, 7\}| + |\{2, 4\}/\{a_1, 1\}| + |\{2, 4\}/\{a_3, 3\}| + 2t_{u v} + t_0 = 2\Delta + 10 + t_0$.
It follows that if $T_1\cup T_3\subseteq T_0$, then $t_1 = 1$, $t_3 = 2$.
One sees from \ref{1325254545}(4) that if $T_1\cup T_3\subseteq C(u_4)$, then $C_{13}\ne \emptyset$.
We discuss the following two subcases.

{\bf Case 1.} $T_1\cup T_3\subseteq C(u_2)$. Then $t_0\ge t_3$.
We distinguish whether or not $T_1\subseteq C(u_4)$:

{Case 1.1.} $T_1\setminus C(u_4)\ne \emptyset$.
Then $2\in C(u_1)$,
and it follows from Corollary~\ref{auv2}(5.1) that $1\in C(u_2)$ or $H$ contains a $(1, 3/4)_{(u_2, u)}$-path,
from Corollary~\ref{auv2}(5.2) that mult$_{S_u}(4)\ge 2$.
One sees from  Corollary~\ref{auv2d(u)5}(7.1) that $2\in C(u_3)$ or $3\in S_u$.
Hence, $\sum_{x\in N(u)\setminus \{v\}}d(x)\ge 2\Delta + 6 + |\{4, 4, 2, 1, 3/2\}| + t_0\ge 2\Delta + 11 + 2 = 2\Delta + 13$.
It follows that $\Delta = 7$, $t_3 = 2$, $T_1\cap C(u_4) = \emptyset$,
mult$_{S_u}(4) = 2$, mult$_{S_u}(i) = 1$, $i = 1, 5, 7$,
and from Lemma~\ref{135766}(3) that $6\in C(u_3)$.
One sees that $W_3 = \{3, 5, 6, 2/4\}$.
It follows from Corollary~\ref{auv2d(u)5}(7.1) that $3\in S_u$ and $2\not\in S_u\setminus C(u_1)$.
\begin{itemize}
\parskip=0pt
\item
     $1\in C(u_2)$.
     One sees that $1, 2\not\in C(u_4)$.
     It follows from Corollary~\ref{auv2}(5.1) and \ref{135Cu42T34B1B3}(1) that $3, 5\in C(u_4)$,
     and then from \ref{1331switch} that $4\in C(u_1)$.
     Then $(u u_1, u u_2, u v)\overset{\divideontimes}\to (T_1, 5, T_3)$
     \footnote{Or, by \ref{135notCu1}(1), we are done. }.
\item
     $1\in C(u_4)$.
     It follows from Corollary~\ref{auv2}(5.1) that $4\in C(u_2)$ and $H$ contains a $(1, 4)_{(u_2, u_4)}$-path,
     and from \ref{131notCu3356Cu1} and \ref{1331switch} that for each $j\in \{3, 5\}$, $j\in C(u_1)$, or $H$ contains a $(2, j)_{(u_1, u_2)}$-path.
     Together with Corollary~\ref{auv2}(5.1) and \ref{135Cu42T34B1B3}(1), we have $3, 5\in C(u_1)$.
     Let $(u u_2, u u_3)\to (5, 2)$.
     If $H$ contains no $(5, j)_{(u_1, v_1)}$-path for some $j\in T_1$,
     then $(u u_1, u v)\overset{\divideontimes}\to (j, T_3)$.
     Otherwise, $H$ contains a $(5, T_1)_{(u_1, u_2)}$-path and $uv\overset{\divideontimes}\to T_1$.
\end{itemize}

{Case 1.2.} $T_1\subseteq C(u_4)$.
Then $\sum_{x\in N(u)\setminus \{v\}}d(x)\ge 2\Delta + 10 + t_0\ge 2\Delta + 10 + t_1 + t_3\ge 2\Delta + 10 + 1 + 2 = 2\Delta + 13$.
It follows that $\Delta = 7$, $t_1 = 1$, $t_3 = 2$, $c_{13} = 2$, $C_{13}\cap (C(u_2)\cup C(u_4)) = \emptyset$,
$\sum_{i\in [1, 4]} $mult$_{S_u}(i) = 4$, mult$_{S_u}(i) = 1$, $i = 5, 7$,
and from Lemma~\ref{135766}(3) that $6\in C(u_3)$.
One sees that $W_3 = \{3, 5, 6, 2/4\}$ and then $3\in S_u$.
It follows from Corollary~\ref{auv2d(u)5}(4.2) that there exists a $\gamma\in \{2, 4\}$ such that mult$_{S_u}(\gamma)\ge 2$.
\begin{itemize}
\parskip=0pt
\item
     $2, 4\in C(u_1)$.
     It follows that $1\not\in S_u$ and from Corollary~\ref{auv2}(5.1) that $3\in C(u_{a_3})$.
     Thus, one sees from Corollary~\ref{auv2}(3) that mult$_{S_u}(j)\ge 2$, $j\in \{2, 4\}\setminus \{a_2\}$,
     a contradiction.
\item
     $\{2, 4\}\setminus C(u_1)\ne \emptyset$.
     Then $a_1 = a_3 = \gamma$, $4\not\in C(u_2)$, $2\not\in C(u_4)$,
     and from Corollary~\ref{auv2}(3) that $1/3\in C(u_\beta)$, where $\{\gamma, \beta\} = \{2, 4\}$.
     One sees that $1\not\in C(u_3)$ and $\beta\not\in C(u_\gamma)$.
     It follows from Corollary~\ref{auv2}(5.1) that $1\in C(u_\gamma)$ and then $3\in C(u_\beta)$.
     Then $(u u_\beta, u u_1, u v)\overset{\divideontimes}\to (1, \beta, T_3)$.
\end{itemize}

{\bf Case 2.} $(T_1\cup T_3)\setminus C(u_2)\ne \emptyset$ and $T_1\cup T_3\subseteq C(u_4)$.
Then $C_{13}\ne \emptyset$, $w_1 + w_3\le 9 - c_{13}$,
and for each $i\in \{1, 3\}$ where $T_i\setminus C(u_2)\ne \emptyset$,
$4\in C(u_i)$, and $H$ contains a $(4, T_i)_{(u_i, u_4)}$-path,
and from Corollary~\ref{auv2}(5.2) and (3) that $2\in B_1\cup B_3$,
for each $j\in \{1, 3\}$, $2\in C(u_j)$, or $H$ contains a $(2, 4)_{(u_j, u_4)}$-path.
One sees clearly that mult$_{S_u}(4)\ge 1$ and mult$_{S_u}(2)\ge 2$.

{Case 2.1.} $T_3\subseteq C(u_2)$ and $T_1\setminus C(u_2)\ne \emptyset$.
Then $4\in C(u_1)$ and follows from  Corollary~\ref{auv2}(5.1) that $1\in C(u_4)$ or $H$ contains a $(1, 2/3)_{(u_4, u)}$-path,
from Corollary~\ref{auv2d(u)5}(7.1) that $4\in C(u_3)$ or $3\in S_u$.
Hence, $\sum_{x\in N(u)\setminus \{v\}}d(x)\ge 2\Delta + 6 + |\{2, 2, 4, 1, 4/3\}| + t_0\ge 2\Delta + 11 + t_3\ge 2\Delta + 11 + 2 = 2\Delta + 13$.
It follows that $t_3 = 2$,  mult$_{S_u}(2) = 2$, mult$_{S_u}(i) = 1$, $i = 5, 7$,
and from Lemma~\ref{135766}(3) that $6\in C(u_3)$.
Since $W_3 = \{3, 5, 6, 2/4\}$, we have $3\in S_u$, $4\not\in S_u\setminus C(u_1)$,
$W_3 = \{3, 5, 6, 2\}$ and from Corollary~\ref{auv2}(5.1) that $3\in C(u_2)$ or $H$ contains a $(1, 3)_{(u_1, u_2)}$-path.
Since $1\in C(u_4)$ or $H$ contains a $(1, 2)_{(u_2, u_4)}$-path, $1\in C(u_2)$, it follows from \ref{1331switch} that $3\in C(u_1)$ or $H$ contains a $(3, 2/4)_{(u_1, u)}$-path.
\begin{itemize}
\parskip=0pt
\item
     $1\in C(u_2)$.
     One sees that $2\in C(u_4)$, and $3\in C(u_1)$ or $H$ contains a $(4, 3)_{(u_1, u_4)}$-path.
     Hence, $3\in C(u_1)$.
     Then $uu_3\to \alpha\in T_3$,
     and $(u u_1, u u_2, u u_4, u v) \overset{\divideontimes} \to (2, 4, 1, T_3\setminus \{j\})$.
\item
     $1\in C(u_4)$.
     One sees that $1, 4\not\in C(u_2)$.
     Hence, $3\in C(u_2)$ and $2\in C(u_1)$.
     Then $(u u_1, u u_2, u u_3, u u_4)\to (T_1, 4, 1, 3)$ reduces the proof to (2.1.2.).
\end{itemize}

{Case 2.2.} $T_1\subseteq C(u_2)$ and $T_3\setminus C(u_2)\ne \emptyset$.
Then $4\in C(u_3)$ and follows from Corollary~\ref{auv2}(5.1) that $3\in C(u_4)$ or $H$ contains a $(3, 1/2)_{(u_4, u)}$-path,
from Corollary~\ref{auv2d(u)5}(7.1) that $4\in C(u_1)$ or $1\in S_u$,
from \ref{135Cu42T34B1B3}(1) that $5\in C(u_4)$, or $H$ contains a $(5, 1/2)_{(u_4, u)}$-path.
Hence, $\sum_{x\in N(u)\setminus \{v\}}d(x)\ge 2\Delta + 6 + |\{2, 2, 4, 3, 4/1\}| + t_0\ge 2\Delta + 11 + t_1\ge 2\Delta + 11 + 1 = 2\Delta + 12$ and $|C_{13}\cap T_0|\le 1$.

When $3, 5\in C(u_1)\cup C(u_4)$,
one sees that $d(u_1) + d(u_3) + d(u_4)\ge 2t_{uv} + |\{1\}| + |\{2\}| + |\{3, 5, 4\}| + |\{4, 6\} | + |\{2, 2, 3, 5, 6\}| = 2\Delta + 7$.
It follows that $2\in C(u_1)$ and $1, 7\not\in C(u_1)\cup C(u_3)\cup C(u_4)$.
If $2\not\in C(u_4)$, one sees that $W_4 = [3, 6]$, $W_3 = [2, 5]$ and $W_1 = \{1, 2, 6\}$,
then $(vu_3, uv)\overset{\divideontimes} \to (T_3, 5)$.
Otherwise, $2\in C(u_4)\setminus C(u_3)$, $2\in B_1$, and from \ref{1343567V1}(1) that $R_1\ne \emptyset$.
It follows from Corollary~\ref{auv2}(5.1) that $H$ contains a $(2, 7)_{(u_2, u_4)}$-path.
Since $u u_1\to 7$ and $(u u_1, u u_3)\to (7, 1)$ reduce to an acyclic edge $(\Delta + 5)$-coloring for $H$,
by \ref{1331switch}, we are done.

In the other case, $\{3, 5\}\setminus (C(u_1)\cup C(u_4))\ne \emptyset$.
It follows that $2\in C(u_4)$ and $H$ contains a $(2, \{3, 5\}\setminus (C(u_1)\cup C(u_4)))_{(u_2, u_4)}$-path.
If $H$ contains no $(2, j)_{(u_2, u_4)}$-path for some $j\in \{3, 5\}\setminus C(u_4)$,
one sees that $1\in C(u_4)$, $H$ contains a $(1, j)_{(u_1, u_2)}$-path,
and then $W_1 = \{1, 2/4, j\}$, $W_4 = \{4, 6, 2, 1\}$.
Then $(u u_1, u u_4, u u_3)\to (\{3, 5\}\setminus \{j\}, j, T_3\cap T_2)$ reduces the proof to (2.4.1.) or (2.4.2.).
Otherwise, $H$ contains a $(2, \{3, 5\}\setminus C(u_4))_{(u_2, u_4)}$-path.
\begin{itemize}
\parskip=0pt
\item
    $1\not\in S_u$. One sees that $2, 4\in C(u_1)$.
    If $5\not\in C(u_1)\cup C(u_4)$,
    then $(u u_1, u u_4, u u_3)\to (5, 1, T_3\cap T_2)$ reduces the proof to (2.4.1.).
    Otherwise, $5\in C(u_1)\cup C(u_4)$.
    It follows that $H$ contains a $(2, 3)_{(u_2, u_4)}$-path, $W_3 = [3, 6]$, and then $2\in B_1$, $R_1\ne \emptyset$.
    Then, by \ref{1331switch}, we are done.
\item
    $1\in S_u$. One sees that if $4\not\in S_u\setminus C(u_3)$,
    then for each $j\in C_{13}$, $H$ contains a $(j, i)_{(u_4, u_i)}$-path for some $i\in [1, 3]$,
    and that if $5\not\in C(u_1)\cup C(u_4)$, then it follows from \ref{315234} that $5/3\in C(u_1)$.
    \begin{itemize}
    \parskip=0pt
    \item
         $4\not\in S_4\setminus C(u_3)$.
         If $W_4 = \{4, 6, 2\}$, one sees that $C_{13}\subseteq C(u_2)$, $3, 5\in C(u_2)$, and $5/3\in C(u_1)$,
         then $\sum_{x\in N(u)\setminus \{v\}}d(x)\ge 2\Delta + 6 + |\{2, 2, 4, 3, 1, 5/3\}| + t_0\ge 2\Delta + 12 + t_1 + c_{13}\ge 2\Delta + 12 + 1 + 1 = 2\Delta + 14$.
         Otherwise, $w_4 = 4$, $c_{13}\ge 2$, and then $1/3\in C(u_4)$.
         It follows that $W_1 = \{1, 2, 3/5\}$, $W_3 = [3, 6]$, and $1, 5\in C(u_2)$.
         Since $\sum_{x\in N(u)\setminus \{v\}}d(x)\ge 2\Delta + 12 + |\{3/5\}| = 2\Delta + 13$,
         $W_4 = \{4, 6, 2, 3\}$.
         It follows from Corollary~\ref{auv2}(5.1) that $H$ contains a $(2, 7)_{(u_2, u_4)}$-path.
         Since $u u_3\to 7$ and $(u u_1, u u_3)\to (7, 1)$ reduce to an acyclic edge $(\Delta + 5)$-coloring for $H$, by \ref{1331switch}, we are done.
    \item
         $4\in C(u_1)\cup C(u_2)$.
         Then $\sum_{x\in N(u)\setminus \{v\}}d(x)\ge 2\Delta + 12 + |\{4\}| = 2\Delta + 13$.
         It follows that $t_0 = t_1 = 1$, mult$_{S_u}(i) = 1$, $i = 1, 3, 5, 7$,
         mult$_{S_u}(j) = 2$, $j = 2, 4, 6$,
         and from Lemma~\ref{135766}(4) that $6\in C(u_3)$ and $H$ contains a $(6, T_3)_{(u_3, u_4)}$-path.
         If $5\not\in C(u_4)$, then $5\in C(u_2)$, $3\in C(u_2)\cup C(u_4)$, and $5/3\in C(u_1)$, a contradiction. Otherwise, $W_4 = \{4, 6, 2, 5\}$, $W_3 = [3, 6]$, and $H$ contains a $(2, 3)_{(u_2, u_4)}$-path.
Then by \ref{1331switch}, we are done.
    \end{itemize}
\end{itemize}

{Case 2.3.} $T_3\setminus C(u_2)\ne \emptyset$ and $T_1\setminus C(u_2)\ne \emptyset$.
Then $4\in C(u_1)\cap C(u_3)$, $H$ contains a $(4, T_1\cup T_3)_{(u_4, u)}$-path,
and for each $j\in \{1, 3\}$, $j\in C(u_4)$, or $H$ contains a $(j, i)_{(u_4, u)}$-path.
When $1, 3\not\in S_u\setminus C(u_2)$,
it follows that $2\in C(u_4)$, $H$ contains a $(2, \{1, 3\})_{(u_2, u_4)}$-path,
then by \ref{1331switch}, we are done.
In the other case, $1/3\in S_u\setminus C(u_2)$.
One sees that $d(u_1) + d(u_3) + d(u_4)\ge 2t_{u v} + |\{1, 4\}| + |[3, 5]| + |\{4, 6\}| + |\{2, 2, 6, 1/3\}| = 2\Delta + 7$.
It follows that $d(u_i) = \Delta = 7$, $i = 1, 3, 4$, $d(u_2)\le 6$, $|\{1, 3\}\cap S_u\setminus C(u_2)|\le 1$,
and then $5, 6, 7\in C(u_2)$.
Further, $2\in C(u_4)$, $6\in C(u_3)$, and $H$ contains a $(2, \{1, 3\}\setminus (S_u\setminus C(u_2)))_{(u_2, u_4)}$-path.
\begin{itemize}
\parskip=0pt
\item
     $1\in C(u_3)\cup C(u_4)$.
     It follows that $W_3\cup W_4 = [3, 6]\cup \{2, 4, 6\}\cup \{1\}$, $W_1 = \{1, 2, 4\}$, and $H$ contains a $(2, 3)_{(u_2, u_4)}$-path.
     Further, it follows from Corollary~\ref{auv2}(5.1) that $H$ contains a $(2, 7)_{(u_2, u_4)}$-path.
     Since $u u_1\to 7$ and $(u u_1, u u_3)\to (3, 7)$ reduces to an acyclic edge $(\Delta + 5)$-coloring for $H$, by \ref{1331switch}, we are done.
\item
     $1\not\in C(u_3)\cup C(u_4)$.
     Then $H$ contains a $(1, 2)_{(u_2, u_4)}$-path.
     It follows from \ref{1331switch} that $3\in C(u_1)$ or $H$ contains a $(3, 2/4)_{(u_1, u)}$-path.
     If $2\in C(u_3)$, since $W_1 = \{1, 3, 4\}$, $W_4 = \{4, 6, 2\}$, and $H$ contains a $(3, 6)_{(u_1, u_3)}$-path, then $(v u_4, u v)\overset{\divideontimes}\to (T_4, 6)$.
     Otherwise, $2\not\in C(u_3)$.
     It follows from Corollary~\ref{auv2}(5.1) that $H$ contains a $(2, 7)_{(u_2, u_4)}$-path.
     Since $u u_3\to 7$ and $(u u_1, u u_3)\to (7, 1)$ reduce to an acyclic edge $(\Delta + 5)$-coloring for $H$, by \ref{1331switch}, we are done.
\end{itemize}

Hence, $T_1\setminus C(u_4)\ne \emptyset$, $T_1\subseteq C(u_2)$, and $T_3\setminus C(u_2)\ne \emptyset$.
\end{proof}

It follows from Corollary~\ref{auv2}(5.1) that $1, 3\in S_u$,
from Corollary~\ref{auv2}(5.2) and (3) that $i\in B_1\cup B_3$, mult$_{S_u}(i)\ge 2$, $i\in \{2, 4\}$,
and $4\in C(u_3)$, $H$ contains a $(4, T_3)_{(u_3, u_4)}$-path,
$2\in C(u_1)$, $H$ contains a $(2, T_1)_{(u_1, u_2)}$-path.
Hence, $\sum_{x\in N(u)\setminus \{v\}}d(x)\ge 6 + \sum_{i\in [1, 4]}$ mult$_{S_u}(i) + |\{6, 6, 5, 7\}| + 2t_{u v} + t_0 = 6 + 2 + 4 + 4 + 2\Delta - 4 + t_0 = 2\Delta + 12 + t_0$.
One sees that $2\le $ mult$_{S_u}(1) + $ mult$_{S_u}(3) \le 3$.
One sees that $1\in C(u_2)$ or $H$ contains a $(1, 3/4)_{(u_2, u)}$-path,
and that $3\in C(u_4)$ or $H$ contains a $(3, 1/2)_{(u_4, u)}$-path.
We distinguish the following two subcases on whether mult$_{S_u}(1) + $ mult$_{S_u}(3) = 2$ or not.

{\bf Case 1.} mult$_{S_u}(1) + $ mult$_{S_u}(3) = 2$.

{Case 1.1.}  $1\in C(u_3)$.
Then $3\in C(u_2)$ and $2\in C(u_4)$.
If $2\not\in C(u_3)$, since $(u u_1, u u_2, u u_3, u u_4)\to (T_1, 1, 2, 3)$ reduces to acyclic edge $(\Delta + 5)$-coloring for $H$,
then by \ref{1331recoloring}, we are done.
Otherwise, $2\in C(u_3)$, mult$_{S_u}(5) = 1$ and by Lemma~\ref{135766}(4), $6\in C(u_3)$.
Hence, $W_3 = [1, 6]$, $W_4 = \{2, 4, 6\}$, and by \ref{135Cu42T34B1B3}(1), $5\in C(u_2)$ and $H$ contains a $(2, 5)_{(u_2, u_4)}$-path.
Then $(u u_1, u v)\overset{\divideontimes}\to (5, T_3)$\footnote{Or, contradiction to \ref{315234}. }.

{Case 1.2.}  $1\in C(u_3)$. Then $4\in C(u_2)$ and $H$ contains a $(4, 1)_{(u_2, u_4)}$-path.
It follows from \ref{1331switch} that $3\in C(u_1)$ or $H$ contains a $(3, 2/4)_{(u_1, u)}$-path.
Hence, $3\in C(u_1)\cup C(u_2)$.
If $3\in C(u_4)$, it follows that $4\in C(u_1)$ and $H$ contains a $(4, 3)_{(u_1, u_4)}$-path,
then $(u u_1, u u_2, u u_3, u v)\overset{\divideontimes}\to (T_1\cap T_4, 3, 1, T_3\cap T_2)$.

Otherwise, $3\in C(u_1)$.
If $2\not\in C(u_4)$, then $(u u_1, u u_2, u u_4, u v)\to (T_1\cap T_4, 1, 2, T_3)$.
If $2\not\in C(u_3)$, since $(u u_1, u u_2, u u_3, u u_4)\to (T_1\cap T_4, 1, 2, 3)$ reduces to acyclic edge $(\Delta + 5)$-coloring for $H$,
then by \ref{1331recoloring}, we are done.
Otherwise, $2\in C(u_3)\cap C(u_4)$ and $4\not\in C(u_1)$.
Since $(u u_1, u u_2, u u_3, u u_4)\to (4, 1, T_3, 3)$ reduces to acyclic edge $(\Delta + 5)$-coloring for $H$,
then by \ref{1331recoloring}, we are done.

{Case 1.3.}  $1\in C(u_2)$.
One sees that $3\in C(u_4)$ or $H$ contains a $(2, 3)_{(u_2, u_4)}$-path.
When $3\in C(u_2)$, it follows that $2\in C(u_4)$, $H$ contains a $(2, 3)_{(u_2, u_4)}$-path,
and from \ref{1331switch} that $2\in C(u_3)$ and $H$ contains a $(1, 2)_{(u_2, u_3)}$-path.
Then $(u u_1, u u_4, u u_3, u v)\overset{\divideontimes}\to (3, 1, T_3\cap T_2, T_1\cap T_4)$
\footnote{In fact, $T_3\cap C(u_2) = \emptyset$, $T_1\cap C(u_4) = \emptyset$,
i.e., $T_3\cap T_2 = T_3$, $T_1\cap T_4 = T_1$. }.
In the other case, $3\in C(u_4)$.
If $H$ contains no $(4, 3)_{(u_2, u_4)}$-path,
then $(u u_1, u u_2, u u_3, u v)\overset{\divideontimes}\to (T_1\cap T_4, 3, 1, T_3\cap T_2)$.
Otherwise, $H$ contains a $(4, 3)_{(u_2, u_4)}$-path.
It follows from \ref{1331switch} that $2\in C(u_3)$ and $H$ contains a $(1, 2)_{(u_2, u_3)}$-path.
If $2\not\in C(u_4)$,
then $(u u_1, u u_2, u u_3, u u_4, u v)\overset{\divideontimes}\to (T_1\cap T_4, 3, 1, 2, T_3\cap T_2)$.
Otherwise, $2\in C(u_4)$, mult$_{S_u}(5) = 1$ and by Lemma~\ref{135766}(4), $6\in C(u_3)$.
Hence, $W_3 = [2, 6]$, $W_4 = \{2, 3, 4, 6\}$, and by \ref{135Cu42T34B1B3}(1), $5\in C(u_2)$ and $H$ contains a $(2, 5)_{(u_2, u_4)}$-path.
Then $(u u_1, u v)\overset{\divideontimes}\to (5, T_3)$\footnote{Or, contradiction to \ref{315234}. }.

{\bf Case 2.} mult$_{S_u}(1) + $ mult$_{S_u}(3) = 3$.
Then mult$_{S_u}(i) = 2$, $i = 2, 4, 6$, mult$_{S_u}(j) = 1$, $j = 5, 7$,
and $t_0 = 0$, i.e., $T_3\cap C(u_2) = \emptyset$, $T_1\cap C(u_4) = \emptyset$.
If $7\not\in C(u_4)\cup C(u_3)$ and $H$ contains no $(7, 1/2)_{(u_4, u)}$-path,
then $(u u_4, u u_3, u v)\overset{\divideontimes}\to (7, T_3, T_1)$.
Otherwise, $7\in C(u_3)\cup C(u_4)$, or $H$ contains a $(7, 1/2)_{(u_4, u)}$-path.
We consider the following three scenarios due to $5\in C(u_i)$, $i = 1, 2, 4$, respectively.

{Case 2.1.} $5\in C(u_2)$.
It follows from \ref{135Cu42T34B1B3}(1) that $2\in C(u_4)$, $H$ contains a $(2, 5)_{(u_2, u_4)}$-path,
from \ref{315234} that $3\in C(u_1)$,
and from \ref{131notCu3356Cu1} that $1\in C(u_4)$ or $H$ contains a $(1, 4)_{(u_3, u_4)}$-path.
If $3\not\in C(u_2)\cup C(u_4)$, then $(u u_1, u u_3, u u_4)\to (5, T_3, 3)$ reduces the proof to (2.4.2.).
Otherwise, $3\in C(u_2)\cup C(u_4)$ and mult$_{S_u}(1) = 1$.
Recall that $1\in C(u_2)$ or $H$ contains a $(1, 3/4)_{(u_2, u)}$-path.
It follows that $1\in C(u_3)$ and $3\in C(u_2)$.
One sees that $(u u_1, u u_2, u u_3, u u_4)\to (T_1, 1, 2, 3)$ reduces to acyclic edge $(\Delta + 5)$-coloring for $H$,
then by \ref{1331recoloring}, we are done.

{Case 2.2.} $5\in C(u_4)$.
It follows from \ref{315234} and \ref{135notCu1}(1) that $H$ contains a $(5, \alpha)_{(u_1, u_\alpha)}$-path and a $(5, \beta)_{(u_2, u_\beta)}$-path, where $\{\alpha, \beta\} = \{3, 4\}$, and $\alpha\in C(u_1)$, $\beta\in C(u_2)$.

(2.2.1.) $\alpha = 4$, $\beta = 3$.
It follows from \ref{131Cu4232T1}(2) that $1\in C(u_4)$ or $H$ contains a $(1, 2/3)_{(u_4, u)}$-path.
If $1\not\in C(u_3)$ and $H$ contains no $(1, 4/5)_{(u_3, u)}$-path,
then $(u u_1, u u_2, u u_3, u v)\overset{\divideontimes}\to (T_1, 5, 1, T_3)$.
Otherwise, $1\in C(u_3)$ or $H$ contains a $(1, 4/5)_{(u_3, u)}$-path.

When $1\not\in C(u_3)\cup C(u_4)$,
it follows that $2\in C(u_4)$, $H$ contains a $(1, 2)_{(u_2, u_4)}$-path and a $(1, 5)_{(u_2, v_1)}$-path,
then $(u u_1, u u_2, u u_3, u u_4, u v)\overset{\divideontimes}\to (T_1, 5, 2, 1, T_3)$.
In the other case, $1\in C(u_3)\cup C(u_4)$.

One sees clearly that $3\not\in C(u_3)\cup C(u_4)$.
If $2\not\in C(u_3)$, then $(u u_1, u u_2, u u_3, u u_4)\to (T_1, 5, 2, 3)$ reduces the proof to (2.4.2.).
Otherwise, $2\in C(u_3)$ and then $1\in C(u_4)$, $3\in C(u_1)$.
Thus, $1\not\in C(u_2)\cup C(u_3)$ and $H$ contains a $(1, 4)_{(u_2, u_4)}$-path and a $(1, 4)_{(u_3, u_4)}$-path, a contradiction.

(2.2.2.) $\alpha = 3$, $\beta = 4$.
It follows from \ref{131notCu3356Cu1} that $1\in C(u_3)$ or $H$ contains a $(1, 2/4)_{(u_3, u)}$-path.
If $1\not\in C(u_4)$, and $H$ contains no $(1, 3/5)_{(u_4, u)}$-path,
then $(u u_1, u u_2, u u_4, u v)\overset{\divideontimes}\to (T_1, 5, 1, T_3)$.
Otherwise, $1\in C(u_4)$ or $H$ contains a $(1, 3/5)_{(u_4, u)}$-path.

When $1\not\in C(u_3)\cup C(u_4)$,
it follows that $2\in C(u_3)$, $H$ contains a $(1, 2)_{(u_2, u_3)}$-path and a $(1, 5)_{(u_2, u_4)}$-path,
then $(u u_1, u u_2, u u_3, u u_4, u v)\overset{\divideontimes}\to (T_1, 5, 1, 2, T_3)$.
In the other case, $1\in C(u_3)\cup C(u_4)$.
One sees clearly that $3\not\in C(u_3)\cup C(u_4)$.
If $H$ contains no $(2, 3)_{(u_2, u_4)}$-path,
then $(u u_1, u u_3, u u_4)\to (5, T_3, 3)$ reduces the reduces the proof to (2.4.2.).
Otherwise, $2\in C(u_4)$, $3\in C(u_2)$, and  mult$_{S_u}(1) = 1$.
Hence, $1\in C(u_4)$.
Thus, $H$ contains a $(1, 4)_{(u_2, u_4)}$-path and a $(1, 4)_{(u_3, u_4)}$-path, a contradiction.

{Case 2.3.} $5\in C(u_1)$. One sees that $1\in C(u_4)$ and $H$ contains a $(1, 5)_{(u_1, u_4)}$-path.
When $3\not\in C(u_1)\cup C(u_4)$,
one sees that $2\in C(u_4)$ and $H$ contains a $(2, 3)_{(u_2, u_4)}$-path,
then $(u u_1, u u_3, u u_4)\to (3, T_3, 5)$ reduces the proof to (2.4.1.).
In the other case, $3\in C(u_1)\cup C(u_4)$.

(2.3.1.) $3\in C(u_4)$.
Then $1\not\in C(u_3)$.
If $1\not\in C(u_2)$, one sees that $4\in C(u_2)$,
then $u u_3\to \alpha$,
and $(u u_1, u u_2, u u_4, u v)\overset{\divideontimes}\to (4, 1, 5, T_3\setminus \{\alpha\})$.
Otherwise, $1\in C(u_2)$ and then $3\not\in C(u_1)\cup C(u_2)$.
If $2\not\in C(u_4)$, then $(u u_1, u u_2, u u_3, u u_4, u v)\overset{\divideontimes}\to (T_1, 3, 1, 2, T_3)$.
Otherwise, $2\in C(u_4)$ and $(u u_1, u u_2, u u_3, u u_4, u v)\overset{\divideontimes}\to (T_1, 3, 2, 5, T_3)$.

(2.3.2.) $3\not\in C(u_4)$ and $3\in C(u_1)$.
When $1\not\in C(u_3)\cup C(u_4)$,
one sees that $4\in C(u_2)$, then $u u_3\to \alpha$,
and $(u u_1, u u_2, u u_4, u v)\overset{\divideontimes}\to (4, 1, 5, T_3\setminus \{\alpha\})$.
In the other case, $1\in C(u_3)\cup C(u_4)$ and  mult$_{S_u}(3) = 1$.
If $1\in C(u_3)$, one sees that $2\in C(u_4)$, $2\in B_1$ and $R_1\ne \emptyset$,
then $(u u_1, u u_2, u u_3, u u_4, u v)\overset{\divideontimes}\\to (4, 1, 2, 3, T_1\cap R_1)$.
Otherwise, $1\not\in C(u_3)$ and $1\in C(u_2)$.

If $2\in C(u_3)$, one sees that $2\not\in C(u_4)$, $4, 7\not\in C(u_1)$, and then $4\in B_3$,
it follows that $7\in C(u_4)$,
then $(u u_1, u u_2, u u_4, u v)\overset{\divideontimes}\to (4, 7, 2, T_3)$.
Otherwise, $2\in C(u_4)\setminus C(u_3)$ and then $2\in B_1$, $(T_1\cup C_{13})\cap R_1\ne \emptyset$.
If $H$ contains no $(7, j)_{(v_1, v_7)}$-path for some $j\in (T_1\cup C_{13})\cap R_1$,
then $(v v_1, v u_3, u v)\overset{\divideontimes}\to (j, 1, T_3)$.
Otherwise, it follows from Lemma~\ref{auvge2}(1.2) that $7\in C(v_1)$,
and $H$ contains a $(7, (T_1\cup C_{13}))_{(v_1, v_7)}$-path.
One sees that $H$ contains no $(j, i)_{(u_4, v)}$-path for each $j\in (T_1\cup C_{13})\cap R_1$.
It follows from Corollary~\ref{auv2}(4) that $6\in B_1\cup B_3$.
One sees that if $C(v_3) = \{3, 6\}\cup T_{u v}$, then $(v v_3, u v)\overset{\divideontimes}\to (4, T_1)$,
and that if $4\in C(v_3)$, then from Corollary~\ref{auv2d(u)5}(6), $(T_3\cup C_{13})\cap R_3\ne \emptyset$.
Note that $\{4, 6\}\setminus C(u_1)\ne \emptyset$ and $\{4, 6\}\cap B_3\ne \emptyset$.
Hence, $(T_3\cup C_{13})\cap R_3\ne \emptyset$ and it follows from Lemma~\ref{auvge2}(1.2) that $7\in C(v_3)$, and $H$ contains a $(7, (T_3\cup C_{13})\cap R_3)_{(v_3, v_7)}$-path.
One sees that if $2\in C(v_3)$, since $2, 3, 4/6, 7\in C(v_3)$ and $c_{13}\le 1$,
then $R_3\ge 2$, and $R_3\cap T_3\ne \emptyset$.
\begin{itemize}
\parskip=0pt
\item
     $7\in C(u_1)$. Then $(u u_3, u v)\overset{\divideontimes}\to (7, T_1\cap R_1)$.
\item
    $7\in C(u_4)$.
    Then $(u u_1, u u_2, u u_3, u u_4)\overset{\divideontimes}\to (T_1, 7, 1, 3, (2\cup T_3)\setminus C(v_3))$.
\item
    $7\in C(u_3)$. One sees that $C_{13} = \emptyset$ and $R_1\subseteq T_1$.
    Then $(u u_1, u u_3, u u_4, u v)\overset{\divideontimes}\to (7, 1, 3, T_1\cap R_1)$.
\item
    $7\in C(u_2)$.
    Then $H$ contains a $(2, 7)_{(u_2, u_4)}$-path.
    One sees that $u u_1\to 7$ or $u u_3\to 7$ reduces to an acyclic edge $(\Delta + 5)$-coloring for $H$.
    Hence, $T_1\cup T_3\subseteq C(v_7)$ and then $T_1\subseteq C(v_1)$, $T_3\subseteq C(v_3)$.
    Thus, $H$ contains a $(7, C_{13})_{(v_1, v_7)}$-path and a $(7, C_{13})_{(v_3, v_7)}$-path, a contradiction.
\end{itemize}

This finishes the inductive step for the case (2.4.3.).

{Case 2.5.} $\{i_1, i_2\} = \{1, 2\}$ and assume w.l.o.g. $c(v u_3) = 5$.
Assume that $5\not\in C(u_4)$.
Let $u u_4\to 5$.
If $5\not\in C(u_1)$, $H$ contains no $(5, 2)_{(u_2, u_4)}$-path,
and $H$ contains no $(3, j)_{(u_1, u_3)}$-path for some $j\in T_1$,
then $u u_1\to j$.
If $H$ contains no $(5, 1)_{(u_4, u_1)}$-path, $4\not\in C(u_2)$,
and $H$ contains no $(4, i)_{(u_2, u_i)}$-path for each $i\in \{1, 3\}$,
then $u u_2\to 4$.
If $H$ contains no $(5, 2)_{(u_4, u_2)}$-path, $4\not\in C(u_1)$,
and $H$ contains no $(4, i)_{(u_1, u_i)}$-path for each $i\in [2, 3]$,
then $u u_1\to 4$.
Assume that $H$ contains no $(5, i)_{(u_4, u_i)}$-path for each $i\in [1, 2]$.
If $4\not\in B_1\cup B_2$, then $u v\overset{\divideontimes}\to 4$.
If for some $i\in [1, 2]$, $4\not\in C(u_i)$ and $H$ contains no $(3, 4)_{(u_i, u_3)}$-path,
then $uu_i\to 4$.
These reduce the proof to Case (2.2.1.) or Case (2.4.1.).
Hence, we proceed with the following proposition, or otherwise we are done:
\begin{proposition}
\label{125u4u1oru4u2}
Assume that $5\not\in C(u_4)$. Then
\begin{itemize}
\parskip=0pt
\item[{\rm (1)}]
	if $5\not\in C(u_1)$, $H$ contains no $(5, 2)_{(u_2, u_4)}$-path,
    then $3\in C(u_1)$, $H$ contains a $(3, T_1)_{(u_1, u_3)}$-path;
\item[{\rm (2)}]
	if $5\not\in C(u_2)$, $H$ contains no $(5, 1)_{(u_1, u_4)}$-path,
    then $3\in C(u_2)$, $H$ contains a $(3, T_2)_{(u_2, u_3)}$-path;
\item[{\rm (3)}]
	if $H$ contains no $(5, 1)_{(u_4, u_1)}$-path,
    then $4\in C(u_2)$ or $H$ contains a $(4, 1/3)_{(u_2, u)}$-path;
\item[{\rm (4)}]
	if $H$ contains no $(5, 2)_{(u_4, u_2)}$-path,
    then $4\in C(u_1)$ or $H$ contains a $(4, 2/3)_{(u_1, u)}$-path;
\item[{\rm (5)}]
	if $H$ contains no $(5, i)_{(u_4, u_i)}$-path for each $i\in [1, 2]$,
    then $4\in B_1\cup B_2$, mult$_{S_u}(4)\ge 2$, and for each $i\in[1, 2]$, $4\in C(u_i)$,
    or $H$ contains a $(3, 4)_{(u_i, u_3)}$-path, $3\in C(u_i)$, $4\in C(u_3)$.
\end{itemize}
\end{proposition}

For each $j\in [6, 7]\setminus C(u_4)$, let $uu_4\to j$.
If $j\not\in C(u_1)$, $H$ contains no $(j, i)_{(u_4, u_i)}$-path for each $i\in [2, 3]$,
and $H$ contains no $(3, l)_{(u_1, u_3)}$-path for some $l\in T_1$,
then $u u_1\to l$.
If $H$ contains no $(j, i)_{(u_4, u_i)}$-path for each $i\in [2, 3]$,
$4\in C(u_1)$ and $H$ contains no $(4, l)_{(u_1, u_l)}$-path for each $l\in [2, 3]$,
then $u u_1\to 4$.
These reduce the proof to Case (2.2.2.) or Case (2.4.3.).
Hence, assume the following proposition, or otherwise we are done:
\begin{proposition}
\label{126notinu4}
For each $j\in [6, 7]\setminus C(u_4)$, then
\begin{itemize}
\parskip=0pt
\item[{\rm (1)}]
	if $j\not\in C(u_1)$ and $H$ contains no $(j, i)_{(u_i, u_4)}$-path for each $i\in [2, 3]$,
    then $3\in C(u_1)$ and $H$ contains a $(3, T_1)_{(u_1, u_3)}$-path;
\item[{\rm (2)}]
	if $j\not\in C(u_2)$ and $H$ contains no $(j, i)_{(u_i, u_4)}$-path for each $i\in \{1, 3\}$,
    then $3\in C(u_2)$ and $H$ contains a $(3, T_2)_{(u_2, u_3)}$-path;
\end{itemize}
\end{proposition}

For each $j\in [6, 7]\setminus C(u_3)$, let $u u_3\to j$.
If $j\not\in C(u_1)$, $H$ contains no $(j, i)_{(u_3, u_i)}$-path for each $i\in \{2, 4\}$,
and $H$ contains no $(4, l)_{(u_1, u_4)}$-path for some $l\in T_1$,
then $u u_1\to l$ reduces the proof to Case (2.4.).
Hence, we proceed with the following proposition, or otherwise we are done:
\begin{proposition}
\label{126notinu3}
For each $j\in [6, 7]\setminus C(u_3)$, then
\begin{itemize}
\parskip=0pt
\item[{\rm (1)}]
	if $j\not\in C(u_1)$ and $H$ contains no $(j, i)_{(u_i, u_3)}$-path for each $i\in \{2, 4\}$,
    then $4\in C(u_1)$ and $H$ contains a $(4, T_1)_{(u_1, u_4)}$-path;
\item[{\rm (2)}]
	if $j\not\in C(u_2)$ and $H$ contains no $(j, i)_{(u_i, u_3)}$-path for each $i\in \{1, 4\}$,
    then $4\in C(u_2)$ and $H$ contains a $(4, T_2)_{(u_2, u_4)}$-path;
\end{itemize}
\end{proposition}

Assume that $j\in T_4\cup [6, 7]\cup C(u_4)$ and $H$ contains no $(j, 3)_{(u_4, u_3)}$-path.
If $H$ contains no $(j, 2)_{(u_4, u_2)}$-path,
$4\in C(u_1)$ and $H$ contains no $(4, l)_{(u_1, u_l)}$-path for each $l\in [2, 3]$,
then $u u_1\to 4$ reduces the proof to Case (2.2.2.) or Case (2.4.3.),
or reduces to an acyclic edge $(\Delta + 5)$-coloring for $H$ such that $|C(u)\cap C(v)| = 1$.
If $H$ contains no $(j, i)_{(u_4, u_i)}$-path for each $i\in [1, 2]$,
then if $4\not\in B_1\cup B_2$, $uv\overset{\divideontimes}\to 4$;
and if for some $i\in [1, 2]$, $4\not\in C(u_i)$ and $H$ contains no $(3, 4)_{(u_i, u_3)}$-path,
then $uu_i\to 4$ and we are done.
Hence, assume the following proposition, or otherwise we are done:
\begin{proposition}
\label{12T4123}
Assume $j\in (T_4\cup [6, 7])\setminus C(u_4)$, $H$ contains no $(j, 3)_{(u_4, u_3)}$-path. Then
\begin{itemize}
\parskip=0pt
\item[{\rm (1)}]
	if $H$ contains no $(j, 2)_{(u_4, u_2)}$-path,
    then $4\in C(u_1)$ or $H$ contains a $(4, 2/3)_{(u_1, u)}$-path;
\item[{\rm (2)}]
	if $H$ contains no $(j, 1)_{(u_4, u_1)}$-path,
    then $4\in C(u_2)$ or $H$ contains a $(4, 1/3)_{(u_2, u)}$-path.
\item[{\rm (3)}]
    if $H$ contains no $(j, i)_{(u_4, u_i)}$-path for each $i\in [1, 2]$,
    then $4\in B_1\cup B_2$, mult$_{S_u}(4)\ge 2$,
    and for each $i\in[1, 2]$, $4\in C(u_i)$,
    or $H$ contains a $(3, 4)_{(u_i, u_3)}$-path, $3\in C(u_i)$, $4\in C(u_3)$.
\end{itemize}
\end{proposition}

Assume $j\in T_3\cup [6, 7]\cup C(u_3)$ and $H$ contains no $(j, 4)_{(u_4, u_3)}$-path.
If $H$ contains no $(j, 2)_{(u_3, u_2)}$-path,
$3\in C(u_1)$ and $H$ contains no $(3, l)_{(u_1, u_l)}$-path for each $l\in \{2, 4\}$,
then $u u_1\to 3$ reduces the proof to Case (2.4.2.) or Case (2.4.3.),
or reduces to an acyclic edge $(\Delta + 5)$-coloring for $H$ such that $|C(u)\cap C(v)| = 1$.
If $H$ contains no $(j, i)_{(u_3, u_i)}$-path for each $i\in [1, 2]$,
then if $3\not\in B_1\cup B_2$, $u v\overset{\divideontimes}\to 3$;
and if for some $i\in [1, 2]$, $3\not\in C(u_i)$ and $H$ contains no $(3, 4)_{(u_i, u_4)}$-path,
then $u u_i\to 3$ and we are done.
Hence, we proceed with the following proposition, or otherwise we are done:
\begin{proposition}
\label{12T3124}
Assume $j\in (T_3\cup [6, 7])\setminus C(u_3)$, $H$ contains no $(j, 4)_{(u_4, u_3)}$-path. Then
\begin{itemize}
\parskip=0pt
\item[{\rm (1)}]
	if $H$ contains no $(j, 2)_{(u_3, u_2)}$-path,
    then $3\in C(u_1)$ or $H$ contains a $(3, 2/4)_{(u_1, u)}$-path;
\item[{\rm (2)}]
	if $H$ contains no $(j, 1)_{(u_3, u_1)}$-path,
    then $3\in C(u_2)$ or $H$ contains a $(3, 1/4)_{(u_2, u)}$-path.
\item[{\rm (3)}]
    if $H$ contains no $(j, i)_{(u_3, u_i)}$-path for each $i\in [1, 2]$,
    then $3\in B_1\cup B_2$, mult$_{S_u}(3)\ge 2$,
    and for each $i\in[1, 2]$, $3\in C(u_i)$,
    or $H$ contains a $(3, 4)_{(u_i, u_4)}$-path, $4\in C(u_i)$, $3\in C(u_4)$.
\end{itemize}
\end{proposition}

By \ref{12T4123}(1--2),
we proceed with the following proposition, or otherwise we are done:
\begin{proposition}
\label{124notinu2u1}
\begin{itemize}
\parskip=0pt
\item[{\rm (1)}]
	If $4\not\in C(u_2)$ and $H$ contains no $(4, l)_{(u_2, u_l)}$-path for each $l\in \{1, 3\}$,
    then $6, 7\in C(u_3)\cup C(u_4)$,
    or $H$ contains a $(1, [6, 7]\setminus (C(u_3)\cup C(u_4)))_{(u_1, u_4)}$-path, $1\in C(u_4)$;
\item[{\rm (2)}]
	If $4\not\in C(u_1)$ and $H$ contains no $(4, l)_{(u_1, u_l)}$-path for each $l\in \{2, 3\}$,
    then $6, 7\in C(u_3)\cup C(u_4)$,
    or $H$ contains a $(2, [6, 7]\setminus (C(u_3)\cup C(u_4)))_{(u_2, u_4)}$-path, $2\in C(u_4)$.
\end{itemize}
\end{proposition}

By \ref{12T3124}(1--2),
we proceed with the following proposition, or otherwise we are done:
\begin{proposition}
\label{123notinu2u1}
\begin{itemize}
\parskip=0pt
\item[{\rm (1)}]
	If $3\not\in C(u_2)$ and $H$ contains no $(3, l)_{(u_2, u_l)}$-path for each $l\in \{1, 4\}$,
    then $6, 7\in C(u_3)\cup C(u_4)$,
    or $H$ contains a $(1, [6, 7]\setminus (C(u_3)\cup C(u_4)))_{(u_1, u_3)}$-path, $1\in C(u_3)$;
\item[{\rm (2)}]
	If $3\not\in C(u_1)$ and $H$ contains no $(3, l)_{(u_1, u_l)}$-path for each $l\in \{2, 4\}$,
    then $6, 7\in C(u_3)\cup C(u_4)$,
    or $H$ contains a $(2, [6, 7]\setminus (C(u_3)\cup C(u_4)))_{(u_2, u_3)}$-path, $2\in C(u_3)$.
\end{itemize}
\end{proposition}

If $2\not\in C(u_3)$, $H$ contains no $(1, 2)_{(u_1, u_3)}$-path,
and $3\not\in C(u_2)$, $H$ contains no $(3, i)_{(u_2, u_i)}$-path for each $i\in \{1, 4\}$,
then $(u u_3, u u_2)\to (2, 3)$ reduces the proof to Case (2.4.).
If $1\not\in C(u_3)$, $H$ contains no $(1, i)_{(u_3, u_i)}$-path for each $i\in \{2, 4\}$,
and $3\not\in C(u_1)$, $H$ contains no $(3, i)_{(u_1, u_i)}$-path for each $i\in \{2, 4\}$,
then $(u u_3, u u_1)\to (1, 3)$ reduces the proof to Case (2.4.).
Hence, we proceed with the following proposition, or otherwise we are done:
\begin{proposition}
\label{1212notu3}
\begin{itemize}
\parskip=0pt
\item[{\rm (1)}]
	If $1\not\in C(u_3)$, $H$ contains no $(1, i)_{(u_3, u_i)}$-path for each $i\in \{2, 4\}$,
    then $3\in C(u_1)$, or $H$ contains a $(3, i)_{(u_1, u_i)}$-path for some $i\in \{2, 4\}$.
\item[{\rm (2)}]
	If $2\not\in C(u_3)$, $H$ contains no $(1, 2)_{(u_1, u_3)}$-path,
    then $3\in C(u_2)$, or $H$ contains a $(3, i)_{(u_2, u_i)}$-path for some $i\in \{1, 4\}$.
\end{itemize}
\end{proposition}

If $1\not\in C(u_4)$, $H$ contains no $(1, i)_{(u_4, u_i)}$-path for each $i\in [2, 3]$,
and $4\not\in C(u_1)$, $H$ contains no $(4, i)_{(u_1, u_i)}$-path for each $i\in [2, 3]$,
then $(u u_4, u u_1)\to (1, 4)$.
If $2\not\in C(u_4)$, $H$ contains no $(2, i)_{(u_4, u_i)}$-path for each $i\in \{1, 3\}$,
and $4\not\in C(u_2)$, $H$ contains no $(4, i)_{(u_1, u_i)}$-path for each $i\in \{1, 3\}$,
then $(u u_4, u u_2)\to (2, 4)$.
These reduce the proof to Case (2.2.2.) or Case (2.4.3.).
Hence, we proceed with the following proposition, or otherwise we are done:
\begin{proposition}
\label{1212notu4}
\begin{itemize}
\parskip=0pt
\item[{\rm (1)}]
	If $1\not\in C(u_4)$, $H$ contains no $(1, i)_{(u_4, u_i)}$-path for each $i\in [2, 3]$,
    then $4\in C(u_1)$, or $H$ contains a $(4, i)_{(u_1, u_i)}$-path for some $i\in [2, 3]$.
\item[{\rm (2)}]
	If $2\not\in C(u_4)$, $H$ contains no $(2, i)_{(u_4, u_i)}$-path for each $i\in \{1, 3\}$,
    then $4\in C(u_2)$, or $H$ contains a $(4, i)_{(u_2, u_i)}$-path for some $i\in \{1, 3\}$.
\end{itemize}
\end{proposition}

If $3\not\in C(v_1)$ and $H$ contains no $(3, i)_{(v_1, v_i)}$-path for each $i\in C(v)$,
then $v v_1\to 3$ reduces the proof to Case (2.4.).
Hence, we proceed with the following proposition:
\begin{proposition}
\label{123v13v2}
\begin{itemize}
\parskip=0pt
\item[{\rm (1)}]
	$3\in C(v_1)$, or $H$ contains a $(3, 2/5/6/7)_{(v_1, v)}$-path.
\item[{\rm (2)}]
	$3\in C(v_2)$, or $H$ contains a $(3, 1/5/6/7)_{(v_2, v)}$-path.
\end{itemize}
\end{proposition}

{(2.5.1)} $c(vu_4) = 2$.
Let $C(u_1) = W_1\cup T_2\cup C_{12}$, $C(u_2) = W_2\cup T_1\cup C_{12}$,
where $T_{u v} = T_1\cup T_2\cup C_{12}$.
Then $T_1\subseteq B_2\subseteq C(u_2)\cap C(u_4)$, $T_2\subseteq B_1\subseteq C(u_1)$.

By \ref{123v13v2}(2), we proceed with the following proposition:
\begin{proposition}
\label{12T43167neemptyset}
\begin{itemize}
\parskip=0pt
\item[{\rm (1)}]
	$3\in C(u_4)$, or $H$ contains a $(3, i)_{(u_4, v_i)}$-path for some $i\in \{1, 6, 7\}$.
\item[{\rm (2)}]
	$3/1/6/7\in C(u_4)$, $T_4\ne \emptyset$, and there exists a $b_2\in [5, 7]\cap C(u_4)$.
    such that $H$ contains a $(b_2, j)_{(u_4, v)}$-path for some $j\in T_4$.
\end{itemize}
\end{proposition}

Assume that $5\not\in C(u_1)$ and $H$ contains no $(5, 4)_{(u_1, u_4)}$-path.
If $H$ contains no $(5, 2)_{(u_2, u_1)}$-path,
then $u u_1\to 5$ reduces the proof to the Case 2.3.
If $1\not\in C(u_2)$, $H$ contains no $(1, i)_{(u_2, u_i)}$-path for each $i\in \{3, 4\}$,
and $H$ contains no $(5, j)_{(u_1, u_3)}$-path for some $j\in B_1$,
then $(u u_1, u u_2, u v)\to (5, 1, j)$.
Hence, we proceed with the following proposition, or otherwise we are done:
\begin{proposition}
\label{125notinu1}
Assume that $5\not\in C(u_1)$ and $H$ contains no $(5, 4)_{(u_1, u_4)}$-path.
\begin{itemize}
\parskip=0pt
\item[{\rm (1)}]
	$H$ contains a $(5, 2)_{(u_1, u_2)}$-path.
\item[{\rm (2)}]
	If $1\not\in C(u_2)$, $H$ contains no $(1, i)_{(u_2, u_i)}$-path for each $i\in \{3, 4\}$,
    then $H$ contains a $(5, B_1)_{(u_1, u_3)}$-path, a $(5, T_4)_{(u_1, u_3)}$-path, and $T_4\subseteq C(u_3)$
\end{itemize}
\end{proposition}

Assume that $2\not\in C(u_1)$ and $H$ contains no $(2, 3)_{(u_1, u_3)}$-path.
Let $u u_1\to 2$.
If $5\not\in C(u_2)$ and $H$ contains no $(4, 5)_{(u_2, u_4)}$-path,
then $u u_2\to 5$ reduces the proof to the Case 2.3.
Together with Corollary~\ref{auv2d(u)5}(2.1),
we proceed with the following proposition, or otherwise we are done:
\begin{proposition}
\label{122notinu12323}
Assume that $2\not\in C(u_1)$ and $H$ contains no $(2, 3)_{(u_1, u_3)}$-path.
\begin{itemize}
\parskip=0pt
\item[{\rm (1)}]
	$5\in C(u_2)$ or $H$ contains a $(4, 5)_{(u_2, u_4)}$-path.
\item[{\rm (2)}]
	$1\in C(u_2)$ or $H$ contains a $(1, i)_{(u_2, u_i)}$-path for some $i\in [3, 4]$.
\end{itemize}
\end{proposition}

Assume that $5\not\in C(u_2)$ and $H$ contains no $(5, i)_{(u_2, u_i)}$-path for each $i\in \{1, 4\}$.
One sees that the same argument applies if $uu_2\to 5$.
Hence, we proceed with the following proposition:
\begin{proposition}
\label{125notu2514}
If $5\not\in C(u_2)$ and $H$ contains no $(5, i)_{(u_2, u_i)}$-path for each $i\in \{1, 4\}$, then
\begin{itemize}
\parskip=0pt
\item[{\rm (1)}]
	$H$ contains a $(5, T_{uv}\setminus B_1)_{(u_2, u_3)}$-path, a $(5, T_1)_{(u_2, u_3)}$-path,
    and $T_1\subseteq T_{u v}\setminus B_1\subseteq C(u_3)$;
\item[{\rm (2)}]
	$T_3\ne \emptyset$ and there exists a $b_5\in \{2, 6, 7\}\cap C(u_3)$ such that $H$ contains a $(b_5, j)_{(u_3, v)}$-path for some $j\in T_3$.
\end{itemize}
\end{proposition}

If $H$ contains no $(4, i)_{(u_4, v)}$-path for each $i\in \{1\}\cup [5, 7]$,
and $H$ contains no $(j, 3)_{(u_4, u_3)}$-path for some $j\in T_4$,
then $(v u_4, u u_4, u v)\overset{\divideontimes}\to (4, j, T_1)$.
Hence, we proceed with the following proposition:
\begin{proposition}
\label{124i4i3j3j}
\begin{itemize}
\parskip=0pt
\item[{\rm (1)}]
	If $H$ contains no $(4, i)_{(u_4, v)}$-path for each $i\in \{1\}\cup [5, 7]$,
    then $H$ contains a $(3, T_4)_{(u_4, u_3)}$-path.
\item[{\rm (2)}]
	If $3\not\in C(u_4)$, or $T_2\setminus C(u_4)\ne \emptyset$, or $T_4\setminus C(u_3)\ne \emptyset$,
    then $H$ contains a $(4, i)_{(u_4, v)}$-path for some $i\in \{1\}\cup [5, 7]$.
\end{itemize}
\end{proposition}

One sees that $t_1 = \Delta - 2 - (d(u_1) - w_1) = \Delta - d(u_1) + w_1 - 2\ge w_1 - 2$,
and $t_2 = \Delta - 2 - (d(u_2) - w_2) = \Delta - d(u_2) + w_2 - 2 \ge w_2 - 2$.

\begin{lemma}
\label{12T2ge1}
$w_2\ge 3$ and $t_2 = \Delta - d(u_2) + w_2 - 2 \ge w_2 - 2 \ge 1$.
\end{lemma}
\begin{proof}
Assume that $w_2\le 2$.
One sees from \ref{125notinu1}(1) that $5/2/4\in C(u_1)$,
and if $5\not\in C(u_1)$ and $H$ contains no $(5, 2)_{(u_1, u_2)}$-path,
then $H$ contains a $(5, 4)_{(u_1, u_4)}$-path.

Condition~(C1) states that $|W_3\cap [1, 5]| + |W_4\cap [1, 5]| + |W_1\cap [1, 5]|\ge 10$,
$6, 7\in C(u_1)\cup C(u_3)\cup C(u_4)$,
and $T_1\subseteq C(u_3)\cap C(u_4)$, $T_{uv}\subseteq C(u_3)\cup C(u_4)$.
Assume that (C1) holds.
Let $\kappa_{u_1} = |[6, 7]\cap C(u_1)|$.
One sees that $t_1\ge w_1 - 2\ge \kappa_{u_1} + |W_1\cap [1, 5]| - 2$.
Thus, $d(u_3) + d(u_4)\ge |W_3\cap [1, 5]| + |W_4\cap [1, 5]| + 2 - \kappa_{u_1} + t_{uv} + t_1\ge |W_3\cap [1, 5]| + |W_4\cap [1, 5]| + 2 - \kappa_{u_1} + t_{uv} + \kappa_{u_1} + |W_1\cap [1, 5]| - 2 \ge \Delta + 10 - 2 = \Delta + 8$,
and none of ($A_{8.1}$)--($A_{8.4}$) holds.
Hence, we proceed with the following proposition:
\begin{proposition}
\label{12Delta+8}
If (C1) holds, then we are done.
\end{proposition}

When $1, 3, 4, 5\not\in C(u_2)$,
it follows from \ref{125notinu1}(1) that $5/4\in C(u_1)$,
from \ref{122notinu12323} that $2/3\in C(u_1)$,
and from \ref{125notu2514}(1) that $T_1\subseteq C(u_3)$.
One sees that $T_4\subseteq B_1$.
It follows from \ref{12T4123}(3) that $3\in C(u_4)$,
$T_4\subseteq C(u_3)$, i.e., $T_{uv}\subseteq C(u_3)\cup C(u_4)$,
from \ref{125u4u1oru4u2}(3) that $5/1\in C(u_4)$,
from \ref{1212notu3}(2) that $2/1\in C(u_3)$,
and from \ref{123notinu2u1}(1) that $6, 7\in C(u_1)\cup C(u_3)\cup C(u_4)$.
One sees that (C1) holds.
Thus, by \ref{12Delta+8}, we are done.

In the other case, $5/1/3/4\in C(u_2)$.
We consider below the four cases, respectively.

{\bf Case 1.} $W_2 = \{2, 5\}$.
It follows from \ref{122notinu12323}(2) that $2/3\in C(u_1)$,
from \ref{125notinu1}(1) that $5/4/2\in C(u_1)$,
from \ref{12T3124} that $T_1\subseteq C(u_3)$,
and from \ref{12T4123}(3) that $3\in C(u_4)$,
$T_4\subseteq C(u_3)$, i.e., $T_{uv}\subseteq C(u_3)\cup C(u_4)$.
It follows from \ref{125u4u1oru4u2}(3) that $5/1\in C(u_4)$,
from \ref{1212notu3}(2) that $2/1\in C(u_3)$,
and from \ref{123notinu2u1}(1) that $6, 7\in C(u_1)\cup C(u_3)\cup C(u_4)$.
It follows from \ref{12Delta+8} that mult$_{C(u_3)}(1) + $mult$_{C(u_3)}(2)= 1$, $4\not\in C(u_3)$,
$[2, 5]\cap C(u_1) = \{2\}$,
and then from \ref{1212notu3}(1) that $1\in C(u_3)$.
Then $2\not\in C(u_3)$, and $(u u_1, u u_2, u u_3, u v)\overset{\divideontimes}\to (3, 1, 2, T_4)$.

{\bf Case 2.} $W_2 = \{2, 1\}$.
It follows from \ref{125notinu1}(1) that $5/4\in C(u_1)$,
from \ref{122notinu12323}(1) that $2/3\in C(u_1)$,
from \ref{12T3124} that $T_1\subseteq C(u_3)$,
and from \ref{12T4123}(3) that $3\in C(u_4)$,
$T_4\subseteq C(u_3)$, i.e., $T_{uv}\subseteq C(u_3)\cup C(u_4)$.
It follows from \ref{125u4u1oru4u2}(5) that $5/1\in C(u_4)$,
from \ref{12T4123}(3) that $6, 7\in C(u_1)\cup C(u_3)\cup C(u_4)$.
If $1/2\in C(u_3)$, then (C1) holds and by \ref{12Delta+8}, we are done.
Otherwise, $1, 2\not\in C(u_3)$.
It follows from \ref{1212notu3}(2) that $H$ contains a $(1, 3)_{(u_1, u_2)}$-path.

If there exists a $j\in [5, 7]\setminus C(u_1)$,
then $(u u_2, u u_3, u v)\overset{\divideontimes}\to (j, 2, 3)$.
Otherwise, $[5, 7]\subseteq C(u_1)$ and $t_1\ge 3$.
Since $d(u_3) + d(u_4)\ge  |\{3, 5\}| + |\{2, 4, 3, 5/1\}| + t_{uv} + t_1 = \Delta + 4 + t_1$,
we have $t_1 = 3$, $W_3 = \{3, 5\}$ and $W_1 = \{1, 3\}\cup [5, 7]$.
Then by \ref{122notinu12323}(1), we are done.


{\bf Case 3.} $W_2 = \{2, 3\}$.
It follows from \ref{125notinu1}(1) that $5/4\in C(u_1)$,
from \ref{122notinu12323}(1) that $2/3\in C(u_1)$,
from \ref{12T43167neemptyset}(2) that $5/6/7\in C(u_4)$,
and from \ref{125notu2514}(1) that $T_1\subseteq T_{uv}\setminus B_1\subseteq C(u_3)$,
from \ref{125notu2514}(2) that $2/6/7\in C(u_3)$.
We distinguish whether or not $H$ contains a $(3, 4)_{(u_2, u_3)}$-path.

{Case 3.1.} $H$ contains no $(3, 4)_{(u_2, u_3)}$-path.
It follows from \ref{12T4123}(3) that $3\in C(u_4)$,
$T_4\subseteq C(u_3)$, i.e., $T_{u v}\subseteq C(u_3)\cup C(u_4)$,
from \ref{125u4u1oru4u2}(5) that $5/1\in C(u_4)$,
from \ref{12T4123}(2) that $6, 7\in C(u_1)\cup C(u_3)\cup C(u_4)$,
and if $1\not\in C(u_4)$, then $6, 7\in C(u_3)\cup C(u_4)$.
One sees that $|\{3, 5\}| + |[2, 4]| + t_{uv} + t_1 = \Delta + 3 + t_1$.
If $1\not\in C(u_4)$, one sees that $5\in C(u_4)$, $6, 7\in C(u_3)\cup C(u_4)$,
and $d(u_3) + d(u_4)\ge \Delta + 3 + t_1 + |[5, 7]| = \Delta + 6 + t_1$,
then $t_1 = 1$, $1, 2\not\in C(u_3)$, $W_1 = \{1, 2, 4/5\}$,
and by Corollary~\ref{auv2d(u)5}(7.2), we are done.
Otherwise, $1\in C(u_4)$,
and $d(u_3) + d(u_4)\ge \Delta + 3 + t_1 + |\{2/6/7\}| + |\{1, 5/6/7\}| = \Delta + 6 + t_1$.
Then $t_1 = 1$, $6, 7\not\in C(u_1)$, $1, 2\not\in C(u_3)$, and $5\not\in C(u_4)$.
Hence, $W_1 = \{1, 2, 5\}$, a contradiction.

{Case 3.2.} $H$ contains a $(3, 4)_{(u_2, u_3)}$-path and $4\in C(u_3)$.
When there exists a $8\in T_{uv}\setminus C(u_3)\cup C(u_4)$,
then assume w.l.o.g. $H$ contains a $(6, 8)_{(u_3, v_6)}$-path, a $(7, 8)_{(u_4, v_7)}$-path,
and $6\in C(u_3)$, $7\in C(u_4)$.
It follows from \ref{12T4123}(3) that $4\in B_1$.
If $3\not\in B_1$, then $(u u_3, u u_2, u v)\overset{\divideontimes}\to (8, 5, 3)$.
Otherwise, $3\in B_1$ and then $3, 4\in C(u_1)$.
If $7\not\in C(u_1)\cup C(u_3)$,
then $(u u_2, u u_3, u u_4)\to (4, 7, 8)$ reduces the proof to (2.4.3.).
Otherwise, $7\in C(u_1)\cup C(u_3)$.
Since $[3, 6]\subseteq C(u_3)$, we have $t_3\ge 2$ and assume $9\in T_3$.
One sees from $T_2\subseteq C(u_3)$ and $T_{u v}\setminus B_1\subseteq C(u_3)$ that $9\in C(u_2)\cap B_1$.
Assume that $6\not\in C(u_1)\cup C(u_4)$ and $9\not\in C(u_4)$.
Let $(u u_1, u u_2)\to (6, 5)$.
If $H$ contains no $(6, 9)_{(u_1, v_6)}$-path,
then $(u u_3, u v)\overset{\divideontimes}\to (8, 9)$.
Otherwise, $H$ contains a $(6, 9)_{(u_1, v_6)}$-path and $(u u_3, u v)\overset{\divideontimes}\to (9, 8)$.
Hence, $6\in C(u_1)\cup C(u_4)$ or $9\in C(u_1)\cap C(u_2)\cap C(u_4)$.
Then $\sum_{x\in N(u)\setminus \{v\}}d(x)\ge |\{1, 3, 4\}| + |[2, 3]| +  |\{3, 5, 4, 6\}| + |\{2, 4, 7\}| + |\{2, 5, 7, 6/9\}| + 2t_{u v} + t_0 \ge 2\Delta + 12 + t_1\ge 2\Delta + 12 + 1 = 2\Delta + 13$.
It follows that $t_1 = 1$, $1\not\in S_u$, and $3\not\in C(u_4)$.
Then $(u u_1, u u_2, u u_3, u v)\overset{\divideontimes}\to (2, 5, 1, 3)$.

In the other case, $T_{u v}\subseteq C(u_3)\cup C(u_4)$.
One sees that $d(u_3) + d(u_4)\ge |\{3, 5, 4, 2/6/7\}| + |\{2, 4, 5/6/7\}| + t_{u v} + t_1 = \Delta + 5 + t_1$.
It follows that $t_1\le 2$ and $1\not\in C(u_3)\cap C(u_4)$.
\begin{itemize}
\parskip=0pt
\item
     $1\in C(u_3)\cup C(u_4)$.
     Then $t_1 = 1$, $3\not\in C(u_4)$,  $\sum_{i\in \{2, 6, 7\}}$ mult$_{C(u_3)}(i) = 1$,
     $\sum_{i\in [5, 7]}$ mult$_{C(u_4)}(i) = 1$,
     and $W_1 = \{1, 2/3, 4/5\}$.
     When $1\in C(u_3)\setminus C(u_4)$,
     it follows from Corollary~\ref{auv2d(u)5}(7.2) that $3\in C(u_1)$ and then $2\in C(u_3)$.
     If $4\not\in C(u_1)$, then $(u u_3, u u_1)\to ([6, 7]\setminus C(u_4), T_1)$ reduces the proof to (2.4.2.).
     Otherwise, $W_1 = \{1, 3, 4\}$ and $5\in C(u_4)$.
     Then $v u_4\to 3$ reduces the proof to (2.4.1.).
     In the other case, $1\in C(u_4)\setminus C(u_3)$.
     If $H$ contains no $(3, j)_{(u_1, u_4)}$-path,
     then $(u u_1, u u_3, u u_4, u v)\overset{\divideontimes}\to (j, 1, 3, T_4)$.
     Otherwise, $H$ contains a $(3, T_1)_{(u_1, u_4)}$-path and then contains a $(4, T_1)_{(u_1, u_4)}$-path.
     Hence, $W_1 = \{1, 3, 4\}$, $2\in C(u_3)$, $5\in C(u_4)$, $6, 7\not\in C(u_3)\cup C(u_4)$,
     and from \ref{12T3124} that $3\in B_1$.
     Thus, $v u_4\to 3$ reduces the proof to (2.4.1.).
\item
    $1\not\in C(u_3)\cup C(u_4)$.
    It follows from Corollary~\ref{auv2d(u)5}(7.2) that $3, 4\in C(u_1)$.
    When $2\in C(u_1)$, then $t_1 = 2$, $5\in C(u_4)$, $W_4 = \{2, 4, 5\}$,
    and $v u_4\to 3$ reduces the proof to (2.4.1.).
    Hence, $2\not\in C(u_1)$.
    If $3\not\in C(u_4)$, then $(u u_1, u u_2, u u_3, u v)\overset{\divideontimes}\to (2, 5, 1, 3)$.
    Otherwise, $3\in C(u_4)$.
    Then $t_1 = 1$, $W_1 = \{1, 3, 4\}$.
    If $H$ contains no $(5, j)_{(u_2, u_3)}$-path for some $j\in T_4$,
    then $(u u_1, u u_2, u u_3, u v)\overset{\divideontimes}\to (2, 5, 1, j)$.
    Otherwise, $H$ contains a $(5, T_4)_{(u_2, u_3)}$-path.
    Then $(u u_1, u u_4, u v)\overset{\divideontimes}\to (5, 1, T_4)$.
\end{itemize}

{\bf Case 4.} $W_2 = \{2, 4\}$.
If $H$ contains no $(4, 5)_{(u_2, u_4)}$-path, then $u u_2\to 5$ reduces the proof to the above Case 3.
Otherwise, $H$ contains a $(4, 5)_{(u_2, u_4)}$-path, $5\in C(u_4)$,
and follows from \ref{125notinu1}(1) that $5\in C(u_1)$.
When $W_1 = \{1, 5\}$,
it follows from \ref{122notinu12323}(2) that $H$ contains a $(4, 1)_{(u_2, u_4)}$-path, $1\in C(u_4)$,
and from \ref{1212notu3}(1) that $1\in C(u_3)$.
It follows from \ref{12T4123}(3) that $3\in C(u_4)$, $T_4\subseteq C(u_3)$ and $6, 7\in C(u_3)\cup C(u_4)$.
Thus, $d(u_3) + d(u_4)\ge |\{1, 3, 5\}| + |[1, 5]| + |[6, 7]| + t_{u v} = \Delta + 8$.

In the other case, $W_1\ge 3$ and $t_1\ge 1$.
One sees that $H$ contains a $(3/4, T_1)_{(u_1, u)}$-path.

{Case 4.1.} $3\not\in C(u_1)$.
Then $4\in C(u_1)$ and follows from Corollary~\ref{auv2}(5.1) that $1, 6, 7\in S_u$,
and if $1\not\in C(u_4)$, then $H$ contains a $(1, 3)_{(u_3, u_4)}$-path,
from Corollary~\ref{auv2}(5.1) and (4) that $T_1\subseteq C(u_1)$.
\begin{itemize}
\parskip=0pt
\item
     There exists a $8\in T_{u v}\setminus (C(u_3)\cup C(u_4))$.
     It follows from \ref{12T3124}(3) that $2\in C(u_3)$,
     from \ref{12T3124}(2) that $H$ contains a $(4, 3)_{(u_2, u_4)}$-path, $3\in C(u_4)$,
     and from \ref{12T43167neemptyset}(2) that $H$ contains a $(6, 8)_{(u_4, v_6)}$-path, $6\in C(u_4)$.
     If $6\not\in C(u_1)\cup C(u_3)$, one sees from \ref{12T3124}(3) that $H$ contains a $(4, 6)_{(u_3, u_4)}$-path, then $(u u_1, u v)\overset{\divideontimes}\to (6, 8)$.
     Otherwise, $6\in C(u_1)\cup C(u_3)$.
     Then $\sum_{x\in N(u)\setminus \{v\}}d(x)\ge |\{1, 4, 5\}| + |\{2, 4\}| + |\{3, 5, 2\}| + |[2, 6]| + |\{1, 6, 7\}| + 2t_{u v} + t_0\ge 2\Delta 2\Delta + 12 + t_1\ge 2\Delta + 13$.
     It follows that $t_1 = 1$, $4\not\in C(u_3)$, $W_1 = \{1, 4, 5\}$,
     and from \ref{122notinu12323}(2) that $1\in C(u_4)\setminus C(u_3)$.
     Then $(u u_1, u u_3, u v)\overset{\divideontimes}\to (3, 1, 8)$.
\item
    $T_{u v}\subseteq C(u_3)\cup C(u_4)$.
    Let $\kappa_{u_1} = |[6, 7]\cap C(u_1)|$.
    One sees that $t_1\ge w_1 - 2\ge \kappa_{u_1} + 3 - 2 = \kappa_{u_1} + 1$.
    Then $d(u_3) + d(u_4)\ge |\{3, 5\}| + |\{2, 4, 5\}| + |\{1\}| + 2 - \kappa_{u_1} + t_{u v} + t_1\ge 8 - \kappa_{u_1} + t_{u v} + \kappa_{u_1} + 1 = \Delta + 7$.
    It follows that $3\not\in C(u_4)$ and then $1\in C(u_4)\setminus C(u_3)$.
    Then $(u u_1, u u_3, u u_4, u v)\overset{\divideontimes}\to (T_1, 1, 3, T_4)$.
\end{itemize}

{Case 4.2.} $3\in C(u_1)$.
When $4\not\in B_1$, it follows from \ref{12T4123}(3) that $3\in C(u_4)$, $T_4\subseteq C(u_3)$,
and $6, 7\in C(u_3)\cup C(u_4)\cup C(u_1)$.
Let $\kappa_{u_1} = |[6, 7]\cap C(u_1)|$.
One sees that $t_1\ge w_1 - 2\ge \kappa_{u_1} + 3 - 2 = \kappa_{u_1} + 1$.
Then $d(u_3) + d(u_4)\ge |\{3, 5\}| + |[2, 5]| + |\{1\}| + 2 - \kappa_{u_1} + t_{u v} + t_1\ge 9 - \kappa_{u_1} + t_{u v} + \kappa_{u_1} + 1 = \Delta + 8$.

In the other case, $4\in B_1\subseteq C(u_1)$ and then $t_1\ge 2$.
One sees that $T_4\subseteq C(u_1)\cap C(u_2)$.

(4.2.1.) $T_1\setminus C(u_3)\ne \emptyset$.
It follows from Corollary~\ref{auv2}(5.1) that $1, 6, 7\in S_u$,
from Corollary~\ref{auv2}(5.2) and (3) that $3\in B_1$ and $H$ contains a $(4, 3)_{(u_2, u_4)}$-path, $3\in C(u_4)$.
If $1\in C(u_4)$, ones sees that $d(u_4)\ge |[1, 5]| + t_1\ge 7$,
it follows that $t_1 = 2$, $W_4 = [1, 5]$, $W_1 = \{1, 3, 4, 5\}$,
and from Lemma~\ref{auvge2}(1.2) that $T_4\subseteq C(u_3)$,
then $C(u_3) = \{3, 5, 6, 7\}\cup T_4$,
and $(u u_1, u u_3, u u_4, u v)\overset{\divideontimes}\to (T_1, 4, 6, T_4)$.
Otherwise, $1\not\in C(u_4)$.
It follows from \ref{122notinu12323}(2) that $2\in C(u_1)$ or $H$ contains a $(2, 3)_{(u_1, u_3)}$-path.
One sees that $\sum_{x\in N(u)\setminus \{v\}}d(x)\ge |\{1, 3, 4, 5\}| + |\{2, 4\}| + |\{3, 5, 1\}| + |[2, 5]| + |\{6, 7, 2\}| + 2t_{u v} + t_0 = 2\Delta + 12 + t_0$.
Since $t_4\ge 2$, we have $T_4\setminus C(u_3)\ne \emptyset$.
Assume w.l.o.g. $8\not\in C(u_3)\cup C(u_4)$ and $H$ contains a $(6, 8)_{(u_4, v_6)}$-path.
ones sees that $d(u_4)\ge |[2, 6]| + t_1\ge 7$,
it follows that $t_1 = 2$, $C(u_4) = [2, 6]\cup T_1$, $T_{uv} = T_1\cup T_4$,
$W_1 = \{1, 3, 4, 5\}$, and $1, 2, 7\in C(u_3)$.
If $H$ contains no $(4, 6)_{(u_4, v_6)}$-path,
then $(v u_4, u u_1, u u_2, u u_4, u v)\overset{\divideontimes}\to (4, T_1\cap T_3, 1, 2, T_4)$.
Otherwise, $H$ contains a $(4, 6)_{(u_4, v_6)}$-path and $4\in C(v_6)$.
If $6\not\in C(u_3)$, then $(u u_1, u v)\overset{\divideontimes}\to (6, 8)$.
Otherwise, $6\in C(u_3)$,
and $\sum_{x\in N(u)\setminus \{v\}}d(x)\ge 2\Delta + 12 + t_0 + |\{6\}|\ge 2\Delta + 13$.
It follows that $t_0 = 0$ and $H$ contains a $(6, T_4)_{(u_4, v_6)}$-path.
One sees that the same argument applies if $u u_2\to 6$, and thus $T_1\subseteq C(v_6)$.
Hence, $C(v_1) = \{4, 6\}\cup T_{uv}$.
Then $(v v_6, u v)\overset{\divideontimes}\to (3, 6)$.

(4.2.2.) $T_1\subseteq C(u_3)$.
One sees that $|\{1, 3, 4, 5\}| + |\{2, 4\}| + |\{3, 5\}| + |\{2, 4, 5\}| + 2t_{u v} + t_0 = 2\Delta + 7 + t_0$.
We distinguish whether or not $T_4\setminus C(u_3) \ne \emptyset$.
\begin{itemize}
\parskip=0pt
\item
    $T_4\setminus C(u_3)\ne \emptyset$.
    Assume w.l.o.g. $8\not\in C(u_3)\cup C(u_4)$ and $H$ contains a $(6, 8)_{(u_4, v_6)}$-path.
    It follows from \ref{12T3124}(2) that $3\in C(u_4)$.
    Ones sees that $d(u_4)\ge |[2, 6]| + t_1\ge 7$.
    It follows that $t_1 = 2$, $C(u_4) = [2, 6]\cup T_1$, $T_{u v} = T_1\cup T_4$,
    $W_1 = \{1, 3, 4, 5\}$,
    and from \ref{122notinu12323}(2) that $2\in C(u_3)$.
    Assume that $6\not\in C(u_3)$.
    If $H$ contains no $(4, 6)_{(u_1, u_4)}$-path, then $(u u_1, u v)\overset{\divideontimes}\to (6, 8)$.
    Otherwise, $H$ contains a $(4, 6)_{(u_1, u_4)}$-path and $(u u_1, u u_2, u u_3, u v)\overset{\divideontimes}\to (2, 1, 6, 8)$.
    Hence, $6\in C(u_3)$,
    and $\sum_{x\in N(u)\setminus \{v\}}d(x)\ge 2\Delta + 7 + t_0 + |\{3, 6, 6, 2\}| = 2\Delta + 11 + t_0\ge 2\Delta + 11 + t_1\ge 2\Delta + 11 + 2 = 2\Delta + 13$.
    It follows that $1, 4, 7\not\in C(u_3)$,
    and from Corollary~\ref{auv2d(u)5}(7.2) that $H$ contains a $(4, T_1)_{(u_1, u_4)}$-path.
    Then $(u u_1, u u_2, u u_3, u u_4, u v)\overset{\divideontimes}\to (T_1, 1, 4, 7, 8)$.
\item
    $T_4\subseteq C(u_3)$.
    One sees from \ref{122notinu12323}(2) that $1\in C(u_4)$ or $2\in C(u_1)\cup C(u_3)$,
    and from \ref{12T43167neemptyset}(1) that $1/3/6/7\in C(u_4)$, $t_4\ge 2$.
    Then $\sum_{x\in N(u)\setminus \{v\}}d(x)\ge 2\Delta + 7 + t_0 + |\{1/2\}|\ge 2\Delta + 8 + t_1 + t_4\ge 2\Delta + 8 +2 + 2 = 2\Delta + 12$.
    Hence, assume w.l.o.g. $6\not\in S_u$,
    and it follows from \ref{12T3124}(3) that $3\in C(u_4)$.
    If $1\not\in C(u_3)$, then $(u u_1, u u_3, u u_4, u v)\overset{\divideontimes}\to (T_1, 1, 6, T_4)$.
    Otherwise, $1\in C(u_3)$ and then $1\not\in C(u_4)$, $2\not\in C(u_1)\cup C(u_3)$.
    Thus, $(u u_1, u u_2, u v)\overset{\divideontimes}\to (2, 1, T_4)$
    \footnote{Or, by \ref{122notinu12323}(2), we are done.  }.
\end{itemize}
Hence, $w_2\ge 3$ and $t_2 = \Delta - d(u_2) + w_2 - 2 \ge w_2 - 2 \ge 1$.
\end{proof}

\begin{lemma}
\label{12T1ge1}
$w_1\ge 3$ and $t_1 = \Delta - d(u_1) + w_1 - 2 \ge w_1 - 2 \ge 1$.
\end{lemma}
\begin{proof}
Assume $w_1\le 2$.
One sees from \ref{125notinu1}(1) that $5/2/4\in C(u_1)$.
Then $W_1 = \{1, 5/2/4\}$.
We consider below the three cases, respectively.

%

{\bf Case 1.} $W_1 = \{1, 5/2\}$.
It follows from \ref{12T4123}(3) that $T_2\subseteq C(u_4)$, $3\in C(u_4)$, $T_4\subseteq C(u_3)$,
from \ref{12T4123}(3) and \ref{12T3124}(3) that $6, 7\in C(u_3)\cup C(u_4)$.
One sees that $|\{3, 5\}| + |[2, 4]| + |\{6, 7\}| + t_{u v} + t_2 = \Delta + 5 + t_2$.
When $5\in C(u_1)$, it follows from \ref{1212notu3}(1) and \ref{1212notu4}(1) that $1\in C(u_3)\cup C(u_4)$.
One sees that the same argument applies if $uu_1\to 6/7$.
Hence, $H$ contains a $(\{1, 6, 7\}, T_4)_{(u_1, u)}$-path.
It follows from \ref{12T43167neemptyset}(2) that $5\in C(u_4)$.
Then $d(u_3) + d(u_4)\ge \Delta + 5 + t_2 + |\{1, 5\}| = \Delta + 7 + t_2\ge\Delta + 7 + 1 = \Delta + 8$.

In the other case, $W_2 = \{1, 2\}$.
Then $H$ contains a $(2, 5)_{(u_1, u_2)}$-path.
It follows from \ref{125u4u1oru4u2}(4) that $5\in C(u_4)$.
Then $d(u_3) + d(u_4)\ge \Delta + 5 + t_2 + |\{5\}| = \Delta + 6 + t_2\ge \Delta + 6 +1 = \Delta + 7$.
It follows that $t_2 = 1$, $1\not\in S_u$, $2\not\in C(u_3)$,
and from \ref{1212notu4}(1) that $4\in C(u_2)$.
Hence, $W_2 = \{2, 4, 5\}$ and $H$ contains a $(4, T_2)_{(u_2, u_4)}$-path.
Then $(u u_1, u u_2, u u_3, u u_4, u v)\overset{\divideontimes}\to (3, T_2, 4, 1, T_4)$.

{\bf Case 2.} $W_1 = \{1, 4\}$.
Then $H$ contains a $(4, 5)_{(u_1, u_4)}$-path.
It follows from \ref{122notinu12323}(1) that $5\in C(u_2)$,
from \ref{122notinu12323}(2) that $1\in C(u_2)$, or $H$ contains a $(1, 3/4)_{(u_2, u)}$-path.

{Case 2.1.} $T_2\setminus C(u_3)\ne\emptyset$.
Then $4\in C(u_2)$ and $T_2\subseteq C(u_4)$.
It follows from Corollary~\ref{auv2}(5.1) that $6, 7\in S_u$,
and from \ref{12T4123}(3) that $3\in C(u_2)\cap C(u_4)$.
If $1\not\in C(u_2)\cup C(u_4)$,
one sees that $H$ contains a $(1, 3)_{(u_2,u_3)}$-path,
then $(u u_1, u u_2, u u_4, u v)\overset{\divideontimes}\to (2, T_2\cap T_3, 1, T_4)$.
Otherwise, $1\in C(u_2)\cup C(u_4)$.
Since $d(u_2) + d(u_4)\ge t_{uv} + 2\times |[2, 5]| + |\{1\}| = \Delta + 7$,
it follows that $d(u_2) = \Delta$, $W_1\cup W_4 = [2, 5]\cup \{1\}$, $6, 7\in C(u_3)$, $d(u_4) = 7$,
and then $C(u_1) = \{1, 4\}\cup T_2\cup T_4$,
and from \ref{12T43167neemptyset}(2) that $T_4\subseteq C(u_3)$.
One sees from $d(u_1) = \Delta$ that $\Delta = 7$ and $d(u_3)\le 6$.
Since $d(u_3)\ge |\{3, 5, 6, 7\}\cup T_4| = 4 + t_4$,
we have $t_4 = 2$ and $W_1 = [1, 5]$, $W_4 = [2, 5]$, $W_3 = \{3, 5, 6, 7\}\cup T_4$.
Then $(v u_4, u u_4, u u_2)\to (4, 2, T_2)$ reduces to an acyclic edge $(\Delta + 5)$-coloring for $H$ such that $|C(u)\cap C(v)| = 1$.

{Case 2.2.} $T_2\subseteq C(u_3)$.
When $3\not\in C(u_2)$,
then $4\in C(u_2)$, $T_2\subseteq C(u_4)$,
and $1\in C(u_2)$, or $H$ contains a $(1, 4)_{(u_2, u_4)}$-path.
It follows from Corollary~\ref{auv2}(5.1) that $6, 7\in S_u$,
from \ref{1212notu3}(2) that $2\in C(u_3)$ or $3\in C(u_4)$.
Let $\kappa_{u_2} = |\{1, 6, 7\}\cap C(u_2)|$.
If $T_4\subseteq C(u_3)$,
then $d(u_3) + d(u_4)\ge |\{3, 5\}| + |\{2, 4, 5\}| + |\{2/3\}| + 3 - \kappa_{u_2} + t_{uv} + t_0 \ge \Delta + 7 - \kappa_{u_2} + t_2\ge \Delta + 7 - \kappa_{u_2} + \kappa_{u_2} + 3 - 2 = \Delta + 8$.
Otherwise, $T_4\setminus C(u_3)\ne \emptyset$.
Assume w.l.o.g. $8\not\in C(u_3)\cup C(u_4)$ and $H$ contains a $(6, 8)_{(u_4, v_6)}$-path.
It follows from \ref{12T3124}(3) that $2\in C(u_3)$, $4\in B_1$,
and from \ref{12T3124}(2) that $3\in C(u_4)$.
Recall that $1\in C(u_2)\cup C(u_4)$.
One sees that $d(u_2) + d(u_4)\ge |\{2, 4, 5\}| + |[2, 6]| + |\{1\}| + t_{uv} = \Delta + 7$.
It follows that $d(u_2) = \Delta$, $C(u_1)= \{1, 4\}\cup T_2\cup T_4 = \{1, 4\}\cup T_{uv}$, $d(u_4) = 7$,
and $7\not\in C(u_4)$.
If $H$ contains no $(4, 6)_{(u_1, u_4)}$-path, then $(u u_1, u v)\overset{\divideontimes}\to (6, 8)$.
Otherwise, $H$ contains a $(4, 6)_{(u_1, u_4)}$-path,
and $(v u_4, u u_4)\to (4, 8)$ reduces to an acyclic edge $(\Delta + 5)$-coloring for $H$ such that $|C(u)\cap C(v)| = 1$.

In the other case, $3\in C(u_2)$.
One sees that $|\{1, 4\}| + |\{2, 3, 5\}| + |\{3, 5\}| + |\{2, 4, 5\}| + 2t_{u v} = 2\Delta + 6$.

(2.2.1.) $T_2\setminus C(u_4)\ne \emptyset$.
It follows from Corollary~\ref{auv2}(5.1) that $6, 7\in S_u$,
Corollary~\ref{auv2}(5.2) and (3) that $4\in B_1$ and $4\in C(u_2)\cup C(u_3)$,
and from Corollary~\ref{auv2d(u)5}(7.2) that $2\in C(u_3)$.
If $H$ contains no $(4, i)_{(u_4, v)}$-path for each $i\in [6, 7]$,
then $(v u_4, u u_4)\to (4, T_2\cap T_4)$ reduces to an acyclic edge $(\Delta + 5)$-coloring for $H$ such that $|C(u)\cap C(v)| = 1$.
Otherwise, assume w.l.o.g. $6\in C(u_4)$ and $H$ contains a $(4, 6)_{(u_4, v)}$-path.
One sees that the same argument applies if let $u u_1\to 6$ and thus $H$ contains a $(6, T_4)_{(u_1, v_6)}$-path.
If $1\not\in C(u_2)\cup C(u_3)$,
one sees that $H$ contains a $(1, 4)_{(u_2, u_4)}$-path,
then $(u u_1, u u_2, u u_3, u v)\overset{\divideontimes}\to (2, T_2\cap T_4, 1, T_4)$.
If $6\not\in C(u_2)\cup C(u_3)$,
then $(u u_2, u u_3)\to (T_2\cap T_4, 6)$ reduces the proof to (2.4.3.).
Otherwise, $1, 6\in C(u_2)\cup C(u_3)$.
Hence, $d(u_2) + d(u_3) \ge t_{u v} + 2\times |\{2, 3, 5\}| + |\{1, 4, 6\}| = \Delta + 7$.
It follows that $d(u_2) = \Delta$, $d(u_3) = 7$, and $7\in C(u_4)\setminus (C(u_2)\cup C(u_3))$.
If $3\not\in C(u_4)$, then $(u u_2, u u_3, u u_4)\to (T_2\cap T_4, 7, 3)$ reduces the proof to (2.4.3.).
Otherwise, $3\in C(u_4)$.
Recall that $4\in C(v_1)$.
It follows from \ref{123v13v2}(1) that $R_1\ne \emptyset$.
Assume w.l.o.g. $8\in R_1$.
Then $8\in C(u_2)\cap C(u_4)$ and $C(u_4) = [2, 8]$.
It follows from \ref{126notinu3}(2) that $4\in C(u_3)$, $H$ contains a $(4, 7)_{(u_3, u_4)}$-path.
One sees clearly that $T_4\setminus C(u_3)\ne \emptyset$.
Assume w.l.o.g. $9\in T_3\cap T_4$ and then $H$ contains a $(7, 9)_{(u_4, v_7)}$-path.
Then $(u u_1, u v)\overset{\divideontimes}\to (7, 9)$.

(2.2.2.) $T_2\subseteq C(u_4)$.
Recall that $1\in S_u$.
One sees that the same argument applies if $i\not\in S_u$ and $u u_1\to i$, $i\in [6, 7]$,
and thus $6, 7\in S_u$.
\begin{itemize}
\parskip=0pt
\item
     $4\not\in C(u_2)$.
     It follows from \ref{122notinu12323}(2) that $1\in C(u_2)\cup C(u_3)$,
     from Corollary~\ref{auv2d(u)5}(7.2) that $2\in C(u_3)$.
     Assume that $6\not\in C(u_2)\cup C(u_3)$.
     It follows from \ref{126notinu3}(2) that $H$ contains a $(4, 6)_{(u_3, u_4)}$-path.
     If $H$ contains no $(6, j)_{(u_4, v)}$-path for some $j\in T_4$,
     then $(u u_1, u v)\overset{\divideontimes}\to (6, j)$.
     Otherwise, $H$ contains a $(6, T_4)_{(u_1, v)}$-path.
     Then $(u u_1, u u_2, u v)\overset{\divideontimes}\to (2, 6, T_4)$.
     Hence, $6, 7\in C(u_2)\cup C(u_3)$,
     and $d(u_2) + d(u_3)\ge 2\times |\{2, 3, 5\}| + |\{1, 6, 7\}| + t_{uv} = \Delta + 7$.
     It follows that $4\not\in C(u_2)\cup C(u_3)$,
     and from \ref{12T4123}(3) that $3\in C(u_4)$ and $T_4\subseteq C(u_3)$.
     Let $\kappa_{u_2} = |\{1, 6, 7\}\cap C(u_2)|$.
     One sees that $t_2\ge  \kappa_{u_2} + |\{2, 3, 5\}| - 2 = \kappa_{u_2} + 1$.
     Thus, $d(u_3) + d(u_4)\ge |\{2, 3, 5\}| + |[2, 5]| + 3 - \kappa_{u_2} + t_{uv} + t_0 \ge \Delta + 8 - \kappa_{u_2} + t_2\ge \Delta + 8 - \kappa_{u_2} + \kappa_{u_2} + 1 = \Delta + 9$.
\item
     $4\in C(u_2)$.
     Let $\kappa_{u_2} = |\{1, 6, 7\}\cap C(u_2)|$.
     One sees that $t_2\ge \kappa_{u_2} + |[2, 5]| - 2 = \kappa_{u_2} + 2$.
     If $T_4\subseteq C(u_3)$,
     then $d(u_3) + d(u_4)\ge |\{3, 5\}| + |\{2, 4, 5\}| + 3 - \kappa_{u_2} + t_{u v} + t_0 \ge \Delta + 6 - \kappa_{u_2} + t_2\ge \Delta + 6 - \kappa_{u_2} + \kappa_{u_2} + 2 = \Delta + 8$.
     Otherwise, $T_4\setminus C(u_3)\ne \emptyset$.
     Assume w.l.o.g. $8\not\in C(u_3)\cup C(u_4)$ and $H$ contains a $(6, 8)_{(u_4, v_6)}$-path.
     It follows from \ref{12T4123}(3) that $4\in B_1$.
     If $H$ contains no $(4, 6)_{(u_1, u_4)}$-path,
     then $(u u_1, u v)\overset{\divideontimes}\to (6, 8)$.
     Otherwise, $H$ contains a $(4, 6)_{(u_1, u_4)}$-path.
     If $H$ contains no $(4, 7)_{(u_4, v)}$-path,
     then $(v u_4, u u_4)\to (4, 8)$ reduces to an acyclic edge $(\Delta + 5)$-coloring for $H$ such that $|C(u)\cap C(v)| = 1$.
     Otherwise, $7\in C(u_4)$ and $H$ contains a $(4, 7)_{(u_4, v)}$-path.
     It follows that $t_2 = 2$, $C(u_4) = \{2\}\cup [4, 7]\cup T_2$, and $W_2 = [2, 5]$, $1\in C(u_3)$.
     If $H$ contains no $(7, 8)_{(u_1, v)}$-path,
     then $(u u_1, u v)\overset{\divideontimes}\to (7, 8)$.
     Otherwise, $H$ contains a $(7, 8)_{(u_1, v)}$-path.
     If $7\not\in C(u_3)$, then $(u u_1, u u_2, u v)\overset{\divideontimes}\to (2, 7, 8)$.
     Otherwise, $7\in C(u_3)$.
     Hence, $\sum_{x\in N(u)\setminus \{v\}}d(x)\ge |\{1, 4\}| + |[2, 5]| + |\{3, 5, 1, 7\}| + |\{2\}\cup [4, 7]| + 2t_{u v} + t_0 = 2\Delta + 11 + t_0\ge 2\Delta + 11 + t_2 = 2\Delta + 13$.
     It follows that $6\not\in C(u_2)\cup C(u_3)$.
     Then $(u u_1, u u_2, u v)\overset{\divideontimes}\to (2, 6, 8)$.
\end{itemize}

Hence, $w_1\ge 3$ and $t_1 = \Delta - d(u_1) + w_1 - 2 \ge w_1 - 2 \ge 1$.
\end{proof}

Assume that $H$ contains a $(6, 8)_{(u_4, v_6)}$-path, $8\in T_4$.
If $6\not\in C(u_1)$ and $H$ contains no $(6, i)_{(u_1, u_i)}$-path for each $i\in [2, 4]$,
then $(u u_1, u v)\overset{\divideontimes}\to (6, 8)$.
Otherwise, $6\in C(u_1)$ or $H$ contains a $(6, 2/3/4)_{(u_1, u)}$-path.
If $6\not\in S_u\setminus C(u_4)$ and $2\not\in C(u_1)$, $H$ contains no $(2, 3)_{(u_1, u_3)}$-path,
one sees that $H$ contains a $(4, 6)_{(u_1, u_4)}$-path,
then $(u u_1, u u_2, u v)\overset{\divideontimes}\to (2, 6, 8)$.
Hence, together with \ref{126notinu3}, \ref{12T3124}(3),
we proceed with the following proposition, or otherwise we are done:
\begin{proposition}
\label{12686862}
Assume that $H$ contains a $(6, 8)_{(u_4, v_6)}$-path, $8\in T_4$. Then
\begin{itemize}
\parskip=0pt
\item[{\rm (1)}]
	 $6\in C(u_1)$, or $H$ contains a $(6, 2/3/4)_{(u_1, u)}$-path;
\item[{\rm (2)}]
	 if $6\not\in S_u\setminus C(u_4)$, then
     \begin{itemize}
     \parskip=0pt
     \item[{\rm (2.1)}]
          $4\in C(u_2)\cap C(u_1)$,
          and $H$ contains a $(4, T_2)_{(u_2, u_4)}$-path, a $(4, T_1)_{(u_1, u_4)}$-path;
     \item[{\rm (2.2)}]
          $3\in B_1\cup B_2$, mult$_{S_u}(3)\ge 2$, and for each $i\in[1, 2]$, $3\in C(u_i)$, or $H$ contains a $(3, 4)_{(u_i, u_4)}$-path, $4\in C(u_i)$, $3\in C(u_4)$.
     \item[{\rm (2.3)}]
          $2\in C(u_1)$, or $H$ contains a $(2, 3)_{(u_1, u_3)}$-path.
	 \end{itemize}
\end{itemize}
\end{proposition}

Assume that $H$ contains a $(7, 8)_{(u_3, v_6)}$-path, $8\in T_4$.
If $7\not\in C(u_1)$ and $H$ contains no $(7, i)_{(u_1, u_i)}$-path for each $i\in [2, 4]$,
then $(u u_1, u v)\overset{\divideontimes}\to (7, 8)$.
Otherwise, $7\in C(u_1)$ or $H$ contains a $(7, 2/3/4)_{(u_1, u)}$-path.
If $7\not\in S_u\setminus C(u_3)$ and $2\not\in C(u_1)$, $H$ contains no $(2, 3)_{(u_1, u_3)}$-path,
one sees that $H$ contains a $(3, 7)_{(u_1, u_3)}$-path,
then $(u u_1, u u_2, u v)\overset{\divideontimes}\to (2, 7, 8)$.
Hence, together with \ref{126notinu4}, \ref{12T4123}(3),
we proceed with the following proposition, or otherwise we are done:
\begin{proposition}
\label{12787872}
Assume that $H$ contains a $(7, 8)_{(u_3, v_7)}$-path, $8\in T_4$. Then
\begin{itemize}
\parskip=0pt
\item[{\rm (1)}]
	 $7\in C(u_1)$, or $H$ contains a $(7, 2/3/4)_{(u_1, u)}$-path;
\item[{\rm (2)}]
	 if $7\not\in S_u\setminus C(u_3)$, then
     \begin{itemize}
     \parskip=0pt
     \item[{\rm (2.1)}]
	      $3\in C(u_1)\cap C(u_2)$,
          and $H$ contains a $(3, T_1)_{(u_1, u_3)}$-path, a $(3, T_2)_{(u_2, u_3)}$-path;
     \item[{\rm (2.2)}]
          $4\in B_1$, mult$_{S_u}(4)\ge 2$, and for each $i\in[1, 2]$, $4\in C(u_i)$, or $H$ contains a $(3, 4)_{(u_i, u_3)}$-path, $3\in C(u_i)$, $4\in C(u_3)$.
     \item[{\rm (2.3)}]
           $2\in C(u_1)$, or $H$ contains a $(2, 3)_{(u_1, u_3)}$-path.
	 \end{itemize}
\end{itemize}
\end{proposition}

By \ref{12686862}(2.1),
we proceed with the following proposition, or otherwise we are done:
\begin{proposition}
\label{126T46multge2}
Assume that $T_2\setminus C(u_4)\ne \emptyset$. Then
\begin{itemize}
\parskip=0pt
\item[{\rm (1)}]
	for each $i\in [6, 7]$, if $H$ contains a $(i, j)_{(u_4, v_i)}$-path for some $j\in T_4$,
    then $i\in S_u\setminus C(u_4)$, i.e., mult$_{S_u}(i)\ge 2$.
\item[{\rm (2)}]
	if mult$_{S_u}(6)\le 1$ and mult$_{S_u}(7)\le 1$,
    then $H$ contains no $(i, j)_{(u_4, v)}$-path, $i = 6, 7$, $j\in T_4$,
    and $5\in C(u_4)$, $H$ contains a $(5, T_4)_{(u_3, u_4)}$-path.
\end{itemize}
\end{proposition}

Assume that $H$ contains a $(4, i)_{(u_4, v_i)}$-path and a $(i, j)_{(u_4, v_i)}$-path for some $j\in T_4$, $i\in [6, 7]\cap C(u_4)$.
If $i\not\in C(u_1)$, and $H$ contains no $(i, l)_{(u_1, u_l)}$-path for each $l\in [2, 3]$,
then $(u u_1, u v)\overset{\divideontimes}\to (i, j)$.
Hence, we proceed with the following proposition, or otherwise we are done:
\begin{proposition}
\label{1264646j6j6Su}
If $H$ contains a $(4, i)_{(u_4, v_i)}$-path and a $(i, j)_{(u_4, v_i)}$-path for some $j\in T_4$, $i\in [6, 7]\cap C(u_4)$, then $i\in S_u\setminus C(u_4)$,
and if $i\not\in C(u_1)$, then $H$ contains a $(i, 2/3)_{(u_1, u)}$-path.
\end{proposition}

\begin{lemma}
\label{12T2u4T1u3}
If $T_2\setminus C(u_4)\ne \emptyset$ and $T_1\setminus C(u_3)\ne \emptyset$,
then we are done.
\end{lemma}
\begin{proof}
Then $4\in C(u_1)$, $H$ contains a $(4, T_1)_{(u_1, u_4)}$-path,
and $3\in C(u_2)$, $H$ contains a $(3, T_2)_{(u_2, u_3)}$-path, $T_2\subseteq C(u_3)$.
It follows from Corollary~\ref{auv2}(5.1) that $6, 7\in S_u$,
$1\in C(u_4)$ or $H$ contains a $(1, 3/2)_{(u_4, u)}$-path,
and $2\in C(u_3)$ or $H$ contains a $(1, 2)_{(u_1, u_3)}$-path,
from Corollary~\ref{auv2}(5.2) and (3) that $4\in B_1$, $3\in B_1\cup B_2$,
and $4\in C(u_2)\cup C(u_3)$, $3\in C(u_1)\cup C(u_4)$.
It follows from \ref{122notinu12323}(1) and \ref{125u4u1oru4u2}(1) that $5\in C(u_1)\cup C(u_4)$,
from \ref{122notinu12323}(1) and \ref{125notu2514}(1) that $5\in C(u_1)\cup C(u_2)$.
We distinguish whether or not mult$_{S_u}(5) = 1$.

{\bf Case 1.} mult$_{S_u}(5) = 1$.
It follows that $5\in C(u_1)\setminus (C(u_2)\cup C(u_4))$,
and from \ref{125notu2514}(1) that $H$ contains a $(1, 5)_{(u_1, u_2)}$-path, $1\in C(u_2)$.
If $H$ contains no $(5, j)_{(u_1, u_4)}$-path for some $j\in T_1\cap T_3$,
then $(u u_4, u u_1, u v)\overset{\divideontimes}\to (5, j, T_4)$.
Otherwise, $H$ contains a $(5, T_1\cap T_3)_{(u_1, u_4)}$-path.
If $1\not\in C(u_4)$ and $H$ contains no $(1, 3)_{(u_3, u_4)}$-path,
then $(u u_1, u u_2, u u_4)\to (T_1\cap T_3, 5, 1)$ reduces the proof to (2.4.2.).
Otherwise, $1\in C(u_4)$ or $H$ contains a $(1, 3)_{(u_3, u_4)}$-path.
Hence, $1\in C(u_3)\cup C(u_4)$ and mult$_{S_u}(1)\ge 2$.
Since $5\not\in C(u_4)$, it follows from \ref{12T43167neemptyset}(2) that for some $j\in T_2\cap T_4$,
$H$ contains a $(6, j)_{(u_4, v_6)}$-path and $6\in C(u_4)$, say $j = 8$.
It follows from \ref{126T46multge2}(1) that $6\in S_u\setminus C(u_4)$ and mult$_{S_u}(6)\ge 2$.
Then $\sum_{x\in N(u)\setminus \{v\}}d(x)\ge 6 + 2\times |\{1, 6, 3, 4\}| + |\{2, 5, 7\}| + 2t_{u v} + t_0 = 2\Delta + 13 + t_0$.
It follows that $t_0 = 0$, $\Delta = 7$,
and mult$_{S_u}(i) = 2$, $i = 1, 3, 4, 5$,
mult$_{S_u}(j) = 1$, $j = 2, 7$.
\begin{itemize}
\parskip=0pt
\item
     $1\in C(u_4)\setminus C(u_3)$.
     It follows that $2\in C(u_3)\setminus C(u_1)$ and from \ref{122notinu12323}(1) that $H$ contains a $(2, 3)_{(u_1, u_3)}$-path, $3\in C(u_1)\setminus C(u_4)$.
     Then $(u u_1, u u_2, u u_3, u u_4, u v)\overset{\divideontimes}\to (2, 5, 1, 3, T_1\cap T_4)$.
\item
    $1\in C(u_3)\setminus C(u_4)$.
    It follows that $3\in C(u_4)\setminus C(u_1)$,
    and from \ref{122notinu12323}(1) that $2\in C(u_1)\setminus C(u_3)$.
    One sees that $H$ contains a $(1, 2)_{(u_1, u_3)}$-path.
    Then $u u_2\to \alpha\in T_2\cap T_4$,
    and $(u u_1, u u_3, u u_4, u v)\overset{\divideontimes}\to (3, 2, 1, T_4\setminus \{\alpha\})$.
\end{itemize}

{\bf Case 2.} mult$_{S_u}(5)\ge 2$.
{\em Condition~(C2)} states that mult$_{S_u}(1) = 1$, mult$_{S_u\setminus C(u_4)}(2) = 1$,
and mult$_{S_u}(i) = 2$, $i = 3, 4$.
When (C2) holds, if $2\in C(u_1)$,
one sees that $1\in C(u_3)$ and $3\in C(u_4)$,
then $uu_2\to \alpha\in T_2\cap T_4$,
and $(u u_1, u u_3, u u_4, u v)\overset{\divideontimes}\to (3, 2, 1, T_4\setminus \{\alpha\})$.
Otherwise, $2\in C(u_3)\setminus C(u_1)$.
Then $1\in C(u_4)$ or $H$ contains a $(1, 3)_{(u_3, u_4)}$-path.
It follows from \ref{122notinu12323} (2) that $1\in C(u_4)$,
and $3\in C(u_1)$ or $4\in C(u_2)$.
Let $(u u_1, u u_2, u u_3)\to (2, T_2\cap T_4, 1)$.
If $4\not\in C(u_3)$, then reduces the proof to (2.4.2.).
Otherwise, $4\in C(u_3)$ and $3\in C(u_1)$.
Then $u u_4\to 3$ reduces the proof to (2.4.2.).

In the other case, $\sigma_{[1, 4]} = \sum_{i\in \{1, 3, 4\}}$ mult$_{S_u}(1) + $mult$_{S_u\setminus C(u_4)}(2)\ge 7$. Let $\sigma_{[5, 7]} = \sum_{i\in [5, 7]}$ mult$_{S_u}(i)$.
Recall that $\sum_{x\in N(u)\setminus \{v\}}d(x)\ge 6 + \sigma_{[1, 4]} + \sigma_{[5, 7]} + 2t_{u v} + t_0 = 2\Delta + 2 + \sigma_{[1, 4]} + \sigma_{[5, 7]} + t_0$.
If $\sigma_{[5, 7]}\ge 5$, one sees that $\sum_{x\in N(u)\setminus \{v\}}d(x)\ge 2\Delta + 2 + \sigma_{[1, 4]} + 5 + t_0 = 2\Delta + 7 + \sigma_{[1, 4]} + t_0$, then $\sigma_{[1, 4]}\le 6$ and (C2) holds.
Otherwise, $\sigma_{[5, 7]}\le 4$.
It follows that mult$_{S_u}(5) = 2$, mult$_{S_u}(i) = 1$, $i = 6, 7$,
$\sum_{x\in N(u)\setminus \{v\}}d(x)\ge 2\Delta + 2 + \sigma_{[1, 4]} + 4 + t_0 \ge 2\Delta + 6 + 7 + t_0 = 2\Delta + 13 + t_0\ge 2\Delta + 13$.
Hence, $\Delta = 7$, $\sigma_{[1, 4]} = 7$, $t_0 = 0$, and $T_2\cap C(u_4) = T_1\cap C(u_3) = \emptyset$.
One sees from \ref{126T46multge2}(2) that $5\in C(u_4)$, $H$ contains a $(5, T_4)_{(u_3, u_4)}$-path.
Since $t_0 = 0$ and $T_4\subseteq C(u_3)$,
we have $C_{12} = \emptyset$, i.e., $T_{uv} = T_1\cup T_2$.
%

When $6\in C(u_3)\cup C(u_4)$, one sees that $6\not\in B_1\cup B_2$,
and it follows from \ref{126notinu3}(2) that if $6\in C(u_4)$, then $H$ contains a $(4, 6)_{(u_3, u_4)}$-path;
from \ref{126notinu4}(1) that if $6\in C(u_3)$, then $H$ contains a $(3, 6)_{(u_3, u_4)}$-path.
Since the same argument applies if $u u_1\to 6$ or $u u_2\to 6$, $C(v_6) = \{4, 6\}$ $\cup T_1\cup T_2$.
Then $(v v_6, u v)\overset{\divideontimes}\to (3, 6)$.
In the other case, $6, 7\not\in C(u_3)\cup C(u_4)$.
Recall that $4\in B_1 \subseteq C(v_1)$.

One sees from \ref{124i4i3j3j}(2) that $H$ contains a $(5, 4)_{(u_3, u_4)}$-path. 
It follows from \ref{125notinu1}(1) that if $5\not\in C(u_1)$, then $H$ contains a $(2, 5)_{(u_1, u_2)}$-path,
and from \ref{125notu2514}(1) that if $5\not\in C(u_2)$, then $H$ contains a $(1, 5)_{(u_1, u_2)}$-path,
i.e., $5\not\in B_1\cup B_2$.
One sees from $\{1, 4\}\cup T_2\subseteq C(v_1)$ and \ref{123v13v2}(1) that $T_1\setminus C(v_1)\ne \emptyset$.
Then $(v u_3, u v)\overset{\divideontimes}\to (T_1\cap R_1, 5)$.
\end{proof}

Recall that $t_1 + t_2 = \Delta - d(u_1) + \Delta - d(u_2) + w_1 + w_2 - 4\ge w_1 + w_2 - 4$.
It follows that $w_1 + w_2\le 4 + t_1 + t_2$,
and $t_1 + t_2 = w_1 + w_2 - 4$ if and only if $d(u_1) = d(u_2) = \Delta$.

\begin{lemma}
\label{12T1T2ge4}
If $t_1 + t_2\ge 4$, then $T_2\setminus C(u_4)\ne \emptyset$ and $T_1\setminus C(u_3)\ne \emptyset$.
\end{lemma}
\begin{proof}
Assume w.l.o.g. $[8, 11]\subseteq T_1\cup T_2$.
It follows that $w_1 + w_2\le 8$ and if $w_1 + w_2 = 8$,
then $d(u_1) = d(u_2) = \Delta = 7$, and $d(u_3)\le 6$.
We consider the following two cases, respectively.

{\bf Case 1.} $T_2\subseteq C(u_4)$.
One sees from \ref{12T43167neemptyset}(2) that $5/6/7\in C(u_4)$.
It follows that $t_1 + t_2 = 4$ and $C(u_4) = \{4, 2, 5/6/7\}\cup [8, 11]$, $C_{12} = T_4 = T_{u v}\setminus [8, 11]$.
Further, it follows from \ref{12T43167neemptyset}(1) that $6\in C(u_4)$, $H$ contains a $(3, 6)_{(u_4, v_6)}$-path,
from Lemma~\ref{auvge2}(1.2) that $H$ contains a $(6, C_{12})_{(u_4, v_6)}$-path,
and from \ref{124i4i3j3j}(2) that $H$ contains a $(6, 4)_{(u_4, v_6)}$-path.
Hence, it follows from \ref{1264646j6j6Su} that $6\in S_u\setminus C(u_4)$.
One sees from \ref{12T4123}(3) that $4\in B_1$, $4\in C(u_2)\cup C(u_3)$,
and from \ref{125notinu1}(1) that $5\in C(u_1)$, or $H$ contains a $(2, 5)_{(u_1, u_2)}$-path.

{Case 1.1.} $H$ contains no $(1, 2)_{(u_2, v)}$-path.
It follows from Corollary~\ref{auv2d(u)5}(7.2) that $3\in C(u_1)$ and $H$ contains a $(3, T_1)_{(u_1, u_3)}$-path, $T_1\subseteq C(u_3)$.
If $1\not\in S_u$, then $(u u_1, u u_3, u u_4, u v)\overset{\divideontimes}\to (T_1, 1, 3, T_4)$.
Since the same argument applies if $7\not\in S_u$ and $uu_1\to 7$, we have $1, 7\in S_u$.
One sees that $6 + |\{3, 4, 4\}| + |\{6\}| + |\{6, 7, 1, 5\}| + 2t_{u v} + t_0 = 2\Delta + 10 + t_0$.
It follows that $T_2\setminus C(u_3)\ne \emptyset$,
and then from Corollary~\ref{auv2}(5.2) that $3\in C(u_2)$.
One sees from $1, 3, 4, 5/2\in C(u_1)$ that $t_1\ge 2$.
Hence, $\sum_{x\in N(u)\setminus \{v\}}d(x)\ge 2\Delta + 10 + t_0 + |\{3\}|\ge 2\Delta + 11 + t_1\ge 2\Delta + 11 + 2 = 2\Delta + 13$.
It follows that $t_1 = 2$, $2\not\in S_u\setminus C(u_4)$, $C_{12}\cap C(u_3) = \emptyset$,
and then $W_1 = \{1, 3, 4, 5\}$.
Thus, $(u u_1, u u_2, u v)\overset{\divideontimes}\to (2, 5, C_{12})$.

{Case 1.2.} $H$ contains a $(1, 2)_{(u_2, v)}$-path.
Then $1\in C(u_2)$.
If $1\not\in C(u_3)$ and $3\not\in C(u_1)$,
then $(u u_1, u u_3, u u_4, u v)\overset{\divideontimes}\to (T_1, 1, 3, T_4)$.
Otherwise, $1\in C(u_3)$ or $3\in C(u_1)$.
When $T_1\cup T_2\subseteq C(u_3)$,
one sees that $6 + |\{1, 4, 4, 6, 5, 6, 1/3\}| + 2t_{uv} + t_0 = 2\Delta + 9 + t_0\ge 2\Delta + 9 + t_1 + t_2 = 2\Delta + 13$.
It follows that $7\not\in S_u$ and mult$_{S_u}(3)\le 1$.
Thus, by \ref{126notinu4}, we are done.

In the other case, $(T_1\cup T_2)\setminus C(u_3)\ne \emptyset$.
It follows from Corollary~\ref{auv2}(5.2) and (3) that $3\in B_1$ and $3\in C(u_2)$,
and from \ref{126notinu4} that $7\in C(u_1)\cup C(u_2)$.
Together with $5\in C(u_1)\cup C(u_2)$, we have $W_1\cup W_2 = [1, 3]\cup \{1, 3, 4, 5\}\cup \{7\}$,
$4, 6\in C(u_3)$.
It follows from \ref{122notinu12323}(1) that $2\in C(u_3)$.
One sees that $6 + |\{1, 4, 4, 3, 3, 6, 6, 5, 7, 2\}| + 2t_{u v} + t_0\ge 2\Delta + 12 + t_2\ge 2\Delta + 12 + 1 = 2\Delta + 13$.
It follows that $t_2 = 1$, $7\in C(u_1)$, $T_1\cap C(u_3) = \emptyset$, and $1\not\in C(u_3)$.
Thus, $(u u_2, u u_3, u v)\overset{\divideontimes}\to (T_2, 7, T_1)$
\footnote{Or, by \ref{126notinu3}(2), we are done. }.

{\bf Case 2.} $T_2\setminus C(u_4)\ne \emptyset$ and $T_2\cup T_1 \subseteq C(u_3)$.
Then $3\in C(u_2)$ and $H$ contains a $(3, T_2)_{(u_2, u_3)}$-path.
It follows from Corollary~\ref{auv2}(5.2) and (3) that $4\in B_1$, and $4\in C(u_2)\cup C(u_3)$,
and from Corollary~\ref{auv2d(u)5}(7.2) that $2\in C(u_3)$ or $H$ contains a $(1, 2)_{(u_1, u_3)}$-path.
Hence, $t_1 + t_2 = 4$, $C(u_3) = \{3, 5, 1/2\}\cup T_1\cup T_2$, $4\in C(u_1)\cap C(u_2)$,
and from \ref{126notinu3}(2) that $6, 7\in C(u_1)\cup C(u_2)$.

If $C(u_3) = \{3, 5, 1\}\cup T_1\cup T_2$, one sees that $2\in C(u_1)$ and $5\not\in C(u_1)\cup C(u_2)$,
then $(v u_3, u u_3, u u_2)\to (4, 5, T_2\cap T_4)$ reduces the proof to (2.4.1.).
Otherwise, $C(u_4) = \{3, 5, 2\}\cup T_1\cup T_2$.
It follows from \ref{126notinu3}(2) that $6, 7\in C(u_2)$.
By Lemma~\ref{12T2ge1} and Lemma~\ref{12T1ge1}, $W_1 = [3, 4]\cup [6, 7]$, $W_2 = \{1, 4, a_2\}$,
and $d(u_1) = d(u_2) = \Delta$.
It follows that $\Delta = 7$ and together with $d(u_3) = 7$, $d(u_4)\le 6$.
If $2\not\in C(u_1)$, then $(u u_1, u u_2, u u_3)\to (2, T_2\cap T_4, 1)$ reduces the proof to (2.4.2.).
Otherwise, $2\in C(u_1)$ and $W_1 = \{1, 4, 2\}$.
It follows from Corollary~\ref{auv2}(5.1) that $1\in C(u_4)$,
from Corollary~\ref{auv2}(4) that $12\in C(u_4)$,
and from \ref{125notinu1}(1) that $5\in C(u_4)$, $H$ contains a $(4, 5)_{(u_1, u_4)}$-path.
Hence, $C(u_4) = \{2, 4, 5, 1, 12\}\cup T_1$,
and by \ref{124i4i3j3j}(2), we are done.

Hence, $T_2\setminus C(u_4)\ne \emptyset$.
One sees from Corollary~\ref{auv2d(u)5}(5) that $T_2\subseteq C(u_3)$.
Then $T_1\setminus C(u_3)\ne \emptyset$.
\end{proof}

By Lemma \ref{12T2u4T1u3} and Lemma \ref{12T1T2ge4}, hereafter we assume that $t_1 + t_2\le 3$,
and $T_1\cup T_2\subseteq C(u_i)$ for some $i\in [3, 4]$.
Together with Lemma \ref{12T2ge1} and Lemma \ref{12T1ge1}, $6\le w_1 + w_2\le 4 + t_1 + t_2 \le 7$.
It follows that $w_1 = w_2 = 3$, or $\{w_1, w_2\} = \{3, 4\}$.
One sees from \ref{125notinu1}(1) that $5\in S_u\setminus C(u_3)$.

Assume that $6\not\in S_u$.
It follows from Corollary~\ref{auv2}(5.1) that $3, 4\in C(u_1)\cap C(u_2)$ and $T_1\cup T_2\subseteq C(u_3)\cap C(u_4)$.
If $5\not\in C(u_1)$,
one sees from $w_1 + w_2\le 7$ that $H$ contains no $(2, 5)_{(u_1, u_2)}$-path and $H$ contains a $(4, 5)_{(u_1, u_4)}$-path,
then $(u u_1, u u_4, u v)\overset{\divideontimes}\to (5, 6, 4)$.
Otherwise, $W_1 = \{1\}\cup [4, 5]$ and $W_2 = [2, 4]$.
It follows from \ref{122notinu12323}(1) that $H$ contains a $(5, 4)_{(u_2, u_4)}$-path,
or a $(3, 2)_{(u_1, u_3)}$-path.
If $H$ contains no $(5, 4)_{(u_2, u_4)}$-path,
one sees that $H$ contains a $(3, 2)_{(u_1, u_3)}$-path,
then $(u u_1, u u_2, u u_3, u v)\overset{\divideontimes}\to (2, 5, 6, 3)$.
Otherwise, $H$ contains a $(5, 4)_{(u_2, u_4)}$-path.
If $2\not\in C(u_3)$, then $(u u_1, u u_2, u u_4, u v)\overset{\divideontimes}\to (2, 5, 6, 4)$.
Otherwise, $2\in C(u_3)$.
One sees that $|\{1, 3, 4, 5\}| + |[2, 4]| + |\{3, 5, 2\}| + |\{2, 4, 5\}| + 2t_{uv} + t_0  = 2\Delta + 9 + t_0$.
If $C_{12}\subseteq C(u_3)\cup C(u_4)$, then $t_0\ge t_{u v} = \Delta - 2$,
and $\sum_{x\in N(u)\setminus \{v\}}d(x)\ge 2\Delta + 9 + \Delta - 2 = 3\Delta + 7\ge 2\Delta + 14$.
Otherwise, $C_{12}\setminus (C(u_3)\cup C(u_4))\ne \emptyset$.
It follows from \ref{12T43167neemptyset}(2) that $7\in C(u_4)$ and $H$ contains a $(7, j)_{(u_4, v_7)}$-path for some $j\in T_4$.
Hence, $C(u_4) = \{2, 4, 5, 7\}\cup T_1\cup T_2$.
It follows from \ref{124i4i3j3j}(2) that $H$ contains a $(4, 7)_{(u_4, v_7)}$-path,
and from \ref{1264646j6j6Su} that $7\in S_u\setminus C(u_4)$.
Thus, $\sum_{x\in N(u)\setminus \{v\}}d(x)\ge 2\Delta + 9 + t_0 + |\{7, 7\}| = 2\Delta + 11 + t_0\ge 2\Delta + 11 + t_1 + t_2 = 2\Delta + 14$.
Hence, we proceed with the following proposition, or otherwise we are done:
\begin{proposition}
\label{1267inSu}
$5\in S_u\setminus C(u_3)$ and $6, 7\in S_u$.
\end{proposition}

If $1\not\in S_u$ and $3\not\in C(u_4)$,
one sees from Corollary~\ref{auv2d(u)5}(7.2) that $3, 4\in C(u_1)$ and $H$ contains a $(3, T_1)_{(u_1, u_3)}$-path, a $(4, T_1)_{(u_1, u_4)}$-path,
then $(u u_1, u u_3, u u_4, u v)\to (T_1, 1, 3, T_4)$.
Hence, we proceed with the following proposition, or otherwise we are done:
\begin{proposition}
\label{121Su3Cu4}
If $1\not\in S_u$, then $3, 4\in C(u_1)$, $T_1\subseteq C(u_3)\cap C(u_4)$, and $3\in C(u_4)$.
\end{proposition}

It follows that $6, 7\in S_u$, $5\in S_u\setminus C(u_3)$, and $1\in S_u$ or $3\in C(u_4)$.
One sees that $c_{12} = t_{uv} - t_1 - t_2 = \Delta - 2 - (t_1 + t_2)\ge \Delta - 2- 3 = \Delta - 5\ge 1$.

Assume that $8\in C_{12}\setminus (C(u_3)\cup C(u_4))$.
It follows from \ref{12T43167neemptyset}(2) that $H$ contains a $(6, 8)_{(u_4, v_6)}$-path.
If  mult$_{S_u\setminus C(u_3)}(5) = 1$, then it follows from mult$_{S_u}(5) = 1$ that $7\in C(u_3)$ and $H$ contains a $(7, 8)_{(u_3, v_7)}$-path.
Hence, we proceed with the following proposition:
\begin{proposition}
\label{12C12Cu3Cu4neemptyset}
Assume that $8\in C_{12}\setminus (C(u_3)\cup C(u_4))$.
Then $6\in C(u_4)$, $H$ contains a $(6, 8)_{(u_3, u_4)}$-path,
and if mult$_{S_u\setminus C(u_3)}(5) = 1$, then $7\in C(u_3)$ and $H$ contains a $(7, 8)_{(u_3, v_7)}$-path.
\end{proposition}

\begin{lemma}
\label{12T2neu4emptyset}
If $T_2\setminus C(u_4)\ne \emptyset$, then we are done.
\end{lemma}
\begin{proof}
Then $3\in C(u_2)$ and $H$ contains a $(3, T_2)_{(u_2, u_3)}$-path.
It follows from Corollary~\ref{auv2}(5.1) that $2\in C(u_3)$ or $H$ contains a $(1, 2)_{(u_1, u_2)}$-path,
from Corollary~\ref{auv2}(5.2) and (3) that $4\in B_1$ and $4\in C(u_2)\cup C(u_3)$.
Let $\sigma_{13} = $ mult$_{S_u}(1) + $mult$_{C(u_1)\cup C(u_4)}(3)$.

When $\sigma_{13}\ge 2$,
one sees that $\sum_{x\in N(u)\setminus \{v\}}d(x)\ge |\{1, 4\}| + |\{2, 3\}| + |\{3, 5\}| + |\{2, 4\}| + |\{2, 4, 5, 6, 7\}| + \sigma_{13} + 2t_{u v} + t_0 = 2\Delta + 9 + \sigma_{13} + t_0
\ge 2\Delta + 9 + 2 + t_1\ge 2\Delta + 11 + 1 = 2\Delta + 12$.
If $C_{12}\subseteq C(u_3)\cup C(u_4)$,
then $\sum_{x\in N(u)\setminus \{v\}}d(x)\ge 2\Delta + 12 + \Delta - 5 = 3\Delta + 7\ge 2\Delta + 14$.
Otherwise, $C_{12}\setminus (C(u_3)\cup C(u_4))\ne \emptyset$.
Assume w.l.o.g. $8\not\in C(u_3)\cup C(u_4)$ and $H$ contains a $(6, 8)_{(u_4, v_6)}$-path.
It follows from \ref{126T46multge2}(1) that $6\in S_u\setminus C(u_4)$.
Hence, $\sum_{x\in N(u)\setminus \{v\}}d(x)\ge 2\Delta + 9 + \sigma_{13} + t_0 + |\{6\}|\ge 2\Delta + 10 + \sigma_{13} + t_1\ge 2\Delta + 10 + 2 + 1 = 2\Delta + 13$.
It follows that $\Delta = 7$, $t_1 = 1$, $\sigma_{13} = 2$, $T_2\cap C(u_4)= \emptyset$,
$C_{12}\cap (C(u_3)\cup C(u_4)) = \emptyset$,
and mult$_{S_u}(i) = 1$, $i = 2, 5, 7$, mult$_{S_u}(j) = 2$, $j = 4, 6$, holds.
It follows from \ref{12C12Cu3Cu4neemptyset} that $7\in C(u_3)$, $H$ contains a $(7, 8)_{(u_3, v_7)}$-path,
and from \ref{12787872}(2.1) that $3\in C(u_1)$.
Hence, $W_1 = \{1, 3, 4\}$.
It follows from \ref{125notinu1}(1) that $H$ contains a $(4, 5)_{(u_1, u_4)}$-path, $5\in C(u_4)$,
from \ref{122notinu12323}(1) that $2\in C(u_3)$,
and from \ref{124i4i3j3j}(2) that $H$ contains a $(4, 6)_{(u_1, u_4)}$-path.
If $6\not\in C(u_3)$, then $(u u_1, u v)\overset{\divideontimes}\to (6, 8)$.
Otherwise, $6\in C(u_3)$ and $C(u_3) = [2, 3]\cup [5, 7]\cup T_1\cup T_2$, $4\in C(u_2)$.
Since $d(u_2) + d(u_3)\ge |[2, 4]| + |[2, 3]| + |[5, 7]| + t_{uv} + t_1 = \Delta + 6 + t_1 \ge \Delta + 6 + 1= \Delta + 7$, we have $1\not\in C(u_2)\cup C(u_3)$.
Then $(u u_1, u u_2, u u_3, u v)\overset{\divideontimes}\to (2, T_2, 1, T_4)$.
In the other case, $\sigma_{13}\le 1$.

One sees from \ref{121Su3Cu4} that $1\in S_u$ or $3\in C(u_4)$.
Hence, $3\not\in C(u_1)$.
It follows from Corollary~\ref{auv2}(5.1) that $1\in C(u_4)$ or $H$ contains a $(1, 2/3)_{(u_4, u)}$-path.
One sees from $1\in S_u$ that $3\not\in C(u_4)$,
and $1\in C(u_4)$ or $H$ contains a $(1, 2)_{(u_2, u_4)}$-path.
Let $(uu_3, uu_4)\to (1, 3)$.
If $1\in C(u_4)$, then $(u u_1, u u_2)\to (2, T_2\cap T_4)$;
if $1\in C(u_2)$, then $u u_1\to T_1$.
These reduce the proof to (2.4.2.).
\end{proof}

Now we assume that $T_2\subseteq C(u_4)$.
One sees from \ref{1267inSu} that $5\in S_u\setminus C(u_3)$ and $6, 7\in S_u$.

When $1\not\in S_u$,
it follows from \ref{124i4i3j3j} that $3, 4\in C(u_1)$, $T_1\subseteq C(u_3)$, $3\in C(u_4)$,
and from \ref{122notinu12323}(2) that $2\in S_u\setminus C(u_4)$.
If $2, 5\not\in C(u_1)$, one sees from \ref{125notinu1}(1) that $H$ contains a $(4, 5)_{(u_1, u_4)}$-path,
then $(u u_1, u u_4, u v)\overset{\divideontimes}\to (5, 1, 4)$.
Otherwise, $5/2\in C(u_1)$.
It follows that $t_1 = 2$ and $t_2 = 1$.
One sees that $\sum_{x\in N(u)\setminus \{v\}}d(x)\ge |\{1, 3, 4\}| + |\{2, a_2\}| + |\{3, 5\}| + |\{2, 4, 3\}| + |\{2, 5, 6, 7\}| + 2t_{uv} + t_0 = 2\Delta + 10 + t_0$.
If $C_{12}\subseteq C(u_3)\cup C(u_4)$, then $t_0\ge c_{12} + t_1 = t_{u v} - t_2 = \Delta - 3$,
and $\sum_{x\in N(u)\setminus \{v\}}d(x)\ge 2\Delta + 10 + t_0\ge 2\Delta + 10 + \Delta - 3 = 3\Delta + 7\ge 2\Delta + 14$.
Otherwise, assume w.l.o.g. $8\not\in C(u_3)\cup C(u_4)$ and $H$ contains a $(6, 8)_{(u_4, v_6)}$-path.
Hence, $C(u_4) = [2, 4]\cup \{6\}\cup T_1\cup T_2$.
It follows that $4\in B_1$, mult$_{S_u}(4)\ge 2$,
from \ref{12T2u4T1u3}(2) that $H$ contains a $(4, 6)_{(u_4, v_6)}$-path,
and from \ref{1264646j6j6Su} that $6\in S_u\setminus C(u_4)$.
Then $\sum_{x\in N(u)\setminus \{v\}}d(x)\ge 2\Delta + 10 + |\{6\}| + t_0\ge 2\Delta + 11 + t_1 = 2\Delta + 13$.
It follows that mult$_{S_u}(i) = 1$, $i = 5, 7$, $a_2 = 4$, $3\not\in C(u_2)$,
and follows from \ref{12C12Cu3Cu4neemptyset} that $7\in C(u_4)$, 
from \ref{12787872}(2.1) that $H$ contains a $(3, 7)_{(u_1, u_3)}$-path.
Thus, by \ref{126notinu4}(2), we are done.
In the other case, $1\in S_u$.
Recall that $\sum_{x\in N(u)\setminus \{v\}}d(x)\le 3\Delta + 6$ for ($A_{8.1}$)--($A_{8.4}$).

Let $\sigma_{[1, 4]} = \sum_{i\in \{1, 3, 4\}}$ mult$_{S_u}(i) + $ mult$_{C(u_1)\cup C(u_3)}(2)$,
 $\sigma_{[5, 7]} = \sum_{i\in [6, 7]}$ mult$_{S_u}(i) + $ mult$_{S_u\setminus C(u_3)}(5)$.

Assume that $\sigma_{[1, 4]}\le 3$.
One sees that if $(T_1\cup T_2)\setminus C(u_3)\ne \emptyset$, then $4\in C(u_1)\cup C(u_2)$ and mult$_{S_u}(3)\ge 2$, $\sigma_{[1, 4]}\ge |\{4, 3, 3, 1\}| = 4$,
and that if $|C(u_1)\cap [3, 4]|\ge 2$, or $|C(u_2)\cap [3, 4]|\ge 2$, then $\sigma_{[1, 4]}\ge 4$.
Then $|C(u_1)\cap [3, 4]| = |C(u_2)\cap [3, 4]| = 1$.
One sees that if $a_2 = 3$, then $2\in C(u_1)\cup C(u_3)$, and $\sigma_{[1, 4]}\ge |\{a_1, a_2, 1, 2\}| = 4$.
Hence, $a_2 = 4$.
It follows from \ref{122notinu12323}(2) $1\in C(u_2)\cup C(u_4)$,
and from Corollary~\ref{auv2d(u)5}(7.2) that $a_1 = 4$.
Thus, $(u u_1, u u_3, u u_4, u v)\overset{\divideontimes}\to (T_1, 1, 3, T_4)$.
Hence, we proceed with the following proposition, or otherwise we are done:
\begin{proposition}
\label{121Su1234ge4}
$1\in S_u$ and $\sigma_{[1, 4]}\ge 4$, $\sigma_{[5, 7]}\ge 3$.
\end{proposition}

Recall that $\sum_{x\in N(u)\setminus \{v\}}d(x)\ge 6 + \sigma_{[1, 4]} + \sigma_{[5, 7]} + 2t_{uv} + t_0 = 2\Delta + 2 + \sigma_{[1, 4]} + \sigma_{[5, 7]} + t_0$.

\begin{lemma}
\label{12686ge2}
If $H$ contains a $(6, 8)_{(u_4, v_6)}$-path for $8\in T_4$, then $6\in S_u\setminus C(u_4)$.
\end{lemma}
\begin{proof}
Assume that $6\not\in S_u\setminus C(u_4)$.
One sees from \ref{12686862}(1) that $H$ contains a $(4, 6)_{(u_1, u_4)}$-path,
from \ref{12686862}(2) that $4\in C(u_1)\cap C(u_2)$,
$3\in C(u_1)\cup C(u_2)$, mult$_{C(u_1)\cup C(u_2)\cup C(u_4)}(3)\ge 2$,
and $2\in C(u_1)$ or $H$ contains a $(2, 3)_{(u_1, u_3)}$-path.
One sees that $\sum_{x\in N(u)\setminus \{v\}}d(x)\ge 6 + \sigma_{[1, 4]} + \sigma_{[5, 7]} + 2t_{u v} + t_0\ge 6 + |\{4, 4, 3, 3, 2, 1\}| + |\{5, 6, 7\}| + 2t_{u v} + t_0 = 2\Delta + 11 + t_0$,
and $d(u_1) + d(u_2) + d(u_4)\ge 4 + |\{3, 3, 4, 4, 5, 6\}| + 2t_{u v} = 2\Delta + 6$.
It follows that mult$_{C(u_1)}(2) + $mult$_{S_u\setminus C(u_3))}(5) + $mult$_{S_u\setminus C(u_3))}(7)\le 2$.

{\bf Case 1.} $H$ contains a $(4, j)_{(u_4, v)}$-path for some $j\in [5, 7]$.

{Case 1.1.} $5\in C(u_4)$ and $H$ contains a $(4, 5)_{(u_3, u_4)}$-path.
It follows from \ref{125notinu1}(1) that $5\in C(u_1)$,
or $H$ contains a $(2, 5)_{(u_1, u_2)}$-path, $2\in C(u_2)$.
Hence, $5\in C(u_1)$, $d(u_1) + d(u_2) + d(u_4)\ge = 2\Delta + 7$,
and then $2\not\in C(u_1)$, $3\in C(u_1)$.
It follows that $W_1 = \{1\}\cup [3, 5]$, $W_2 = [2, 4]$, $C(u_4) = \{2\}\cup [4, 6]\cup T_1\cup T_2$, $C(u_3) = [1, 5]\cup \{7\}$.
Thus, $(u u_1, u u_4, u v)\overset{\divideontimes}\to (T_1, 1, T_4)$
\footnote{Or, by Corollary~\ref{auv2d(u)5}(7.2), we are done}.

{Case 2.1.} $H$ contains a $(4, 7)_{(u_4, v)}$-path.
Then $d(u_1) + d(u_2) + d(u_4)\ge 2\Delta + 6 + |\{7\}| = 2\Delta + 7$.
It follows that $d(u_3)\le 6$, $2\not\in C(u_1)$, $1, 2\in C(u_4)$, and then $3\in C(u_1)$.
One sees that $5\in C(u_1)$ or $H$ contains a $(4, 5)_{(u_1, u_4)}$-path,
and $d(u_1) + d(u_4)\ge t_{u v} + t_2 + |\{1, 3, 4\}| + |\{2, 4, 6, 7\}| + |\{5\}| = \Delta + 6 + t_2\ge \Delta + 6 + 1 = \Delta + 7$.
It follows that $t_2 = 1$, $3\not\in C(u_4)$, and $W_2 = [2, 4]$,
from \ref{125notu2514}(1) that $T_1\subseteq C(u_3)$.
If $T_2\setminus C(u_3)\ne \emptyset$,
one sees from \ref{125u4u1oru4u2}(2) that if $5\not\in C(u_4)$, then $T_2\subseteq C(u_3)$,
then $5\in C(u_4)$, and $(u u_1, u u_2, u u_4)\to (5, T_2\cap T_3, 1)$ reduces the proof to (2.4.2.).
Otherwise, $T_2\subseteq C(u_3)$.
Thus, $d(u_3)\ge |\{5\}\cup [1, 3]\cup T_1\cup T_2|\ge 7$.

{\bf Case 2.} $H$ contains no $(4, i)_{(u_4, v_i)}$-path for each $i\in [5, 7]$.

One sees from \ref{124i4i3j3j}(2) that $3\in C(u_4)$ and $T_4\subseteq C(u_3)$, i.e., $C_{12}\subseteq C(u_3)\cup C(u_4)$, and $t_0\ge c_{12} = \Delta - 2 - t_1 - t_2\ge \Delta - 2 - 3 = \Delta - 5\ge 1$.
Recall that $\sum_{x\in N(u)\setminus \{v\}}d(x)\ge 2\Delta + 11 + t_0\ge 2\Delta + 11 + c_{12}\ge 2\Delta + 11 + 1 = 2\Delta + 12$.
We have $c_{12}\le t_0\le 2$ and $\Delta \le 7$.
When $w_4\ge 5$, one sees that $d(u_4)\ge w_4 + t_1 + t_2\ge 5 + 1 + 1 = 7$,
it follows that $\Delta = 7$, $t_1 = t_2 = 1$, and thus $t_0\ge c_{12} = \Delta - 4\ge 3$.
In the other case, $W_4 = [2, 4]\cup \{6\}$.

{Case 2.1.} $2\in C(u_1)$.
Then $d(u_1) + d(u_2) + d(u_4)\ge 4 + |\{3, 3, 4, 4, 5, 6, 2\}| + 2t_{uv} = 2\Delta + 7$.
It follows that $d(u_i) = \Delta = 7$, $i = 1, 2, 4$, $d(u_3)\le 6$,
and $\{1, 7, 3, 5\}\subseteq C_{12}\subseteq C(u_3)$, mult$_{S_u\setminus C(u_3)}(5) = 1$.
One sees from $w_1 + w_2\ge |\{1, 2\}| + |\{4, 4, 3, 5, 2\}| = 7$ that $t_1 + t_2 = 3$ and then $c_{12} = 2$.
Hence, $C(u_3) = \{1, 7, 3, 5\}\cup C_{12}$. 
Since mult$_{S_u\setminus C(u_3)}(5) = 1$, it follows from Corollary~\ref{auv2}(4) that $H$ contains a $(7, T_2)_{(u_3, v)}$-path.
If $H$ contains no $(3, 7)_{(u_1, u_3)}$-path, then $(u u_1, u v)\to (7, T_2)$.
Otherwise, $3\in C(u_1)\setminus C(u_2)$ and $H$ contains a $(3, 7)_{(u_1, u_3)}$-path.
Thus, by \ref{126notinu4}(2), we are done.

{Case 2.2.} $2\not\in C(u_1)$.  Then $3, 5\in C(u_1)$ and $W_1 = \{1\}\cup [3, 5]$.
It follows that $t_1 = 2$, $t_2 = 1$, $C(u_4) = \{2, 3, 4, 6\}\cup T_1\cup T_2$,
and $d(u_i) = \Delta = 7$, $i = 1, 2, 4$, $d(u_3)\le 6$.
Hence $W_1 = \{2, 4, \alpha\}$, $C(u_3) = \{2, 3, 5, \beta\}\cup C_{12}$, where $\{\alpha, \beta\} = \{1, 7\}$.
One sees that $5, 3\not\in C(u_2)$ and $1, 5\not\in C(u_4)$.
Thus, by \ref{125u4u1oru4u2}(2), we are done.
\end{proof}

One sees that if $T_1\cup T_2\subseteq C(u_3)$ and $C_{12}\subseteq C(u_3)\cup C(u_4)$,
then $\sum_{x\in N(u)\setminus \{v\}}d(x)\ge 2\Delta + 2 + \sigma_{[1, 4]} + \sigma_{[5, 7]} + t_0\ge 2\Delta + 2 + 4 + 3 + t_{u v}\ge 2\Delta + 9 + \Delta - 2 = 3\Delta + 7$.
Hence, if $T_1\cup T_2\subseteq C(u_3)$, then $C_{12}\setminus (C(u_3)\cup C(u_4))\ne \emptyset$,
and if $C_{12}\subseteq C(u_3)\cup C(u_4)$, then $(T_1\cup T_2)\setminus C(u_3) \ne \emptyset$.

\begin{lemma}
\label{12787ge2}
If $H$ contains a $(7, 8)_{(u_4, v_7)}$-path, where $7\in T_1\cup T_2\cup T_4$,
then $7\in S_u\setminus C(u_3)$.
\end{lemma}
\begin{proof}
Assume that $7\not\in S_u\setminus C(u_3)$.
One sees from \ref{12787872}(1) that $H$ contains a $(3, 7)_{(u_i, u_3)}$-path for some $i\in [1, 2]$,
from \ref{12787872}(2) that $3\in C(u_1)\cap C(u_2)$, $T_1\cup T_2\subseteq C(u_3)$,
$4\in B_1$, mult$_{S_u}(4)\ge 2$.

Hence, it suffices to assume that $8\in T_4$ and $3\in C(u_1)$, $H$ contains a $(3, 7)_{(u_1, u_3)}$-path.
One sees from \ref{12T43167neemptyset}(2) that $H$ contains a $(6, 8)_{(u_4, v_6)}$-path.
It follows from Lemma~\ref{12686ge2} that $6\in S_u\setminus C(u_4)$.
One sees that $\sum_{x\in N(u)\setminus \{v\}}d(x)\ge 6 + \sigma_{[1, 4]} + \sigma_{[5, 7]} + 2t_{uv} + t_0
\ge 6 + |\{4, 4, 3, 3, 1\}| + |\{5, 6, 6, 7\}| + 2t_{u v} + t_0 = 2\Delta + 11 + t_0\ge 2\Delta + 11 + t_1 + t_2\ge 2\Delta + 11 + 1 + 1 = 2\Delta + 13$.
It follows that $t_1 = t_2 = 1$, $W_1 = \{1, 3, 4\}$,
$2\not\in S_u\setminus C(u_4)$, and mult$_{S_u}(1) = $ mult$_{S_u\setminus C(u_3)}(5) = 1$.
Hence, $5\in C(u_4)$ and $H$ contains a $(4, 5)_{(u_1, u_4)}$-path.
Thus, $(u u_1, u u_2)\to (2, 5)$ reduces the proof to (2.3.)
\footnote{Or, by \ref{122notinu12323}(1), we are done}.
\end{proof}

If $6\not\in C(u_1)\cup C(u_2)$ and $H$ contains no $(6, i)_{(u_i, u_j)}$-path,
$i = 1, 2$, $j = 3, 4$,
one sees that the same argument applies if $uu_1\to 6$,
and it follows from \ref{122notinu12323}(2) that $2\in C(u_1)$ or $H$ contains a $(3, 2)_{(u_1, u_3)}$-path.
Otherwise, we proceed with the following proposition, or otherwise we are done:
\begin{proposition}
\label{126notCu1Cu2}
If $6\not\in C(u_1)\cup C(u_2)$ and $H$ contains no $(6, i)_{(u_i, u_j)}$-path,
$i = 1, 2$, $j = 3, 4$,
then $2\in C(u_1)$ or $H$ contains a $(3, 2)_{(u_1, u_3)}$-path.
\end{proposition}

Assume $6\not\in S_u\setminus C(u_4)$ and $H$ contains no $(4, 6)_{(u_3, u_4)}$-path.
One sees from \ref{126notinu3} that $4\in C(u_1)\cap C(u_2)$,
and from \ref{12T3124}(3) that mult$_{S_u}(3)\ge 2$,
then $\sigma_{[1, 4]}\ge |\{4, 4, 3, 3, 1\}|\ge 5$.
Together with \ref{126notCu1Cu2}, we proceed with the following proposition, or otherwise we are done:
\begin{proposition}
\label{12464614ge5}
Assume $6\not\in S_u\setminus C(u_3)$ and $H$ contains no $(4, 6)_{(u_3, u_4)}$-path. Then
\begin{itemize}
\parskip=0pt
\item[{\rm (1)}]
	$4\in C(u_1)\cap C(u_2)$, mult$_{S_u}(3)\ge 2$, and $\sigma_{[1, 4]}\ge |\{4, 4, 3, 3, 1\}|\ge 5$;
\item[{\rm (2)}]
	 if $H$ contains a $(4, 6)_{(u_4, v)}$-path, then $\sigma_{[1, 4]}\ge 6$,
     and $2\in C(u_1)$ or $H$ contains a $(3, 2)_{(u_1, u)}$-path.
\end{itemize}
\end{proposition}

Assume $7\not\in S_u\setminus C(u_3)$ and $H$ contains no $(3, 7)_{(u_3, u_4)}$-path.
One sees from \ref{126notinu4} that $3\in C(u_1)\cap C(u_2)$, $T_1\cup T_2\subseteq C(u_3)$,
and from \ref{12T4123}(3) that $4\in B_1$, mult$_{S_u}(4)\ge 2$,
then $\sigma_{[1, 4]}\ge |\{4, 4, 3, 3, 1\}|\ge 5$.
Together with \ref{126notCu1Cu2}, we proceed with the following proposition, or otherwise we are done:
\begin{proposition}
\label{12474714ge5}
Assume $7\not\in S_u\setminus C(u_3)$ and $H$ contains no $(3, 7)_{(u_3, u_4)}$-path. Then
\begin{itemize}
\parskip=0pt
\item[{\rm (1)}]
	$3\in C(u_1)\cap C(u_2)$, $T_1\cup T_2\subseteq C(u_3)$,
    $4\in B_1$, mult$_{S_u}(4)\ge 2$, and $\sigma_{[1, 4]}\ge |\{4, 4, 3, 3, 1\}|\ge 5$;
\item[{\rm (2)}]
	 if $H$ contains a $(3, 7)_{(u_3, v)}$-path, then $\sigma_{[1, 4]}\ge 6$,
     and $2\in C(u_1)$ or $H$ contains a $(3, 2)_{(u_1, u)}$-path.
\end{itemize}
\end{proposition}

Assume that $8\in C_{12}\setminus (C(u_3)\cup C(u_4))$, and $H$ contains a $(6, 8)_{(u_3, u_4)}$-path.
It follows from Lemma~\ref{12686ge2} that $6\in S_u\setminus C(u_4)$.
One sees from Corollary~\ref{auv2}(4) that if $H$ contains no $(7, 8)_{(u_3, v)}$-path,
then mult$_{S_u\setminus C(u_3)}(5)\ge 2$.
If $H$ contains a $(7, 8)_{(u_3, v)}$-path,
then it follows Lemma~\ref{12787ge2} that $7\in S_u\setminus C(u_3)$.
Hence, $\sigma_{[5, 7]}\ge |\{5, 6, 6, 7, 5/7\}|\ge 5$,
and we proceed with the following proposition:
\begin{proposition}
\label{12C12notinCu3Cu4}
If $C_{12}\setminus (C(u_3)\cup C(u_4))\ne \emptyset$, then $\sigma_{[5, 7]}\ge |\{5, 6, 6, 7, 5/7\}|\ge 5$.
\end{proposition}

\begin{lemma}
\label{12T2inCu4T1T2inCu3}
If $T_2\cup T_1\subseteq C(u_4)\cap C(u_3)$, then we are done.
\end{lemma}
\begin{proof}
One sees that if $C_{12}\subseteq C(u_3)\cup C(u_4)$,
then $\sum_{x\in N(u)\setminus \{v\}}d(x)\ge 2\Delta + 2 + \sigma_{[1, 4]} + \sigma_{[5, 7]} + t_0\ge 2\Delta + 2 + 4 + 3 + t_{u v}\ge 2\Delta + 9 + \Delta - 2 = 3\Delta + 7$.
Hence, $C_{12}\setminus (C(u_3)\cup C(u_4))\ne \emptyset$, say $8\not\in C(u_3)\cup C(u_4)$.
It follows from \ref{12T4123} that $4\in B_1$ and $4\in C(u_2)\cup C(u_3)$,
from \ref{12T43167neemptyset}(2) that $6\in C(u_4)$ and $H$ contains a $(6, 8)_{(u_4, v_6)}$-path,
from \ref{124i4i3j3j}(2) that $H$ contains a $(4, 5/6/7)_{(u_4, v)}$-path.
One sees from \ref{12C12notinCu3Cu4} that $\sigma_{[5, 7]}\ge 5$.
Hence, $\sum_{x\in N(u)\setminus \{v\}}d(x)\ge 2\Delta + 2 + \sigma_{[1, 4]} + \sigma_{[5, 7]} + t_0\ge 2\Delta + 2 + 4 + 5 + t_0 = 2\Delta + 11 + t_0\ge 2\Delta + 11 + t_1 + t_2 = 2\Delta + 13$.
It follows that $\sigma_{[1, 4]} = 4$, $t_1 = t_2 = 1$, $\sigma_{[5, 7]} = 5$,
and mult$_{S_u}(6) = 2$.
One sees from \ref{12464614ge5} that if mult$_{S_u}(7) = 1$ and $H$ contains a $(4, 7)_{(u_4, v_7)}$-path,
then $\sigma_{[1, 4]} \ge 5$ and we are done.
Otherwise, mult$_{S_u}(7) \ge 2$ or $H$ contains no $(4, 7)_{(u_4, v_7)}$-path.

When $W_1 = \{1, 4, 5\}$, it follows from \ref{12686862}(1) that $H$ contains a $(4, 6)_{(u_1, u_4)}$-path,
and from \ref{122notinu12323}(1) that $5\in C(u_2)$ or $H$ contains a $(4, 5)_{(u_2, u_4)}$-path.
Hence, mult$_{S_u}(7) = 1$ and $H$ contains a $(4, 7)_{(u_4, v_7)}$-path, a contradiction.
In the other case, $5\not\in C(u_1)$.

When $H$ contains no $(7, 8)_{(u_3, v_7)}$-path,
it follows from Corollary~\ref{auv2}(4) that $5\in B_2\subseteq C(u_2)\cap C(u_4)$, $W_2 = \{2, a_2, 5\}$,
and $H$ contains a $(4, 5)_{(u_1, u_4)}$-path.
Then mult$_{S_u}(7) = 1$.
We further assume that $H$ contains a $(4, 6)_{(u_4, v_6)}$-path,
and it follows from \ref{12686862}(1) that $6\in C(u_1)$ or $H$ contains a $(3, 6)_{(u_1, u_3)}$-path.
If $H$ contains no $(2, 3)_{(u_1, u_2)}$-path, then $(u u_1, u u_2, u v)\overset{\divideontimes}\to (2, 6, 8)$.
Otherwise, $3\in C(u_1)$ and $H$ contains a $(3, \{2, 6\})_{(u_1, u_3)}$-path.
Thus, $\sigma_{[1, 4]} \ge |\{3, 4, a_2, 2, 1\}| = 5$.

Now assume $H$ contains a $(7, 8)_{(u_3, v_7)}$-path.
It follows from Lemma~\ref{12787ge2} that $7\in S_u\setminus C(u_3)$ and mult$_{S_u\setminus C(u_3)}(5) = 1$.
If $5\in C(u_2)$, since $W_1 = \{1, 2, 4\}$, $W_2 = \{2, 5, a_2\}$,
and it follows from \ref{12686862}(1) that $H$ contains a $(4, 6)_{(u_1, u_4)}$-path,
and then a $(4, 7)_{(u_4, v_7)}$-path,
then $(uu_1, uv)\overset{\divideontimes}\to (7, 8)$
\footnote{Or, by \ref{12787872}(1), we are done. }.
Otherwise, $5\in C(u_4)$ and $H$ contains a $(4, 5)_{(u_1, u_4)}$-path.
It follows from \ref{122notinu12323}(1) that $2\in C(u_1)$ or $H$ contains a $(3, 2)_{(u_1, u_3)}$-path.
Since $\sigma_{[1, 4]} = 4$, $W_1 = \{1, 4, 2\}$.
If $H$ contains no $(4, 6)_{(u_1, u_4)}$-path,
then $(u u_1, u u_2, u v)\overset{\divideontimes}\to (5, 6, 8)$.
Otherwise, $H$ contains a $(4, 6)_{(u_1, u_4)}$-path, a $(4, 7)_{(u_1, u_4)}$-path,
and it follows from \ref{12787872}(1) that $W_2 = \{2, 4, 7\}$.
Thus, $(u u_1, u u_4, u v)\overset{\divideontimes}\to (T_1, 1, T_4)$
\footnote{Or, by Corollary~\ref{auv2d(u)5}(7.2), we are done. }.
\end{proof}

Hence, $(T_1\cup T_2)\setminus C(u_3)\ne \emptyset$.
It follows from Corollary~\ref{auv2}(5.2) and (3) that $3\in B_1\cup B_2$,
and for each $i\in [1, 2]$, $3\in C(u_i)$ or $H$ contains a $(4, 3)_{(u_i, u_4)}$-path.
And for each $j\in [1, 2]$ such that $T_j\setminus C(u_3)\ne \emptyset$, $4\in C(u_j)$ and $H$ contains a $(4, T_j)_{(u_j, u_4)}$-path.

Assume that $2\not\in C(u_1)$ and $H$ contains no $(2, 3)_{(u_1, u_3)}$-path.
One sees from \ref{122notinu12323}(1) that $5\in C(u_2)$ or $H$ contains a $(4, 5)_{(u_2, u_4)}$-path,
from \ref{122notinu12323}(2) that $1\in C(u_2)$ or $H$ contains a $(1, 3/4)_{(u_2, u)}$-path,
and from \ref{125notinu1}(1) that $5\in C(u_1)$ or $H$ contains a $(4, 5)_(u_1, u_4)$-path.
It follows that mult$_{S_u\setminus C(u_3)}(5)\ge 2$.
If $6\not\in C(u_1)\cup C(u_2)\cup C(u_4)$,
then by \ref{126notCu1Cu2}, we are done.
Otherwise, $6, 7\in C(u_1)\cup C(u_2)\cup C(u_4)$.
One sees that $d(u_1) + d(u_2) + d(u_4)\ge 4 + |\{4, 3, 3, 5, 5, 6, 7\}| + 2t_{uv} = 2\Delta + 7$.
It follows that mult$_{C(u_1)\cup C(u_2)}(4) = 1$, mult$_{S_u)}(3) = 2$,
$3\in C(u_2)$ and $H$ contains a $(1, 3)_{(u_1, u_3)}$-path,
and from Corollary~\ref{auv2d(u)5}(7.2) that $3\in C(u_1)$, $T_1\subseteq C(u_3)$.
One sees from \ref{12T4123}(3) that $4\in B_1$.
Thus, $T_2\subseteq C(u_3)$, a contradiction.
In the other case, we proceed with the following proposition, or otherwise we are done:
\begin{proposition}
\label{122inCu1Cu3}
$2\in C(u_1)$ or $H$ contains a $(2, 3)_{(u_1, u_3)}$-path.
\end{proposition}

Assume that $C_{12}\subseteq C(u_3)\cup C(u_4)$, i.e., $T_4\subseteq C(u_3)$.
Then $t_0\ge c_{12} = t_{u v} - (t_1 + t_2)\ge t_{uv} - 3 = \Delta - 5$.
One sees that $\sum_{x\in N(u)\setminus \{v\}}d(x)\ge 2\Delta + 2 + \sigma_{[1, 4]} + \sigma_{[5, 7]} + t_0 \ge 2\Delta + 2 + |\{4, 3, 3, 1, 2\}| + |[5, 7]| + t_0 = 2\Delta + 10 + t_0$.
When mult$_{S_u}(4)\ge 2$, $\sum_{x\in N(u)\setminus \{v\}}d(x)\ge 2\Delta + 10 + |\{4\}| + t_0 \ge 2\Delta + 11 + c_{12} \ge 2\Delta + 11 + \Delta - 5 = 3\Delta + 6$.
Hence, $t_0 = c_{12} = \Delta - 5$, $(T_1\cup T_2)\cap C(u_3) = \emptyset$,
and mult$_{S_u}(i) =  $mult$_{S_u\setminus C(u_3)}(5) = 1$, $i = 6, 7$.
It follows from Lemma~\ref{12686ge2} and Lemma~\ref{12787ge2} that $5\in C(u_4)$,
and from \ref{125notinu1}(1) that $H$ contains a $(4, 5)_{(u_1, u_4)}$-path.
Thus, $(uu_2, uv)\to (5, T_1)$
\footnote{Or, by \ref{125notu2514}(1), we are done. }.
In the other case, mult$_{S_u}(4) = 1$.
It follows that $3\in C(u_4)$ and $T_i\subseteq C(u_i)$ for some $i\in [1, 2]$ such that $4\not\in C(u_i)$.
One sees from \ref{12T43167neemptyset}(2) that $5/6/7\in C(u_4)$.
Hence, $t_0\ge c_{12} + t_i\ge t_4 + t_i\ge 2 + 1 = 3$,
and $\sum_{x\in N(u)\setminus \{v\}}d(x)\ge 2\Delta + 10 + t_0 \ge 2\Delta + 10 + 3 = 2\Delta + 13$.
It follows that $t_0 = 3$, $c_{12} = t_4 = 2$, $t_i = 1$,
and mult$_{S_u}(i) =  $mult$_{S_u\setminus C(u_3)}(5) = 1$, $i = 6, 7$.
It follows from Lemma~\ref{12686ge2} and Lemma~\ref{12787ge2} that $5\in C(u_4)$,
and from \ref{125notinu1}(1) that $H$ contains a $(4, 5)_{(u_1, u_4)}$-path, $4\in C(u_1)$.
Thus, it follows from \ref{125notu2514}(1) that $T_1\subseteq C(u_3)$.
Since $4\not\in C(u_2)$, $T_2\subseteq C(u_3)$, a contradiction.
Hence, we proceed with the following proposition, or otherwise we are done:
\begin{proposition}
\label{12CunotinCu3}
$C_{12}\setminus (C(u_3)\cup C(u_4))\ne \emptyset$.
\end{proposition}

Assume w.l.o.g. $8\not\in C(u_3)\cup C(u_4)$ and $H$ contains a $(6, 8)_{(u_4, v_6)}$-path.
It follows from Lemma~\ref{12686ge2} that $6\in S_u\setminus C(u_4)$,
from \ref{12C12notinCu3Cu4} that $\sigma_{[5, 7]}\ge |\{5, 6, 6, 7, 5/7\}|\ge 5$,
and from \ref{12T4123}(3) that $4\in B_1$ and mult$_{S_u}(4)\ge 2$,
from \ref{124i4i3j3j}(2) that $H$ contains a $(4, 5/6/7)_{(u_4, v)}$-path.
Hence, $\sum_{x\in N(u)\setminus \{v\}}d(x)\ge 2\Delta + 2 + \sigma_{[1, 4]} + \sigma_{[5, 7]} + t_0\ge 2\Delta + 2 + |\{4, 4, 3, 3, 1, 2\}| + |\{5, 6, 6, 7, 5/7\}| + t_0 = 2\Delta + 13 + t_0$.
It follows that $T_0 = \emptyset$, $4\in C(u_1)\cap C(u_2)$,
and $H$ contains a $(4, T_i)_{(u_, u_4)}$-path, $i = 1, 2$.

Let
\[
S'_u = \biguplus_{i\in \{1, 2, 4\}}(C(u_i)\setminus \{c(u u_i)\}).
\]
One sees from $d(u_1) + d(u_2) + d(u_4)\ge 4 + |\{4, 4, 3, 3, 5, 6, 5/7\}| + 2t_{uv} = 2\Delta + 7$.
Hence, $d(u_i) = 7$, $i = 1, 2, 4$, and $d(u_3)\le 6$, $1, 2\in C(u_3)$.
It follows from Corollary~\ref{auv2d(u)5}(7.2) that $3\in C(u_4)$, $H$ contains a $(1, 3)_{(u_3, u_4)}$-path,
and from \ref{122inCu1Cu3} that $3\in C(u_1)$.

One sees from \ref{125notu2514} that $5\in C(u_2)$ or $H$ contains a $(4, 5)_{(u_2, u_4)}$-path.
If $5\not\in C(u_1)$, and $H$ contains no $(4, 5)_{(u_1, u_4)}$-path,
then $(u u_1, u u_2, u v)\overset{\divideontimes}\to (5, 1, T_2)$.
Otherwise, $5\in C(u_1)$, or $H$ contains a $(4, 5)_{(u_1, u_4)}$-path.
It follows that mult$_{S'_u}(5) = 2$.
One sees
Then $d(u_1) + d(u_4)\ge t_{u v} + t_2 + |\{1, 4, 3\}| + |\{2, 3, 4, 6\}| + |\{5\}| = \Delta + 6 + t_2\ge \Delta + 6 + 1 = \Delta + 7$.
It follows that $t_2 = 1$, $W_2 = \{2, 4, 5\}$, $W_1 \biguplus W_4 = \{1, 4, 3\}\cup \{2, 3, 4, 6\}\cup \{5\}$,
and $C(u_3) = [1, 3]\cup [5, 7]$.
One sees that $5\in C(u_4)$ or $H$ contains a $(4, 5)_{(u_1, u_4)}$-path.
It follows that $H$ contains a $(4, 6)_{(u_4, v)}$-path.
If $H$ contains no $(6, j)_{(u_2, v_6)}$-path, then $(u u_2, u v)\overset{\divideontimes}\to (6, j)$.
Otherwise, $H$ contains a $(6, T_1)_{(u_2, v_6)}$-path.
Since $\{1, 3, 4\}\cup T_2\cup T_4\subseteq C(v_1)$, $T_1\setminus C(v_1)\ne \emptyset$.
It follows from Lemma~\ref{auvge2}(1.2) that $H$ contains a $(7, T_1\cap R_1)_{(u_3, v_1)}$-path,
and from Lemma~\ref{12787ge2} that $7\in S'_u$, a contradiction.

This finishes the inductive step for the Case 2.5.1.

{(2.5.2.)} $c(vu_4) = 6$.
One sees from \ref{123v13v2}(1)--(2) that for each $j\in [3, 4]$,
$j\in C(v_i)$ or $H$ contains a $(j, l)_{(v_1, v_l)}$-path for some $l\in C(v)\setminus \{i\}$, $i\in [1, 2]$.
Hence, together with Corollary~\ref{auv2d(u)5}(2.1),
we have the following proposition, or otherwise we are done:
\begin{proposition}
\label{12outRineemptyset}
\begin{itemize}
\parskip=0pt
\item[{\rm (1)}]
	$R_1\cup R_2\ne \emptyset$.
\item[{\rm (2)}]
	If $R_i = \emptyset$, $i\in [1, 2]$,
    then $C(v_i) = \{i, 7\}\cup T_{uv}$ and $H$ contains a $(7, [3, 4])_{(v_1, v_7)}$-path.
\item[{\rm (3)}]
    If $2\not\in C(u_1)$, $H$ contains no $(2, i)_{(u_1, u_i)}$-path for each $i\in [3, 4]$,
    then $1\in C(u_2)$ or $H$ contains a $(1, i)_{(u_2, u_i)}$-path for some $i\in [3, 4]$.
\end{itemize}
\end{proposition}

By \ref{126notinu3} and \ref{12T3124},
we have the following proposition, or otherwise we are done:
\begin{proposition}
\label{12out6notinu3}
Assume that $6\not\in C(u_3)$. Then
\begin{itemize}
\parskip=0pt
\item[{\rm (1)}]
	if $6\not\in C(u_1)$, $H$ contains no $(6, 2)_{(u_2, u_3)}$-path,
    then $4\in C(u_1)$, $H$ contains a $(4, T_1)_{(u_1, u_4)}$-path;
\item[{\rm (2)}]
	if $6\not\in C(u_2)$, $H$ contains no $(6, 1)_{(u_1, u_3)}$-path,
    then $4\in C(u_2)$, $H$ contains a $(4, T_2)_{(u_2, u_4)}$-path;
\item[{\rm (3)}]
	if $H$ contains no $(6, 1)_{(u_3, u_1)}$-path,
    then $3\in C(u_2)$ or $H$ contains a $(3, 1/4)_{(u_2, u)}$-path;
\item[{\rm (4)}]
	if $H$ contains no $(6, 2)_{(u_3, u_2)}$-path,
    then $3\in C(u_1)$ or $H$ contains a $(3, 2/4)_{(u_1, u)}$-path;
\item[{\rm (5)}]
	if $H$ contains no $(6, i)_{(u_3, u_i)}$-path for each $i\in [1, 2]$,
    then $3\in B_1\cup B_2$, mult$_{S_u}(3)\ge 2$, and for each $i\in[1, 2]$, $3\in C(u_i)$,
    or $H$ contains a $(3, 4)_{(u_i, u_4)}$-path, $4\in C(u_i)$, $3\in C(u_4)$.
\end{itemize}
\end{proposition}

Assume that $5\not\in C(u_1)$ and $H$ contains no $(5, 4)_{(u_1, u_4)}$-path.
If $H$ contains no $(5, 2)_{(u_2, u_1)}$-path,
then $u u_1\to 5$ reduces the proof to the Case (2.5.1.).
If $1\not\in C(u_2)$, $H$ contains no $(1, i)_{(u_2, u_i)}$-path for each $i\in \{3, 4\}$,
and $(u u_1, u u_2)\to (5, 1)$ reduces the proof to the Case (2.5.1.).
Hence, and by symmetry, we proceed with the following propositions, or otherwise we are done:
\begin{proposition}
\label{12out5notinu1u2}
\begin{itemize}
\parskip=0pt
\item[{\rm (1)}]
	If $5\not\in C(u_1)$ and $H$ contains no $(5, 4)_{(u_1, u_4)}$-path, then
    \begin{itemize}
    \parskip=0pt
    \item[{\rm (1.1)}]
	    $H$ contains a $(5, 2)_{(u_1, u_2)}$-path.
    \item[{\rm (1.2)}]
	    $1\in C(u_2)$, or $H$ contains a $(1, i)_{(u_2, u_i)}$-path for some $i\in \{3, 4\}$.
    \end{itemize}
\item[{\rm (2)}]
	If $5\not\in C(u_2)$ and $H$ contains no $(5, 4)_{(u_2, u_4)}$-path, then
    \begin{itemize}
    \parskip=0pt
    \item[{\rm (2.1)}]
	    $H$ contains a $(5, 1)_{(u_1, u_2)}$-path.
    \item[{\rm (2.2)}]
	    $2\in C(u_1)$, or $H$ contains a $(2, i)_{(u_1, u_i)}$-path for some $i\in \{3, 4\}$.
    \end{itemize}
\end{itemize}
\end{proposition}

\begin{proposition}
\label{12out6notinu1u2}
\begin{itemize}
\parskip=0pt
\item[{\rm (1)}]
	If $6\not\in C(u_1)$ and $H$ contains no $(6, 3)_{(u_1, u_3)}$-path, then
    \begin{itemize}
    \parskip=0pt
    \item[{\rm (1.1)}]
	    $H$ contains a $(6, 2)_{(u_1, u_2)}$-path.
    \item[{\rm (1.2)}]
	    $1\in C(u_2)$, or $H$ contains a $(1, i)_{(u_2, u_i)}$-path for some $i\in \{3, 4\}$.
    \end{itemize}
\item[{\rm (2)}]
	If $6\not\in C(u_2)$ and $H$ contains no $(6, 3)_{(u_2, u_3)}$-path, then
    \begin{itemize}
    \parskip=0pt
    \item[{\rm (2.1)}]
	    $H$ contains a $(6, 1)_{(u_1, u_2)}$-path.
    \item[{\rm (2.2)}]
	    $2\in C(u_1)$, or $H$ contains a $(2, i)_{(u_1, u_i)}$-path for some $i\in \{3, 4\}$.
    \end{itemize}
\end{itemize}
\end{proposition}

By \ref{12out5notinu1u2}(1.1), (2.1), and \ref{12out6notinu1u2}(1.1), (2.1),
we have the following proposition:
\begin{proposition}
\label{12out56u1u2}
$5/2/4, 6/2/3\in C(u_1)$, $5/1/4, 6/1/3\in C(u_2)$, and $5, 6\in C(u_1)\cup C(u_2)$.
\end{proposition}

\begin{lemma}
\label{12outT2ge1}
For each $i\in [1, 2]$, $w_i\ge 3$, $t_i\ge 1$,
and for each $j\in T_i$, there exists a $a_i\in [3, 4]\cap C(u_i)$ such that $H$ contains a $(a_i, j)_{(u_i, u_{a_i})}$-path.
\end{lemma}
\begin{proof}
Assume that $w_1\le 2$.
It follows from \ref{12out56u1u2} that $W_1 = [1, 2]$, $5, 6\in C(u_2)$,
and then $C(u_1) = [1, 2]\cup T_{uv}$, there exists a $a_2\in [3, 4]\cap C(u_2)$,
and from \ref{12out5notinu1u2}(1.1) that $H$ contains a $(2, 5)_{(u_1, u_2)}$-path,
from \ref{12out6notinu1u2}(1.1) that $H$ contains a $(2, 6)_{(u_1, u_2)}$-path.

One sees clearly that $3, 4\not\in C(u_1)$.
It follows from \ref{125u4u1oru4u2}(5) that $5\in C(u_4)$,
and \ref{12out6notinu3}(5) that $6\in C(u_3)$.
Hence, $T_3\ne \emptyset$ and $T_4\ne \emptyset$.
It follows from \ref{12T4123}(3) and
Corollary~\ref{auv2}(4) that $T_2\subseteq C(u_4)$, for each $j\in T_4$,
there exists a $a_4\in [1, 3]\cap C(u_4)$, $b_6\in \{1, 2, 5, 7\}\setminus \{a_4\}$ such that $H$ contains a $(a_4, j)_{(u_4, u)}$-path, a $(b_6, j)_{(u_4, v)}$-path,
and from \ref{12T3124}(3) and
Corollary~\ref{auv2}(4) that $T_2\subseteq C(u_3)$, for each $j\in T_3$,
there exists a $a_3\in \{1, 2, 4\}\cap C(u_3)$, $b_5\in \{1, 2, 6, 7\}\setminus \{a_3\}$ such that $H$ contains a $(a_3, j)_{(u_3, u)}$-path, a $(b_5, j)_{(u_3, v)}$-path.
Thus, $\sum_{x\in N(u)\setminus \{v\}}d(x)\ge |[1, 2]| + |\{2, 5, 6, a_2\}| + |\{3, 5, 6, a_3, b_5\}| + |\{4, 5, 6, a_4, b_6\}| + 2t_{uv} + t_0\ge 2\Delta + 12 + 2 = 2\Delta + 14$, a contradiction.
\end{proof}

Assume that $1\not\in S_u$.
It follows from Corollary~\ref{auv2d(u)5}(7.2) that $3, 4\in C(u_1)$, $T_1\subseteq C(u_3)\cap C(u_4)$,
from \ref{12outRineemptyset}(3) that $2\in C(u_1)$ or $H$ contains a $(2, i)_{(u_2, u_i)}$-path,
from \ref{12out5notinu1u2}(1.1) that $5\in C(u_1)$ or $H$ contains a $(4, 5)_{(u_1, u_4)}$-path,
from \ref{12out5notinu1u2}(2.1) that $5\in C(u_2)$ or $H$ contains a $(4, 5)_{(u_2, u_4)}$-path,
and from \ref{12out6notinu1u2}(1.1) that $6\in C(u_1)$ or $H$ contains a $(3, 6)_{(u_1, u_3)}$-path,
from \ref{12out6notinu1u2}(2.1) that $6\in C(u_2)$ or $H$ contains a $(3, 6)_{(u_2, u_3)}$-path.
It follows that mult$_{S_u\setminus C(u_3)}(5)\ge 2$ and mult$_{S_u\setminus C(u_4)}(6)\ge 2$.
Since the same argument applies if $7\not\in S_u$ and $uu_2\to 7$, $7\in S_u$.
Then $\sum_{x\in N(u)\setminus \{v\}}d(x)\ge |\{1, 3, 4\}| + |\{2, a_2\}| + |\{3, 5\}| + |\{4, 6\}| + |\{2, 7\}| + 2\times |\{5, 6\}| + 2t_{u v} + t_0 = 2\Delta + 11 + t_0\ge 2\Delta + 11 + t_1$.
One sees from $1, 3, 4, 2, a_2, 5, 6\in W_1\biguplus W_2$ that $t_1 + t_2\ge 3$.
It follows that $T_2\setminus T_0\ne \emptyset$.
Assume w.l.o.g. $T_2\setminus C(u_4)\ne \emptyset$.
Then $3\in C(u_2)$ and $T_2\subseteq C(u_3)$.
It follows from Corollary~\ref{auv2}(5.1) that $2\in C(u_3)$,
or $H$ contains a $(2, i)_{(u_3, u_i)}$-path for each $i\in \{1, 4\}$,
and from Corollary~\ref{auv2}(5.2) that mult$_{S_u}(4)\ge 2$.
One sees that if $2\not\in C(u_1)\cup C(u_4)$,
then $H$ contains a $(2, 4)_{(u_1, u_4)}$-path and a $(2, 4)_{(u_3, u_4)}$-path, a contradiction.
Hence, $2\in C(u_1)\cup C(u_3)$, and thus,
$d(u_1) + d(u_2) + d(u_3)\ge |\{1, 2, 3, 5\}| + 2\times |\{3, 4, 6\}| + |\{2, 5\}| + 2t_{u v} = 2\Delta + 8$.
Since the same argument applies if $7\not\in S_u$ and $u u_1\to 7$,
we proceed with the following proposition, or otherwise we are done:
\begin{proposition}
\label{12out127Su}
$1, 2\in S_u$, and $7\in S_u$
\end{proposition}

\begin{lemma}
\label{12outT2ge2}
For each $i\in [1, 2]$, $4\le w_i\le 5$, $t_i\ge 2$ and $\min\{w_1, w_2\} = 4$.
\end{lemma}
\begin{proof}
Assume that $w_1 = 3$.
By \ref{12out56u1u2}, assume w.l.o.g. $W_1 = \{1, 4, 6/2/3\}$.
Below we distinguish three cases where $W_1 = \{1, 4, i\}$, $i = 3, 6, 2$, respectively.

{\bf Case 1.} $W_1 = \{1, 4, 3\}$.
It follows from \ref{12out5notinu1u2}(1.1) that $H$ contains a $(4, 5)_{(u_1, u_4)}$-path, $5\in C(u_4)$,
from \ref{12out6notinu1u2}(1.1) that $H$ contains a $(6, 3)_{(u_1, u_3)}$-path, $6\in C(u_3)$,
and from \ref{12out5notinu1u2}(2.1) that $5\in C(u_2)$,
from \ref{12out6notinu1u2}(2.1) that $6\in C(u_2)$.
One sees that $\sum_{x\in N(u)\setminus \{v\}}d(x)\ge |\{1, 3, 4\}| + |\{2, 5, 6, a_2\}| + |\{3, 5, 6\}| + |\{4, 6, 5\}| + 2t_{u v} + t_0 = 2\Delta + 9 + t_0$.

{Case 1.1.} $1\not\in C(u_3)$ and $H$ contains no $(1, i)_{(u_3, u_i)}$-path for each $i\in \{2, 4\}$.
One sees from Corollary~\ref{auv2d(u)5}(7.2) that $T_1\subseteq C(u_4)$.
If $H$ contains no $(6, j)_{(u_1, u_4)}$-path for some $j\in T_2$,
then $(u u_1, u u_3, u v)\overset{\divideontimes}\to (6, 1, j)$.
Otherwise, $H$ contains a $(6, T_2)_{(u_1, u_4)}$-path and $T_2\subseteq C(u_4)$.
If there exists a $j\in (T_2\cup T_1)\setminus C(u_3)$,
then $(u u_1, u u_3)\to (6, j)$ reduces the proof to Case (2.5.1.).
Otherwise, $T_1\cup T_2\subseteq C(u_3)$.
Thus, $\sum_{x\in N(u)\setminus \{v\}}d(x)\ge 2\Delta + 9 + t_0 + |\{1, 2, 7\}|\ge 2\Delta + 12 + t_1 + t_2\ge 2\Delta + 12 + 1 + 2 = 2\Delta + 15$,
and none of ($A_{8.1}$)--($A_{8.4}$) holds.

{Case 1.2.} $1\in C(u_3)$ or $H$ contains a $(1, 2/4)_{(u_3, u)}$-path,
and $1\in C(u_4)$ or $H$ contains a $(1, 2/3)_{(u_4, u)}$-path.
It follows that there exists a $a_3\in \{1, 2, 4\}\cap C(u_3)$ and a $a_4\in \{1, 2, 3\}\cap C(u_4)$.
One sees clearly that $1\in C(u_3)\cup C(u_4)$, say $1\in C(u_3)$,
and if $1, 3\not\in C(u_4)$, i.e., $a_4 = 2$, then $1\in C(u_2)$.
Then $\sum_{x\in N(u)\setminus \{v\}}d(x)\ge 2\Delta + 9 + t_0 + |\{1\}| + |\{a_4\}| + |\{7, 2/1/3\}| = 2\Delta + 13 + t_0$.
It follows that $t_0 = 0$, $C(u_2)\cap [3, 4] = \{a_2\}$,
and $4\not\in C(u_3)$.
Then $T_2\subseteq C(u_{a_2})$,
and it follows from Corollary~\ref{auv2}(5.2) that mult$_{S_u}(i)\ge 2$, where $\{a_2, i\} = [3, 4]$.
Hence, $a_2 = 4$, $a_4 = 3$, $1\not\in C(u_2)\cup C(u_4)$,
and $T_2\subseteq C(u_4)$, $T_2\cap C(u_3) = \emptyset$.
Thus, $(u u_1, u u_3, u u_4, u v)\overset{\divideontimes}\to (5, 4, 1, T_2)$.

{\bf Case 2.} $W_1 = \{1, 4, 6\}$.
It follows from \ref{12out5notinu1u2}(1.1) that $H$ contains a $(4, 5)_{(u_1, u_4)}$-path, $5\in C(u_4)$,
from \ref{12out5notinu1u2}(2.1) that $5\in C(u_2)$,
and from Corollary~\ref{auv2}(5.1) that $1\in C(u_4)$, or $H$ contains a $(1, 2/3)_{(u_4, u)}$-path.
When $6\not\in C(u_2)\cup C(u_3)$,
one sees from \ref{12out6notinu1u2}(2.1) that $1\in C(u_2)$ and $H$ contains a $(1, 6)_{(1, 2)}$-path.
Then it follows from \ref{12out6notinu1u2}(2.2) that $2\in C(u_4)$,
from \ref{12out6notinu3}(2) that $T_2\subseteq C(u_4)$,
and from \ref{12out6notinu3}(4) that $3\in C(u_4)$.
One sees from $1, 2, 5, a_2\in C(u_2)$ that $t_2\ge 2$.
Thus, $d(u_4)\ge |[2, 6]| + t_1 + t_2\ge 5 + 1 + 2 = 8$.

In the case, $6\in C(u_2)\cup C(u_3)$.
One sees that $\sum_{x\in N(u)\setminus \{v\}}d(x)\ge |\{1, 4, 6\}| + |\{2, 5, a_2\}| + |\{3, 5\}| + |\{4, 6, 5\}| + |\{1, 2, 6, 7\}| + 2t_{u v} + t_0 = 2\Delta + 11 + t_0$.
We distinguish whether or not $T_2\subseteq C(u_4)$:

{Case 1.1.} $T_2\setminus C(u_4)\ne \emptyset$.
Then $3\in C(u_2)$, $T_2\subseteq C(u_3)$,
and it follows from Corollary~\ref{auv2}(5.2) that $4\in C(u_2)\cup C(u_3)$.
One sees that $T_1\subseteq C(u_3)$ or $3\in C(u_4)$.
Then $\sum_{x\in N(u)\setminus \{v\}}d(x)\ge 2\Delta + 11 + |\{4\}| + |\{3/T_1\}|\ge 2\Delta + 13 + t_0\ge 2\Delta + 13$.
It follows that mult$_{S_u}(i) = 1$, $i = 1, 2$, mult$_{S_u}(4) = 2$.
One sees from Corollary~\ref{auv2}(5.1) that $2\not\in C(u_4)$ or $H$ contains a $(2, 4)_{(u_3, u_4)}$-path.
It follows from \ref{12outRineemptyset}(3) that $1\in C(u_2)$ or $H$ contains a $(1, 3/4)_{(u_2, u)}$-path.
Hence, $1\in C(u_2)\cup C(u_4)$.

If $1\in C(u_2)$, one sees that $2\in C(u_4)$, $4\in C(u_3)$,
then $(u u_1, u u_2, u u_3, u u_4)\to (T_1, 4, 2, 1)$ reduces the proof to Case 2.3.
Otherwise, $1\in C(u_4)$, one sees that $4\in C(u_2)$, $2\in C(u_3)$,
then $(u u_1, u u_2, u u_3, u u_4)\to (T_1, 1, 2, 4)$ reduces the proof to Case 2.4.3.

{Case 1.2.} $T_2\subseteq C(u_4)$.
When $3\not\in C(u_4)$, it follows from Corollary~\ref{auv2}(5.2) that $T_1\cup T_2\subseteq C(u_3)$,
then $\sum_{x\in N(u)\setminus \{v\}}d(x)\ge 2\Delta + 11 + t_0\ge 2\Delta + 11 + t_1 + t_2\ge 2\Delta + 11 + 1 + 1 = 2\Delta + 13$.
It follows that $t_1 = t_2 = 1$, mult$_{S_u}(i) = 1$, $i = 1, 2$, and $4\not\in C(u_3)$, $3\not\in C(u_4)$.
One sees that $W_2 = \{2, 5, a_2\}$.
It follows from \ref{12out56u1u2} that $a_2 = 3$.
Hence, by Corollary~\ref{auv2}(5.1), $1\in C(u_4)\setminus C(u_3)$ and $2\in C(u_3)\setminus C(u_4)$.
Thus, by \ref{12outRineemptyset}, we are done.
In the other case, $3\in C(u_4)$.
Then $\sum_{x\in N(u)\setminus \{v\}}d(x)\ge 2\Delta + 11 + t_0 + |\{3\}| = 2\Delta + 12 + t_0$.
It follows that $(T_1\cup T_2)\setminus C(u_3)\ne \emptyset$,
and from Corollary~\ref{auv2}(5.2) and (3) that $3\in B_1$.
One sees that $4\in C(u_2)$ or $T_2\subseteq C(u_2)$.
Then $\sum_{x\in N(u)\setminus \{v\}}d(x)\ge 2\Delta + 12 + |\{4/T_2\}| = 2\Delta + 13$.
It follows that mult$_{S_u}(2) = 1$, $4\not\in C(u_3)$,
and if $4\in C(u_2)$, then $H$ contains a $(4, T_2)_{(u_2, u_4)}$-path.
If $2\not\in C(u_4)$, then $(u u_2, u u_3, u u_4)\to (T_2, 4, 2)$ reduces the proof to Case 2.4.3.
Otherwise, $2\in C(u_4)$ and $C(u_4) = [2, 6]\cup T_1\cup T_2$.
One sees clearly that $W_2 = \{2, 3, 5\}$ and $C(u_2) = \{3, 5, 6, 1, 7\}\cup T_2$.
Thus, $(u u_2, u u_3, u v)\overset{\divideontimes}\to (T_2, 2, T_1)$
\footnote{Or, by Corollary~\ref{auv2d(u)5}(7.2), we are done. }.

{\bf Case 3.} $W_1 = \{1, 4, 2\}$.
Then $T_1\subseteq C(u_1)$.
It follows from Corollary~\ref{auv2d(u)5}(5.1) that $1\in C(u_4)$,
or $H$ contains a $(1, 2/3)_{(u_4, u)}$-path,
and from \ref{12out5notinu1u2}(1.1) that $H$ contains a $(4, 5)_{(u_1, u_4)}$-path,
or a $(2, 5)_{(u_1, u_2)}$-path.
Together with \ref{12out5notinu1u2}(2.1) and \ref{125u4u1oru4u2}(1), we have $5\in C(u_2)\cap C(u_4)$.
It follows from \ref{12out6notinu1u2}(1) that $H$ contains a $(2, 6)_{(u_1, u_2)}$-path,
and from \ref{12out6notinu1u2}(1.2) that $1\in C(u_2)$, or $H$ contains a $(1, 3/4)_{(u_2, u)}$-path.
Hence, $1\in C(u_2)\cup C(u_4)$.
One sees that $t_2\ge 2$ and $\sum_{x\in N(u)\setminus \{v\}}d(x)\ge |\{1, 2, 4\}| + |\{2, 5, 6, a_2\}| + |\{3, 5\}| + |\{4, 6, 5\}| + |\{1, 7\}| + 2t_{u v} + t_0 = 2\Delta + 10 + t_0$.
Further, if $3\in C(u_4)$, then $T_2\setminus C(u_4)\ne \emptyset$,
or else $d(u_2) + d(u_4)\ge t_{uv} + t_1 + w_1 + w_4\ge \Delta + 8$.
One sees from \ref{12out6notinu3}(5) that $6\in C(u_3)$ or $3\in C(u_4)$.

If $T_2\subseteq C(u_4)$, one sees that $3\not\in C(u_4)$, then $6\in C(u_3)$,
and it follows from Corollary~\ref{auv2}(5.1) that $T_1\cup T_2\subseteq C(u_3)$,
thus, $\sum_{x\in N(u)\setminus \{v\}}d(x)\ge 2\Delta + 10 + t_0 + |\{6\}|\ge 2\Delta + 11 + t_1 + t_2\ge 2\Delta + 11 + 1 + 2 = 2\Delta + 14$.
Otherwise, $T_2\setminus C(u_4)\ne \emptyset$.
Then $3\in C(u_3)$, and it follows from Corollary~\ref{auv2}(5.2) that $4\in C(u_2)\cup C(u_3)$,
from Corollary~\ref{auv2}(5.1) that $2\in C(u_3)\cup C(u_4)$ or $1\in C(u_3)$,
and $3\in C(u_4)$ or $T_1\subseteq C(u_3)$.
Hence, $\sum_{x\in N(u)\setminus \{v\}}d(x)\ge 2\Delta + 10 + t_0 + |\{4, 1/2, 3/T_1\}/|\ge 2\Delta + 13$.
It follows that $6\not\in C(u_3)$, $3\in C(u_4)$, and $t_0 = 0$, i.e., $T_1\cap C(u_3) = T_2\cap C(u_4) = \emptyset$,
mult$_{S_u}(1) + $mult$_{S_u}(2) = 3$, and mult$_{S_u\setminus C(u_4)}(4) = 2$.
If $1\in C(u_2)$, one sees that $2\in C(u_4)$, $4\in C(u_3)$,
then $(u u_1, u u_2, u u_3, u u_4)\to (T_1, 4, 2, 1)$ reduces the proof to Case 2.3.
Otherwise, $1\in C(u_4)$ and $H$ contains a $(1, 4)_{(u_2, u_3)}$-path.
One sees that $4\in C(u_2)$, $2/1\in C(u_3)$.
Then $(u u_1, u u_2, u u_3, u u_4, u v)\overset{\divideontimes}\to (6, 1, 4, 2, T_2)$.

Hence, for each $i\in [1, 2]$, $w_i\ge 4$.
Together with $w_1 + w_2 + t_{uv}\le d(u_1) + c(u_2)\le \Delta + 7$,
$w_1 + w_2\le 9$.
It follows that $4\le w_1, w_2\le 5$ and $\min\{w_1, w_2\} = 4$.
\end{proof}

Hereafter, for each $i\in [1, 2]$, $w_i\ge 4$ and $t_i\ge 2$.

\begin{lemma}
\label{12outT1T2notinCu3Cu4}
For each $i\in [3, 4]$, $(T_1\cup T_2)\setminus C(u_i)\ne \emptyset$.
\end{lemma}
\begin{proof}
Assume that $T_1\cup T_2\subseteq C(u_4)$.
When $T_1\cup T_2\subseteq C(u_3)$,
$\sum_{x\in N(u)\setminus \{v\}}d(x)\ge |\{1, a_1\}| + |\{2, a_2\}| + |\{3, 5\}| + |\{4, 6\}| + |\{5, 6, 1, 2, 7\}| + 2t_{u v} + t_0\ge 2\Delta + 9 + t_1 + t_2\ge 2\Delta + 9 + 2 + 2 = 2\Delta + 13$.
It follows that mult$_{S_u\setminus C(u_3)}(5) = $ mult$_{S_u\setminus C(u_4)}(6) = 1$, $t_0 = t_1 + t_2 = 2 + 2 = 4$,
and from Corollary~\ref{auv2}(4) that $1/2/7/6\in C(u_3)$ and $1/2/7/5\in C(u_4)$.
Thus, $\sum_{x\in N(u)\setminus \{v\}}d(x)\ge w_1 + w_2 + w_3 + w_4 + 2t_{uv} + t_0\ge 2\times 4 + 2\times 3 + 2t_{uv} + 4 = 2\Delta + 14$.
In the other case, $(T_1\cup T_2)\setminus C(u_3)\ne \emptyset$.
It follows that $4\in C(u_1)\cup C(u_2)$,
and from  Corollary~\ref{auv2}(5.1) that mult$_{S_u}(3)\ge 2$.

One sees from Corollary~\ref{auv2}(7.2) that if $1/2/3\not\in C(u_4)$,
then $T_1\cup T_2\subseteq C(u_3)$.
Hence, $1/2/3\in C(u_4)$, $C(u_4) = \{4, 6, 1/2/3\}\cup T_1\cup T_2$, and $t_1 = t_2 = 2$.
It follows that $w_i = 4$, $d(u_i) = \Delta$, $i = 1, 2$.

{\bf Case 1.} $C(u_4) = \{4, 6, 3\}\cup T_1\cup T_2$.
It follows from \ref{125u4u1oru4u2}(5) that mult$_{S_u}(4)\ge 2$,
from Corollary~\ref{auv2}(4) that mult$_{S_u\setminus C(u_4)}(6)\ge 2$.
One sees that $\sum_{x\in N(u)\setminus \{v\}}d(x)\ge 6 + 2\times |\{3, 4, 6\}| + |\{5, 1, 2, 7\}| + 2t_{uv} + t_0\ge 2\Delta + 12 + t_0$.
It follows that $T_i\setminus C(u_3)\ne \emptyset$, $i = 1, 2$,
and from \ref{125u4u1oru4u2}(1--2) that $5\in C(u_1)\cap C(u_2)$.
Hence, assume w.l.o.g. $W_1 = \{1, 4, 5, 6\}$, $W_2 = [2, 5]$, $C(u_3) = [1, 3]\cup [5, 7]$,
and from Corollary~\ref{auv2}(5.1) that $H$ contains a $(1, 3)_{(u_3, u_4)}$-path.
Thus, by \ref{12outRineemptyset}(3), we are done.

{\bf Case 2.} $C(u_4) = \{4, 6, 1\}\cup T_1\cup T_2$.
Then $3\in C(u_1)\cap C(u_2)$.
When $H$ contains no $(3, j)_{(u_1, u_3)}$-path for some $j\in T_1$,
then $4\in C(u_1)$, and it follows from \ref{125u4u1oru4u2}(1) that $5\in C(u_1)$,
from \ref{126notinu4}(1) that $7\in C(u_1)$,
thus, $W_1\ge |\{1, 3, 4, 5, 7\}| = 5$.
In the other case, $H$ contains a $(3, T_1)_{(u_1, u_3)}$-path.

One sees from Corollary~\ref{auv2}(3--4) that  mult$_{S_u}(4)\ge 2$ or mult$_{S_u\setminus C(u_4)}(6)\ge 2$.

Then $\sum_{x\in N(u)\setminus \{v\}}d(x)\ge 6 + |\{3, 3, 4, 6, 5, 1, 2, 7, 4/6\}| + 2t_{uv} + t_0\ge 2\Delta + 11 + t_1 = 2\Delta + 11 + 2 = 2\Delta + 13$.
It follows that $T_2\cap C(u_3) = \emptyset$ and then $4\in C(u_2)$.
One sees from Corollary~\ref{auv2d(u)5}(7.2) that $2\in C(u_1)$,
and from \ref{12out5notinu1u2}(2.1) that $5\in C(u_2)$,
from \ref{126notinu4}(2) that $7\in C(u_1)\cup C(u_2)$.
Thus, $w_1 + w_2\ge |[2, 5]| + |[1, 3]| + |[6, 7]| = 9 > 8$, a contradiction.

Otherwise, $(T_1\cup T_2)\setminus C(u_4)\ne \emptyset$,
and likewise, $(T_1\cup T_2)\setminus C(u_3)\ne \emptyset$.
\end{proof}

Assuming w.l.o.g. $T_2\setminus C(u_4)\ne \emptyset$.
Then $3\in C(u_2)$, $H$ contains a $(3, T_2)_{(u_2, u_3)}$-path.
One sees from Corollary~\ref{auv2d(u)5}(5) that $T_2\subseteq C(u_3)$.
By Lemma~\ref{12outT1T2notinCu3Cu4}, $T_1\setminus C(u_3)\ne \emptyset$ and then $T_1\subseteq C(u_4)$.
Then $4\in C(u_1)$, $H$ contains a $(4, T_1)_{(u_1, u_4)}$-path.
It follows from Corollary~\ref{auv2}(5.2) and (3) that $3, 4\in B_1\cup B_2$ and mult$_{S_u}(i)\ge 2$, $i = 3, 4$.
One sees that $\sum_{x\in N(u)\setminus \{v\}}d(x)\ge |\{6\}| + 2\times |\{3, 4\}| + |\{1, 2, 7, 5, 6\}| + 2t_{uv} + t_0 = 2\Delta + 11 + t_0$.

When mult$_{S_u\setminus C(u_3)}(5)\ge 2$ and mult$_{S_u\setminus C(u_4)}(6)\ge 2$,
then $\sum_{x\in N(u)\setminus \{v\}}d(x)\ge +  2\Delta + 11 + t_0 + |\{5, 6\}| = 2\Delta + 13 + t_0$.
It follows that mult$_{S_u}(i) = 1$, $i = 1, 2$, mult$_{S_u}(j) = 2$, $j = 1, 2$.
If $1, 2\in C(u_3)\cup C(u_4)$,
one sees from  Corollary~\ref{auv2}(5.1) that $1\not\in C(u_3)$ or $H$ contains a $(1, 3)_{(u_3, u_4)}$-path,
and $2\not\in C(u_4)$ or $H$ contains a $(2, 4)_{(u_3, u_4)}$-path.
Thus, by \ref{12outRineemptyset}(3), we are done.
Otherwise, assume w.l.o.g. $1\in C(u_2)$.
Then $2\in C(u_4)$, $4\in C(u_3)$,
and $(u u_1, u u_2, u u_3, u u_4)\to (T_1, 4, 2, 1)$ reduces the proof to Case 2.3.
In the other case, assume w.l.o.g. mult$_{S_u\setminus C(u_3)}(5) = 1$.

If $5\in C(u_2)$, one sees from \ref{12out5notinu1u2}(1.1) that $H$ contains a $(2, 5)_{(u_1, u_2)}$-path,
then by \ref{125u4u1oru4u2}(1), we are done.
Otherwise, $5\in C(u_1)$,
and it follows from \ref{12out5notinu1u2}(2.1) that $1\in C(u_2)$, $H$ contains a $(1, 5)_{(u_1, u_2)}$-path.
If $H$ contains no $(5, j)_{(u_1, u_4)}$-path for some $j\in T_1\cap T_3$,
then $(u u_1, u u_4)\to (T_1\cap T_3, 5)$ reduces the proof to Case 2.4.1.
Otherwise, $H$ contains a $(5, T_1\cap T_3)_{(u_1, u_4)}$-path.
If $1\not\in C(u_4)$ and $H$ contains no $(1, 3)_{(u_3, u_4)}$-path,
then $(u u_1, u u_2, u u_4)\to (T_1\cap T_3, 5, 1)$ reduces the proof to Case 2.4.2.
Otherwise, $1\in C(u_4)$ or $H$ contains a $(1, 3)_{(u_3, u_4)}$-path.
Likewise, if mult$_{S_u\setminus C(u_4)}(6) = 1$, then $6\in C(u_2)$, $2\in C(u_1)$,
and $2\in C(u_3)$ or $H$ contains a $(2, 4)_{(u_3, u_4)}$-path.

Then $\sum_{x\in N(u)\setminus \{v\}}d(x)\ge +  2\Delta + 11 + t_0 + |\{1, 6/2\}| = 2\Delta + 13 + t_0$.
It follows that $t_0 = 0$, $(C_{12}\cup T_2)\cap C(u_4) = \emptyset$, $(C_{12}\cup T_1)\cap C(u_3) = \emptyset$,
and mult$_{S_u}(i) = 2$, $i = 3, 4$, and mult$_{S_u}(7) = 1$.

One sees that $5\not\in C(u_4)$ and $C_{12}\cap C(u_3) = \emptyset$.
Now we want to show that for some $j^*\in T_2\cup C_{12}\cup \{3\}$,
$H$ contains no $(i, j)_{(u_4, v)}$-path for each $i\in 1, 2, 7$.

Assume that $H$ contains a $(7, j)_{(u_4, v_7)}$-path for some $j\in T_2$.
If $H$ contains no $(4, 7)_{(u_1, u_4)}$-path,
then $(u u_1, u v)\overset{\divideontimes}\to (7, j)$.
Otherwise, $H$ contains a $(4, 7)_{(u_1, u_4)}$-path.
Thus, by \ref{126notinu3}(2), we are done.
Hence, $H$ contains no $(7, j)_{(u_4, v_7)}$-path for each $j\in T_2$.

Further, we assume that $H$ contains a $(2, T_2)_{(u_4, v_2)}$-path.
If $3\in C(v_2)$, one sees from \ref{12outRineemptyset}(1) that $C_{12}\setminus C(v_2) \ne \emptyset$,
and it follows from Lemma~\ref{auvge2}(1.2) that $H$ contains a $(7, C_{12}\cap R_2)_{(v_2, v_7)}$-path.
Then $j^*\in C_{12}\cap R_2$.
Otherwise, $3\not\in C(v_2)$.
Since $3\in B_1$, $3\in C(u_1)\setminus C(u_4)$.
One sees from \ref{123v13v2}(2) that $H$ contains a $(3, 7)_{(v_2, v_7)}$-path.
Then $j^* = 3$.
Hence, by Corollary~\ref{auv2}(4) or $vu_4\to 3$,
we have $6\in B_1\cup B_2$, and mult$_{S_u}(6) = 2$.
It follows that mult$_{S_u}(2) = 1$.

One sees from \ref{12out5notinu1u2}(2.2) that $2\in C(u_1)$ or $H$ contains a $(2, 3/4)_{(u_1, u)}$-path.
Together with $2\in C(u_3)$ or $H$ contains a $(2, 1/4)_{(u_3, u)}$-path, $2\in C(u_1)\cup C(u_3)$.
If $2\in C(u_2)$, one sees that $1\in C(u_3)$, $3\in C(u_4)$,
then $(u u_1, u u_2, u u_3, u u_4)\to (3, T_2, 2, 1)$ reduces the proof to Case 2.3.
Otherwise, $2\in C(u_3)$ and then $3\in C(u_1)$, $1\in C(u_4)$.
One sees from \ref{126notinu4}(1) that $7\in C(u_4)$,
and from \ref{126notinu3}(2) that $H$ contains a $(4, 7)_{(u_3, u_4)}$-path.
Then $(u u_1, u u_2)\to (7, 5)$ reduces the proof to Case (2.5.1.)
This finishes the inductive step for the Case (2.5.1.)

This finishes the inductive step for the case where $G$ contains the configuration ($A_6$)--($A_8$),
and thus completes the proof of Theorem~\ref{thm01}.

\section{Conclusions}
Proving the acyclic edge coloring conjecture (AECC) on general graphs has a long way to go,
but on planar graphs it is very close to reach with recent efforts.
In particular, the AECC has been confirmed true for planar graphs without $i$-cycles, for each $i \in \{3, 4, 5, 6\}$,
and for planar graphs that containing no intersecting $3$-cycles.
In this paper, we deal with the general planar graph.
Prior to our work, their acyclic chromatic index was shown to be at most $\Delta + 6$,
and in this work we made another step that $a'(G) \le \Delta + 5$.
The main technique is a discharging method to show that at least one of the several specific local structures must exist in the graph,
followed by an induction on the number of edges in the graph.

The discharging method is classic;
different from the previous work, we do not always validate $\omega'(x)\ge 0$ for every element in $x\in V(H)\cup F(H)$ independently,
but sometimes collectively for subsets of elements.
For example, in Discharging-Lemma~\ref{demma03} we validate $\sum_{x\in \{v\}\cup F_3(v)}\omega'(x)\le 0$,
and in Discharging-Lemma~\ref{demma04}, when $u$ is a $4$-vertex in $V(H)$ adjacent to a vertex $v$,
we validate $\sum_{x\in \{u, v\}\cup (F_3(u)\cup F_3(v))}\omega'(x)\le 0$ if $d_H(v) = 4$,
and $\sum_{x\in \{u\}\cup F_3(u)}\omega'(x)\le 0$ if $d_H(v)\ge 5$.

Some notations defined and some properties unearthed in the paper are of importance during the proof.
For example, the notations of $T_i$, $W_i$, $C_{ij}$, $S_u$, $T_0$,
and if $C(u)\cap C(v) = \{1, 2\}$, then $C(u_1) = W_1\cup T_2\cup C_{12}$, $C(u_2) = W_2\cup T_1\cup C_{12}$,
and the characteristic stated in Lemma~\ref{auvge2}, Corollary~\ref{auv2} and Corollary~\ref{auv2d(u)5}.

Surely, our detailed construction and proofs are built on top of several existing works,
and we believe our new design of the local structures and the associated discharging method can be helpful
for further lowering the upper bound of the acyclic chromatic number on planar graphs, or addressing the AECC eventually.
On the other hand, for general graphs, improvements from the approximation algorithm perspective,
that is, to reduce the coefficient for $\Delta$ in the upper bound (from the currently best $3.74$), are impactful pursuit.


\end{document}